\documentclass[10pt]{book}

\usepackage[braket, qm]{qcircuit}
\usepackage[makeroom]{cancel}

\usepackage{epigraph}
\usepackage{amsmath, amsfonts, mathtools}
\usepackage{amssymb, amsthm}
\usepackage{authblk}
\usepackage{braket}
\usepackage{graphicx}
\usepackage{upgreek}
\usepackage{bm}
\usepackage{tabularx}
\usepackage[numbers]{natbib}
\usepackage{nicefrac}
\usepackage{color}

\usepackage{afterpage}

\usepackage{imakeidx}
\usepackage[
       type={CC},
       modifier={by},
       version={4.0},
       imageposition={right},
     ]{doclicense}

\usepackage[top=2.1cm,left = 2.1cm, right=2.1cm]{geometry}

\usepackage{dsfont}

\usepackage{hyperref}
\PassOptionsToPackage{hyphens}{url}
\usepackage{url} 
\usepackage[noanswer]{exercise}

\counterwithin{Exercise}{section}

\usepackage{subcaption}
\newtheorem{thm}{Theorem}[section]
\newtheorem{lem}[thm]{Proposition}

\usepackage[capitalize]{cleveref}

\usepackage{enumitem}
\newlist{abbrv}{itemize}{1}
\setlist[abbrv,1]{label=,labelwidth=1in,align=parleft,itemsep=0.1\baselineskip,leftmargin=!}

\usepackage{listings}
\newcommand{\inline}[1]{\lstinline[columns=fixed]{#1}}

  \usepackage{comment}

\newcommand{\nodeflux}[2][]{\ifthenelse{\equal{#2}{dot}}{
		\ifthenelse{\equal{#1}{}}{\dot{\Phi}}{\dot{\Phi}_{#1}}}
	{\ifthenelse{\equal{#1}{}}{\Phi}{\Phi_{#1}}}}
\newcommand{\branchflux}[2][]{\ifthenelse{\equal{#2}{dot}}{
		\ifthenelse{\equal{#1}{}}{\dot{\Phi}_{\mathfrak{b}}}{\dot{\Phi}_{\mathfrak{b}=(#1)}}}
	{\ifthenelse{\equal{#1}{}}{\Phi_{\mathfrak{b}}}{\Phi_{\mathfrak{b}=(#1)}}}}
\newcommand{\extflux}[2][]{\ifthenelse{\equal{#2}{dot}}{
		\ifthenelse{\equal{#1}{}}{\dot{\Phi}_{\mathrm{ext}}}{\dot{\Phi}_{\mathrm{ext}}^{\mathrm{#1}}}}
	{\ifthenelse{\equal{#1}{}}{\Phi_{\mathrm{ext}}}{\Phi_{\mathrm{ext}}^{\mathrm{#1}}}}}
\newcommand{\nodephi}[2][]{\ifthenelse{\equal{#2}{dot}}{
		\ifthenelse{\equal{#1}{}}{\dot{\phi}}{\dot{\phi}_{#1}}}
	{\ifthenelse{\equal{#1}{}}{\phi}{\phi_{#1}}}}
\newcommand{\extphi}[2][]{\ifthenelse{\equal{#2}{dot}}{
		\ifthenelse{\equal{#1}{}}{\dot{\phi}_{\mathrm{ext}}}{\dot{\phi}_{\mathrm{ext}}^{\mathrm{#1}}}}
	{\ifthenelse{\equal{#1}{}}{\phi_{\mathrm{ext}}}{\phi_{\mathrm{ext}}^{\mathrm{#1}}}}}
\newcommand{\lagrangian}{\mathcal{L}}
\newcommand{\hamiltonian}{\mathcal{H}}

\newcommand{\Ohm}{\mathrm{\Omega}}
\newcommand{\mat}[1]{\bm{#1}}
\newcommand{\vect}[1]{\bm{#1}}
\newcommand{\vv}{\mathrm{v}}
\newcommand{\ii}{\mathrm{i}}
\newcommand{\zz}{\mathrm{z}}
\newcommand{\yy}{\mathrm{y}}
\newcommand{\scat}{\mathrm{s}}
\newcommand{\vb}{\vv_{\mathfrak{b}}}
\newcommand{\ib}{\ii_{\mathfrak{b}}}
\newcommand{\vbm}[1]{\bm{\vv}_{\mathfrak{b}_{#1}}}
\newcommand{\ibm}[1]{\bm{\ii}_{\mathfrak{b}_{#1}}}
\newcommand{\vbvec}{\bm{\vv}_{\mathfrak{b}}}
\newcommand{\ibvec}{\bm{\ii}_{\mathfrak{b}}}
\newcommand{\Phibvec}{\bm{\Phi}_{\mathfrak{b}}}

\begin{document}

\begin{titlepage}
\centering
\vspace*{1.5\baselineskip}

\rule{16cm}{1.6pt}\vspace*{-\baselineskip}\vspace*{2pt} 
	\rule{16cm}{0.4pt} 
	
		\vspace{0.75\baselineskip} 

{\Huge Lecture Notes on Quantum Electrical Circuits}

\vspace{0.2\baselineskip} 
	\rule{16cm}{0.4pt}\vspace*{-\baselineskip}\vspace{3.2pt} 
	\rule{16cm}{1.6pt} 
	
		\vspace{7.75\baselineskip} 

{\huge Alessandro Ciani} \\
\vspace{0.3cm}
{\Large \it Forschungszentrum J\"ulich} \\
\vspace{0.8cm}
{\huge David P. DiVincenzo} \\
\vspace{0.3cm}
{\Large \it Forschungszentrum J\"ulich} \\
\vspace{0.8cm}
{\huge Barbara M. Terhal}\\
\vspace{0.3cm}
{\Large \it Delft University of Technology} \\

\end{titlepage}

\noindent Title Open Textbook: Lecture Notes on Quantum Electrical Circuits \\
Authors: Alessandro Ciani, David P. DiVincenzo, Barbara M. Terhal \\
Publisher: TU Delft OPEN Publishing \\
Year of publication: 2024 \\
ISBN (softback/paperback): 978-94-6366-814-9 \\
ISBN (E-book): 978-94-6366-815-6  \\
DOI:  \url{https://doi.org/10.59490/tb.85} \\
Attribution cover image: 17-transmon quantum processor wire-bonded to a printed circuit board from DiCarlo lab, Delft University of Technology (2020) \\

\doclicenseThis 

This copyright license, which TU Delft OPEN Publishing uses for their original content, does not extend to or include any special permissions granted to us by the rights holders for our use of their content. CC-BY conditions do not apply to Figure 6.1(b).

Every attempt has been made to ascertain the correct source of images and other potentially copyrighted material and ensure that all materials included in this book have been attributed and used according to their license. If you believe that a portion of the material infringes someone else’s copyright, please contact the author b.m.terhal@tudelft.nl.

\vfill
\noindent Alessandro Ciani \\
$\circ$ Peter Gr\"unberg Institute: Institute for Quantum Computing Analytics (PGI-12), Forschungszentrum J\"ulich, 52425 J\"ulich, Germany \\
\href{mailto:a.ciani@fz-juelich.de}{a.ciani@fz-juelich.de} 
\vspace{0.5cm}

\noindent David P. DiVincenzo \\
$\circ$ Peter Gr\"unberg Institute: Theoretical Nanoelectronics (PGI-2), Forschungszentrum J\"{u}lich, 52425 J\"{u}lich, Germany \\
$\circ$ JARA Institute for Quantum Information, RWTH Aachen University, 52074 Aachen, Germany \\
$\circ$ QuTech, Lorentzweg 1, 2628 CJ Delft, The Netherlands \\
\href{mailto:d.divincenzo@fz-juelich.de}{d.divincenzo@fz-juelich.de} 
\vspace{0.5cm}

\noindent Barbara M. Terhal \\
$\circ$ Department of Applied Mathematics, Faculty EEMCS, Delft University of Technology, Mekelweg 4, 2628 CD Delft, The Netherlands \\
$\circ$ QuTech, Lorentzweg 1, 2628 CJ Delft, The Netherlands \\
\href{mailto:B.M.Terhal@tudelft.nl}{b.m.terhal@tudelft.nl} 

\tableofcontents

\newpage

\section*{Symbols and nomenclature}
\subsection*{List of constants}

\begin{abbrv}
\item[$h$ (resp. $\hbar$)]                  Planck constant resp. reduced Planck constant:  $\hbar=h/2 \pi$             
\item[$e$]                   Electron charge
\item[$m_e$]                 Electron mass
\item[$\Phi_0$]              Superconducting flux quantum: $\Phi_0 = \frac{h}{2 e}$
\item[$c$] Speed of light
\item[$k_B, T$] Boltzmann constant resp. temperature: $\beta=\frac{1}{k_B T}$
\end{abbrv}

\subsection*{List of symbols}
 \begin{abbrv}
\item[$\lagrangian$]  Lagrangian 
\item[$\mathfrak{L}$] Lagrangian density
\item[$\hamiltonian$] Classical Hamiltonian
\item[$H$]            Quantum Hamiltonian
\item[$C$]            Capacitance
\item[$E_C$]          Charging energy associated with a capacitance $C$: $E_C = \frac{e^2}{2 C}$
\item[$L$]            Inductance
\item[$E_L$]          Inductive energy associated with an inductance $L$: $E_L = \frac{\Phi_0^2}{4 \pi^2 L}$
\item[$E_J$]          Josephson energy associated with a Josephson junction: $E_J=\frac{\Phi_0 I_c}{2\pi}$ with critical current $I_c$
\item[$L_J$]          Josephson inductance associated with a Josephson junction: $L_J = \frac{\Phi_0^2}{4 \pi^2 E_J}$
\item[$\Phi(t)$]      Flux 
\item[$\phi(t)$]       Reduced (dimensionless) flux : $\phi(t) = 2 \pi \Phi(t)/\Phi_0=2 e\Phi(t)/\hbar$
\item[$\varphi(t)$] Phase in $[0,2\pi)$
\item[$Q(t)$]      Charge 
\item[$q(t)$]       Reduced (dimensionless) charge: $q(t) = \frac{Q(t)}{2 e}$
\item[$\vv(t)$ ($V(s)$)]       Voltage in time domain (resp. Laplace domain)
\item[$\ii(t)$ ($I(s)$)]       Current in time domain (resp. Laplace domain)
\item[$\zz(t)$ ($Z(s)$)]      Impedance in time domain (resp. Laplace domain)
\item[$\yy(t)$ ($Y(s)$)]      Admittance in time domain (resp. Laplace domain)
\item[$X,Y,Z$ and $\mathds{1}$] The three $2 \times 2$ Pauli matrices and the identity matrix or identity operator
\end{abbrv}

\clearpage

\chapter{Introduction}
\label{chap:intro}
\epigraph{``And I proceeded to think intensively about how one could modify the laws of mechanics so as to fit them in with this non-commutation.''}{Paul Dirac, \href{https://www.youtube.com/watch?v=2GwctBldBvU}{Lecture 1} on Quantum Mechanics, 1975}

The classical theory of electrical circuits and networks is a very well established subject \cite{newcomb,peikari}. During the last 30 years, stimulated by the quest to build superconducting quantum processors, a theory of quantum electrical circuits has emerged, which is called circuit quantum electrodynamics or circuit QED. Early pioneering work on building a quantum network theory for electrical circuits using canonical quantization was done by Yurke and Denker \cite{YD:qnetwork}. For quantum information processing, the circuits in question are superconducting and a central role is played by the Josephson junction.

The circuit theory aims to model and describe the dynamics of superconducting chips in which some elements readily admit a lumped-element description, i.e., capacitors or inductors. Other superconducting structures we will encounter are co-planar resonators, transmission lines and amplifiers, all operating in the microwave (roughly between $300 \, \mathrm{MHz} $ and $10\, \mathrm{GHz}$) regime. The goal of the theory is to provide a quantum description of the most relevant degrees of freedom. The central objects to be derived and studied are the Lagrangian and the Hamiltonian governing these degrees of freedom. Central concepts in classical network theory such as impedance and scattering matrices can be used to obtain the Hamiltonian and Lagrangian description for the lossless part of the circuits. Methods of analysis, both classical and quantum, can also be developed for nonreciprocal circuits. Losses can later be introduced and modeled in the quantum theory via dynamical equations such as the Lindblad master equation.  

These lecture notes aim at giving a comprehensive overview of this subject for theoretically-oriented Master or PhD students in physics and electrical engineering, as well as Master and PhD students who work on experimental superconducting quantum devices and wish to learn more theory. The text is supplemented by various exercises; answers can be found at the end of the book for the version available at \href{https://textbooks.open.tudelft.nl/textbooks}{TU Delft OPEN Publishing}. This book is based on lecture notes developed for Master courses in quantum technology at TU Delft and RWTH Aachen in the period 2016-2023, and was first published in book form in 2024. 

Compared to the theory of classical electrical networks, the theory of coherent superconducting quantum devices is by no means a finished topic and we anticipate that new elements and descriptions thereof will be added over time. We hope that the reader's curiosity is piqued by our descriptions and that this work encourages further reading into the subject; some suggestions are given below. We also hope that this book makes clear that the richness of the subject of circuit QED goes beyond the experimental effort of making good qubits or even a functioning superconducting quantum computer. 

Let us emphasize what we do {\em not} cover in this book. We do not discuss the important experimental subjects of chip fabrication, chip design, cooling $\&$ refrigeration and electronics, nor the topic of classical electromagnetic simulations of devices. Although we cover some aspects of noise, we do not focus on the (many) physical mechanisms which introduce loss and decoherence. We do not heavily focus on what can be done with the resulting circuit Hamiltonians, that is, the ways in which single or two-qubit gates and couplers can be realized with these Hamiltonians.

We refer the interested reader to the following papers which either have introduced or cover elements of the circuit QED formalism extensively: Devoret \cite{devoret}, Burkard, Koch and DiVincenzo \cite{BKD:circuit, burkard:circuit}, Girvin \cite{girvin}, Vool $\&$ Devoret \cite{Vool2016} and Minev {\em et.~al.} \cite{Minev2021}. The application of circuit QED for manipulating superconducting qubits can be found in the review by Blais, Grimsmo, Girvin and Walraff~\cite{Blais_2021}, and the review by Oliver and coworkers \cite{Krantz_2019}. Besides the QuCAT software package discussed in Section~\ref{sec:qucat}, a Python software package Scqubits for analyzing several superconducting qubits, discussed in Section~\ref{sec:examples} and Chapter~\ref{chap:transmon}, is described in \cite{software:scqubits}.

\section{Overview}

In Chapter~\ref{chap:lagr_ham} we consider electrical circuits and how each element contributes to the Lagrangian of the circuit. Chapter~\ref{chap:lagr_ham} also discusses some simple circuits that provide the basic understanding for the successful transmon qubit and the description of resonators, covered in detail in Chapter~\ref{chap:transmon}.

Chapter~\ref{chap:cq-app} is devoted to applying the theory to various circuits and understanding how this leads to the definition of various qubits. Some of the circuit examples discussed in this chapter provide interesting new qubits under current investigation which may have advantages over the transmon qubit. We hope to illustrate some of the richness and `plug-and-play' character of the formalism by including these.

Chapter~\ref{chap:transmon} is devoted to an in-depth discussion of the transmon qubit, resonators and transmission lines and how these can be coupled to manipulate the transmon qubit. A discussion of the many ways to realize two-qubit gates between transmon qubits is beyond the scope of this book.

Chapter~\ref{chap:ln} provides an introduction to the description of classical multi-node electrical networks using admittance, impedance and scattering matrices. This description is useful for understanding and generating the Hamiltonian of large quantum circuits and using classical EM simulations as input to the quantum Hamiltonian description.

Chapter~\ref{chap:nonrec} is devoted to non-reciprocal circuit elements and how these are included in the Lagrangian and Hamiltonian description of the circuit.

Chapter~\ref{chap:noise} gives a concise overview of the description of noise and the concept of noise protection.

Chapter~\ref{chap:add} collects some additional `stand-alone' exercises on quantum amplification.

Appendix~\ref{app:cc}, which accompanies Chapter~\ref{chap:lagr_ham}, reviews the general formalism of canonical quantization,
which is a standard technique in classical and quantum mechanics.
Appendix~\ref{app:anharmonic-born-opp} describes the general treatment of harmonic systems, as well as the formal elimination of high-energy variables for non-harmonic systems. This appendix also includes a standard circuit representation of an $N$-port linear electrical network which is used in so-called black-box quantization (Section~\ref{sec:bb}).

\section{Acknowledgements}
We thank Francesco Battistel, Mario Gely, Yaroslav Herasymenko, Martin Rymarz, Mac Hooper Shaw, Maarten Stroeks and Boris Varbanov for contributing their ideas, teaching assistance and input. We thank Leo DiCarlo and Christian Andersen for sharing images of superconducting processors fabricated in their labs. We thank Sander Bais for generous hospitality at Le Vialat at which some part of this work was completed.
This work was supported by QuTech NWO funding 2020-2024 – Part I “Fundamental Research” with project number 601.QT.001-1. B.M.T. and A.C. thank the OpenSuperQPlus100 project (no.~101113946) of the EU Flagship on Quantum Technology (HORIZON-CL4-2022-QUANTUM-01-SGA) for support. A. C. acknowledges funding from the Deutsche Forschungsgemeinschaft (DFG, German Research Foundation) under Germany's Excellence Strategy – Cluster of Excellence Matter and Light for Quantum Computing (ML4Q) EXC 2004/1 – 390534769 and from the German Federal Ministry of Education and Research in the funding program ``quantum technologies – from basic research to market" (contract number 13N15585). 

\chapter{Lagrangians and Hamiltonians for electrical circuits }
\label{chap:lagr_ham}

In this chapter we introduce the method of canonical quantization for electrical circuits, a procedure also known as circuit quantization \cite{Vool2016, BKD:circuit, rasmussen2021}. The general goal is to describe a lossless electrical circuit in terms of a Lagrangian that depends on appropriate (independent) degrees of freedom of the circuit. The Lagrangian determines the dynamics of the circuit, which as the reader might already imagine, corresponds to Kirchhoff's well-known laws on current and voltages. As explained in Appendix~\ref{app:cc}, from the Lagrangian we can obtain a Hamiltonian and quantize the system by promoting the variables to operators with appropriate commutation relations. As we will see, since we will assume that the circuit is made out of superconducting material, two main additions are needed to the standard theory of canonical quantization of electrical circuits: the presence of an additional, nonlinear element, the Josephson junction, and the additional condition of flux (or better fluxoid) quantization \cite{clarke2006squid, tinkham}. 

We start by introducing the concept of a branch and in particular the two main types of branches, capacitive and inductive, and their energies. These fundamental elements will represent the building blocks for larger circuits. We then provide several examples of simple circuits, starting with the LC oscillator, in order to build some intuition. We then move to describe the general procedure and explain how to handle the presence of external fluxes, as well as voltage or current sources.

\section{Branch voltages and branch fluxes}
\label{sec:canq_el}

As explained above, in order to apply canonical quantization to an electrical circuit we need to determine its Lagrangian which is in accordance with the known (classical) equations of motion, i.e., Kirchhoff's laws. Here we assume that the circuits we consider are lossless and thus the dynamics is energy-conserving. We will show that it is also possible to associate a classical Lagrangian to elements whose microscopic origin has its roots in quantum mechanics, and thus do not have a classical analog, such as the Josephson junction. In fact, we can still use known equations of motion, expressing the underlying dynamics of superconductors, to obtain a classical Lagrangian. In this chapter we mostly restrict ourselves to circuit elements which are branches, while in Chapter~\ref{chap:nonrec} we consider more general circuit elements. The collection of branches and nodes gives a graph, namely the electrical circuit graph $G$.  

We start by stating some conventions and identifying the relevant dynamical variables. A general lumped-element branch is shown in Fig.~\ref{fig:el_conv}. For an electrical circuit, each branch is characterized by the voltage $\vb(t)$ across the branch and the current $\ib(t)$ through the branch.

\begin{figure}[htb]
\centering
\includegraphics[height= 4 cm]{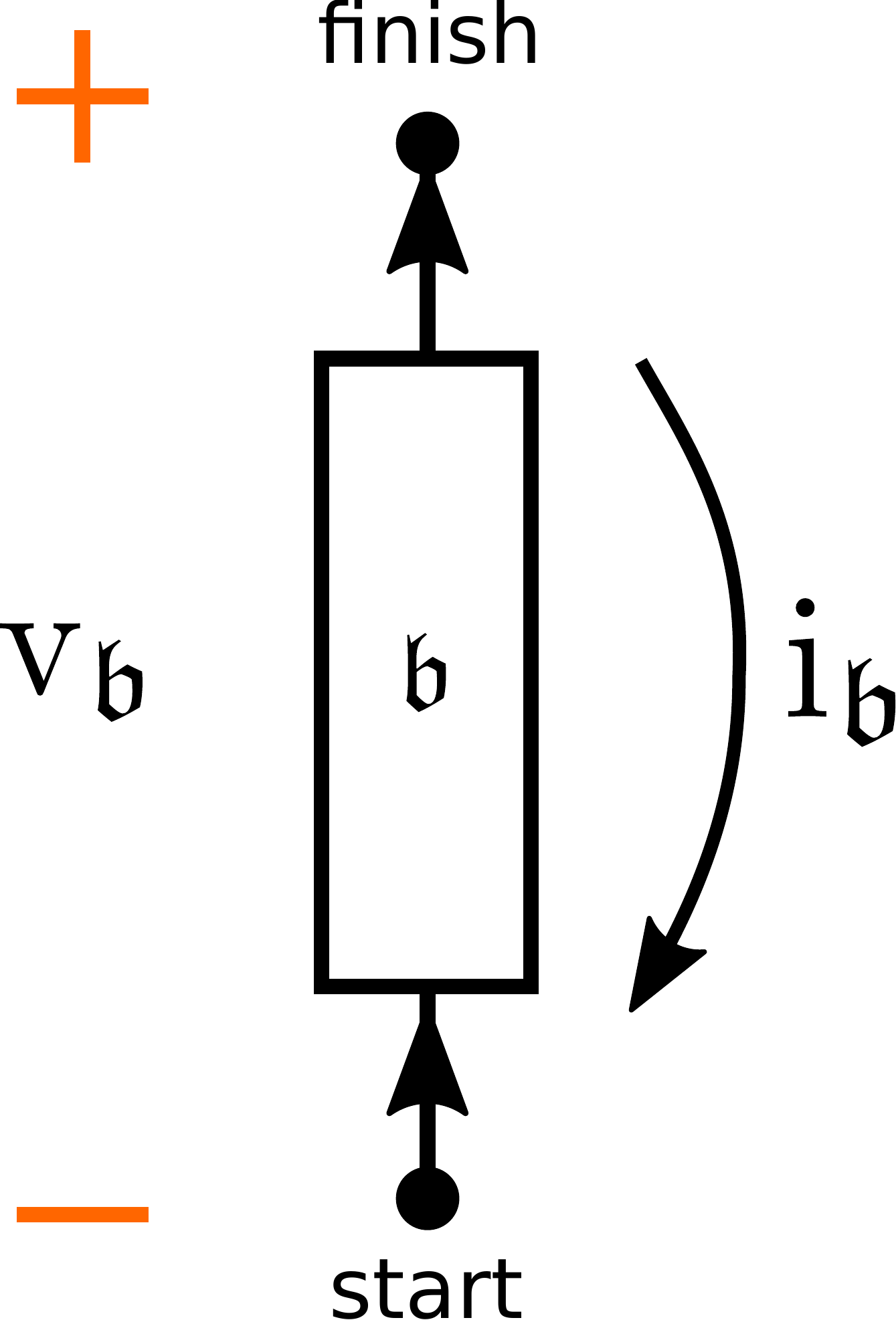}
\caption{General lumped-element branch $\mathfrak{b}$ with standard passive voltage and current convention. The $+$ and $-$ nodes are positive resp. negative voltages and positive current flows in the opposite direction of the arrow (orientation) on the branch.}
\label{fig:el_conv}
\end{figure}

Branches of a superconducting electrical circuit can be capacitors, inductors, Josephson junctions or voltage or current sources. All these branches except the voltage and current sources are conservative elements, meaning that they store or conserve energy, and do not dissipate it or increase it. Passive, as opposed to active elements such as current or voltage sources, amplifiers, or explicitly time-dependent elements do not increase energy. A resistor is an example of a non-conservative passive element.

As we see from Fig.~\ref{fig:el_conv}, each branch is a {\em two-terminal} element, i.e.,~an edge between two nodes or terminals. One can also consider electrical circuits which have four-terminal elements, as depicted in Fig.~\ref{fig:twop_network}. Each pair of terminals is then also called a port (see more in Section~\ref{sec:mport}). In this case, the voltages and currents at the two branches are functionally related to each other. Fundamental two-port elements are the mutual inductor, and the closely related ideal transformer, covered in Section~\ref{sec:mi}, and the gyrator, which we will study in Chapter~\ref{chap:nonrec}. 

In order to be precise, we need to assign a standard orientation (arrow) to each branch (so the graph $G$ has directed edges). As shown in Fig.~\ref{fig:el_conv}, when the arrow points from a node called start to a node called finish, the current variable of the branch is the positive current $\ib(t)$ which flows from finish to start. The voltage across the branch is defined as $\vb(t) \equiv \vv_{\rm finish}(t)-\vv_{\rm start}(t)$. 
Thus, we see the branch voltage as a difference of two node voltages, using the orientation on a branch.

Instead of working with voltages and currents ---which is natural from the electrical engineering perspective--- in order to obtain a Lagrangian formulation, we will work with (generalized) branch-flux variables $\Phi_{\mathfrak{b}}(t)$. By {\em definition}, this branch flux relates to the voltage across a branch as 
\begin{equation}
\vb(t)=\frac{d \Phi_{\mathfrak{b}}(t)}{dt}.
\label{eq:def-flux}
\end{equation}

We assume that at $t \rightarrow - \infty$ all electromagnetic fields are absent. This implies that voltages, currents, magnetic fluxes and charges can all taken to be zero at $t \rightarrow - \infty$. In what follows we will always work with this assumption, but we will specify this explicitly when it plays a role. Thus, we can write 
\begin{equation}\label{eq:branchflux}
\Phi_{\mathfrak{b}}(t)=\int_{-\infty}^t \vb(t') dt'.
\end{equation}

Similar to node voltages, we can also define node fluxes, see Section~\ref{sec:dep}, i.e., we have by convention for an oriented branch:
\begin{equation}
\Phi_{\mathfrak{b}}(t)=\Phi_{\rm finish}(t)-\Phi_{\rm start}(t),
\label{eq:def-orien}
\end{equation}
 where $\Phi_{\rm start}, \Phi_{\rm finish}$ are the node fluxes at the nodes labeled start and finish in Fig.~\ref{fig:el_conv}.
The unit of the branch flux is the weber ($1$ weber is $1$ volt $\times$ $1$ sec). If we were to sum the oriented branch fluxes around a loop $\gamma=\partial S$ in a circuit, we obtain the total magnetic flux through the surface $S$ defined by the loop. 

For superconducting circuits the branch-flux variable also directly relates to the superconducting phase drop over the piece of superconducting material representing the branch; we will come back to this in Section~\ref{subsec:joscpb} and later.

\begin{figure}
    \centering
    \includegraphics[height=4 cm]{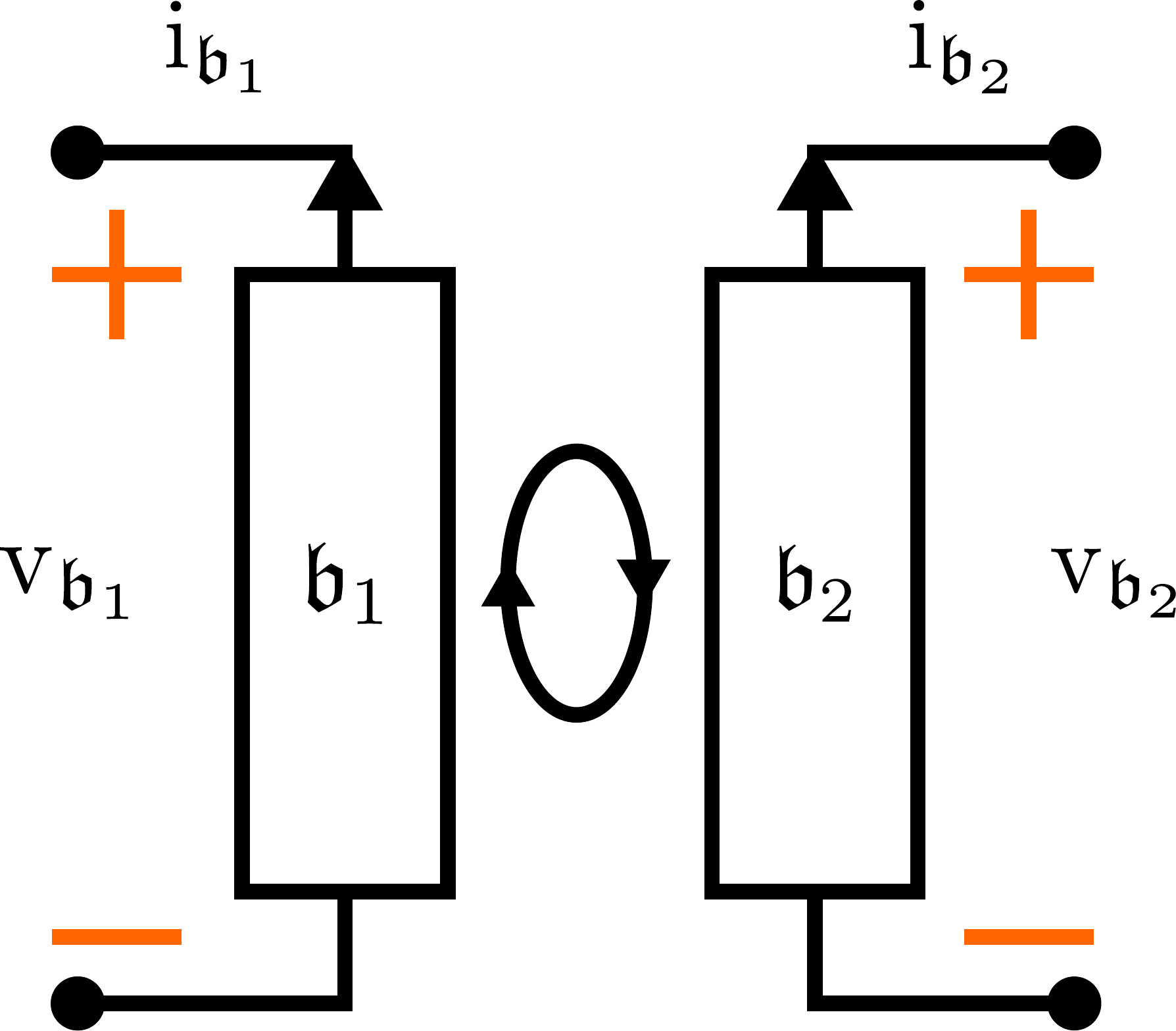}
    \caption{General two-port, or four-terminal, element.}
    \label{fig:twop_network}
\end{figure}

 \section[Inductive and capacitive branches]{Constructing the Lagrangian: inductive and capacitive branches}
\label{sec:ind-cap}

 For a given electrical circuit and its graph, the first task is to write down the Lagrangian of the circuit in terms of our chosen variables $\Phi_{\mathfrak{b}}(t)$ and their time derivatives $\dot{\Phi}_{\mathfrak{b}}(t)$. These variables will be the analog of positions and velocities in the mechanical setting. We can obtain the Lagrangian by considering the potential or kinetic energy stored in each conservative element. For the circuit elements that we consider in this chapter, the Lagrangian will be of the form
\begin{equation*}
 \lagrangian=T(\{\dot{\Phi}_\mathfrak{b}\})-U(\{\Phi_\mathfrak{b}\}),
 \end{equation*}
 where $T$ plays the role of the kinetic energy and $U$ the role of the potential energy. We refer to Appendix~\ref{app:cc} for a general formalism which describes how one obtains a Hamiltonian
 from a Lagrangian and how one then quantizes the system.
 The reason to use branch variables in $\lagrangian$ is that energy is stored in branches (not in nodes), but, as we will see later, it is also possible to obtain formulations in terms of independent node variables.
 
In order to write down the Lagrangian of the electrical circuit, we have to express the potential energy stored in the inductive branch elements and the kinetic energy stored in the capacitive branch elements at some time $t$.

Let us start our analysis with the inductive elements. A current $\ib(t)$ that flows on a branch for an infinitesimal time $d t$ transfers a charge $dt \ib(t)$ from the finish to the start node in Fig.~\ref{fig:el_conv}. This charge experiences a change of potential $\vb(t)$ and thus the infinitesimal work done ``on" the element is $d W = dt\, \ib(t) \vb(t)$. The total energy stored in the element is 

\begin{equation}
U_{\mathfrak{b}}(t)=\int_{-\infty}^t dt'  \ib (t') \vb(t').
\label{eq:def-U}
\end{equation}
For a general inductive branch one has the constitutive relation 
\begin{equation}
\label{eq:indbranch}
\ib(t)=g(\branchflux{}(t)),
\end{equation} 
where $g$ is some function. For a linear inductor, the function $g$ is linear, i.e.,~
\begin{equation}
\ib(t)=\frac{\Phi_{\mathfrak{b}}(t)}{L},
\label{eq:ind}
\end{equation}
with inductance $L$ (unit: henry=weber/ampere). For a piece of coil, the inductance can be the magnetic self-inductance (related to the number of windings of the coil) plus a so-called kinetic inductance, which comes from the kinetic energy of the charge carriers. The origin of the kinetic inductance for superconductors is that the superconducting phase variable $\varphi(t)$ (associated with the flux) varies over a piece of superconducting wire or a branch (as it can vary between two nodes in a lumped circuit). The energetic cost of this variation is captured by the kinetic inductance of the superconducting material, see also Section~\ref{subsec:fluxonium}.

Generally, the (magnetic) energy stored in an inductive branch is thus
\begin{equation}
U_{\mathfrak{b}}(t)=\int_{-\infty}^t  dt'g(\Phi_{\mathfrak{b}}(t'))\frac{d \Phi_{\mathfrak{b}}(t')}{dt'}  =\int_{\Phi_{\mathfrak{b}}(t=-\infty)}^{\Phi_{\mathfrak{b}}(t)}  d\Phi_{\mathfrak{b}}g(\Phi_{\mathfrak{b}})=G(\Phi_{\mathfrak{b}}(t)),
\label{eq:genpot}
\end{equation}
with primitive (or antiderivative) function $G(\Phi_{\mathfrak{b}}(t))$. Hence, for a linear inductor, one has 

\begin{equation}
U_{\mathfrak{b}}(t)=G(\Phi_b(t))=\frac{\Phi_{\mathfrak{b}}^2(t)}{2L}. 
\end{equation}

A general capacitive branch is defined as one in which the current is given by
\begin{equation}
\ib(t)= \frac{df(\vb(t))}{dt},
\end{equation}
 where $f$ is some function. For a linear capacitor the function $f$ is linear, namely 

 \begin{equation}
    \ib(t) = C \vb(t), 
 \end{equation}
 with capacitance $C > 0$ (unit: farad=coulomb/volt). This definition gives the well-known relation for a linear capacitor $Q_{\mathfrak{b}}=C \vb$ where the charge is defined as $Q_{\mathfrak{b}}(t)=\int _{-\infty}^t dt' \ib(t')$ or $C d \vv_{\mathfrak{b}}(t)/dt=d Q_{\mathfrak{b}}(t)/dt=\ib(t)$. 

For a linear capacitor we can write its kinetic energy as
\begin{align}
T_{\mathfrak{b}}(t)&=\int_{-\infty}^t dt'\frac{d f(\vb(t'))}{dt'} \vb(t') =\int_{\vb(t=-\infty)}^{\vb(t)} d \vb \frac{d f(\vb)}{d\vb}\vb  =\frac{1}{2} C\vb^2(t)=\frac{1}{2} C \dot{\Phi}_{\mathfrak{b}}^2(t),
\label{eq:en-cap}
\end{align} 
which can also be written as $T_{\mathfrak{b}}(t)=\frac{Q_{\mathfrak{b}}^2(t)}{2C} $. 
The kinetic energy can also be viewed as the integral of the amount of work done to induce a charge change $d Q_{\mathfrak{b}}$ given the current voltage $\vb=Q_{\mathfrak{b}}/C$. 

In what follows we omit writing the time dependence of the dynamical variables explicitly. We will write the time dependence explicitly only inside integral symbols and when needed for clarity.
 
\subsection{The LC oscillator and its canonical quantization}
\label{subsec:lc}

The parallel LC oscillator consists of a capacitor $C$ in parallel with an inductor $L$, see Fig.~\ref{fig:LCa}. Clearly, the voltages across the capacitive and the inductive branches are not independent variables since they are the same due to Kirchhoff's law for voltages, and thus there is only one independent branch variable, $\Phi_{\mathfrak{b}}$. This branch variable can be written as the difference between two node variables $\Phi_{\mathfrak{b}} = \Phi - \Phi_{\mathrm{ground}}$, where $\Phi$ is the node flux on the top node and $\Phi_{\mathrm{ground}}$ the node flux on the bottom (ground) node in Fig.~\ref{fig:LCa}. Since only the difference between the node fluxes matters, we can arbitrarily set $\Phi_{\mathrm{ground}} =0$, in which case we get simply $\Phi_{\mathfrak{b}} = \Phi$. The capacitor then contributes the kinetic energy $T=\frac{C}{2} \dot{\Phi}^2$ and the inductor contributes the potential energy $U=\frac{\Phi^2}{2L}$. We get the Lagrangian

\begin{equation}
\label{eq:lagr_lc}
\lagrangian=T-U=\frac{C}{2} \dot{\Phi}^2-\frac{\Phi^2}{2L}.
\end{equation}
Now we use the method described in Appendix~\ref{app:cc} to determine the Hamiltonian. In Appendix~\ref{app:cc} we can read that this requires that $\lagrangian$ is a convex function of $\dot{\Phi}$ which is clearly the case here.
So we define the conjugate variable 
\begin{equation}
Q=\frac{\partial \lagrangian}{ \partial \dot{\Phi}}=C \dot{\Phi},
\end{equation}
which leads to
\begin{equation}
\hamiltonian=\frac{Q^2}{2C} +\frac{\Phi^2}{2L}.
\end{equation} 
Note that this conjugate variable is precisely what we defined as the charge variable above, but here its definition comes about formally, not physically \footnote{For more general elements the conjugate variable defined through the Legendre transformation in going from the Lagrangian to the Hamiltonian may not be identifiable with a physical charge.}.

This Hamiltonian is like that of a mechanical harmonic oscillator with $\Phi$ as position, and $Q$ as momentum, $C$ playing the role of the mass and $L^{-1}$ being the spring constant. The classical equation of motion, i.e., the Euler-Lagrange equation associated with the Lagrangian $\lagrangian$, is 
\begin{equation}
\frac{d}{dt}\left(\frac{\partial\lagrangian}{\partial\dot{\Phi}}\right)-\frac{
\partial\lagrangian}{\partial\Phi}=0.
\end{equation}
For the Lagrangian in Eq.~\eqref{eq:lagr_lc} this gives 
\begin{equation}
C \ddot{\Phi}+ \frac{\Phi}{L}=0,
\end{equation}
which is analogous to the equation of motion of a spring.
In this electrical context we can also interpret this equation as Kirchhoff's current law on the current at the top node with the current through the inductor ($\Phi/L$) equal in magnitude to the current through the capacitor ($C \ddot{\Phi}$) with opposite sign. Notice that if Kirchhoff's current law is satisfied on the top node, then it is automatically satisfied on the bottom node. In our formalism, this is the reason why we can always set the node flux of an arbitrary node, the ground node, to zero, effectively removing one variable. 

Having identified the conjugate variables, we can quantize our system and replace $Q$ by $\hat{Q}$, and $\Phi$ by $\hat{\Phi}$, i.e., by Hermitian operators which take eigenvalues in $\mathbb{R}$ that satisfy the commutation relation
\begin{equation}
  [\hat{\Phi}, \hat{Q}]=i \hbar \mathds{1}.  
\end{equation} 

Let us also introduce useful dimensionless
charges and flux variables, which will be used throughout this book. They are defined as 
\begin{subequations}
\begin{equation}
   \phi = \frac{2 \pi \Phi}{\Phi_0},
   \label{phidef}
\end{equation}
\begin{equation}
   q  =\frac{Q}{2e},
\end{equation}
\label{eq:dimless}
\end{subequations}
with 
\begin{equation}
\Phi_0 = \frac{h}{2 e}
\end{equation}
the superconducting flux quantum ($\Phi_0 \approx 2 \times 10^{-15}$ weber). In terms of these variables the Hamiltonian becomes
\begin{equation}
\mathcal{H} = 4 E_C q^2 + \frac{E_L}{2} \phi^2,
\label{eq:cpb_Ham}
\end{equation}
where we define the charging energy $E_C$ 
\begin{equation}
   E_C = \frac{e^2}{2 C}, 
\label{eq:char_en}
\end{equation}
and the inductive energy $E_L$ 
\begin{equation}
   E_L = \frac{\Phi_0^2}{4 \pi^2 L}.
\label{eq:ind_en}
\end{equation} 
 
The commutation relation for the rescaled quantum operators $\hat{\phi}$ and $\hat{q}$ reads

\begin{equation}
  [\hat{\phi}, \hat{q}]=i \mathds{1}.   
   \label{eq:cr}
\end{equation} 

The quantum states of the system, i.e. 
\[\ket{\psi}=\int_{\mathbb{R}}d\phi \;\psi(\phi)\ket{\phi},
\]
correspond to square-integrable functions $\ell^2(\mathbb{R})$ with $\int_{\mathbb{R}}d\phi \;|\psi(\phi)|^2 < \infty$. 

As is standard for the quantum harmonic oscillator, one can introduce the annihilation operator 
\begin{equation}
\hat{a}=\frac{1}{\sqrt{2 L \hbar \omega_r}} \hat{\Phi}+i \frac{1}{\sqrt{2C \hbar \omega_r}} \hat{Q},\label{eq:annih}
\end{equation}
with $\omega_r$ the resonant frequency of the LC oscillator defined as

\begin{equation}
    \omega_r = \frac{1}{\sqrt{L C}}.
\end{equation}
The annihilation $\hat{a}$ and creation $\hat{a}^{\dagger}$ satisfy the bosonic commutation relations

\begin{equation}
[\hat{a}, \hat{a}^{\dagger}]=\mathds{1}.
\end{equation}
Using the operators $\hat{a}$ and $\hat{a}^{\dagger}$, the Hamiltonian can be rewritten as
\begin{equation}
H=\hbar \omega_r \biggl(\hat{a}^{\dagger}\hat{a}+\frac{1}{2} \biggr).
\label{eq:harmosc}
\end{equation}
In this book, when we quantize an (an)harmonic oscillator, we will often omit the vacuum energy term $\hbar \omega_r/2$ in Eq.~\eqref{eq:harmosc}.

We can observe that
\begin{subequations}
\begin{equation}
\hat{\Phi} = \Phi_{\rm zpf} (\hat{a}^\dagger + \hat{a}) =  \sqrt{\frac{\hbar Z_0}{2}} (\hat{a}^\dagger + \hat{a}), 
\end{equation}
\begin{equation}
\hat{Q} = i Q_{\rm zpf}(\hat{a}^{\dagger} - \hat{a})= i \sqrt{\frac{\hbar }{2 Z_0}} (\hat{a}^{\dagger} - \hat{a}),
\end{equation}
\label{eq:chargeflux_aadag}
\end{subequations}
with $Z_0 = \sqrt{L/C}$ the characteristic impedance of the oscillator.

The coefficients 
\begin{equation}
    \Phi_{\rm zpf}=\sqrt{\frac{\hbar}{2C \omega_r}}=\sqrt{\frac{\hbar Z_0}{2}}, \;\;Q_{\rm zpf}=\sqrt{\frac{\hbar C \omega_r}{2}}=\sqrt{\frac{\hbar}{2Z_0}},
    \label{eq:zpf}
\end{equation} 
represent the zero point fluctuations (zpf) of the flux and charge variable, respectively. They correspond to the standard deviation of the corresponding variables in the vacuum state, since $\ket{0}$ is $\bra{0} (\Delta\hat{\Phi})^2 \ket{0}=\Phi_{\rm zpf}^2$ and $\bra{0} (\Delta \hat{Q})^2 \ket{0}=Q_{\rm zpf}^2$, together obeying Heisenberg uncertainty $(\Delta \hat\Phi)^2 (\Delta \hat Q)^2 \geq \hbar^2/4$.

\begin{figure}
\centering
\begin{subfigure}[t]{0.45 \textwidth}
\centering
\includegraphics[height=4cm]{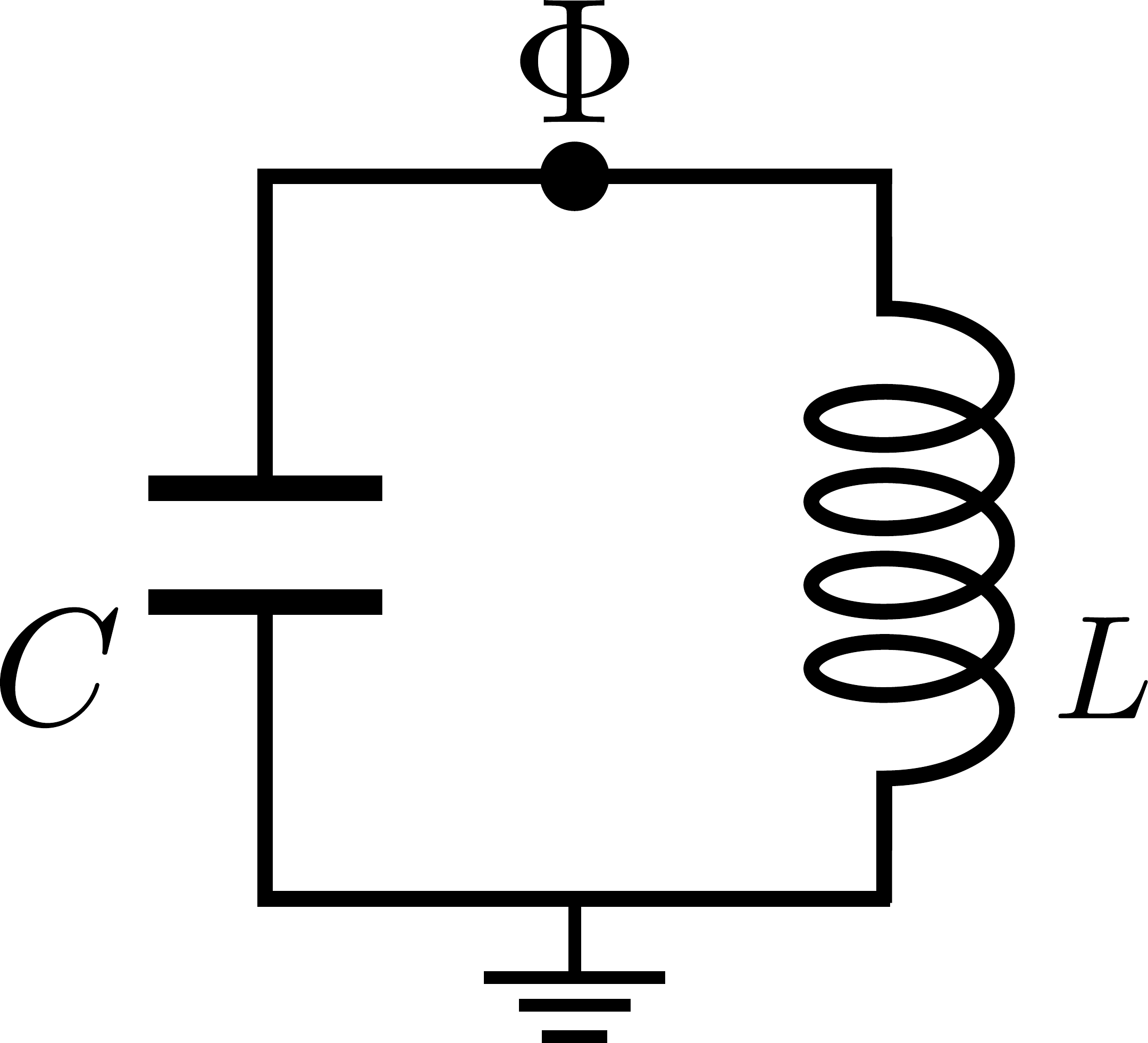}
\subcaption{}
\label{fig:LCa}
\end{subfigure}
\begin{subfigure}[t]{0.45 \textwidth}
\centering
\includegraphics[height=4cm]{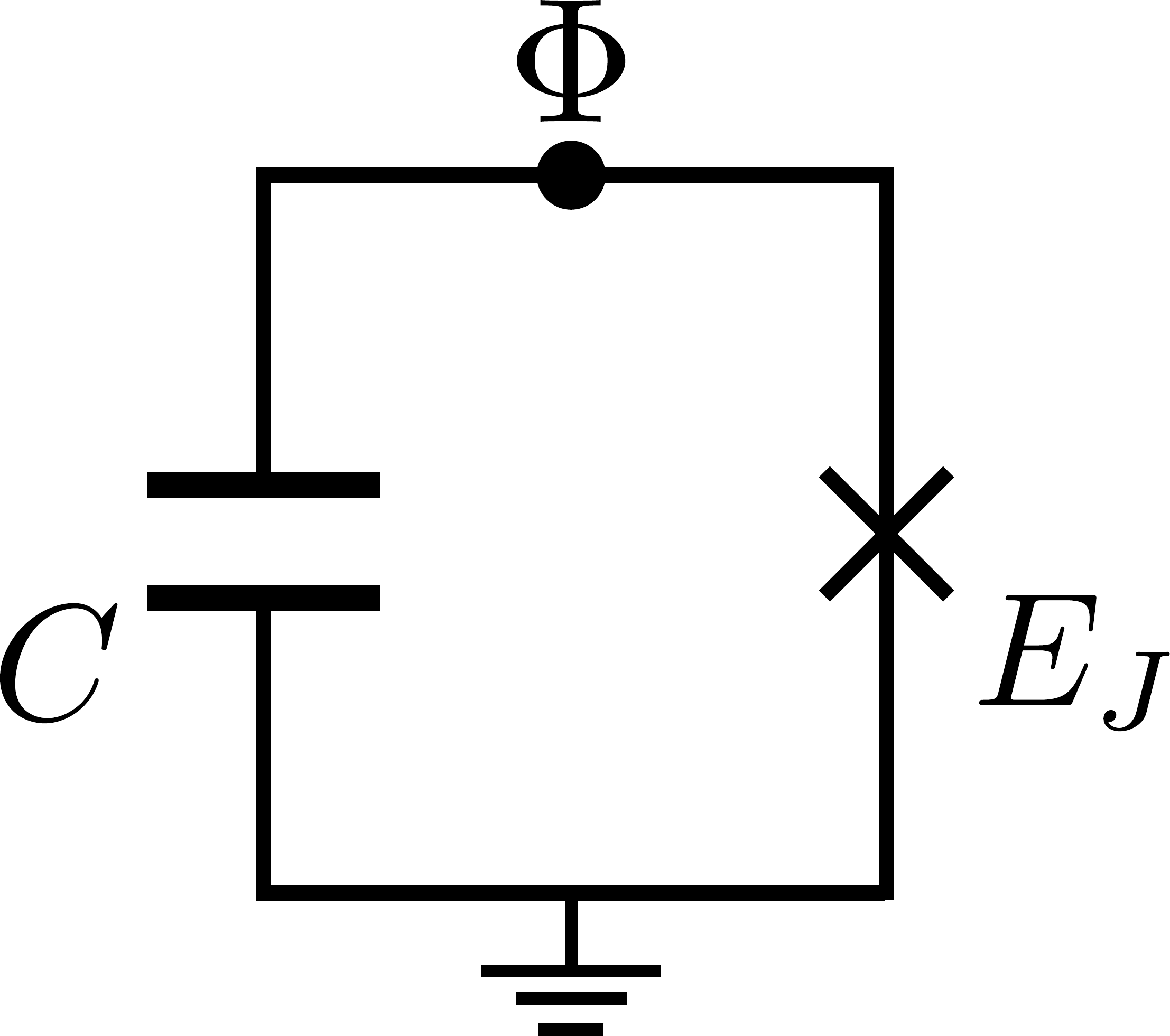}
\subcaption{}
\label{fig:LCb}
\end{subfigure}
\caption{(a) LC oscillator. (b) Replacing the inductor in the LC oscillator by a Josephson junction. The ground in these figures merely represents our arbitrary choice of reference node.}
\label{fig:LC}
\end{figure}

For completeness, we can also express the rescaled operators $\hat{\phi}$ and $\hat{q}$ in terms of annihilation and creation operators:
\begin{subequations}
\begin{equation}
\hat{\phi} = \phi_{\rm zpf} (\hat{a} + \hat{a}^{\dagger}) =  \biggl(\frac{2 E_C}{E_L}\biggr)^{\nicefrac{1}{4}}(\hat{a} + \hat{a}^{\dagger}), 
\label{eq:zpf-reduced}
\end{equation}
\begin{equation}
\hat{q} = i q_{\rm zpf}(\hat{a}^{\dagger} - \hat{a})= \frac{i}{2}\biggl (\frac{E_L}{2 E_C}\biggr)^{\nicefrac{1}{4}}(\hat{a}^{\dagger} - \hat{a}).
\end{equation}
\label{eq:zpf_resc}
\end{subequations}

\noindent\fbox{
    \parbox{\textwidth}{
        \centerline{\bf Why do we need a quantum description?} 
        The superconducting structures are made of metals such as aluminium (Al), niobium (Nb), tantalum (Ta) and niobium-titanium-nitride (NbTiN) patterned on, say, silicon substrates; these chips are cooled to tens of milli-Kelvins in dilution fridges, which is far below their superconducting transition temperature $T_c$. Due to being a superconducting instead of a normal metal, very little dissipation occurs; an oscillator can thus oscillate for at least $Q=10^5$ rounds or more. The dilution fridge temperature is much lower than the microwave frequency (using the conversion 50mK $\approx$ 1GHz) which implies that thermal excitations have a relatively small effect. At the same time, microwave photons are much less energetic than the superconducting gap (say, $T_c=1.1$K for bulk Al at zero magnetic field and zero pressure \cite{kittel}). Thus if we apply microwaves at very low intensity ---done by strongly attenuating room-temperature generated pulses--- we don't break superconductivity. Hence, we are in the regime of few excitations and few photons, necessitating a quantum description employing quantized energy levels. See more about loss and noise in Section \ref{sec:loss} and Chapter \ref{chap:noise}.}}

\subsection{The Josephson junction and the Cooper pair box}
\label{subsec:joscpb}

\begin{figure}
\centering
\includegraphics[height=5cm]{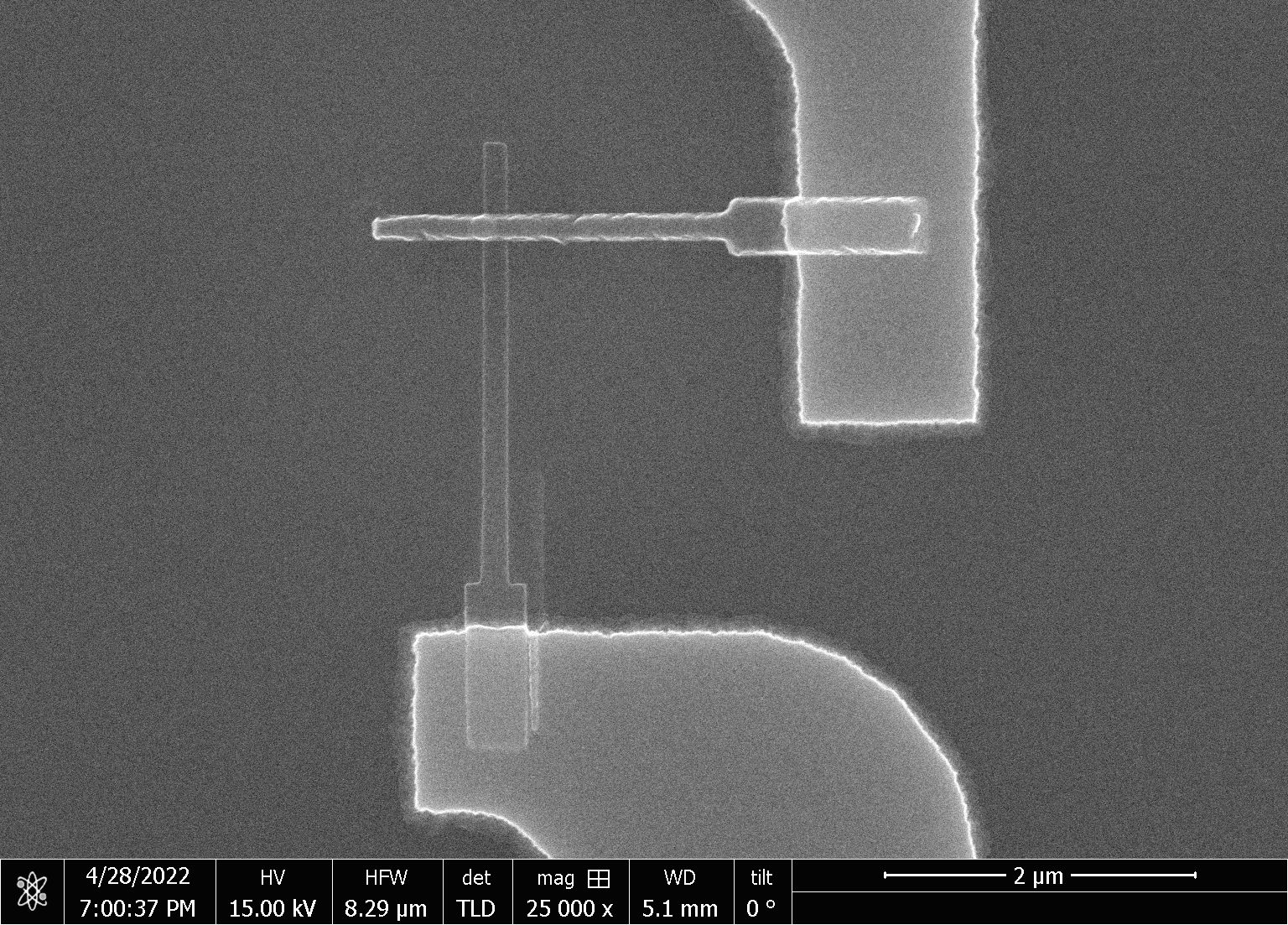}
\caption{So-called Manhattan-style Josephson junction from Ref.~\cite{nand}: the junction is formed where the two vingers overlap.}
\label{fig:JJ}
\end{figure}

A new inductive element for superconductors is the Josephson junction, see an example device in Fig.~\ref{fig:JJ}, for which the function $g(\cdot)$ in Eq.~\eqref{eq:indbranch} is nonlinear. A textbook introduction to the Josephson effect usually expresses the Josephson relations in terms of the superconducting phase difference $\varphi$ between two superconductors \cite{tinkham}, and thus not directly in terms of branch flux variables. In fact, the first Josephson relation reads
\begin{equation}
\ib=I_c \sin \varphi,
\label{eq:JJ-relation-phi}
\end{equation}
with $I_c >0$ the critical current, that is, the maximum current that can flow through the junction. However, the connection between the superconducting phase difference $\varphi$ and the branch flux variable $\Phi_{\mathfrak{b}}$ follows from the second Josephson relation:
\begin{equation}
    \frac{d\varphi}{dt}=\frac{2\pi}{\Phi_0}\vb. 
\end{equation}

If we identify $\varphi \equiv 2 \pi \Phi_{\mathfrak{b}}/\Phi_0$, we see that the second Josephson relation translates to the definition $\dot{\Phi}_{\mathfrak{b}}=\vb$ while the first Josephson relation gives a constitutive relation for the Josephson junction as a nonlinear inductor:
\begin{equation}
\ib=I_c \sin\left(\frac{2 \pi }{\Phi_0} \Phi_{\mathfrak{b}}\right).
\label{eq:JJ-relation}
\end{equation}
Observe that a possibly more precise modeling of a Josephson junction could only affect the constitutive relation in Eq.~\eqref{eq:JJ-relation} but not the second Josephson relation, as, given the relation of phase with flux, it is simply a definition. However, the identification $\varphi \equiv 2 \pi \Phi_{\mathfrak{b}}/\Phi_0$ cannot be entirely correct, as $\varphi$ is a phase in $[0,2\pi)$ while $\Phi$ takes values in $\mathbb{R}$. We will come back to this in Section~\ref{subsec:fp} in greater detail.
 
In~\cref{eq:JJ-relation-phi,eq:JJ-relation} the critical current $I_c$ of the Josephson junction depends on the junction's material, thickness and area. The critical current can be estimated from the so-called Ambegaokar-Baratoff formula \cite{ambegaokar1963} at zero temperature $I_c=\frac{\pi \Delta}{2 e R_n}$ with superconducting gap $\Delta$, and normal barrier resistance $R_n$. Variations in the barrier thickness ---the junction is about 100~nm~x~100~nm with only O(1)~nm thickness--- give variations in $R_n$ of $2\%$ or more~\cite{IBMannealing}, leading to variations in $I_c$ and thus the qubit parameters which depend on it.

The energy stored in a Josephson-junction branch can be derived using Eqs.~\eqref{eq:def-U},\eqref{eq:JJ-relation}:
\begin{equation}
U_J=-E_J\cos \biggl(\frac{2 \pi}{\Phi_0} \Phi_{\mathfrak{b}}\biggr),
\end{equation}
 with the Josephson energy
\begin{equation}
E_J=\frac{\Phi_0 I_c}{2\pi}.
\label{eq:defineEJ}
\end{equation}
Note that we have dropped a constant term $E_J$ in $U_J$. For later convenience we also define the Josephson inductance as
\begin{equation}\label{eq:jos_ind}
L_J = \frac{\Phi_0^2}{4 \pi^2 E_J}.
\end{equation}
If we Taylor expand the energy $U_J$ around $\Phi_{\mathfrak{b}}(t)=0$ up to second order and neglect a constant term we get $U_J=  \frac{\Phi_{\mathfrak{b}}^2}{2L_J}$, motivating the definition of $L_J$.

We can consider what happens when we replace the linear inductor by a Josephson junction, see Fig.~\ref{fig:LCb}, obtaining a so-called Cooper pair box (CPB). As in Section~\ref{subsec:lc} we can identify the branch flux $\Phi_{\mathfrak{b}}$ with the node flux $\Phi$ setting the ground node flux to zero, and write the Lagrangian of the CPB as
\begin{equation}
\lagrangian=\frac{1}{2} C \dot{\Phi}^2+E_J \cos\left(\frac{2\pi }{\Phi_0} \Phi\right),
\label{eq:Ltrans}
\end{equation}
with equation of motion (Euler-Lagrange equation)
\begin{equation}
C \ddot{\Phi}+I_c \sin\left(\frac{2\pi }{\Phi_0} \Phi \right)=0,
\end{equation}
which again corresponds to Kirchhoff's current law.

The system has an immediate mechanical analogy as a pendulum in a gravitational field. In fact, we can introduce a phase variable, i.e.,~$2\pi \Phi/\Phi_0=\varphi+2\pi k$ with phase $\varphi \in [0,2\pi)$ and integer `winding number' $k$. This system then is identical to a pendulum swinging in a gravitational field, where $\varphi$ is the angle between the pendulum and the $z$-axis, letting $F=-\frac{g}{l} \sin(\varphi)$ be the gravitational force with $l$ the length of the pendulum. The equation of motion of such a pendulum is
\begin{equation}
    \ddot{\varphi}=-\frac{g}{l}\sin \varphi.
\end{equation}
One can thus identify $\frac{2 \pi I_c}{\Phi_0 C}$ with $g/l$. The integer $k$ then keeps track of how many times the pendulum has swung full circle. 

To quantize the system in Eq.~\eqref{eq:Ltrans}, similar as for the LC oscillator, we can define the conjugate variable $Q =\partial \lagrangian/\partial \dot{\Phi}=C \dot{\Phi}$ and promote variables to operators satisfying the commutation relation $[\hat{\Phi}, \hat{Q}] = i \hbar \mathds{1}$. The quantum Hamiltonian then reads
\begin{equation}
H = \frac{\hat{Q}^2}{2 C} - E_J \cos \biggl(\frac{2 \pi}{\Phi_0} \hat{\Phi} \biggr).
\end{equation}

Introducing the rescaled flux $\hat{\phi}$ and the rescaled charge $\hat{q}$ operators, as in Section~\ref{subsec:lc}, the Hamiltonian can be written as
\begin{equation}\label{eq:hcpb}
H = 4 E_C \hat{q}^2 - E_J \cos \hat{\phi},
\end{equation}
where we have defined the charging energy $E_C$ as in Eq.~\eqref{eq:char_en}. The spectrum of the CPB Hamiltonian and its interpretation are treated in detail in Section~\ref{sec:cpbspectrum}. 

In Fig.~\ref{fig:LCb}, the Josephson junction is shown in parallel with a capacitance and the cross representing the junction only accounts for the Josephson potential. However, any Josephson junction always comes with its own capacitance $C_J$, independently of there being an additional shunting capacitance in parallel. In the literature, to depict this situation one often finds the Josephson junction represented as in Fig.~\ref{fig:jjcj_symbol}. One could thus have a shunting capacitor $C_{s}$ in parallel to increase the effective capacitance $C = C_J +C_s$. While from a circuit theory point of view nothing changes, these considerations matter in the design of a superconducting qubit, especially for the case of the transmon qubit: in order to enter the transmon regime characterized by a relatively large ratio $E_J/E_C$, the intrinsic capacitance $C_J$ is not sufficient and an additional capacitive shunt is necessary.  

\begin{figure}[htbp]
\centering
\includegraphics[height=4cm]{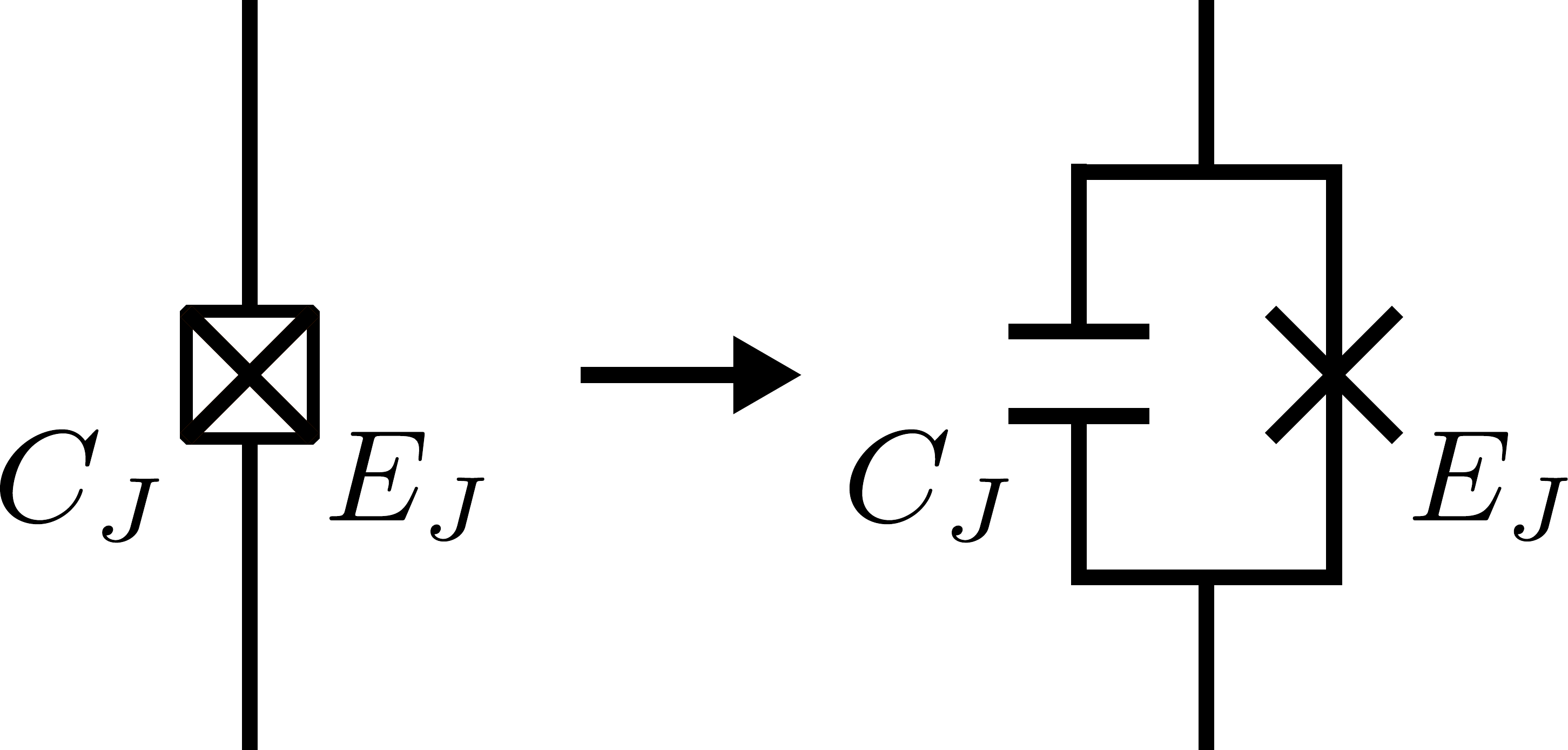}
\caption{Compact symbol representing a Josephson junction with Josephson energy $E_J$ with its intrinsic capacitance $C_J$.}
\label{fig:jjcj_symbol}
\end{figure}

    In the examples of the LC oscillator and the CPB, we have seen that we have identified the energy associated with the capacitive elements as the ``kinetic" energy, while the inductive contributions are seen as the ``potential" energy in the mechanical analogy. However, we point out that there is no a priori reason to make this identification. For example, in Ref.~\cite{UH:dual} a (dual) quantization method is described which takes classical charge variables $Q_{\mathfrak{b}}$ and their time-derivatives $\dot{Q}_{\mathfrak{b}}$ as dynamical variables in the Lagrangian in order to deal with a phase-slip junctions which have a non-convex energy contribution $\sim \cos(\pi Q_{\mathfrak{b}}(t)/e)$. 
Since the non-convex term then appears in the potential part of the Lagrangian, one can follow the canonical quantization method and associate a conjugate variable with $Q_{\mathfrak{b}}$. This however requires that the circuit contains no elements which are nonlinear in $\ib=\dot{Q}_{\mathfrak{b}}$, hence no Josephson junctions. We refer the reader to Refs.~\cite{osborne2023, parra2023} for a symplectic approach to quantize circuits with nonlinear elements in both flux and charge degrees of freedom. \\

\noindent\fbox{
    \parbox{\textwidth}{
        \centerline{\bf Size matters} 
        As you will see in some of the figures of devices in this book, cf.~Fig.~\ref{fig:riste}, superconducting qubits such as the transmon qubits, are not small, but have a linear footprint of hundreds of $\mu$meters. Compare this to modern CMOS chips with hundreds of millions of transistors per square mm! If one wants to make large chips with millions of qubits or more \cite{ibm:arch}, size becomes an issue. Size constraints on superconducting qubits arise from structures having to be resonant at microwave frequencies. To get a sense of scale, the wavelength corresponding to a frequency $f=10$GHz with speed of light $\mathfrak{c}$ is $\lambda=\frac{\mathfrak{c}}{f}\approx 3$ cm.}}

\section{Mutual inductances and the ideal transformer}
\label{sec:mi}

An electrical circuit can also include a mutual inductance between two inductive branches as in Fig.~\ref{fig:mutual_inductance}. The origin of such mutual inductance is simply the fact that two current-carrying wires can attract or repel each other by inducing a local magnetic field. We refer the reader to Ref.~\cite{alexanderSadiku}, Chapter~$13$ for a basic introduction to magnetically coupled circuits. Here we focus on how to treat them in the Lagrangian formalism. If have two branches $\mathfrak{b}_1$ and $\mathfrak{b}_2$, one has

\begin{eqnarray}
\left( \begin{array}{c}\Phi_{\mathfrak{b}_1}\\ \Phi_{\mathfrak{b}_2}\end{array}\right) & = & 
\left(\begin{array}{cc} L_{1} & M \\ 
M & L_2 \end{array} \right)
\left( \begin{array}{c}  \ibm{1} \\ \ibm{2} \end{array}\right),
\label{eq:mutual-2}
\end{eqnarray}
where $L_{1, 2} > 0$ are the self-inductances in Eq.~\eqref{eq:ind} and $M$ is the mutual inductance. The fact that the inductance matrix of the mutual inductor is symmetric is a consequence of the reciprocity of the circuit (see e.g. Ref.~\cite{Griffiths} for concrete expressions for $M$ for current-carrying wires).  We have $M=k\sqrt{L_{1} L_{2}}$ with coupling coefficient $0 \le k < 1$. 
This implies that the determinant of the $2 \times 2$ times matrix in Eq.~\eqref{eq:mutual-2} is $L_{1} L_{2}(1-k) > 0$. Since the trace of the $2 \times 2$ matrix, i.e., ~$L_{1}+L_{2}$, is also positive, this implies that both eigenvalues are positive and the matrix is invertible. The limit $k \rightarrow 1$ is referred to as the \emph{perfect coupling} limit. 

 If we have many inductive branches and several such pairwise interactions, we can thus write
\begin{equation}
\Phibvec=\mat{M} \ibvec,
\label{eq:mutual}
\end{equation}
where $\Phibvec$ is a column vector with all the branch fluxes $\Phi_{\mathfrak{b}}$, and the matrix $\mat{M}$ is a sum of pairwise interactions, each represented as a $2 \times 2$ submatrix, between inductive branches. We note that $\mat{M}$ is a symmetric, positive-definite matrix with positive eigenvalues as it is the sum over $2 \times 2$ submatrices with this property.

As energy contribution to the Lagrangian one obtains
\begin{equation}
U_L=\int_{-\infty}^t dt' \ibvec(t') \cdot \vbvec(t') =\int_{-\infty}^t dt' (\mat{M}^{-1} \Phibvec(t'))^T \frac{d \Phibvec}{dt'}=\frac{1}{2} \Phibvec^T(t) \mat{M}^{-1} \Phibvec(t).
\label{eq:mutual_l}
\end{equation}
We see that the invertibility of the matrix $\mat{M}$, which we have argued, is needed to obtain this expression. In particular, using the inverse of the matrix in Eq.~\eqref{eq:mutual-2}, one gets, for the circuit in Fig.~\ref{fig:mutual_inductance},
\begin{equation}
    U_L=\frac{1}{2(L_1L_2-M^2)}\begin{pmatrix}\Phi_{\mathfrak{b}_1} & \Phi_{\mathfrak{b}_2} \end{pmatrix}
    \left(\begin{array}{rr} L_{2} & -M \\ 
-M & L_1 \end{array}\right)
    \begin{pmatrix}
    \Phi_{\mathfrak{b}_1} \\ \Phi_{\mathfrak{b}_2} \end{pmatrix}.
    \label{eq:ul-ind}
\end{equation}
\\

\begin{figure}
\centering
\begin{subfigure}[t]{0.45 \textwidth}
\centering
\includegraphics[height=4 cm]{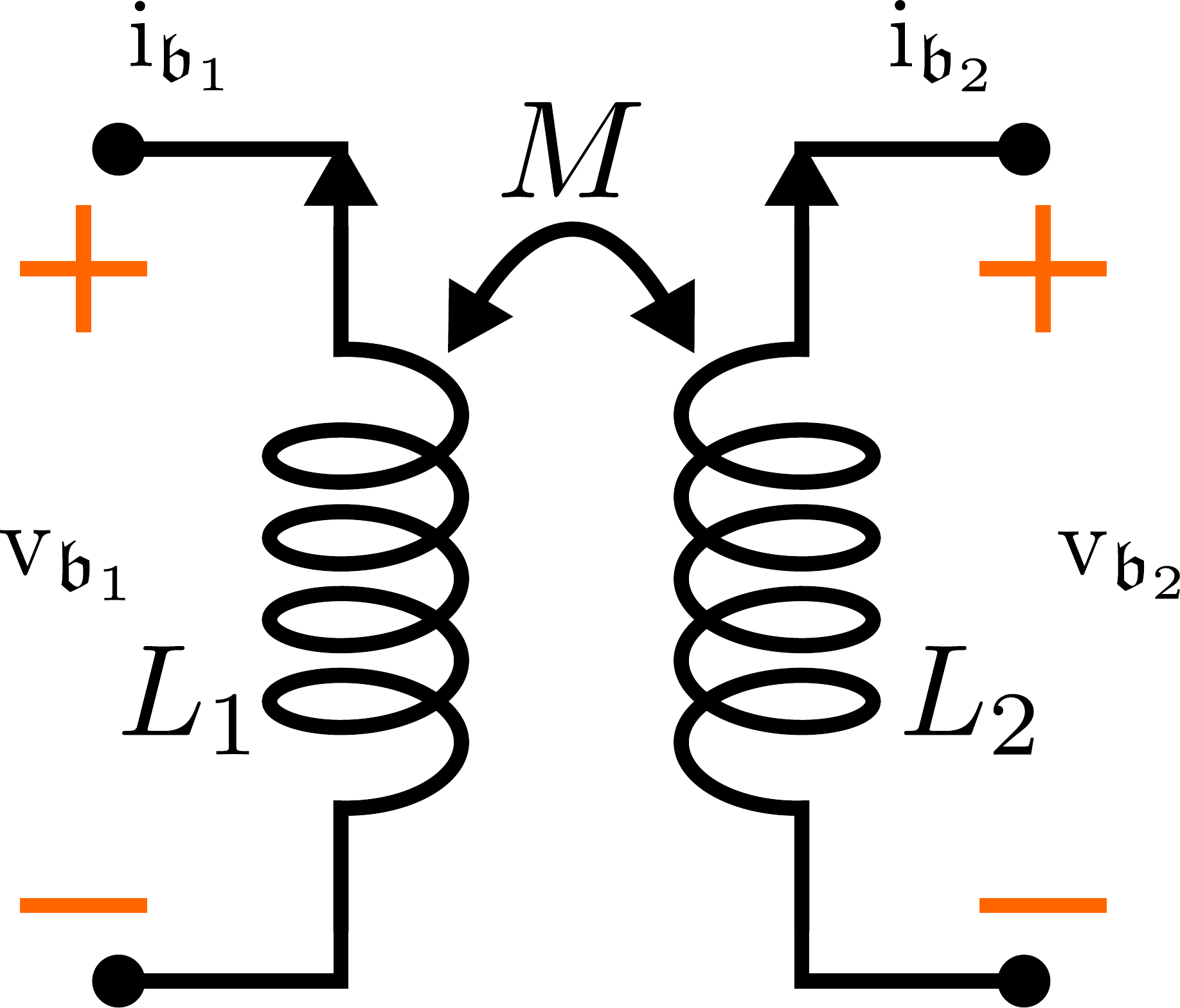}
\subcaption{}
\label{fig:mutual_inductance}
\end{subfigure}
\begin{subfigure}[t]{0.45 \textwidth}
\centering
\includegraphics[height=4 cm]{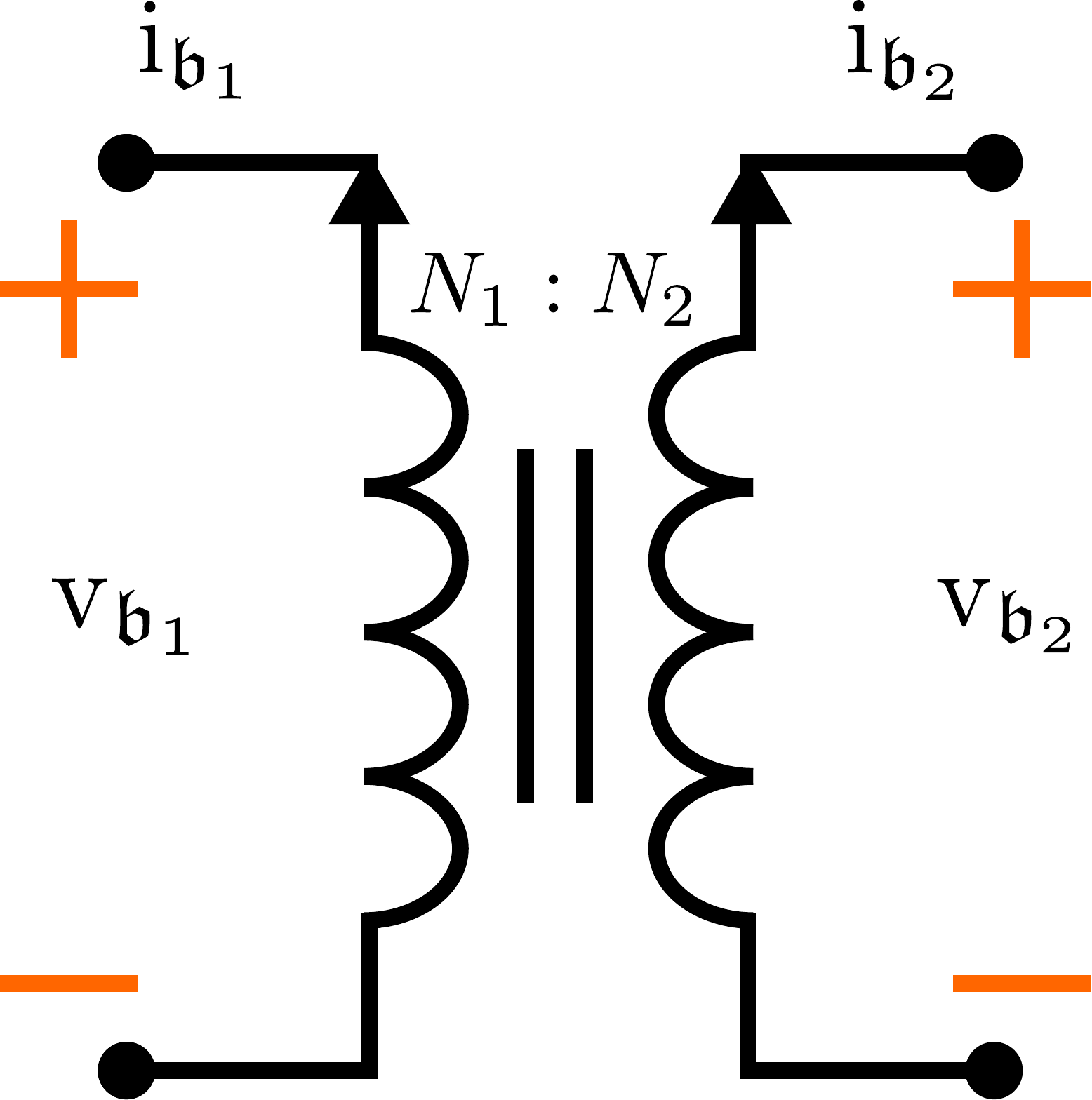}
\subcaption{}
\label{fig:ideal_transf}
\end{subfigure}
\caption{(a) Magnetically coupled inductors with self-inductances $L_{1, 2}$ and a mutual inductance $M$. (b) Circuit symbol of an ideal transformer with turns ratio $t_r = N_2/N_1$.}
\label{fig:mag_coupled_circuits}
\end{figure}

\par 

\subsection{Transformer}
\label{subsec:ideal-trafo}

We now introduce another circuit element, the ideal transformer, which is closely related to a mutual inductance between two inductive branches. Its circuit symbol is shown in Fig.~\ref{fig:ideal_transf}. The ideal transformer is a lossless element that stores no energy. This means that an instantaneous power balance between the two inductive branches must hold, which mathematically translates into the condition
\begin{equation}
\vv_{\mathfrak{b}_1}\ii_{\mathfrak{b}_1} =- \vv_{\mathfrak{b}_2} \ii_{\mathfrak{b}_2},
\label{eq:transform}
\end{equation}
that holds at any fixed time. 

The fact that there is no energy stored in the ideal transformer tells us that we should not expect any energy term associated with it in the Lagrangian formalism for electrical circuits. Instead, the ideal transformer imposes a simple algebraic constraint between the voltages across the terminals, and accordingly, on the corresponding fluxes, that is
\begin{equation}
\vv_{\mathfrak{b}_2} = t_r\vv_{\mathfrak{b_1}} \implies \Phi_{\mathfrak{b}_2} = t_r \Phi_{\mathfrak{b}_1},
\label{eq:id_transf}
\end{equation}
where $t_r \in \mathbb{R}$ is a constant called the turns ratio of the ideal transformer. An ideal transformer can be viewed as a mutual inductance in the limit $L_{1, 2} \rightarrow \infty$ and perfect coupling $k \rightarrow 1$, see Exercise \ref{exc:transf}. If the two inductors are made out of simple coils with number of turns $N_1$ and $N_2$, the parameter $t_r$ in the ideal transformer limit would be $t_r = N_2/N_1$, which justifies the name turns ratio. 

\begin{Exercise}[title={Ideal transformer},label=exc:transf]
Argue that the ideal transformer is indeed the limit $k \rightarrow 1$ and $L_i \rightarrow \infty$ of a mutual inductance circuit as in Fig.~\ref{fig:mutual_inductance} by considering its potential energy in Eq.~\eqref{eq:mutual_l}, in this limit.
\end{Exercise}

\begin{Answer}[ref={exc:transf}]
The eigenvalues of the symmetric inductance matrix $\mat{M}$ are 
\[\lambda_1=\frac{1}{2}(L_1+L_2-\sqrt{(L_1+L_2)^2+4(k^2-1)L_1L_2}) \approx 0,
\]
in the limit $k\rightarrow 1$ with eigenvector $\Phi_1=-\sqrt{\frac{L_2}{L_1}}\Phi_{\mathfrak{b}_1}+\Phi_{\mathfrak{b}_2}$, and $\lambda_2 \rightarrow L_1+L_2$ with eigenvector $\Phi_2=\sqrt{\frac{L_1}{L_2}}\Phi_{\mathfrak{b}_1}+\Phi_{\mathfrak{b}_2}$. So the potential in Eq.~\eqref{eq:mutual_l} equals $U_L=\frac{1}{2}\left(\frac{1}{\lambda_1} \Phi_1^2+\frac{1}{\lambda_2}\Phi_2^2\right)$ and hence there is one (infinitely) steep direction constraining its eigenvector $\Phi_1=0$, which leads to the condition $\Phi_{\mathfrak{b}_2}=\sqrt{\frac{L_2}{L_1}}\Phi_{\mathfrak{b}_1}$. Since $L_i \sim N_i^2$ for an inductor with $N_i$ coil turns, we find Eq.~\eqref{eq:id_transf} with $t_r=N_2/N_1$. Note that since the ideal transformer otherwise should contribute no potential energy, we require that $\lambda_2\rightarrow \infty$ which implies that $L_1 \rightarrow \infty$ and $L_2 \rightarrow \infty$ (but $L_2/L_1$ finite).
\end{Answer}

From Eq.~\eqref{eq:id_transf} we see that an ideal transformer simply imposes a constraint between the branch fluxes, thus reducing the number of degrees of freedom by one. While the ideal transformer is a mathematical abstraction, it is useful as an ideal limit of a real transformers in the context of circuit synthesis. In fact, Cauer's construction for multi-port networks, which we will discuss in Chapter~\ref{chap:ln} and Appendix~\ref{app:norm_mode}, can be interpreted as an ideal circuit of harmonic oscillators coupled to ports via ideal transformers. 

\section{Conservation laws in electrical circuit graphs}
\label{sec:dep}

Electrical circuits generally contain loops and the branch fluxes around a loop cannot be independent variables, as the voltage drop around a loop must be zero. Thus, the branch fluxes and their derivatives, which are useful to express the energy content in branches, are not always {\em independent} variables, while the canonical quantization method in Appendix~\ref{app:cc} requires the variables to be independent. In addition, in some loops one may apply some external magnetic flux (essentially by having another current-carrying circuit nearby and coupling to the loop via a mutual inductance). 
For example, one can have two parallel Josephson junctions in a loop ---a SQUID--- and some flux is threaded through the loop, see Fig.~\ref{fig:flux-JJ}.

In addition, an electrical circuit can contain active elements such voltage and current sources and we would also like to understand how to include these in the Lagrangian. 

\begin{figure}
\centering
\begin{subfigure}[t]{0.45 \textwidth}
\centering
\includegraphics[height=4 cm]{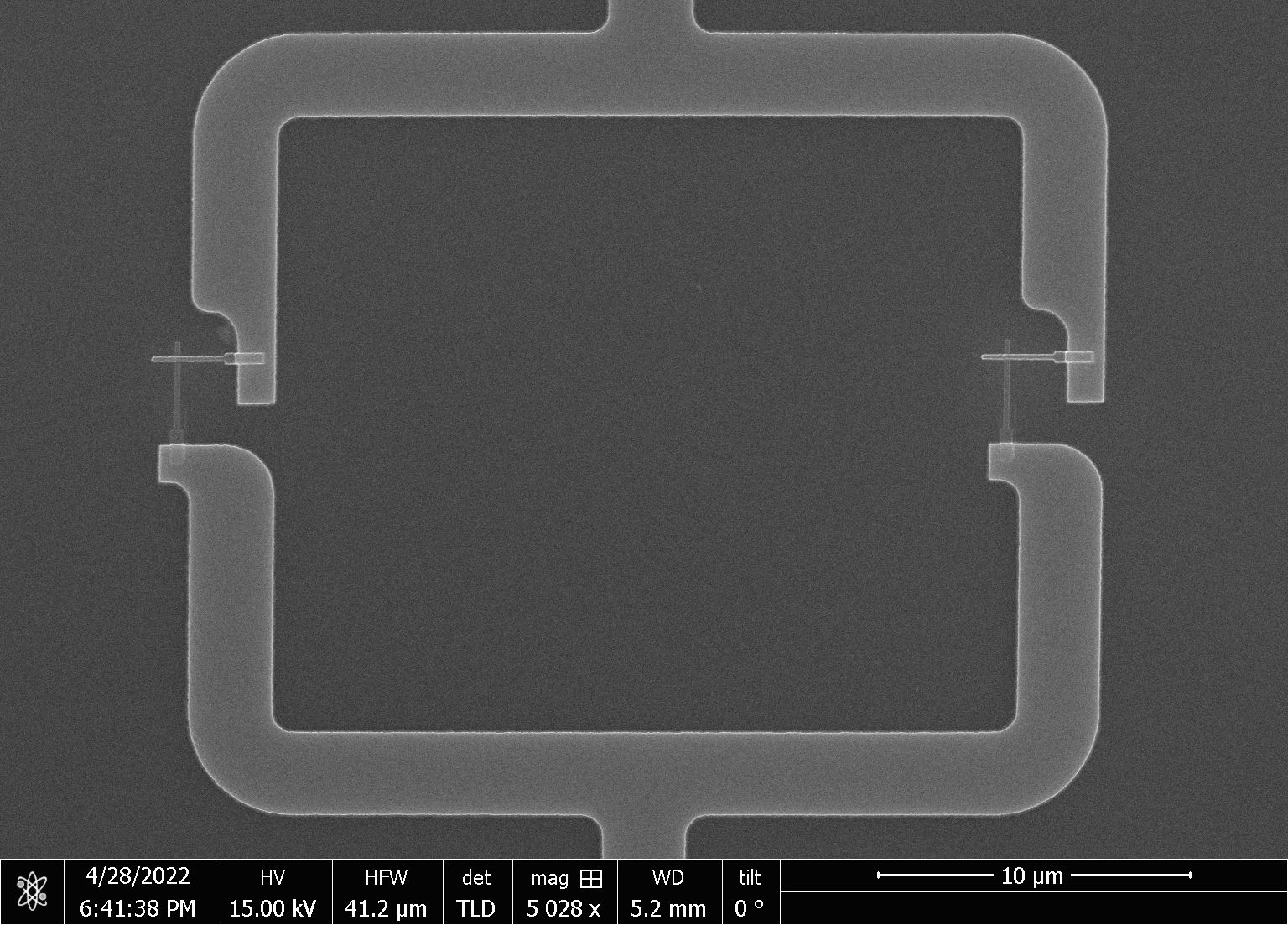}
\subcaption{}
\label{fig:squid}
\end{subfigure}
\begin{subfigure}[t]{0.45 \textwidth}
\centering
\includegraphics[height=4 cm]{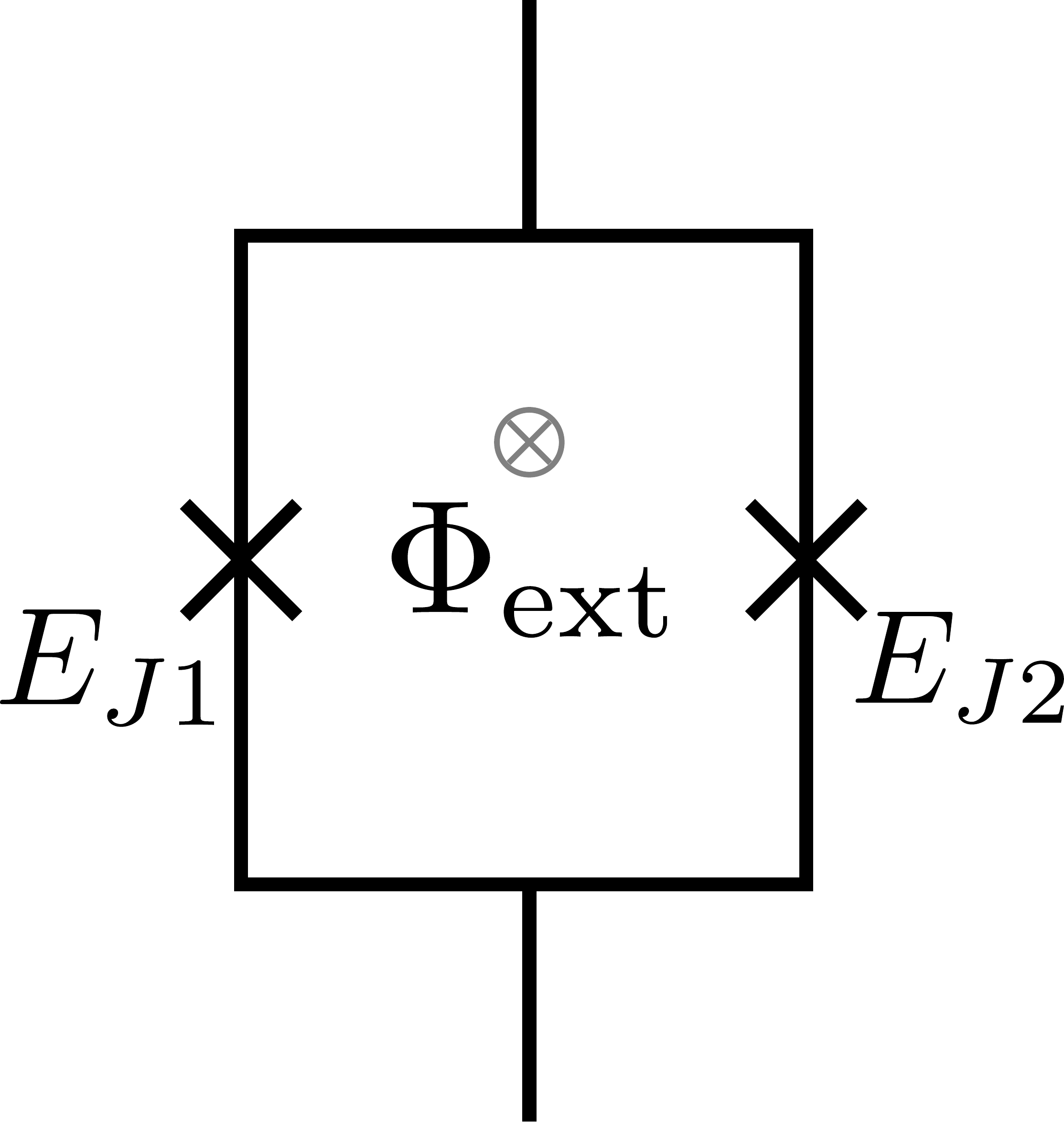}
\subcaption{}
\label{fig:fluxSQUID}
\end{subfigure}
\caption{(a) A SQUID device from Ref.~\cite{nand} using the Josephson junctions depicted in Fig.~\ref{fig:JJ}. (b) Corresponding electric circuit where some external flux could be threading the loop with the two Josephson junctions. The magnetic flux equals $\Phi_{\rm ext}$ {\em into} the plane of the page.}
\label{fig:flux-JJ}
\end{figure}

In this section, we show how one can identify a set of independent node-flux variables in an electrical circuit, using Kirchhoff's voltage conservation law. We then examine how we systematically represent current and voltage sources in the circuit Lagrangian and Hamiltonian in Section~\ref{sec:CV-source}. We discuss how to include external fluxes in Section~\ref{sec:ex-flux}.

Consider an electrical circuit, which can be viewed as an oriented graph $G$ with, say, $N+1$ nodes and $M$ oriented branches. Some branches may in addition have mutual-inductive or other couplings, but we assume that the graph is connected. One arbitrary node will be designated as ground node $n=0$, at which the voltage $\vv_0(t)=0$ at all times. Hence there are $N$ `live' nodes. The choice of which node is set to be the ground node is not relevant for the dynamics, as it just expresses that only voltage (and flux) differences are physically meaningful. 

As mentioned before, the orientation of each branch is in principle arbitrary, but we will fix it in a convenient way in a moment. One can define an $N \times M$ incidence matrix $\mat{A}$ with entries $a_{ij}$ given by:
\begin{eqnarray*}
a_{ij}& =& +1 \mbox{ if node $i$ is incident to branch $j$ with orientation pointing away from $i$},\\
 a_{ij}& =& -1 \mbox{ if node $i$ is incident to branch $j$ with orientation pointing towards $i$},\\
a_{ij}& =& 0  \mbox{ otherwise}.
\end{eqnarray*}

Kirchhoff's current law says that at every node the sum of the ingoing and outgoing currents is zero. Let the column vector $\vect{\ii}_{\mathfrak{b}}=(\ibm{1}, \ldots, \ibm{M})^T$ (where $T$ stands for transpose). Then Kirchhoff's law states that
\begin{equation}
\mat{A}\vect{\ii}_{\mathfrak{b}}=0\;\;(\mbox{Kirchhoff's current law})
\label{eq:KCL}
\end{equation}

The second condition on any electrical circuit is that the voltage drop is zero around any closed loop. In order to express this condition, we want to enumerate the number of independent loops in the electrical circuit graph, so that if the voltage drop is zero for any of these independent loops, it is zero for any loop. The enumeration of these independent, so-called fundamental loops, uses some basic facts of graph theory as follows. First, one chooses a spanning tree $T$ for the graph $G$; its defining feature is that, starting from a reference node ---the ground node--- it reaches every node and has no loops, see an example in Fig.~\ref{fig:tree}. Note that the choice of a spanning tree, as well as that of the reference node, is not unique.

Given such a tree $T$, any branch that is not included in the tree ---such a branch is called a chord---  will create a unique loop. The number of branches in the spanning tree is $N$ since it has to connect all nodes. This means that the number of branches which are not in the tree, i.e.,~the number of chords, is $M-N$, and this coincides with the number of fundamental loops.

Given our choice of ground node and spanning tree, we give the branches in the tree an orientation, that is, we simply take the orientations to point away from the ground node, see Fig.~\ref{fig:tree}. We leave the orientation of the chords free for the moment.

\begin{figure}[htbp]
\centering
	\includegraphics[height=5cm]{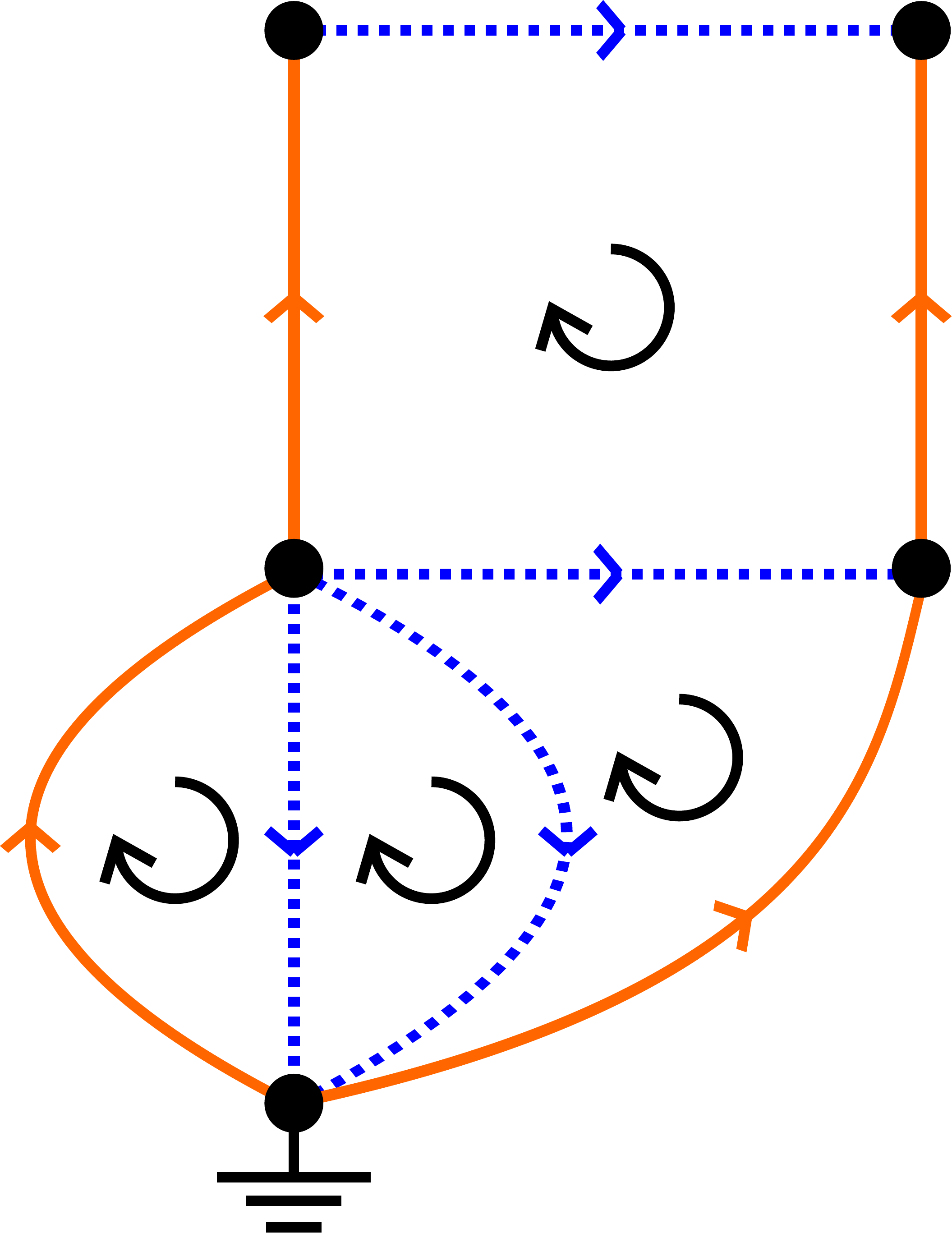}
	\caption{A graph $G$ with $N+1=5$ nodes, of which one is the designated ground node and $M=8$ branches, so that there are $M-N=4$ fundamental loops. The branches in the chosen spanning tree are in orange. Each fundamental loop is associated with a branch (a chord, in dotted blue) which is not in the tree. For convenience, we choose an orientation of each branch in the tree such that it is parallel to the path from the designated ground node to all other nodes.  We can choose an arbitrary (but consistently used) orientation of the chords, so that it aligns with the (arbitrary) orientation of the loops that they close. For a planar circuit graph, it is natural and simplest to take the orientation of the loops to be the same for every loop and we always choose a clockwise orientation as shown in the Figure.}
	\label{fig:tree}
\end{figure}

We also give each fundamental loop an orientation (clockwise or counterclockwise) as shown in the example in Fig.~\ref{fig:tree}. In principle, we can choose this orientation arbitrarily since it just sets a convention, but for simplicity it is best to standardize. This allows us to define a fundamental loop matrix $\mat{B}$. $\mat{B}$ is a $M-N \times N$ matrix with entries defined as
\begin{eqnarray*}
b_{ij}& =& +1 \mbox{ if branch $j$ is part of loop $i$ and has the same orientation},\\
b_{ij} & =&-1\mbox{ if branch $j$ is part of loop $i$ and has the opposite orientation},\\
b_{ij}&=&0 \mbox{ otherwise}.
\end{eqnarray*}
Kirchhoff's voltage law can be then expressed as a constraint on the branch voltage vector $\vbvec=(\vbm{1}, \ldots, \vbm{N})^T$, namely
\begin{equation}
\mat{B} \vbvec=0\;\; (\mbox{Kirchhoff's voltage law}).
\label{eq:KVL}
\end{equation}
One can read this as follows: each fundamental loop, given by a row in $\mat{B}$, imposes an independent constraint. Since $\vbvec(t)=\frac{d\vect{\Phi}_{\mathfrak{b}}}{dt}$ in Eq.~\eqref{eq:def-flux} and $\Phi_{\mathfrak{b}}(t=-\infty)=0$, it follows from this law that
\begin{equation}
\mat{B} \vect{\Phi}_{\mathfrak{b}}=0. 
\label{eq:KVL-flux}
\end{equation}

The law thus shows that the dynamical branch variables $\Phi_{\mathfrak{b}}(t)$ are not independent. In fact, the number of independent variables is simply the number of branches in the tree $T$: there are $M-N$ constraints on $M$ branch variables leaving $N$ independent variables. We also observe, that according to our definition of branch fluxes, the sum of branch fluxes along a loop (taken with appropriate signs) has to be zero. However, in the literature one usually sees the condition written as
\begin{equation}
\mat{B} \vect{\Phi}_{\mathfrak{b}}=\vect{\Phi}_{\mathrm{ext}}, 
\label{eq:KVL-flux_ext}
\end{equation}
with $\vect{\Phi}_{\mathrm{ext}}$ a vector of external fluxes associated with each fundamental loop. In drawings with $\Phi_{\rm ext}$ going into the plane of the loop, one has to choose a clockwise orientation of the loop in order to have Eq.~\eqref{eq:KVL-flux_ext} with $\vect{\Phi}_{\mathrm{ext}}$ (instead of $-\vect{\Phi}_{\mathrm{ext}}$). This is because we have to obey Lenz's law, given our conventions in Fig.~\ref{fig:el_conv} and Eq.~\eqref{eq:def-flux}, i.e., by increasing $\Phi_{\rm ext}$ linearly, one generates a current which generates a flux which opposes the increase.

As we will see in Section~\ref{sec:ex-flux}, we can always interpret Eq.~\eqref{eq:KVL-flux_ext} as Eq.~\eqref{eq:KVL-flux}, where an additional branch flux is forced to take a value equal to the external flux.

Given that there can be fewer independent degrees of freedom than the number of branch fluxes, we can find these independent degrees of freedom as {\em node fluxes} $\Phi_n$ at each node $n=1, \ldots N$ of the graph. We can convert each branch flux to a node flux by considering the unique path in the tree from the node $n$ to the ground node $n=0$. Since the orientations of the branches in this tree path are all the same, we have 
 \begin{equation}
\Phi_n=\sum_{\mathfrak{b} \in {\rm Path}(0 \rightarrow n)} \Phi_{\mathfrak{b}}, 
\end{equation}
so that for two adjacent nodes $n_1$ and $n_2$ with orientation pointing away from $n_1$, connected by a branch $\mathfrak{b}$, we have
\begin{equation}
\Phi_{n_2}-\Phi_{n_1}=\Phi_{\mathfrak{b}},
\label{eq:define-node}
\end{equation}
consistent with Eq.~\eqref{eq:def-orien}.
We can thus take Eq.~\eqref{eq:define-node} as the defining equation for the branch fluxes in the tree in terms of the independent variables.

In the absence of externally-applied magnetic fluxes, the branch-flux variable for each chord is then given as the difference of the two node fluxes of the chord, where an arbitrarily chosen orientation determines the sign. In practice, for convenience, the orientation of these chords can be chosen to align with the orientation of the loop which they close. Hence the sum of branch fluxes around each loop sums up to zero, satisfying Eq.~\eqref{eq:KVL-flux}. 

In the presence of an external flux, we need to satisfy Eq.~\eqref{eq:KVL-flux_ext} and thus all branch fluxes cannot simply be the difference between nodes fluxes, as one has to include the external flux constraint. Thus one can conveniently opt to include this constraint via the chord branch, since the chord branch is in one-to-one correspondence with a loop. This implies that for a chord closing a loop $\ell$ we have
\begin{equation}
    \Phi_{\mathfrak{b}_{\ell}}=\Phi_{n_2}-\Phi_{n_1}+\Phi_{\rm ext,\ell}.
    \label{eq:flux-l1}
\end{equation}

We will continue a discussion of handling external fluxes in Section~\ref{sec:ex-flux} after treating current sources.

We will sometimes omit the orientation of branches or loops in this book, but one can always derive expressions using our standard orientation convention with (1) edges in the spanning tree pointing away from ground, (2) clockwise orientation of loops, (3) chord orientation aligning with the loop which they close, and (4) $\Phi_{\rm ext}$ pointing into the plane.

\begin{Exercise}[title={Cooper pair box grounding variations}, label=exc:grounding]
\begin{figure}[htb]
\centering
\begin{subfigure}[t]{0.5\textwidth}
\centering
\includegraphics[height=5cm]{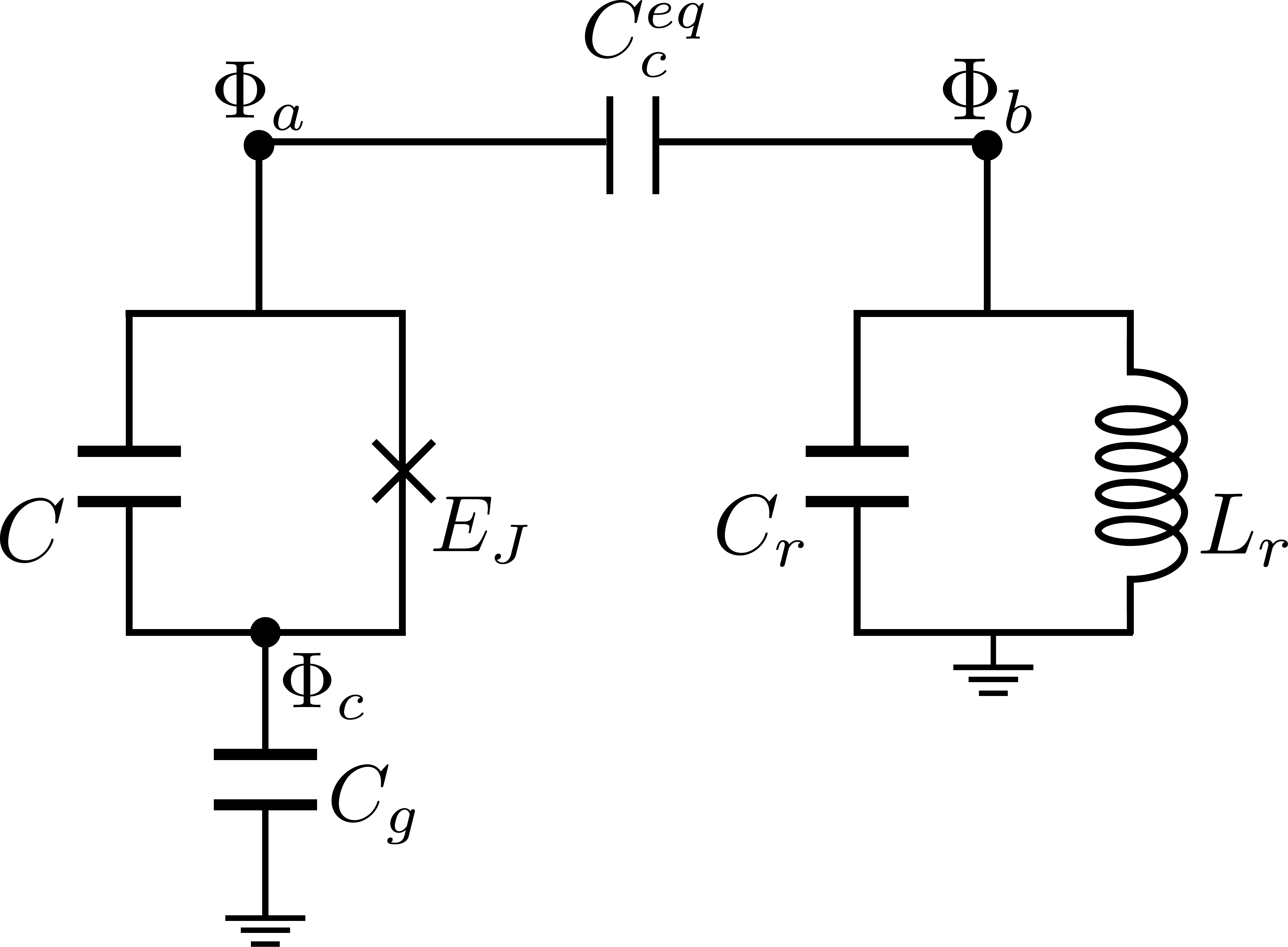}
\subcaption{}
\label{fig::unground}
\end{subfigure}
\begin{subfigure}[t]{0.5\textwidth}
\vspace{0.5cm}
\centering
\includegraphics[height=3.9cm]{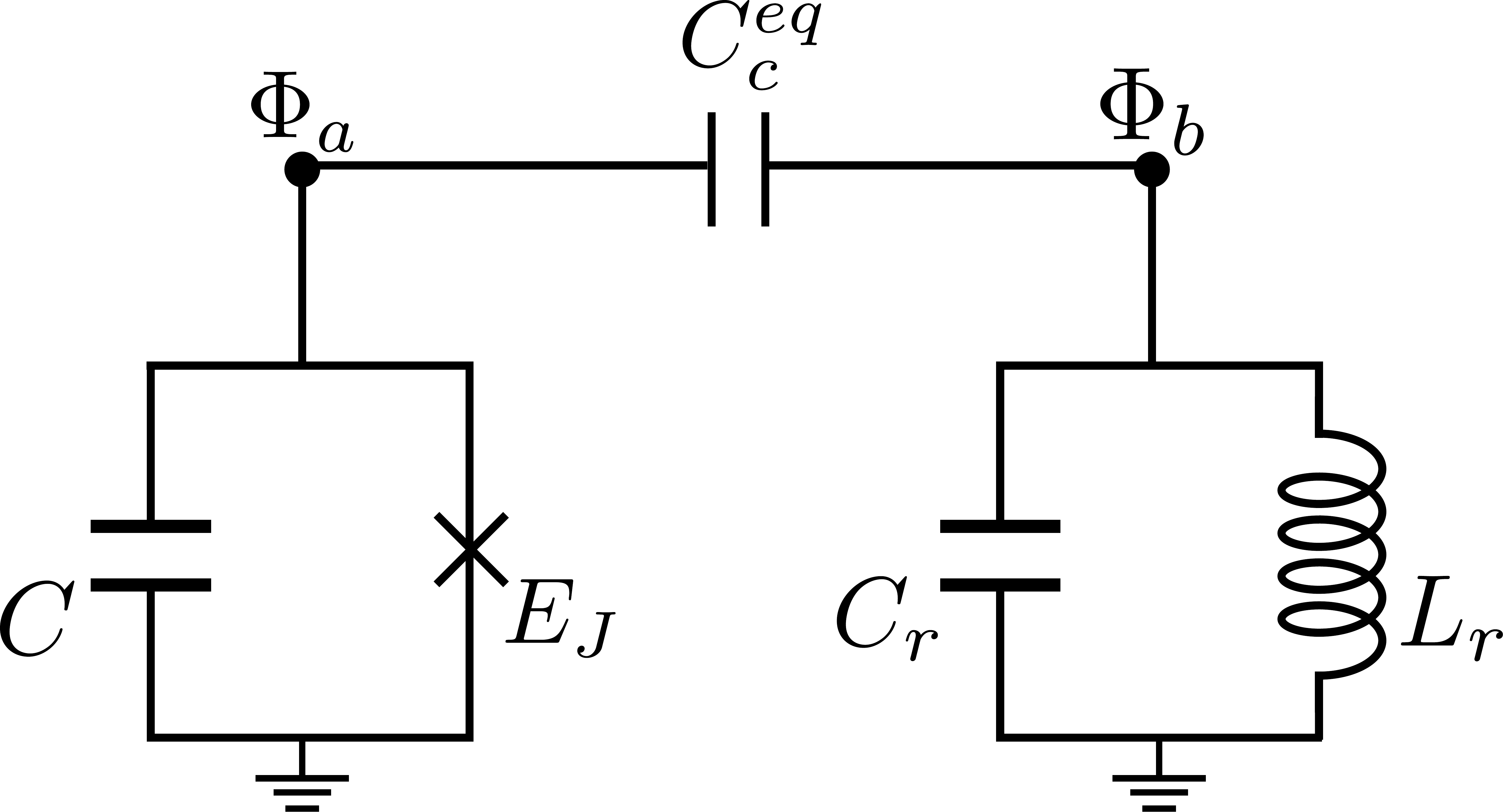}
\subcaption{}
\label{fig::ground}
\end{subfigure}
\caption{(a) Ungrounded transmon capacitively coupled to a LC resonator. (b) Equivalent grounded circuit.}
\end{figure}
Consider the circuit in Fig.~\ref{fig::unground}, showing a typical circuit QED setup with a Cooper pair box (or transmon) qubit capacitively coupled to a LC resonator. We see that the transmon does not share a common ground node with the resonator, but there is a capacitor~$C_g$ to ground. This has been the design choice for instance in the DiCarlo lab at QuTech in Delft, where the transmons are not short-circuited to the ground plane; see for example the two capacitive plates of the transmon in Fig.~8(a) in \cite{versluis}. Other groups such as the Wallraff group at ETH Z\"urich use grounded transmons; see for example Fig.~2(b) in \cite{andersen2020} showing a single (yellow) island (the location of node $\Phi_a$ in Fig.~\ref{fig::ground}). The choice is merely practical, since, as we show in this exercise, the physics is the same. In principle, the top transmon node will also have a capacitance to ground, but here we do not include it for simplicity.

With respect to~\cref{fig::unground}, write down the Lagrangian in terms of the variables $\Phi_1 = \Phi_a - \Phi_c$, $\Phi_2 = \Phi_b$, $\Phi_3 = \Phi_c$, and, using the Lagrangian formalism, show that we can arrive at an effective circuit where the transmon is also grounded, as in Fig.~\ref{fig::ground}. What is the interpretation of the effective coupling capacitance $C_{c}^{eq}$?
\end{Exercise}

\begin{Answer}[ref={exc:grounding}]
The Lagrangian reads
\begin{multline}\label{eq::lunground}
\mathcal{L}(\Phi_1, \Phi_2, \Phi_3; \dot{\Phi}_1, \dot{\Phi}_2, \dot{\Phi}_3 ) = \frac{C_g}{2} \dot{\Phi}_3^2 + \frac{C}{2} \dot{\Phi}_1^2 + \frac{C_r}{2} \dot{\Phi}_2^2 + \frac{C_c}{2} \bigl(\dot{\Phi}_1 + \dot{\Phi}_3- \dot{\Phi}_2 \bigr)^2 
+E_J \cos \biggl[\frac{2 \pi}{\Phi_0} \Phi_1\biggr] - \frac{\Phi_2^2}{2 L_r}.
\end{multline}
The Euler-Lagrange equation associated with $\Phi_3$ yields
\begin{equation}
C_g \ddot{\Phi}_3 + C_c \bigl(\ddot{\Phi}_3 + \ddot{\Phi}_1 - \ddot{\Phi}_2 \bigr) =0 \implies \dot{\Phi}_3 = \frac{C_c}{C_g + C_c} \bigl(\dot{\Phi}_2 - \dot{\Phi}_1 \bigr) + \lambda,
\end{equation}
where $\lambda$ is an arbitrary constant. Substituting into Eq.~\eqref{eq::lunground} we get
\begin{equation}
\label{eq::lground}
\mathcal{L}(\Phi_1, \Phi_2; \dot{\Phi}_1, \dot{\Phi}_2) = \frac{C}{2} \dot{\Phi}_1^2 + \frac{C_r}{2} \dot{\Phi}_2^2 + \frac{C_c^{eq}}{2} \bigl(\dot{\Phi}_1 - \dot{\Phi}_2 \bigr)^2 \\
+E_J \cos \biggl[\frac{2 \pi}{\Phi_0} \Phi_1\biggr] - \frac{\Phi_2^2}{2 L_r} + \frac{(C_c+C_g)}{2}\lambda^2,
\end{equation}
where the equivalent coupling capacitance is given by
\begin{equation}
C_c^{eq} = \frac{C_c C_g}{C_c + C_g},
\end{equation}
and where the constant term~$(C_c+C_g)\lambda^2/2$ can be dropped. Note that $C_c^{eq}$ is simply the equivalent capacitance of a series of $C_g$ and $C_c$. Switching back to the original node variables, we can write
\begin{align}
\mathcal{L}(\Phi_a-\Phi_c, \Phi_b; \dot{\Phi}_a-\dot{\Phi}_c, \dot{\Phi}_b) = 
 \frac{C}{2} (\dot{\Phi}_a-\dot{\Phi}_c)^2 + \frac{C_r}{2} \dot{\Phi}_b^2 + \frac{C_c^{eq}}{2} \bigl(\dot{\Phi}_a-\dot{\Phi}_c - \dot{\Phi}_b \bigr)^2 
+E_J \cos \biggl[\frac{2 \pi}{\Phi_0} (\Phi_a-\Phi_c)\biggr] - \frac{\Phi_b^2}{2 L_r}.
\end{align}
Note that by identifying the difference $\Phi_a-\Phi_c$ with the node $\Phi_a$ in Fig.~\ref{fig::ground}, we get the Lagrangian of the circuit in Fig.~\ref{fig::ground}. 
\end{Answer}

\section{Current and voltage sources}
\label{sec:CV-source}

\begin{figure}
\centering
\begin{subfigure}[t]{0.5\textwidth}
	\includegraphics[height=3.5cm]{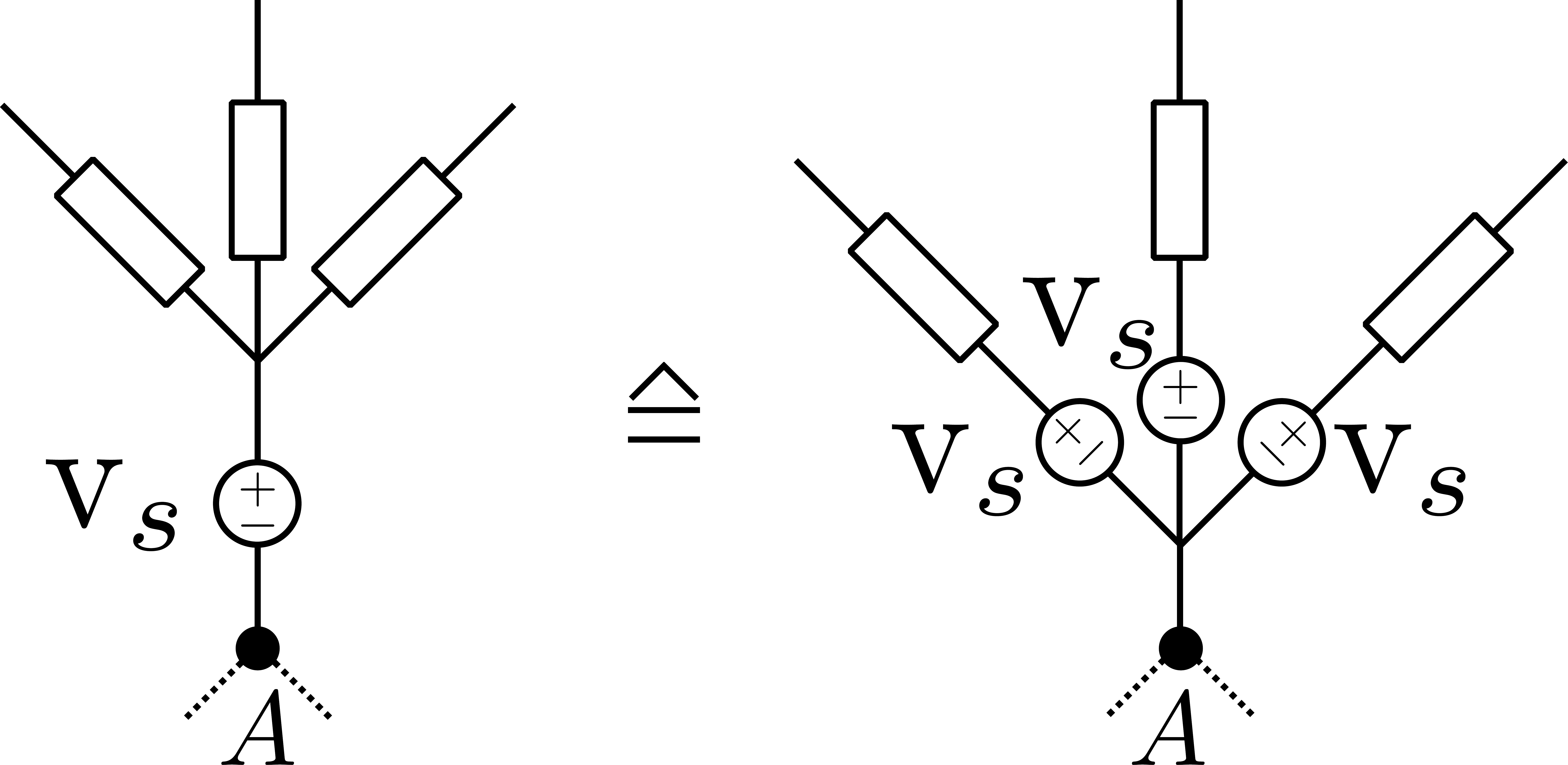}
	\subcaption{}
\end{subfigure}
\begin{subfigure}[t]{0.5\textwidth}
	\includegraphics[height=3.5cm]{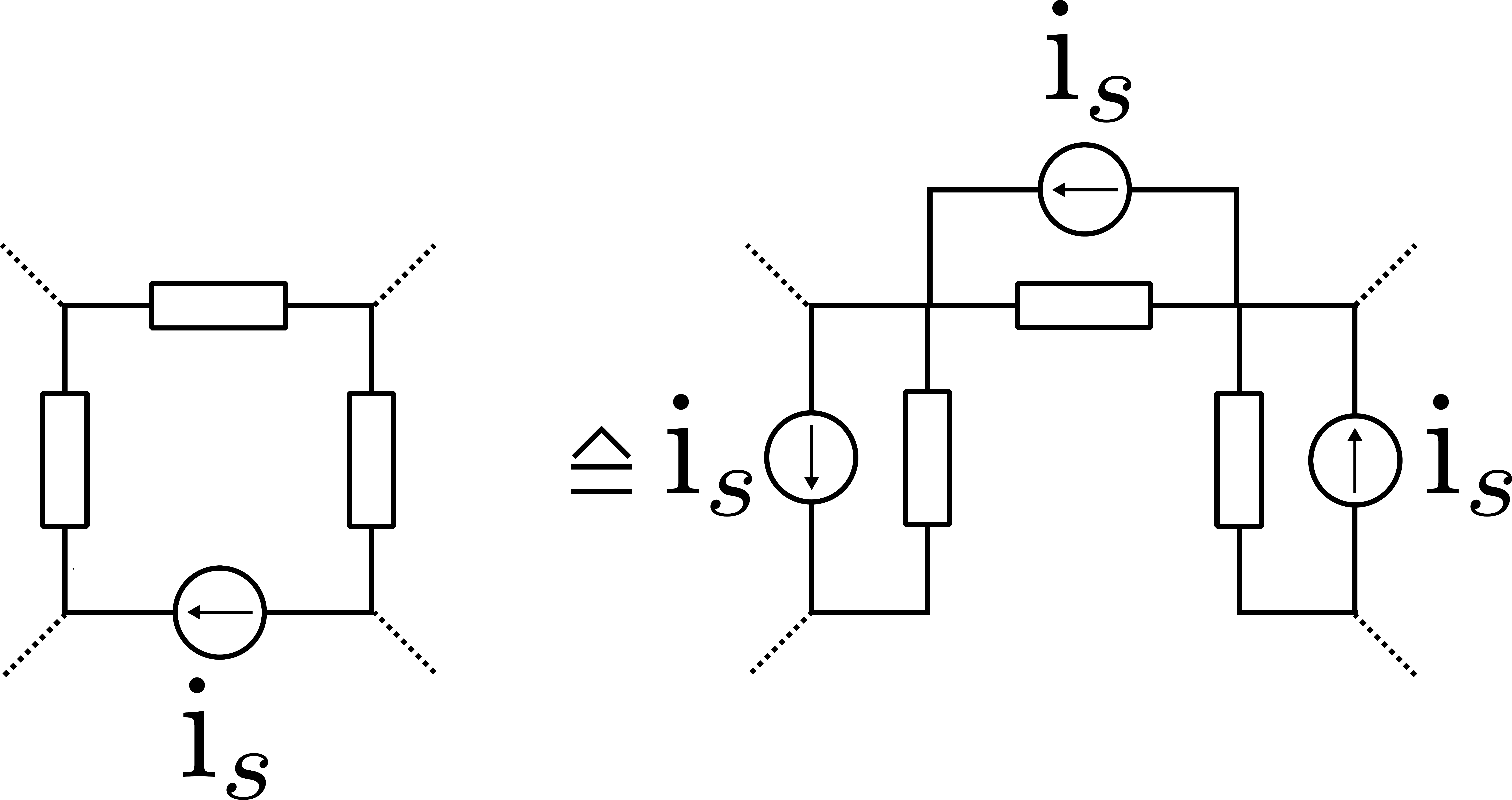}
	\subcaption{}
\end{subfigure}
	\caption{Simple circuit identities for making ideal voltage and current sources part of a branch, leading to a general branch of the form in Fig.~\ref{fig:gen-branch}. Boxes denote elements which can be inductors, capacitors etc.}

	\label{fig:circuit-ID}
\end{figure}

\begin{figure}[htbp]
\centering
	\includegraphics[height=4cm]{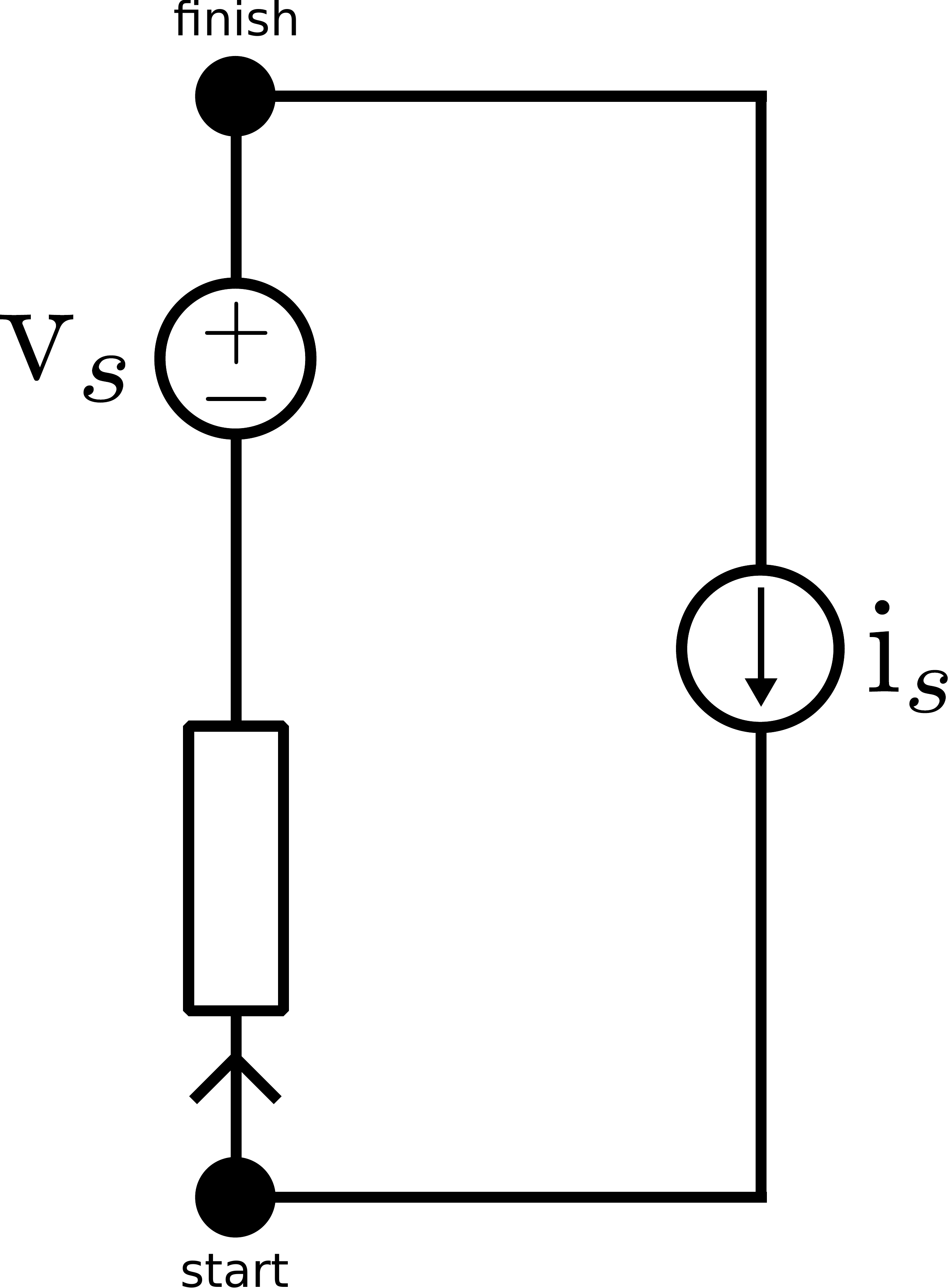}
	\caption{General form of a branch after applying circuit identities in Fig.~\ref{fig:circuit-ID}. The box is an element which can be an inductor, capacitor etc. and has some voltage drop $\vv_{\rm element}$ over it.}
	\label{fig:gen-branch}
\end{figure}

In an electrical circuit, active elements, such as voltage and current sources, can also be present in some of the branches. In this section, we explain how to treat these elements within the Lagrangian formalism. 
First of all, we define an ideal voltage source as a two-terminal element that imposes a possibly time-dependent voltage $\vv_s(t)$ across its terminals. Analogously, an ideal current source is a two-terminal element in which a current $\ii_s(t)$ flows from one terminal to the other. The electrical symbols for voltage and current sources shown for example in Fig.~\ref{fig:circuit-ID} always represent ideal voltage and current sources.

Let us also discuss some other general facts about electrical circuits. Given an electrical circuit, one can always modify it so that any voltage source is in series with another energy-conserving branch element, while each current source is in parallel with each branch element, see Fig.~\ref{fig:gen-branch}. The circuit identities which accomplish this are given in Fig.~\ref{fig:circuit-ID}. The left circuit identity follows simply by adding a countercurrent in parallel on each branch around the loop. This leads to the general form of a branch element of the form shown in Fig.~\ref{fig:gen-branch}, with an ideal current source $\ii_s(t)$ and an ideal voltage source $\vv_s(t)$.

\subsection{Voltage sources}
\label{subsec:capvol}

\begin{figure}[htbp]
\centering
\includegraphics[height=4 cm]{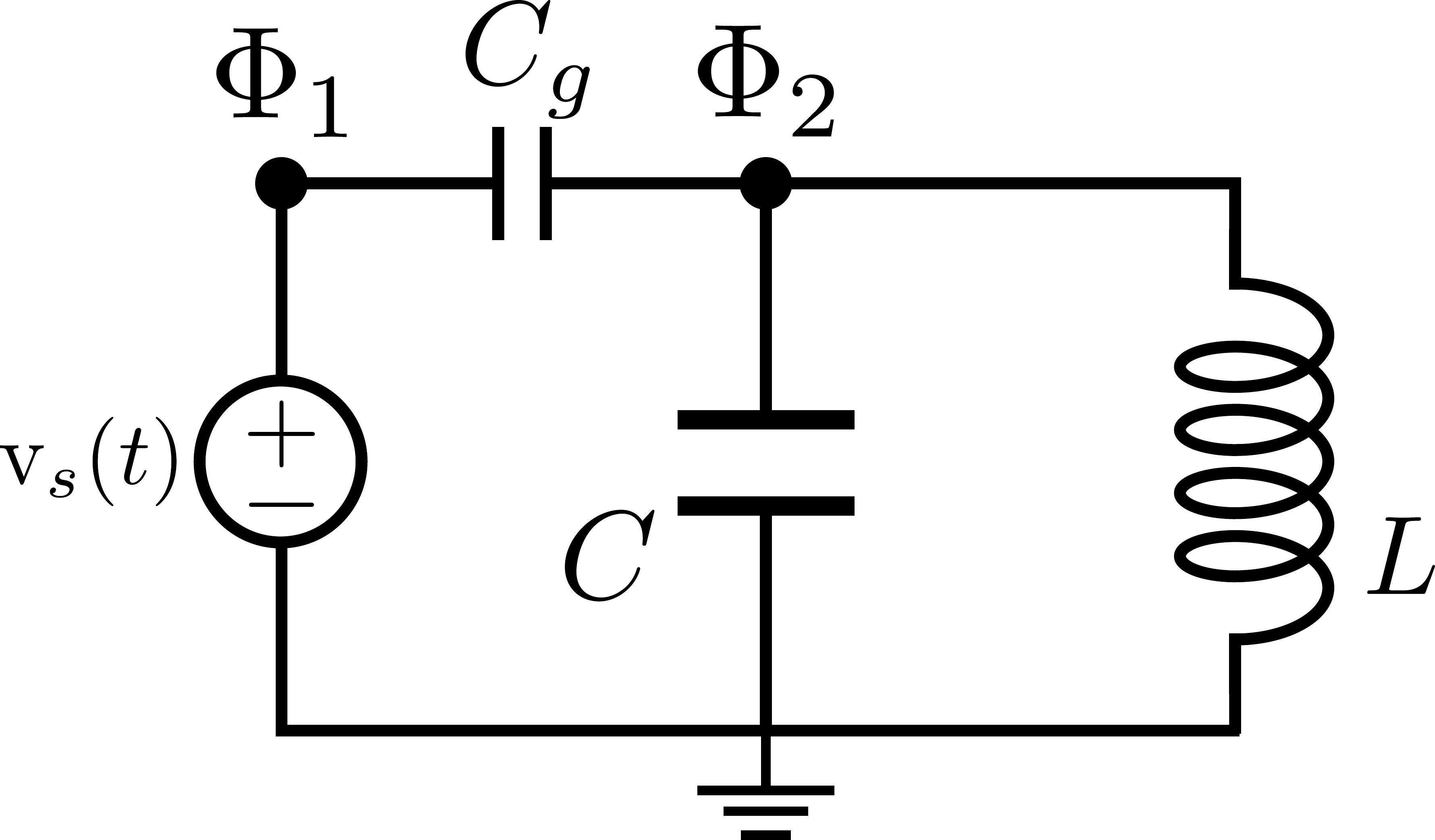}
\caption{Simple circuit with an ideal, time-dependent voltage source capacitively coupled to a LC circuit allowing the harmonic oscillator to be driven and get excited.}
\label{fig:driven_LC}
\end{figure}

Let us start by considering the implications for voltage sources and deal with current sources in Section~\ref{subsec:cs}. 
Considering the general branch in Fig.~\ref{fig:gen-branch}, we can express the voltage drop across such a branch as
\begin{equation}
\vbm{}(t)=\vv_{\rm element}(t)+\vv_s(t),
\label{eq:gen-el}
\end{equation}
where $\vv_{\rm element}(t)$ is the voltage drop across the element in the box. 
In terms of fluxes we have
 \begin{equation}
 \dot{\Phi}_{\mathfrak{b}}=\dot{\Phi}_{\rm element}+\vv_s(t) \implies \quad \Phi_{\mathfrak{b}}(t)=\Phi_{\rm element}(t)+\int_{-\infty}^{t} \vv_s(t') dt'.
\label{eq:branch-rel}
\end{equation}

If the element in Fig.~\ref{fig:gen-branch} is a capacitor $C$, it will add a kinetic energy $\frac{1}{2} C \dot{\Phi}_{\rm element}^2$ to the Lagrangian. We can then express this energy using Eq.~\eqref{eq:branch-rel} in terms of $\dot{\Phi}_{\mathfrak{b}}$ and the externally applied voltage $\vv_s(t)$, which could be time-dependent. If the element in series with the voltage source is inductive, it will add a potential energy $U(\Phi_{\rm element})$ to the Lagrangian and we can express $\Phi_{\rm element}$ in terms of $\Phi_{\mathfrak{b}}$ and the time-integral over the voltage using Eq.~\eqref{eq:branch-rel}. This can make the Lagrangian have an explicit time dependence and hence can lead to an explicitly time-dependent Hamiltonian. Fig.~\ref{fig:driven_LC} shows an example of an LC oscillator capacitively coupled to a time-dependent voltage source.

Note that the branch variable $\Phi_{\mathfrak{b}}$ or $\dot{\Phi}_{\mathfrak{b}}$ can be part of a larger network (say a loop) and hence other terms in the Lagrangian may depend on it. Using Eq.~\eqref{eq:branch-rel} we can reduce this to a dependence on $\Phi_{\rm element}$ and $\vv_s(t)$. In general, a voltage source can also be due to another quantized electrical circuit such as an electrical circuit representing a transmission line or a resonator which we wish to treat in a quantum manner. The point is then that one can replace the classical $\vv_s$ in the Hamiltonian of the circuit by the quantized voltage operator $\hat{\vv}_s$ of the other electrical circuit (with its own quantum degrees of freedom), giving a representation of their coupling.

Let us also consider an alternative approach in which we replace a constant voltage source $\vv_s$ by a capacitor $C_s$ and show that, in an appropriate limit, it leads to an identical result as the arguments above for the small circuit in Fig.~\ref{fig:driven_LC}. With the voltage source as capacitor, the circuit has Lagrangian
\begin{multline}
\lagrangian = \frac{1}{2} \begin{pmatrix}\dot{\Phi}_1 & \dot{\Phi}_2
\end{pmatrix} \underbrace{\begin{pmatrix}C_s+C_g & -C_g \\ -C_g & C+C_g\end{pmatrix} }_{\mat{C}}\begin{pmatrix}\dot{\Phi}_1 \\ \dot{\Phi}_2 \end{pmatrix}-\frac{1}{2L} \Phi_2^2 
=\frac{1}{2}\begin{pmatrix}\dot{\Phi}_1 & \dot{\Phi}_2
\end{pmatrix} \mat{C} \begin{pmatrix}\dot{\Phi}_1 \\ \dot{\Phi}_2 \end{pmatrix}-\frac{1}{2L} \Phi_2^2,
\label{eq:lagrangV}
\end{multline}
where we defined the capacitance matrix $\mat{C}$. Using conjugate variables $Q_1=\partial {\cal L}/\partial \dot{\Phi}_1=C_s \dot{\Phi}_1-C_g(\dot{\Phi}_2-\dot{\Phi}_1)$ and  $Q_2=\partial {\cal L}/\partial \dot{\Phi}_2=C \dot{\Phi}_2+C_g (\dot{\Phi}_2-\dot{\Phi}_1)$ gives the Hamiltonian (see Appendix~\ref{app:cc}):
\begin{equation}
\hamiltonian =\frac{1}{2} \begin{pmatrix}
Q_1 & Q_2 
\end{pmatrix}\mat{C}^{-1} \begin{pmatrix}Q_1 \\ Q_2 \end{pmatrix}+\frac{1}{2L} \Phi_2^2.
\end{equation}
When $C_s \gg C$ and $C_s \gg C_g$, the kinetic energy of this Hamiltonian reads  
\begin{multline}
T=\frac{1}{C C_g+ C_s(C+C_g)}\biggl[\frac{C+C_g}{2} Q_1^2+\frac{C_g+C_s}{2} Q_2^2 + C_g Q_1 Q_2\biggr] 
\approx
\frac{Q_1^2}{2 C_s}+\frac{Q_2^2}{2(C+C_g)} +\frac{C_g Q_1 Q_2}{C_s(C+C_g)}, 
\end{multline}
which can be rewritten, using $\vv_s \approx Q_1/C_s$ in the limit $C_s \gg C_g$, as 
\begin{equation}
T \approx \frac{1}{2(C+C_g)}(Q_2+C_g \vv_s)^2+ \mathrm{const.}
\label{eq:meth1}
\end{equation}

If instead we use Eq.~\eqref{eq:gen-el}, then $\dot{\Phi}_1=\vv_s$, and so the Lagrangian would be
\begin{equation}\label{eq:llcv}
\lagrangian =\frac{C_g}{2}(\dot{\Phi}_2-\vv_s)^2+\frac{C}{2} \dot{\Phi}_2^2-\frac{1}{2L} \Phi_2^2,
\end{equation}
which using $Q_2=(C+C_g)\dot{\Phi}_2-C_g \vv_s$ gives the Hamiltonian as
\begin{equation}
\hamiltonian =Q_2 \dot{\Phi}_2-{\cal L} =\frac{1}{2(C+C_g)}(Q_2+C_g \vv_s)^2 +  \frac{1}{2L} \Phi_2^2.
\label{eq:volt}
\end{equation}
This Hamiltonian has identical kinetic energy as in Eq.~\eqref{eq:meth1} (modulo the harmless constant energy term).

\begin{Exercise}[title={Voltage-driven LC oscillator},label=exc:driven-sol]
For the Hamiltonian in Eq.~\eqref{eq:volt}, what are the classical equations of motion for the conjugate variables $\Phi_2(t)=\Phi(t)$ and $Q_2=Q(t)$? 
Assume a harmonic drive $\vv_s(t)=V_0 \cos (\omega t)$. Solve the equations of motion for $Q(t)$ and $\Phi(t)$ for this drive by introducing the complex variable $a(t)=\frac{\Phi}{\sqrt{2 L\omega_r}}+i\frac{Q}{\sqrt{2C_{\Sigma}\omega_r}}$ (compare Eq.~\eqref{eq:annih}) with $\omega_r=\frac{1}{\sqrt{L C_{\Sigma}}}$ and $C_{\Sigma}=C+C_g$. How can one maximize the effect of the drive?
\end{Exercise}

\begin{Answer}[ref={exc:driven-sol}]
We use the notation $Q_d(t)\coloneqq C_g \vv_s(t)$. Hamilton's equations of motion in Eqs.~\eqref{eq:hameq1} and~\eqref{eq:hameq2} give the coupled equations of motion for $Q(t)$ and $\Phi(t)$. They read
$\dot{\Phi}=\frac{\partial \hamiltonian}{\partial Q}=\frac{Q(t)+Q_d(t)}{C_{\Sigma}}$ and $\dot{Q}=-\frac{\partial \hamiltonian}{\partial \Phi}=-\frac{\Phi}{L}$. 
For the variable $a(t)$, we then obtain the differential equation \[\dot{a}=-i \omega_r a+\sqrt{\frac{\omega_r}{2C_{\Sigma}}}Q_d(t),\]
where $Q_d(t)$ can be expressed in terms of the given drive $v_s(t)$. The solution for this linear, first-order, inhomogeneous equation is
\begin{align}
    a(t)&=a(0) e^{-i \omega_r t}+ \sqrt{\frac{\omega_r}{2C_{\Sigma}}} C_g V_0
        e^{-i \omega_r t}\int_0^t dt' \cos(\omega t') e^{i \omega_r t'}  \notag \\
    &= a(0) e^{-i \omega_r t}+
        \sqrt{\frac{\omega_r}{2C_{\Sigma}}} C_g V_0 \Biggl[ e^{i(\omega-\omega_r)t/2}  \frac{\sin\bigl((\omega+\omega_r)t/2\bigr)}{\omega+\omega_r} +
        e^{-i(\omega+\omega_r)t/2}  \frac{\sin\bigl((\omega-\omega_r)t/2\bigr)}{\omega-\omega_r}
        \Biggr].
\end{align}

Taking the complex conjugate to get the corresponding equation for~$a^*(t)$, we obtain
\begin{align}
    \Phi(t) &= \frac{\sqrt{2L\omega_r}}{2} \bigl(a(t)+a^*(t)\bigr) \notag \\
    &= \Phi(0)\cos(\omega_r t) + \sqrt{\frac{L}{C_\Sigma}}Q(0)\sin(\omega_r t) \nonumber\\
    &\quad +\frac{C_g V_0}{C_\Sigma} \Biggl[ \frac{\cos\bigl((\omega-\omega_r)t/2\bigr)\sin\bigl((\omega+\omega_r)t/2\bigr)}{\omega+\omega_r}
    + \frac{\cos\bigl((\omega+\omega_r)t/2\bigr)\sin\bigl((\omega-\omega_r)t/2\bigr)}{\omega-\omega_r}
    \Biggr],
\end{align}
and
\begin{align}
    Q(t) &= \frac{\sqrt{2C_\Sigma\omega_r}}{2i} \bigl(a(t)-a^*(t)\bigr) \notag \\
    &= -\sqrt{\frac{C_\Sigma}{L}}\Phi(0)\sin(\omega_r t) + Q(0)\cos(\omega_r t)  - 2C_g V_0 \frac{\omega_r^2}{\omega^2-\omega_r^2} 
    \sin\bigl((\omega-\omega_r)t/2\bigr)\sin\bigl((\omega+\omega_r)t/2\bigr).
\end{align}

The strength of the driven solution is linearly increasing with coupling capacitance $C_g$, with voltage drive strength $V_0$, and with how on-resonant the drive is (i.e.,~when $\omega\to\omega_r$). 
\end{Answer}

\subsection{Current sources}
\label{subsec:cs}

\begin{figure}[htbp]
\centering
	\includegraphics[width=8cm]{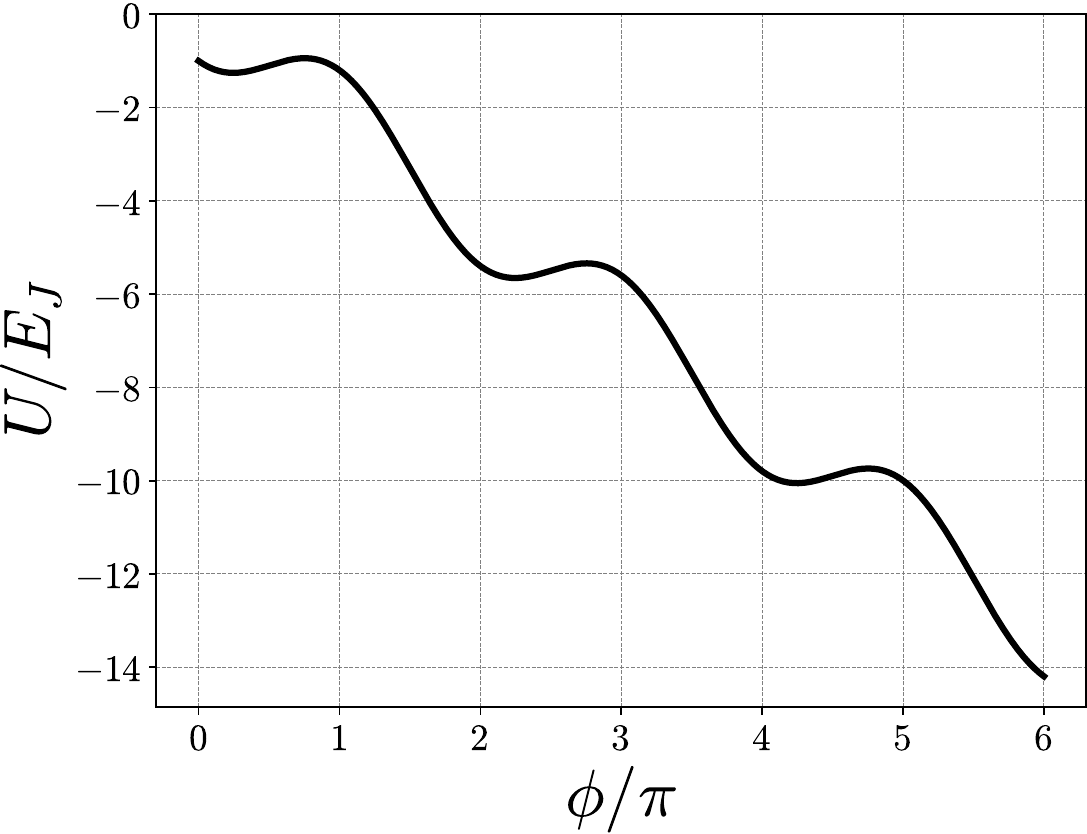}
	\caption{Example of a washboard potential $U(\phi)/E_J = \ii_s\phi/I_c - E_J \cos \phi$ as a function of the reduced (dimensionless) flux $\phi$ for $\ii_s/I_c=-0.7$.}
	\label{fig:washboard_potential}
\end{figure}

For a current source, one does not need to use any circuit identities, but can just treat the current source as a branch in the circuit, in principle. In any case, for the branch in Fig.~\ref{fig:gen-branch} we have the identity $\ibvec(t)=\ii_{s}(t)+\ii_{\rm element}(t)$, where $\ii_{\rm element}(t)$ is the current which goes through the element (the box in the Figure).

If we have a constant current source in parallel with an inductor, we can write the potential energy due to both the current source and the inductor as
\begin{equation}
U=\int_{-\infty}^t dt' \vbm{}(t') \ibm{}(t') =\ii_s \int_{-\infty}^t dt' \vbm{}(t') + U_L=\ii_s \Phi_{\mathfrak{b}}(t)+\frac{1}{2L} \Phi_{\mathfrak{b}}^2(t).
\label{eq:include-i}
\end{equation}
Similarly, a current-biased Josephson junction ---a junction in parallel with a current source--- has potential energy 
\begin{equation}
U=\ii_s \Phi_{\mathfrak{b}}(t)-E_J \cos\left(\frac{2 \pi }{\Phi_0} \Phi_{\mathfrak{b}}(t)\right).
\label{eq:current_source_def_potential}
\end{equation}
This latter potential has the character of a washboard, a linear term superposed with an oscillating cosine as shown in the example in Fig.~\ref{fig:washboard_potential}. One can determine the value of the critical current $I_c$ by varying the strength of the current source $\ii_s$ and observing when the dynamics of the flux variable $\Phi_{\mathfrak{b}}$ is such that it rolls down the washboard instead of being confined in a well (whose depth is set by $E_J$ in Eq.~\eqref{eq:defineEJ} and hence $I_c$) as was first done in \cite{DMC:cb}. The earliest observations of macroscopic quantum behavior in superconducting circuits by Martinis, Devoret and Clarke \cite{MDC:quant} were done for the current-biased junction circuit: the earliest superconducting qubit, the so-called phase qubit, is based on picking the lowest two-levels in one of the wells on the washboard. \\

One should observe that there is an asymmetry in how we deal with a current source versus a voltage source. For the current source we add the extra potential energy due to the source to the Lagrangian.
For the voltage source we only add the kinetic energy stored in the element to the Lagrangian and use $\vv_s$ to obey the constraints of Kirchhoff's voltage law. 

\begin{Exercise}[title={Current and flux-driven LC oscillator}, label=exc:flux]
\begin{figure}[htb]
\centering
\begin{subfigure}[t]{0.4\textwidth}
\centering
\includegraphics[height=3cm]{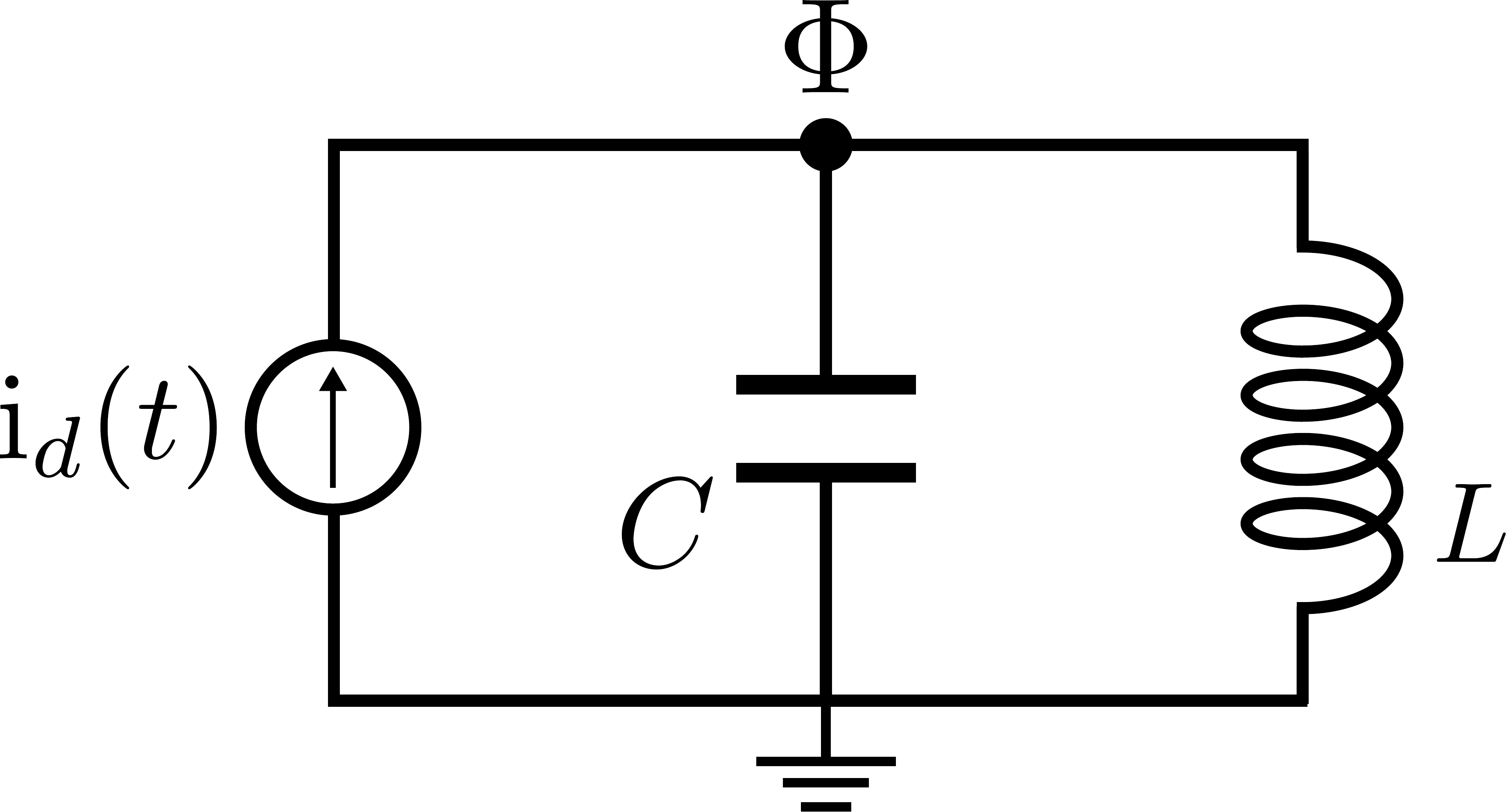}
\subcaption{}
\label{fig:LC-current}
\end{subfigure}
\begin{subfigure}[t]{0.5\textwidth}
\centering
\includegraphics[height=3cm]{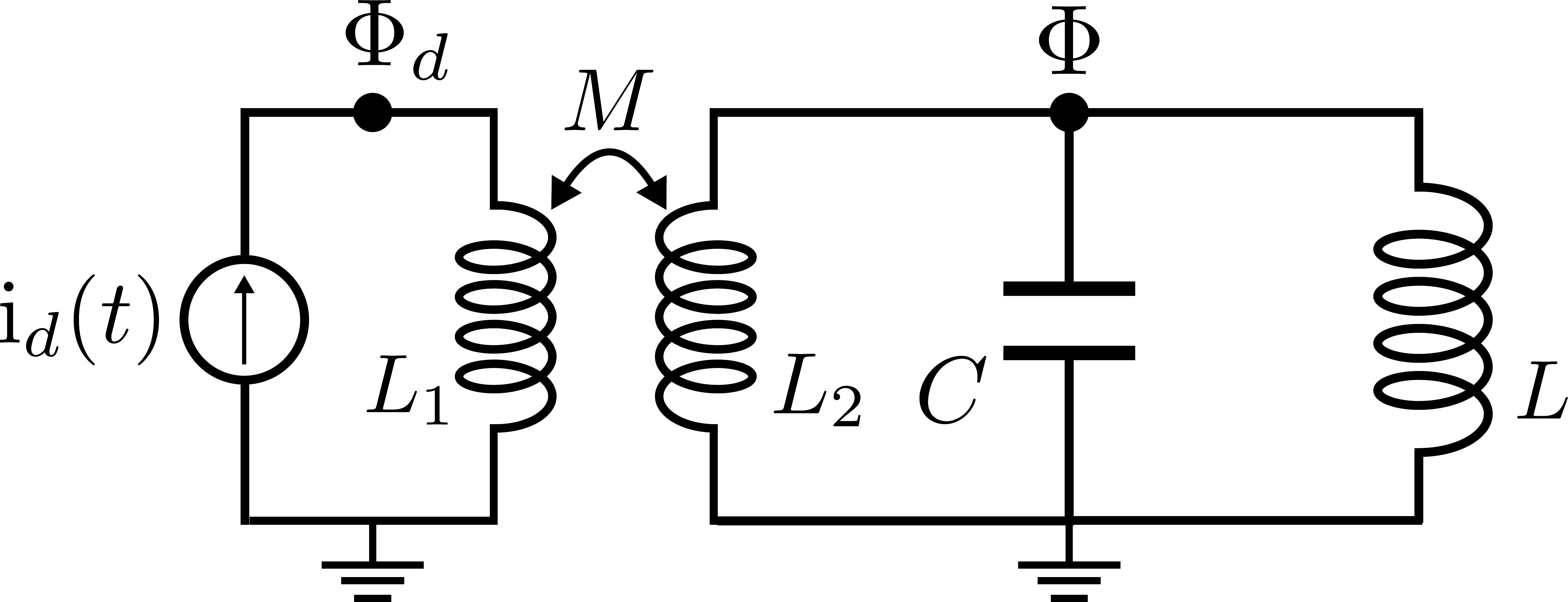}
\subcaption{}
\label{fig:LC-flux}
\end{subfigure}
	\caption{(a) Current-driven LC circuit; (b) flux-driven LC circuit.}
	\label{fig:LC-current_flux}
\end{figure}
\Question Consider the circuit in Fig.~\ref{fig:LC-current} connected to a current source $\ii_d$. Write down the equation of motion for the flux variable $\Phi(t)$, and from that recover the corresponding Lagrangian, consistent with using Eq.~\eqref{eq:include-i}. 
\Question Consider the LC circuit in Fig.~\ref{fig:LC-flux} coupled via a mutual inductance to a current source. Write down the Lagrangian of the circuit and derive the equations of motions for the variables $\Phi_d$ and $\Phi$. By eliminating the variable $\Phi_d$ using its Euler-Lagrange equation, demonstrate that the resulting equation for the variable $\Phi$ is equivalent to that of question~1. So the effect of the direct current source is the same as that of the flux induced by the inductively coupled circuit.
 \end{Exercise}

\begin{Answer}[ref={exc:flux}]
\Question Current conservation means $C \ddot{\Phi}+\frac{\Phi}{L}=\ii_d$ and can be obtained as the equation of motion from the Lagrangian $\lagrangian=\frac{C \dot{\Phi}}{2}-(\frac{\Phi^2}{2L}-\ii_d \Phi)$. \\
\Question Using Eq.~\eqref{eq:ul-ind} and Eq.~\eqref{eq:include-i}, the Lagrangian equals 
\[\lagrangian=\frac{C}{2}\dot{\Phi}^2-
 \frac{1}{\Delta}\left(\frac{L_2 \Phi_d^2}{2}+\frac{L_1 \Phi^2}{2}-M \Phi \Phi_d \right)
 -\frac{1}{2L}\Phi^2
+\ii_d \Phi_d,
\] with $\Delta=L_1 L_2 -M^2 > 0$. This results in the following equations of motion: $L_2 \Phi_d-M\Phi=\Delta \ii_d$ and $C \ddot{\Phi}+\frac{\Phi}{L}+\frac{\Phi L_1}{\Delta}-\frac{M \Phi_d}{\Delta}=0$. If we eliminate the variable $\Phi_d$, we obtain the equation of motion for $\Phi$ as $C\ddot{\Phi}+\left(\frac{1}{L}+\frac{1}{L_2} \right)\Phi=\frac{M}{L_2}\ii_d$ which is of identical form as in circuit (a) with an effective inductance $\frac{1}{L_{eq}}=\frac{1}{L}+\frac{1}{L_2}$ (which is the formula for effective inductances in parallel) and an effective current source $\frac{M}{L_2}\ii_d$.
\end{Answer}

\section{External fluxes}
\label{sec:ex-flux}

In this section, we explain what we mean by constant external fluxes applied to loops and how to include them in the description; what we want to show is how to pass from the general condition Eq.~\eqref{eq:KVL-flux} to Eq.~\eqref{eq:KVL-flux_ext}, and under which conditions this is valid. 
Let us start by considering the actual physical implementation of an external flux by means of the application of a nearby classical current source. To make this concrete, we can consider the case of an inductively shunted CPB shown in Fig.~\ref{fig:iscpb}, where on the left we have the physical circuit. We want to understand in which limit we can treat the current source and the mutual inductance as an effective external flux threading the loop formed by the Josephson junction and the inductance as depicted in Fig.~\ref{fig:iscpb} on the right.  

\begin{figure}
\centering
	\includegraphics[height=6 cm]{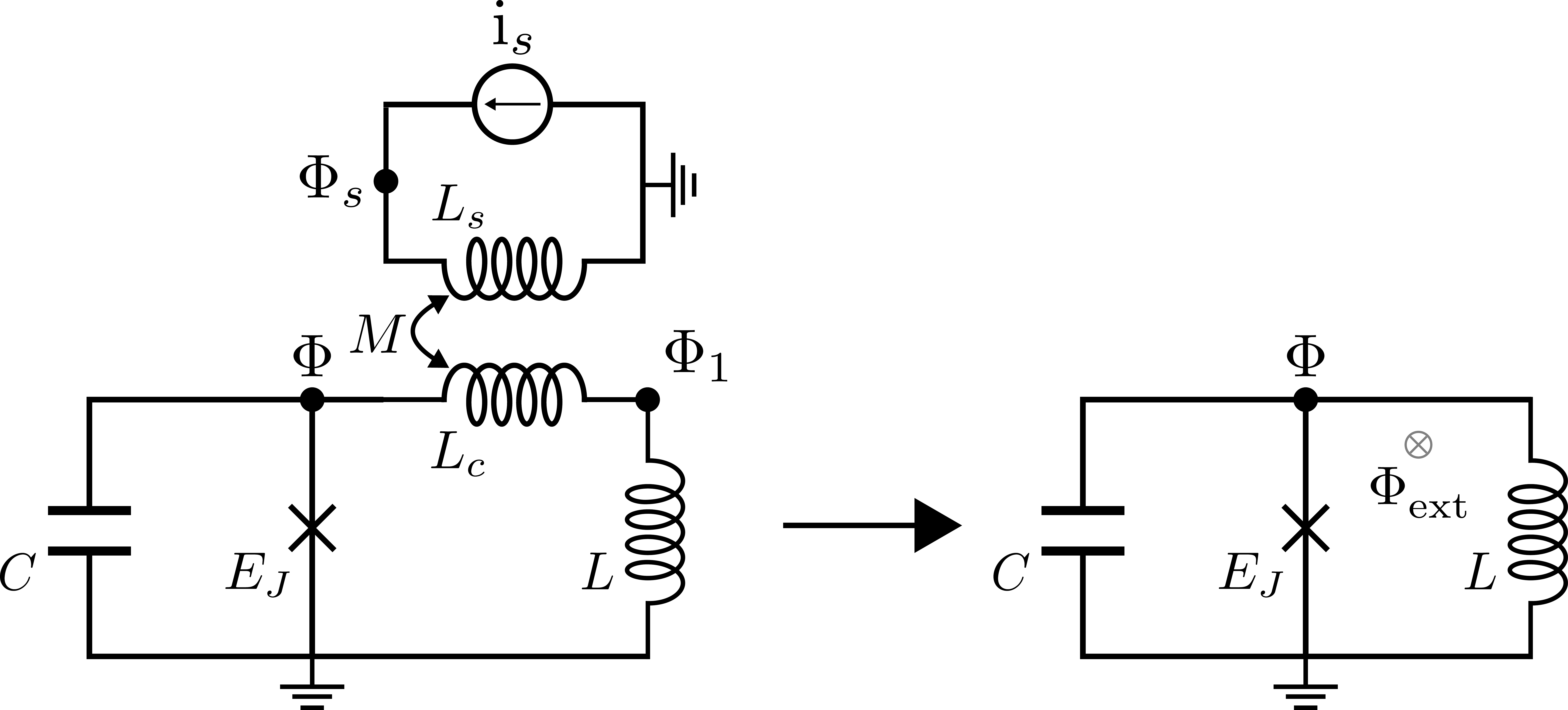}
	\caption{Circuit of an inductively shunted CPB with an applied external flux. (left) Full circuit with current source and mutual inductance. (right) Equivalent circuit with an effective external flux $\Phi_{\mathrm{ext}} = \lim_{M \rightarrow 0} M \mathrm{i}_s$.}
	\label{fig:iscpb}
\end{figure}

The current source $\ii_s$ is related to the flux $\Phi_{s}$ by $\ii_s=\Phi_s/L_s$. We now consider the limit in which the mutual inductance goes to zero, while the current source is adjusted so that their product is constant, i.e., let $\Phi_{\mathrm{ext}}=\lim_{M\rightarrow 0} M \ii_s$. This is modeling the absence of back action on the external current source: it is treated as a non-dynamical classical source. Using Eq.~\eqref{eq:ul-ind}, the potential energy associated with the self-inductances $L_c$ and $L_s$ and the mutual inductance $M$ equals
\begin{equation}
U=\frac{L_c}{2(L_s L_c-M^2)} \Phi_{c}^2+ \frac{L_s}{2(L_s L_c-M^2)} \Phi_s^2-\frac{M}{(L_s L_c-M^2)} \Phi_s \Phi_{c},
\label{eq:mi-exp}
\end{equation}
where we have called $\Phi_c = \Phi_1-\Phi$ the branch flux across the inductance $L_c$. 
In the limit $M \rightarrow 0$ and $\Phi_{\rm ext}=\lim_{M\rightarrow 0} M \ii_s$, this becomes
\begin{equation}
U(\Phi_{\rm ext})=\frac{1}{2L_c} \Phi_{c}^2+ \frac{L_s}{2} \ii_s^2-\frac{1}{L_c} \Phi_{\rm ext} \Phi_{c}=\frac{1}{2L_c}(\Phi_{c}-\Phi_{\rm ext})^2+\mathrm{constant}.
\end{equation}
 When we take $L_c \rightarrow 0$ ---so that it models a piece of material of zero length--- we see that this potential energy just enforces that $\Phi_{c}=\Phi_{\rm ext}$. Thus, in this limit, the external flux can be incorporated by letting the sum of the branch fluxes around the loop be zero, but include an extra branch with a mutual inductance coupled to a circuit with a current source, effectively giving Eq.~\eqref{eq:KVL-flux_ext}.
 
 While we took the case of an inductively shunted CPB as an example, this procedure can be repeated for the flux through any loop in an electrical circuit, i.e., by invoking a mutual inductive coupling to a current source in the proper limit \footnote{Obviously, a similar modeling would hold if one generates the magnetic flux directly via a magnet, which is, however, not done in practice.}. Thus, when we consider external fluxes, we can always keep this limiting procedure in mind.  

 \subsection{Time-dependent external fluxes}
 \label{subsec:t_dep_ext_fluxes}
 
\begin{figure}
\centering
\begin{subfigure}[h]{0.45 \textwidth}
\centering
\includegraphics[height=4cm]{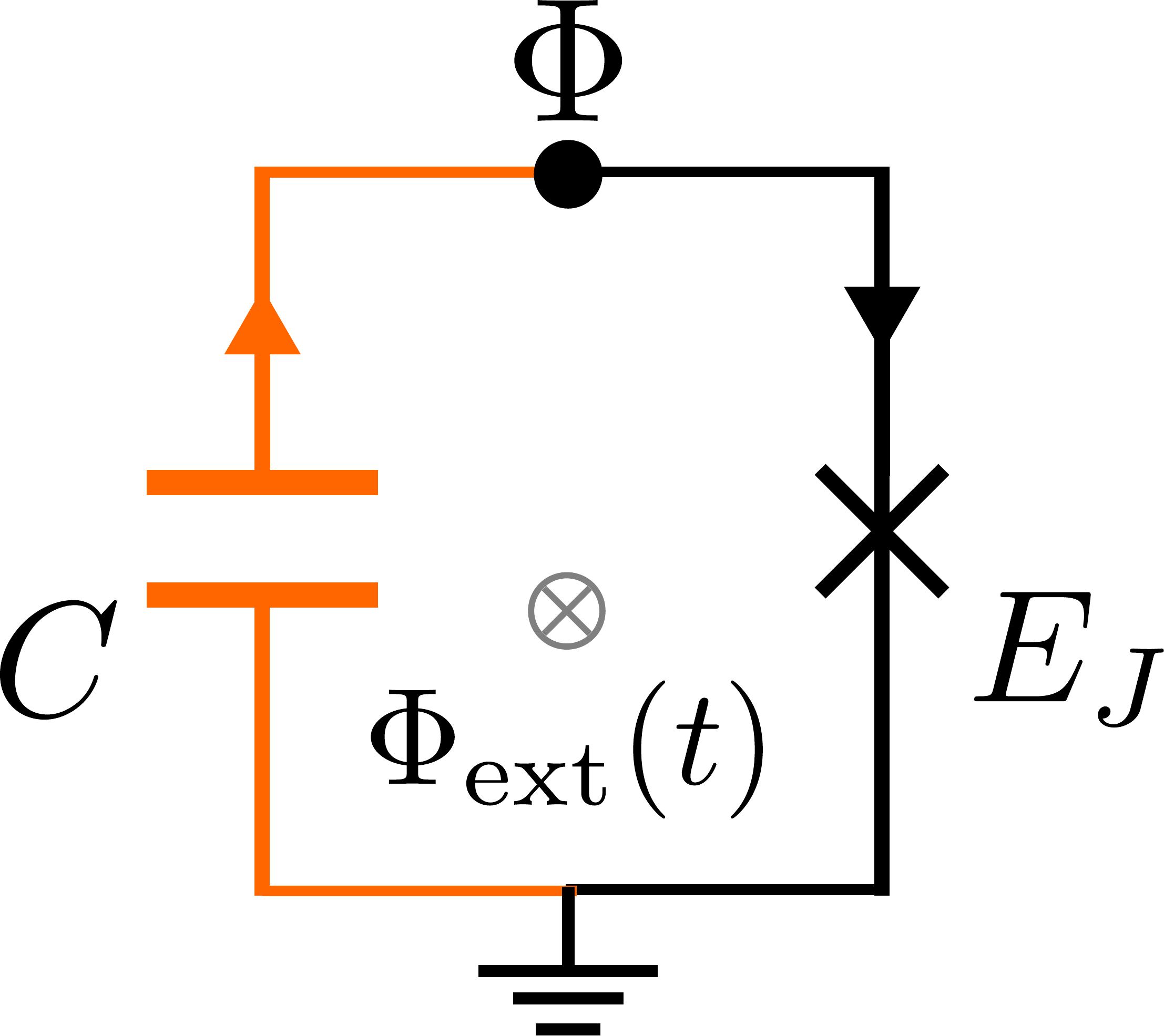}
\subcaption{}
\label{fig:2spana}
\end{subfigure}
\begin{subfigure}[h]{0.45 \textwidth}
\centering
\includegraphics[height=4cm]{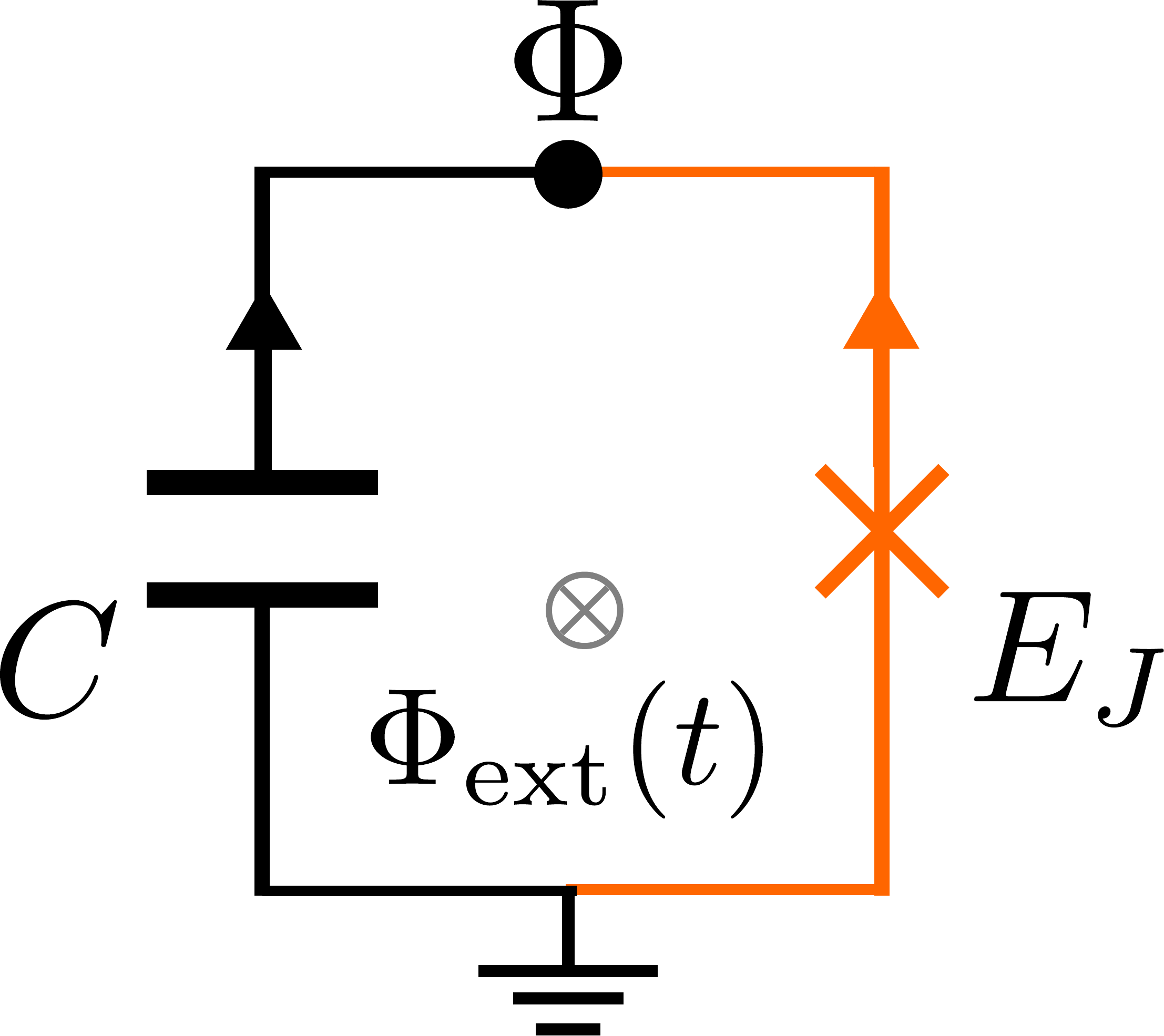}
\subcaption{}
\label{fig:2spanb}
\end{subfigure}
	\caption{Simple circuit in which the external flux $\Phi_{\rm ext}(t)$ in the loop is time-dependent.
	Two different spanning trees are shown in orange, with the respective branch orientations.}
	\label{fig:2span}
\end{figure}
 
 Let us briefly consider how to incorporate the effect of time-dependent external fluxes.
 If the externally applied flux in the loop is time-dependent, it generates an external `electromotive force' around the loop.
 This implies that Kirchhoff's voltage law, Eq.~\eqref{eq:KVL}, is modified to 
 \begin{equation}
    \mat{B} \vbvec =\vect{\vv}_{\rm ext}=\vect{\dot{\Phi}}_{\rm ext}.
 \end{equation}
 This implies that the chord-branch variable $\dot{\Phi}_{\mathfrak{b}_{\ell}}$ should be expressed in node-branch variables $\dot{\Phi}_{n_1}$ and $\dot{\Phi}_{n_2}$ and the externally induced voltage $\vv_{\rm ext}=\dot{\Phi}_{\rm ext,\ell}$ in the loop $\ell$.
 This is accomplished by taking the time-derivative of Eq.~\eqref{eq:flux-l1}, i.e.,
 \begin{equation}
    \dot{\Phi}_{\mathfrak{b}_{\ell}}=\dot{\Phi}_{n_2}-\dot{\Phi}_{n_1}+\dot{\Phi}_{\rm ext,\ell}.
    \label{eq:flux-l2}
\end{equation}
 
 Note that $\dot{\Phi}_{\rm ext,\ell}$ will only appear in the Lagrangian when the chord-branch $\mathfrak{b}_{\ell}$ is capacitive, as only in that case does its energy depend on the variable $\dot{\Phi}_{\mathfrak{b}_{\ell}}$.

Since the answer to the question, which branch is a chord and which one is not, depends on the choice of spanning tree, it might seem ambiguous at first sight whether $\Phi_{\rm ext}(t)$ or its time derivative enters the Lagrangian and eventually the Hamiltonian.
Consider for example the circuit in Fig.~\ref{fig:2span}, in which we have drawn two different spanning trees.
In the circuit in Fig.~\ref{fig:2spana}, the chord branch is capacitive and we have ${\Phi}_{\mathfrak{b}=C}+{\Phi}_{\mathfrak{b}=J}={\Phi}_{\rm ext}$, and ${\Phi}_{\mathfrak{b}=C}=\Phi$ according to the procedure around Eq.~\eqref{eq:define-node}, and the Lagrangian reads
\begin{equation}
    \lagrangian=\frac{C}{2}\dot{\Phi}^2 + E_J\cos\left(\frac{2\pi}{\Phi_0}(\Phi-\Phi_{\rm ext})\right).
    \label{eq:first}
\end{equation}
Alternatively, in the circuit in Fig.~\ref{fig:2spanb}, the chord branch is a Josephson junction. Since now $\mathfrak{b}=J$ is on the spanning tree, we have ${\Phi}_{\mathfrak{b}=J}=\Phi$ and ${\Phi}_{\mathfrak{b}=C}={\Phi}_{\rm ext}+{\Phi}_{\mathfrak{b}=J}={\Phi}_{\rm ext}+\Phi$, such that 
\begin{equation}
    \lagrangian=\frac{C}{2}(\dot{\Phi}+\dot{\Phi}_{\rm ext})^2 + E_J\cos \left(\frac{2\pi }{\Phi_0} \Phi \right).
    \label{eq:sec}
\end{equation}

The difference between Eq.~\eqref{eq:first} and Eq.~\eqref{eq:sec} may seem even more dramatic when~$\dot{\Phi}_{\rm ext}=0$, since then Eq.~\eqref{eq:first} has a dependence on the external flux, but this is absent in Eq.~\eqref{eq:sec}.
However, one can apply a simple change of variables~$\Phi\mapsto {\Phi}+{\Phi}_{\rm ext}$ to Eq.~\eqref{eq:first} and get Eq.~\eqref{eq:sec}.
Of course, this simple connection between the two Lagrangians also holds when~$\dot{\Phi}_{\rm ext}\neq 0$, so all seems fine.

\begin{figure}
\centering
\includegraphics[height=4 cm]{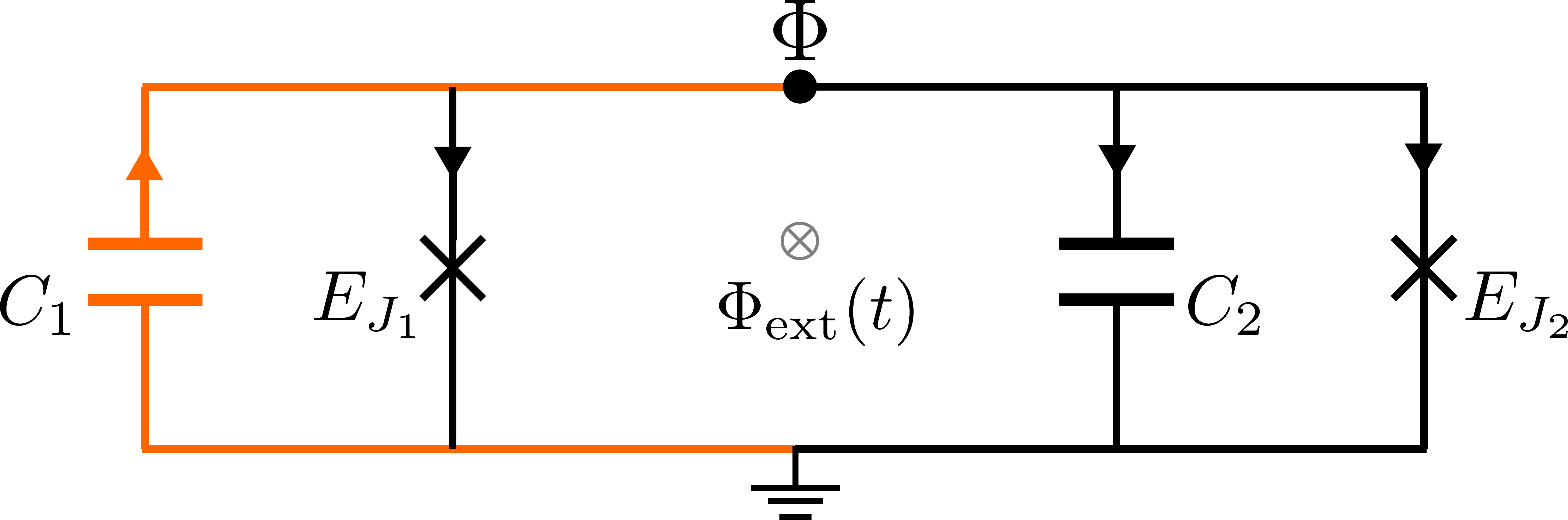}
	\caption{SQUID circuit in which the external flux $\Phi_{\rm ext}(t)$ in the central loop is time-dependent.
	The spanning tree is shown in orange, with the respective branch orientations.}
	\label{fig:squid_ext}
\end{figure}

Now let's look at a more involved circuit ---a SQUID loop---  see Fig.~\ref{fig:squid_ext}, as was done in Ref.~\cite{time-dep-flux1}. This circuit can represent, for example, a tunable transmon qubit, discussed in Section~\ref{sec:cpbspectrum}.
In this case we have the relations, ${\Phi}_{\mathfrak{b}=C_1}=\Phi$, ${\Phi}_{\mathfrak{b}=C_1}+{\Phi}_{\mathfrak{b}=C_2}= {\Phi}_{\rm ext}$, ${\Phi}_{\mathfrak{b}=C_1} =-{\Phi}_{\mathfrak{b}=J_1}$, ${\Phi}_{\mathfrak{b}=C_2} ={\Phi}_{\mathfrak{b}=J_2}$. We get the Lagrangian
\begin{equation}
    \lagrangian = \frac{C_1}{2}\dot{\Phi}^2 +\frac{C_2}{2} (\dot{\Phi}-\dot{\Phi}_{\rm ext})^2 + E_{J_1}\cos\left(\frac{2\pi }{\Phi_0} \Phi \right) + E_{J_2}\cos\Bigl(\frac{2\pi}{\Phi_0}(\Phi-{\Phi}_{\rm ext})\Bigr).
    \label{eq:lagr_squid_extflux}
\end{equation}
The conjugate charge is
\begin{equation}
    Q = \frac{\partial\lagrangian}{\partial \dot{\Phi}}
    = (C_1+C_2)\dot{\Phi} - C_2 \dot{\Phi}_{\rm ext}(t),
    \label{eq:cc-timedep}
\end{equation}
where we note the explicit time dependence via the external flux.
The Hamiltonian is then
\begin{equation}
    \hamiltonian = \frac{(Q+C_2\dot{\Phi}_{\rm ext})^2}{2(C_1+C_2)}
    - E_{J_1}\cos\Bigl(\frac{2\pi}{\Phi_0}\Phi\Bigr) - E_{J_2}\cos\Bigl(\frac{2\pi}{\Phi_0}(\Phi-{\Phi}_{\rm ext})\Bigr).
    \label{eq:hamil_squid_extflux_1}
\end{equation}
The dependence of the Hamiltonian on $\dot{\Phi}_{\rm ext}$, i.e., in the capacitive part, may not be desirable for computational purposes, and we can ask whether there is always a variable transformation in the Lagrangian which ensures that $\dot{\Phi}_{\rm ext}$ only enters the potential term. 

Consider the transformation $\Phi\mapsto \Phi + \alpha {\Phi}_{\rm ext}$, with a `gauge' parameter~$\alpha$ to be determined.
The Lagrangian in Eq.~\eqref{eq:lagr_squid_extflux} becomes
\begin{multline}
    \lagrangian = \frac{C_1+C_2}{2}\dot{\Phi}^2 + \bigl(C_1\alpha + (1-\alpha)C_2\bigr) \dot{\Phi}\dot{\Phi}_{\rm ext} + 
    E_{J_1}\cos\biggl(\frac{2\pi}{\Phi_0}(\Phi+ \alpha {\Phi}_{\rm ext})\biggr) + E_{J_2}\cos\biggl(\frac{2\pi}{\Phi_0}(\Phi-(1-\alpha){\Phi}_{\rm ext})\biggr),
\end{multline}
where in the last line we have dropped the terms that do not contain~$\dot{\Phi}$ or~${\Phi}$.
To get rid of~$\dot{\Phi}_{\rm ext}$ we can set $(C_1\alpha + (1-\alpha)C_2)\equiv 0 \implies \alpha=\frac{C_2}{C_2-C_1}$.
This is called the `irrotational gauge' in Ref.~\cite{time-dep-flux1}.
It follows directly that the Hamiltonian is
\begin{multline}
    \hamiltonian = \frac{Q^2}{2(C_1+C_2)}- E_{J_1}\cos\biggl(\frac{2\pi}{\Phi_0}\biggl(\Phi- \frac{C_2}{C_1-C_2} {\Phi}_{\rm ext}\biggr)\biggr) 
    - E_{J_2}\cos\biggl(\frac{2\pi}{\Phi_0}\biggl(\Phi- \frac{C_1}{C_1-C_2}{\Phi}_{\rm ext} \biggr)\biggr).
    \label{eq:hamil_squid_extflux_3}
\end{multline}
While one could in principle solve for the system dynamics using either~\cref{eq:hamil_squid_extflux_1} or \cref{eq:hamil_squid_extflux_3} depending on the chosen reference frame, Eq.~\eqref{eq:hamil_squid_extflux_3} seems to make this task easier.
Furthermore, Ref.~\cite{time-dep-flux1} makes the point that on a more fundamental level, only Eq.~\eqref{eq:hamil_squid_extflux_3} leads to consistent predictions of the qubit relaxation time~$T_1$ due to flux noise, using Fermi's golden rule. According to Fermi's golden rule (cf. Eq.~\eqref{eq:fermi}), one approximately has $T_1^{-1}\propto |\braket{0|\partial_{{\Phi}_{\rm ext}}H|1}|^2$, where $0,1$ are the ground and first-excited state of the qubit, respectively. Based on this formula, in the case of unequal junctions, the way in which the external flux is split across the two junctions leads to different predictions. 
While one approach could be to adapt this formula, the approach taken by Ref.~\cite{time-dep-flux1} is to adapt the circuit-quantization procedure.
In particular, the authors generalize the change of variables considered above to a systematic procedure to split the flux across multiple Josephson junctions in possibly multiple loops, effectively replacing the standard procedure in Eq.\eqref{eq:flux-l1} and Eq.~\eqref{eq:flux-l2} to assign the external (time-dependent) flux solely to the chord branch which closes the fundamental loop through which the flux is threading. The theoretical analysis of \cite{time-dep-flux1} has been recently experimentally confirmed in Ref.~\cite{byron2023}.

Ref.~\cite{time-dep-flux2} extends these results from the lumped-element case to the case of continuous structures.
In this case they identify the Coulomb gauge as the gauge for the electromagnetic field that restricts the dependence on external magnetic fluxes to the potential in the Hamiltonian (and not in the kinetic energy via $\dot{\Phi}_{\rm ext}$).
Besides, the results of \cite{time-dep-flux1} are validated in the sense that it is shown that the continuous limit of their method (based on the irrotational gauge) leads precisely to the Coulomb gauge.
Further considerations about treating time-dependent external fluxes in circuit quantization can be found in Refs.~\cite{time-dep-flux1, byron2023, time-dep-flux2}.

\subsection{External fluxes in superconductors and fluxoid quantization}

\begin{figure}
\centering
\includegraphics[height=4 cm]{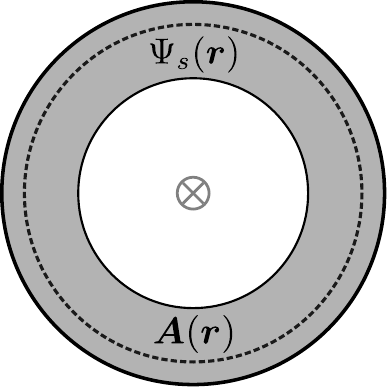}
	\caption{Superconducting loop (grey) pierced by a magnetic field. The dashed line depicts a possible integration contour in the bulk of the superconductor for the integrals in Eq.\eqref{eq:fluxoid_quant}.}
	\label{fig:suploop}
\end{figure}

The circuit theory developed so far applies to any lossless electrical circuit. The only way superconductivity has entered the discussion is via the introduction of the Josephson junction which does not exist for normal metals. However, we know that superconductors are different than normal metals with zero resistance and so this raises the question which other properties of superconductors could play a role in a lumped-element description. Two prominent superconducting effects should be considered: the Meissner effect and fluxoid quantization. The Meissner effect consists of the exclusion of the magnetic field out of a superconductor during the superconducting phase transition \cite{tinkham, girvin_yang_2019}. Thus, in a superconductor not only the electric field is zero, as in normal metals, but also the magnetic field. 
 
Fluxoid quantization instead plays a nontrivial role. For a superconductor let the superconducting order parameter be given as
\begin{equation}
\Psi_s(\vect{r}) = \lvert \Psi_s(\vect{r}) \rvert e^{i \varphi(\vect{r})}.
\label{eq:op}
\end{equation}
The superconducting current density $\vect{J}_s(\vect{r})$ then equals
\begin{equation}
\label{eq:supercurrent}
\vect{J}_s(\vect{r}) = - \frac{e}{m_e} |\Psi_s(\vect{r})|^2 \biggl( \hbar \nabla \varphi(\vect{r}) + 2 e \vect{A}(\vect{r}) \biggl).
\end{equation}
with vector potential $\vect{A}(\vect{r})$  \cite{tinkham, girvin_yang_2019, timm_notes}. Eq.~\eqref{eq:supercurrent} is the fundamental starting equation for deriving the Josephson effect, and, in the absence of magnetic field, the phase $\varphi$ in the first Josephson relation Eq.~\eqref{eq:JJ-relation-phi} coincides with the difference between the phases $\varphi_1$ and $\varphi_2$ of the superconducting order parameter of the two superconductors forming the junction (see Appendix A in \cite{clarke2006squid} for a simple derivation). Assuming $|\Psi_s(\vect{r})|\neq 0, \forall \vect{r}$ in the superconductor and integrating over a closed loop we obtain
\begin{equation}
\frac{m_e}{2 e^2} \oint d \vect{r} \cdot \frac{\vect{J}_s(\vect{r})}{|\Psi_{s}(\vect{r})|^2} +  \oint d \vect{r} \cdot \vect{A}(\vect{r}) = \frac{m_e}{2 e^2} \oint d \vect{r} \cdot \frac{\vect{J}_s(\vect{r})}{|\Psi_{s}(\vect{r})|^2} +  \Phi_{\mathrm{loop}} = m \Phi_0, \quad m \in \mathbb{Z},
\label{eq:fluxoid_quant}
\end{equation}
where we have identified the magnetic flux enclosed by the loop as $\Phi_{\mathrm{loop}} = \oint d \vect{r} \cdot \vect{A}(\vect{r})$. The term on the left-hand side of Eq.~\eqref{eq:fluxoid_quant} is called the \emph{fluxoid} and the condition in Eq.~\eqref{eq:fluxoid_quant} itself fluxoid quantization.

Let us analyze Eq.~\eqref{eq:fluxoid_quant}. We notice that the fluxoid is given by the sum of a magnetic flux $\Phi_{\mathrm{loop}}$ and a ``kinetic" term that is due to the line integral of the superconducting current density divided by the modulus squared of the order parameter. This term has not entered in our discussion about Lagrangians and Hamiltonians in electrical circuits. However, as argued in Ref.~\cite{tinkham}, the superconducting current density $\vect{J}_s(\vect{r})$ is only non-zero close to the surface of the superconductor, more precisely within the London penetration depth of the material. Thus, if our circuit is made out of superconductors that are thicker than the London penetration depth, we can always take the line integral in an inner region where $\bm{J}_s(\bm{r})\approx0$ (see Ref.~\cite{tinkham} for more details). If we neglect the kinetic term, we arrive at the condition  

\begin{equation}
\Phi_{\mathrm{loop}} = m \Phi_0, \quad m \in \mathbb{Z},
\label{eq:flux_quant}
\end{equation}
which, unsurprisingly, goes under the name of flux quantization. 

There is still one important question to be answered: what is the interpretation of the integer $m$, i.e., the number of fluxoids, in Eqs.~\eqref{eq:fluxoid_quant}, \eqref{eq:flux_quant}? Given our derivation, we immediately conclude that $m$ represents the (signed) number of times that the phase $\varphi(\bm{r})$ of the superconducting order parameter goes around $2\pi$ in the loop. Notice that it is only the fluxoid number that matters, and not the specific distribution of the phase $\varphi(\bm{r})$ around the path. Moreover, Eq.~\eqref{eq:flux_quant} seems to suggest that the general condition in Eq.~\eqref{eq:KVL-flux_ext} with superconductors should be modified as

\begin{equation}
\mat{B} \vect{\Phi}_{\mathfrak{b}}=\vect{\Phi}_{\mathrm{ext}} + \vect{m} \Phi_0,
\label{eq:KVL-fluxoid}
\end{equation}
where $\vect{m} \in \mathbb{Z}^{N-M}$ with $N+1$ the number of nodes and $M$ the number of branches in the circuit graph. Looking at Eq.~\eqref{eq:KVL-fluxoid}, if $\bm{m}$ is a constant and cannot change in time, we can interpret it simply as Eq.~\eqref{eq:KVL-flux_ext} with a new effective vector of external fluxes $\vect{\Phi}_{\mathrm{ext}}'= \vect{\Phi}_{\mathrm{ext}} + \vect{m} \Phi_0$. Thus, if $\bm{m}$ is constant our previous description is still valid with potentially shifted flux parameters. Also, if we only have Josephson junctions, or in general, $2 \pi$-periodic potentials, $\vect{m}\Phi_0$ would not even enter in the Lagrangian and Hamiltonian.

However, is it possible for $\vect{m}$ to vary in time in a superconductor? And if so, shall we treat it as an additional dynamical degree of freedom that we should somehow quantize and to which we assign orthogonal states $|\bm{m} \rangle$ in the Hilbert space? The answer to both questions is: in theory yes, but in most practical cases no. A change in the number of fluxoids $m$ in a superconducting loop goes under the name of phase slip. In order for the process to manifest itself one needs very thin superconducting nanowires \cite{Mooij2006, lau2001, zaikin1997}. At a temperature close to the critical temperature phase slips are thermally activated \cite{golubev2008} but at lower temperatures coherent tunneling between neighboring fluxoid states $|m \rangle \leftrightarrow |m \pm 1\rangle$ should be possible in nanowires. Although hard to observe experimentally, this has led to the notion that coherent phase slip tunneling is a dual effect to the Josephson effect in superconducting nanowires \cite{Mooij2006}. 
As we will not deal with nanowires and their treatment in this book, we can assume that the fluxoid degree of freedom is fixed to a certain value; specifically, we can assume $\bm{m} = 0$, since, if before cooling down the superconductor we do not apply any magnetic flux, then the $\bm{m}=0$ case is the one that minimizes the energy in a superconducting loop (see however Ref.~\cite{kirtley2003} for an experimental study of spontaneous fluxoid formation).

The previous discussion shows that it is quite tricky to model a piece of superconductor as a lumped-element and eventually theoretical difficulties in the description arise from this fact. We refer the reader to Ref.~\cite{semenov2016} for an example of a possible continuous model of a thin superconducting wire. Another possible far-reaching approach is to model the superconductor as a chain of Josephson junctions \cite{matveev2002}. In this case, phase slips correspond to changes by (approximately) $2 \pi$ in one of the junctions in the array and the coherent tunneling between them can be rigorously identified via a perturbative approach. This analysis is also relevant for the fluxonium that we will briefly introduce in Section~\ref{subsec:fluxonium}. We refer the reader to the original fluxonium references for a more detailed discussion \cite{ManucharyanPhd, Manucharyan113}. We conclude by remarking that there is still some debate in the literature about the correct, effective treatment of superconductors in the presence of phase slips (see Refs.~\cite{le_grimsmo2019, koliofoti2023}). 

\newpage
 
\begin{Exercise}[title={Circuit equivalences using the Lagrangian formalism}, label=exc:series]
\begin{figure}[htb]
\centering
\begin{subfigure}[h]{0.45 \textwidth}
\centering
\includegraphics[height=5cm]{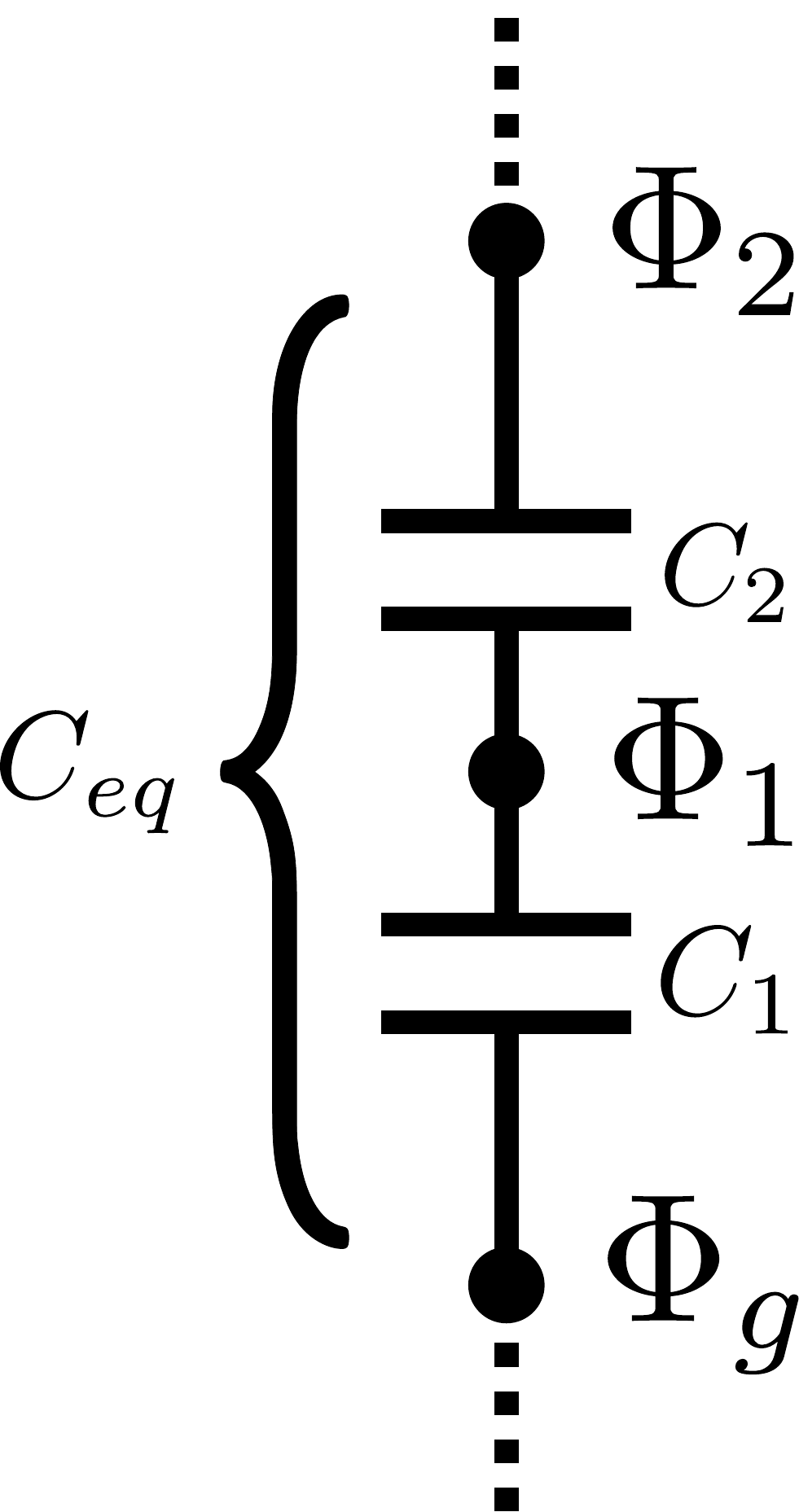}
\subcaption{}
\label{fig:series_cap}
\end{subfigure}
\begin{subfigure}[h]{0.45 \textwidth}
\centering
\includegraphics[height=5cm]{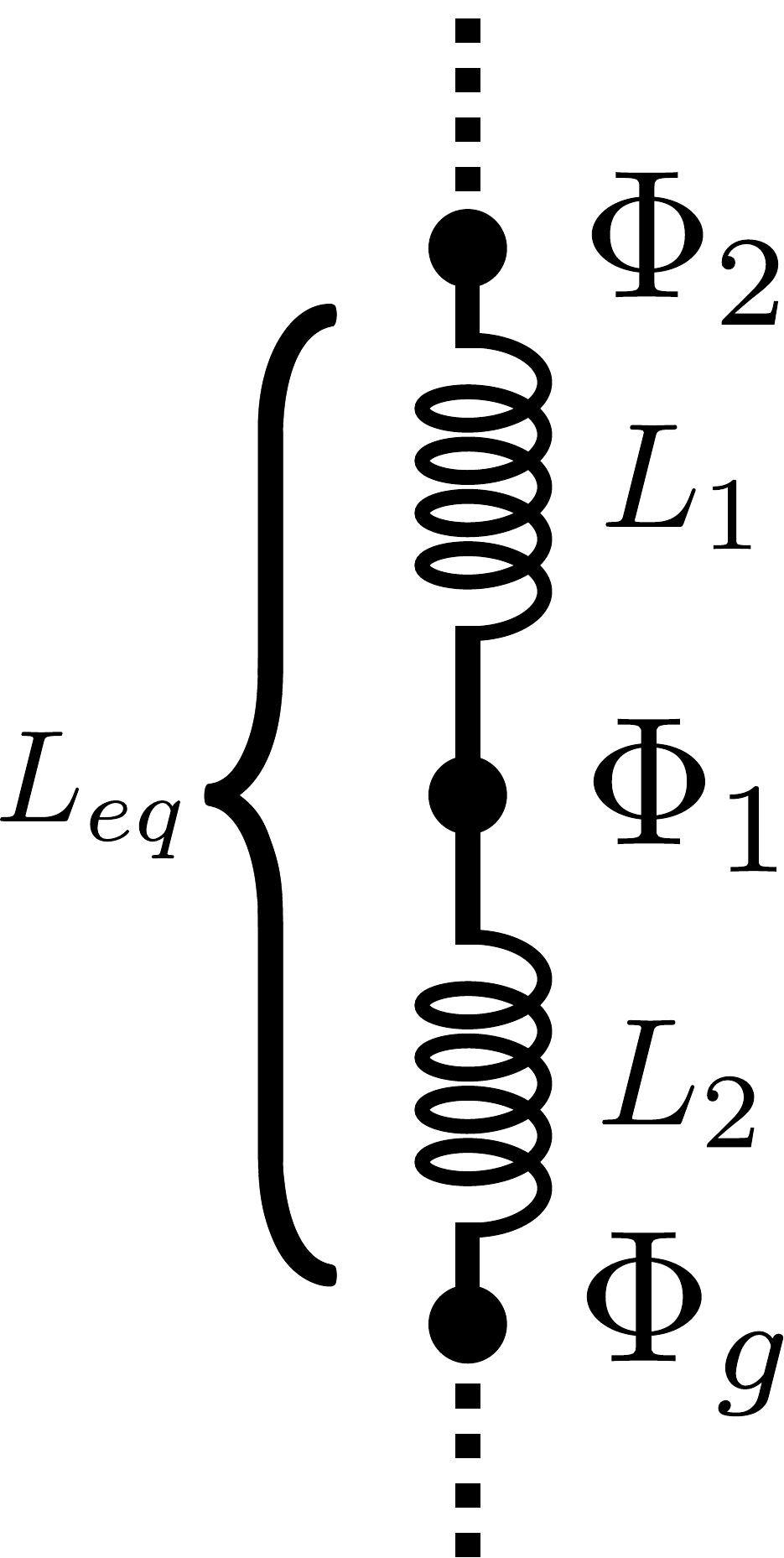}
\subcaption{}
\label{fig:series_ind}
\end{subfigure}
\caption{(a) Series of capacitances (b) series of inductances.}
\label{fig:series_cap_ind}
\end{figure}

In this exercise we verify that the well-known circuit equivalences can be derived using the Lagrangian formalism. Consider the series of capacitances in Fig.~\ref{fig:series_cap} which can be part of a larger circuit at the nodes $\Phi_g$ and $\Phi_2$ (but not at node $\Phi_1$). Write down the Lagrangian for the circuit and use the Euler-Lagrange equations to show that the capacitances can be substituted by an equivalent capacitance $C_{eq}$ with $1/C_{eq}=1/C_1+1/C_2$. Do the same for the series of inductances in Fig.~\ref{fig:series_ind}, which can again be part of a larger circuit, and obtain the equivalent inductance $L_{eq}=L_1+L_2$. 
\end{Exercise}

\begin{Answer}[ref={exc:series}]
The Lagrangian is $\lagrangian=\lagrangian_{\rm else}(\Phi_2,\Phi_g)+\frac{C_2}{2}(\dot{\Phi}_2-\dot{\Phi}_1)^2+\frac{C_1}{2}(\dot{\Phi}_1-\dot{\Phi}_g)^2$ as these two branches are in any spanning tree, since they are the only way to reach node $\Phi_1$.  $\lagrangian_{\rm else}(\Phi_2,\Phi_g)$ represents the other part of the circuit. Hence the Euler-Lagrange equation for $\Phi_1$ reads $\ddot{\Phi}_1=\frac{C_1 \ddot{\Phi}_g+C_2 \ddot{\Phi}_2}{C_2+C_1}$, which allows one to eliminate $\dot{\Phi}_1(t)=\int_{-\infty}^t \ddot{\Phi}_1(t')dt'$ in $\lagrangian$, resulting in $\lagrangian_{\rm else}(\Phi_2,\Phi_g)+\frac{C_{eq}}{2}(\dot{\Phi}_2-\dot{\Phi}_g)^2$. Similarly, for the inductances we have $\lagrangian=\lagrangian_{\rm else}(\Phi_2,\Phi_g)+\frac{1}{2L_1}(\Phi_2-\Phi_1)^2+\frac{1}{2L_2}(\Phi_1-\Phi_g)^2$, which gives rise to $\Phi_1=\frac{L_2 \Phi_2+L_1 \Phi_g}{L_2+L_1}$ which is used to eliminate $\Phi_1$. This reasoning is completely correct for classical dynamics; it has problems in the quantum theory because of zero-point fluctuations. This issue is treated in detail in Section~\ref{subsec:patho} and Appendix~\ref{sec:elim}.
\end{Answer}

\chapter[Applying canonical quantization]{Applying canonical quantization to an electrical circuit}
\label{chap:cq-app}

In Chapter~\ref{chap:lagr_ham} we have outlined a general procedure to obtain the Hamiltonian of an electrical circuit and in this chapter we will show how to apply this procedure to analyze various qubits. 

The idea is that we express the Lagrangian in terms of the energy contributions using branch variables. Then, given the electrical circuit and its associated graph we identify the independent variables in the circuit by finding the spanning tree and we re-express the Lagrangian in terms of these independent node flux variables. 
Once we have the Lagrangian and a set of independent flux variables $\Phi_i$, we can write down the classical equations of motion
\begin{equation}
    \frac{d}{dt} \left( \frac{\partial \mathcal{L}}{\partial \dot{\Phi}_k} \right)
	- \frac{\partial \mathcal{L}}{\partial \Phi_k} = 0.
	\label{eq:genEL}
\end{equation}
When $\Phi_k$ are node fluxes, we observe that one can read these classical equations of motion as expressing Kirchhoff's current law in Eq.~\eqref{eq:KCL}: the sum of currents is zero at each node (note that the dimension of $\partial \lagrangian/\partial \Phi_k$ is that of ampere).

\section{Invertibility of the capacitance matrix}

In Appendix~\ref{app:cc} we have seen that we can properly define a Hamiltonian from a Lagrangian when the symmetric matrix $\mat{C}$ in Eq.~\eqref{eq:standard-form} is positive-definite, making it invertible. When the variables $\dot{x}_i$ are the variables $\dot{\Phi}_i$, the matrix $\mat{C}$ is called the capacitance matrix.
Here we consider under what conditions the capacitance matrix is indeed invertible.
 The capacitance matrix is found by considering all capacitive (with linear capacitors) branches in the graph, in order to write down the kinetic energy $T$ of the Lagrangian and then switching to node variables.

Let $G$ be a graph with $N+1$ nodes of which one is chosen as the ground node. Each capacitive branch flux can be written as $\Phi_{\mathfrak{b}}=\pm(\Phi_n-\Phi_{n'})$ with $\mathfrak{b}=(n n')$ where the $\pm$ depends on orientation, so that $\dot{\Phi}_{\mathfrak{b}}=\pm(\dot{\Phi}_n-\dot{\Phi}_{n'})$ and $\dot{\Phi}_{n=0}=0$ as it is the designated ground node. We have
\begin{equation}
\label{eq:gencapenergy}
T=\frac{1}{2}\sum_{{\rm cap. branch} \;\; \mathfrak{b}} C_{\mathfrak{b}} \dot{\Phi}_{\mathfrak{b}}^2=\frac{1}{2} \sum_{n,n'=1}^N \dot{\Phi}_n C_{nn'} \dot{\Phi}_{n'},
\end{equation}
defining the $N \times N$ symmetric capacitance matrix $\mat{C}$ with entries $C_{nn'}$.  We note that in Eq.~\eqref{eq:gencapenergy} the orientation sign $\pm$ is actually irrelevant. The matrix elements of the capacitance matrix are given by
\begin{eqnarray*}
\forall n,m=1,\ldots, N, \colon C_{nn}=\sum_{n'=0\colon n\neq n'}^{N} C_{\mathfrak{b}=(nn')}\\
{\mbox{$m\neq n$,}}\;\; C_{nm}= C_{mn}=- C_{\mathfrak{b}=(mn)} < 0.
\end{eqnarray*}
Thus the diagonal entries ($C_{nn}$) sum over all the capacitances of the capacitive branches in which the node $n$ participates, including a capacitance to the ground node. $C_{nm}$ is minus the capacitance of the branch $(mn)$.
An example is the $2 \times 2$ matrix in Eq.~\eqref{eq:lagrangV}.

Let us now prove the following simple proposition.

\begin{lem}
The capacitance matrix $\mat{C}$ is invertible if the graph $G$ has a spanning tree (connected to the ground node) which consists of only capacitive branches $\mathfrak{b}$, each with non-zero capacitance $C_{\mathfrak{b}} >0$.
\label{lem:captree}
\end{lem}

\begin{proof}
Let $G_{\rm cap} \subseteq G$ be the capacitive subgraph of $G$ with the same number of nodes as $G$ such that two nodes in $G_{\rm cap}$ are connected if there is a capacitive branch between them. If $G$ has a spanning tree which consists of capacitive branches, then $G_{\rm cap}$ is connected (naturally $G_{\rm cap}$ can consist of more branches than those in the spanning tree).
The capacitance matrix $\mat{C}$ is (almost!) the weighted Laplace matrix of the form $\mat{L}=\mat{D}-\mat{A}$ where $\mat{A}$ is a symmetric (adjacency) matrix of weights with $A_{ii}=0$ and $\mat{D}$ is the sum of weights on each row of $\mat{A}$. Here $\mat{A}$ is the adjacency matrix of the capacitive subgraph $G_{\rm cap} \subseteq G$. The `almost' relates to the fact that the capacitance matrix $\mat{C}$ is a $N \times N$ matrix, while the weighted Laplace matrix $\mat{L}$ is a $(N+1) \times (N+1)$ matrix: we obtain $\mat{C}$ from $\mat{L}$ by just removing the row and column corresponding to the ground node $\Phi_{n=0}$.
A weighted Laplace matrix $\mat{L}$ is positive-semidefinite and it has a unique eigenvector $(1,1,\dots,1)$ with zero eigenvalue, when the underlying graph, in this case $G_{\rm cap}$, is connected and includes each node in the underlying graph $G$. 
We observe that under the constraint $V_0=\dot{\Phi}_0=0$, this unique zero-eigenvalue vector is the zero vector with $\dot{\Phi}_n=V_n=0$ $\forall n$ (all voltages same as ground). 
Moving from $\mat{L}$ to $\mat{C}$ exactly eliminates this zero eigenvector. This implies that when $G_{\rm cap}$ is a connected graph and it is connected to the ground node, $\mat{C}$ only has positive eigenvalues and hence is invertible. 
\end{proof}

As an example in which the condition of the proposition is not fullfilled, image that $G_{\rm cap}$ has two connected components, one connected to ground, the other one `freely floating'. In that case $\mat{C}$ breaks up into two submatrices, one of which has a zero eigenvalue. The zero eigenvector of the submatrix corresponding to the freely floating subgraph corresponds to a constant voltage on the nodes. By grounding this `freely floating' subgraph, one removes the zero eigenvalue and makes both submatrices invertible. 

It is not a necessity to find a capacitive spanning tree between independent node variables, as we will see in some of the examples in Section~\ref{sec:examples}. What matters is that we identify the set of {\em independent} variables in the circuit; we are free to choose the most convenient set ---leading to a clear interpretation of the properties of the circuit---  and we can eliminate degrees of freedom which have no dynamics `by hand'. So in fact, when there is no spanning tree, as in Proposition \ref{lem:captree}, we know that we should be able to reduce the number of independent dynamical variables. 

\section{Non-locality of capacitive interactions in the Hamiltonian}

Before we discuss some examples, it is interesting to make another observation on the form and non-local connectivity of the Hamiltonian due to the capacitive couplings in the electrical circuit. 

The information about capacitive couplings in an electrical circuit is stored in the capacitance matrix $\mat{C}$. In the Hamiltonian that we obtain, i.e., 
\begin{equation}
    \lagrangian=\frac{1}{2}\vect{\dot{\Phi}}^T\mat{C} \vect{\dot{\Phi}}-U(\vect{\Phi}) \rightarrow \hamiltonian=\frac{1}{2} \vect{Q}^T \mat{C}^{-1} \vect{Q} +U(\vect{\Phi}),
\end{equation}
the coupling is via the inverse of the capacitance matrix $\mat{C}^{-1}$. 
When $\mat{C}$ is sparse and `local' as a matrix, representing a (natural) small number of couplings per node, its inverse $\mat{C}^{-1}$ is not. Thus, the couplings in $\hamiltonian$ go beyond nearest-neighbor nodes which are capacitively coupled.

The nonlocality of $\mat{C}^{-1}$ is usually dealt with perturbatively, which is warranted when some couplings are much larger than others. A sufficiently fast fall-off of such non-nearest neighbor interactions is important for design, error control and cross-talk in quantum electrical circuits.

For example, if the capacitive coupling between a pair of nodes is strong (say, equal to $C_{\rm strong}$), while capacitive couplings to the other remaining nodes are weak (i.e.,~of strength $\epsilon=C_{\rm weak}$), we can use perturbative methods to approximate the inverse capacitance matrix. We can write
\begin{equation}
\mat{C}=\left(\begin{array}{cc} \mat{C}_{\rm strong}  & 0  \\ 0 & \mat{C}_{\rm rest} \end{array}\right)+\mat{E}\equiv \mat{C}_0+\mat{E}, 
\end{equation}
with e.g.~$\mat{C}_{\rm strong}=\left(\begin{array}{cc}C_{\rm strong} & -C_{\rm strong} \\ -C_{\rm strong} & C_{\rm strong} \end{array}\right)$ and $\mat{C}_{\rm rest}$ is the capacitance matrix among the remaining nodes, representing couplings among themselves. The perturbative matrix $\mat{E}$ represents capacitive couplings between the two subsets of nodes (the pair and the rest) and the matrix entries in $\mat{E}$ are at most $O(\epsilon)$.   We have 
\begin{equation}
\mat{C}^{-1}=(\mat{C}_0(\mathds{1}+\mat{C}_0^{-1} \mat{E}))^{-1}=(\mathds{1}+\mat{C}_0^{-1} \mat{E})^{-1} \mat{C}_0^{-1} \approx 
(\mathds{1}-\mat{C}_0^{-1} \mat{E})\mat{C}_0^{-1}=\mat{C}_0^{-1}-\mat{C}_{0}^{-1} \mat{E} \mat{C}_0^{-1},
\label{eq:expansion}
\end{equation}
via Taylor expanding a matrix inverse. The inverse of the block matrix $\mat{C}_0$ is of course a block matrix with the same block connectivity structure.
More generally, if we divide the circuit into strongly-coupled small subsets of nodes, this coupling structure of the perturbatively expanded $\mat{C}^{-1}$ can be easily determined. The first-order correction proportional to $\mat{E}$ in Eq.~\eqref{eq:expansion} has connectivity set by the pertubatively weak capacitive couplings. A next-order correction $O(||\mat{E}||^2)$ then has connectivity determined by $\mat{E}^2$, i.e.,~there are non-zero couplings between nodes in different subsets which are weakly coupled via one intermediate node and so on.

\section{Examples}
\label{sec:examples}

In this section, we go through a variety of quantum circuits illustrating the circuit-QED quantization method and defining various qubits. We start with a case in which one cannot find a spanning tree with only capacitive branches in the electrical circuit, so one cannot apply Proposition~\ref{lem:captree}.

\subsection{A pathological case?}
\label{subsec:patho}

Consider the electrical circuit in Fig.~\ref{fig:patho} with two 
node fluxes $\Phi_1$ and $\Phi_2$. This seemingly pathological case has been studied in Ref.~\cite{rymarz:msc}, with an extensive update in Ref.~\cite{rymarz:sing}. The Lagrangian reads
\begin{equation}
\lagrangian=\frac{C}{2} \dot{\Phi}_2^2-\frac{1}{2L}(\Phi_2-\Phi_1)^2+E_J\cos \biggl(\frac{2 \pi }{\Phi_0}\Phi_1 \biggr).
\label{lagpath}
\end{equation}
Observe that if we do not put a small capacitance $C_J$ in parallel with the Josephson junction, the circuit does not have a spanning tree with only capacitive branches. Note that in principle each Josephson junction comes with a small capacitance $C_J$ so this is physically warranted, but here we consider the case when it is absent.

The additional capacitance $C_J$ would add a kinetic energy $C_J \dot{\Phi}_1^2/2 $ to the Lagrangian. 
In the limit of very small $C_J$, the corresponding term in the Hamiltonian, proportional to $C_J^{-1}$, becomes very large and the mode associated with it has high energy (tiny mass). If one is used to working with perturbed Hamiltonians, it may be counterintuitive that a term which we omit as it is small in the circuit actually has a large strength in $\hamiltonian$. 

For the circuit in Fig.~\ref{fig:patho}, in the limit of $C_J\rightarrow 0$, the Euler-Lagrange equation for the variable $\Phi_1$ is 
\begin{equation}
C_J\ddot{\Phi}_1+\frac{\partial U}{\partial \Phi_1}=0 \rightarrow_{C_J \rightarrow 0} \frac{1}{L}(\Phi_2-\Phi_1)=I_c \sin \biggl(\frac{2\pi}{\Phi_0}\Phi_1 \biggr),
\label{eq:sol}
\end{equation}
suggesting that we should simply put $\Phi_1$ at the minimum of its potential. The last equation states that the current through the inductive branch should be the same as the current through the Josephson junction branch. Eq.~\eqref{eq:sol} 
is of the form
$\Phi_2=f(\Phi_1)$; if we can invert this to form $\Phi_1=f^{-1}(\Phi_2)$, we can eliminate $\Phi_1$ from Eq.~(\ref{lagpath}) and have a Lagrangian only in terms of $\Phi_2$.

This classical approach obviously has some problems, in that the inverse function $f^{-1}$ may not even exist (i.e., be multi-valued). However, beyond this, this simple classical elimination is not obviously warranted when we quantize the system. Classically, the ground state of the high-energy variable $\Phi_1$ is the minimum of the potential, but quantumly the ground state has its zero-point energy. Thus it is found that treating the system quantum-mechanically, the classical elimination strategy of Eq.~\eqref{eq:sol} is never valid, see Ref.~\cite{rymarz:msc}. 
We revisit this problematic circuit, and give a full systematic treatment in Appendix~\ref{sec:elim-BO}, using a Born-Oppenheimer approximation.
 
More generally, a safe method is to explicitly include any small capacitance in the circuit (to make for a capacitive spanning tree) and determine whether the quantum fluctuations of the fast, to-be-eliminated, degrees of freedom affect the dynamics of the remaining degrees of freedom.

\begin{figure}[htbp]
\centering
	\includegraphics[height=4cm]{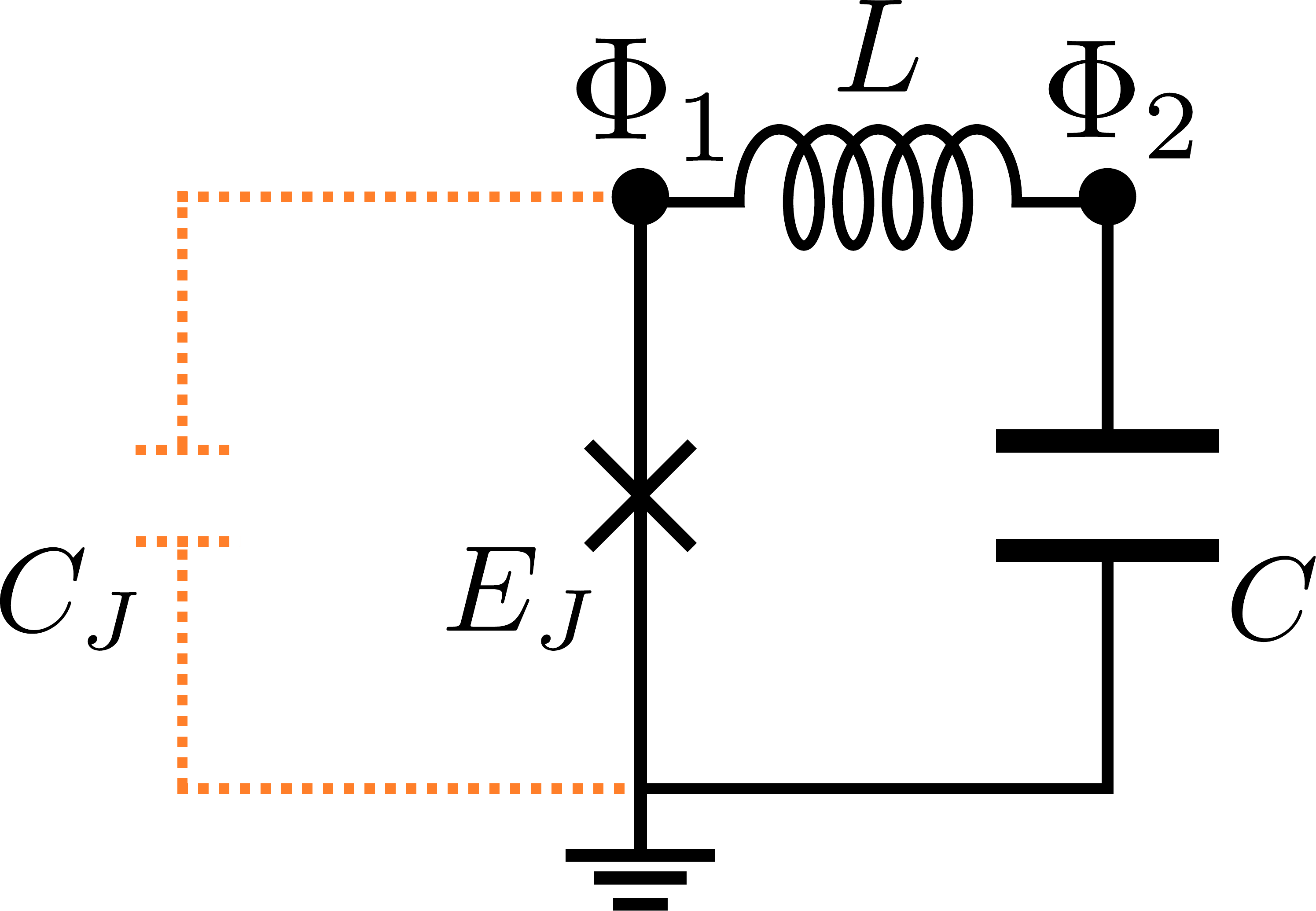}
	\caption{The spanning tree does not contain only capacitive branches, so the capacitance matrix is not invertible and a variable has to be removed by hand. How, though? }
	\label{fig:patho}
\end{figure}

\subsection{Two coupled flux qubits}

Consider the circuit in Fig.~\ref{fig:flux-qubits}, which represents two inductively coupled flux qubits in an rf SQUID configuration \cite{zagoskin_2011}. For both electrical circuit graphs $G_1$ and $G_2$ we choose a spanning tree.
In this case, these are the single capacitive branches $\mathfrak{b}=C_1$ and $\mathfrak{b}=C_2$ in orange in Fig.~\ref{fig:flux-qubits}. Thus, the kinetic energy equals
\begin{equation}
T=\frac{1}{2} C_{J_1} \dot{\Phi}^2_{C_1}+\frac{1}{2} C_{J_2} \dot{\Phi}^2_{C_2}.
\end{equation} 
Now consider all the branches which are not in the tree. No flux is threading through the fundamental loop which is made by the branch associated with the Josephson junction; hence, for those branches we have $\Phi_{J_1}=-\Phi_{C_1},\Phi_{J_2}=\Phi_{C_2}$ (given the orientations) and potential energy $U_{J}=-E_{J_1} \cos\left(2 \pi \Phi_{J_1}/\Phi_0 \right)-E_{J_2} \cos\left(2 \pi \Phi_{J_2}/\Phi_0 \right)$.  

The energy from the mutual inductances is of the form in Eq.~\eqref{eq:ul-ind} and since the branches are associated with loops which thread some external flux, we need to re-express the branch variables $\Phi_{L_1}$ and $\Phi_{L_2}$ as
\begin{equation}
\Phi_{L_1}=-\Phi_{C_1}+\Phi_{{\rm ext}}^{(1)}, \; \Phi_{L_2}=\Phi_{C_2}+\Phi_{{\rm ext}}^{(2)}.
\end{equation}
All this looks a bit heavy-handed in this simple example; since there is only one independent branch variable per circuit, we can directly replace it by the node fluxes $\Phi_{i}'=\Phi_{C_i}$, $i=1,2$ and convert the total Lagrangian to a Hamiltonian:
\begin{equation}
\hamiltonian=\sum_{i=1,2} \frac{ Q_i^2}{2C_{J_i}}-\sum_{i=1,2} E_{J_i} \cos\left(\frac{2\pi}{\Phi_0} \Phi_i'\right)
+\frac{1}{2} \begin{pmatrix}
    -\Phi_1' + \Phi_{\mathrm{ext}}^{(1)} & \Phi_2' + \Phi_{\mathrm{ext}}^{(2)}  
\end{pmatrix} \mat{M}^{-1} \begin{pmatrix}
    -\Phi_1' + \Phi_{\mathrm{ext}}^{(1)} \\ \Phi_2' + \Phi_{\mathrm{ext}}^{(2)}  
\end{pmatrix},
\end{equation}
with, cf. Eq.~\eqref{eq:ul-ind},
\begin{equation}
\mat{M}^{-1}=\frac{1}{L_1 L_2-M^2} \left(\begin{array}{cc} L_2 & -M \\ -M & L_1
\end{array}\right).
\label{eq:inverseL}
\end{equation}

\begin{figure}[htbp]
\centering
	\includegraphics[height=4 cm]{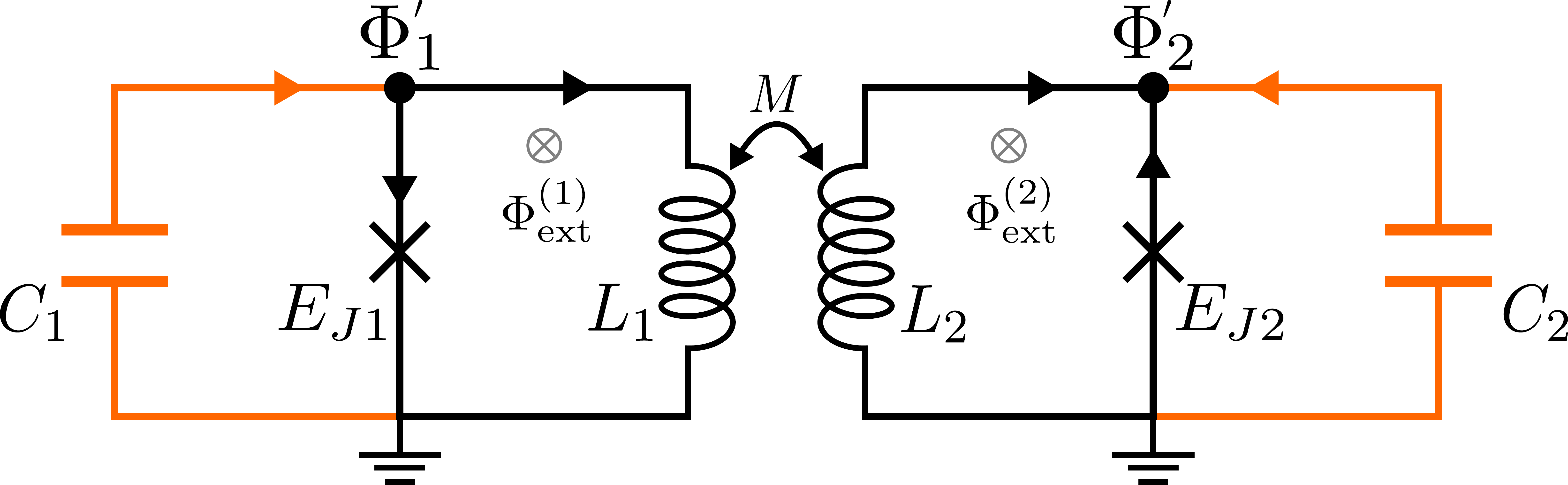}
	\caption{Two inductively coupled rf SQUID flux qubits, each representing one independent degree of freedom. The arrows indicate our choice for the orientation of the branches.}
	\label{fig:flux-qubits}
\end{figure}

Let us first neglect the mutual inductance $M$, i.e.,~take $M=0$, so that we have two uncoupled circuits. After the change of variables $\Phi_1 = -\Phi_1' + \Phi_{\rm ext}^{(1)}$ and $\Phi_2 = \Phi_2'+\Phi_{\rm ext}^{(2)}$ and setting both external fluxes to half a flux quantum $\Phi_{\rm ext}^{(2)} = \Phi_{\rm ext}^{(1)}=\Phi_0/2$, the Hamiltonian of the single flux qubits takes the form
\begin{equation}
\label{eq:fqhamiltonian}
    \mathcal{H}_{i} = \frac{Q_i^2}{2 C_i} + E_{J_i} \cos \left( \frac{2 \pi}{\Phi_0} \Phi_i \right) + \frac{\Phi_i^2}{2 L_i} 
    = 4 E_{C_i} q_i^2 + E_{J_i} \cos \phi_i + \frac{E_{L_i}}{2} \phi_i^2, \quad i=1,2,
\end{equation}
where we introduced dimensionless variables as well as charging and inductive energy as in Section~\ref{subsec:lc}. The advantage of working at half-flux quantum is that the potential is first-order insensitive to flux noise around this point due to the vanishing first derivative of the $\cos(\cdot)$ function, see also Section~\ref{subsec:fluxss}.

For large $E_J/E_L$, the cosine in the  potential $U(\phi)=E_J \cos \phi +E_L \phi^2/2$ creates double wells symmetrically around $\phi=0$ and more wells come into play for smaller $E_L$ versus $E_J$.
\begin{figure}
\centering
\begin{subfigure}[h]{0.49 \textwidth}
\centering
\includegraphics[width=7cm]{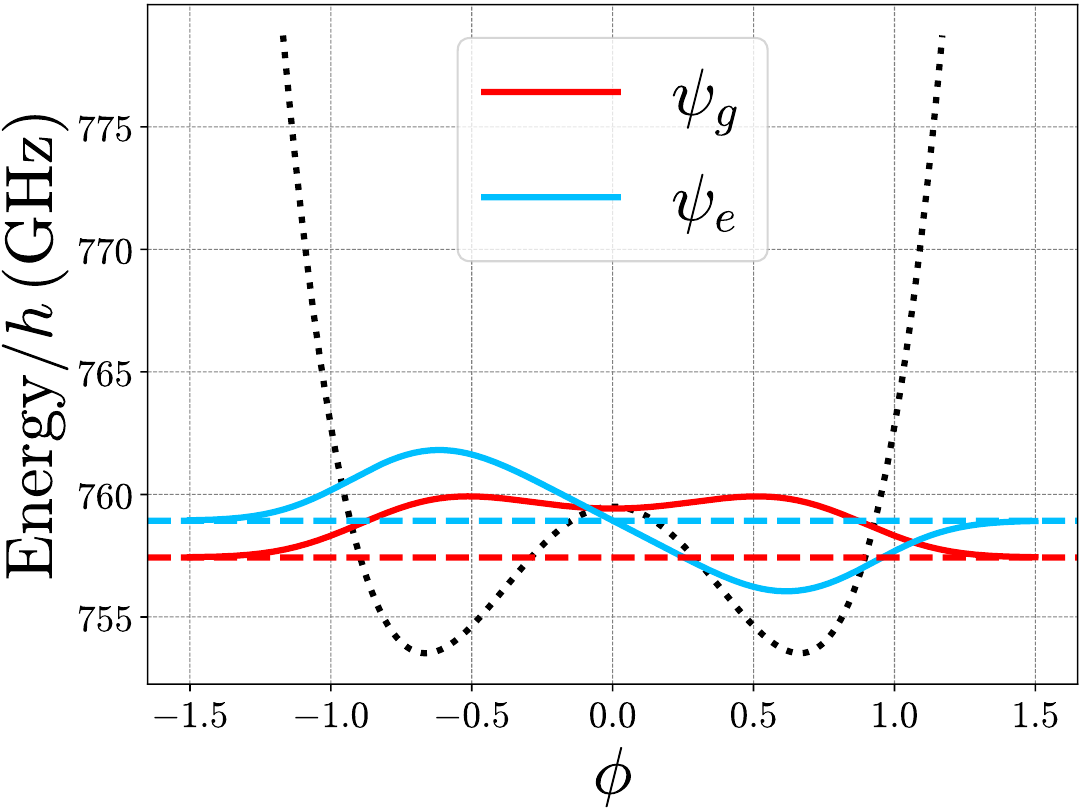}
\subcaption{}
\label{fig:ge}
\end{subfigure}
\begin{subfigure}[h]{0.49 \textwidth}
\centering
\includegraphics[width=7cm]{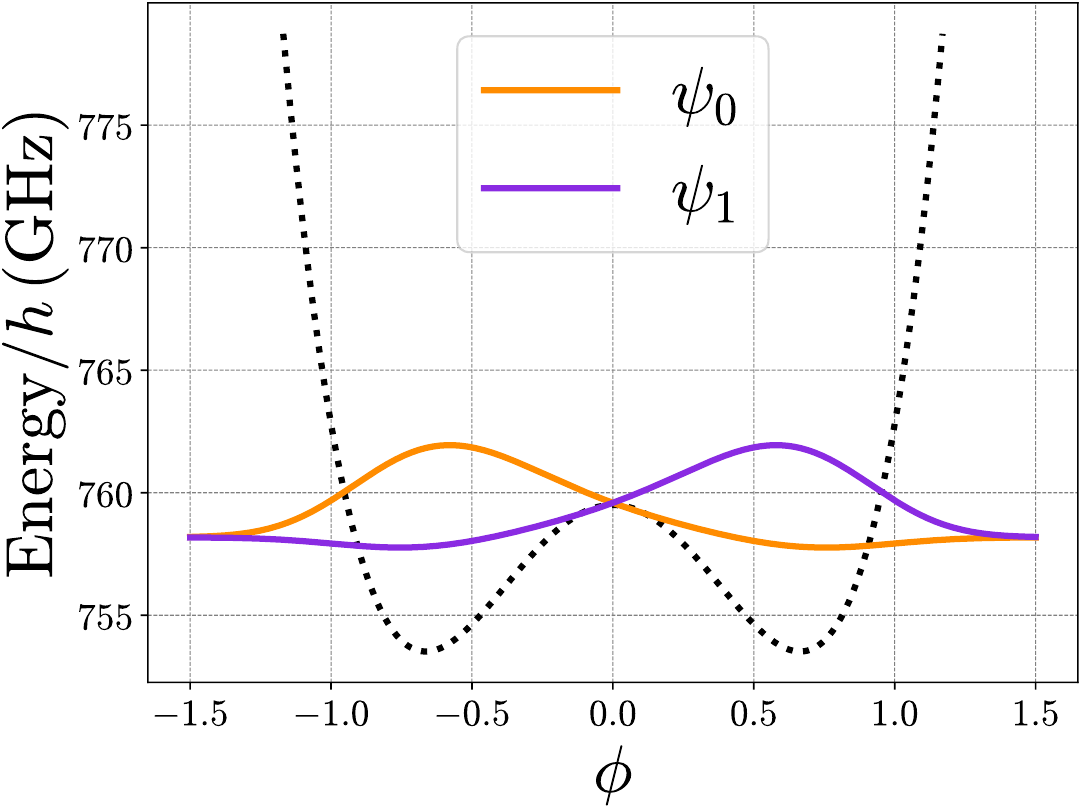}
\subcaption{}
\label{fig:cbasis}
\end{subfigure}
\caption{Flux qubit with symmetric potential. (a) Ground and first-excited wavefunctions as a function of $\phi$ for the Hamiltonian in Eq.~\eqref{eq:fqhamiltonian}, and their energies (dashed lines).
(b) Computational basis states obtained as symmetric and anti-symmetric combination of the first two eigenstates, that is, $\ket{0} = (\ket{g} + \ket{e})/\sqrt{2}$, $\ket{1} = (\ket{g} - \ket{e})/\sqrt{2}$ . The relevant parameters are taken as $E_C/h = 0.124 \mathrm{GHz}$, $E_J/E_C = 6129$, $E_J/E_L = 1.08$, which are typical parameters for rf SQUID flux qubits (see for instance Ref.~\cite{ozfidan2020}). Figure is reproduced from Ref.~\cite{cianiStoq}. }
\label{fig:fqplot}
\end{figure}
The flux qubit operates in the regime where we just have one double well. Given a finite tunnel barrier ---the relative height of the tunnel barrier is determined by $E_C/E_J$---  the lowest two energy eigenstates are symmetric and anti-symmetric superpositions of states localized at in each well, as shown in Fig.~\ref{fig:ge}. By taking linear combinations of these eigenstates, one can define a computational basis of finite flux states, see Fig.~\ref{fig:cbasis}, corresponding to a current running clockwise or counterclockwise through the SQUID loop. These states are sometimes called fluxons. When the relative height of the tunnel barrier increases, $E_C/E_J \rightarrow 0$, the energy eigenstates get more degenerate, and thus the fluxon states become (degenerate) eigenstates. 

In case $E_L$ is sufficiently small ---for this limit, see also Section~\ref{subsec:fluxonium} on the fluxonium qubit---  and more than two wells come into play, there could be multiple states inside a single well; these states are often called plasmons. Such plasmons have a characteristic `plasma eigenfrequency' if one models the well as a quadratic (harmonic) potential.

In the flux qubit case, one can alter the double well potential, e.g.~make it asymmetric, by applying additional magnetic flux through the loop closed by the Josephson branch. You can consider for yourself how this changes the potential energy, but we point out that for typical flux-qubit parameters such as those reported in the caption of Fig.~\ref{fig:fqplot}, the potential is extremely sensitive to fluctuations in the external flux.

\subsubsection{Inductive coupling}

Consider now the two coupled flux qubits operated at half a flux quantum with $M\neq 0$. The mutual inductance adds an additional coupling potential which one can calculate as 
\begin{equation}
U_{\rm coupl}(\Phi_1, \Phi_2)=-\frac{M}{L_1 L_2-M^2} \Phi_1 \Phi_2.
\end{equation}
After quantizing the system, we can project the coupling potential into the computational flux qubit basis $\ket{0_i}, \ket{1_i}$ so that its action in this subspace equals
\begin{equation}
\Pi_{\rm comp} U_{\rm coupl} \Pi_{\rm comp}=-\frac{M }{(L_1 L_2-M^2)}\bra{0_1} \hat{\Phi}_1 \ket{0_1} \bra{0_2} \hat{\Phi}_2 \ket{0_2}Z_1 Z_2.
\end{equation}
Here $\Pi_{\rm comp}$ is the projector onto the two-qubit subspace and we have used that $\bra{0_i} \hat{\Phi}_i \ket{1_i}=0$ and $\bra{1_i} \hat{\Phi}_i \ket{1_i} = -\bra{0_i} \hat{\Phi}_i \ket{0_i}$. These facts follow from the double-well symmetry, see Exercise \ref{exc:fluxq} in Section~\ref{sec:sym-prot}. We see that the mutual inductance generates an entangling $ZZ$~coupling between the flux qubits and we see that the strength of the $ZZ$ coupling depends on the expectation value of the flux operator with respect to the $\ket{0}$ state. The state for, say, qubit $1$, is characterized by a `persistent' current $\ii_1=\bra{0_1}\hat{
\Phi}_1\ket{0_1}/L_1$ and it is the strength of this current that quantifies the coupling strength. This also follows from a purely classical perspective on the coupling: it is the fact that $\ket{0}$ and $\ket{1}$ correspond to opposing currents and these currents produce opposing magnetic fields in the other flux qubit loop which causes the magnetic coupling.

 For large $E_J/E_L$, the double-well minima are approximately at $\Phi_0/2$, in which case we can write
\begin{equation}
\Pi_{\rm comp} U_{\rm coupl} \Pi_{\rm comp}\approx-\frac{M \Phi_0^2}{4(L_1 L_2-M^2)} Z_1 Z_2.
\end{equation}

Here we have discussed the simplest flux qubit, while there are other types of flux qubits which distinguish themselves by the number of junctions and the parameter regime. In particular, flux qubits were first realized not by using an inductive branch, but replacing this branch by two small Josephson junctions. This qubit will be treated in the next section, and is usually referred to as either the persistent-current flux qubit or the capacitively-shunted flux qubit \cite{Yan2016}, depending on the parameter regime. In Section~\ref{sec:sym-prot} we discuss some basic symmetries of these flux qubits. See Ref.~\cite{Krantz_2019} and Refs.~\cite{orlando1999, harris2010, younori, Yan2016} and references there in these papers to learn more about flux qubits. 

\subsection{Flux qubit: replacing the inductor by two Josephson junctions}

\begin{figure}
\centering
	\includegraphics[height=4cm]{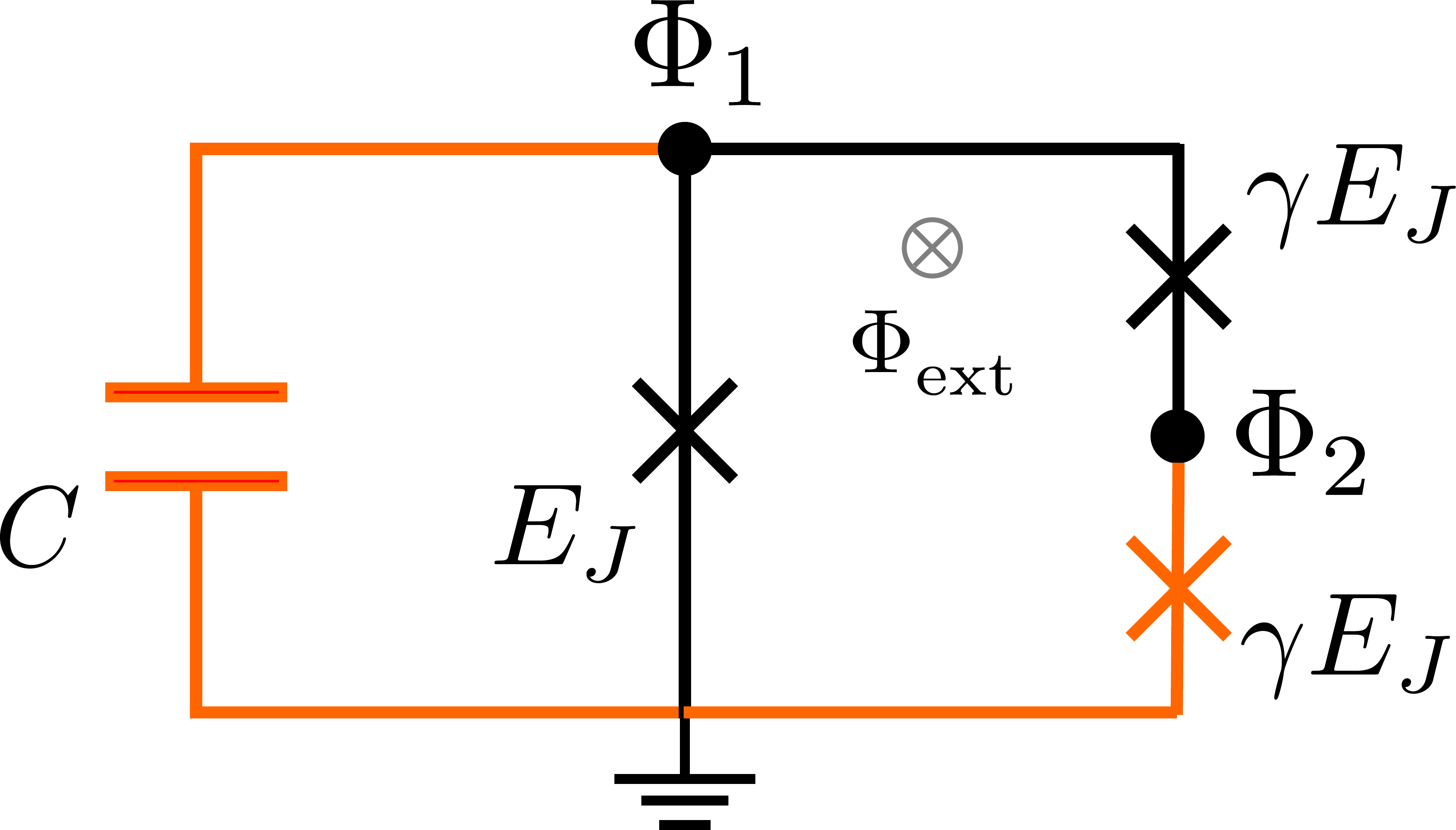}
	\caption{A flux qubit with three Josephson junctions.}
	\label{fig:flux2JJ}
\end{figure}

Flux qubits have been realized using the circuit shown in Fig.~\ref{fig:flux2JJ}. Here the two Josephson junctions in series, each with Josephson energy $\gamma E_J$, play the role of `inductive' branch. If we do not include a small capacitance in parallel with the Josephson junctions, we see that the circuit does not contain a capacitive spanning tree. In terms of the node variables $\Phi_1$ and $\Phi_2$, we have 
\begin{equation}
\lagrangian=\frac{C}{2}\dot{\Phi}_1^2+E_J \cos\biggl(\frac{2 \pi}{\Phi_0}\Phi_1 \biggr)+\gamma E_J \cos\biggl(\frac{2\pi}{\Phi_0} \Phi_2 \biggr)+\gamma E_J \cos\biggl(\frac{2\pi}{\Phi_0} (\Phi_2-\Phi_1+\Phi_{\rm ext})\biggr).
\end{equation}
As a first approach, we try to find a constraint which eliminates the variable $\Phi_2$ in $\mathcal{L}$ and get a Hamiltonian for just the remaining variable $\Phi_1$. Introducing the dimensionless variables $\phi_i$, the Euler-Lagrange equation for the variable $\phi_2$ reads 
\begin{equation}
\frac{\partial U}{\partial \phi_2}=0 \Rightarrow \sin \phi_2+\sin(\phi_2-\phi_1+\phi_{\rm ext})=0.
\label{eq:ELflux2}
\end{equation}
Note that the second derivative gives
\begin{equation*}
\frac{\partial^2 U}{\partial \phi_2^2} \propto \cos \phi_2+\cos(\phi_2-\phi_1+\phi_{\rm ext}).
\end{equation*}
There are two solutions to Eq.~\eqref{eq:ELflux2} labeled by $k=0,1$ namely $\phi_2^{k=0,1}=\frac{1}{2}(\phi_1-\phi_{\rm ext})+k \pi$. Each one is stable, i.e., $\frac{\partial^2 U}{\partial \phi_2^2}\big\vert_{\phi_2^k}\geq 0$ in a different region for $\phi_1$. Namely, the $k=0$ solution is stable when $|\phi_1-\phi_{\rm ext}|\leq \pi$ and the $k=1$ solution is stable when $|\phi_1-\phi_{\rm ext}+2\pi|\leq \pi$. Setting $\phi_2=\phi_2^{k=0}$ gives the potential $U=-E_J \cos \phi_1-2\gamma E_J \cos((\phi_1-\phi_{\rm ext})/2)$ which, defining $2\phi \equiv \phi_1-\phi_{\rm ext}$, gives rise to the one-dimensional Hamiltonian \begin{equation}
    \hamiltonian=E_C q^2-E_J \cos(2\phi+\phi_{\rm ext})-2 \gamma E_J \cos(\phi).
    \label{eq:fluxH}
\end{equation}
Note the factor of $4$ difference as compared to a standard charging energy term in Eq.~\eqref{eq:cpb_Ham}.
As long as $E_C\ll \gamma E_J$ (requiring a large shunt capacitance), $\phi$ will have small fluctuations ensuring stability of the working point as long as $|\phi|<  \pi/2$. The other solution $\phi_2^{k=1}$ leads to the same Hamiltonian with the variable change $\phi \rightarrow \phi+\pi$ and stability of the working point means that $|\phi+\pi|< \pi/2$. 

Thus, we obtain a reasonable reduction to a one-dimensional Hamiltonian, warranted in the regime $E_C \ll E_J$. But the careful reader may worry that there is some problem here, similar to the situation featured in Section~\ref{subsec:patho}. That is, it is safest to check if the result we have just derived is consistent with the scenario in which we retain a small capacitance $C_2$ connected to node $\Phi_2$. This is most properly treated using the Born-Oppenheimer procedure, which can be found in Appendix~\ref{sec:elim}. We can briefly state the physics of the situation: the capacitance $C_2$ should not be too large; in particular we require $C_2 \ll C$. Otherwise, the system would have two fully independent degrees of freedom. At the same time, the Euler-Lagrange procedure that we just outlined, requires the variable $\phi_2$ to behave classically, in the sense that its zero-point fluctuations, $\phi_{2,\text{zpf}}$ should be small,
$\phi_{2, \text{zpf}} \ll1$. Using Eq.~\eqref{eq:zpf-reduced}, we find 
\begin{equation}
\phi_{2,\text{zpf}}=
\biggl(\frac{2 E_{C_2}}{\gamma E_J}\biggr)^{1/4}. 
\label{eq:phi2}
\end{equation}
To satisfy our constraint, $C_2$ should actually not be too small; even for $C_2=C/10$, the condition $\phi_{2,\text{zpf}}\ll 1$ is not so well satisfied. However, the consequences of this are not too serious; if $\phi_{2,\text{zpf}}\sim 1$, we find \cite{rymarz:sing} that Eq.~\eqref{eq:fluxH} is very nearly correct, except that the last term becomes $-2\gamma\tilde{E}_J \cos(\phi)$, where the coefficient is somewhat reduced compared to its ``bare" value: $\tilde{E}_J\lesssim E_J$. We say that the effective Josephson energy in this term is ``renormalized by quantum fluctuations'' of $\phi_2$. A problem really arises if $C_2$ is very small, since then our coefficient $\tilde{E}_J$ really becomes zero, and all contributions (of the branch with two Josephson junctions) to the circuit dynamics are entirely lost. 

A final word about the Hamiltonian in Eq.~\eqref{eq:fluxH}: for the operating point $\phi_{\rm ext}=\pi$, Eq.~\eqref{eq:fluxH} has a single well when $\gamma \geq 2$ (and double wells when $\gamma < 2$) and then represents the so-called (anharmonic) C-shunted flux qubit. 

A similar, but stronger, reduction in the number of variables is performed to describe the fluxonium qubit in the next section.

\subsection{The fluxonium qubit}
\label{subsec:fluxonium}

\begin{figure}[htbp]
\centering
	\includegraphics[width=8cm]{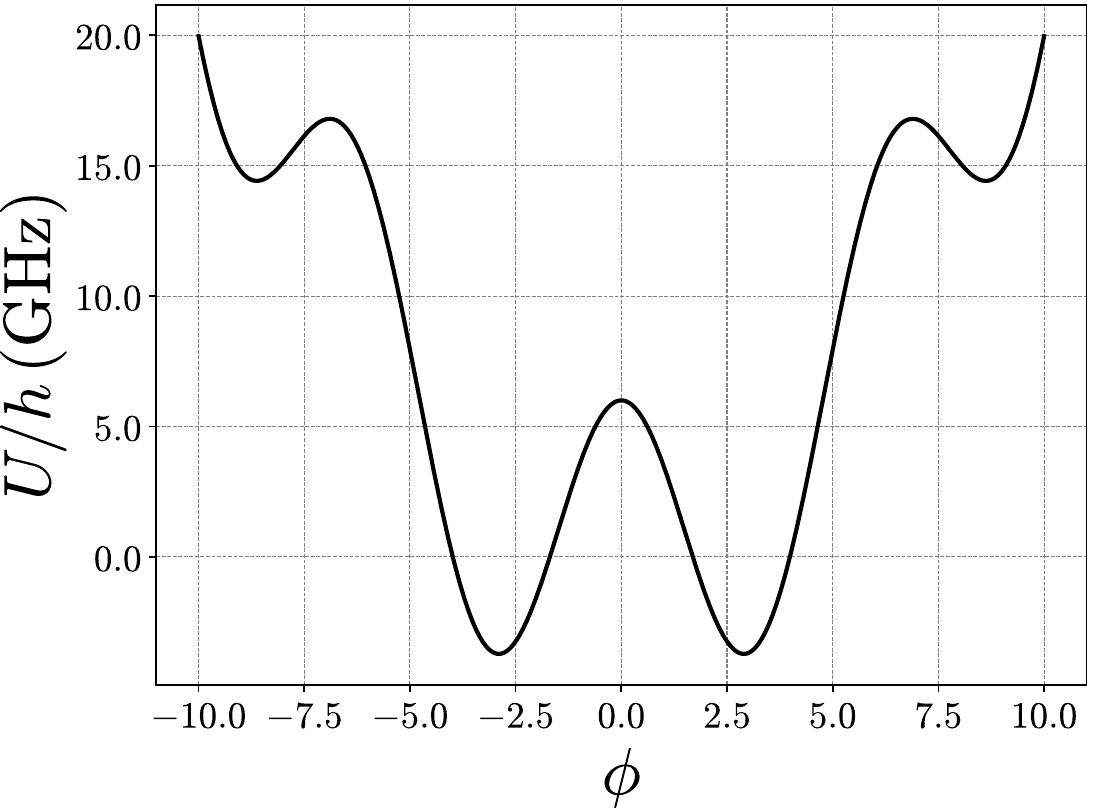}
	\caption{Example of a fluxonium potential $U = - E_J \cos(\phi + \phi_{\mathrm{ext}}) + E_L \phi^2/2$ as a function of the reduced (dimensionless) flux $\phi$ for typical fluxonium parameters $E_J/h = 6.0 \, \mathrm{GHz}$, $E_L/h=0.5 \, \mathrm{GHz}$, $\phi_{\mathrm{ext}} = \pi$.}
	\label{fig:fluxonium_potential}
\end{figure}

Let us again consider an inductively shunted CPB as in Fig.~\ref{fig:iscpb} on the right. Its basic Hamiltonian can be written (in terms of dimensionless variables) as
\begin{equation}
\label{eq:hfluxonium-start}
\hamiltonian = 4 E_C q^2 - E_J \cos(\phi) + \frac{E_L}{2} (\phi-\phi_{\mathrm{ext}})^2,
\end{equation} 
or, with a change of variables
\begin{equation}
\label{eq:hfluxonium}
\hamiltonian = 4 E_C q^2 - E_J \cos(\phi+\phi_{\rm ext}) + \frac{E_L}{2} \phi^2.
\end{equation} 

There are three energy scales at play, namely the charging energy $E_C$ defined in Eq.~\eqref{eq:en-cap}, the Josephson energy $E_J$, and the inductive energy $E_L$ defined in Eq.~\eqref{eq:ind_en}. We can ask what happens when we vary these energy scales. In particular, the fluxonium regime is defined as the regime where $E_L/E_J \ll 1$ while $1 \lesssim E_J/E_C \lesssim 10$ but not higher~\cite{nguyen2019}. 

In this regime, the small parabolic potential due to $E_L$ is not very confining and multiple wells are present due to the Josephson potential. An example of the fluxonium potential for typical parameters is shown in Fig.~\ref{fig:fluxonium_potential}. In order to work in the fluxonium regime, we need to have a large effective inductance $L=L_{\mathrm{eff}}$, which should not be accompanied by an additional large effective capacitance $C_{\mathrm{eff}}$ in parallel with it. This is because we would like the wavefunctions to be delocalized in the flux degree of freedom. A spurious, large shunting capacitance would instead have the tendency to localize the wavefunctions and make the system, in each well, only weakly anharmonic, similar to the transmon qubit, which we treat extensively in Chapter~\ref{chap:transmon} \cite{nguyen2019}. In other words, in the circuit the Josephson junction should have an element with a large impedance $Z=\sqrt{L_{\rm eff}/C_{\rm eff}}$ in parallel with it and we call this element a `superinductance'. Ideally, the dimensionless quantity $Z/R_Q \gtrsim 1$ for this element\footnote{In current designs \cite{nguyen2019} fluxonium qubits have $Z \ge 1 \, \mathrm{k\Ohm}$, and thus $Z \lesssim R_Q$.}, where $R_Q$ is the resistance quantum 
\begin{equation}
  R_Q = \frac{h}{(2 e)^2} \approx 6.5 \, \mathrm{k\Omega}.
  \label{eq:resist-quant}
\end{equation}  
Physically, it is nontrivial to obtain a superinductance, in particular when the inductance is solely due to the (geometric) self-inductance of the material: an inductance which grows with the size of a piece of superconducting material (like the number of coils) also comes with a large capacitance, reducing the impedance $Z$.

A solution is to use a material with a large kinetic inductance such as granular aluminum to get an effective superinductance~\cite{Maleeva2018}. Another option is to engineer an effective inductance by using an array of $N$ ($N \gg 1$) Josephson junctions in series. In fact, one can also model a material such as granular aluminum (and the origin of its kinetic inductance) as a random array of a macroscopic number of superconducting islands connected by Josephson junctions \cite{Maleeva2018}.

\begin{figure}[htbp]
\centering
\includegraphics[width=5cm]{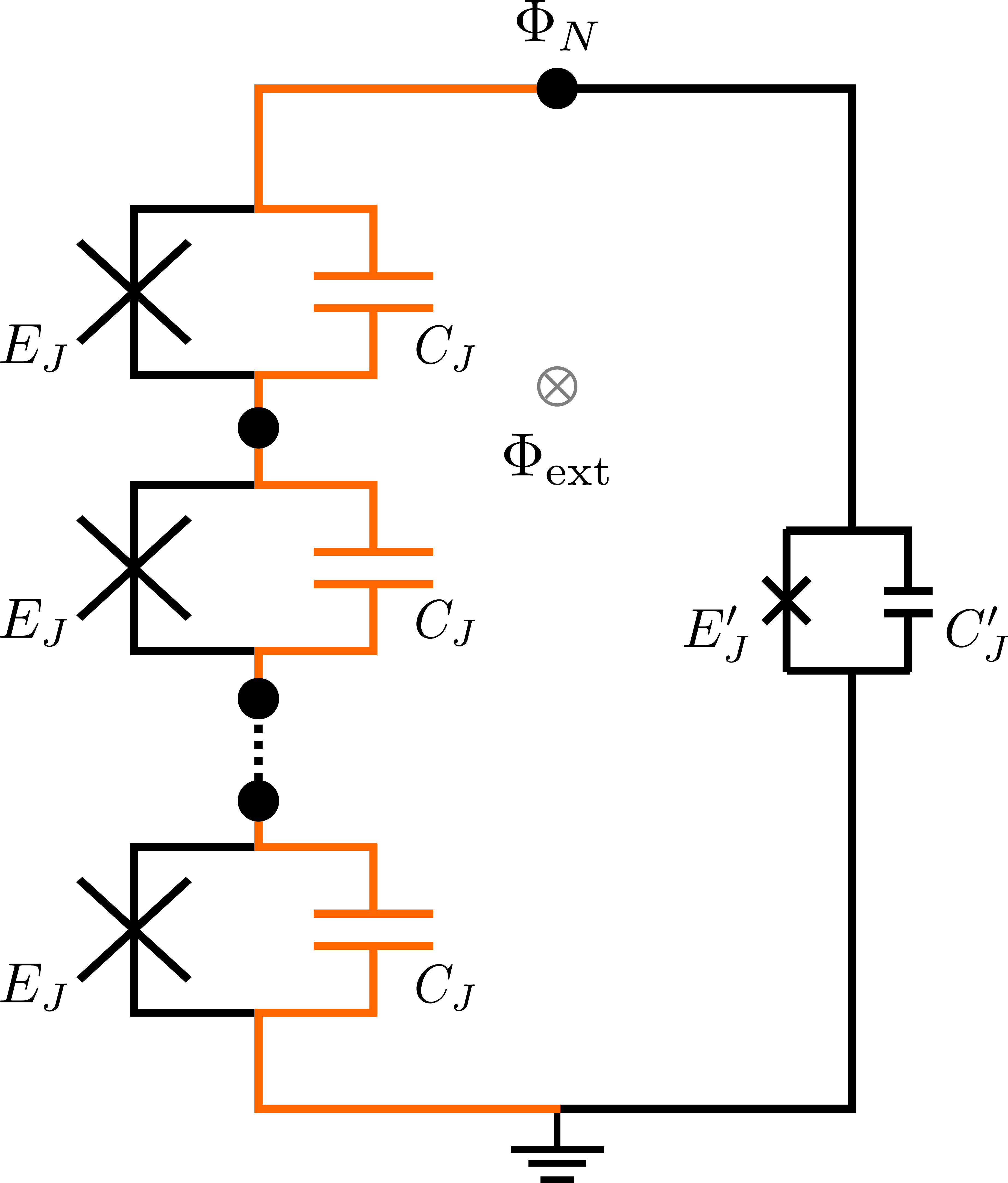}
\caption{An array of $N$ identical Josephson junctions in series, in parallel with a single junction with $E_J \gg E'_{J}$. It provides an effective large inductance used in the fluxonium qubit. See Fig.~\ref{fig:fluxonium-fab}(b) for a fabricated array.}
\label{fig:fluxonium}
\end{figure}

\begin{figure}[htbp]
\centering
	\includegraphics[width=5cm]{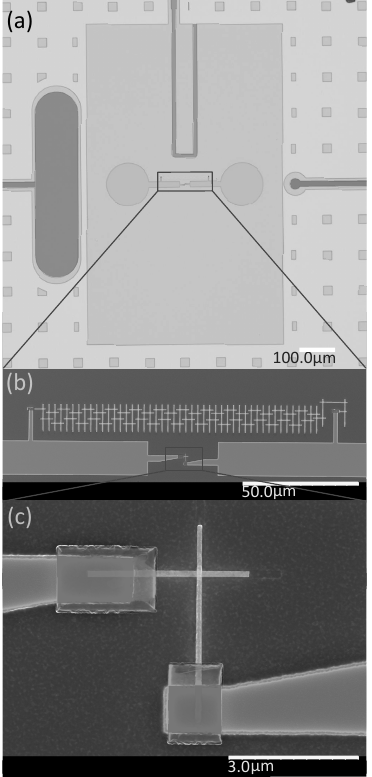}
	\caption{Optical micrograph of a chip with a fluxonium qubit (Andersen Lab, Delft University of Technology, 2023; figure courtesy of Figen Yilmaz). (a) The round pads provide the larger shunting capacitance. On the left (stadium-shaped pad) is the capacitive coupling to a readout resonator, on the top one sees a flux line for controlling the amount of magnetic flux through the loop, and on the right one sees a `charge' line for driving the fluxonium qubit. (b) The long array of junctions consists of 100 junctions, each 410 x 410 nm in size, in parallel with a single junction. (c) Zoom-in of the single junction.}
	\label{fig:fluxonium-fab}
\end{figure}

The electrical circuit is given in Fig.~\ref{fig:fluxonium} and a fabricated fluxonium qubit is shown in Fig.~\ref{fig:fluxonium-fab}. To analyze the circuit, we use the tree with orange branches as indicated in the circuit with nodes $1, 2, \ldots, N$. We denote the capacitive branch flux across the $i$-th junction in the array by $\Phi_{C_i}$ . The Lagrangian in terms of these variables reads
\begin{equation}
T=\frac{1}{2}C_J  \sum_{i=1}^{N}  \dot{\Phi}_{C_i}^2+\frac{1}{2}C_{J'}\dot{\Phi}_{N}^2,
\end{equation}
with $\Phi_{N}=\sum_{i=1}^N \Phi_{C_i}$. The potential energy equals
 \begin{equation}
 U=-E'_{J} \cos\left(\frac{2\pi}{\Phi_0} \Phi_{\mathfrak{b}} \right)-E_J \sum_{i=1}^{N}
\cos\left(\frac{2 \pi}{\Phi_0} \Phi_{C_i}\right),
\end{equation}
with $\Phi_{\mathfrak{b}}=-\Phi_{N}+\Phi_{\rm ext}$. We assume that the branch fluxes in the array all represent `heavy' (large `mass' $C_J$) degrees of freedom with $\phi_{\rm zpf}< 1$. Mathematically, this means that we assume $E_J/E_{C_{J}} > 1$, and we approximate the potential (modulo a dropped constant) as
\begin{equation}
U \approx U_{\rm approx}=-E'_{J} \cos\left(\frac{2\pi}{\Phi_0} (\Phi_{N}-\Phi_{\rm ext}) \right)+\frac{1}{2 L_J}\sum_{i=1}^N\Phi_{C_i}^2,
\end{equation}
with Josephson inductance defined in Eq.~\eqref{eq:jos_ind}. 
Each potential $\frac{\Phi_{C_i}^2}{2 L_J}$ with small $L_J$ provides a tight confinement for the branch flux $\Phi_{C_i}$. Thus, given that $\Phi_{N}=\sum_{i=1}^N \Phi_{C_i}$ while $\Phi_{N}$ is favored to take the value $\Phi_{\rm ext}$ due to the parallel junction with scale $E'_{J}$, it is natural to consider a low-energy approximation in which each variable $\Phi_{C_i}=\Phi_{N}/N$ has no further independent dynamics by itself. This corresponds to a local minimum of the energy of the Josephson-junction array, but there are other minima as well, see Exercise \ref{exc:jjarray}. We refer the reader to Ref.~\cite{ManucharyanPhd} for an excellent discussion of Josephson junction arrays.

In this approximation the potential energy reads
\begin{equation}
U_{\rm approx}=-E'_{J} \cos\left(\frac{2\pi}{\Phi_0} (\Phi_{N}-\Phi_{\rm ext})\right)+\frac{1}{2 N L_J}\Phi_{N}^2, 
\end{equation}
where we see that one has obtained a large effective inductance $L_{\rm eff} \equiv N L_J$. In this approximation the kinetic energy equals
\begin{equation}
T=\frac{1}{2}\left(C'_{J} +\frac{C_J}{N}  \right)\dot{\Phi}_{N}^2 \rightarrow  \frac{Q_N^2}{2(C'_{J}+C_J/N)}.
\end{equation} 
Thus, we see that the array adds an effective capacitance $C_{\mathrm{eff}} \approx C_J/N$ in parallel to the capacitance of the small junction $C'_{J}$. Note that $C_{J}/N$ corresponds to the equivalent capacitance of $N$ capacitances $C_J$ in series.
 The effective characteristic impedance of the array is then $Z_{\mathrm{array}} = \sqrt{L_{\mathrm{eff}}/C_{\mathrm{eff}}}= N \sqrt{L_J/C_J}$ and it increases linearly with the number of junctions, making it possible to reach the superinductance regime. 

 \begin{Exercise}[title={Josephson junction array}, label=exc:jjarray]
\begin{figure}[htb]
\centering
\includegraphics[scale=0.5]{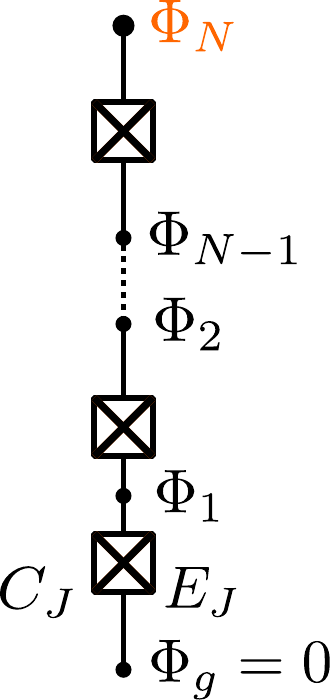}
\caption{Array of $N$ Josephson junctions.}
\label{fig:jjarray}
\end{figure}

Consider the circuit in Fig.~\ref{fig:jjarray} representing an array of $N$ equal Josephson junctions, each with capacitance $C_J$ and Josephson energy $E_J$. Arrays of junctions are used to realize the large inductance needed for the fluxonium qubit. In this exercise we will obtain an effective Lagrangian representing the circuit, which will be a function of the dimensionless node flux at the end of the array $\phi_N = 2 \pi \Phi_N/\Phi_0$ only, as sketched by the arguments above.
\Question First of all, write down the Lagrangian in terms of dimensionless node fluxes $\phi_k = 2 \pi \Phi_k /\Phi_0$ ($k=1, \dots, N$). The problem is more easily understood by introducing the dimensionless branch fluxes across each junction $\theta_k= \phi_{k+1}-\phi_k$ and writing the Lagrangian in terms of these new variables. For compactness we will write $\vec{\theta}=\{\theta_1, \dots, \theta_N\}$.
\Question We now assume that $\phi_N$ is set externally and enters as a parameter in the problem. In terms of the variables $\theta_k$ this translates into the constraint: $\phi_N = \sum_{k=1}^N \theta_k$. 
In the limit of $E_J/E_{C_J} \gg 1$ we expect that {\em classically} the system is found in one of the (metastable) minima of the potential. Formulate the problem as a constrained optimization and confirm that the point 
\begin{equation}
\label{eq:uniformPhase}
\vec{\theta}^0\colon \forall k,\;\theta_k^{0}= \frac{\phi_N}{N},
\end{equation}
and the points labeled by $\pm,l=1, 2, \dots, N$: 
\begin{equation}\label{eq:otherphases}
\vec{\theta}^{\pm,l}\colon \theta_l^{\pm 1}= \frac{\phi_N \mp 2 \pi}{N} \pm 2 \pi,\quad \theta_{k\neq l}^{\pm 1}= \frac{\phi_N \mp 2\pi}{N},
\end{equation}
are extrema of the potential. Notice that the configurations labeled by different $l$ are indistinguishable as they lead to identical contributions in the Lagrangian, and so we may label them only by $\pm$. 
\emph{Comment: These are not the only extrema of the potential, as we show in the answer.}
\Question Show that $\vec{\theta}^0$ is a minimum of the potential when $\phi_N/N$ is sufficiently small.
\Question Quantum-mechanically, we expect to have localized states at such a minimum $\vec{\theta}^0$. In which parameter limit do you expect the tunnelling from one minimum to another to be suppressed? 
\Question We now promote $\phi_N$ to a dynamical variable. Let's assume that the system initially is in the minimum given by Eq.~\eqref{eq:uniformPhase}. If the variable $\phi_N$ varies relatively slowly, we can assume that the system follows the minimum adiabatically. Within this approximation, obtain an effective Lagrangian for the system in terms of the variable $\phi_N(t)$. Why does this system behave as a large inductance in the large $N$ limit? 
\end{Exercise}

\begin{Answer}[ref={exc:jjarray}]
\Question The Lagrangian in terms of dimensionless node fluxes reads 
\begin{equation}
\mathcal{L}(\phi_1, \dots, \phi_{N}; \dot{\phi}_1, \dots, \dot{\phi}_{N})= \sum_{k=0}^{N-1} \left(\frac{\hbar^2}{16 E_{C_J}} (\dot{\phi}_{k+1}-\dot{\phi}_k)^2+E_J \cos(\phi_{k+1}-\phi_{k})\right),
\end{equation}
where $\phi_0=0=\dot{\phi}_0$, or in terms of the newly-defined phases:
\begin{equation}
\label{eq:lagrArray}
\mathcal{L}(\theta_1, \dots, \theta_{N}; \dot{\theta}_1, \dots, \dot{\theta}_{N})= \sum_{k=1}^{N} \left(\frac{\hbar^2}{16 E_{C_J}} \dot{\theta}_k^2+E_J \cos \theta_k \right).
\end{equation}
We also identify the potential
\begin{equation}
U(\theta_1, \dots, \theta_{N})= -\sum_{k=1}^{N} E_J \cos \theta_k.
\end{equation}
\Question We formulate the problem as a constrained optimization problem using the method of Lagrange multipliers where the constraint is given by
\begin{equation}
g(\vec{\theta})= \sum_{k=1}^N \theta_k-\phi_N=0.
\end{equation}
We introduce the Lagrange function (not to be confused with the Lagrangian of the system):
\begin{equation}
\Lambda(\vec{\theta}, \lambda)=  \lambda g(\vec{\theta})+U(\vec{\theta}), 
\end{equation} 
with Lagrange multiplier $\lambda$. Thus we just have to check that the gradient of the Lagrange function is zero in Eqs.~\eqref{eq:uniformPhase} and \eqref{eq:otherphases}. Setting the gradient to zero we get the following equations
\begin{subequations}
\begin{equation}
\label{eq:zeroGrad1}
\frac{\partial \Lambda}{\partial \theta_k}=0 \implies E_J \sin \theta_k + \lambda=0,
\end{equation}
\begin{equation}
\label{eq:zeroGrad2}
\frac{\partial \Lambda}{\partial \lambda}=0 \implies \phi_N - \sum_{k=1}^N \theta_k=0.
\end{equation}
\label{eq:zeroGrad}
\end{subequations}

Plugging Eq.~\eqref{eq:uniformPhase} in Eqs.~\eqref{eq:zeroGrad2} indeed gives an extremum when $\lambda=-E_J \sin (\phi_N/N)$.~Plugging Eq.~\eqref{eq:otherphases} into Eq.~\eqref{eq:zeroGrad1} gives
\begin{equation}
E_J \sin \biggl( \frac{\phi_N \mp 2 \pi}{N} \pm 2 \pi \delta_{kl}\biggr)+\lambda=0 \implies E_J \sin \biggl(\frac{\phi_N \mp 2 \pi}{N}\biggr)+\lambda=0,
\end{equation}
which is satisfied if $\lambda= -E_J \sin [ (\phi_N \mp 2 \pi))N)]$. Additionally, Eq.~\eqref{eq:zeroGrad2} is easily checked:
\begin{equation}
\phi_N - \sum_{k=1}^N \left( \frac{\phi_N \mp 2 \pi}{N} \pm 2 \pi \delta_{kl}\right)= \phi_N-\phi_N \mp 2 \pi \pm 2 \pi=0. 
\end{equation}
A fully general solution of Eqs.~\eqref{eq:zeroGrad} is obtained by picking for each $k$ either
\begin{equation}
    \theta_k=\alpha+2\pi m_k \mbox{ or } \theta_k=\pi-\alpha+2\pi m_k,
\end{equation}
and determining the choice of $\alpha$ which satisfies Eq.~\eqref{eq:zeroGrad2}. Assume we pick $\theta_k=\alpha+2\pi m_k$ for a set $S$ $\theta_k$s, and the other solution for the rest. This leads to
\begin{equation}
    \alpha=\frac{\phi_N}{2|S|-N}-\pi \frac{2m+N-|S|}{2|S|-N},
\end{equation}
with $m = \sum_{i=1}^N m_k$. Hence a fully general solution can be labeled by $\vec{m}=(m_1,m_2,\ldots,m_N)$ with $m_i \in \mathbb{Z}$ and a set $S\subseteq [N]$:
\begin{align}
\label{eq:phasSlipSol}
\vec{\theta}^{\vec{m},S}\colon \begin{cases}
\theta_{k\in S} = \frac{\phi_N-2 \pi m+\pi(|S|-N)}{2|S|-N}+2 \pi m_k ,  \\
\theta_{k\notin S} = \frac{-\phi_N+2 \pi m+\pi|S|}{N}+2 \pi m_k. 
\end{cases}
\end{align}
 Note that these configurations are in principle only distinguished by the total number $m \in \mathbb{Z}$ and $|S|$, since the Lagrangian in Eq.~\eqref{eq:lagrArray} does not have any dependence on the particular $m_k$ (as long as these numbers have no dynamics) nor $S$ when the junctions are all identical.
When $|S|=N$, Eq.~\eqref{eq:phasSlipSol} simplifies to
\begin{align}
\label{eq:phasSlipSolsimp}
\vec{\theta}^{\vec{m}}\colon 
\theta_{k} = \frac{\phi_N-2 \pi m}{N}+2 \pi m_k.  
\end{align}
This represents a configuration with $m$ phase slips in total, with junction $k$ having a $m_k\in \mathbb{Z}$ phase slips. 
\Question We have to verify whether the $N-1 \times N-1$ Hessian matrix $H\geq 0$, with matrix elements $i,j=1,\ldots N-1$:
\begin{align}
H_{ij}/E_J=\frac{1}{E_J}\frac{\partial^2 U}{\partial \phi_j \partial \phi_i}=-\frac{\partial}{\partial \phi_j}\left(\sin(\phi_{i+1}-\phi_i)-\sin(\phi_i-\phi_{i-1})\right)= \notag \\ -\delta_{j,i+1}\cos(\phi_{i+1}-\phi_i)+\delta_{ij}(\cos(\phi_{i+1}-\phi_i)+\cos(\phi_i-\phi_{i-1}])-\delta_{j,i-1}\cos(\phi_i-\phi_{i-1}]),\notag
\end{align}
where $\phi_0=0$ and $\phi_N$ is given. At the point $\vec{\theta}^0$ this gives
\begin{align}
    H_{ij}/E_J=-\cos(\phi_N/N)(\delta_{j,i+1}-2\delta_{ij}+\delta_{j,i-1}) \approx -\delta_{j,i+1}+2\delta_{ij}-\delta_{j,i-1}+O\left(\frac{\phi_N^2}{N^2}\right),\notag 
\end{align}
which represents a positive-semidefinite matrix.
 \Question We expect the tunneling to be suppressed in the limit in which no tunneling can exist, i.e. in the classical limit $E_J/E_{C_J} \gg 1$ or $\phi_{\rm zpf}$ in Eq.~\eqref{eq:zpf-reduced}, where $E_L$, taken as $E_{L_J}$, is small. The suppression is indeed exponential in this ratio $\sim \exp(-\sqrt{8 E_J/E_{C_J}})$.
\Question Within the assumption explained in the text, this approximation just assumes that the relation in Eq.~\eqref{eq:uniformPhase} is always satisfied even when the variable $\phi_N$ varies with time. Note that this is a low-energy approximation and quantum-mechanically it crucially relies on the fact that tunneling between different minima is suppressed, as explained in the previous point. The effective Lagrangian in terms of the phase across the array is obtained by plugging Eq.~\eqref{eq:uniformPhase} into Eq.~\eqref{eq:lagrArray} to get
\begin{equation}
\label{eq:effLagr}
\mathcal{L}_{\mathrm{eff}}= \frac{\hbar^2 }{16 E_{C} N} \dot{\phi}^2 +N E_J \cos \frac{\phi}{N}, 
\end{equation}
where for simplicity we simply called $\phi_N= \phi$. From Eq.~\eqref{eq:effLagr} we notice that in terms of the collective variable $\phi_N$, the array provides an effective capacitance $C_J/N$, which is just the series of all the capacitances and an effective potential
\begin{equation}
U_{\mathrm{eff}}= -N E_J \cos \frac{\phi}{N}.
\end{equation}
In order to understand why we expect this to behave like a large inductance, we just expand the previous equation up to second order. Neglecting the constant term
\begin{equation}
U_{\mathrm{eff}} \approx  \frac{E_J \phi^2}{2  N}.
\end{equation}
So while the equivalent capacitance goes down as $1/N$, the equivalent inductance increases proportionally to $N$. 
\end{Answer}

\subsection{Circuits for two Cooper pair tunneling:  \texorpdfstring{$0$-$\pi$}{TEXT} qubit}
\label{subsec:0pi}

The idea of a $0$-$\pi$ qubit is to effectively engineer a Hamiltonian of the approximate form 
\begin{equation}
H=4E_C \hat{q}^2-E_J^{\rm eff} \cos(2\hat{\phi}),
\label{eq:H0pi}
\end{equation}
so that the potential term corresponds to the tunneling of `two Cooper pairs'. As we will discuss in Section~\ref{subsec:fp}, the term $\cos \hat{\phi}=\frac{1}{2}(e^{i\hat{\phi}}+e^{-i \hat{\phi}})$ can be viewed as describing the tunneling of a single Cooper pair from one island to the other and vice versa. 

When $E_C \ll E_J^{\rm eff}$, the Hamiltonian has two approximately degenerate ground states with wavefunctions $\psi_{0,1}(\phi)$ localized at $\phi=0$ and $\phi=\pi$, where they attain their minimal potential energy. Since these wavefunctions $\psi_{0,1}(\phi)$, being localized at very different values of $\phi$, have little overlap, there is an intrinsic protection against noise, see the discussion in Section~\ref{sec:bip}. Many features of the $0$-$\pi$ qubit are discussed in Ref.~\cite{Paolo_2019}. 
 
The derivation of the effective potential in Eq.~\eqref{eq:H0pi} starts by writing a full Hamiltonian and then, as in other examples, reducing the number of degrees of freedom. 

\begin{figure}[htbp]
\centering
	\includegraphics[width=9cm]{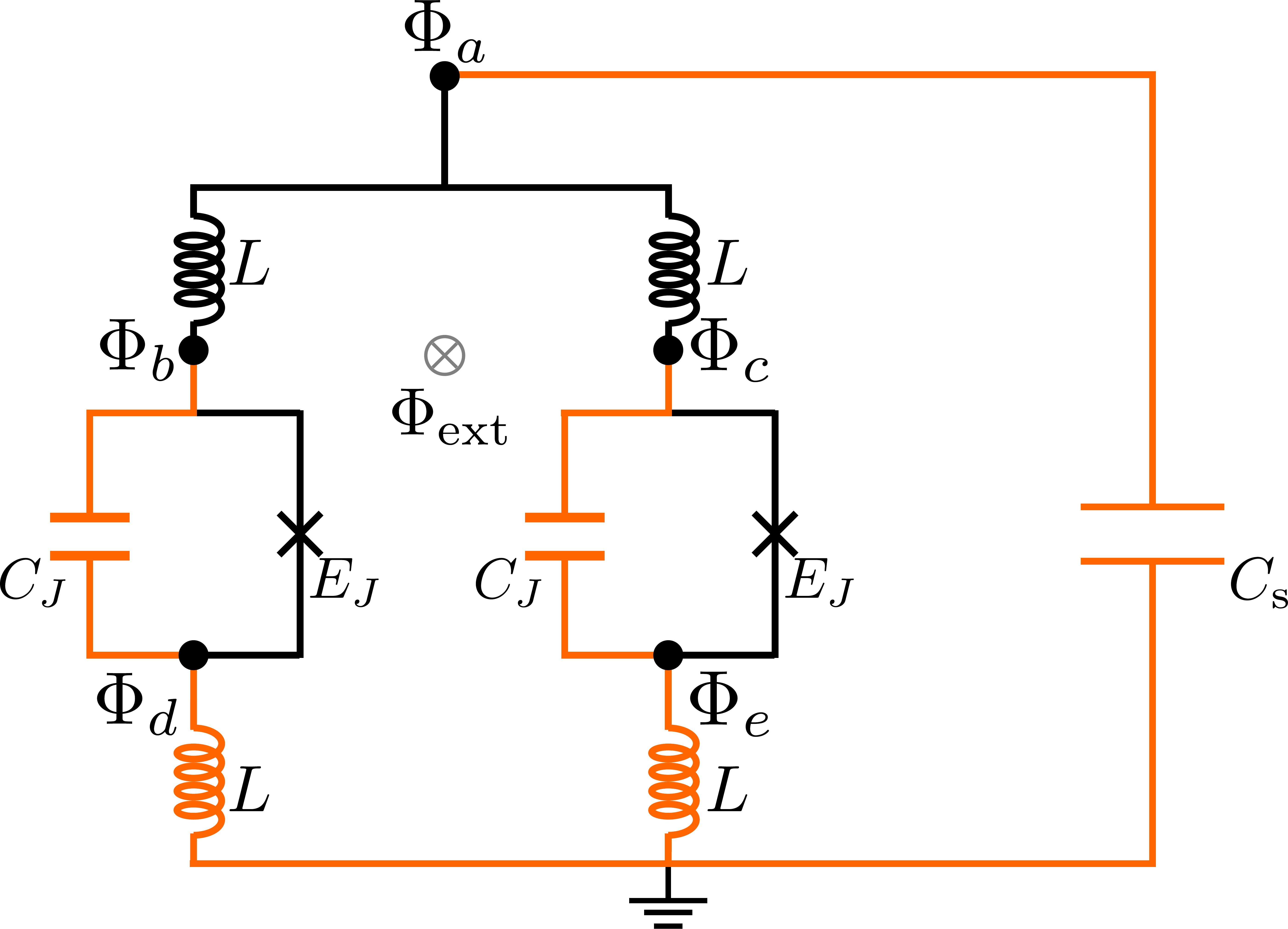}
	\caption{The superconducting circuit analyzed in Ref.~\cite{smith2019}: a spanning tree is drawn in orange.}
	\label{fig:0pi}
\end{figure}

Here we first discuss an example of a circuit designed to craft a $\cos(2\phi)$ potential which was proposed in Ref.~\cite{smith2019}.
The circuit is shown in Fig.~\ref{fig:0pi}. We will not focus on the particular interesting features of this two-Cooper-pair qubit here, but just look at the technical steps of getting the Hamiltonian. Exercise \ref{exc:0pi} will analyze the standard circuit for the $0$-$\pi$ qubit introduced in Ref.~\cite{brooks2013}.

We note that for the circuit in Fig.~\ref{fig:0pi}, there is no spanning tree in the graph where all branches are capacitive (required for Proposition \ref{lem:captree}), and hence we will work with independent branch-flux variables. Using the spanning tree highlighted in orange in Fig.~\ref{fig:0pi}, we see that there are five independent node fluxes, $\Phi_a, \ldots, \Phi_e$. The capacitance matrix is not invertible \footnote{In this case, adding a very small capacitance between, e.g., nodes $a$ and $b$, to make the capacitance matrix invertible will do nothing pathological; it will only add an innocuous high-frequency degree of freedom.} in terms of these variables, so we use the following branch variables

\begin{equation}
\Phi_1=\Phi_b-\Phi_d, \quad \Phi_2=\Phi_c-\Phi_e, \quad \Phi_3=\Phi_a, \quad \Phi_4=\Phi_d, \quad \Phi_5=\Phi_e.
\end{equation}
In terms of these variables, the Lagrangian is
\begin{multline}
\lagrangian =\frac{C_s}{2} \dot{\Phi}_3^2+ \frac{C_J}{2}(\dot{\Phi}_1^2+\dot{\Phi}_2^2)+E_J \cos \biggl(\frac{2\pi }{\Phi_0} \Phi_1 \biggr)+E_J \cos \biggl(\frac{2\pi \Phi_2}{\Phi_0} \Phi_2 \biggr)  \\
-\frac{1}{2L}[\Phi_4^2+\Phi_5^2+(\Phi_3-\Phi_1-\Phi_4+\Phi_{\rm ext})^2+(\Phi_3-\Phi_5-\Phi_2)^2].
\end{multline}

The variables $\Phi_4$ and $\Phi_5$ have no dynamics (their kinetic energy is zero); hence, we set these variables to the minimum of their potential energy (a stable extremum), i.e.,~$\frac{\partial \lagrangian}{\partial\Phi_4}=0$ and $\frac{\partial \lagrangian}{\partial\Phi_5}=0$. This leads to  
\begin{equation}
\Phi_4=\frac{1}{2}(\Phi_3-\Phi_1+\Phi_{\rm ext}), \quad \Phi_5=\frac{1}{2}(\Phi_3-\Phi_2),
\end{equation}
which can be put back into $\lagrangian$ to remove the dependence on these variables. The remaining three variables are again linearly transformed to new independent and rescaled variables
\begin{equation}
\phi =\frac{2 \pi}{\Phi_0}(\Phi_1-\Phi_2), \quad \chi=\frac{2 \pi}{\Phi_0}\biggl[\frac{1}{2}(\Phi_1+\Phi_2) \biggr], \quad \theta=\frac{2 \pi}{\Phi_0} \biggl[\frac{1}{2}(2\Phi_3-\Phi_2-\Phi_1+\Phi_{\rm ext}) \biggr].
\end{equation}
Hence the kinetic energy $T$ and the potential energy $U$ are equal to
\begin{eqnarray}
T& =& \frac{C_s}{2} \frac{\Phi_0^2}{4 \pi^2}(\dot{\theta}+\dot{\chi})^2+C_J \frac{\Phi_0^2}{4 \pi^2}\dot{\chi}^2+\frac{C_J}{4}\frac{\Phi_0^2}{4 \pi^2}\dot{\phi}^2,\nonumber \\
U& =& -E_J \cos\biggl(\chi + \frac{\phi}{2} \biggr)-E_J \cos\biggl(\chi - \frac{\phi}{2} \biggr)+\frac{E_L}{2} \left[\frac{1}{4} (\phi-\phi_{\rm ext})^2+\theta^2\right],
\end{eqnarray}
where $\phi_{\mathrm{ext}} = 2 \pi \Phi_{\mathrm{ext}}/\Phi_0$; we introduced the inductive energy $E_L$ defined in Eq.~\eqref{eq:ind_en}.
The capacitance matrix associated with the $\phi, \chi$ and $\theta$ coordinates is thus
\begin{equation}
\mat{C}=\left(\begin{array}{ccc} \frac{C_J}{2} & 0 & 0 \\
0 & C_s+2 C_J & C_s \\
0 & C_s & C_s \end{array}\right) \rightarrow \mat{C}^{-1}=\frac{1}{C_J}\left(\begin{array}{ccc}2 &0 & 0 \\ 
0 & \frac{1}{2} & -\frac{1}{2} \\
0 & -\frac{1}{2} & \frac{1}{2}+\frac{C_J}{C_s}
\end{array}\right).
\end{equation}
Using the inverse capacitance matrix $\mat{C}^{-1}$, the definition of the dimensionless conjugate momenta, 
\begin{equation}
q_{\phi}=\frac{1}{\hbar}\frac{\partial \lagrangian}{\partial \dot{\phi}}, \quad q_{\chi}=\frac{1}{\hbar}\frac{\partial \lagrangian}{\partial \dot{\chi}}, \quad  q_{\theta}=\frac{1}{\hbar}\frac{\partial \lagrangian}{\partial \dot{\theta}},
\end{equation} 
and Eq.~\eqref{eq:Hnondiag}, we get the final Hamiltonian
\begin{equation*}
\hamiltonian=4 E_{C_J}\left[4 q_{\phi}^2+(q_{\chi}-q_{\theta})^2+2 \frac{C_J}{C_s} q_{\theta}^2 \right]-2 E_J \cos\left(\chi\right) \cos\left(\frac{\phi}{2}\right)+\frac{E_L}{2} \left[\frac{1}{4} (\phi-\phi_{\rm ext})^2+\theta^2\right],
\end{equation*}
with $E_{C_J}$ the charging energy associated with $C_J$ (see Eq.~\eqref{eq:char_en}). In Ref.~\cite{smith2019}, it is further shown that in the fluxonium-like limit and for $\phi_{\rm ext}=\pi$, the system effectively behaves as an element with a $\cos(2 \chi)$ potential. 

\begin{Exercise}[title={The $0$-$\pi$ qubit}, label=exc:0pi]
\begin{figure}[htb]
\centering
\includegraphics[height=7cm]{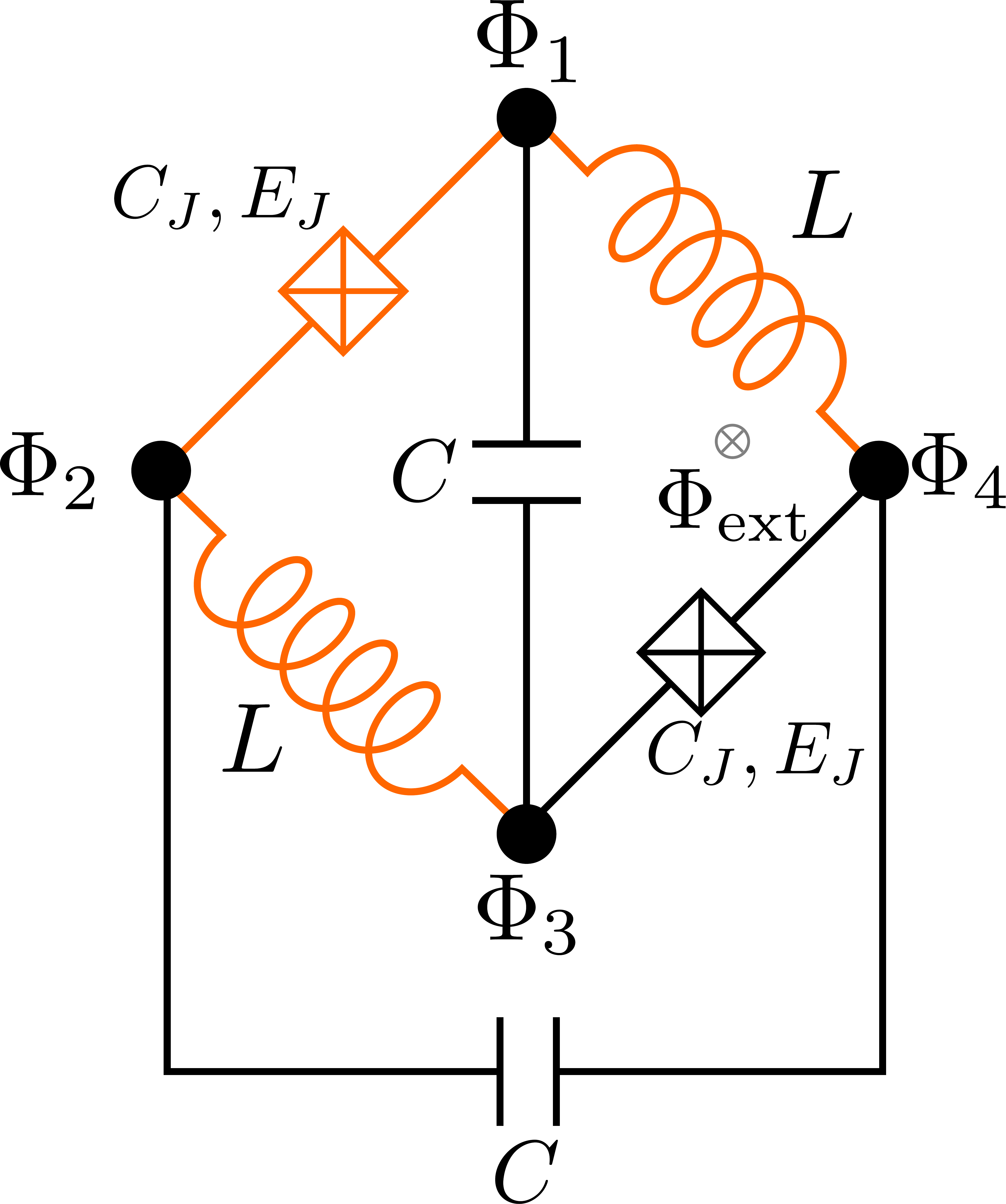}
\caption{Circuit of the $0$-$\pi$ qubit in Exercise~\ref{exc:0pi}.}
\label{fig::zeropi}
\end{figure} 
In this exercise we study the original circuit of the $0$-$\pi$ qubit shown in Fig.~\ref{fig::zeropi}.
\begin{enumerate}
\item Obtain the Lagrangian of the circuit, taking the spanning tree in orange and considering node $3$ as conventional ground node, i.e., setting $\Phi_3 = 0$. 
\item Rewrite the previous Lagrangian in terms of the following variables 
\begin{subequations}\label{eq:newv}
\begin{equation}
\Phi = \frac{1}{2}\bigl[(\Phi_2 - \Phi_3) + (\Phi_4 - \Phi_1)\bigr],
\end{equation}
\begin{equation}
X = \frac{1}{2}\bigl[(\Phi_2 - \Phi_3) - (\Phi_4 - \Phi_1)\bigr],
\end{equation}
\begin{equation}
\Theta = \frac{1}{2}\bigl[(\Phi_2 - \Phi_1) - (\Phi_4 - \Phi_3)\bigr].
\end{equation}
\end{subequations}
Obtain the Hamiltonian as well. How do you interpret the problem when expressed in these variables? Which degree of freedom should be treated as periodic (see also Section~\ref{subsec:fp} for a more formal mapping from real to periodic variables)? What do you expect to happen if you have disorder in the parameters and they are not all perfectly equal as assumed in Fig.~\ref{fig::zeropi}?
\item Considering rescaled phase variables $\phi = 2 \pi  \Phi/\Phi_0 $ and $\theta = 2 \pi \Theta/\Phi_0 $, $\phi_{\mathrm{ext}} = 2 \pi \Phi_{\mathrm{ext}}/\Phi_0$, plot the potential $U(\phi, \theta)$ for these variables with parameters $E_J = 1, E_J/E_L = 150$, $E_L = \frac{\Phi_0^2}{4 \pi^2 L}$, $\Phi_{\mathrm{ext}}=0$ and in the range $\phi \in [-20, 20]$, $\theta \in [-\pi/2, 3 \pi/2]$. 
\item The $0$-$\pi$ qubit is operated in the regime $C \gg C_J$. Which variables have large kinetic energy and which have small kinetic energy (and can thus be approximately put at the minimum of their potential energy, as their $\phi_{\rm zpf}$ is small)?
\item Looking at the potential, qualitatively explain how you expect the low-energy wavefunctions to be distributed and why the circuit could be said to realize a $0$-$\pi$ qubit.
\end{enumerate}
\end{Exercise}

\begin{Answer}[ref={exc:0pi}]
\Question In terms of node fluxes the Lagrangian of the circuit reads
\begin{multline}
\mathcal{L}(\Phi_1, \Phi_2, \Phi_4; \dot{\Phi}_1, \dot{\Phi}_2, \dot{\Phi}_3) = \frac{C}{2} \dot{\Phi}_1^2 + \frac{C}{2} \bigl(\dot{\Phi}_2 - \dot{\Phi}_4 \bigr)^2 + \frac{C_J}{2} \dot{\Phi}_4^2 + \frac{C_J}{2} \bigl(\dot{\Phi}_1 - \dot{\Phi}_2 \bigr)^2 \\ 
- \frac{\Phi_2^2}{2 L} - \frac{1}{2 L} \bigl(\Phi_1 - \Phi_4 \bigr)^2 + E_J \cos \biggl[\frac{2 \pi}{\Phi_0}\bigl(\Phi_1 - \Phi_2 \bigl) \biggr] + E_J \cos \biggl[\frac{2 \pi}{\Phi_0}\bigl(\Phi_4 - \Phi_{\mathrm{ext}}\bigl) \biggr].
\end{multline}
\Question By inverting Eq.~\eqref{eq:newv} we obtain 
\begin{subequations}
\begin{equation}
\Phi_1 = X-\Theta,
\end{equation}
\begin{equation}
\Phi_2 = \Phi+X,
\end{equation}
\begin{equation}
\Phi_4 = \Phi-\Theta.
\end{equation}
\end{subequations}
We rewrite the Lagrangian in terms of the new variables as
\begin{multline}
\mathcal{L}(\Phi ,X, \Theta; \dot{\Phi}, \dot{X}, \dot{\Theta}) = C \dot{X}^2 + C_J \dot{\Phi}^2 + (C+C_J) \dot{\Theta}^2 \\ - \frac{X^2}{L} - \frac{\Phi^2}{L} + E_J \cos \biggl[\frac{2 \pi}{\Phi_0}\bigl(\Phi + \Theta \bigl) \biggr] + E_J \cos \biggl[\frac{2 \pi}{\Phi_0}\bigl(\Phi - \Theta -\Phi_{\mathrm{ext}}\bigl) \biggr].
\end{multline}
We see that the variable substitution diagonalizes the kinetic energy term, i.e., we do not have terms $\dot{X} \dot{\Phi}$ for example. In addition, the variable $X$ is decoupled from the other two, and it simply gives rise to a harmonic oscillator. By defining the conjugate variables as
\begin{subequations}
\begin{equation}
Q_{\Phi} = \frac{\partial \mathcal{L}}{\partial \dot{\Phi}} = 2 C_J \dot{\Phi},
\end{equation}
\begin{equation}
Q_{X} = \frac{\partial \mathcal{L}}{\partial \dot{X}} = 2 C \dot{X},
\end{equation}
\begin{equation}
Q_{\Theta} = \frac{\partial \mathcal{L}}{\partial \dot{\Theta}} = 2 (C + C_J) \dot{\Theta},
\end{equation}
\end{subequations}
the Hamiltonian reads
\begin{multline}
\mathcal{H}(\Phi, X, \Theta; Q_{\Phi}, Q_{X},Q_{\Theta} ) = \frac{Q_X^2}{4 C} + \frac{Q_{\Phi}^2}{4 C_J} + \frac{Q_{\Theta}^2}{4 (C+C_J)} \\ 
+ \frac{X^2}{L} + \frac{\Phi^2}{L} - E_J \cos \biggl[\frac{2 \pi}{\Phi_0}\bigl(\Phi + \Theta \bigl) \biggr] - E_J \cos \biggl[\frac{2 \pi}{\Phi_0}\bigl(\Phi - \Theta - \Phi_{\mathrm{ext}}\bigl) \biggr].
\end{multline}
The Hamiltonian is periodic in $\Theta$ with period $\Phi_0$. This means that we can effectively treat $\Theta$ as a periodic (compact) variable. Note that this is due to the fact that the circuit of the $0$-$\pi$ qubit is made out of two superconducting islands. If we have disorder in the parameters we expect that the $X$ variable is not perfectly decoupled from the other two, but some spurious coupling will be present. 
\Question In terms of the new variables we identify the potential $U(\phi, \theta)$ as
\begin{equation}
U(\phi, \theta) = E_L \phi^2 -E_J \cos(\phi + \theta) - E_J \cos(\phi-\theta - \phi_{\mathrm{ext}}).
\end{equation}
\begin{center}
\includegraphics[scale=0.4]{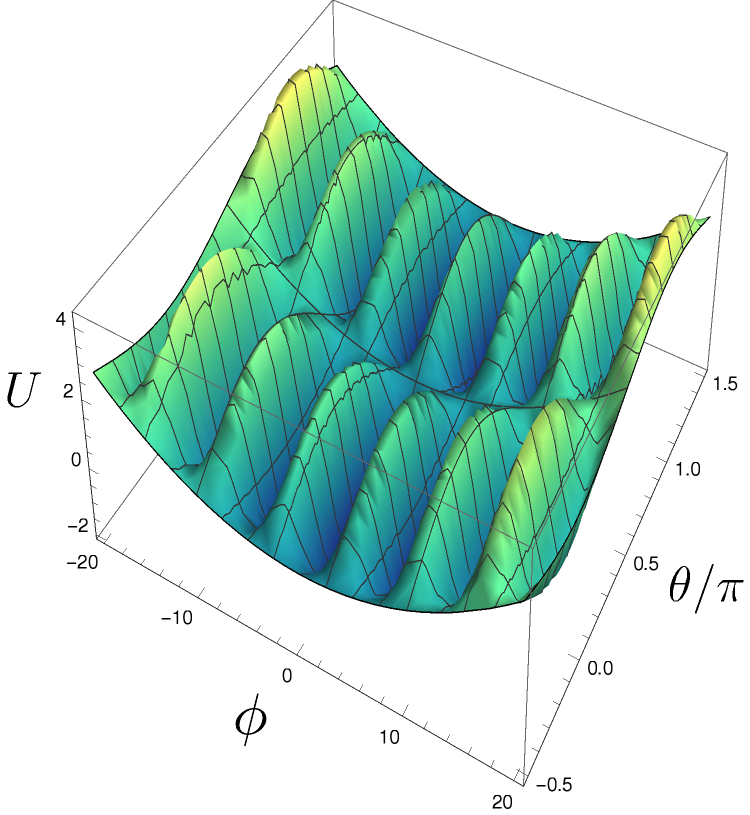}
\includegraphics[scale=0.4]{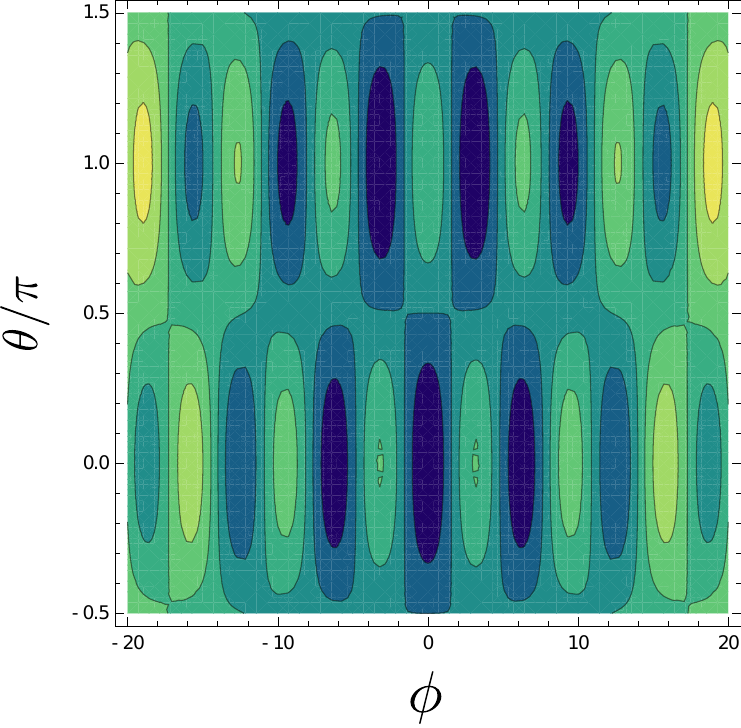}
\captionof{figure}{Potential of the $0$-$\pi$ qubit.}
\label{fig::zeropiPot}
\end{center}
We plot the potential for the required parameters in Fig.~\ref{fig::zeropiPot}. 
\Question When $C \gg C_J$, the variables $X$ and $\Theta$ have relatively low kinetic energy (also depending on $E_J$ and $1/L$), while $\Phi$ has relatively high kinetic energy. 
\Question We see that the potential possesses two collections of local minima at $\theta = 0$ and $\theta = \pi$. This suggests that if $C \gg C_J$, the low-energy wavefunctions will also split into two sets localized at $\theta=0, \pi$, but delocalized in the $\phi$ variable. This justifies the name $0$-$\pi$ qubit.   
\end{Answer}

\subsection{The Möbius-strip circuit}
\label{subsec:mobius}

Due to the fact that superconducting chips are fabricated on 2D~substrates, typical electrical circuit graphs are planar graphs. However, one can consider non-planar electrical circuits, which might require fabricating off-planar Josephson junctions \cite{blochnium}, or using air-bridges, or non-planar couplers \cite{ibm:arch}. For example, one can imagine long-range couplers which are planar, and thus not crossing, on both sides of a 2D chip, as in a so-called biplanar graph.

Leaving fabrication issues aside, one can consider non-planar structures and ask whether they allow for novel qubits with enhanced protection, possibly by exploiting symmetries.
In this section and in Section \ref{subsec:tetrahedron} we consider two possibly non-planar example circuits from the literature to which we apply the method of obtaining the circuit Lagrangian.

\begin{figure}[htb]
\centering
\includegraphics[width=5cm]{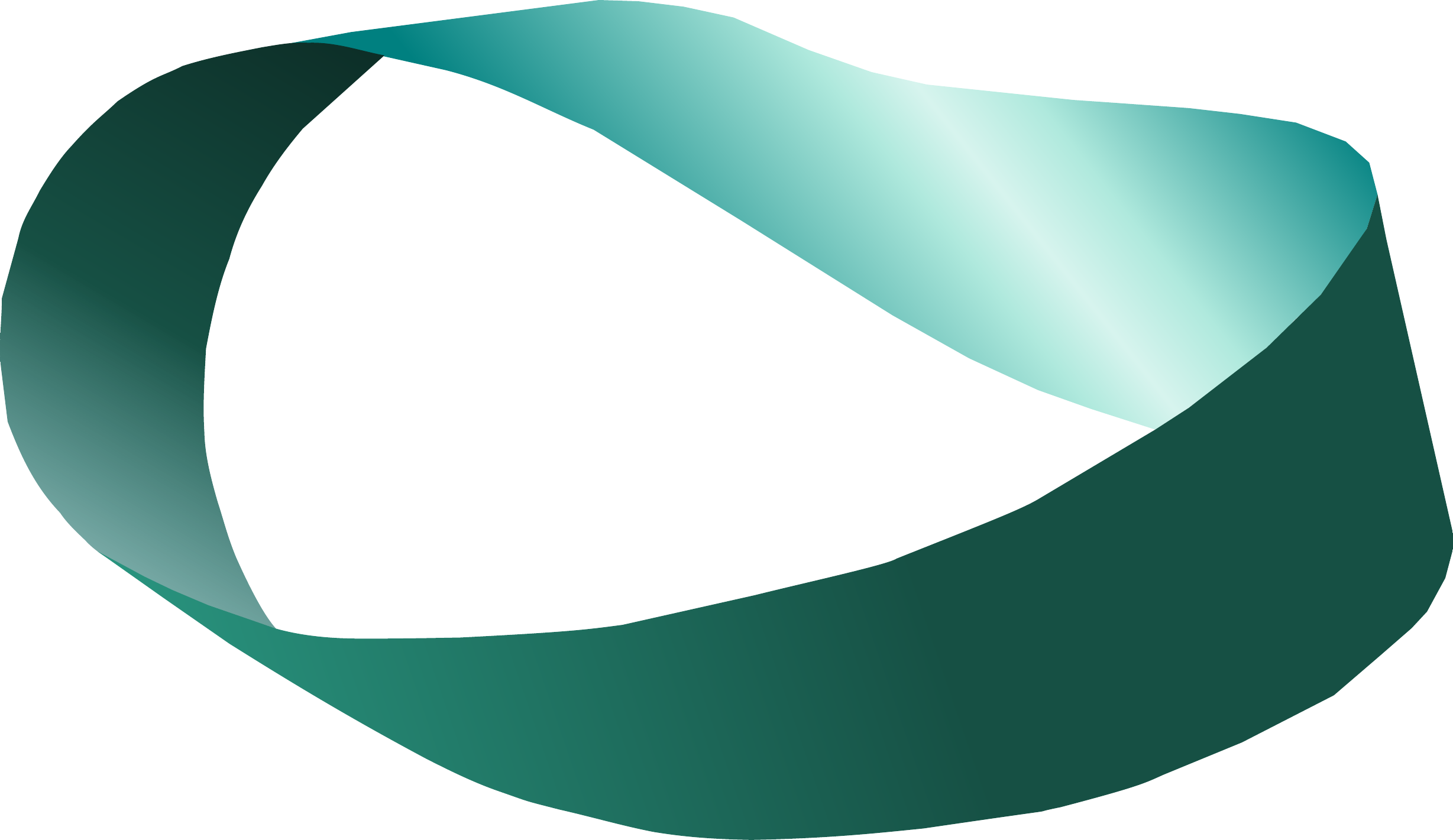}
\caption{A Möbius strip.}
\label{fig:mobius_strip}
\end{figure}

\begin{figure}[htb]
\centering
\includegraphics[scale=0.07]{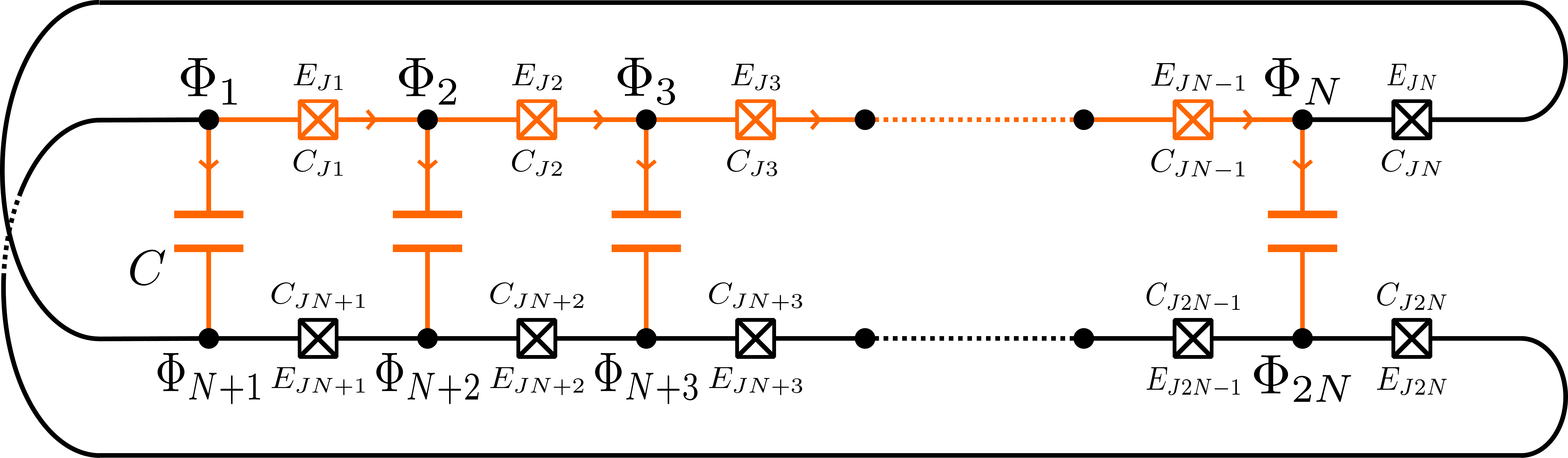}
\caption{Besides the branches shown in this circuit, each node is connected to ground via a capacitance $C_g$. The spanning tree is in orange and is assumed to start at the ground node via a capacitive branch connecting $\Phi_1$ to ground. The orientation of the branches in the tree is shown explicitly.}
	\label{fig:mobius}
\end{figure}

The Möbius strip, see~\cref{fig:mobius_strip}, is a well-known non-orientable surface with a single boundary curve.
We consider a circuit with the topology of a Möbius strip  in which the boundary curve is made of a series of Josephson junctions, further connected by capacitances across the strip, as in Fig.~\ref{fig:mobius}. This circuit was first introduced in~\cite{Kitaev06} and analyzed in detail in \cite{Weiss19}.
We consider a total of~$2N$ junctions with Josephson energy~$E_{J_j}$ and intrinsic capacitance~$C_{J_j}$, each between the nodes~$j$ and~$j+1$.
The nodes~$j$ and~$j+N$ are connected by another capacitance~$C_j$, for~$j\in\{1,\ldots,N\}$. In addition, each node is capacitively coupled to ground via a capacitance $C_g$ (not shown in the figure).

The Lagrangian equals
\begin{align}
\lagrangian =  \sum_{j=1}^{N} \frac{C}{2} (\nodeflux[N+j]{dot}-\nodeflux[j]{dot})^2+\sum_{j=1}^{2N}\frac{C_g}{2}\nodeflux[j]{dot}^2 +\sum_{j=1}^{N-1} \lagrangian_{J_j}(\nodeflux[j+1]{}-\nodeflux[j]{}) 
 +  \lagrangian_{J_N}(\nodeflux[N+1]{}-\nodeflux[N]{}) \notag \\
+ \lagrangian_{J_{N+1}}(\nodeflux[N+2]{}-\nodeflux[N+1]{}) 
+ \sum_{j=N+2}^{2N-1} \lagrangian_{J_j}(\nodeflux[j+1]{}-\nodeflux[j]{})  
+ \lagrangian_{J_{2N}}(\nodeflux[1]{}-\nodeflux[2N]{}),  
\end{align}
where we have defined
\begin{align}
\lagrangian_{J_j}(\nodeflux{}) \coloneqq \frac{C_{J_j}}{2} \nodeflux[]{dot}^2 
+ E_{J_j} \cos\left(\frac{2\pi}{\Phi_0}\Phi \right).
\end{align}
The Hamiltonian can be found analytically by inverting the corresponding capacitance matrix, introducing $Q_j=\frac{\partial \lagrangian}{\partial \dot{\Phi}_j}$ with $[\hat{\Phi}_i,\hat{Q}_j]=i\hbar \delta_{ij} \mathds{1}$.

If this M\"obius circuit were realized by a 3D structure and subjected to an external magnetic field, one would have to additionally consider the magnetic flux $\Phi_{\rm ext}$ through some of the loops. 

This system, proposed first in Ref.~\cite{Kitaev06} and analyzed in detail in Ref.~\cite{Weiss19}, constitutes an alternative construction of a protected $0$-$\pi$ qubit, see Section~\ref{subsec:0pi}. Specifically, assuming~$E_{J_j}=E_J$, $C_{J_j}=C_J$ for all~$j$, in the regime where
\begin{align}
N\gg 1, \qquad E_J \ll E_{C_J}, \qquad C \gg C_J, C \gg C_g, 
\end{align}
there are two lowest-energy states (conventionally labeled as~$\ket{0}$ and~$\ket{1}=\ket{\pi}$) which are exponentially degenerate (in $N$) and non-overlapping as wavefunctions.

One can understand the emergence of this degeneracy as follows. When $C \gg C_g$, it is energetically favorable when a Cooper pair tunnels from, say, node $\Phi_{N+1}$ to node $\Phi_{N+2}$, that it is accompanied by another Cooper pair tunneling from node $\Phi_1$ to $\Phi_2$ so as to keep the charge balance on the capacitive rungs of the ladder. Such a two-Cooper-pair tunneling process is effectively described (and can be perturbatively obtained \cite{Weiss19, VCT:homo}) as $-E_{J_{\rm eff}}\cos\left(\frac{2\pi}{\Phi_0} (\hat{\Phi}_{N+2}-\hat{\Phi}_{N+1}+\hat{\Phi}_1-\hat{\Phi}_2\right)$ with some $E_{J_{\rm eff}}$ which can be estimated with perturbation theory. At the end of the ladder, due to the twisted boundary of the M\"obius strip, the effective two-Cooper-pair tunneling term is $-E_{J_{\rm eff}}\cos\left(\frac{2\pi}{\Phi_0} (\hat{\Phi}_{1}-\hat{\Phi}_{2N}+\hat{\Phi}_N-\hat{\Phi}_{N+1}\right)$. All these potential terms, one for each square face of the strip, commute with the operator $e^{i\frac{\pi}{2e}\sum_{i=1}^N \hat{Q}_i}$; in particular, the twisted face at the end only commutes with $e^{i\frac{\theta}{2e}\sum_{i=1}^N \hat{Q}_i}$ when $\theta=\pi$. The Hamiltonian will thus commute with this product of $\pi$ shifts and its spectrum will have a double-degeneracy reflecting this symmetry.
In \cite{VCT:homo} it is argued that this $0$-$\pi$ qubit can be viewed as an example of a homological rotor code which can encode a single qubit due to the non-orientability of the underlying two-dimensional manifold.

\section{Symmetries and forbidden transitions}
\label{sec:sym-prot}

In a superconducting device, certain energy transitions can be forbidden due to symmetries or properties of the Hamiltonian, just like in atoms. In the next section, we discuss an example showing how this might lead to useful, qubit-encoding, degeneracies. Here, instead, we mention some simple symmetry facts which apply to basic qubits, such as transmon and flux qubits. These allow one to understand how to couple to these qubits and which transitions are forbidden.

One can consider the following unitary transformation $U_{\pi}$ on a single set of conjugate variables
\begin{equation}
    U_{\pi} \colon \hat{\phi} \rightarrow -\hat{\phi},
    \,\hat{q} \rightarrow -\hat{q}.
    \label{eq:upi}
\end{equation}
This can be viewed as a spatial parity transformation if $\hat{\phi}$ and $\hat{q}$ are seen as (dimensionless) positions and momenta, which in phase space corresponds to a rotation of $180$ degrees. The unitary $U_{\pi}$ is known as the parity operator and one can verify that it is defined as 
\begin{equation}
    U_{\pi} = i e^{-i \frac{\pi}{2}(\hat{\phi}^{2} + \hat{q}^{2})} = e^{- i \pi \hat{a}^{\dagger} \hat{a}} = (-1)^{\hat{a}^{\dagger} \hat{a}}. 
\end{equation}
Note that $U_{\pi}$ is not only a unitary, but also a Hermitian, transformation, satisfying $U_{\pi}^2=\mathds{1}$, or equivalently $U_{\pi}= U_{\pi}^{\dagger}$.

A Hamiltonian which only contains even functions of $\hat{q}$ and $\hat{\phi}$ is thus invariant under this transformation, implying that its eigenvectors $\ket{\psi_k}$ can be chosen as eigenstates of $U_{\pi}$ with eigenvalues $\pm 1$, i.e., the eigenstates either have even or odd parity. The symmetry implies that for two eigenstates $\ket{\psi_k}$ with the same symmetry (both odd or both even), one has
\begin{equation}
    \langle \psi_k | \hat{q} |\psi_l \rangle=
    \langle \psi_k | U_{\pi}^{\dagger}  \hat{q} U_{\pi}|\psi_l \rangle=-\langle \psi_k | \hat{q} |\psi_l \rangle=0.
    \label{eq:Q0}
\end{equation}
Similarly, $\langle \psi_k | \hat{\phi} |\psi_l \rangle=0$ between eigenstates with the same symmetry. More generally, 
\begin{align}
    \langle \psi_k | f_{\rm odd}(\hat{\phi}) |\psi_l \rangle=0,
    \label{eq:sym0-2}
\end{align} 
between eigenstates with the same symmetry where $f_{\rm odd}()$ is an odd function.

Hamiltonians which depend on $\hat{q}^2$ and any even potential $U(\phi)$ are thus parity-symmetric; examples are the transmon qubit, the flux qubit and fluxonium qubit Hamiltonian, both at $\phi_{\rm ext}=0$ or $\pi$, see e.g. the potential in  Fig.~\ref{fig:fluxonium_potential}. It implies that driving the qubit `via its dipole', i.e., coupling via the operator $\hat{q}$, will only lead to transitions from even to odd and vice-versa: due to Eq.~\eqref{eq:Q0}, the other transitions are forbidden. The same holds when one drives inductively, via the $\hat{\phi}$ operator. This is important to know when one considers how to manipulate the qubit, or estimate matrix elements of $\hat{\phi}$ and $\hat{q}$ in the qubit basis. For example, for those qubits one has $\bra{g} \hat{\phi} \ket{g}=\bra{e} \hat{\phi} \ket{e}=\bra{g} \hat{q} \ket{g}=\bra{e} \hat{q} \ket{e}=0$. 

Further considerations can be made for the operator $\hat{q}$ when the Hamiltonian only depends on $\hat{q}^2$, also when the potential is not parity-symmetric, i.e., $U(\phi) \neq U(-\phi)$. Since in the flux representation $\hat{q}=-i \frac{d}{d \phi}$, any Hamiltonian which solely depends on $\hat{q}^2$ acts as a real-valued operator in the flux basis. This implies that the eigenfunctions $\psi_k(\phi)$  can be taken as real-valued functions and, if $\psi_k(\phi)$ are normalizable, this leads to
\begin{subequations}
\begin{equation}
    \bra{\psi_k} \hat{q} \ket{\psi_k} =  -i \int_{-\infty}^{\infty} d\phi  \,\psi_k^*(\phi) \frac{d}{d\phi}\psi_k(\phi)=0, 
    \end{equation}
    \begin{equation}
     \bra{\psi_k} \hat{q} \ket{\psi_l} = -  \bra{\psi_l} \hat{q} \ket{\psi_k}. 
     \label{eq:pure-im}
     \end{equation}
\end{subequations}
This property is formally equivalent to a time-reversal invariance.  The flux qubit or fluxonium qubit Hamiltonian in Eq.~\eqref{eq:hfluxonium} at arbitrary external flux is an example, showing that driving via $\hat{q}$ induces no $Z$-component if the computational states are taken as the qubit levels $\ket{0},\ket{1}$. In addition, $
\hat{q}$ then acts proportional to Pauli $Y$ (and not $X$) on the qubit due to Eq.~\eqref{eq:pure-im}.

\begin{Exercise}[label=exc:fluxq]
Argue why $\bra{0} \hat{\phi} \ket{1}=0$ for the flux qubit at $\phi_{\rm ext}=\pi$ where the potential is parity-symmetric in $\phi$ and $\ket{0}=\frac{1}{\sqrt{2}}(\ket{g}+\ket{e})$ and $\ket{1}=\frac{1}{\sqrt{2}}(
\ket{g}-\ket{e})$ with ground and first-excited states $\ket{g}$ and $\ket{e}$.
\end{Exercise}

\begin{Answer}[ref={exc:fluxq}]
Due to parity symmetry we have $\bra{g} \hat{\phi} \ket{g}=\bra{e} \hat{\phi} \ket{e}=0$. Thus $\bra{0} \hat{\phi} \ket{1}=\frac{1}{2}(
-\bra{e} \hat{\phi} \ket{g}+\bra{g} \hat{\phi} \ket{e})=0$ as the wavefunctions $\psi_e(\phi)$ and $\psi_g(\phi)$ are real.
\end{Answer}

\subsection{Symmetries: discussion of a tetrahedral qubit}
\label{subsec:tetrahedron}

In this section, we use the example of a tetrahedral circuit~\cite{Feigelman04} to discuss symmetry-protected subspaces. We first review the consequences of symmetries of a Hamiltonian more generally, as it is applied in atomic, molecular and condensed-matter quantum physics.

Recall that a representation of a group $G$ is a group homomorphism in which each element~$g$ of the group~$G$ is mapped to a matrix or operator~$D_g$ such that $g_1g_2=g_3\implies D_{g_1}D_{g_2}=D_{g_3}$ (we refer the reader to Ref.~\cite{zee}) for the basics of group theory).
If for a subspace~$V$, $\forall g, D_g(V)\subseteq V$, i.e.,~vectors in $V$ are mapped to vectors in $V$, then $V$ is called an invariant subspace.
If a representation does not possess any non-trivial invariant subspaces, it is said to be irreducible.
A symmetry of a quantum Hamiltonian $H$ is a representation such that $[D_g,H]=0$ for all~$g$.
An important consequence is that if~$\ket{\psi}$ is an eigenstate of~$H$, i.e.,~$H\ket{\psi}=E\ket{\psi}$, then
\begin{align}
H (D_g\ket{\psi}) = D_g H\ket{\psi} = E(D_g\ket{\psi}),
\end{align}
i.e.,~$D_g\ket{\psi}$ for any $g \in G$ is an eigenstate as well, degenerate with $\ket{\psi}$. Every group representation $D_g$ thus breaks down into a direct sum of irreducible representations, and the associated invariant subspaces form degenerate eigenspaces of the Hamiltonian.

\begin{figure}[htb]
\centering
	\includegraphics[scale=0.12]{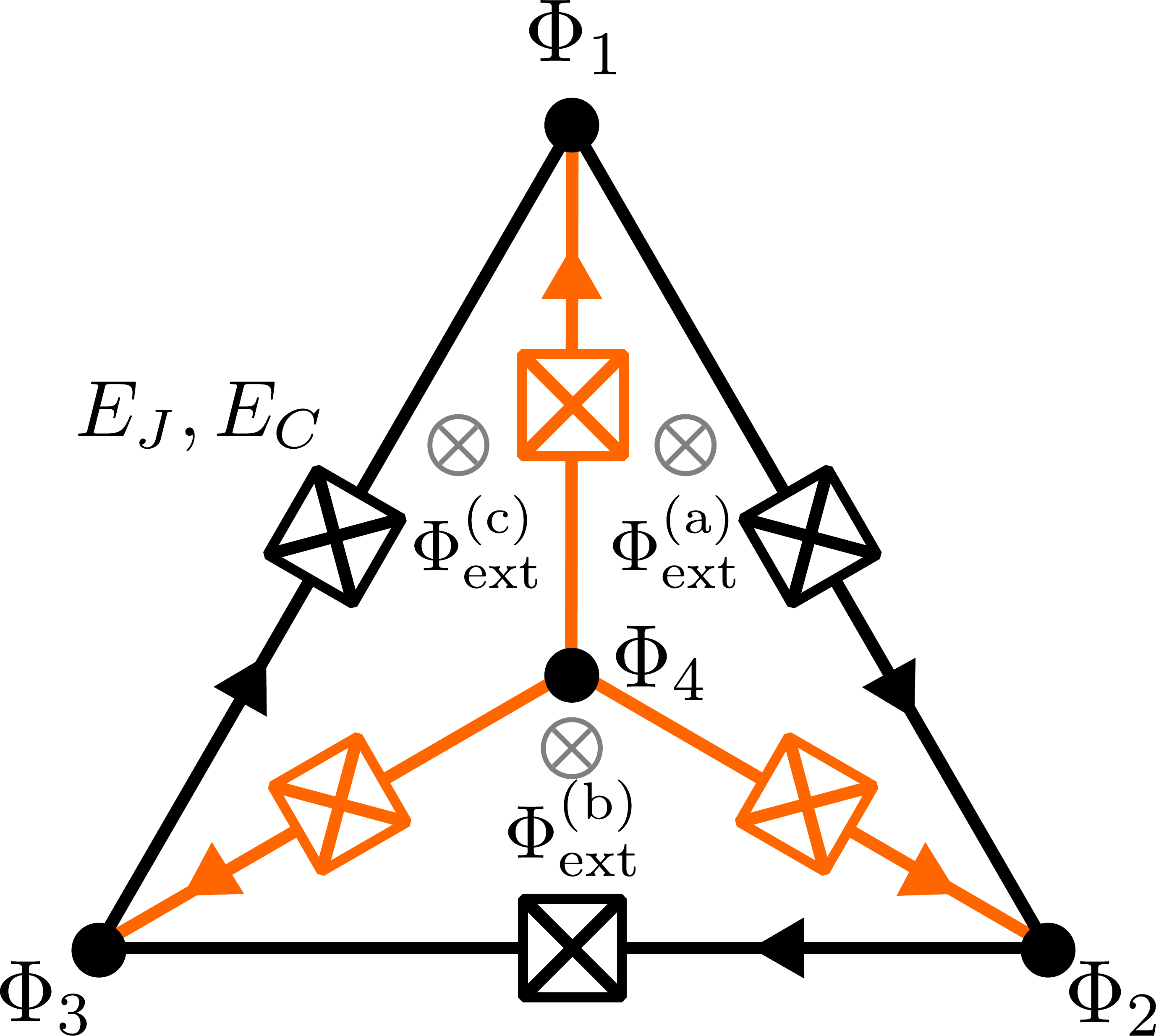}
	\caption{A flattened tetrahedron. We assume all Josephson junctions to have the same~$E_J$ and~$E_C$. The considered spanning tree is shown in orange and $\Phi_i$ are the node fluxes (we assume that node $\Phi_4$ is grounded). The orientation of the branches is shown explicitly and we consider all loops to have a clockwise orientation.
	}
	\label{fig:tetrahedron}
\end{figure}

In practice, spectral degeneracy based on symmetry can be useful to engineer `protected' qubits, assuming that the noise does not commute with the symmetry. If the degenerate space is two-dimensional, we can say that this space encodes a qubit, and one can define logical $X$ and $Z$ operators which act on this qubit. Consider a perturbative term in the Hamiltonian that does not commute with the symmetry.
Such a term typically maps states in the degenerate subspace to states outside of it. At lowest order in perturbation theory, see e.g. Eq.~\eqref{eq:firstorderSW} in Section~\ref{sec:SW}, such an `off-diagonal' term does not break the degeneracy in the spectrum. This implies that errors due to such a term are suppressed at this order. In fact, only errors which can generate the logical operators of the encoded qubit can lead to a breaking of the degeneracy, at some order in perturbation theory. By choosing an encoding such that these logical operators relate to unlikely processes ---affecting, say, many degrees of freedom, hence generated only at high-order in perturbation theory, as in the case of topological order \cite{BHM:topo,VCT:homo}--- one can thus gain protection. Furthermore, if one could activate logical processes at will, while the environment cannot, then one can perform gates on the qubit while it remains protected. In Chapter~\ref{sec:bip} we discuss noise protection in more detail.

Coming back to the symmetry group, we note that if the invariant subspaces are one-dimensional, then $D_g\ket{\psi}\propto\ket{\psi}$ and there is no degeneracy.
Hence, to find a symmetry-protected subspace that can host a qubit or a qudit, it is necessary that the group has at least one two-dimensional irreducible representation.
If the symmetry group is Abelian (all group elements commute), then all irreducible representations are one-dimensional, hence no protected qubit can be found. Note that in the theory of stabilizer quantum error correction, an example of the non-Abelian group $G$ is the centralizer of the stabilizer group which is generated by the Abelian stabilizer group plus the (non-commuting) logical operators of the encoded qubits \cite{thesis:gottesman}.

If the electrical circuit Hamiltonian possesses a non-Abelian symmetry, like the (spatial) symmetry group of a three-dimensional tetrahedron for example, then a degenerate subspace of dimension larger than one is a possibility. More generally, one could imagine Hamiltonians of electrical circuits which have the symmetry group of, say, three-dimensional polyhedra or hyperbolic surfaces.

Obviously, it may be overall easier to find electrical circuits which only {\em approximately} possess a symmetry, i.e., parameters are chosen such that some terms can be neglected or degrees of freedom removed to get an effective Hamiltonian with the targeted symmetry.

We now consider a tetrahedron with one Josephson junction per edge, where all junctions have the same $E_J$ and $E_C$; see a flattened representation of this tetrahedron in \cref{fig:tetrahedron}. In the absence of external magnetic fields, the electrical circuit has tetrahedral symmetry due to all branch elements being equal. 

However, if this is genuinely a three-dimensional tetrahedron and external magnetic fields are present, then the net flux through the tetrahedron must be zero by the Maxwell equation $\nabla \cdot \vect{B}=0$.
This implies that if the flux is positive through some faces, it has to be negative through at least one other face.
As a consequence, in the presence of external magnetic fields, the tetrahedral symmetry is generally broken, since not all faces can be pierced by the same flux. However, note that this is not true if the threading flux is a half-flux quantum, since $\phi_{\rm ext}=\pm \pi$ gives the same contribution to the Hamiltonian.  

In addition, in the absence of magnetic fields, if one of the flux nodes is grounded, we also lose full tetrahedral symmetry. Indeed, as we will see, by grounding one node, the number of degrees of freedom is only three and the symmetry group of the Lagrangian is reduced. One could also consider a circuit in which we have an {\em additional} ground node to which each of the four nodes is coupled capacitively via some capacitance $C_g$. In that scenario, without external fluxes or at $\phi_{\rm ext}=\pm \pi$, the Lagrangian and Hamiltonian will possess tetrahedral symmetry.

Here we will restrict ourselves to the case where one node is grounded, say, $\nodeflux[4]{}=0$. We can then just consider a tetrahedron collapsed onto its base, as in~\cref{fig:tetrahedron}. Having a two-dimensional structure also makes an experimental implementation more feasible. The external fluxes in \cref{fig:tetrahedron} are then simply piercing the plane and, for symmetry, we will set $\extflux[(\mathrm{a})]{}=\extflux[(\mathrm{b})]{}=\extflux[(\mathrm{c})]{}\equiv\extflux[]{}$. Let us explicitly derive the Hamiltonian of this structure. In our derivation we will first include $\nodeflux[4]{}$ explicitly in the Lagrangian and then set it to zero.

Given the choice of spanning tree in~\cref{fig:tetrahedron}, the loop conditions with the external fluxes, see Eq.~\eqref{eq:KVL-flux_ext}, are
\begin{align}
\branchflux[4,1]{} + \branchflux[1,2]{}-\branchflux[4,2]{}  &= \extflux[(\mathrm{a})]{}, \nonumber \\
\branchflux[4,2]{}+\branchflux[2,3]{}-\branchflux[4,3]{}   &= \extflux[(\mathrm{b})]{}, \nonumber \\
\branchflux[4,3]{}+\branchflux[3,1]{} - \branchflux[4,1]{}  &= \extflux[(\mathrm{c})]{},
\end{align}
where for branches on the spanning tree we use the definition  $\branchflux[j,k]{}=\nodeflux[k]{}-\nodeflux[j]{}$, as in Eq.~\eqref{eq:define-node}. It follows that
\begin{align}
\branchflux[1,2]{} &= \nodeflux[2]{} - \nodeflux[1]{} + \extflux[(\mathrm{a})]{}, \nonumber \\
\branchflux[2,3]{} &= \nodeflux[3]{} - \nodeflux[2]{} + \extflux[(\mathrm{b})]{}, \nonumber \\
\branchflux[3,1]{} &= \nodeflux[1]{} - \nodeflux[3]{} + \extflux[(\mathrm{c})]{}.
\end{align}
Assuming $\extflux[(\mathrm{a})]{}=\extflux[(\mathrm{b})]{}=\extflux[(\mathrm{c})]{}\equiv\extflux[]{}$, the Lagrangian is
\begin{align}
\lagrangian &= \lagrangian_J(\nodeflux[1]{} - \nodeflux[4]{}) + \lagrangian_J(\nodeflux[2]{} - \nodeflux[4]{}) + \lagrangian_J(\nodeflux[3]{} - \nodeflux[4]{}) \nonumber\\
&\quad + \lagrangian_J(\nodeflux[2]{} - \nodeflux[1]{} + \extflux[]{}) + \lagrangian_J(\nodeflux[3]{} - \nodeflux[2]{} + \extflux[]{}) + \lagrangian_J(\nodeflux[1]{} - \nodeflux[3]{} + \extflux[]{}),
\end{align}
where we have defined
\begin{align}
\lagrangian_{J}(\nodeflux{}) = \frac{C_{J}}{2} \nodeflux[]{dot}^2 
+ E_{J} \cos\Bigl(\frac{2\pi}{\Phi_0}\Phi\Bigl).
\end{align}
We also assume~$\extflux[]{dot}=0$ in the following and set $\Phi_4=0$. 
This means that there are only three actual degrees of freedom, and the symmetry group of the Lagrangian is $S_3$, the permutation group of three elements. In particular, it is easy to see that~$\lagrangian$ for $\Phi_{\rm ext} \neq 0$, $\Phi_4=0$, is invariant under permutations of the three node fluxes as long as we add a negative sign to the odd permutations. For example, the transformation
\begin{align}
\nodeflux[1]{} \to - \nodeflux[2]{}, \qquad \nodeflux[2]{} \to - \nodeflux[1]{}, \qquad
\nodeflux[3]{} \to  -\nodeflux[3]{}, 
\end{align}
corresponds to a {\em signed} permutation of~$1$ and~$2$. The classical Hamiltonian in terms of dimensionless variables $\phi_j$, $q_j$, defined as in Eq.~\eqref{eq:dimless}, is found by inverting the capacitance matrix:
\begin{align}
\hamiltonian &= 2 E_{C_J} \Bigl( q_1^2 + q_2^2 + q_3^2 + q_1 q_2 + q_2 q_3 + q_3 q_1 \Bigr) \nonumber\\
&\quad - E_J\Bigl[ \cos \nodephi[1]{} + \cos\nodephi[2]{} + \cos\nodephi[3]{} 
 + \cos(\nodephi[2]{} - \nodephi[1]{} + \extphi[]{}) + \cos(\nodephi[3]{} - \nodephi[2]{} + \extphi[]{}) + \cos(\nodephi[1]{} - \nodephi[3]{} + \extphi[]{}) \Bigr].
\label{eq:hamil_tetrahedron}
\end{align}
The classical Hamiltonian $\hamiltonian$ can be quantized to $H$, imposing the standard commutation relations.

\subsubsection{Degenerate eigenspaces and irreducible representations}
 
The group~$S_3$ has six elements, namely the identity~$e$, the transpositions $(12),(23),(31)$ and the cyclic permutations $(123),(132)$.
Furthermore, it is non-Abelian, since, for example, $(12)(23)\neq(23)(12)$. We consider this symmetry in the flux basis. First, we write any state as 
\[
\ket{\psi}=\int_{\mathbb{R}^3}
 d\phi_1 d\phi_1 d\phi_2 \,\psi(\phi_1,\phi_2,\phi_3) \ket{\phi_1,\phi_2,\phi_3},
\]
and we wish to find properties of eigenfunctions $\psi(\phi_1,\phi_2,\phi_3)\in \ell^2(\mathbb{R}^3)$ which arise solely from symmetry. 

The Hamiltonian $H$ is clearly invariant under signed transformations $D_{\sigma}$
for any permutation $\sigma\in S_3$ on the three spaces. More precisely, the action of the group elements $D_{\sigma}$ is
\begin{align}
\ket{\phi_1,\phi_2,\phi_3} \: \overset{D_e}{\mapsto} \:\:\,&\ket{\phi_1,\phi_2,\phi_3}, \nonumber \\
\overset{D_{(12)}}{\mapsto} &\ket{-\phi_2,-\phi_1,-\phi_3}, \nonumber\\ 
\overset{D_{(23)}}{\mapsto} &\ket{-\phi_1,-\phi_3,-\phi_2}, \nonumber\\
\overset{D_{(31)}}{\mapsto} &\ket{-\phi_3,-\phi_2,-\phi_1}, \nonumber\\
\overset{D_{(123)}}{\mapsto} &\ket{\phi_2,\phi_3,\phi_1}, \nonumber\\
\overset{D_{(132)}}{\mapsto} &\ket{\phi_3,\phi_1,\phi_2}.
\label{eq:six}
\end{align}
We note that for specific values of $\phi_1, \phi_2, \phi_3$, e.g. $\phi_i=0$, $\phi_i=\phi_j$, some of the basis states in Eq.~\eqref{eq:six} may not be distinct. This fact will restrict the support of the found eigenwavefunctions (the wavefunctions will have some zeros or `nodes'). We will not pursue this analysis here, but will consider the consequences of this group action on basis states $\ket{\phi_1,\phi_2,\phi_3}$ for general positions $\phi_1,\phi_2,\phi_3$ and comment on this issue at the end of the analysis.

We can view each signed transposition in Eq.~\eqref{eq:six} as the application of the parity operator $U_{\pi}^{\otimes 3}$ in Eq.~\eqref{eq:upi} on the three spaces, followed by the transposition (swap) of a pair of spaces. 

For each choice of~$\phi_1,\phi_2,\phi_3$, and under the assumptions $\phi_i \neq 0$, and $\phi_i \neq \pm \phi_j$, the six vectors in Eq.~\eqref{eq:six} determine a six-dimensional space. This six-dimensional (regular) representation contains once the one-dimensional trivial representation, once the one-dimensional sign representation, and twice the two-dimensional representation of $S_3$, so together $2+2+1+1=6$ (this is a property of a regular representation, see page 107 in \cite{hamermesh}). This means we can label eigenvectors and (degenerate) eigenspaces by these four possible irreducible representations (irreps).

For example, the eigenstates corresponding to the trivial irreducible representation (each group element acting as the identity) are of the form
\begin{align}
\ket{\psi^{(S)}}=\int_{\mathbb{R}^3}
 d\phi_1 d\phi_1 d\phi_2 \,\psi_{
 \rm sym}(\phi_1,\phi_2,\phi_3) \ket{\phi_1,\phi_2,\phi_3},
 \label{eq:S}
\end{align}
where, since $\forall \sigma,\; D_{\sigma}\ket{\psi^{(S)}}=\ket{\psi^{(S)}}$, we have $\forall \sigma,\;\psi_{\rm sym}(\phi_1,\phi_2,\phi_3)=\psi_{\rm sym}(D_{\sigma}(\phi_1,\phi_2,\phi_3))$. The other one-dimensional irreducible representation applies $-1$ for odd permutations $\sigma$ and $+1$ for even permutations and corresponds to eigenvectors of the form
\[
\ket{\psi^{(AS)}}=\int_{\mathbb{R}^3}
 d\phi_1 d\phi_1 d\phi_2 \,\psi_{
 \rm asym}(\phi_1,\phi_2,\phi_3) \ket{\phi_1,\phi_2,\phi_3},
\]
with $\psi_{\rm asym}(\phi_1,\phi_2,\phi_3)={\rm sign}(\sigma)\psi_{\rm asym}(D_{\sigma}(\phi_1,\phi_2,\phi_3))$ for all $\sigma \in S_3$. Alternatively, this state can be written as
\begin{align}
\ket{\psi^{(AS)}}& =\int_{\mathbb{R}^3}
 d\phi_1 d\phi_1 d\phi_2 \,\chi(\phi_1,\phi_2,\phi_3)\:\Bigl( \ket{\phi_1,\phi_2,\phi_3}-\ket{-\phi_1,-\phi_3,-\phi_2}
 \nonumber\\
&\qquad
-\ket{-\phi_2,-\phi_1,-\phi_3}-\ket{-\phi_3,-\phi_2,-\phi_1}+
 \ket{\phi_3,\phi_1,\phi_2}+\ket{\phi_2,\phi_3,\phi_1}\Bigr),
 \label{eq:SS}
\end{align}
where the wavefunction $\chi(\phi_1,\phi_2,\phi_3)$ is still entirely free. Since the state space of $H$ is infinite-dimensional, one expects an infinite number of eigenvalues and eigenwavefunctions (such as $\chi(\phi_1,\phi_2,\phi_3)$) for each sector labelled by an irrep. 

\begin{Exercise}[label=exc:ortho-irreps]
Verify that $D_{\sigma}\ket{\psi^{(AS)}}={\rm sign}(\sigma)\ket{\psi^{(AS)}}$ with $\ket{\psi^{(AS)}}$ as in Eq.~\eqref{eq:SS}. 
Expand $\ket{\psi^{(S)}}$ in Eq.~\eqref{eq:S} in a similar way as Eq.~\eqref{eq:SS}, using $\forall \sigma,\; D_{\sigma}\ket{\psi^{(S)}}=\ket{\psi^{(S)}}$. Show explicitly that $\bra{\psi^{(S)}}\psi^{(AS)}\rangle=0$, using the fact that these states transform differently under some $D_{\sigma}$. 
\end{Exercise}
 
\begin{Answer}[ref=exc:ortho-irreps]
We can verify that $D_{\sigma=(12)}\ket{\psi^{(AS)}}=-\ket{\psi^{(AS)}}$ etc.
We have 
\begin{align}
\ket{\psi^{(S)}}& =\int_{\mathbb{R}^3}
 d\phi_1 d\phi_2 d\phi_3 \,\tilde{\chi}(\phi_1,\phi_2,\phi_3)\:\Bigl( \ket{\phi_1,\phi_2,\phi_3}+
 \ket{\phi_3,\phi_1,\phi_2}+\ket{\phi_2,\phi_3,\phi_1} 
 \nonumber\\
&\qquad\qquad\qquad\qquad+
\ket{-\phi_1,-\phi_3,-\phi_2}+\ket{-\phi_2,-\phi_1,-\phi_3}+\ket{-\phi_3,-\phi_2,-\phi_1}\Bigr),
\end{align}
where $\tilde{\chi}(\phi_1,\phi_2,\phi_3)$ is still free again. It is easy to see that the eigenspaces labelled by irreps are orthogonal, since $\bra{\psi^{(S)}}\psi^{(SS)}\rangle=0$ as $\bra{\psi^{(S)}}\psi^{(SS)}\rangle=\bra{\psi^{(S)}}D_{\sigma}\ket{\psi^{(SS)}}=-\bra{\psi^{(S)}}\psi^{(SS)}\rangle$ for odd permutations $\sigma$.
\end{Answer}

We label the two two-dimensional irreps by $k=1,2$. Let's build these irreps. Take a vector
\begin{align}
\ket{\psi_1^{(k=1)}} & \propto \int_{\mathbb{R}^3} d\phi_1d\phi_2d\phi_3 \: \chi'(\phi_1,\phi_2,\phi_3) \:
\Bigl( \ket{\phi_1,\phi_2,\phi_3} - \ket{-\phi_1,-\phi_3,-\phi_2} \nonumber\\
&\qquad\qquad\qquad\qquad\qquad\qquad\qquad\qquad\quad+ \ket{-\phi_2,-\phi_1,-\phi_3} - \ket{\phi_2,\phi_3,\phi_1} \Bigr).
\end{align}
This vector is orthogonal to $\ket{\psi^{(S)}}$ and $\ket{\psi^{(AS)}}$, assuming that $\chi(\phi_1,\phi_2,\phi_3), \chi'(\phi_1,\phi_2,\phi_3)$ etc. are such that 
all states over which we integrate are orthogonal for different values of $\phi_1,\phi_2,\phi_3$. Applying $D_{(13)}$ to this vector gives
\begin{align}
D_{(13)}\ket{\psi_1^{(k=1)}} & \propto \int_{\mathbb{R}^3} d\phi_1d\phi_2d\phi_3 \: \chi'(\phi_1,\phi_2,\phi_3) \:
\Bigl( \ket{-\phi_3,-\phi_2,-\phi_1} - \ket{\phi_2,\phi_3,\phi_1} \nonumber\\
&\qquad\qquad\qquad\qquad\qquad\qquad\qquad\qquad\quad+ \ket{\phi_3,\phi_1,\phi_2} - \ket{-\phi_1,-\phi_3,-\phi_2} \Bigr),
\end{align}
which is not orthogonal to $\ket{\psi_1^{(k=1)}}$, but one can Gram-Schmidt orthogonalize and obtain an orthogonal vector which is the partner to $\ket{\psi_1^{(k=1)}}$ in the two-dimensional irrep:
\begin{align}
\ket{\psi_2^{(k=1)}} & \propto  \int_{\mathbb{R}^3} d\phi_1d\phi_2d\phi_3 \: \chi'(\phi_1,\phi_2,\phi_3) \:
\Bigl( \ket{\phi_1,\phi_2,\phi_3} + \ket{-\phi_1,-\phi_3,-\phi_2} \nonumber\\
&\qquad\qquad\qquad\qquad\qquad\qquad\qquad\qquad\qquad+ \ket{-\phi_2,-\phi_1,-\phi_3} 
- 2\ket{-\phi_3,-\phi_2,-\phi_1} \nonumber\\
&\qquad\qquad\qquad\qquad\qquad\qquad\qquad\qquad\qquad- 2\ket{\phi_3,\phi_1,\phi_2} + \ket{\phi_2,\phi_3,\phi_1}\Bigr).
\end{align}

\begin{Exercise}[label=exc:ortho]
In the simplified language of three qutrit states (which is applicable here), verify the orthogonality of $\ket{a}=\ket{123}-\ket{132}+\ket{213}-\ket{231}$ and 
$\ket{b}=\ket{123}+\ket{132}+\ket{213}-2\ket{321}-2\ket{312}+\ket{231}$ and the orthogonality of both these states, $\ket{a}$ and $\ket{b}$, with respect to $\ket{S}=\ket{123}+\ket{312}+\ket{231}+\ket{132}+\ket{213}+\ket{321}$ and 
$\ket{AS}=\ket{123}+\ket{312}+\ket{231}-\ket{132}-\ket{213}-\ket{321}$. Verify that $(13)\ket{a}=\ket{321}-\ket{231}+\ket{312}-\ket{132}$ lies in the span of $\ket{a}$ and $\ket{b}$.
\end{Exercise}
 
\begin{Answer}[ref=exc:ortho]
Use Mathematica, MATLAB or Python software to represent vectors as 6-dimensional vectors and verify.
\end{Answer}

To construct the basis functions of the last two-dimensional irrep, we start with a vector orthogonal to all previous ones, $\ket{\psi_{1,2}^{(k=1)}},\ket{\psi^{(S)}},\ket{\psi^{(AS)}}$, using simple math as in Exercise \ref{exc:ortho} to represent orthogonality. Note that the $\pm$ signs in the states $\ket{\pm \phi_1,\pm \phi_2,\pm \phi_3}$ play no role in these orthogonality facts, since they are fixed by the state being an even or odd permutation from $\ket{\phi_1,\phi_2,\phi_3}$. Thus, we can take
\begin{align}
\ket{\psi_1^{(k=2)}} & \propto \int_{\mathbb{R}^3} d\phi_1d\phi_2d\phi_3 \: \chi''(\phi_1,\phi_2,\phi_3) \:
\Bigl( \ket{\phi_1,\phi_2,\phi_3} - \ket{-\phi_2,-\phi_1,-\phi_3} \nonumber\\
&\qquad\qquad\qquad\qquad\qquad\qquad\qquad\qquad\quad+ \ket{-\phi_3,-\phi_2,-\phi_1} - \ket{\phi_3,\phi_1,\phi_2} \Bigr),
\end{align}
and its Gram-Schmidt orthogonalized partner is
\begin{align}
\ket{\psi_2^{(k=2)}} & \propto  \int_{\mathbb{R}^3} d\phi_1d\phi_2d\phi_3 \: \chi''(\phi_1,\phi_2,\phi_3) \:
\Bigl( \ket{\phi_1,\phi_2,\phi_3} + 2\ket{-\phi_1,-\phi_3,-\phi_2} \nonumber\\
&\qquad\qquad\qquad\qquad\qquad\qquad\qquad\qquad\qquad- \ket{-\phi_2,-\phi_1,-\phi_3} 
- \ket{-\phi_3,-\phi_2,-\phi_1} \nonumber\\
&\qquad\qquad\qquad\qquad\qquad\qquad\qquad\qquad\qquad +\ket{\phi_3,\phi_1,\phi_2} - 2\ket{\phi_2,\phi_3,\phi_1}\Bigr).
\end{align}

We have assumed that the wavefunction only has support on states $\ket{\phi_1,\phi_2,\phi_3}$ 
which avoid any indistinguishability by the group action, to derive all the invariant subspaces and their basis vectors. But what if $\chi(\phi_1,\phi_2,\phi_3)$ only has support on states for which, say, $\phi_1=\phi_2=\phi_3$? It means that the action of $D_{\sigma}$ in Eq.~\eqref{eq:six} is only that of a two-dimensional representation, composed of the trivial irrep and the sign irrep; one can see that the other vectors $\ket{\psi_{1,2}^{(k)}}$ become null vectors and drop out as possible eigenstates.

If we wish to further analyze this circuit, we could focus on these four different symmetry sectors and project the Hamiltonian onto these spaces to solve for the $\chi(\psi_1,\psi_2,\psi_3)$ etc. wavefunctions. 

An interesting working point to execute this analysis could be at 
$\phi_{\mathrm{ext}}=\pi$ in $\hamiltonian$ in~\cref{eq:hamil_tetrahedron}, where the potential has a continuous set of minima~\cite{Feigelman04} for
\begin{align}
\nodephi[1]{} = 0, \qquad \nodephi[2]{} - \nodephi[3]{} = \pi, \qquad \nodephi[3]{} \:\: \text{arbitrary}
\end{align}
Further studies of the tetrahedron Hamiltonian can be found in Ref.~\cite{Feigelman04}.

\chapter{The transmon qubit, resonators and their coupling}
\label{chap:transmon}

In this chapter, we analyze the transmon qubit in detail and discuss the mathematical treatment of resonators to which it can be coupled. At the time of writing, the transmon qubit can be considered the most successful superconducting qubit. In fact, all superconducting, multi-qubit (between, say, $10$ and $1000$ qubits) chips that are currently in use all employ transmon qubits. The transmon qubit consists of a large capacitance shunting a single Josephson junction (fixed-frequency transmon) or a SQUID (flux-tunable transmon) as shown in Fig.~\ref{fig:ft_transmon}. Thus, the transmon, from a circuit theory point of view, is essentially a CPB as described in Section~\ref{subsec:joscpb}, but operated in a specific parameter regime. 

Transmons have been called artifical atoms, as the circuit functions somewhat like an atom with electronic states which have an eletric dipole coupling to the electromagnetic field, stored, say, in a resonator. The advantage of the transmon as compared to genuine atom-light interactions (cavity QED) is that in circuit QED this coupling can be engineered to be very strong. The strength of $g/2\pi$ in the extended Jaynes-Cummings model (see Eq.~\eqref{eq:gen_jc}) is typically of the order of $100 \, \mathrm{MHz}$, so that for a transmon qubit with typical frequency of $5 \, \mathrm{GHz}$, $g/\omega=0.02$. Importantly for entangling operations, the coupling is also much stronger than the transmon and resonator decay rates, which are typically less than $1/(10 \, \mathrm{\mu s})=10^{-1} \, \mathrm{MHz}$ (strong coupling regime).

The deserved success of the transmon is mainly due to the following reasons:
\begin{itemize}
    \item The simplicity of its circuit, which comes with reliable and well-established fabrication techniques.
    \item The resilience against charge and flux noise, leading to relatively long relaxation times $T_1$ (varying from $10-100\,\mathrm{\mu s}$ or even higher, depending on material properties) and dephasing times $T_2$.
    \item The possibility to perform single- and two-qubit gates, as well as measurements with relatively simple protocols.
\end{itemize}

In what follows, we introduce commonly-used analytical and numerical methods to study transmon qubits. We start by studying its Hamiltonian and the associated spectrum. 

\begin{figure}[h]
    \centering
    \includegraphics[height=4 cm]{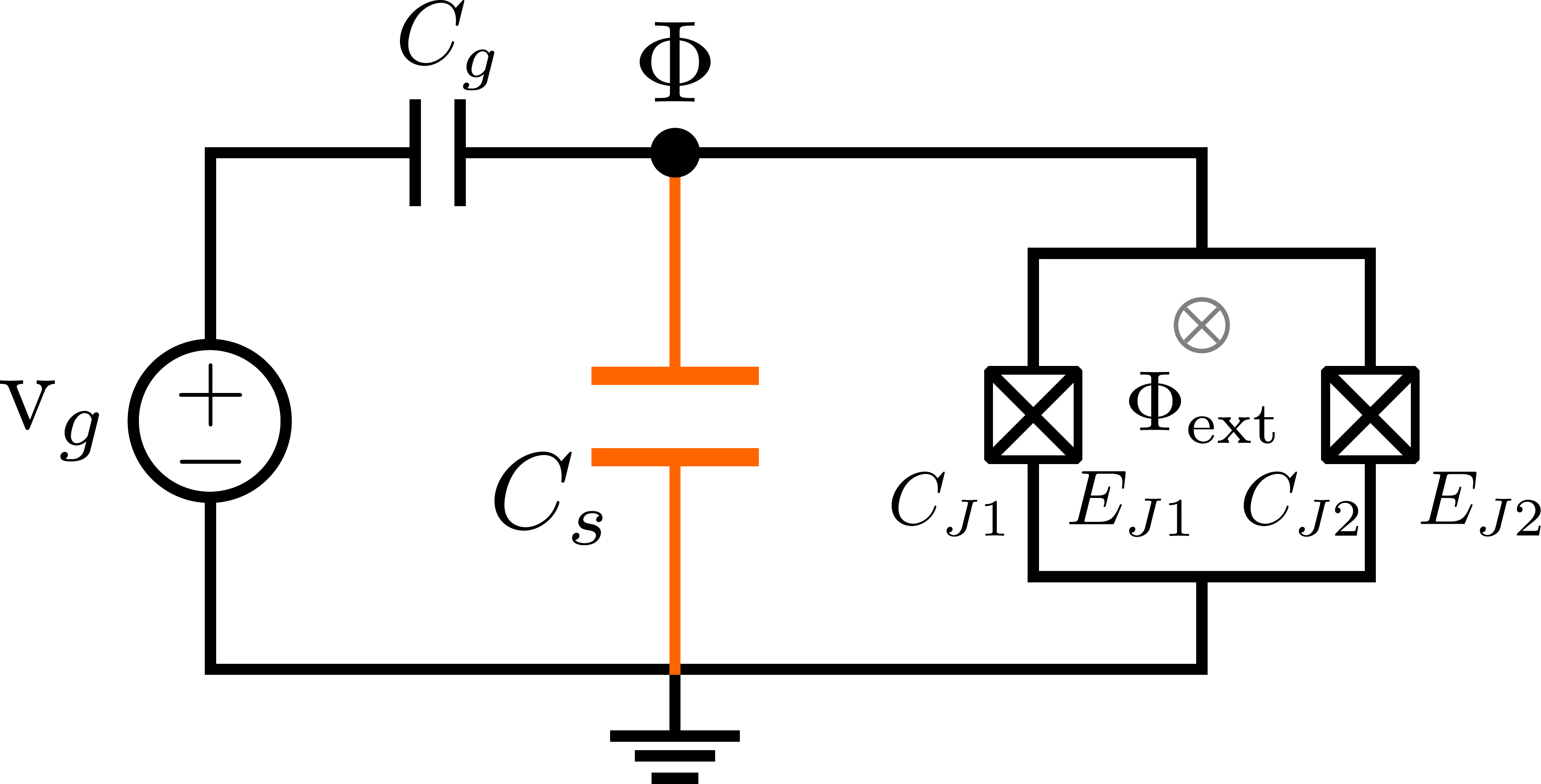}
    \caption{Circuit of a flux-tunable transmon. The voltage source models the effect of charge noise or an externally-applied voltage via, say, a transmission line. Our choice of spanning tree is highlighted in orange.}
    \label{fig:ft_transmon}
\end{figure}

\begin{figure}[h]
\centering
\includegraphics[height=5cm]{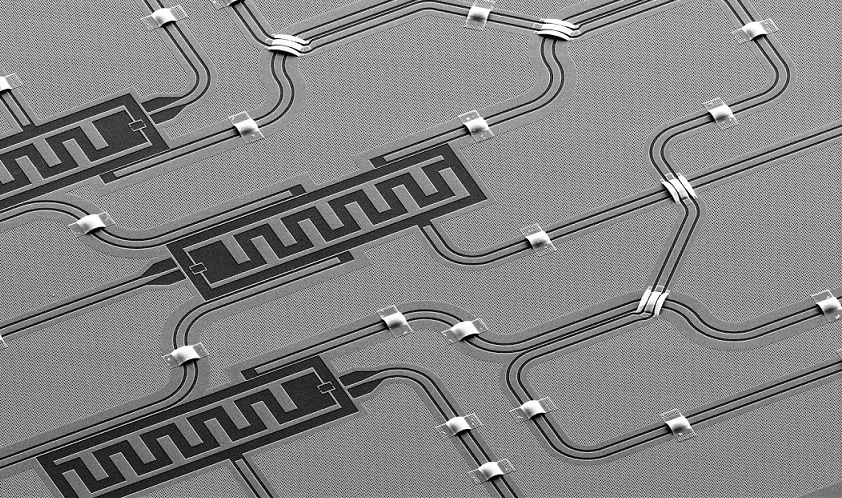}
\caption{Processor (DiCarlo lab, Delft University of Technology, 2016) with three transmon qubits each with an interdigitated capacitor (for a large shunting capacitance), in parallel with a tiny SQUID. One can see the capacitive coupling to co-planar resonators, transmission lines, and flux lines next to the SQUID. One can see airbridges ('bandages') to ensure a well-connected ground plane and cross-overs to let lines cross each other in the 2D plane. This image is of the processor in Ref.~\cite{asaad}. See also Fig.~\ref{fig:riste} for a different chip with its input lines.}
\label{fig:dicarlo-transmon-reson}
\end{figure}

\section{The CPB Hamiltonian and its spectrum}\label{sec:cpbspectrum}

We consider the circuit of a flux-tunable transmon shown in Fig.~\ref{fig:ft_transmon} in the general case in which the Josephson energies of the two junctions (SQUID) are not necessarily equal, but some asymmetry may be present. This is practically always the case due to inaccuracies of the fabrication process leading to uncertainty in $E_J$. Nonetheless, we will see that the essential form of the Hamiltonian will not change. \par 
Using the spanning tree highlighted in orange in Fig.~\ref{fig:ft_transmon}, the Lagrangian of the circuit reads
\begin{equation}\label{eq:trans_lagr}
\lagrangian = \frac{\tilde{C}}{2} \dot{\Phi}^2 + \frac{C_g}{2} \bigl( \dot{\Phi} - \vv_g)^2 + E_{J1} \cos \biggl( \frac{2 \pi}{ \Phi_0} \Phi \biggr) + E_{J2} \cos \biggl( \frac{2 \pi}{ \Phi_0}\bigl( \Phi - \Phi_{\mathrm{ext}}\bigr) \biggr),
\end{equation}
with $\tilde{C}$ the total shunting capacitance $\tilde{C} = C_s + C_{J1} + C_{J2}$. The kinetic term has two contributions: one due to the capacitance $\tilde{C}$ and the other due to the gate capacitance $C_g$.
We assume that the external flux bias $\Phi_{\mathrm{ext}}$ is a constant over time. In order to bring the Lagrangian in Eq.~\eqref{eq:trans_lagr} to a standard form with an effective Josephson energy, we perform the following change of variables: $\Phi \mapsto \Phi +\Phi_{\mathrm{ext}}/2$. After some trigonometric manipulations, the Lagrangian becomes
\begin{equation}
\lagrangian = \frac{\tilde{C}}{2} \dot{\Phi}^2 + \frac{C_g}{2} \bigl( \dot{\Phi} - \vv_g)^2 + E_{J} \cos \biggr[\frac{2 \pi}{\Phi_0} \bigl(\Phi - \Phi_{\mathrm{ext}}^{(d)} \bigr) \biggl],
\end{equation}
where the effective Josephson energy of the system is given by
\begin{equation}
E_J = E_{J}(\Phi_{\mathrm{ext}}) = (E_{J1} + E_{J2}) \cos \biggl(\frac{\pi \Phi_{\mathrm{ext}}}{\Phi_0} \biggr) \sqrt{1 + d^2 \tan^2 \biggl(\frac{\pi \Phi_{\mathrm{ext}}}{\Phi_0} \biggr)},  
\end{equation}
with $d=\frac{E_{J1} - E_{J2}}{E_{J1} + E_{J2}}$ capturing the asymmetry of the junctions, and where the flux $\Phi_{\mathrm{ext}}^{(d)}$ is defined as
\begin{equation}
 \Phi_{\mathrm{ext}}^{(d)} = \arctan \biggl[ d \tan \biggl( \frac{\pi \Phi_{\mathrm{ext}}}{\Phi_0} \biggr) \biggr].  
\end{equation}
Finally, with the additional shift $\Phi \mapsto \Phi + \Phi_{\mathrm{ext}}^{(d)}$ the Lagrangian becomes
\begin{equation}\label{eq:trans_lagr2}
\lagrangian = \frac{\tilde{C}}{2} \dot{\Phi}^2 + \frac{C_g}{2} \bigl( \dot{\Phi} - \vv_g)^2 + E_{J}\cos \biggr(\frac{2 \pi}{\Phi_0} \Phi \biggl).
\end{equation}

The Lagrangian in Eq.~\eqref{eq:trans_lagr2} is the general Lagrangian of a CPB, similar to the one in Eq.~\eqref{eq:Ltrans} with only the addition of the voltage term, as in Section~\ref{subsec:capvol}. Following identical steps as in Section~\ref{subsec:capvol}, we identify the conjugate variable 
\begin{equation}
Q = \frac{\partial \mathcal{L}}{\partial \dot{\Phi}} = C \dot{\Phi} - C_g \vv_g,
\end{equation} 
where we have defined the effective capacitance $C = \tilde{C} + C_g$. We obtain the Hamiltonian (neglecting constant terms):
\begin{equation}
\hamiltonian = \frac{(Q + C_g \vv_g)^2}{2 C} - E_{J}\cos \biggr(\frac{2 \pi}{\Phi_0} \Phi \biggl).
\end{equation}
Using dimensionless variables as defined in Eq.~\eqref{eq:dimless} we write the quantum Hamiltonian as
\begin{equation}\label{eq:h_cpb}
H = 4 E_{C} (\hat{q} + n_g)^2 - E_{J} \cos \hat{\phi},
\end{equation}
with the reduced gate charge $n_g= \frac{C_g \vv_g}{2 e}$ and charging energy $E_C$ as defined in Eq.~\eqref{eq:char_en}. Note that the operator $\hat{\phi}$ can take values in $\mathbb{R}$ and so it should not be identified directly with the phase variable $\hat{\varphi}$ as mentioned in Section~\ref{subsec:joscpb}. We will discuss this topic in the next section.

\subsection{From flux to phase}
\label{subsec:fp}

  For a superconducting material, the superconducting phase $\varphi({\bf r})$, which varies inside the material, is given in Eq.~\eqref{eq:op}. Then why do we take the dimensionless flux variable $\phi \in \mathbb{R}$ as the fundamental degree of freedom of an electrical circuit to be quantized? We do so because the electrical circuit is a course-grained description of the material: we capture how, between two points ${\bf r}_1$ and ${\bf r}_2$ in the material, the phase $\varphi(\bf{r})$ varies, {\em including} its windings. This variable thus takes values in $\mathbb{R}$ and equals the (dimensionless) branch flux $\phi_{\mathfrak{b}}$ between nodes ${\bf r}_1$ and ${\bf r}_2$. When the electrical network contains only periodic terms in the branch fluxes, one can revert back to phase variables. In this section, we take a mathematical excursion to explain this in detail: it can be viewed as working in a restricted Hilbert space.

First of all, let us introduce the translation operators in flux and charge
\begin{equation}
T_{\phi}(r) = e^{-i r \hat{q}}, \quad T_{q}(r) = e^{i r \hat{\phi}}.
\end{equation}
For any operators $A$ and $B$ it holds that
\begin{equation}\label{eq:bch}
e^{A} B e^{-A} = B + [A, B] + \frac{1}{2!} [A,[A, B]] + \frac{1}{3!} [A, A,[A, B]]] + \dots = \sum_{n=0}^{+\infty} C_n,
\end{equation}
where we define 
\begin{subequations}
\begin{align}
C_{0} &= B, \\
C_n &= [A,  C_{n-1}], \quad n \ge 1.
\end{align}
\end{subequations}
From Eq.~\eqref{eq:bch} and the commutation relation Eq.~\eqref{eq:cr}, one can readily show that
\begin{subequations}
\begin{align}
T_{\phi}(r)^{\dagger} \hat{\phi} T_{\phi}(r) &= \hat{\phi} + r, \\
T_{q}(r)^{\dagger} \hat{q} T_{q}(r) &= \hat{q} + r.
\end{align}
\end{subequations}
Using the previous equation we see that $T_{\phi}(2 \pi)^{\dagger} H T_{\phi}(2 \pi) = H$ for the CPB Hamiltonian $H$ in Eq.~\eqref{eq:h_cpb}. Thus, the translation operator $T_{\phi}(2 \pi)$ commutes with $H$, which implies that we can simultaneously diagonalize the operators $H$ and $T_{\phi}(2 \pi)$~\cite{sakurai}. This fact essentially leads to Bloch's theorem for our Hamiltonian~\cite{Ashcroft76}.  We start by rewriting the operator $\hat{q}$ as 
\begin{equation}
\hat{q} = \hat{n} + \hat{n}_{\alpha},
\end{equation} 
where the operators $\hat{n}$ and $\hat{n}_{\alpha}$ act on a charge eigenstate $\ket{q \in \mathbb{R}}$ as
\begin{equation}
\hat{n} \ket{q} = n \ket{q}= \mathrm{floor}(q) \ket{q}, \quad \hat{n}_{\alpha} \ket{q} = n_{\alpha} \ket{q} = (q\!\!\! \mod 1) \ket{q}.
\end{equation}
In this way, we can label the eigenstates of $\hat{q}$ as $\ket{q} = \ket{n, n_{\alpha}}$, satisfying $\hat{q}\ket{n, n_{\alpha}} = (n + n_{\alpha})\ket{n, n_{\alpha}}$, with $n_{\alpha} \in [0, 1)$, $n \in \mathbb{Z}$. We proceed similarly for the flux operator and write 
\begin{equation}
\hat{\phi} = \hat{\varphi} + 2 \pi \hat{k}.
\end{equation} 
We can label the eigenstates of $\hat{\phi}$ as $\ket{\phi \in \mathbb{R}} = \ket{\phi=\varphi+2\pi k}$ with $\varphi \in [0, 2 \pi)$, $k \in \mathbb{Z}$. As the Hamiltonian commutes with the translation operator $T_{\phi}(2 \pi)$, we can restrict ourselves to subspaces with fixed eigenvalues of $T_{\phi}(2 \pi)$. 

Let us consider the eigensubspace $\mathcal{H}_{n_{\alpha}}$ with eigenvalue $e^{-i 2\pi n_{\alpha}}$ of $T_{\phi}(2 \pi)$, i.e., 
\begin{equation}\label{eq:hnalpha}
\mathcal{H}_{n_{\alpha}} = \{\ket{\psi} | \quad T_{\phi}(2 \pi) \ket{\psi} = e^{-i 2 \pi n_{\alpha}} \ket{\psi} \},
\end{equation} 
which we will call a rotor subspace. Note that the condition specified by Eq.~\eqref{eq:hnalpha} can be interpreted as a boundary condition on the wavefunction. A basis for the rotor subspace is given by the (unnormalized) states
\begin{equation}
\ket{\varphi, n_{\alpha}} = \sum_{k \in \mathbb{Z}} e^{i 2 \pi n_{\alpha}k} \ket{\phi=\varphi+2\pi k},
\end{equation}
since
\begin{multline}
T_{\phi}(2 \pi)\ket{\varphi, n_{\alpha}} = e^{-i 2\pi \hat{q}} \sum_{k \in \mathbb{Z}} e^{i 2 \pi n_{\alpha} k} \ket{\varphi+2\pi k} 
= \sum_{k \in \mathbb{Z}} e^{i 2 \pi n_{\alpha} k} \ket{\varphi+ 2\pi(k+1)} = e^{-i 2 \pi n_{\alpha}} \ket{\varphi, n_{\alpha}}.
\end{multline}
Note also that the projector onto this subspace is
\begin{equation}\label{eq:rotor_proj}
\Pi_{n_{\alpha}} = \int_{0}^{ 2 \pi} d \varphi \ket{\varphi, n_{\alpha}}\bra{\varphi, n_{\alpha}},
\end{equation}
and that by integrating over all these projectors we obtain the standard identity on the total Hilbert space
\begin{equation}
\int_{0}^1 d n_{\alpha} \Pi_{n_{\alpha}} = \int_{-\infty}^{+ \infty} d \phi \ket{\phi}\bra{\phi} = \mathds{1}.
\end{equation}
It is also useful to define a \emph{charge} basis for the rotor subspaces. In fact, since the translation operator $T_{\phi}(2 \pi)$ has eigenvalues $e^{-i 2 \pi n_{\alpha}}$, the subspace $\mathcal{H}_{n_{\alpha}}$ must be spanned by eigenkets of the operator $\hat{n}$ with eigenvalue $n + n_{\alpha}$, i.e., the states $\ket{n, n_{\alpha}}$. As a consequence, the projector onto the rotor subspaces given in Eq.~\eqref{eq:rotor_proj} can also be written as
\begin{equation}\label{eq:proj_n}
\Pi_{n_{\alpha}} = \sum_{n \in \mathbb{Z}} \ket{n, n_{\alpha}}\bra{n, n_{\alpha}},
\end{equation}
a result that one can derive mathematically by writing
\begin{multline}\label{eq:nc_basis}
\ket{n, n_{\alpha}}= \ket{n + n_{\alpha}} = \frac{1}{\sqrt{2 \pi}} \int_{- \infty}^{\infty} d \phi \,e^{i (n + n_{\alpha}) \phi} \ket{\phi} \\ = \frac{1}{\sqrt{2 \pi}} \int_{0}^{2 \pi} d \varphi \sum_{k \in \mathbb{Z}} e^{i (n + n_{\alpha}) (\varphi + 2 \pi k)} \ket{\varphi+2\pi k} = \frac{1}{\sqrt{2 \pi}} \int_{0}^{2 \pi} d \varphi e^{i (n + n_{\alpha}) \varphi}\ket{\varphi, n_{\alpha}},
\end{multline}
and using the following representation of the Dirac comb
\begin{equation}
\sum_{k \in \mathbb{Z}} e^{i k x} = 2 \pi \sum_{k \in \mathbb{Z}} \delta (x -2 \pi k). 
\end{equation}
From this it also follows that the states $\ket{\varphi, n_{\alpha}}$ can be written as a linear combination of the $\ket{n, n_{\alpha}}$ states:
\begin{equation}
\ket{\varphi, n_{\alpha}} = \frac{1}{\sqrt{2 \pi}} \sum_{n \in \mathbb{Z}} e^{-i (n+n_{\alpha}) \varphi} \ket{n, n_{\alpha}}. 
\end{equation}

\subsection{CPB spectrum}
\label{subsec:cpbspectrum}

We now study the eigenvalue problem associated with the CPB Hamiltonian in Eq.~\eqref{eq:h_cpb}.
To do so, we can project onto each rotor subspace to get the projected Hamiltonians
\begin{equation}
H_{n_{\alpha}} = \Pi_{n_{\alpha}} H\, \Pi_{n_{\alpha}},
\end{equation}
parametrized by 
$n_{\alpha}$. Once we project onto this subspace, the problem is equivalent to that of a pendulum or rotor, where the discrete charge $\hat{n}$ plays the role of the angular momentum, and the phase $\hat{\varphi}$ the role of the angle. In the discrete charge basis defined in Eq.~\eqref{eq:nc_basis}, using Eq.~\eqref{eq:proj_n}, and since
\begin{equation}
e^{i \hat{\varphi}} \ket{n, n_{\alpha}} =  \ket{n + 1, n_{\alpha}},
\label{eq:tun}
\end{equation}
we get 
\begin{multline}\label{eq:h_cpb_n}
H_{n_{\alpha}} = \sum_{n \in \mathbb{Z}} 4 E_{C} (n + n_{\alpha} + n_{g})^2 \ket{n, n_{\alpha}}\bra{n, n_{\alpha}} 
- \frac{E_J}{2} \sum_{n \in \mathbb{Z}} ( \ket{n+1, n_{\alpha}}\bra{n, n_{\alpha}} + \ket{n, n_{\alpha}}\bra{n+1, n_{\alpha}}).
\end{multline}
In the last term here we see that the operator $\cos(\hat{\varphi})$ of $\Pi_{n_{\alpha}}H \Pi_{n_{\alpha}}$ acts as a Cooper pair tunneling process, with $n$ being the difference in number of Cooper pairs on one of the superconducting islands. Eq.~\eqref{eq:h_cpb_n} also shows that the effect of the offset charge $n_g$ is equivalent to changing the rotor subspace. The spectrum of the Hamiltonian $H_{n_{\alpha}}$ restricted to the rotor subspace $\mathcal{H}_{n_{\alpha}}$ depends on $n_g$. 
The spectrum of the full Hamiltonian is the union of all the spectra obtained within the rotor subspaces and it {\em does not depend} on $n_g$. Eq.~\eqref{eq:h_cpb_n} also shows that given a fixed $n_{\alpha} \in [0, 1)$, the spectrum of $H_{n_{\alpha}}$ as a function of $n_g$ must be periodic with period $1$. Eq.~\eqref{eq:h_cpb_n} immediately suggests a numerical method to diagonalize $H_{n_{\alpha}}$: we can simply consider a finite number of charge states $\ket{n, n_{\alpha}}$ and obtain a finite matrix that we can numerically diagonalize \footnote{Of course, these states shall be taken to be symmetric around the state $\ket{n, n_{\alpha} }$ for which $(n+n_{\alpha} + n_g)^2$ is at a minimum.}. 

The energy levels of $H_{n_{\alpha}}$ for $n_{\alpha}=0$ as a function of $n_g$ are shown in Fig.~\ref{fig:energyng}. These plots can also be interpreted as the band structure of the CPB, given the equivalence between $n_g$ and $n_{\alpha}$. The quantum information is usually encoded in the first two levels (bands), which define the qubit subspace. One immediately notices that the energy levels are less sensitive to the parameter $n_g$ as the parameter $E_J/E_C$ is increased, and they are practically insensitive for $E_J/E_C = 50$. Thus, in this parameter regime of the CPB, even the low-lying levels of the Hamiltonians $H_{n_{\alpha}}$ show a very weak dependency on $n_g$. This is clearly a desirable property of the system, since the dephasing rate associated with charge noise, i.e., noise in the parameter $n_g$, is connected to derivatives of the energy levels as a function of this parameter~\cite{koch2007} (see also the concept of ``sweet spots" in Section~\ref{subsec:fluxss}). This is the important feature that makes the transmon resilient to charge noise. The qualitative reason behind this phenomenon is that for $E_J/E_C \gg 1$ the system behaves more and more as a \emph{heavy} harmonic oscillator whose low-lying eigenstates are well localized in the minimum of the cosine potential. This implies that at least these low-lying levels \emph{do not see} the effect of the boundary condition in Eq.~\eqref{eq:hnalpha}, since the wavefunction quickly decays to zero away from the potential minimum. As we will see in Section~\ref{sec:tr_approx}, this justifies the so-called transmon approximation in which we treat the CPB as an anharmonic Duffing oscillator for $E_J/E_C \gg 1$. As a side effect, the system will be weakly anharmonic as appears from Fig.~\ref{fig:energyng_d}, which leads to the problem of leakage out of the computational qubit subspace.\\     \par 
\begin{figure}[htb]
\centering
\begin{subfigure}[t]{0.45 \textwidth}
\includegraphics[scale=0.3]{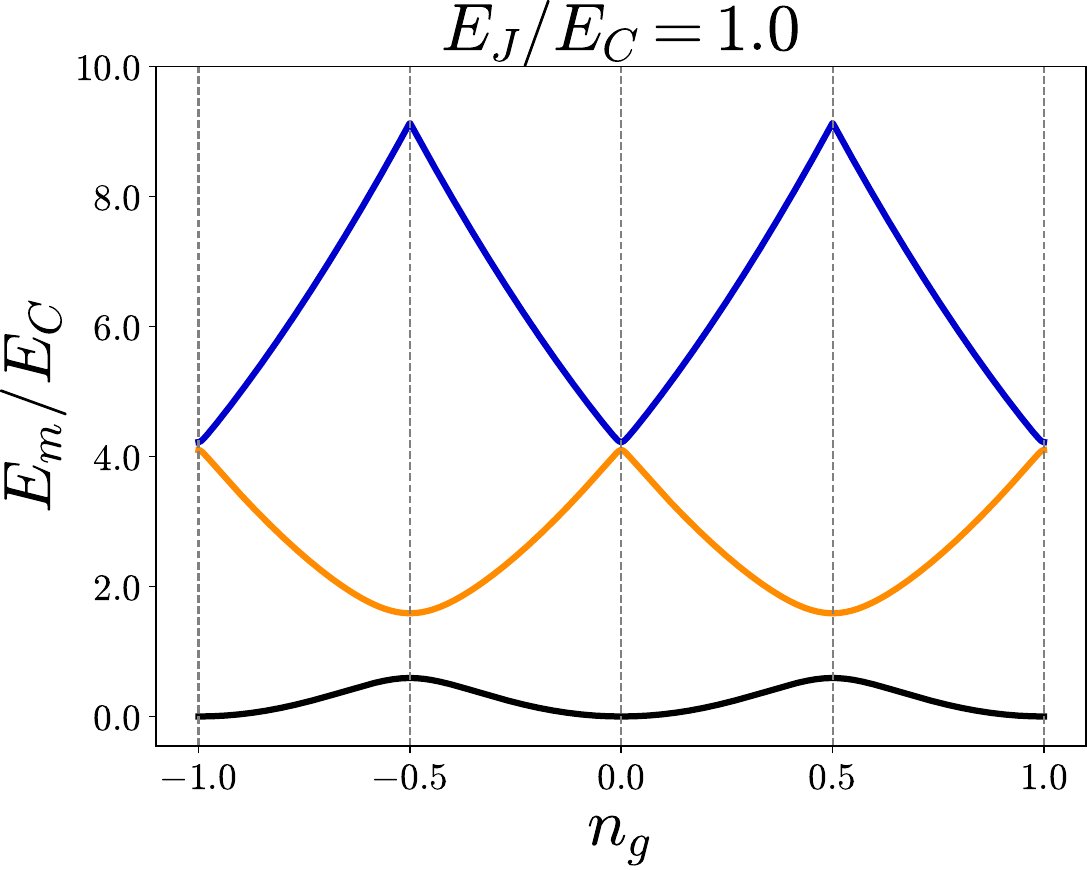}
\subcaption{}
\end{subfigure}
\begin{subfigure}[t]{0.45\textwidth}
\includegraphics[scale=0.3]{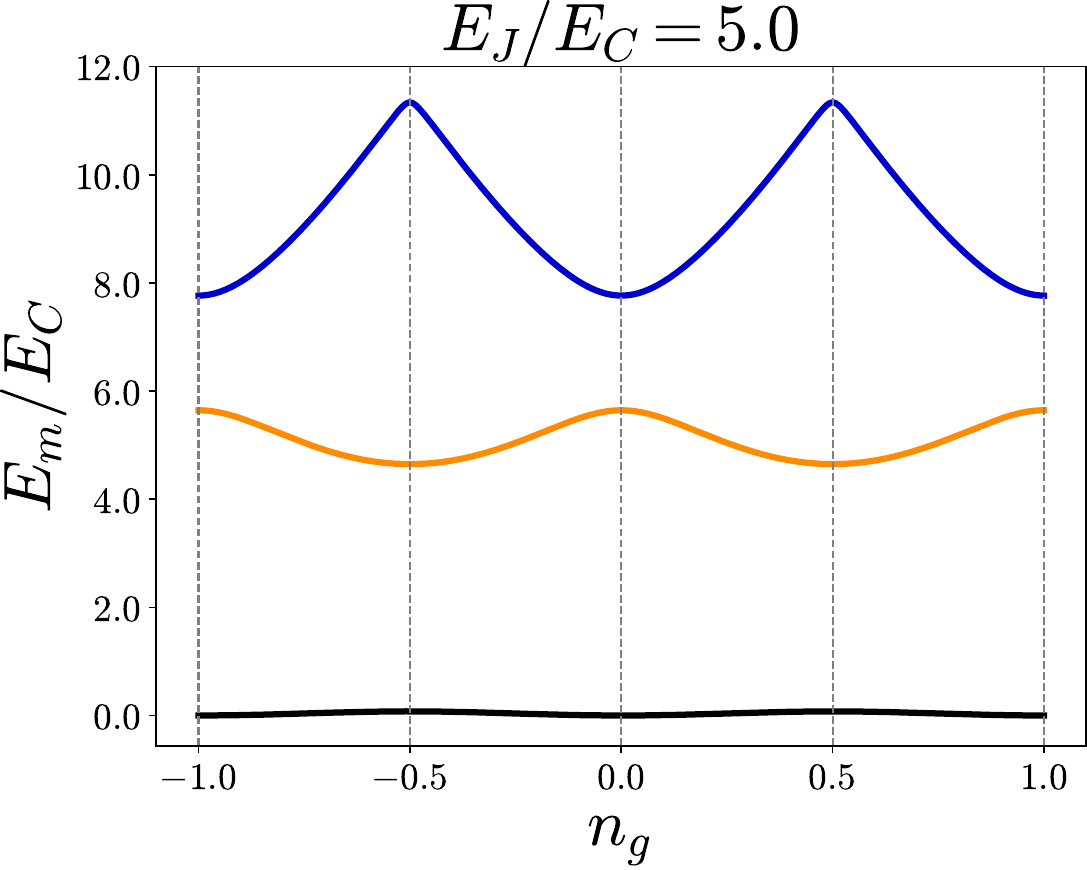}
\subcaption{}
\end{subfigure}
\begin{subfigure}[t]{0.45 \textwidth}
\includegraphics[scale=0.3]{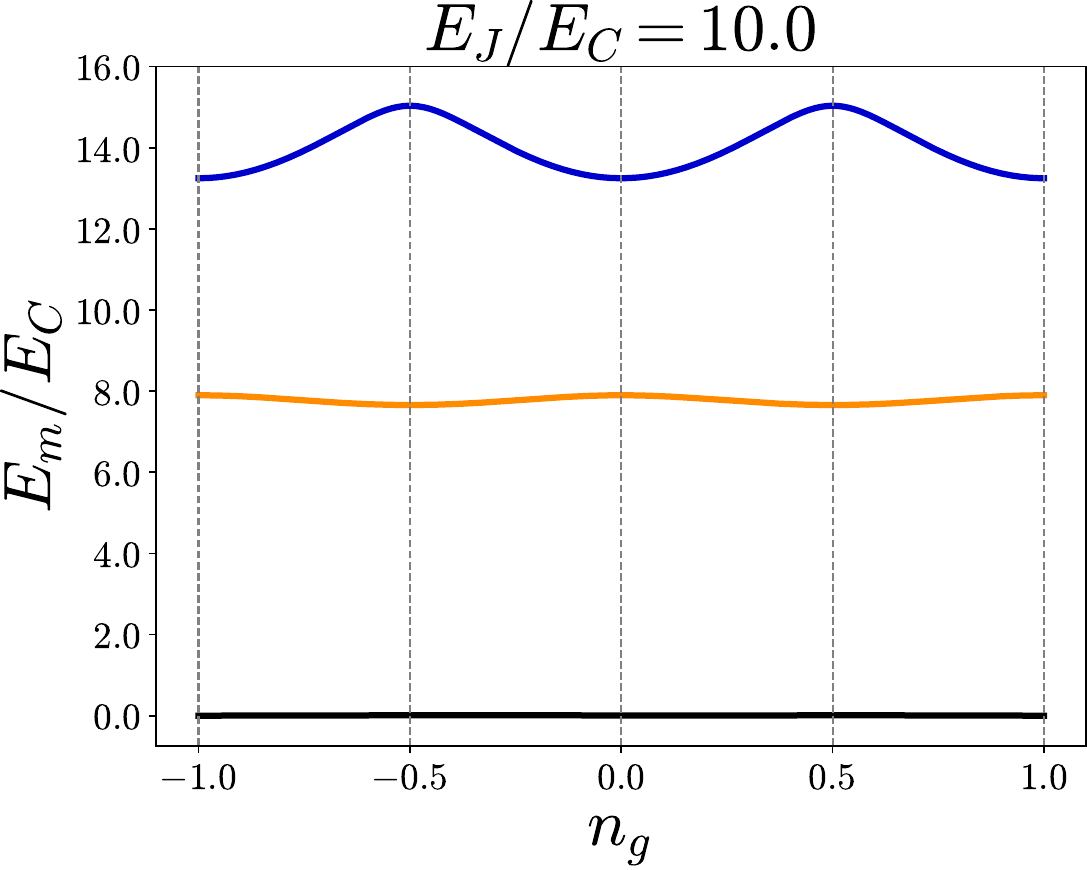}
\subcaption{}
\end{subfigure}
\begin{subfigure}[t]{0.45 \textwidth}
\includegraphics[scale=0.3]{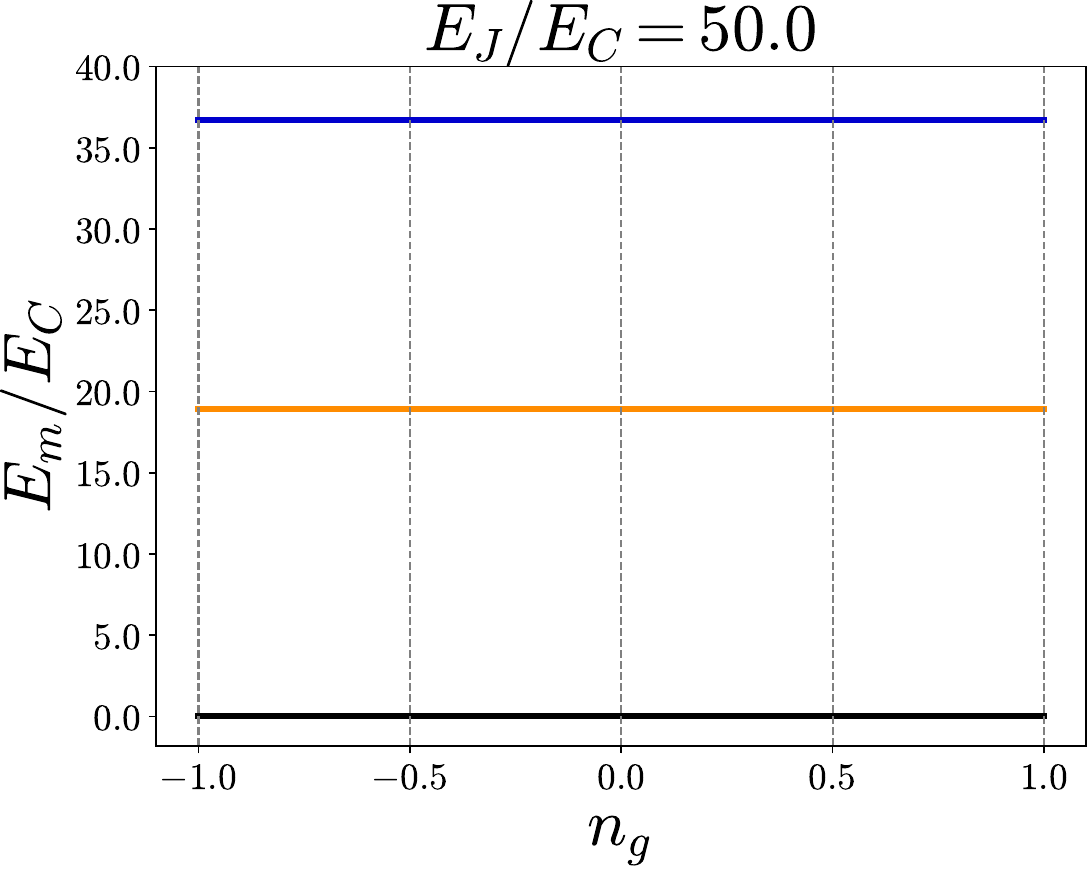}
\subcaption{}
\label{fig:energyng_d}
\end{subfigure}
\caption{First three energy levels of the CPB Hamiltonian restricted to the rotor subspace Eq.~\eqref{eq:h_cpb_n} with $n_{\alpha}=0$ as a function of the offset charge $n_g$ and for different values of the parameter $E_J/E_C$. In each plot the zero of the energy is taken to be the minimum of the energies of the lowest level. Note how for large $E_J/E_C$, e.g., $E_J/E_C = 50.0$, the spectrum loses its dependency on $n_g$, at the price of being only weakly anharmonic.}
\label{fig:energyng}
\end{figure} 

It is worth discussing the eigenvalue problem in the phase basis $\ket{\varphi, n_{\alpha}}$, which reads
\begin{equation}
H_{n_{\alpha}} \ket{\varphi, n_{\alpha}} = E \ket{\varphi, n_{\alpha}}  \implies \biggl[4 E_C \biggl(-i \frac{d}{d \varphi} + n_{\alpha} + n_g \biggr)^2 - E_{J}\cos \varphi \biggr] \psi_{n_{\alpha}}(\varphi) = E \psi_{n_{\alpha}}(\varphi), 
\end{equation}
with boundary condition
\begin{equation}
T_{q}(2 \pi) \ket{\varphi, n_{\alpha}} = e^{-i 2 \pi n_{\alpha}} \ket{\varphi, n_{\alpha}}  \implies \psi_{n_{\alpha}}(\varphi + 2 \pi) = e^{-i 2 \pi n_{\alpha}} \psi_{n_{\alpha}}(\varphi).
\end{equation}
This eigenvalue problem has an analytical solution in terms of Mathieu functions~\cite{koch2007}. In particular, for $n_g \in [0, 1/2)$ and $n_{\alpha} = 0 $, the eigenenergies are given by 
\begin{equation}
E_m(n_g) = E_C \mathcal{M}_A\biggl(k_f(m, n_g), - \frac{E_{J}}{2 E_C} \biggr), \quad m \in \mathbb{N},
\label{eq:mathieu}
\end{equation}
where $\mathcal{M}_A(r,q)$ is the Mathieu characteristic value for even Mathieu functions \footnote{We use the same definition as the Wolfram Mathematica software.} and where we defined
\begin{equation}
k_f(m, n_g) = m + 1 - (m+1 \mod 2) + 2 n_g (-1)^m.
\end{equation}
Note that from the previous equation we can reconstruct the eigenenergies for any $n_{\alpha}\in[0, 1)$. 

\subsubsection{Dynamics in the phase basis}

At this point we should ask ourselves how we should treat a CPB. Shall we assume that the system is allowed to explore the whole oscillator Hilbert space or should we work in a specific subspace $\mathcal{H}_{n_{\alpha}}$? The answer is that it actually depends on how the system is coupled to other circuit elements, and also on the history of the CPB. If the CPB is never coupled to circuit elements that do not commute with the operator $T_{\phi}(2 \pi)$, which can cause mixing between the rotor subspaces, then the state of the CPB will naturally be confined to one of these subspaces. This is what happens when we cool down a simple CPB, which is capacitively coupled to other parts of the circuit. Since the system is confined to the rotor subspace, one should accordingly consider a thermal state defined in this subspace. Also, assuming that one works in the regime $E_J/E_C \gg 1$ it is not relevant to know in which specific rotor subspace $\mathcal{H}_{n_{\alpha}}$ we are working, since the spectrum will have a small dependency on it. If $E_{J}/E_C \lesssim 1$ instead, we can understand in which subspace $\mathcal{H}_{n_{\alpha}}$ we are if we know $n_g$ and we are able to do spectroscopy. Practically one can always see the problem as the one defined in $\mathcal{H}_{n_{\alpha}=0}$ with a shifted $n_g$. This is the reason why in usual treatments of the CPB one usually assumes immediately that $n_{\alpha} = 0$. \par 

Nevertheless, there are situations in which there can be mixing between the rotor subspaces and we should consider states that are normalized in the oscillator Hilbert space. We provide two examples in Fig.~\ref{fig:tr_li}. In Fig.~\ref{fig:tr_li_a} a CPB is biased by a current source. We have already encountered this circuit in Section~\ref{subsec:cs}, where we saw that it gives rise to the washboard potential shown in Fig.~\ref{fig:washboard_potential}. The current source adds a term proportional to $I \hat{\phi}$ to the CPB Hamiltonian which does not commute with $T_{\phi}(2 \pi)$. In this case, the charge on the capacitor is not forced to assume discrete values and the whole oscillator Hilbert space can be explored. Another example is provided in Fig.~\ref{fig:tr_li_b}, where the CPB is shunted by an inductance, which adds a term $E_L\hat{\phi}^2/2$. Depending on the parameter regime the system would be called a flux qubit or a fluxonium. Again, the inductance allows an arbitrary charge on the capacitor, and the system would have solutions spread out in the oscillator space. Let us suppose that the system in the presence of the inductor is in a state $\ket{\psi}$, which is a superposition of states lying in different rotor spaces labeled by $n_{\alpha}$. At this point we suddenly switch off the inductance. The system is now again a CPB, but the initial state is not confined to one rotor subspace. Its time dynamics can be described in parallel in each rotor subspace and it does not mix states in different rotor subspaces. This case is also discussed in Ref.~\cite{le2020}.

\begin{figure}[htb]
\centering
\begin{subfigure}[t]{0.45 \textwidth}
\centering
\includegraphics[height=4 cm]{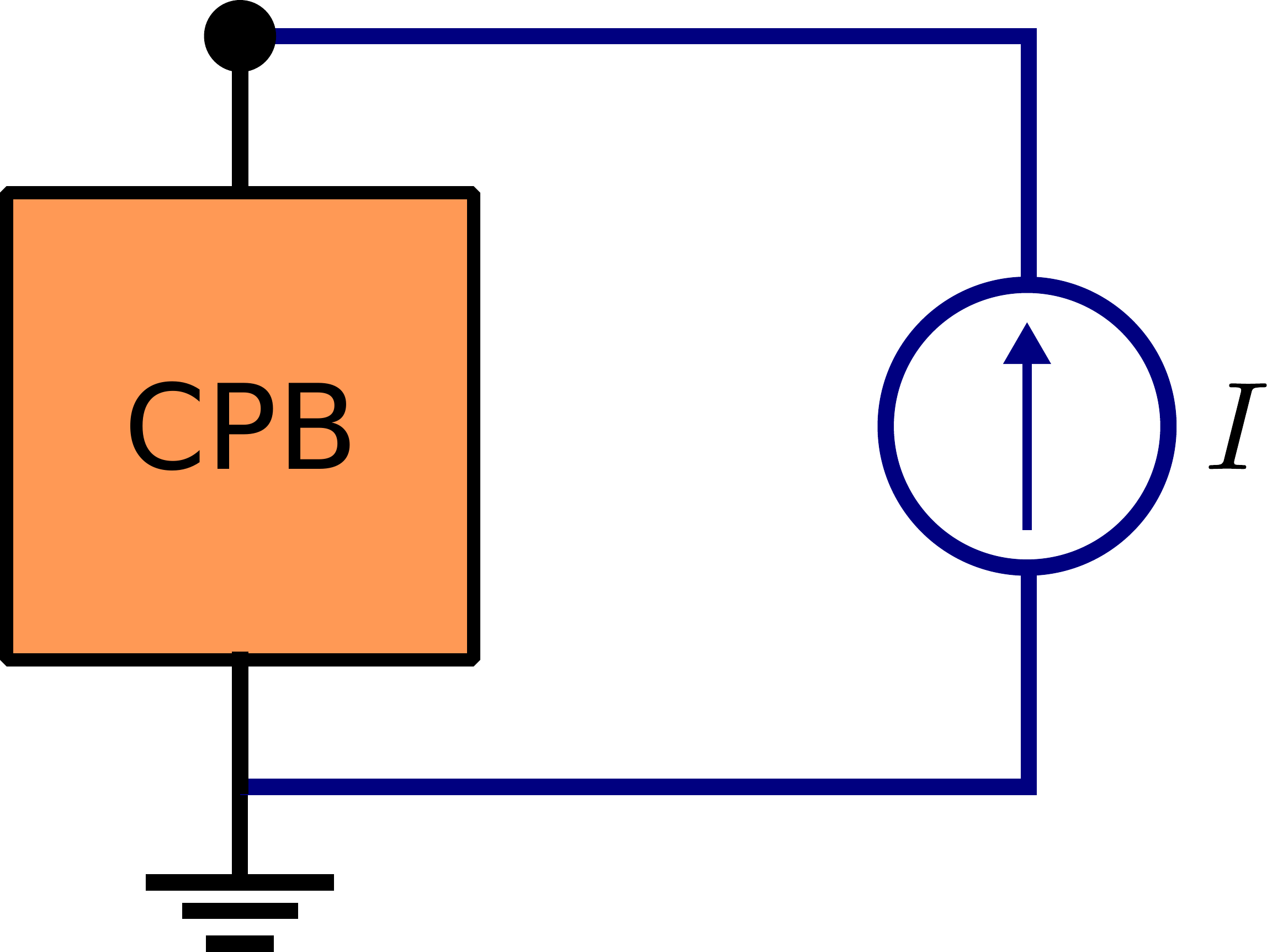}
\subcaption{}
\label{fig:tr_li_a}
\end{subfigure}
\begin{subfigure}[t]{0.45\textwidth}
\centering
\includegraphics[height=4 cm]{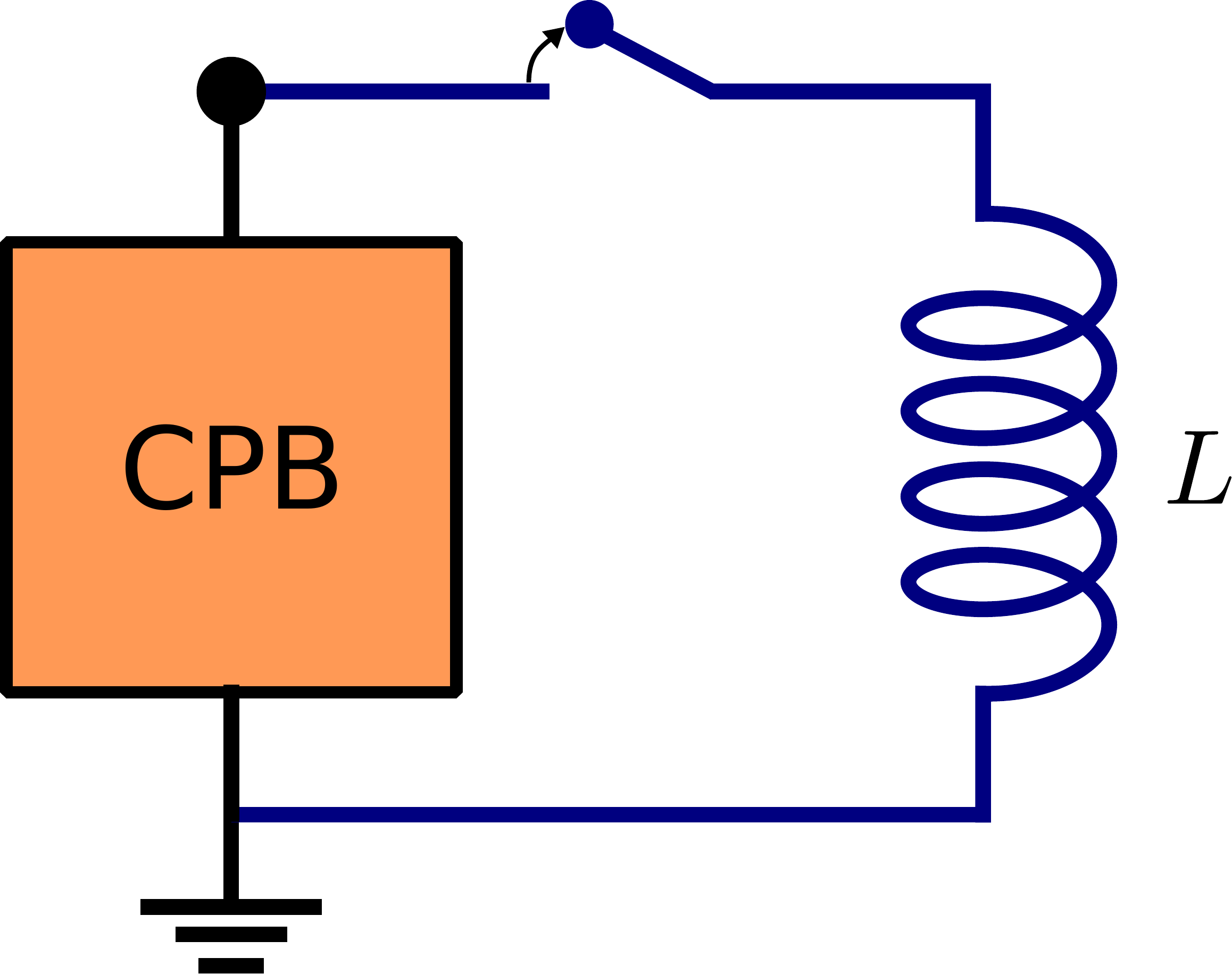}
\subcaption{}
\label{fig:tr_li_b}
\end{subfigure}
\caption{Examples of circuits that cause mixing between the rotor subspaces: (a) CPB with a current source. (b) CPB with a switchable inductance.}
\label{fig:tr_li}
\end{figure} 

\section{The anharmonic approximation}
\label{sec:tr_approx}

In the previous section, we have formulated in detail the eigenvalue problem of the CPB. We have also argued that in the regime $E_J/E_C \gg 1$ the system behaves as a harmonic oscillator (at low energies) and its energy levels show a weak dependency on the parameter $n_g$. This observation suggests an approximation that we discuss here. Let us consider the Hamiltonian in Eq.~\eqref{eq:h_cpb} and let us immediately neglect $n_g$. In the regime $E_J/E_C \gg 1$ the low-energy wavefunctions (in each rotor subspace) will be localized very close to $\phi = 0$. It is thus reasonable to Taylor expand the cosine around $\phi=0$, or equivalently around any $\phi = 2 \pi k$, $k \in \mathbb{Z}$. In order to take into account nonlinear effects we Taylor expand up to fourth order, to obtain the transmon Hamiltonian
\begin{equation}\label{eq:h_trans}
H_{\rm tr} = 4 E_C \hat{q}^2 + \frac{E_J}{2} \hat{\phi}^2 - \frac{E_J}{24}\hat{\phi}^4,
\end{equation} 
where we also omitted constant terms. We introduce annihilation and creation operators that diagonalize the quadratic part of the Hamiltonian as
\begin{subequations}
\label{eq:secqu}
\begin{equation}
\hat{\phi}= \biggl(\frac{2 E_C}{E_J}\biggr)^{\nicefrac{1}{4}}(\hat{b}+\hat{b}^{\dagger}),
\end{equation}
\begin{equation}
\hat{q}= \frac{i}{2}\biggl (\frac{E_J}{2 E_C}\biggr)^{\nicefrac{1}{4}}(\hat{b}^{\dagger}-\hat{b}),
\end{equation}
\end{subequations}
with $[\hat{b}, \hat{b}^{\dagger}] = \mathds{1}$. For completeness, we also write the original operators $\hat{\Phi}$ and $\hat{Q}$ in terms of annihilation operators as
\begin{subequations}
\begin{equation}
\hat{\Phi} = \sqrt{\frac{\hbar Z_t}{2}} (\hat{b} + \hat{b}^{\dagger}), 
\end{equation}
\begin{equation}
\label{eq:zpf_trans}
\hat{Q} =  i \sqrt{\frac{\hbar }{2 Z_t}} (\hat{b}^{\dagger} - \hat{b}),\end{equation}
\end{subequations}
with $Z_t = \sqrt{L_J/C}$ the characteristic impedance of the transmon. The introduction of annihilation and creation operators here is of course completely equivalent to the case of an LC oscillator discussed in Section~\ref{subsec:lc}. 

Substituting Eqs.~\eqref{eq:secqu} into Eq.~\eqref{eq:h_trans}, and omitting again constant terms, one gets
\begin{equation}\label{eq:h_tr4}
H_{\rm tr}= \sqrt{8 E_J E_C}\hat{b}^{\dagger} \hat{b} -\frac{E_C}{12}(\hat{b}^{\dagger}+\hat{b})^4.
\end{equation}
We now perform a further approximation. The quartic term in Eq.~\eqref{eq:h_tr4} can be seen as a perturbation on top of the harmonic oscillator Hamiltonian $\sqrt{8 E_J E_C}\hat{b}^{\dagger} \hat{b}$. Since we are assuming $E_J/E_C \gg 1$, then also
$\frac{E_C}{\sqrt{8 E_J E_C}} \ll 1$ and we can treat the perturbation using lowest-order perturbation theory, which amounts to keeping only diagonal terms from $(\hat{b} + \hat{b}^{\dagger})^4$, i.e., terms with an equal number of annihilation and creation operators, which do not change the eigenstates of the unperturbed oscillator Hamiltonian. This approximation is sometimes also called the Rotating Wave Approximation (RWA), since if we work in the interaction or rotating-frame picture it amounts to neglecting fast-rotating terms. 

Keeping only diagonal terms in the expansion of $(\hat{b}^{\dagger}+\hat{b})^4$, and iteratively applying the commutation relation $[\hat{b}, \hat{b}^{\dagger}] = \mathds{1}$, we obtain
\begin{equation}
(\hat{b}^{\dagger}+\hat{b})^4 \overset{\rm RWA}{\approx} 6 \hat{b}^{\dagger} \hat{b}^{\dagger} \hat{b} \hat{b}+12 \hat{b}^{\dagger} \hat{b}.
\end{equation}
Hence, we approximate the transmon Hamiltonian with the Hamiltonian of a so-called Duffing oscillator
\begin{equation}
\label{eq:hduff}
H_{\rm tr} \overset{\rm RWA}{\approx} H_{\mathrm{Duffing}}= \hbar \Omega \hat{b}^{\dagger} \hat{b} + \hbar \frac{\delta}{2} \hat{b}^{\dagger} \hat{b}^{\dagger} \hat{b} \hat{b},
\end{equation}
where, denoting by $E_n$ the eigenenergies associated with the Fock state $\ket{n}$, we have defined the transmon frequency
\begin{equation}
\Omega= \frac{E_1-E_0}{\hbar} = \frac{ \sqrt{8 E_J E_C}-E_C}{\hbar},
\label{eq:freqt}
\end{equation} 
and the anharmonicity
\begin{equation}
\delta= \frac{(E_2 - E_1)- (E_1 - E_0)}{\hbar}=-\frac{E_C}{\hbar}.
\label{eq:anharmont}
\end{equation}  
It is worth pointing out that the Hamiltonian in Eq.~\eqref{eq:hduff} is in diagonal form, which allows us to immediately obtain all the energy levels. Thus, for $E_J/E_C \gg 1$ the anharmonicity of the CPB is negative, which one can also notice in Fig.~\ref{fig:energyng_d}. The relative anharmonicity $\delta_r = \delta/\Omega$ quantifies how harmonic the system is, and as expected $\delta_r \rightarrow 0$ for $E_J/E_C \rightarrow \infty$, corresponding to completely harmonic behaviour. 

If the system is too harmonic, the transmon suffers from the problem of leakage. If we start from the ground state and apply a drive at angular frequency approximately $\Omega$, the system would not be confined to the qubit subspace, but higher levels will be populated as well. In particular, a driven harmonic oscillator evolves to a coherent state $\ket{\alpha}$ \cite{book:Gerry.Knight:QuantumOptics}:
\begin{equation}
    \ket{\alpha}=e^{-|\alpha|^2/2}\sum_{n=0}^{\infty}\frac{\alpha^n}{\sqrt{n!}}\ket{n},
    \label{eq:coh-state}
\end{equation}
where the coherent amplitude $\alpha \in \mathbb{C}$ depends on drive amplitude, frequency and phase. This leakage problem, which can also occur due to two-qubit gates, qubit measurement or cross-qubit driving, is a serious issue, since standard qubit-based quantum error correcting codes, such as the surface code~\cite{Terhal.2013:QECReview}, are not designed to correct against this kind of errors. 

Transmon qubits are roughly operated in the regime $40 \le E_J/E_C \le 100$, for which we have a compromise between charge noise sensitivity and anharmonicity.  In addition, the transmons are designed to have frequencies in the microwave regime, i.e., $\Omega/2 \pi = 4$-$8 \, \mathrm{GHz}$, while the charging energy is usually chosen to be $E_C/h = 200$-$300 \, \mathrm{MHz}$, corresponding to an effective capacitance $C$ of $C=60-80$ fF.

\begin{Exercise}[title={CPB in the charge basis},label=exc:coop]
Consider the Hamiltonian of the CPB in the discrete charge basis as in Eq.~\eqref{eq:h_cpb_n}. For convenience, we will set $n_{\alpha}=0$.
\Question Considering only a finite number of discrete charge states, and taking them to be symmetric around the state with zero charge, obtain the eigen-energies numerically reproducing the plots in Fig.~\ref{fig:energyng_d}.
\emph{Hint: include at least 21 charge states. You can take $E_C=1$ as your unit of energy.}
\Question Suppose that we work at the \emph{sweet spot} $n_g=1/2$. Plot the relative anharmonicity $\delta_r=(E_{12}-E_{01})/E_{01}$, with $E_{ij}=E_j-E_i$ as a function of $E_J/E_C \in [5, 80]$. What happens when you increase the ratio $E_J/E_C$?
\end{Exercise}

\begin{Answer}[ref={exc:coop}]
\Question See Fig.~\ref{fig:energyng_d}.
\Question The plot is shown in Fig.~\ref{fig:relAnhFig}. In the large $E_J/E_C$ limit, the anharmonicity approaches zero through negative values. This means that the system behaves more and more as a harmonic oscillator.
\begin{center}
\includegraphics[width=8cm]{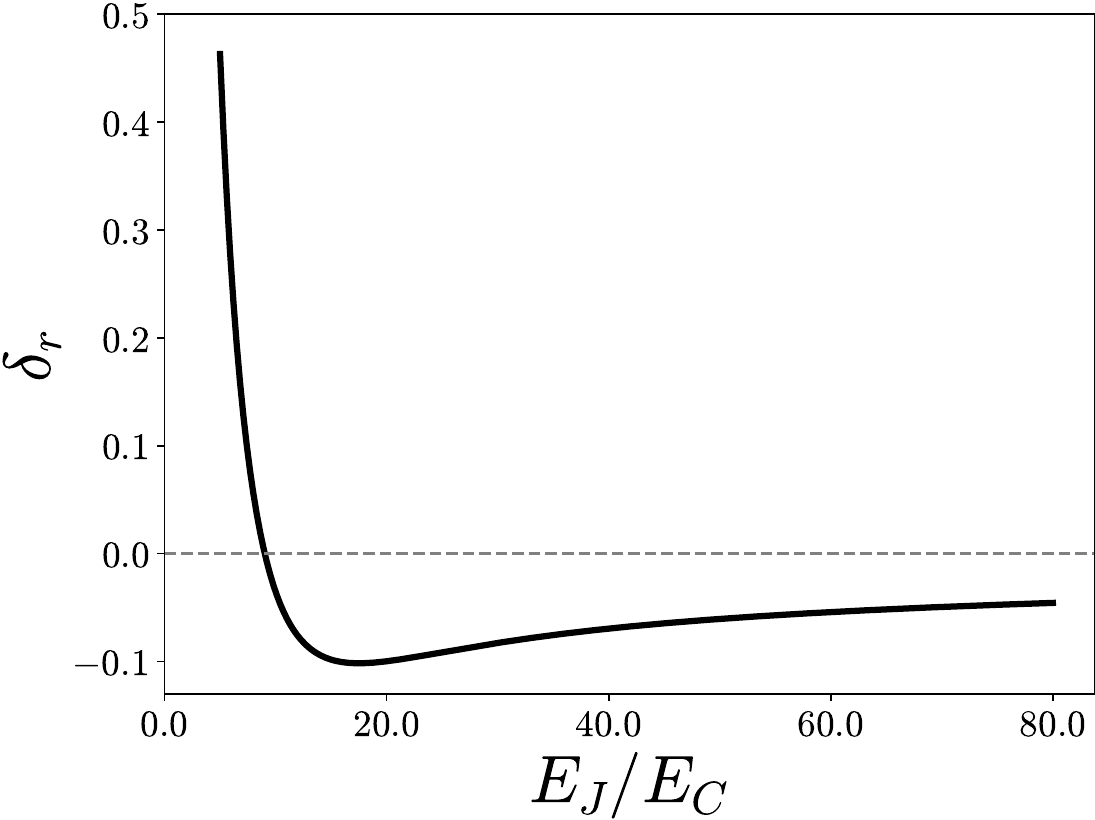}
\captionof{figure}{Relative anharmonicity of the CPB as a function of $E_J/E_C$ at $n_g=1/2$ and $n_{\alpha}=0$ for the Hamiltonian in Eq.~\eqref{eq:h_cpb_n}.}
\label{fig:relAnhFig}
\end{center}
\end{Answer}

\section{Driving a transmon qubit}
\label{sec:tr_drive}

Let us consider a CPB as in Fig.~\ref{fig:ft_transmon}, where now the voltage source $\vv_g \rightarrow \vv_d(t)$ is time-dependent. The system is coupled to it via a capacitance $C_g \rightarrow C_d$. Assuming that the CPB is operated in the transmon regime and performing the approximations we described in Section~\ref{sec:tr_approx}, we obtain the Hamiltonian
\begin{equation}
\frac{H(t)}{\hbar} = \Omega \hat{b}^{\dagger} \hat{b} + \frac{\delta}{2} \hat{b}^{\dagger} \hat{b}^{\dagger} \hat{b} \hat{b}+ i \varepsilon_d(t)(\hat{b}^{\dagger} - \hat{b}),
\end{equation}
where we have defined
\begin{equation}
\varepsilon_d(t) = 4 \frac{E_C}{\hbar} \biggl(\frac{E_J}{2 E_C} \biggr)^{\nicefrac{1}{4}} \frac{C_d \vv_d(t)}{2 e}.
\end{equation}
In order to drive transitions between the ground and first excited state, we would ideally like a drive with a single frequency~$\Omega_d$. Thus, let us consider a voltage drive of the form 

\begin{equation}
\label{eq:vdriveform}
\vv_d(t) = \vv_d^{\mathrm{max}} \cos(\Omega_d t).
\end{equation}
In this case, we obtain
\begin{equation}
\frac{H(t)}{\hbar}  =  \Omega \hat{b}^{\dagger} \hat{b} +  \frac{\delta}{2} \hat{b}^{\dagger} \hat{b}^{\dagger} \hat{b} \hat{b} + \frac{i}{2} \mathcal{E} \bigl(e^{i \Omega_d t} + e^{-i \Omega_d t}\bigr)(\hat{b}^{\dagger} - \hat{b}),
\end{equation}
with $\mathcal{E} = \varepsilon_d (0)$.

\begin{figure}
\centering
\begin{subfigure}[h]{0.4 \textwidth}
\centering
\includegraphics[width=5cm]{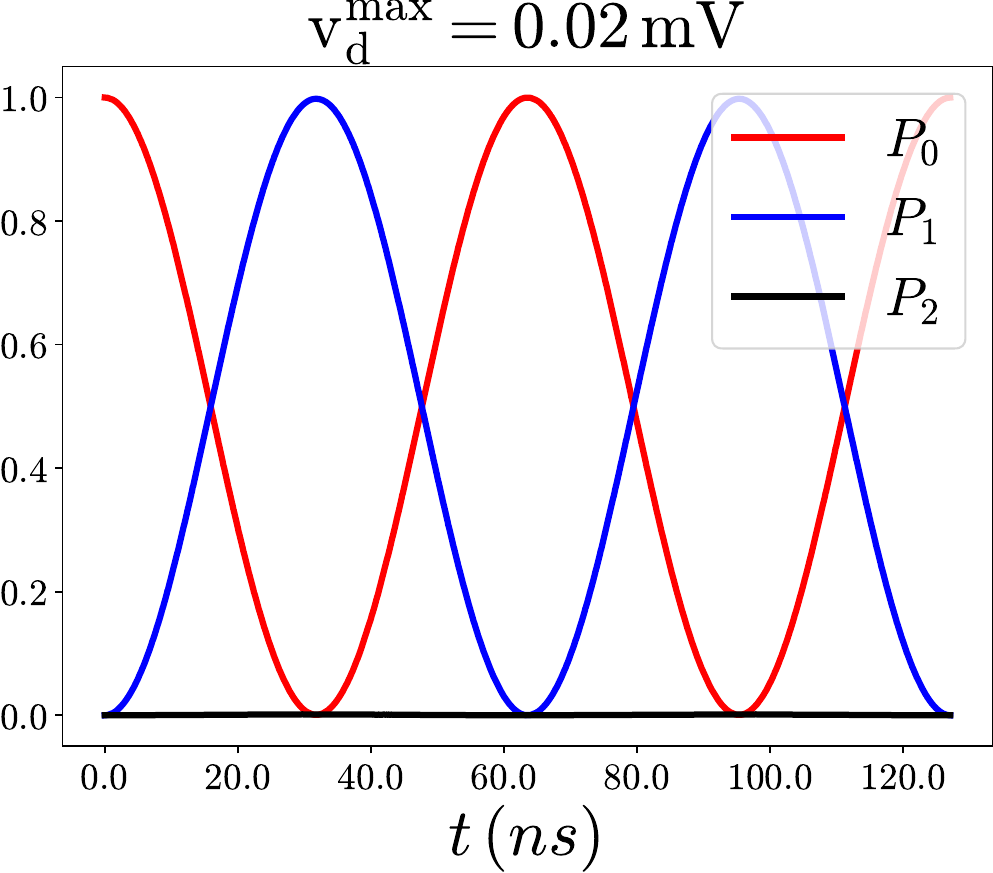}
\subcaption{}
\label{fig:tr_drive_a}
\end{subfigure}
\begin{subfigure}[h]{0.4\textwidth}
\centering
\includegraphics[width=5cm]{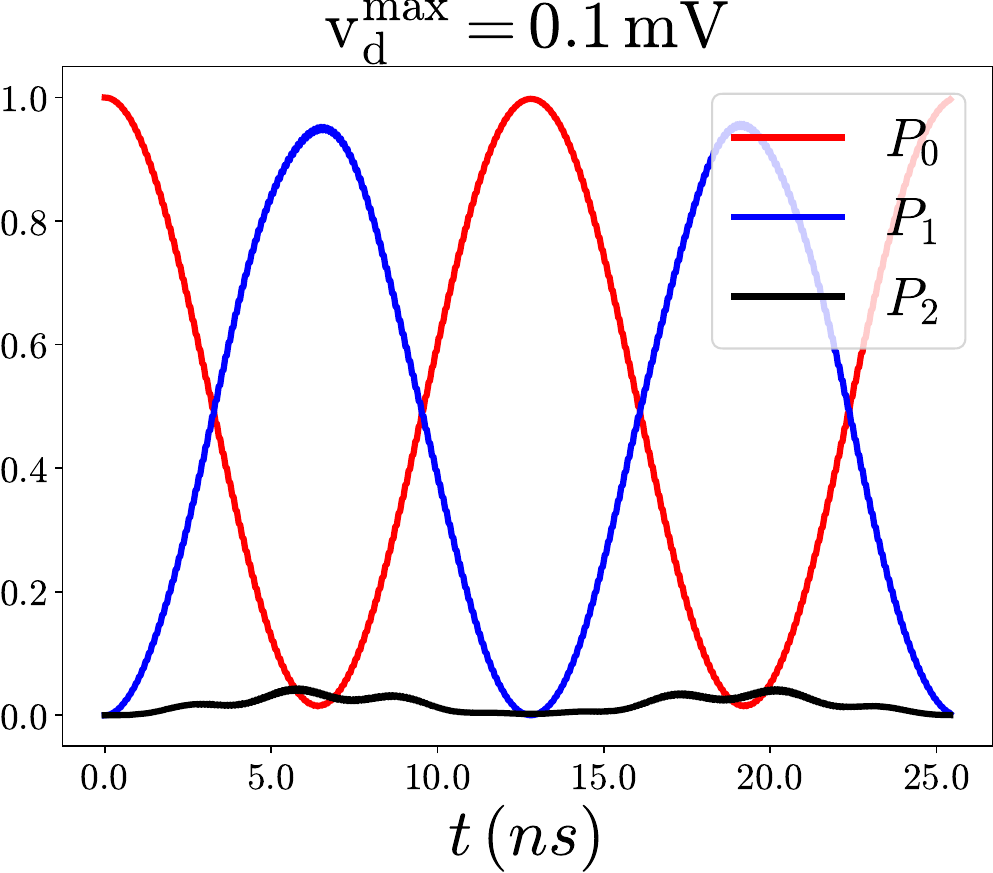}
\subcaption{}
\label{fig:tr_drive_b}
\end{subfigure}
\begin{subfigure}[h]{0.4\textwidth}
\centering
\includegraphics[width=5cm]{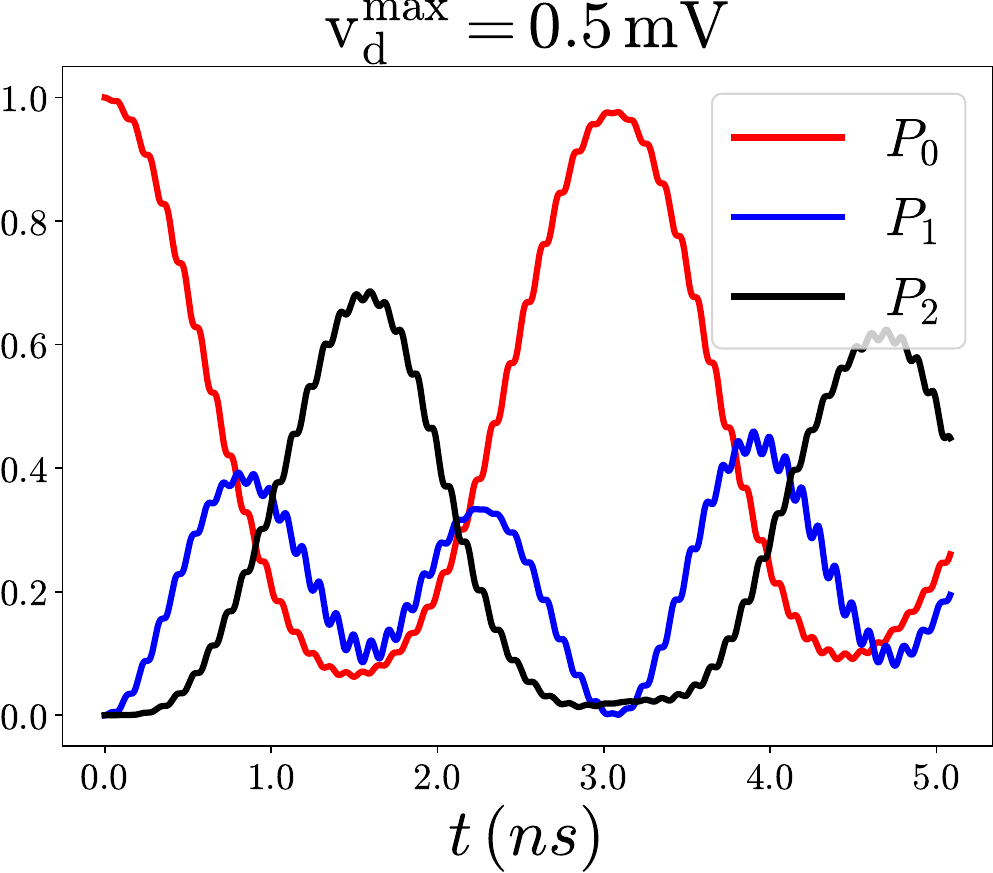}
\subcaption{}
\label{fig:tr_drive_c}
\end{subfigure}
\caption{Driven transmon qubit. We approximate the transmon as a Duffing oscillator with frequency $\Omega/2 \pi = 5 \, \mathrm{GHz}$ and anharmonicity $\delta/2 \pi = - 300 \, \mathrm{MHz}$. The transmon is coupled to the drive line via a capacitance $C_d= 0.1 \, \mathrm{fF}$. We initialize the transmon in the ground state $\ket{0}$ and drive it with a resonant drive of the form given in Eq.~\eqref{eq:vdriveform}. The plots show the probabilities $P_0, P_1, P_2$ of finding the system in $\ket{0}, \ket{1}, \ket{2}$, respectively, as a function of time. We see that by increasing the maximum voltage $\vv_d^{\mathrm{max}}$ the probability of leaking out of the computational subspace, i.e., populating the $\ket{2}$ state, increases as expected. Practically, this sets a lower bound on the gate time of single-qubit gates, which are limited to be at least $10-20 \, \mathrm{ns}$. We also observe that the curves show small wiggles for large $\vv_d$, which is due to the fast-oscillating terms in Eq.~\eqref{eq:hit}.}
\label{fig:tr_drive}
\end{figure}

We now make an excursion to the interaction picture taking the Hamiltonian $H_0 = \hbar \Omega \hat{b}^{\dagger} \hat{b}$ as reference Hamiltonian. Said differently, we go to the rotating frame of the oscillator at frequency~$\Omega$. The Hamiltonian in the interaction picture is (see Chapter~$3$ in Ref.~\cite{petruccione} for instance or Exercise \ref{exc:rot-frame}):
\begin{align}
\frac{\tilde{H}(t)}{\hbar} &= \frac{1}{\hbar} \biggr[e^{i H_0 t/\hbar}H(t) e^{-i H_0 t/\hbar}-H_0 \biggr] \notag \\
&= \frac{\delta}{2} \hat{b}^{\dagger} \hat{b}^{\dagger} \hat{b} \hat{b} + \frac{i}{2} \mathcal{E}\bigl(e^{i (\Omega - \Omega_d)t } \hat{b}^{\dagger} - e^{-i (\Omega - \Omega_d)t } \hat{b} \bigr)  
+ \frac{i}{2} \mathcal{E}\bigl(e^{i (\Omega + \Omega_d)t } \hat{b}^{\dagger} - e^{-i (\Omega + \Omega_d)t } \hat{b} \bigr).
\end{align}
If $\Omega_d = \Omega$, we obtain
\begin{equation}\label{eq:hit}
\frac{H(t)}{\hbar} = \frac{\delta}{2} \hat{b}^{\dagger} \hat{b}^{\dagger} \hat{b} \hat{b} + \frac{i}{2} \mathcal{E}\bigl( \hat{b}^{\dagger} - \hat{b} \bigr) 
+ \frac{i}{2} \mathcal{E}\bigl(e^{i 2 \Omega t } \hat{b}^{\dagger} - e^{-i 2\Omega t } \hat{b} \bigr).
\end{equation}
We see that in this picture the last drive term is fast-oscillating with a frequency that is twice the qubit frequency. This means that its effect will be to give small but fast oscillations of the transition amplitudes, which can be neglected to first approximation (for strong drives, one can go to next-order terms in a Magnus expansion, see e.g.~Ref.~\cite{Zeuch_2020}, which leads to a well-known Bloch-Siegert shift in the qubit frequency).

If we further approximate the evolution of the system to be restricted to the computational subspace defined by the first two levels, we obtain
\begin{equation}
\frac{\tilde{H}_{q}}{\hbar} = \frac{\mathcal{E}}{2} Y.
\end{equation}
Thus, the drive Hamiltonian will cause rotations around the $Y$-axis of the Bloch sphere in the interaction picture. Choosing a different phase for the drive, i.e., taking $\varepsilon_d(t) = \mathcal{E} \cos(\Omega_d t + \theta)$, makes it possible to choose an arbitrary rotation axis in the $X$-$Y$ plane of the Bloch sphere. However, if the relative ratio between the anharmonicity and the drive power $\lvert \delta_r/\mathcal{E} \rvert$ is not large enough, the previous qubit approximation is not warranted, and the system can leak to higher computational states $\ket{2},\ket{3}$ etc. This means that the anharmonicity limits the drive amplitude, and accordingly, the gate speed in transmon qubits. We see this explicitly in Fig.~\ref{fig:tr_drive}. 
In practice, leakage during single-qubit gates is reduced by carefully designing the pulse shape, for example using the DRAG technique~\cite{drag}, which makes it possible to achieve single-qubit gate times of the order of a few $\mathrm{ns}$, with error rates $10^{-3}$ or less, without the dramatic leakage that we see in Fig.~\ref{fig:tr_drive_c}. 

\begin{Exercise}[title={Rotating or interaction frame},label=exc:rot-frame]
A common method of analysis and operation in quantum computation is to consider the dynamics of qubits or oscillators in a rotating reference frame in which (part of) their self-evolution is cancelled. Given is a quantum system with Hamiltonian $H$ and dynamics according to the Schr\"odinger-von Neumann equation
\[
i\hbar \frac{d\ket{\psi(t)}\bra{\psi(t)}}{dt}=[H, \ket{\psi(t)}\bra{\psi(t)}].
\]
Given a time-independent reference Hamiltonian $H_{\rm ref}$, generating the unitary $U_{\rm ref}=e^{-i t H_{\rm ref}/\hbar}$, show that the state vector $\ket{\tilde{\psi}(t)}=U_{\rm ref}^{\dagger}\ket{\psi(t)}$ evolves according to the Schr\"odinger equation with `rotating frame or interaction' Hamiltonian $\tilde{H}$ given by
\[
\tilde{H}=U_{\rm ref}^{\dagger} H U_{\rm ref}+i \hbar\frac{dU_{\rm ref}^{\dagger}}{dt} U_{\rm ref}=U_{\rm ref}^{\dagger} H U_{\rm ref}-H_{\rm ref}.
\]
If $H=H_0+V(t)$, and $H_{\rm ref}=H_0$, what is $\tilde{H}$? Here $H_0$ may be the Hamiltonian of a set of uncoupled qubits and oscillators and $V(t)$ a coupling which is used to realize qubit dynamics for computation.
\end{Exercise}

\begin{Answer}[ref={exc:rot-frame}]
We have 
\begin{multline}
    i\hbar \frac{d\ket{\tilde{\psi}(t)\bra{\tilde{\psi}(t)}}}{dt}=i\hbar \frac{dU_{\rm ref}^{\dagger}}{dt}\ket{\psi(t)}\bra{\psi(t)}U_{\rm ref}+ i\hbar U_{\rm ref}^{\dagger}\ket{\psi(t)}\bra{\psi(t)}\frac{dU_{\rm ref}}{dt}+ \\
U_{\rm ref}^{\dagger} [H, \ket{\psi(t)}\bra{\psi(t)}] U_{\rm ref}
=\left[i\hbar \frac{dU_{\rm ref}^{\dagger}}{dt}U_{\rm ref}+U_{\rm ref}^{\dagger} H U_{\rm ref},\ket{\tilde{\psi}(t)}\bra{\tilde{\psi}(t)}\right],
\end{multline}
by inserting $U_{\rm ref} U_{\rm ref}^{\dagger}=\mathds{1}$, using the Schr\"odinger-von Neumann equation and $i\hbar U_{\rm ref}^{\dagger} \frac{dU_{\rm ref}}{dt}=H_{\rm ref}$.
The rotating frame Hamiltonian is $\tilde{H}=e^{i H_0t/\hbar} V(t)e^{-it H_0 t/\hbar}$.
\end{Answer}

\section{Transmission lines and co-planar resonators}
\label{sec:tl-cp}

In this section, we discuss the mathematical treatment of transmission lines and co-planar resonators. The first are used for (microwave) input and output on the chip, whereas the latter are used to couple qubits, store quantum information and perform qubit measurements. Some pictures of these structures and transmon qubits are shown in Figs.~\ref{fig:dicarlo-transmon-reson}, \ref{fig:riste}, \ref{fig:reson} and \ref{fig:chip-trans}.

Besides co-planar resonators, very high-$Q$ 3D microwave cavities are also widely used in circuit QED \cite{paik2011}: finite-element electromagnetic simulations (using software such as Ansys or COMSOL) rather than circuit theory are used to model the eigenmodes and corresponding structure of the electric field inside these cavities \cite{reagorPhd}; we will come back to how to use this modeling as input to constructing a Hamiltonian in Chapter~\ref{chap:ln}. 

\begin{figure}[htb]
\centering
\includegraphics[height=3.5cm]{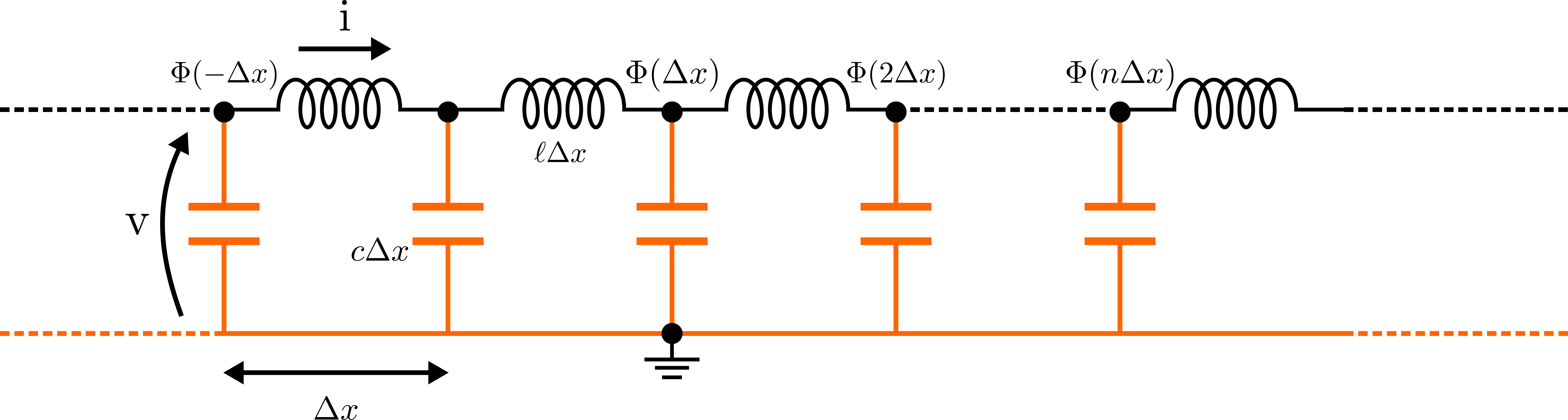}
\caption{Discrete model of an (infinite) transmission line.}
\label{fig:inftl}
\end{figure}

Consider the discrete model of an infinite, lossless, transmission line with capacitance per unit length denoted as $c$ (not the speed of light!), and inductance per unit of length denoted as $\ell$, shown in Fig.~\ref{fig:inftl}. Let us apply the method of circuit quantization to this circuit.

If we take the spanning tree highlighted in orange in Fig.~\ref{fig:inftl}, we can write down the Lagrangian for this circuit following the standard prescription. 
The Lagrangian of the discrete circuit in terms of the node fluxes is
\begin{equation}
\lagrangian=\sum_{k=-\infty}^{+\infty} \frac{c \Delta x}{2} \dot{\Phi}(k \Delta x, t)^2-\sum_{k=-\infty}^{+\infty} \frac{1}{2 \ell \Delta x} [\Phi((k+1) \Delta x, t)-\Phi(k \Delta x, t)]^2.
\end{equation}
In the limit of $\Delta x \rightarrow 0$ we can represent the collection of fluxes $\{\Phi(k\Delta x,t)\}$, $k\in \mathbb{Z}$ effectively as a field $\Phi(x,t)$ and $\sum_{k\in \mathbb{Z}} \Delta x f(k \Delta x) \rightarrow \int_{-\infty}^{\infty} dx f(x)$. In the limit of $\Delta x \rightarrow 0$, one recovers the Lagrangian of a scalar, one-dimensional Klein-Gordon field $\Phi(x,t)$, i.e.,
\begin{equation}
\lagrangian= \int_{-\infty}^{+\infty} dx \,\mathfrak{L}(\partial_x \Phi(x,t), \partial_t \Phi(x,t)) ,
\label{eq:KG}
\end{equation}
with Lagrangian density
\begin{equation}
\mathfrak{L}(\partial_x \Phi(x,t), \partial_t \Phi(x,t)) = \frac{c}{2} \biggl(\frac{\partial \Phi}{\partial t}\biggr)^2-\frac{1}{2 \ell} \biggl(\frac{\partial \Phi}{\partial x}\biggr)^2.
\end{equation}
Physically, the current in each inductive branch is proportional to the branch flux and hence becomes in this limit 
\begin{equation}
    \ii(x,t)=-\frac{1}{\ell} \frac{\partial \Phi}{\partial x}.
    \label{eq:iphi}
\end{equation}
By definition, see Eq.~\eqref{eq:def-flux}, the voltage along the line is given by the time-derivative of the flux, i.e.,~
\begin{equation}
\label{eq:defVI}
\vv(x,t)= \frac{\partial \Phi}{\partial t}.
\end{equation}

The Euler-Lagrange equations for a general scalar field~$\Phi(x,t)$, see the derivation of Eq.~\eqref{eq:ELfield} in Exercise~\ref{exc:field}, are given by
\begin{equation}
\frac{\partial }{\partial t} \biggl(\frac{\partial \mathfrak{L} }{\partial(\partial_t \Phi)}\biggr)+\frac{\partial }{\partial x} \biggl(\frac{\partial \mathfrak{L}}{\partial(\partial_x \Phi)}\biggr)-\frac{\partial \mathfrak{L}}{\partial \Phi}=0,
\end{equation}
and this leads to
\begin{equation}
\label{eq:waveEq}
\frac{\partial^2 \Phi}{\partial t^2}-v_p^2 \frac{\partial^2 \Phi}{\partial x^2}=0.
\end{equation}
This is the one-dimensional wave equation with phase velocity 
\begin{align}
v_p= \frac{1}{\sqrt{\ell c}}.
\end{align}
For a co-planar transmission line resonator on a silicon substrate $v_p \approx \mathfrak{c}/2.5$, where $\mathfrak{c}$ is the speed of light. 

From Eqs.~\eqref{eq:iphi},~\eqref{eq:defVI} and the wave equation itself one can readily see that
\begin{equation}
    \frac{\partial \vv}{\partial x}=-\ell \frac{\partial \ii}{\partial t}, \;\; \frac{\partial \ii}{\partial x}=-c\frac{\partial \vv}{\partial t}.
\end{equation}
These are known as the Telegrapher's equations (see Ref.~\cite{pozar}). By rewriting the wave equation as
\begin{equation}
(\partial_t-v_p \partial_x)(\partial_t+v_p \partial_x)\Phi(x,t)=0,
\end{equation}
we realize that solutions of the transport equations 
\begin{equation}
\frac{\partial \Phi}{\partial t} \pm v_p \frac{\partial \Phi}{\partial x}=0,
\label{eq:transp}
\end{equation}
are also solutions of the wave equation. These are the right-propagating solutions $\Phi^{\rightarrow}(x,t)=\Phi^{\rightarrow}(t-x/v_p)$, and the left-propagating solutions $\Phi^{\leftarrow}(x,t)=\Phi^{\leftarrow}(t+x/v_p)$, that solve Eq.~\eqref{eq:transp} when the middle sign is plus or minus, respectively. The general solution $\Phi(x,t)$, by linearity, is of course a linear combination of $\Phi^{\rightarrow}(x,t)$ and $\Phi^{\leftarrow}(x,t)$. In analogy with Eqs.~\eqref{eq:iphi} and~\eqref{eq:defVI}, we can define right- and left-propagating currents and voltages as

\begin{subequations}
    \begin{equation}
        \ii^{\rightarrow}(x, t) = -\frac{1}{\ell} \frac{\partial \Phi^{\rightarrow}}{\partial x}, \quad \ii^{\leftarrow}(x, t) = -\frac{1}{\ell} \frac{\partial \Phi^{\leftarrow}}{\partial x},
    \end{equation}
\begin{equation}
\label{eq:vleftrightdef}
        \vv^{\rightarrow}(x, t) = \frac{\partial \Phi^{\rightarrow}}{\partial t}, \quad \vv^{\leftarrow}(x, t) = -\frac{\partial \Phi^{\leftarrow}}{\partial t}.
    \end{equation}
\end{subequations}

Using Eqs.~\eqref{eq:iphi} and~\eqref{eq:defVI}, Eq.~\eqref{eq:transp} implies 
\begin{equation}
    \vv^{\rightarrow}(x,t)=Z_0 \ii^{\rightarrow}(x,t),\;\; \vv^{\leftarrow}(x,t)=-Z_0 \ii^{\leftarrow}(x,t),
    \label{eq:char-imp}
\end{equation}
where we defined the characteristic impedance of the transmission line $Z_0=\sqrt{\ell/c}$. This characteristic impedance equals $50\,\Omega$ for commonly-used transmission lines and depends on the geometric properties of the center conductor and the distance to the outer conductor \cite{pozar}.

The transmitted power of the transmission line is given by
\begin{equation}P(x,t)=\vv^{\rightarrow}(x,t)\ii^{\rightarrow}(x,t)+\vv^{\leftarrow}(x,t)\ii^{\leftarrow}(x,t)=\frac{1}{Z_0}({\vv^{\rightarrow}}^2-{\vv^{\leftarrow}}^2),
\end{equation}
with contributions from left- and right-flowing energy fluxes. 

A general real solution of the wave equation can be written as
\begin{equation}
\label{eq:expPhi}
\Phi(x,t)=\Phi^{\rightarrow}(x,t)+\Phi^{\leftarrow}(x,t)= \sqrt{\frac{\hbar}{4 \pi c}} \int_{-\infty}^{+\infty} d k \frac{1}{\sqrt{\omega(k)}} \biggl(b_ke^{-i \omega(k)t+i kx}+b^*_k e^{i \omega(k)t-i kx}\biggr),
\end{equation}
with angular frequency $\omega(k)=v_p |k|$, wavenumber $k$ and complex coefficients $b_k$. Here we have pulled out some prefactors so that the complex coefficient $b_k$ has the same dimension as $k^{-1/2}$.

The right-propagating field associated with $k>0$ is thus
\begin{equation}
\Phi^{\rightarrow}(x,t)= \sqrt{\frac{\hbar}{4 \pi c}} \int_{0}^{+\infty} d k \frac{1}{\sqrt{\omega(k)}} \biggl( b_ke^{-i \omega(k)(t-x/v_p)}+\mathrm{c.c.}\biggr),
\label{eq:rightprop}
\end{equation}
while the left-propagating field associated with $k< 0$ equals
\begin{equation}
\Phi^{\leftarrow}(x,t)= \sqrt{\frac{\hbar}{4 \pi c}} \int_{-\infty}^{0} d k \frac{1}{\sqrt{\omega(k)}} \biggl(b_ke^{-i \omega(k)(t+x/v_p)}+\mathrm{c.c.} \biggr).
\label{eq:leftprop}
\end{equation}
Bosonic quantization means that we replace the real and imaginary value of a complex number by the expectation of two conjugate Hermitian operators. In other words, for every $k \in \mathbb{R}$ we promote
\begin{eqnarray}
    b_k \rightarrow \hat{b}_k, \; b^*_k \rightarrow \hat{b}^{\dagger}_k,\;
[\hat{b}_k,\hat{b}^{\dagger}_{k'}]=\delta(k-k')\mathds{1}, \;[\hat{b}_k,\hat{b}_{k'}]=0.
\label{eq:comm}
\end{eqnarray}
From this definition we see again that $\hat{b}_k$, and thus $b_k$, has the same dimension as $k^{-1/2}$, as the delta function has the dimension of $k^{-1}$.

The classical Hamiltonian is also readily obtained by defining the conjugate field
\begin{equation}
Q(x,t)= \frac{\partial \mathfrak{L}}{\partial (\partial_t \Phi)}= c \frac{\partial \Phi(x,t)}{\partial t},
\label{eq:chargedef}
\end{equation}
which physically represents the charge per unit of length. This leads to the classical Hamiltonian
\begin{equation}
\hamiltonian =\int_{- \infty}^{+\infty} dx\;  Q(x,t) \frac{\partial \Phi(x,t)}{\partial t}-\lagrangian=  \int_{-\infty}^{+\infty} dx \biggl\{ \frac{Q^2(x,t)}{2 c} +\frac{1}{2 \ell} \biggl(\frac{\partial \Phi}{\partial x}\biggr)^2 \biggl\}.
\label{eq:ham-ift}
\end{equation}
The classical Hamiltonian expresses the energy in a field configuration at a certain moment in time using the explicit time dependence of $Q(x,t)$ and $\Phi(x,t)$. Since energy is conserved, the energy is a time-independent quantity of course. If we quantize the system, the time-dependent field $\Phi(x,t)$ represents the operator $\hat{\Phi}(x)$ in the Heisenberg representation, where $\hat{\Phi}(x)$ is defined as
\begin{equation}
    \hat{\Phi}(x)=\sqrt{\frac{\hbar}{4 \pi c}} \int_{-\infty}^{+\infty} d k \frac{1}{\sqrt{\omega(k)}} \biggl(\hat{b}_ke^{i kx}+\hat{b}^{\dagger}_k e^{-i kx}\biggr).
    \label{eq:phix}
\end{equation}
Similarly, the quantized charge operator $\hat{Q}(x)$ can be obtained by using the classical relation Eq.~\eqref{eq:chargedef} and then removing the time dependence, leading to
\begin{equation}
    \hat{Q}(x)= \sqrt{\frac{\hbar c}{4 \pi}} \int_{\infty}^{+\infty} dk \sqrt{\omega(k)} \biggl(-i \hat{b}_ke^{ikx}+i \hat{b}^{\dagger}_ke^{-i kx}\biggr).
    \label{eq:qx}
\end{equation}

Indeed, if the quantized Hamiltonian $H$ is as in Eq.~\eqref{eq:ham-ift-sim} in Exercise \ref{exc:ift} below, then the Heisenberg operators have the time dependence
\begin{equation}
\hat{b}(k,t)=e^{i H t/\hbar} \hat{b}_k e^{-i H t/\hbar}=\hat{b}_ke^{-i \omega(k) t}, \;\hat{b}^{\dagger}(k,t)=\hat{b}^{\dagger}_ke^{i \omega(k) t}.
\end{equation}
The Heisenberg evolution of a bosonic operator will be explicitly verified in Eq.~\eqref{eq:bosheis} in Exercise \ref{exc:JCmod}.

\begin{Exercise}[title={Quantized Infinite Transmission Line}, label={exc:ift}]
\Question Verify that if we apply quantization to the classical Hamiltonian in Eq.~\eqref{eq:ham-ift}, i.e.,~replace $b_k \rightarrow \hat{b}_k$ and $b^*_k \rightarrow \hat{b}^{\dagger}_k$, the quantized Hamiltonian $H$ indeed reads  
\begin{equation}
H= \int_{-\infty}^{+\infty}dk \, \hbar \omega(k) \left(\hat{b}^{\dagger}_k \hat{b}_k+\frac{\delta(0)}{2}\right).
\label{eq:ham-ift-sim}
\end{equation}
You can do this by entering the operators $\hat{\Phi}(x)$ and $\hat{Q}(x)$ in the quantized (time-independent) version $H$ of $\hamiltonian$ in Eq.~\eqref{eq:ham-ift}. Due to the Dirac delta function $\delta(0)$ the vacuum energy is infinite, in principle.
\Question Verify that 
\begin{equation}
[\hat{\Phi}(x), \hat{Q}(x \sp{\prime})]=i \hbar \delta(x-x \sp{\prime}) \mathds{1},
\end{equation}
using the expressions for $\hat{Q}(x)$ and $\hat{\Phi}(x)$ in Eqs.~\eqref{eq:phix} and~\eqref{eq:qx}. That is, verify that they satisfy the canonical commutation relations as expected.
\end{Exercise}

\begin{Answer}[ref={exc:ift}]
\Question Use $\frac{1}{2\pi}\int_{-\infty}^{\infty} dx \,e^{i (k-k')x}=\delta(k-k')$, $\omega(k)=v_p|k|=\frac{1}{\sqrt{\ell c}} |k|$ and the commutation relations in Eq.~\eqref{eq:comm}.
\Question Use the commutation relations in Eq.~\eqref{eq:comm} and 
\begin{align}
    \int_{-\infty}^{\infty} dk\, \left(e^{i k (x-x')}+e^{-i k (x-x')}\right)=4 \pi \delta(x-x').
\end{align}
\end{Answer}

\begin{figure}[htb]
    \centering
    \includegraphics[height=6cm]{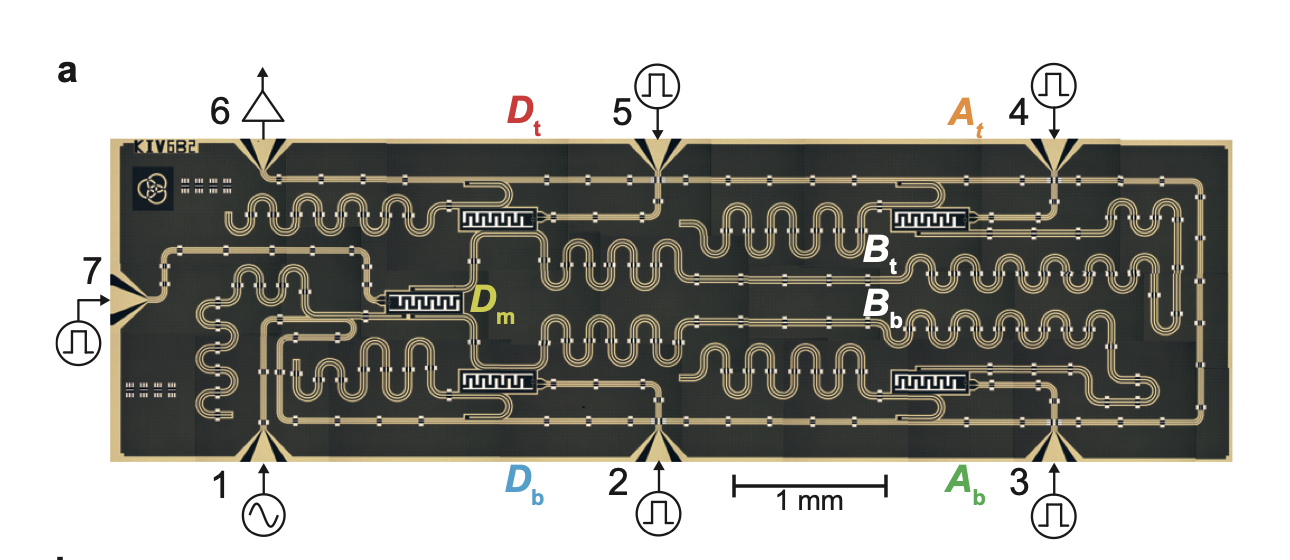}
    \caption{Optical micrograph of a five transmon qubit chip from Ref.~\cite{Riste.etal.2015:BitflipChip} (Figure taken from \href{https://arxiv.org/abs/1411.5542}{ArXiv version of the paper}). Visible are one multiplexed (AC) transmission line coming on- and off-chip at 1 and 6, five individual (DC) flux-lines, each controlling one qubit, three transmon data qubits ($D_t,D_m,D_b$), two transmon ancilla qubits ($A_t$, $A_b$) and two coupling bus resonators $(B_t,B_b)$ coupling each ancilla qubit with two data qubits. In addition, each qubit has its own readout resonator.}
    \label{fig:riste}
\end{figure}

\subsection{Boundaries and resonators}
\label{subsec:br}

If the transmission line can indeed be modeled as infinite, we see that any $k \in \mathbb{R}$ is allowed and the solutions are traveling plane waves characterized by wavenumber $k$, with which we can associate a bosonic mode with annihilation operator $\hat{b}_k$.
In this scenario it is incorrect to associate a single annihilation operator with a frequency $\omega=v_p|k|$, as there are two different mode operators, $\hat{b}_k$ and $\hat{b}_{-k}$, which have that frequency.

Consider now a semi-infinite transmission line which extends to infinity in one direction, but has an open-circuit boundary condition at $x=0$, meaning that 
\begin{equation*}
\ii(x=0,t)=-\frac{1}{\ell}\left.\frac{\partial \Phi}{\partial x}\right\vert_{x=0}=0.
\end{equation*}
We see that if we choose $b_{-k}=-b_k$ and thus $b^*_{-k}=-b^*_k$ in Eq.~\eqref{eq:expPhi}, then for any $k
\in \mathbb{R}$,
\begin{align}
    \Phi_k(x,t) = & b_k[e^{-i \omega(k)t+i kx}-e^{-i \omega(k)t-i kx}]+b^*_k[e^{i \omega(k)t-i kx}-e^{i \omega(k)t+i kx}]    \notag \\
     = &  2 i \cos(k x)(b_ke^{-i \omega(k)t}-b^*_ke^{i \omega(k) t}),
\end{align}
is a solution which satisfies this boundary condition. If we quantize this system, we associate an operator $\hat{b}_k$ with $b_k$, but $b_{-k}$ is simply replaced by $-\hat{b}_k$. Therefore, we can replace the wavenumber $k$ with the angular frequency, as there is only one mode at each frequency. Physically, this means that the mode that travels to the right towards $x=0$ with amplitude $a$ is perfectly reflected at this boundary and continues with negative amplitude $-a$ towards the left. 
We can also consider a semi-infinite transmission line which is grounded at $x=0$, i.e.,~$\vv(x=0,t)=0$. In this case, any solution with $b_k=b_{-k}$ is valid, and physically it means that the reflected wave does not acquire a $-1$ phase shift at this boundary, but retains its amplitude.

If instead the transmission line has a finite length, we need to impose some suitable boundary conditions at its ends. In this case, the set of possible solutions labeled by $k$ becomes discrete corresponding to resonant modes. The most common boundary conditions at a boundary point $x_b$ are the  open-circuit boundary condition $\ii(x_b, t)=0$, and the short-circuit boundary condition $\vv(x_b, t)=0$. A transmission line with open-circuit boundaries on both ends is called a $\lambda/2$ resonator, as the lowest discrete wavevector $k=2\pi/\lambda$ corresponds to a wavelength $\lambda=2 d$ with $d$ the length of the line, see Exercise~\ref{exc:ftl}.

If we terminate a transmission line by one grounded boundary and an open circuit at the other end, one obtains a so-called $\lambda/4$-resonator. Let the length of the transmission line be $d$. 
First, the grounded boundary restricts the classical solutions to the general form
    \begin{equation}
\label{eq:expPhi-l4}
\Phi(x,t)= \sqrt{\frac{\hbar}{4 \pi c}} \int_{0}^{+\infty} d k \frac{\sin(kx)}{\sqrt{\omega(k)}} \biggl(b_ke^{-i \omega(k)t}+b^*_k e^{i \omega(k)t}\biggr).
\end{equation}
Then, the condition 
\begin{equation*}
\ii(x=d,t)=-\frac{1}{\ell}\left.\frac{\partial \Phi}{\partial x}\right\vert_{x=d}=0,
\end{equation*}
implies that $k$ is restricted such that $\cos(k d)=0$, or $k d=\pi/2+n\pi$ with $n\in \mathbb{Z}$, allowing for angular frequencies 
\begin{equation}
 \label{eq:freql4}
    \omega_n=\frac{v_p \pi}{2 d}+\frac{v_p n \pi}{d}.
\end{equation}
Hence the largest wavelength which is supported ---the lowest energy mode--- is $\lambda=4 d$. 

Both $\lambda/4$ and $\lambda/2$ co-planar resonators can be used and one usually omits the description of all but one mode, namely the mode whose frequency is closest to the coupled qubits of interest. The co-planar resonator can be coupled to a transmon qubit, as will be discussed in Section~\ref{sec:JC-coupling}, and when used as a readout resonator, it is in turn coupled to a transmission feedline, see Fig.~\ref{fig:reson} and Fig.~\ref{fig:riste}. 

\begin{figure}[htb]
    \centering
    \includegraphics[height=4cm]{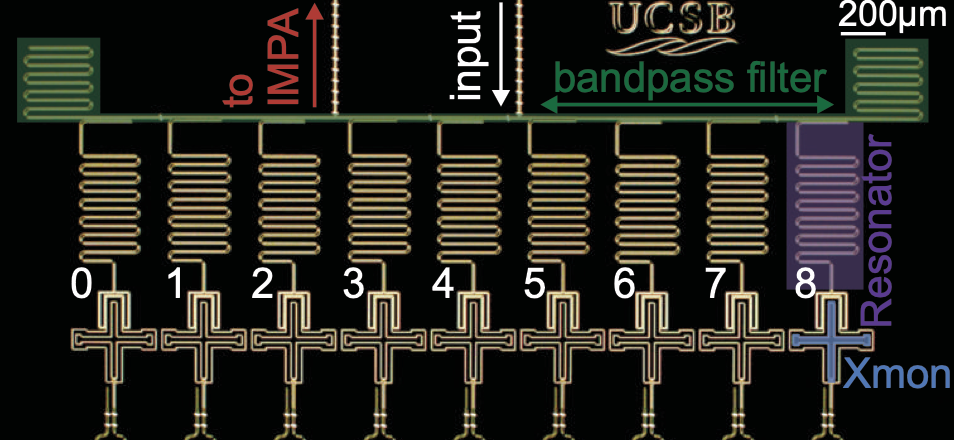}
    \caption{False-colored optical micrograph from Ref.~\cite{kelly:rep-code} with nine transmon qubits (Xmons), each with its $\lambda/2$ readout resonator which are all coupled to one common `feed' transmission line at the top of the chip. Modern industrial superconducting qubit architectures can use a sandwich of two separate but coupled chips, for example one hosting qubits and (tunable) couplers, the other, control, trip hosting readout resonators and I/O lines. Such architecture allows for multi-qubit scalability in 2D, see e.g. Google's Sycamore chip discussed on  \href{https://youtu.be/IWQvt0RBclw}{YouTube}.}
    \label{fig:reson}
\end{figure}

We can also consider semi-infinite transmission lines that are terminated by a so-far unspecified other linear network. Using Eqs.~\eqref{eq:vleftrightdef},~\eqref{eq:rightprop}, and \eqref{eq:leftprop} we have 
\begin{subequations}
\label{eq:lr-volt}
\begin{equation}
    \vv^{\rightarrow}(x,t)=v_p \int_0^{\infty} dk \frac{\vv_0^{\rightarrow}(k)}{2\pi} e^{i v_p |k| (t-x/v_p)}+{\rm c.c.},    
\end{equation}
\begin{equation}
     \vv^{\leftarrow}(x,t)=v_p \int_{-\infty}^0 dk \frac{\vv_0^{\leftarrow}(k)}{2\pi} e^{i v_p |k| (t+x/v_p)}+{\rm c.c.},
\end{equation}
\end{subequations}
where the amplitudes $\vv_0^{\rightarrow}(k)$,  $\vv_0^{\leftarrow}(k)$ are defined as 
\begin{subequations}
\begin{equation}
    \vv_0^{\rightarrow}(k)=i\sqrt{\frac{\hbar k Z_0}{4\pi }}b_k2\pi, \quad k>0,
\end{equation}
\begin{equation}
    \vv_0^{\leftarrow}(k)=i\sqrt{ -\frac{\hbar k Z_0}{4\pi }}b_k2\pi, \quad k <0.
\end{equation}
\end{subequations} 
When the transmission line is connected on one side to another electrical network, say, a linear network which can only reflect the incoming voltage signal, then $\vv^{\rightarrow}(x,t)$ relates to $\vv^{\leftarrow}(x,t)$ at the location $x$ where the reflection occurs (see Exercise~\ref{exc:2port-TL}). 

\begin{Exercise}[title={$\lambda/2$-transmission line resonator},label=exc:ftl]
Consider a finite transmission line of length $d$ ($x \in [0, d]$) with open-circuit boundary conditions: $\ii(0,t)=\ii(d,t)=0$. We first consider the problem completely classically.
\Question Consider the wave equation Eq.~\eqref{eq:waveEq} for $\Phi(x,t)$ and look for solutions of the form $\Phi(x,t) =\xi(t) f(x)$. Show that the boundary conditions impose some constraints on the form of the functions $f(x)$ and that we can find solutions of the previous form $\Phi_n(x,t)=\xi_n(t) f_n(x)$ parametrized by an integer $n=0,1,2, \dots$ and 
\begin{equation}
    \omega(k_n)=\omega_n=v_p k_n=\frac{v_p \pi n}{d},
    \label{eq:oml2}
\end{equation}
supporting a largest wavelength $\lambda=2 d$.
\Question Due to linearity we can write a general solution as
\begin{equation}
\label{eq:phiFiniteExp}
\Phi(x,t)=\sum_{n=0}^{+\infty} \xi_n(t) f_n(x).
\end{equation}
Plug this expansion into Eq.~\eqref{eq:KG} and integrate over~$x$ to get the Lagrangian for the generalized dynamical variables $\xi_n$, $\dot{\xi}_n$. Obtain the Hamiltonian in terms of $\xi_{n}$ and the related canonical momenta $q_n$. \par 
\emph{Hint: remember that for $n, m$ integers}
\begin{equation}
\frac{2}{\pi} \int_0^{\pi} dx \cos (n  x) \cos (m  x)= \frac{2}{\pi} \int_0^{\pi} dx \sin (n  x) \sin (m  x)= \delta_{n m}.
\end{equation}
\Question Quantize and write the Hamiltonian in terms of bosonic annihilation and creation operators.
Additionally, write the expression for the voltage $\hat{\vv}(x)$ in the Schr{\"o}dinger picture in terms of these operators. 
\Question Suppose you want to capacitively couple a system, say a transmon, to the resonator. In particular, you want it to interact strongly with the mode $n=2$. At which position $x_0 \in[0,d]$ would you put the coupling capacitance $C_c$?  
\end{Exercise}

\begin{Answer}[ref={exc:ftl}]
\Question
We solve the 1D wave equation with standard methodology. We look for non-trivial solutions (non-zero) of the form $\Phi(x,t)= \xi(t)f(x)$. Inserting this into Eq.~\eqref{eq:waveEq} and dividing by $\Phi(x,t)$ we get
\begin{equation}
\frac{1}{v_p^2 \xi(t)} \frac{d^2 \xi}{dt^2}= \frac{1}{f(x)} \frac{d^2 f(x)}{dx^2}.
\end{equation}
We have two expressions that have to be equal but are functions of different variables; this is possible only if they are equal to a constant which we call $-k^2$ with real $k$. We then get the two equations
\begin{subequations}
\begin{equation}
\frac{d^2 \xi}{dt^2}= -k^2 v_p^2 \xi(t),
\end{equation}
\begin{equation}
\label{eqf}
\frac{d^2 f}{dx^2}= -k^2  f(x).
\end{equation}
\end{subequations}
The open-circuit boundary conditions imply that $\ii(x,t)=\ii(d,t)=0$ at all $t$ and thus
\begin{equation}
\frac{\partial \Phi}{\partial x}(0,t)= \frac{\partial \Phi}{\partial x}(d,t),
\end{equation}
translating into conditions for $f'(x)$
\begin{equation}
f'(0)=f'(d)=0.
\end{equation}
The general solution of Eq.~\eqref{eqf} is
\begin{equation}
f(x)= A \cos(k x)+B \sin (k x),
\end{equation}
with $A, B$ real constants. The condition $f'(0)=0$ implies $B=0$. Consequently, the condition $f'(d)=0$ reads
\begin{equation}
-A k \sin(k  d)=0,
\end{equation}
which gives non trivial solutions only if 
\begin{equation}
k=k_n= \frac{\pi n}{d} \quad n=\{0, 1,2, \dots\},
\end{equation}
leading to Eq.~\eqref{eq:oml2}.
Note that you can restrict $n$ to be positive, as the negative solutions are the same as the positive ones since the cosine is an even function. We have obtained a family of solutions parametrized by an integer
\begin{equation}
\Phi_n(x,t)= \xi_n(t) f_n(x),
\end{equation}
$n=\{0, 1,2, \dots\}$, with
\begin{equation}
f_n(x)= A_n \cos(k_n x).
\end{equation}
The choice of the constants $A_n$ is arbitrary. However, we can take them such that the $\ell^2$-norm of $f_n(x)$ is equal to one. The $\ell^2$-norm is defined as
\begin{equation}
\lVert f_n \rVert= \biggl( \frac{1}{d} \int_{0}^{d} dx \, f_{n}(x)^2 \biggr)^{1/2}.
\end{equation}
Since for $n \ge 1$
\begin{equation}
\frac{1}{d} \int_{0}^{d} dx \, \cos^2 \biggl( \frac{\pi n}{d}\biggr)= \frac{1}{2},
\end{equation}
we take $A_n=\sqrt{2} \quad \forall n \ge 1$, while for $n=0$, we take trivially $A_0=1$. Compactly
\begin{subequations}
\begin{equation}
f_0(x)=1,
\end{equation}
and
\begin{equation}
f_n(x)= \sqrt{2} \cos \biggl(\frac{\pi n}{d} x \biggr).
\end{equation}
\end{subequations}\par 
From the linearity of the wave equation a linear combination of solutions is again a solution, so that we can write a generic solution as \footnote{Basically, the interpretation of this expansion is that the functions $f_n(x)$ form an orthonormal basis for the vector space of $\ell^2$ (square-integrable) functions defined in $[0,d]$ which satisfy the previous boundary conditions.}
\begin{equation}
\label{expPhiRes}
\Phi(x,t)= \sum_{n=0}^{+\infty} \xi_n(t) f_n(x).
\end{equation} 
Note that the $\xi_n(t)$ now satisfy
\begin{equation}
\frac{d^2 \xi_n}{dt^2}=-\omega_n^2 \xi_n(t),
\label{eq:xi}
\end{equation}
which is the equation of motion of a harmonic oscillator with characteristic frequency $\omega_n = v_p k_n= v_p \pi n/d$. We expect that this is reproduced when we plug our expansion in the Lagrangian. 
\Question Before plugging the expansion Eq.~\eqref{expPhiRes} into the Lagrangian, we notice two properties of the functions $f_n(x)$
\begin{subequations}
\label{innerProducts}
\begin{equation}
\frac{1}{d} \int_0^{d} dx f_n(x) f_m(x) = \delta_{nm},
\end{equation}
\begin{equation}
\frac{1}{d} \int_0^{d} dx \frac{d f_n}{dx} \frac{d f_m}{dx}  = k_n^2 \delta_{nm}.
\end{equation}
\end{subequations} 
These are useful facts when we plug the expansion, Eq.~\eqref{expPhiRes}, into the Lagrangian in Eq.~\eqref{eq:KG}:
\begin{equation}
\lagrangian= \int_0^d dx \; \biggl\{\frac{c}{2} \biggl(\sum_{n=0}^{+\infty} \dot{\xi}_n(t) f_n(x) \biggr)^2 -\frac{1}{2 \ell} \biggl(\sum_{n=1}^{+\infty} \xi_n(t) \frac{d f_n}{dx} \biggr)^2 \biggr\}\,,
\end{equation}
which, using Eqs.~\eqref{innerProducts}, gives
\begin{equation}
\lagrangian= \frac{1}{2}d c \sum_{n=1}^{+\infty}(\dot{\xi}_n^2-\omega_n^2 \xi_n^2) +\frac{1}{2}d c \dot{\xi}_0^2.
\end{equation}
Note that mode $0$ is a free particle, corresponding to a constant voltage (charge) and thus a linearly increasing flux (in time), as $\omega_r=0$ implies the solution $x_0(t)=C_1+_2 t$ in Eq.~\eqref{eq:xi}. We will neglect it from now on, meaning that our flux is always the difference between $\Phi(x,t)$ and the linearly increasing flux. Anyway, in most cases this DC voltage is really zero. Neglecting the free-particle mode, we obtain the conjugate variables
\begin{equation}
q_n= \frac{\partial \lagrangian}{\partial \dot{\xi}_n}= d c  \,\dot{\xi}_n,
\end{equation}
and the Hamiltonian
\begin{equation}
\hamiltonian= \sum_{n=1}^{+\infty} \frac{q_n^2}{2 d  c}+\frac{d c}{2} \omega_n^2 \xi_n^2.
\end{equation}
\Question We now quantize the theory as usual by promoting variables to operators denoted by hats
\begin{equation}
H= \sum_{n=1}^{+\infty} \frac{\hat{q}_n^2}{2 d c}+\frac{d c}{2} \omega_n^2 \hat{\xi}_n^2 ,
\end{equation}
and imposing commutation relations between conjugate variables
\begin{equation}
[\hat{\xi}_n, \hat{q}_n]= i \hbar \mathds{1}.
\end{equation}
We observe that we just have the Hamiltonian of a collection of independent harmonic oscillators with equal `mass' $dc$ (equal to the total capacitance) and frequencies $\omega_n$. We then introduce annihilation and creation operators
\begin{subequations}
\begin{equation}
\hat{\xi}_n= \sqrt{\frac{\hbar}{2 d c \omega_n}}(\hat{a}_n+\hat{a}_n^{\dagger}),
\end{equation}
\begin{equation}
\hat{q}_n= i\sqrt{\frac{\hbar}{2 }d c \omega_n}(\hat{a}_n^{\dagger}-\hat{a}_n),
\end{equation}
\end{subequations}
which satisfy the commutation relations $[\hat{a}_n, \hat{a}_m^{\dagger}]= \delta_{nm}\mathds{1}$. The Hamiltonian becomes
\begin{equation}
H= \sum_{n=1}^{+\infty}\hbar \omega_n \left(\hat{a}_n^{\dagger}\hat{a}_n+ \frac{1}{2}\right).
\end{equation}
In the Heisenberg picture we have
\begin{equation}
\hat{\Phi}(x,t) = \sum_{n=1}^{+\infty} \sqrt{\frac{\hbar}{d c \omega_n}} \cos(k_n x) \biggl( \hat{a}_n e^{-i \omega_n t}+\hat{a}_n^{\dagger}e^{i \omega_n t} \biggr),
\end{equation}
from which we can deduce the expression for the voltage operator in the Heisenberg picture
\begin{equation}
\hat{\vv}(x,t)= \frac{\partial \hat{\Phi}}{\partial t}= \sum_{n=1}^{+\infty} i \sqrt{\frac{\hbar \omega_n}{d c}} \cos(k_n x)\biggl(\hat{a}_n^{\dagger} e^{i \omega_n t} -\hat{a}_n e^{-i \omega_n t} \biggr),
\end{equation}
which in the Schr{\"o}dinger picture equals 
\begin{equation}
\hat{\vv}(x)= \sum_{n=1}^{+\infty} i \sqrt{\frac{\hbar \omega_n}{d c}} \cos(k_n x)\biggl(\hat{a}_n^{\dagger}  -\hat{a}_n  \biggr). 
\end{equation}
\Question If we couple capacitively it means that we are basically coupling to the voltage field $\hat{\vv}(x)$ at a specific position $x_0$. If we want to couple strongly to the mode $n=2$, we would like to couple at a position that maximizes $f_2(x)= \sqrt{2}\cos(2 \pi x/d)$ in modulus. This happens for $x_0=\{0, d/2, d\}$, so the smartest thing to do is to couple either at the end or in the middle of the transmission line resonator.
\end{Answer}

\section{Input-output formalism: Heisenberg-Langevin equations}
\label{sec:HL_derive}

If we have a transmon qubit or a resonator, or any other system of modes corresponding to those of an electrical circuit, we can couple a transmission line to it, for instance via a capacitive coupling. The transmission line is used for sending electromagnetic pulses to the system that can either serve to control it or to obtain information about it by performing measurements on a reflected and/or transmitted signal.

One can write down the Heisenberg equations of motion of the modes of the system which include the effect of such input or output fields on a transmission line where, by definition, the input field represents the incoming signal (launched at some early time $t_0$) and the output field is the outgoing signal at some later time $t_1 > t_0$.
These Heisenberg-Langevin equations are useful in modeling the quantum measurement of a qubit, as well as modeling amplifiers which transform the input field into an amplified output field (see Exercises \ref{exc:deg_par_amp},\ref{exc:JRM} in Chapter~\ref{chap:add}). This formalism, which focuses on input and output fields, was originally developed for the description of input-output dynamics of damped quantum optical cavities \cite{GW:input-output}.

In the measurement of a (transmon) qubit, the qubit is dispersively coupled to a resonator, shifting the resonant frequency of the cavity depending on the qubit state. The cavity is then probed by the input field and the output field is phase-shifted depending on the state of the qubit. We refer a reader to \cite{Blais_2021} for more detailed modeling of the transmon qubit measurement. In Exercise~\ref{exc:disp-meas} we examine a simple qubit model to study some features of this dispersive measurement, omitting amplification.

In this section, we derive the Heisenberg-Langevin equations (see also Ref.~\cite{book:WM}). We imagine a semi-infinite transmission line capacitively coupled to a resonator with annihilation operator $\hat{a}$. The resonator can in turn be part of a larger electrical circuit, together described by some system Hamiltonian $H_{\rm sys}$. As discussed in Section~\ref{subsec:br}, for the semi-infinite transmission line, there is only one mode with frequency $\omega$ on the transmission line, so we use operators $\hat{b}_{\omega}$, labelled with frequency $\omega$ instead of $k$. We start with the Hamiltonian $H=H_{\rm sys}+H_{\rm TL} + H_{\rm coupl}$ with
\begin{equation}
H_{\rm TL}=\int_0^{\infty} d\omega \,\hbar \omega \,\hat{b}^{\dagger}_{\omega} \hat{b}_{\omega}, \quad
H_{\rm coupl}=i\hbar \sqrt{\frac{1}{2\pi}} \int_0^{\infty} d\omega \sqrt{\kappa(\omega)} \left(\hat{a} \hat{b}^{\dagger}_{\omega}- \hat{a}^{\dagger} \hat{b}_{\omega} \right).
\end{equation}
Here the decay rate $\kappa(\omega)$ effectively models the capacitive coupling of the transmission line through the resonator. Taking $\kappa(\omega)$ to be frequency-independent, i.e., $\kappa(\omega)=\kappa$, is a simple Markovian approximation that we take from now on. 
We have the Heisenberg equation for $\hat{b}_{\omega}(t)$:
\begin{equation}
\frac{d \hat{b}_{\omega}}{dt}=\frac{i}{\hbar}[H,\hat{b}_{\omega}(t)]=-i \omega \hat{b}_{\omega}(t)+\sqrt{\frac{\kappa}{2\pi}}\hat{a}(t), 
\end{equation}
which can be formally solved backwards and forwards in time, i.e., for $t > t_0$ and $t < t_1$ we have
\begin{align}
    \hat{b}_{\omega}(t)=e^{-i \omega(t-t_0)}\hat{b}_{\omega}(t_0)+\sqrt{\frac{\kappa}{2\pi}}\int_{t_0}^t d\tau \,e^{-i \omega(t-\tau)} \hat{a}(\tau), \label{eq:for} \\
\hat{b}_{\omega}(t)=e^{-i \omega(t-t_1)}\hat{b}_{\omega}(t_1)-\sqrt{\frac{\kappa}{2\pi}}\int_{t}^{t_1} d\tau \,e^{-i \omega(t-\tau)} \hat{a}(\tau). 
\label{eq:back}
\end{align}
Now one {\em defines} the input field $\hat{b}_{\rm in}(t)$ and output field $\hat{b}_{\rm out}(t)$ as
\begin{align}
    \hat{b}_{\rm in}(t)=-\sqrt{\frac{1}{2\pi}}\int_0^{\infty}d\omega e^{-i\omega(t-t_0)} \hat{b}_{\omega}(t_0),\;\;\hat{b}_{\rm out}(t)=\sqrt{\frac{1}{2\pi}}\int_0^{\infty}d\omega e^{-i\omega(t-t_1)} \hat{b}_{\omega}(t_1).
    \label{eq:define-inout}
\end{align}
The reason for the sign-flip in these definitions is that in the absence of any resonator, the incoming wave simply reflects at the end of the transmission line and flips the sign of its amplitude \footnote{One is however free to choose a different convention.}.

We can consider the Heisenberg equation for $\hat{a}(t)$
\begin{align}
\frac{d \hat{a}(t)}{dt}=\frac{i}{\hbar}[H_{\rm sys},\hat{a}(t)]-\sqrt{\frac{\kappa}{2\pi}}\int_0^{\infty} d\omega \,\hat{b}_{\omega}(t),    
\end{align}
and use Eqs.~\eqref{eq:for},\eqref{eq:back}, the definitions of $\hat{b}_{\rm in}(t)$ and $\hat{b}_{\rm out}(t)$ and $\int_{t_0}^t d\tau \,\delta(t-\tau) \hat{a}(\tau)=\frac{1}{2} \hat{a}(t)$ to obtain 
\begin{align}
    \frac{d\hat{a}(t)}{dt}=\frac{i}{\hbar} [H_{\rm sys},\hat{a}(t)]-\frac{\kappa}{2} \hat{a}(t)+\sqrt{\kappa} \hat{b}_{\rm in}(t).
    \label{eq:in}
\end{align}
and 
\begin{align}
    \frac{d\hat{a}(t)}{dt}=\frac{i}{\hbar} [H_{\rm sys},\hat{a}(t)]+\frac{\kappa}{2} \hat{a}(t)-\sqrt{\kappa} \hat{b}_{\rm out}(t).
    \label{eq:out}
\end{align}
Together, these equations lead to
\begin{align}
    \hat{b}_{\rm in}(t)+\hat{b}_{\rm out}(t)=
    \sqrt{\kappa}\hat{a}(t).
    \label{eq:inout}
\end{align}

Now, if one knows $\langle b_{\rm in}(t) \rangle$, then Eq.~\eqref{eq:in} and Eq.~\eqref{eq:inout} allow one to resolve the dynamics of the resonator mode, using $H_{\rm sys}$, and find the outgoing field $\langle b_{\rm out}(t) \rangle$. For a simple undriven resonator ($\langle b_{\rm in}(t) \rangle =0$), Eq.~\eqref{eq:in}, with the resonator starting in a coherent state $\ket{\alpha}$ defined in Eq.~\eqref{eq:coh-state}, expresses the decay of the amplitude of the coherent state $\alpha(t)=\bra{\alpha} \hat{a}(t) \ket{\alpha}$ at rate $\kappa$.

\section{Capacitively coupling a transmon to a resonator}
\label{sec:JC-coupling}

\begin{figure}[htb]
    \centering
    \includegraphics[height=4cm]{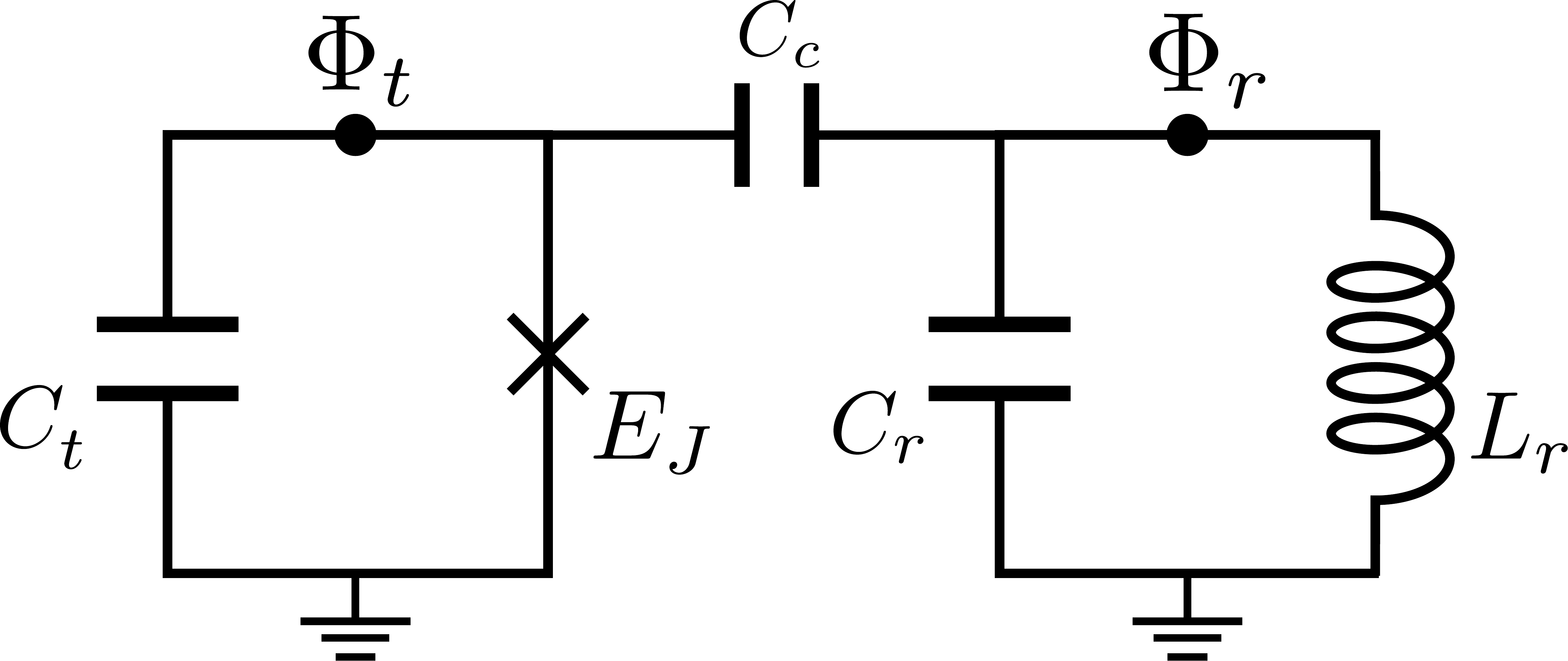}
    \caption{A transmon capacitively coupled to a resonator.}
    \label{fig:transmon_res}
\end{figure}

In superconducting circuits for quantum information processing, transmons are routinely capacitively coupled to microwave resonators or to other transmons. Besides a capacitive coupling, galvanic coupling and inductive coupling also occur in some setups, in particular for flux qubits. By galvanic coupling is meant that the qubit and its coupled system (which could be a resonator or something modeled as another qubit) form connected pieces of superconducting material, unlike for a capacitive or mutual inductive coupling. 

Capacitive couplings between transmons and resonators occur in three main circumstances:
\begin{enumerate}
\item The resonator is used to perform measurements of the qubit state via the so-called dispersive readout technique~\cite{blais_2004, Mallet2009, Blais_2021}. In this case one usually talks about a readout resonator. 
\item Two (or more) transmons are coupled to a resonator that mediates the interaction between them (see for instance the setup in Ref.~\cite{DiCarlo2009} and Fig.~\ref{fig:riste}). In this case, the resonator is called a `bus' resonator. Two transmons can also be capacitively coupled via an intermediate (flux-tunable) transmon so that changes in the resonant frequency of the intermediate transmon can enhance or decrease the effect of such a coupler \cite{martinis:tune, dial:tune}.
\item The transmon is used to steer the state of the resonator, in which the quantum information is encoded, or a transmon is used to couple two resonators which store quantum information. In this type of setup the information-containing states in the resonators can be advantageously chosen to suppress and correct errors using the general paradigm of bosonic quantum error correction (see e.g. Refs.~\cite{Terhal_2020, MA20211789} for some review). 
\end{enumerate}
Here we study the fundamental circuit behind these different applications; it is shown in Fig.~\ref{fig:transmon_res} where a transmon is capacitively coupled to a lumped LC oscillator. While this circuit contains all the physics that we need, in real-world applications the resonator is not physically realized as a lumped-element circuit. Instead, the resonator mode is usually selected as one of the modes of a finite, 2D transmission line or of a 3D microwave cavity. By selecting we mean that the relevant mode is closer in frequency to the characteristic frequency of the transmon, compared to all other modes, and in addition exhibits a large coupling to it. This allows us to neglect the other modes that are present in a resonator, and effectively map the system to that shown in Fig.~\ref{fig:transmon_res} with some effective capacitance and inductance characterizing the selected mode. We refer the reader to Refs.~\cite{gely2017, malekakhlagh2017, Parra-Rodriguez_2018} for several examples that go beyond this simple model and take into account the multi-mode character of the resonators. 

\begin{figure}[htb]
    \centering
    \includegraphics[width=0.4\textwidth]{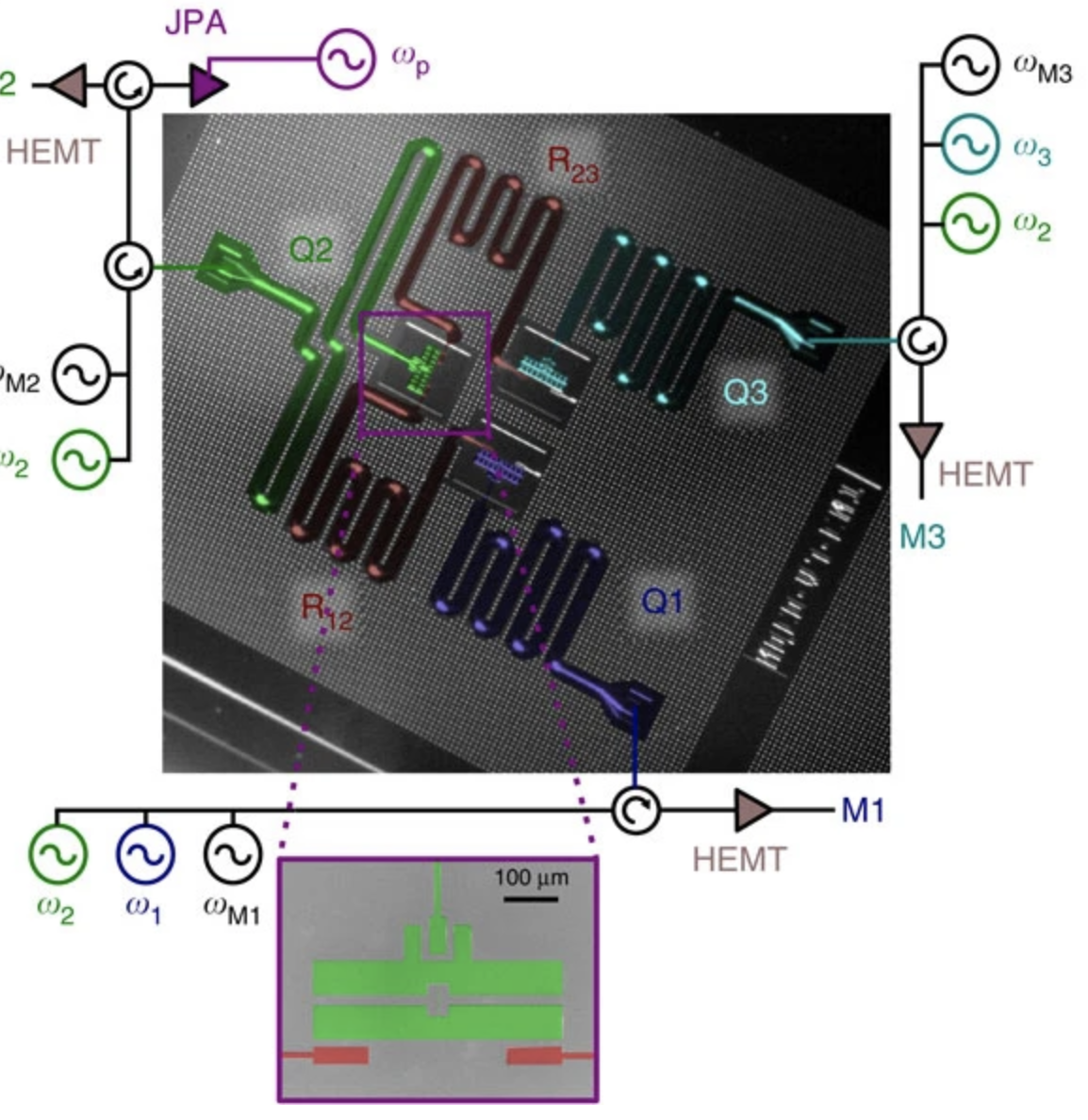}
    \caption{False-colored optical micrograph (from Ref.~\cite{chow:strand}) with three transmon qubits (boxes) $Q_1,Q_2,Q_3$, each with its own readout resonator in purple, green and turquoise respectively. Qubit $Q_2$ is coupled via the red bus resonators to $Q_1$ and $Q_3$. Zoomed in is the green transmon, consisting of two capacitive islands (connected via a Josephson junction, not visible) and its capacitive coupling to the green readout resonator and the red bus resonator.
    }
    \label{fig:chip-trans}
\end{figure}

We follow the general prescription for obtaining the Lagrangian of an electrical circuit developed in Section~\ref{sec:canq_el} to obtain the Lagrangian of the circuit in Fig.~\ref{fig:transmon_res}, which reads

\begin{equation}
\lagrangian = \frac{C_t}{2} \dot{\Phi}_t^2 + \frac{C_r}{2} \dot{\Phi}_r
 + \frac{C_c}{2} \bigl( \dot{\Phi}_t - \dot{\Phi}_r \bigr)^2 + E_J \cos \biggl( \frac{2 \pi}{\Phi_0} \Phi_t \biggr) - \frac{\Phi_r^2}{2 L_r}.
\end{equation}

Given our choice of independent variables, the capacitance matrix of the circuit is
\begin{equation}
\mat{C} =  \begin{pmatrix} C_t + C_c & -C_c \\
-C_c & C_r+ C_c\end{pmatrix}.
\end{equation}
which gives the definition of the conjugate variables 
\begin{equation}
\begin{pmatrix} Q_t \\
Q_r 
\end{pmatrix}= \mat{C} \begin{pmatrix} \dot{\Phi}_t \\
\dot{\Phi}_r 
\end{pmatrix}.
\end{equation}
We obtain the Hamiltonian of a transmon (or more generally of a CPB) capacitively coupled to a resonator, which in quantized form reads
\begin{equation}
H = \frac{\hat{Q}_t^2}{2 \tilde{C}_t} + \frac{\hat{Q}_r^2}{2 \tilde{C}_r} + \frac{\hat{Q}_r \hat{Q}_t}{\tilde{C}_c}  - E_J \cos \biggl( \frac{2 \pi}{\Phi_0} \hat{\Phi}_t \biggr) + \frac{\hat{\Phi}_r^2}{2 L_r}.
\end{equation}
Here we define the `equivalent' capacitances as 
\begin{subequations}
\begin{equation}
\tilde{C}_t = \frac{\mathrm{det}(\mat{C})}{C_r + C_c},
\end{equation}
\begin{equation}
\tilde{C}_r = \frac{\mathrm{det}(\mat{C})}{C_t + C_c}, 
\end{equation}
\begin{equation}
\tilde{C}_c = \frac{\mathrm{det}(\mat{C})}{C_c},
\end{equation}
\end{subequations}
with $\mathrm{det}(\mat{C}) = C_r C_t + C_c C_t + C_c C_r$ the determinant of the capacitance matrix. As one can easily check, the equivalent transmon capacitance $\tilde{C}_t$ corresponds to the capacitance that would shunt the transmon if we replace the inductance $L_r$ of the resonator with an open circuit, and apply the simple rules for capacitances in series and in parallel. Analogously, the equivalent resonator capacitance $\tilde{C}_r$ can be obtained by replacing the Josephson junction with an open circuit using the same procedure. 

We can rewrite the Hamiltonian in terms of rescaled variables as
\begin{equation}
H = 4 E_{C_t} \hat{q}_t^2 + 4 E_{C_r} \hat{q}_r^2 + 8 E_{C_c} \hat{q}_r \hat{q}_t -E_J \cos \hat{\phi}_t + \frac{E_{L_r}}{2} \hat{\phi}_r^2,
\end{equation}
where $E_{C_t}$, $E_{C_r}$ and $E_{C_c}$ are the charging energies associated with $\tilde{C}_t, \tilde{C}_r$ and $\tilde{C}_c$, respectively, and $E_{L_r}$ the inductive energy of the resonator.

We now perform the transmon approximation described in Section~\ref{sec:tr_approx}, valid for $E_J/E_{C_t} \gg 1$. Note that in order to evaluate the validity of this approximation it is the total transmon capacitance $\tilde{C}_t$ (which defined $E_{C_t}$) that matters, and not only the ``physical" shunting capacitance $C_{t}$.  Introducing annihilation and creation for the transmon $\hat{b}, \hat{b}^{\dagger}$ as in Section~\ref{sec:tr_approx}, and for the oscillator $\hat{a}, \hat{a}^{\dagger}$ as in  Section~\ref{subsec:lc}, we obtain the approximate Hamiltonian (neglecting constant terms)

\begin{equation}
H \approx \tilde{H}_{\rm tr+osc} = \hbar \Omega_t \hat{b}^{\dagger} \hat{b} + \hbar \frac{\delta_t}{2} \hat{b}^{\dagger} \hat{b}^{\dagger} \hat{b} \hat{b} + \hbar \omega_r \hat{a}^{\dagger} \hat{a} - \hbar g (\hat{a}^{\dagger} - \hat{a} ) (\hat{b}^{\dagger} - \hat{b}),
\label{eq:tr-osc-first}
\end{equation}
with $\hbar \Omega_t = \sqrt{8 E_J E_{C_t}} -E_{C_t} $, $\hbar \delta_t = -E_{C_t}$, $\hbar \omega_r = \sqrt{8 E_{L_r} E_{C_r}} = \hbar /\sqrt{\tilde{C}_r L_r} $, and where we define the coupling coefficient $g$ as 
\begin{equation}
\label{eq:g_coup}
\hbar g = \frac{\hbar}{2 \tilde{C}_{c} \sqrt{Z_r Z_t}} = 2 E_{C_c} \biggl (\frac{E_J}{2 E_{C_t}}\biggr)^{\nicefrac{1}{4}} \biggl (\frac{E_{L_r}}{2 E_{C_r}}\biggr)^{\nicefrac{1}{4}},
\end{equation}
with $Z_{r (t)} = \sqrt{L_{r (t)}/\tilde{C}_{r (t)}}$ the characteristic impedance of the resonator (resp. transmon). Note that the parameter dependence on the energy scales is the same as for the uncoupled transmon in Eqs.~\eqref{eq:freqt} and \eqref{eq:anharmont}, except that now the scales use the equivalent capacitance $\tilde{C}_t, \tilde{C}_r,\tilde{C}_c$. Eq.~\eqref{eq:g_coup} shows that given a certain equivalent coupling capacitance $\tilde{C}_{c}$, the coupling coefficient $g$ depends only on the characteristic impedances of the resonator and the transmon, or equivalently on the ratios $E_J/E_{C_t}$ and $E_{L_r}/E_{C_r}$. In particular, $g$ decreases with $Z_{r}$ and $Z_t$. 

To gain further insight, we also rewrite Eq.~\eqref{eq:g_coup} as
\begin{equation}
\hbar g = \frac{\hbar \sqrt{\tilde{C}_r  \tilde{C}_t \omega_r (\Omega_t - \delta_t)}}{2 \tilde{C}_c}.
\end{equation}
This shows that given a fixed capacitive network, the coupling coefficient increases with the resonator frequency $\omega_r$. However, its effect on the eigenvalues of the problem vanishes if we keep the transmon frequency fixed, since then $\lim_{\omega_r \rightarrow \infty}\frac{g}{\lvert \omega_r - \Omega_t \rvert} = 0$. 

Finally, similarly to Sections~\ref{sec:tr_approx},~\ref{sec:tr_drive}, we neglect terms with unequal numbers of annihilation and creation operator (also called the `nonsecular' terms) to obtain the Hamiltonian of an extended Jaynes-Cummings model~\cite{koch2007,bishop:vacuum-rabi}:

\begin{equation}
\frac{H}{\hbar} \overset{\mathrm{RWA}}{\approx}\frac{H_{\rm tr+osc}}{\hbar} =  \Omega_t \hat{b}^{\dagger} \hat{b} + \frac{\delta_t}{2} \hat{b}^{\dagger} \hat{b}^{\dagger} \hat{b} \hat{b} +  \omega_r \hat{a}^{\dagger} \hat{a} +  g (\hat{a}^{\dagger} \hat{b} + \hat{a} \hat{b}^{\dagger} ).
\label{eq:gen_jc}
\end{equation}

If we truncate the anharmonic transmon mode to a qubit Hilbert space, one obtains the regular Jaynes-Cummings (JC) model
\begin{equation}
    \frac{H_{\rm JC}}{\hbar}=\frac{\Omega_t}{2} \sigma_z + \omega_r \hat{a}^{\dagger}\hat{a}+g (\hat{a}^{\dagger} \sigma^-+\hat{a}\sigma^+),
    \label{eq:jc}
\end{equation}
with $\sigma^{+} = \ket{e}\bra{g} = (\sigma^-)^{\dagger}$ and $\sigma_z= \ket{e}\bra{e}-\ket{g}\bra{g}$ \footnote{We use $\sigma_x,\sigma_y,\sigma_z$ for the Pauli matrices instead of $X,Y,Z$ when we use levels $\ket{g}$ and $\ket{e}$ instead of $\ket{0}$ and $\ket{1}$. In the quantum information convention, i.e. $Z=\ket{0}\bra{0} - \ket{1}\bra{1}$, we would need an additional sign in front of $\Omega_t/2$.}.
Note that we can shift the spectrum upwards by adding $\frac{\Omega_t}{2} \mathds{1}\otimes \mathds{1}$ to the Hamiltonian so that the state $\ket{0,g}$ has zero energy. The model in which we truncate the transmon space to a qubit while also keeping the nonsecular terms is called the Rabi model. The Rabi model has less symmetry than the Jaynes-Cummings model but is still integrable \cite{braak}, and has been probed experimentally using a transmon qubit in e.g. \cite{langford}.

\begin{Exercise}[label=exc:JCmod]
\Question Show explicitly that the operator associated with the total number of excitations $\hat{n} = \hat{a}^{\dagger} \hat{a} + \hat{b}^{\dagger} \hat{b}$ is an invariant under the dynamics generated by the Hamiltonian in Eq.~\eqref{eq:gen_jc}. 
\Question Use Eq.~\eqref{eq:bch} to show that 
\begin{equation}
e^{i \theta \hat{a}^{\dagger} \hat{a}} \hat{a} e^{-i \theta \hat{a}^{\dagger}\hat{a} } = e^{-i \theta} \hat{a},
\label{eq:bosheis}
\end{equation}
and the same is valid for $\hat{b}$. Use this result to argue that the spectrum of the Hamiltonian in Eq.~\eqref{eq:gen_jc} only depends on the absolute value of $g$. \par 
\Question Suppose that, starting from the Hamiltonian in Eq.~\eqref{eq:gen_jc}, you want to use the first transition of the transmon, i.e., between ground $\ket{g}$ and first-excited state $\ket{e}$, to emulate a two-level system coupled to a harmonic oscillator via a Jaynes-Cummings Hamiltonian in Eq.~\eqref{eq:jc}. First of all, what is the parameter $g_{JC}$ in this case? How would you choose the resonator frequency to approximately realize this model (explain why)? In particular, what could happen if you chose $\omega_r < \Omega_t$? \par 
\emph{Comment: note that the JC Hamiltonian also preserves the number of excitations $\ket{e}\bra{e} + \hat{a}^{\dagger} \hat{a}$. }
\Question What is the spectrum of the Jaynes-Cummings Hamiltonian in Eq.~\eqref{eq:jc} in the resonant case $\omega_r = \Omega_t$ in the single-excitation sector (the observation of these two energy levels is called vacuum Rabi splitting). What is it in the $n$-excitation sector? How could you experimentally determine $|g_{JC}|$? 
\end{Exercise}

\begin{Answer}[ref={exc:JCmod}]
\Question The Heisenberg equation of motion for a generic operator $\hat{A}$ evolving under Hamiltonian $H$ reads
\begin{equation}
\frac{d \hat{A}}{dt} = \frac{i}{\hbar} [H,\hat{A}(t)],
\end{equation}
 Thus, an operator is an invariant of the motion if at any time it commutes with the Hamiltonian. If the Hamiltonian is time-independent, it suffices to show this at a certain time. 
In our case we simply need to show that $\hat{n}=\hat{a}^{\dagger} \hat{a} + \hat{b}^{\dagger} \hat{b}$ commutes with $H_{\rm tr+osc}$ in Eq.~\eqref{eq:gen_jc}. Obviously, $\hat{n}$ commutes with $\hat{a}^{\dagger} \hat{a}$ and $\hat{b}^{\dagger} \hat{b}$. Let us analyze the two remaining terms. First consider namely $\hat{b}^{\dagger} \hat{b}^{\dagger} \hat{b} \hat{b}$ and $\hat{a}\hat{b}^{\dagger} + \hat{a}^{\dagger}\hat{b}$:
\[
[\hat{n}, \hat{b}^{\dagger} \hat{b}^{\dagger} \hat{b} \hat{b}] = [\hat{n}, \hat{b}^{\dagger} (\mathds{1}-\hat{b} \hat{b}^{\dagger}) \hat{b}]=0.
\]
Now consider
\begin{multline}
[\hat{n}, \hat{b}^{\dagger} \hat{a} + \hat{a}^{\dagger} \hat{b}] = [\hat{a}^{\dagger} \hat{a}, \hat{b}^{\dagger} \hat{a} + \hat{a}^{\dagger} \hat{b}] + [\hat{b}^{\dagger} \hat{b}, \hat{b}^{\dagger} \hat{a} + \hat{a}^{\dagger} \hat{b}] = \\
[\hat{a}^{\dagger} \hat{a}, \hat{a}^{\dagger}] \hat{b} + [\hat{a}^{\dagger} \hat{a}, \hat{a}] \hat{b}^{\dagger} + \hat{a}^{\dagger}[\hat{b}^{\dagger} \hat{b}, \hat{b}] + \hat{a}[\hat{b}^{\dagger} \hat{b}, \hat{b}^{\dagger}] = \hat{a}^{\dagger} \hat{b} - \hat{a}\hat{b}^{\dagger} -\hat{a}^{\dagger} \hat{b} + \hat{a} \hat{b}^{\dagger} = 0.
\end{multline}
This shows that $[\hat{n}, H_{\rm tr+osc}] = 0$. This result means that we can simultaneously diagonalize $\hat{n}$ and $H_{\rm tr+osc}$, or equivalently block-diagonalize $H_{\rm tr+osc}$ within subspaces with a fixed eigenvalue of $\hat{n}$. 
\Question Note that
\begin{equation}
[\hat{a}^{\dagger} \hat{a}, \hat{a}] = -\hat{a}.
\end{equation}
Thus,
\begin{equation}
C_n = [i \theta \hat{a}^{\dagger} \hat{a}, C_{n-1}] = (i \theta)^n (-1)^n \hat{a},
\end{equation}
with $C_0 = \hat{a}$. Hence,
\begin{equation}
e^{i \theta \hat{a}^{\dagger} \hat{a}} \hat{a} e^{-i \theta \hat{a}^{\dagger}\hat{a} } = \hat{a} \sum_{n=0}^{+ \infty} \frac{(-i\theta )^n}{n !} = e^{-i \theta} \hat{a}.
\end{equation}
This shows that applying the unitary $e^{i \pi \hat{a}^{\dagger} \hat{a}}$ to $H_{\rm tr+osc}$ we can change the sign of $g$. Since the spectrum of a Hamiltonian is invariant under unitary transformations, this shows that it can only depend on $\lvert g \rvert$. 
\Question In order to obtain a Jaynes-Cummings Hamiltonian using the first two levels we need to apply the projector $\Pi = (\ket{e}\bra{e} + \ket{g} \bra{g}) \otimes I$ to $H_{\rm tr+osc}$. We get 
\begin{equation}
\frac{\Pi H_{\rm tr+osc} \Pi}{\hbar} = \frac{\Omega_t}{2} \ket{e}\bra{e} -\frac{\Omega_t}{2} \ket{g}\bra{g} + \omega_r \hat{a}^{\dagger} \hat{a} + g \bigl(\hat{a}^{\dagger} \sigma^- +\hat{a}\sigma^+  \bigr).
\end{equation}
Thus, we get $g = g_{JC}$. In order for the previous approximation to be valid, we need to assume that the eigenstates of the oscillator are mostly hybridized with the first two transmon eigenstates $\ket{g}$ and $\ket{e}$. This is the case if we take $\omega_r$ relatively close to the transmon frequency $\Omega_t$. If we take $\omega_r < \Omega_t$ we run the risk of having the resonator frequency very close to the second transition of the transmon between the first-excited state $\ket{e}$ and second-excited state $\ket{f}$. In fact, this transition has frequency $\Omega_t + \delta_t < \Omega_t$ since the anharmonicity of the transmon is negative. When $\omega_r \approx \Omega_t + \delta_t$ we would rather obtain a JC model using the levels $\ket{e}$ and $\ket{f}$. 

\Question Since the JC Hamiltonian preserves the total number of excitations, it can be diagonalized in each subspace with a fixed number of excitations. If we shift the spectrum upwards so that $\ket{0,g}$ has zero energy, then, in the single-excitation manifold in the resonant case, we get
\begin{equation}
H_{\rm JC}^{(1)}/\hbar = \omega_r \ket{0,e} \bra{0,e} + \omega_r \ket{1,g} \bra{1,g} + g_{JC} (\ket{1,g} \bra{0,e} + \mathrm{h.c}). 
\end{equation}
Explicitly in matrix form in the $\ket{0,e}, \ket{1,g}$ basis we have 
\begin{equation}
H_{\rm JC}^{(1)}/\hbar = \omega_r \mathds{1} + g_{JC} \sigma_x,
\end{equation}
with $\mathds{1}$ the $2 \times 2$ identity matrix and $\sigma_x$ is the Pauli matrix. Trivial diagonalization of this matrix gives eigenvalues $E_{\pm }=\hbar(\omega_r\pm g_{JC})$ whose difference is $2 \hbar g_{JC}$. Hence one can find $g_{JC}$ by measuring the energy difference between these two energy levels in the single-excitation manifold, for example using spectroscopy, that is, driving the transmon+resonator system and observing enhanced absorption or transmission at those frequencies, depending on the set-up. Instead, in the $n$-excitation manifold, in the resonant case, we have
\[
H_{\rm JC}^{(n)}/\hbar = n\omega_r \ket{n-1,e} \bra{n-1,e} + n\omega_r \ket{n,g} \bra{n,g} + \sqrt{n}g_{JC} (\ket{n,g} \bra{n-1,e} + \mathrm{h.c}). 
\]
with eigenvalues $E_{\pm}=\hbar(n\omega_r \pm \sqrt{n}g_{JC})$ and thus an energy level splitting equal to $2\hbar \sqrt{n} g_{JC}$.
\end{Answer}

\section{Perturbative analysis tools}
\label{sec:SW}
In circuit QED it is often useful to derive effective Hamiltonians describing low-energy dynamics in the presence of perturbations. One such tool is the Born-Oppenheimer approximation discussed in Appendix~\ref{sec:elim}. Another tool is Schrieffer-Wolff (SW) perturbation theory which is commonly used for the analysis of two-qubit gates and couplers. Here we review some of the basic features, while we refer to \cite{BDL:SW, TW:nonadiabatic, petrescu:nonRWA} for a general mathematical analysis and some applications.

Given is a Hamiltonian $H=H_0+\epsilon V$ where $H_0$ is easy to (block)-diagonalize in some number of $m$ blocks. Typically $H_0$ is the Hamiltonian of some uncoupled qubits and oscillators where one block represents the low-energy uncoupled states. Let $\Delta$ be the minimal gap between the spectra in different blocks of $H_0$. $V$ is a perturbation which can be a sum of a block-off-diagonal perturbation $V_{\rm OD}$ and a block-diagonal perturbation $V_{\rm D}$. In principle, one assumes that $\epsilon ||V|| < \Delta/2$ for the perturbation theory to apply, but for many-body systems in which $||V||$ grows with system-size while $\Delta$ is constant, one has to work under relaxed constraints \cite{BDL:SW}.

Pictorially we have for, say, two blocks:
\begin{equation*}
H_0=\begin{pmatrix}
   \colorbox{black}{\,} &   \\    
   & \colorbox{black}{\,}  \\    
\end{pmatrix}, 
V_{\rm OD}=\begin{pmatrix}
   & \colorbox{black}{\,}  \\    
  \colorbox{black}{\,} & \\    
\end{pmatrix},
V_{\rm D}=\begin{pmatrix}
   \colorbox{black}{\,} &   \\    
   & \colorbox{black}{\,}  \\    
\end{pmatrix},
\end{equation*}
where the black boxes only contain non-zero matrix elements. The goal is to absorb the effect of the perturbation $V$ by performing a unitary rotation $U=\exp(S)$ on $H$ such that 
\begin{equation}
H_{\rm eff}=U H U^{\dagger},
\label{eq:defeff}
\end{equation}
is again block-diagonal in the same way as $H_0$. Additionally, $S$ itself is block-off-diagonal and anti-Hermitian, satisfying $S=-S^{\dagger}$ by definition. The purpose of this is to obtain an $H_{\rm eff}$ which has a relatively simple description, allowing one to analytically understand the effect of the perturbation. This is useful in particular for many-body systems which resist numerical investigation.
The unitary $U$ is the Schrieffer-Wolff (SW) transformation. This unitary maps the uncoupled qubit basis set by $H_0$ to a so-called `dressed' basis.
$S$ can be perturbatively determined, that is, one takes
\begin{equation}
    S=\sum_{k=1}^{\infty} \epsilon^k S^{(k)},
\end{equation}
and requires that each $S^{(k)}$ is block-off diagonal, i.e., for two blocks $S^{(k)}=\begin{pmatrix}
   & \colorbox{black}{\,}  \\    
  \colorbox{black}{\,} & \\    
\end{pmatrix}$. The requirement that $H_{\rm eff}$ is fully block-diagonal at each perturbative order constrains, and allows one to resolve, $S^{(k)}$ at each order in $\epsilon$. This goes as follows.
From Eq.~\eqref{eq:defeff}, we have 
\[
H_{\rm eff}=H+[S,H]+\frac{1}{2}[S,[S,H]]+\ldots
\]
In this expansion, we consider terms of order $\epsilon^0$, of order $\epsilon^1$, of order $\epsilon^2$ etc. separately. At zeroth order we have $H_{\rm eff}^{(0)}=H_0$ and $U=I$. In first order we have 
\begin{equation}
H_{\rm eff}^{(1)}=H_0+\epsilon V+\epsilon [S^{(1)}, H_0],
\end{equation}
which needs to be block-diagonal. Since $H_0$ is block-diagonal (D) and $S^{(1)}$ is block-off-diagonal (OD), $[S^{(1)}, H_0]$ is block-off-diagonal, in short $[{\rm D},{\rm OD}]={\rm OD}$. Thus, to cancel any block-off-diagonal part at first order we need to have 
\begin{equation}
    V_{\rm OD}=[H_0,S^{(1)}].
\end{equation}
This allows one to determine $S^{(1)}$, say, by evaluating this expression in the eigenbasis $\{\ket{\psi_k}\}$ of $H_0$, i.e.,
\[
\bra{\psi_k}  S^{(1)} \ket{\psi_l}=\frac{\bra{\psi_k} V_{\rm OD} \ket{\psi_l}}{E_k-E_l}.
\]
Thus in first order we have 
\[
H_{\rm eff}^{(1)}=H_0+\epsilon V_{\rm D}.
\label{eq:firstorderSW}
\] In second order we have
\begin{multline}
 e^{S}(H_0+\epsilon V) e^{-S}= H_{\rm eff}^{(1)} +\epsilon^2 [S^{(1)}, V_{\rm OD}]+\epsilon^2 [S^{(1)},V_{\rm D}]+ \frac{\epsilon^2}{2}[S^{(1)},\underbrace{[ S^{(1)}, H_0]]}_{-V_{\rm OD}} +\epsilon^2 [S^{(2)},H_0] +O(\epsilon^3) \\ = H_0+\frac{\epsilon^2}{2}[S^{(1)}, V_{\rm OD}] +\epsilon^2 [S^{(1)},V_{\rm D}]+\epsilon^2 [S^{(2)},H_0] +O(\epsilon^3).
\end{multline}
Again we impose that $H_{\rm eff}^{(2)}$ must be block-diagonal at order $\epsilon^2$. 
We can now use that $[{\rm D},{\rm OD}]={\rm OD}$ and {\em only} for two blocks $[{\rm OD},{\rm OD}]={\rm D}$. Since $S^{(2)}$ is block-off-diagonal, this implies that for the case of two blocks one has
\[
H_{\rm eff}^{(2)}=H_0+\frac{\epsilon^2}{2}[S^{(1)}, V_{\rm OD}],
\]
and
\begin{align}
[S^{(1)},V_{\rm D}]+[S^{(2)},H_0]=0,
\label{eq:s2}
\end{align}
which can be solved to determine $S^{(2)}$ (when $V_{\rm D}=0$, $S^{(2)}=0$). If we have more than two blocks, we can split $[S^{(1)},V_{\rm OD}]$ into a block-diagonal part which contributes to $H_{\rm eff}^{(2)}$ and a block-off-diagonal part which is added to Eq.~\eqref{eq:s2}. One can continue with order $\epsilon^3$ etc. to develop a perturbative expansion of $H_{\rm eff}$. 

It is important not to forget that the effective Hamiltonian is the Hamiltonian of the degrees of freedom in the {\em dressed} basis. For example, unwanted weak couplings can often lead to effective Hamiltonians with $ZZ$ crosstalk between transmon qubits, so this should be interpreted as $ZZ$ crosstalk between dressed qubits. Dressed qubits are the computational qubits in which the logic of the computation should take place, as the perturbative couplings cannot be removed. When time-dependent microwave drives or DC current or voltages pulses are present, then one should consider the effect of these temporal drives in this computational basis.

As a simple example we can apply a SW analysis to the Jaynes-Cummings model in which two qubits and a resonator are off-resonantly coupled, see Exercise~\ref{exc:SW}. Note that the eigenstates and eigenspectrum of the Jaynes-Cummings model are in principle straightforward to determine analytically, as the Hamiltonian is block-diagonal in two-dimensional sectors with a fixed number of excitations, so in principle no perturbation theory is needed here. The strengths of frequency shifts and couplings that one obtains in this manner deviate from a more complete analysis in which one does not truncate the transmon to a two-level system, but uses Eq.~\eqref{eq:gen_jc} instead in a SW analysis, see e.g. the expressions in Refs.~\cite{Blais_2021, juelich-lec}. 

\begin{Exercise}[title={Schrieffer-Wolff transformation for two qubits off-resonantly coupled to a resonator},label=exc:SW]
We consider the problem of two (transmon) qubits linearly coupled to a resonator. In this calculation we directly use a two-level approximation of the transmon qubit and consider the qubits to have Hamiltonian
\begin{equation}
\frac{H_{k}}{\hbar}= \frac{\Omega_k}{2} \sigma_k^z, \quad k=1,2,
\end{equation}
where $\Omega_k$ is the characteristic frequency of the qubit $k$. Here we assume the quantum optics convention $\sigma_{k}^z = \ket{e}\bra{e}_k - \ket{g}\bra{g}_k$, $\sigma_{k}^+ = \ket{e}\bra{g}_k$, $\sigma_{k}^- = \ket{g}\bra{g}_k$. The interaction with the resonator is described by the bilinear term
\begin{equation}
\frac{V_{k}}{\hbar}= g_k (\hat{a}^{\dagger} \sigma_k^{-} + \hat{a} \sigma_k^{+} ).
\end{equation}
We thus consider the total Hamiltonian
\begin{equation}
\frac{H}{\hbar}= H_r+ H_1+ H_2 + V_1+V_2 = \omega_r \hat{a}^{\dagger} \hat{a} +  \frac{\Omega_1}{2} \sigma_1^z +  \frac{\Omega_2}{2} \sigma_2^z + g_1 (\hat{a}^{\dagger} \sigma_1^{-} + \hat{a} \sigma_1^{+}) + g_2 (\hat{a}^{\dagger} \sigma_2^{-} + \hat{a} \sigma_2^{+}),
\end{equation}
with $\omega_r$ the resonator frequency. We want to construct an effective Hamiltonian using the Schrieffer-Wolff (SW) transformation, i.e., $H_{\rm eff}= e^{S} H e^{-S}$ via a perturbative expansion of the operator $S$. In this exercise we ignore $\hbar$ from now (set it to 1) for notational convenience.
\Question  What are the small (dimensionless) parameters of the perturbation in our case when the qubit frequencies and the resonator frequency are far-detuned? 
\Question Recalling the perturbative expansion of the $S$ operator of the SW transformation,
\begin{equation}\label{eq:SW1}
V=[H_0, S^{(1)}],
\end{equation}
confirm that the first order approximation of the operator $S$ equals
\begin{equation}
S^{(1)} = \sum_{k=1}^{2} S_k^{(1)}=\sum_{k=1}^{2} \lambda_k \bigl(\sigma_k^{+} \hat{a} - \hat{a}^{\dagger} \sigma_k^{-} \bigr), 
\end{equation}
and determine the value of $\lambda_{k}$. Notice that the subscript $k$ denotes that the operator acts only on the $k$-th qubit. 
\Question Obtain the effective Hamiltonian corresponding to the first order approximation $S^{(1)}$ of $S$ and show that this, amongst others, includes a coupling between the two qubits of the form
\begin{equation}
J \bigl(\sigma_1^+ \sigma_2^- + \sigma_1^- \sigma_2^+ \bigr), J = \frac{g_1 g_2}{2} \biggl(\frac{1}{\Delta_1} + \frac{1}{\Delta_2} \biggr), 
    \label{eq:ff}
\end{equation}
where $\Delta_k \equiv \Omega_k-\omega_r$. How do you interpret the result, i.e., what is the meaning of all the new terms? In your results, use the definition $\chi_k \equiv \frac{g_k^2}{\Delta_k}$.
\end{Exercise}

\begin{Answer}[ref={exc:SW}]
\Question Going to the rotation frame at frequency $\omega_r$ for each qubit and the resonator (see Exercise \ref{exc:rot-frame}) gives the Hamiltonian
\begin{equation}
    \frac{\tilde{H}}{\hbar}=\sum_{k=1,2} H_{0,k}+V_k,
    \end{equation}
with $H_{0,k}=\frac{\Omega_k-\omega_r}{2}\sigma_k^z$. In this rotating frame the perturbative parameters become clear. We can apply perturbation theory when the gap in $H_{0,k}$ is much larger than the strength of the perturbation $V_k$. The gap is $|\omega_r-\Omega_k|$ while $||V_k||=\max_{\psi, ||\psi||=1} \sqrt{\bra{\psi} V_k^{\dagger} V_k\ket{\psi}}=|g_k|\sqrt{n}$, i.e., the norm increases if we evaluate it on higher Fock states.
    Thus one requires
\begin{equation}
\biggl \lvert \frac{\sqrt{n} g_k}{\omega_r -\Omega_k}  \biggr \rvert \ll 1,\quad k=1,2.
\end{equation} 
in order for the approximation to work when the resonator has $n$ excitations. A less accurate answer is that at least one should have 
\begin{equation}
\biggl \lvert \frac{g_k}{\omega_r -\Omega_k} \biggr \rvert \ll 1, \quad k=1,2.
\end{equation}
\Question 
We do the SW analysis in the original frame and get 
\begin{equation}
[H_0, S^{(1)}]= \biggl[\omega_r \hat{a}^{\dagger} \hat{a} +\sum_{k=1}^{2} H_{k}, \sum_{k=1}^2 S_{k}^{(1)} \biggr]= \sum_{k=1}^2 \left([\omega_r \hat{a}^{\dagger} \hat{a}, S_{k}^{(1)}]+[H_{k}, S_{k}^{(1)}]\right),
\end{equation}
and we need to compute two kinds of commutators
\begin{equation}
[\omega_r \hat{a}^{\dagger} \hat{a}, S_{k}^{(1)}]= [\omega_r \hat{a}^{\dagger} \hat{a}, S_{k}^{(1)}] = \bigl[\omega_r \hat{a}^{\dagger} \hat{a}, \lambda_k \bigl(\sigma_k^{+} \hat{a} - \hat{a}^{\dagger} \sigma_k^{-} \bigr) \bigr] = -\omega_r \lambda_k \bigl(\sigma_k^{+} \hat{a} + \hat{a}^{\dagger} \sigma_k^{-} \bigr),
\end{equation}
and
\begin{equation}
[H_k, S_{k}^{(1)}] = \biggl[ \frac{\Omega_k}{2} \sigma_k^z, \lambda_k \bigl(\sigma_k^{+} \hat{a} - \hat{a}^{\dagger} \sigma_k^{-} \bigr) \biggr] =  \Omega_k \lambda_k \bigl(\sigma_k^{+} \hat{a} + \hat{a}^{\dagger} \sigma_k^{-} \bigr).
\end{equation}
Thus, in order for Eq.~\eqref{eq:SW1} to hold, we need to set
\begin{equation}
\lambda_k = \frac{g_k}{\Omega_k - \omega_r} = \frac{g_k}{\Delta_k}, \quad k=1,2,
\end{equation}
so that $\chi_k = g_k \lambda_k$.
\Question The effective Hamiltonian corresponding to the first-order expansion $S^{(1)}$ is given by
\begin{equation}
H_{\rm eff}= H_0+\frac{1}{2} [S^{(1)}, V].
\end{equation}
 Thus we need to compute the commutator
\begin{equation}
[S^{(1))}, V]= \sum_{k=1}^{2} \sum_{k'=1}^{2} [S_k^{(1)}, V_{k'}].
\end{equation}

Again we have to compute two kinds of commutators:
\begin{multline}
[S_k^{(1)}, V_k] =  \chi_k \bigl[ \sigma_k^{+} \hat{a} - \hat{a}^{\dagger} \sigma_k^{-}, \sigma_k^{+} \hat{a} + \hat{a}^{\dagger} \sigma_k^{-} \bigr] = \chi_k \hat{a} \hat{a}^{\dagger}[\sigma_k^+, \sigma_k^-] - \chi_k\hat{a}^{\dagger}\hat{a} [\sigma_k^-, \sigma_k^+] = \\
\chi_k \hat{a} \hat{a}^{\dagger}\sigma_k^z + \chi_k \hat{a}^{\dagger}\hat{a} \sigma_k^z = \chi_k \sigma_k^z + 2 \chi_k \sigma_k^{z} \hat{a}^{\dagger}\hat{a}. 
\end{multline}
and for $k' \neq k$
\begin{equation}
[S_k^{(1)}, V_{k'}] =  \frac{g_{k'} g_k}{\Delta_k}\bigl[ \sigma_k^{+} \hat{a} - \hat{a}^{\dagger} \sigma_k^{-}, \sigma_{k'}^{+} \hat{a} + \hat{a}^{\dagger} \sigma_{k'}^{-} \bigr] =  \frac{g_{k'} g_k}{\Delta_k} \bigl( \sigma_k^{+} \sigma_{k'}^- + \sigma_k^{-} \sigma_{k'}^+  \bigr).
\end{equation}
We obtain the effective Hamiltonian 
\begin{multline}
H_{\rm eff} = \frac{1}{2} ( \Omega_1 + \chi_1) \sigma_1^{z} + \frac{1}{2} ( \Omega_2 + \chi_2 ) \sigma_1^{z} +  (\omega_r + \chi_1 \sigma_1^z + \chi_2 \sigma_2^z ) \hat{a}^{\dagger} \hat{a} 
+ J \bigl(\sigma_1^+ \sigma_2^- + \sigma_1^- \sigma_2^+ \bigr),
\end{multline}
where we have defined the coupling 
\begin{equation}
J = \frac{g_1 g_2}{2} \biggl(\frac{1}{\Delta_1} + \frac{1}{\Delta_2} \biggr). 
\label{eq:Jcoup}
\end{equation}
The interpretation is the following. Both qubits get a shift of the original frequency by $\chi_{1}$ resp. $\chi_2$. This shift is usually called the `Lamb shift' as it is a shift in the frequency of a (hydrogen) atom coupled to the electromagnetic field, which occurs even when the electromagnetic field is in its vacuum state. In addition, the resonator frequency is shifted by $\pm \chi_k$ depending on the state of the $k$th qubit. This term is called the dispersive shift. Alternatively, this term can be viewed as a frequency shift of the qubit depending on the number of excitations in the resonator; then, it is sometimes called an AC-Stark shift. In addition, the resonator mediates an exchange (flip-flop) coupling between the qubits with coupling parameter $J$ whose strength varies with $g_i$ in Eq.~\eqref{eq:g_coup} and the detuning as shown in Eq.~\eqref{eq:Jcoup}.
\end{Answer}

In the previous exercise you have seen that resonators can be used to couple qubits via a flip-flop $J (\sigma^+_1 \sigma^-_2+\sigma^-_1 \sigma^+_2)$ interaction. For this flip-flop coupling between transmons, either through a direct interaction or via a bus resonator, one targets $J/2\pi=5-20$ MHz when the coupling should be `on'. This strength indirectly determines the speed at which one can apply a CZ gate between two transmons coupled via a bus-resonator as in, for example, Ref.~\cite{rol} where the resulting CZ gate time is 40 ns. 

\begin{Exercise}[title=Dispersive readout of a superconducting qubit, label=exc:disp-meas]
In Eq.~\eqref{eq:jc} in Section~\ref{sec:JC-coupling} we have seen that the linear interaction of a two-level system with a resonator mode can be described by a Jaynes-Cummings model. If we further assume that qubit and resonator are far-detuned (as in Exercise \ref{exc:SW}) the Jaynes-Cummings Hamiltonian in Eq.~\eqref{eq:jc} can be approximated to second order in perturbation theory by the dispersive Hamiltonian 
\begin{equation}
 H_{\rm disp}= \frac{\hbar}{2} (\Omega + \chi ) \sigma_z + \hbar (\omega_r+\chi \sigma_z )\hat{a}^{\dagger} \hat{a},
 \label{eq:dispersH}
 \end{equation}
with $\chi=g^2/\Delta$ with $\Delta= \Omega-\omega_r$.
In practice, the dispersive shift ranges from $\chi/2\pi=0.5-10$ MHz, so quite small compared to the transmon frequency.
Let us consider the scenario in which the resonator mode also weakly couples capacitively to a semi-infinite transmission line and we can use the input-output formalism discussed in Section~\ref{sec:HL_derive}, resulting in Eqs.~\eqref{eq:in}, \eqref{eq:out}, \eqref{eq:inout} which we will analyze here.
\Question Consider the two cases in which the qubit is either in the excited state $\ket{e}$ or in the ground state $\ket{g}$. By taking the Fourier transform of Eqs.~\eqref{eq:in} and \eqref{eq:inout} show that the Fourier transform of the output field $\hat{b}_{\rm out}[\omega]$ is related to the Fourier-transformed input field as 
\begin{equation}
\hat{b}_{\rm out}[\omega]= r_{e,g}(\omega) \hat{b}_{\rm in}[\omega],
\end{equation}
where the reflection coefficient $r_{e, g}(\omega)$ has modulus 1. Give the phase of the reflection coefficient depending on the state of the qubit. Here we define the Fourier transform of a generic Heisenberg operator $A(t)$ \footnote{Please note the difference between the Fourier transform of a Heisenberg operator and the mode label $\omega'$ of a (non-Heisenberg) operator such as $\hat{b}_{\omega'}$. Switching to the Heisenberg picture, one has $\hat{b}_{\omega'}(t)$ which can be Fourier-transformed as $\hat{b}_{\omega'}[\omega]\propto \delta(\omega-\omega')$.} as  
\begin{equation}
\hat{A}[\omega]=\frac{1}{\sqrt{2 \pi}} \int_{-\infty}^{+\infty} dt \,e^{i \omega t} \hat{A}(t).
\label{eq:defFT}
\end{equation}
 \Question Suppose that the input state at $t=t_0=0$ is some multi-mode coherent state defined by the equation $\ket{\Psi_{\rm in}}=\ket{\{\alpha(\omega)\}}= D(\{\alpha(\omega)\}) \ket{0}$ where we use the continuous displacement operator
\begin{equation}
D(\{\alpha(\omega)\}) = \exp \biggr[ \int_{0}^{\infty} d\omega \biggl(\alpha(\omega)\hat{b}^{\dagger}_{\omega}-\alpha^*(\omega) \hat{b}_{\omega}  \biggr) \biggl],
\end{equation}
with $\alpha(\omega)$ an envelope function and $\ket{0}$ the vacuum state for all modes. Given the result of question 1, compute $\langle b_{\rm out}[\omega]\rangle $ when the qubit is either in the state $e$ or $g$. Around which value $\omega_c$ would you center the envelope function $\alpha(\omega)$ so that the output field has a large dependence on the state of the qubit?  
 \par 
\emph{Hint: remember the property
\begin{equation}
D^{\dagger}(\{\alpha(\omega)\}) \hat{b}_{\omega \sp{\prime}} D(\{\alpha(\omega)\})= \hat{b}_{\omega \sp{\prime}} + \alpha(\omega\sp{\prime}).
\label{eq:hint_disp}
\end{equation}
}
\Question By using the Fourier transform of Eq.~\eqref{eq:in}, assuming the qubit is in either the state $e$ or $g$, determine an expression for $\hat{a}^{\dagger}(t) \hat{a}(t)$ whose expectation is the number of photons in the resonator as a function of time, i.e., $\bar{n}(t)=\langle \hat{a}^{\dagger}(t) \hat{a}(t) \rangle$. Determine $\bar{n}(t)$ in case $\alpha(\omega)=\alpha \sqrt{2\pi} \delta(\omega-\omega_c)$ in the multimode input state $\ket{\Psi_{\rm in}}$, using the result from the previous question.
\end{Exercise}

\begin{Answer}[ref={exc:disp-meas}]
 \Question Let the qubit be in the excited state $\ket{e}$. In this subspace the resonator mode still behaves as a harmonic oscillator but with a shifted frequency $\omega_e= \omega_r+\chi$. If instead the qubit is in the state $\ket{g}$, the resonator has frequency $\omega_g= \omega_r-\chi$. Consequently, in the Heisenberg picture we expect the operator $\hat{a}$ to evolve differently depending on the state of the qubit and to produce a qubit-state dependent output field. \par 
Taking the Fourier transform of the evolution equation in the Heisenberg picture for $\hat{a}(t)$ we readily get
\begin{equation}
-i \omega \hat{a}[\omega]= -i \omega_{e/g} \hat{a}[\omega]-\frac{\kappa}{2} \hat{a}[\omega]+\sqrt{\kappa} \hat{b}_{\rm in}[\omega],
\end{equation}
which gives immediately
\begin{equation}
\hat{a}[\omega]= \frac{\kappa \hat{b}_{\rm in}[\omega]}{i (\omega_{e/g}-\omega)+\kappa/2}.
\label{eq:FTa}
\end{equation}
From the Fourier-transformed input-output relation $\hat{b}_{\rm out}[\omega]= \sqrt{\kappa} \hat{a}[\omega] -\hat{b}_{\rm in}[\omega]$ we get
\begin{equation}
\hat{b}_{\rm out}[\omega] = \frac{\kappa/2-i(\omega_{e/g}-\omega)}{\kappa/2+i(\omega_{e/g}-\omega)} \hat{b}_{\rm in}[\omega] = r_{e/g}(\omega) \hat{b}_{\rm in}[\omega],
\end{equation}
where we define the reflection coefficient
\begin{equation}
r_{e/g}(\omega) = \frac{\kappa/2+i(\omega-\omega_{e/g})}{\kappa/2-i(\omega-\omega_{e/g})}.
\end{equation}
The reflection coefficient clearly has modulus one, being of the general form $r_e=(C+i D)/(C-i D)$ with phase $\theta= \arctan [-2 C D/(C^2-D^2)]= -\arctan [2 C D/(C^2-D^2)]$. We can use the formula 
\begin{equation}
\frac{1}{2}\arctan \biggl(\frac{2 x}{1-x^2}\biggr) = \arctan x
\end{equation} 
to get 
\begin{equation}
r_{e/g}(\omega)= e^{i \theta_{e/g} (\omega)},
\end{equation}
with
\begin{equation}
\theta_{e/g} (\omega)= -2 \arctan \biggl[ \frac{2(\omega-\omega_{e/g})}{\kappa} \biggr]= -2 \arctan \biggl[ \frac{2 (\omega-\omega_r)}{\kappa} \pm 2 \frac{\chi}{\kappa} \biggr].
\end{equation}
\Question Using Eq.~\eqref{eq:hint_disp} we have 
\[
\bra{\Psi_{\rm in}} \hat{b}_{\omega}(t_0=0) \ket{\Psi_{\rm in}}=\alpha(\omega),
\]
since $\bra{0} \hat{b}_{\omega'} \ket{0}=0$.
Then we write down and work out the expectation of the operator $\hat{b}_{\rm out}[\omega]$
\begin{multline}
    \bra{\Psi_{\rm in}} \hat{b}_{\rm out}[\omega] \ket{\Psi_{\rm in}}=
\bra{\Psi_{\rm in}}  r_{e/g}(\omega) \hat{b}_{\rm in}[\omega]  \ket{\Psi_{\rm in}} = \\
\frac{-r_{e/g}(\omega)}{2 \pi} \int_{-\infty}^{+\infty} dt e^{i \omega t} \int  d \omega \sp{\prime}   e^{-i \omega \sp{\prime}t}  \bra{\Psi_{\rm in}}\hat{b}_{\omega'}(t_0=0)  \ket{\Psi_{\rm in}}. \notag
\end{multline}
where we used the definition in Eq.~\eqref{eq:in} and the Fourier transform.

With $\frac{1}{2\pi}\int_{-\infty}^{+\infty} dt e^{i\omega t}= \delta(\omega)$, we get
\begin{equation}
\langle \hat{b}_{\rm out}[\omega]\rangle =  -r_{e/g}(\omega) e^{i \omega t_0} \alpha(\omega).
\end{equation}
From the result of the first point of the exercise we know that the functions $r_{e,g}(\omega)$ are particularly sensitive to the state of the qubit for values of $\omega$ close to the bare cavity frequency $\omega_r$. For this reason we would like the output field to preserve this dependency. We conclude that we would like the envelope function $\alpha(\omega)$ to be particularly peaked at values close to $\omega_r$. Typically, one chooses the carrier (center) frequency to be the bare cavity frequency $\omega_r$ and a narrow frequency profile around it, corresponding to the pulse shape.
\Question From Eq.~\eqref{eq:FTa} and Fourier-transforming back, we get
\[
\bar{n}(t)=\langle \hat{a}^{\dagger}(t) \hat{a}(t) \rangle = \frac{1}{2\pi}\int d\omega \int d\omega' e^{i t (\omega-\omega')} \frac{\kappa\langle \hat{b}_{\rm in}^{\dagger}[\omega]]\hat{b}_{\rm in}[\omega']\rangle}{
(\kappa/2 - i(\omega_{e/g}-\omega))(\kappa/2 + i(\omega_{e/g}-\omega'))}.
\]
We again use 
\[
\hat{b}_{\rm in}[\omega]=\frac{-1}{2\pi}\int dt e^{i \omega t}\int d\omega' e^{-i \omega(t-t_0)}\hat{b}_{\omega'}(t_0)=-\hat{b}_{\omega}(t_0=0).
\]
(note that the minus sign is just due to the convention of how we define $b_{\rm in}$). Then we have 
\begin{multline}
   \bra{\Psi_{\rm in}} \hat{b}^{\dagger}_{\omega}(t_0=0) \hat{b}^{\dagger}(\omega')(t_0=0) \ket{\Psi_{\rm in}}=\bra{0} (\hat{b}^{\dagger}_{\omega}+\alpha^*(\omega))(\hat{b}_{\omega'}+\alpha(\omega'))\ket{0}=\\
   \alpha^*(\omega)\alpha(\omega')=|\alpha|^2 2\pi \delta(\omega-\omega_c) \delta(\omega'-\omega_c).
 \end{multline}
This gives
\[
\bar{n}(t)=\frac{\kappa |\alpha|^2}{\kappa^2/4+(\omega_{e/g}-\omega_{c})^2}.
\]
which is a steady-state number of photons in the resonator. When $\omega_c$ is chosen as $\omega_e$ (qubit in $e$) or $\omega_g$ (when qubit is in $g$), this number of photons is maximal (driving at resonance), and only depends on $\kappa$ and the amplitude $\alpha$ of the plane wave.
\end{Answer}

\chapter{Linear networks and black-box quantization}
\label{chap:ln}

In order to obtain the Lagrangian and the Hamiltonian of a circuit, we have so far made the tacit assumption that a lumped-element circuit representation of our system is available and accurate. This is clearly not always the case, for instance when a transmon qubit is put inside a 3D rectangular microwave cavity~\cite{paik2011}. 

At this point, it should not be immediately obvious how a Hamiltonian for this system can be derived and most of all that an electrical-circuit representation of the problem exists. The theory underlying the superconducting devices and chips is in principle quantum electrodynamics ---with a continuous set of modes---  interacting with superconducting matter. On the other hand, a finite electrical circuit always has a finite number of modes, although one can of course consider the limit of the number of modes going to infinity, like for a transmission line as in the previous chapter.

Having a physics or electrical-engineering background, one should know that continuous microwave structures are usually described in terms of scattering parameters, a concept that we review in Section~\ref{subsec:scat}. By numerically simulating the designed microwave structures, using electromagnetic analysis software tools, one can determine parameters such as the impedance (matrix) as a function of frequency, which directly relate to the scattering parameters. These parameters represent complete knowledge of the behaviour of the system and, accordingly, we expect that if a Hamiltonian or Lagrangian formulation exists, it must be possible to derive it from them. In this chapter we show how this can be done via a method called black-box quantization (Section~\ref{sec:bb}), extended here in comparison to its first introduction in Ref.~\cite{Nigg.etal.2012:BlackBoxCqed}. In this method, we do not need to pass necessarily through the Lagrangian of the system, but can derive the Hamiltonian directly. 

IBM has developed software for quantum device design, known as Qiskit Metal, that uses electromagnetic simulation as input, available at \url{https://qiskit.org/metal/}. It is based on the participation ratio modeling proposed in Ref.~\cite{Minev2021} which we will discuss in Section~\ref{subsec:epratio}. Another software tool for converting parameters obtained in electromagnetic simulations to a Hamiltonian description including losses was developed in Ref.~\cite{Gely_2020}; there are some exercises using this software in Section~\ref{sec:qucat}.

\section{LTI networks}
\label{sec:mport}

In this section, we provide a short introduction to {\em linear time-invariant} (LTI) networks. While the concept can be applied to any network in which we have some inputs and some outputs, we only consider the case in which inputs and outputs are either currents or voltages. \par  

In Fig.~\ref{fig:el_network} we depict a general $N$-port electromagnetic network. We can think of a port as a pair of terminals with a well-defined voltage between them, such that the current that goes in at one terminal is equal to the current that goes out at the other terminal. We refer the reader to Section~\ref{sec:mtn} and to Chapter~$4$ of Ref.~\cite{pozar} for a more thorough discussion of the concept of a port. Let $\vect{\vv}(t)$ and $\vect{\ii}(t)$ denote the $N$-dimensional vectors of port voltages and port currents respectively at time $t$. A general electromagnetic network relates $\vect{\vv}(t)$ and $\vect{\ii}(t)$ by some rule. When the currents are the inputs and the voltages are the outputs, an electrical LTI network relates currents and voltages by a convolution, i.e.,
\begin{equation}\label{eq:lti}
\vect{\vv}(t) = \int_{-\infty}^{+\infty} d \tau \,\mat{\zz}(t - \tau) \vect{\ii}(\tau).
\end{equation}
The $N \times N$ matrix $\bm{\zz}(t)$ is called the {\em impedance} matrix in the time domain (unit: ohm/sec.). The fact that $\mat{\zz}$ depends only on $t-\tau$ captures the time-invariance of the response. If, instead, we exchange the role of inputs and outputs, we have
\begin{equation}\label{eq:lti2}
\vect{\ii}(t) = \int_{-\infty}^{+\infty} d \tau \,\mat{\yy}(t - \tau) \vect{\vv}(\tau),
\end{equation}
where $\mat{\yy}(t)$ is called the {\em admittance} matrix in the time domain (units: siemens/sec. or inverse (ohm*sec.)).

The matrix elements of $\mat{\zz}(t)$ have a simple interpretation. Namely, $\zz_{kl}(t)$ is the voltage $\vv_{k}(t)$ when a Dirac-delta current $\delta(t)$ is applied at port $l$, while all other port currents are zero (hence open-circuited). Thus, $\mat{\zz}(t)$ has the interpretation of an impulse-response function. Similarly, 
$y_{kl}(t)$ is the current $i_k(t)$ when a Dirac delta voltage $\delta(t)$ is applied at port $l$ while all other ports voltages are zero (hence short-circuited). The defining equation of an LTI network, Eq.~\eqref{eq:lti}, aims at capturing delay effects in the system, while still keeping the relation linear in the input.

\begin{figure}
    \centering
    \includegraphics[scale=0.15]{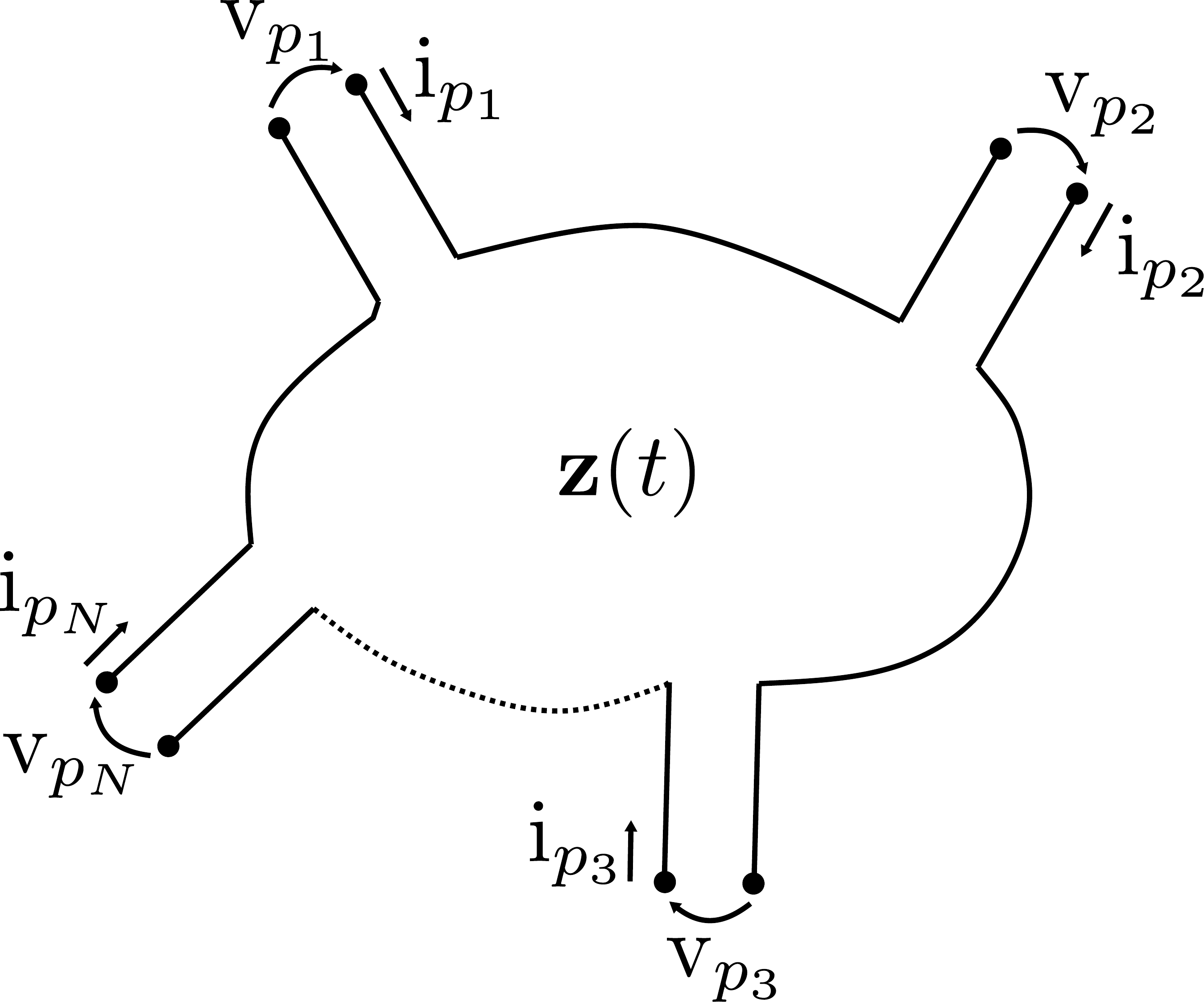}
    \caption{Electrical $N$-port network with voltage $\vv_{p_i}$ across port $p_i$ and current $\ii_{p_i}$ entering and exiting port $p_i$.}
    \label{fig:el_network}
\end{figure}

We can argue that a relation like Eq.~\eqref{eq:lti} for general inputs and outputs can be a good description of many systems\footnote{Many input-output systems, not only electromagnetic, can be treated to a good approximation as LTI systems. Let us consider a phenomenon that was influencing our daily lives a while ago: as input we take a variable that quantifies the government measures to limit the spread of COVID-19, while as output we have a variable which quantifies the effect of the measures, like the reduction of the number of cases. As we know from our experience, if the government applies a measure at time $t$ it will only have an effect on the reduction of cases after some characteristic time $t_c$. Moreover, we somehow expect that at least for \emph{soft} measures the system will behave linearly in the input, while one should expect a saturation when the measures are \emph{strong}. Therefore, if the measures are soft, we expect a convolution relation to be able to capture the behaviour of the system that can thus be approximately treated as an LTI system. There are many other examples where similar reasoning applies, which explains why LTI networks are ubiquitous in nature.}. However, Eq.~\eqref{eq:lti} still allows a property that is completely nonsensical, namely we do not expect that an input at time $t$ can have an effect on the output at time $t'<t$, while this is still allowed by Eq.~\eqref{eq:lti}. Consequently, we add the additional requirement that the network is causal, which can be formulated in terms of the impulse response function as \footnote{Note that here we are assuming \emph{instantaneous} causality, while in principle one should also take the speed of light into account.}
\begin{equation}\label{eq:z_caus}
\forall t<0,\;\mat{\zz}(t) = 0,  \;\; \forall t<0,\;\mat{\yy}(t) = 0.
\end{equation}
Notice that this implies that the upper limit of the convolution integrals in Eqs.~\eqref{eq:lti} and~\eqref{eq:lti2} can be taken to be $t$ and not $+\infty$. 
The interested reader can read more on causality in LTI networks in Ref.~\cite{triverio2007}. We note that, since the response of the electrical network is linear, it precludes the use of Josephson junctions: these are exactly the elements that one can attach at the ports. The network can otherwise be lossy (include resistances) or lossless.

\subsection{Laplace and Fourier domain}
\label{subsec:LF}

Instead of time-dependent impedance and admittance matrices, it is more common to use these matrices in the Laplace and Fourier domain which we define here. Let $\vect{f}(t)$ be a generic time-dependent vector with coefficients $f_i: \mathbb{R} \rightarrow \mathbb{R}^N$. The bilateral Laplace transform of $\vect{f}(t)$ is defined as 
\begin{equation}\label{eq:laplace}
\vect{F}(s) = \int_{-\infty}^{\infty} dt\, e^{-s t} \vect{f}(t), \quad s = \sigma + i \omega \in \mathbb{C}.
\end{equation}
The Laplace transform is defined only in the so-called region of convergence (ROC), i.e., the region of the complex plane in which the integral in Eq.~\eqref{eq:laplace} converges. 
The Fourier domain can be seen as a particular case of the Laplace domain with $\sigma=0$. We will use the notation
\begin{equation}
\vect{\mathcal{F}}(\omega) = \vect{F}(i \omega), \quad \omega \in \mathbb{R}.
\end{equation}

Using the convolution theorem for the Laplace or Fourier transform, Eq.~\eqref{eq:lti} in the Laplace domain becomes
\begin{equation}
\vect{V}(s) = \mat{Z}(s) \vect{I}(s),
\end{equation}
and similarly Eq.~\eqref{eq:lti2} gives
\begin{equation}
\vect{I}(s) = \mat{Y}(s) \vect{V}(s),
\end{equation}
from which one can deduce that 
\begin{equation}
\vect{Z}(s)=\vect{Y}^{-1}(s).
\label{eq:inverseY}
\end{equation}
A note on units: one sees that the units of $V(s)$ are volts*sec.~({\em not} volts!), units of $I(s)$ are amps*sec., units of $Z(s)$ are ohms, and units of $Y(s)$ are 1/ohms.

In close analogy with the time domain case, we can give an interpretation to the matrix elements $Z_{kl}(s)$ of the impedance matrix in the Laplace domain as the ratio between $V_{k}(s)$ and $I_{l}(s)$, when all but the $l$th port are open-circuited (currents $I_{k\neq l}(s)=0$).  Mathematically, this translates to
\begin{equation}
Z_{kl}(s) = \frac{V_k(s)}{I_{l}(s)} \biggl \lvert_{I_{k}(s)=0, \, k \neq l}.
\label{eq:def_impedance_mat_elems}
\end{equation}
Similarly, when all but the $l$th port are short-circuited ($V_{k\neq l}(s)=0$) and a voltage $V_l(s)$ is applied, its current response is captured by
\begin{equation}
    Y_{kl}(s)=\frac{I_k(s)}{V_l(s)}\biggl\vert_{V_k(s)=0, \,k \neq l}.
\end{equation}
It should also be mentioned that an impedance or admittance matrix does not always exist for an arbitrary $N$-port circuit. For example, the circuit can be such that $\ii_l(t) \neq 0$ also implies that $\ii_k(t) \neq 0$ (like for a transformer, a two-port element, see Section~\ref{subsec:ideal-trafo}), making Eq.~\eqref{eq:def_impedance_mat_elems} ill-defined and thus the impedance matrix non-existent. In Section~\ref{sec:mtn}, we discuss the physical interpretation of the matrix elements of the admittance matrix in the Fourier domain in relation to the property of non-reciprocity. 

\subsection{Single port}
\label{subsec:sp}

In case there is a single port $N=1$, the network could be a single two-terminal branch, or it can represent a more general network to which we simply have access through this single port. We thus have
\begin{equation*}
V(s)=Z(s) I(s),
\end{equation*}
and 
\begin{equation*}
I(s)=Y(s) V(s),
\end{equation*}
so that function-wise $Z(s)=Y^{-1}(s)$. Let us first review the impedance and admittance of some simple elements and how these are constructed.
In case the branch is an inductor $L$ with $\vv(t)=d\Phi/dt=L d\ii/dt$, its impedance can be found as $Z(s)=s L$ by a Laplace transform and thus $Y(s)=\frac{1}{sL}$. For a capacitor $C$, $\ii(t)=C d\vv(t)/dt$, and thus, $Y(s)=sC$, while $Z(s)=\frac{1}{s C}$. We see that for these branches $Z(i \omega)=\mathcal{Z}(\omega)$ is purely imaginary. For a resistor $\vv(t)=\ii(t) R$, so $\mathcal{Z}(\omega)=R$ is real.

For a general linear network, its impedance can be constructed from its constituent components. For this, we should remember that impedances $Z_1$ and $Z_2$ (like resistors) add when we place them in series, i.e.,~$Z=Z_1+Z_2$. When placed in parallel, we can use $1/Z=1/Z_1 + 1/Z_2$, to obtain $Z = \frac{Z_1 Z_2}{Z_1 + Z_2}$.  The reverse is naturally true for admittances. It follows that a general lossless, passive linear network, i.e., one composed of inductors and capacitors, $Z(i \omega)$ is purely imaginary.  

\begin{Exercise}[label=exc:admit-LC]
Check for yourself that the admittance of an LC oscillator equals
\begin{equation}
\mathcal{Y}(\omega)=i\omega C+\frac{1}{i \omega L}=\frac{1-\omega^2 LC}{i \omega L}.
\label{eq:singleLC}
\end{equation}
\end{Exercise}

\begin{Answer}[ref=exc:admit-LC]
The admittances of the inductor and the capacitor add. 
\end{Answer}

For the LC oscillator, we see that $\mathcal{Y}(\omega)$ has a zero at the resonant frequency $\omega_r=1/\sqrt{LC}$, and, equivalently, $\mathcal{Z}(\omega)=\mathcal{Y}^{-1}(\omega)$ has a pole at this resonant frequency.
As it turns out, a single-port lossless, passive linear network can always be represented as a series of LC oscillators (Foster's theorem) ---we come back to this in Section~\ref{subsec:ssr} and Appendix~\ref{app:norm_mode}--- but the idea is that the zeros of $\mathcal{Y}(\omega)$ are the resonant frequencies of these LC oscillators.

\subsection{Scattering matrix}
\label{subsec:scat}
So far we have not made any assumption about the nature of the system that is attached to the ports of the network. This system could simply be a current source, or it could be a nonlinear element, like a Josephson junction, as we will consider in the black-box quantization method in Section~\ref{sec:bb}. 

However, it is useful to first consider the case in which we attach transmission lines to the ports of the network as depicted in Fig.~\ref{fig:el_network_tl}. In that case the linear network will scatter incoming EM waves to outcoming EM waves, captured by a scattering matrix $\vect{s}$ which is defined as follows. 

\begin{figure}
    \centering
    \includegraphics[scale=0.15]{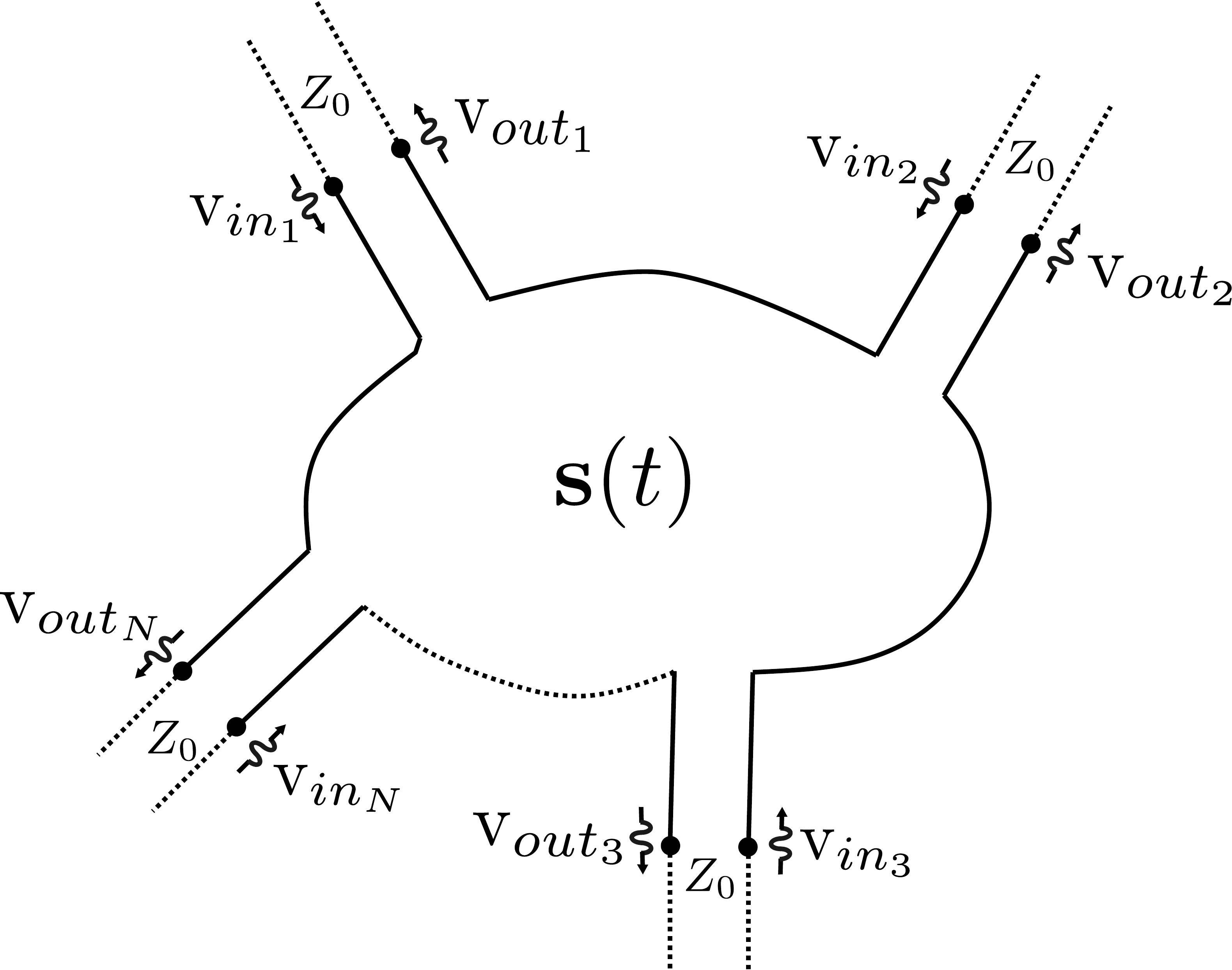}
    \caption{Electrical $N$-port network with transmission lines attached at the ports. The response of the network to incoming radiation is characterized by its scattering matrix $\mat{s}(t)$. The voltage on transmission line $k$ is a sum of an incoming voltage $\vv_{{\rm in},k}(t)$ plus an outgoing voltage $\vv_{{\rm out},k}(t)$, as discussed in Section~\ref{sec:tl-cp}.}
    \label{fig:el_network_tl}
\end{figure} 

Let $\vect{\vv}_{\rm in}(t)$, $\vect{\vv}_{\rm out}(t)$ be the vectors of ingoing and outgoing voltages on the transmission lines that are attached to the ports. Output voltages $\vect{\vv}_{\rm out}(t)$ and input voltages $\vect{\vv}_{\rm in}(t)$ are related by a convolution relation
\begin{equation}\label{eq:scat}
\vect{\vv}_{\rm out}(t) = \int_{-\infty}^t d \tau \,\mat{\scat}(t - \tau) \vect{\vv}_{\rm in}(t),
\end{equation}
where $ \mat{\scat}(t)$ is called the scattering matrix in the time domain. In the Laplace domain, Eq.~\eqref{eq:scat} gives 
\begin{equation*}
\vect{V}_{\rm out}(s)=\mat{S}(s)\vect{V}_{\rm in}(s).
\end{equation*}

Let us us now relate the scattering matrix $\mat{S}(s)$
to the impedance $\mat{Z}(s)$ and admittance matrix $\mat{Y}(s)$ of the network. At a single port, we have the total voltage $\vv(t)=\vv_{\rm in}(t)+\vv_{\rm out}(t)$ and the total current $\ii(t)=\ii_{\rm in}(t)-\ii_{\rm out}(t)$ (here we take the convention that $\ii_{\rm in}$ resp. $\ii_{\rm out}$ is the current associated with the ingoing resp. outgoing wave). For a transmission line the relation between voltage and current is: $\vv_{\rm in}(t)=Z_0 \ii_{\rm in}(t)$ and $\vv_{\rm out}(t)=Z_0 \ii_{\rm out}(t)$ where $Z_0$ is the characteristic impedance, as shown in Eq.~\eqref{eq:char-imp} in Section~\ref{sec:tl-cp} \footnote{Comparing the analysis here with that of the single-port setting of Section~\ref{sec:tl-cp}, note that $\vv_{\rm in}$ and $\vv_{\rm out}$ are denoted there as $\vv^{\rightarrow}$ and $\vv^{\leftarrow}$, see Eqs.~\eqref{eq:char-imp}. Further, note that there is a sign difference in Eq.~\eqref{eq:char-imp}, namely $\vv^{\leftarrow}(x,t)=-Z_0 \ii^{\leftarrow}(x,t)$, since in this derivation the current of the left-going, i.e. outgoing, wave is taken with respect to the current definition in Eq.~\eqref{eq:iphi} and $i(t)=i^{\rightarrow}(t)+i^{\leftarrow}(t)$. This different convention gives of course the same Eqs.~\eqref{eq:vin_out}.}. This implies the (vectorial) relations for all ports:
 \begin{subequations}\label{eq:vin_out}
\begin{align}
\vect{\vv}_{\rm in}(t) & = 
\frac{\vect{\vv}(t) + Z_0 \vect{\ii}(t)}{2}, \\
\vect{\vv}_{\rm out}(t) & = 
\frac{\vect{\vv}(t) - Z_0 \vect{\ii}(t)}{2}.
\end{align}
\end{subequations}
Using Eq.~\eqref{eq:lti} and Eqs.~\eqref{eq:vin_out}, we can obtain the following integral relation between $\vect{\vv}_{\rm in}(t)$ and $\vect{\vv}_{\rm out}(t)$ in a (causal) LTI network
\begin{equation}\label{eq:zs_td}
\int_{- \infty}^t d \tau [\mat{\zz}(t-\tau) + Z_0\delta(t -\tau) \mathds{1} ] \vect{\vv}_{\rm out}(\tau) = \int_{- \infty}^t d \tau [\mat{\zz}(t-\tau) - Z_0\delta(t -\tau) \mathds{1} ]\vect{\vv}_{\rm in}(\tau).
\end{equation}

By Laplace transforming Eq.~\eqref{eq:zs_td} we obtain the following relation between impedance and scattering matrices in the Laplace domain (see also
page~52 in Ref.~\cite{newcomb}):
\begin{equation}\label{eq:scat-lap}
\mat{S}(s) = \left(\mat{Z}(s) + Z_0 \mathds{1} \right)^{-1} \left(\mat{Z}(s) - Z_0 \mathds{1} \right) = -\left(Z_0 \mat{Y}(s) + \mathds{1} \right)^{-1} \left(Z_0 \mat{Y}(s) - \mathds{1} \right).
\end{equation}
Here we have assumed that the matrices $\mat{Z}(s) + Z_0 \mathds{1}$ and $Z_0 \mat{Y}(s) + \mathds{1}$ are invertible. But an interesting fact is that $\mat{S}$ always exists for an electrical structure, even when $\mat{Z}$ or $\mat{Y}$ do not (e.g., for a transformer), but then an analysis independent of Eq.~\eqref{eq:scat-lap} is required.

\begin{Exercise}[label=exc:SZ]
Verify that equivalently $\mat{S}(s)=\left(\mat{Z}(s) - Z_0 \mathds{1} \right)\left(\mat{Z}(s) + Z_0 \mathds{1} \right)^{-1}$, i.e the matrices on the right-hand side of Eq.~\eqref{eq:scat-lap} commute. Verify the last equality in Eq.~\eqref{eq:scat-lap}.
\end{Exercise}

\begin{Answer}[ref=exc:SZ]
Using Eq.~\eqref{eq:scat-lap}, we have 
\begin{equation*}
(\mat{Z}(s) + Z_0 \mathds{1} ) \mat{S}(s) (\mat{Z}(s) + Z_0 \mathds{1} )=\mat{Z}^2(s)-Z_0^2 \mathds{1}=(\mat{Z}(s) + Z_0 \mathds{1} )(\mat{Z}(s) -Z_0 \mathds{1} ),
\end{equation*}
which leads to the expression, see~\cite{newcomb}. Alternatively, use the formal expression \begin{equation*}
(\mathds{1}+\mat{E})^{-1}=\sum_{k=0}^{\infty} (-\mat{E})^k,
\end{equation*}
which shows that the inverse $\left(\mat{Z}(s) + Z_0 \mathds{1} \right)^{-1}$ only contains powers of $\mat{Z}(s)$ and hence commutes with $\mat{Z}(s) - Z_0 \mathds{1}$. To check the last equality, we write 
\begin{equation*}
(\mat{Z}(s)+Z_0 \mathds{1})^{-1}=(\mat{Z}(s) (\mathds{1}+Z_0 \mat{Y}(s)))^{-1}=(Z_0 \mat{Y}(s)+\mathds{1})^{-1} \mat{Y}(s),
\end{equation*}
and absorb $\mat{Y}(s)$ in the next factor.
\end{Answer}

\subsection{Additional properties of LTI networks}

We discuss some additional properties of electrical networks. The first one is reciprocity. A linear network is called reciprocal if and only if its impedance matrix (assuming it exists) is symmetric 
\begin{equation}
\mat{Z}(s)=\mat{Z}^T(s).
\label{eq:def-rec}
\end{equation}
This property holds of course in any domain and would hold also for the admittance matrix. From Eq.~\eqref{eq:scat-lap} we see that in a reciprocal network the scattering matrix is also symmetric, i.e., $\mat{S}(s)=\mat{S}^T(s)$. Linear networks that are nonreciprocal will be extensively discussed in Chapter~\ref{chap:nonrec}. Linear networks involving capacitances, self- and mutual inductances, and resistances are reciprocal. Another important notion is that of a lossless linear network. This notion can be easily formulated in the Fourier domain, when the network is also reciprocal. If we remove resistances from reciprocal networks, we obtain linear networks that are reciprocal and lossless. A linear network is reciprocal and lossless if and only if its impedance matrix in the Fourier domain $\mat{\mathcal{Z}}(\omega)$ is symmetric and purely imaginary \cite{pozar}. This clearly generalizes the single-port case in Section~\ref{subsec:sp}. Also, one can show that the scattering matrix in the Fourier domain, $\mat{\mathcal{S}}(\omega)=\mat{S}(i\omega)$, of a reciprocal, lossless linear network is possibly complex, but always {\em unitary}. You can verify this by checking that
$\vect{\mathcal{S}}^{\dagger}(\omega) \vect{\mathcal{S}}(\omega)=\mathds{1}$ using Eq.~\eqref{eq:scat-lap}.

\begin{Exercise}[{label=exc:refl}]
A transmission line is terminated by a load impedance $Z_{\rm load}(s)$ where the load represents some linear electrical network.
\Question Show that the reflection coefficient $\Gamma(s)$ connecting the incoming to the outgoing voltage on the transmission line is given by
\begin{equation*}
 \Gamma(s)= \frac{Z_{\mathrm{load}}(s)-Z_0}{Z_{\mathrm{load}}(s)+Z_0}.
\end{equation*}
\Question Show that $|\Gamma(s=i\omega)|^2=1$ when $Z_{\rm load}(s=i\omega)$ represents a lossless linear network. Why would you expect this? Show that $|\Gamma(s)|^2 \leq 1$ more generally.
\end{Exercise}

\begin{Answer}[{ref={exc:refl}}]
\Question Since we have a single port to which a transmission line is attached, the scattering matrix is one-dimensional, namely a single complex number $S(s)=\Gamma(s)$. In this case, Eq.~\eqref{eq:scat-lap} becomes
\begin{equation*}
 \Gamma(s)= \frac{Z_{\mathrm{load}}(s)-Z_0}{Z_{\mathrm{load}}(s)+Z_0}.
\end{equation*}
\Question When the load is lossless, $Z_{\rm load}(s=i\omega)$ is purely imaginary. Since then the magnitudes of the real and imaginary parts of the numerator and denominator are equal, we see that $|\Gamma(i\omega)|^2=1$. If there is no loss, all radiation should be reflected. More generally, let $Z_{\rm load}(s)=a(s)+i b(s)$, with $a(s), b(s) \in \mathbb{R}$. Then $|\Gamma(s)|^2=\frac{(a(s)-Z_0)^2+b^2(s)}{(a(s)+Z_0)^2+b^2(s)}\leq 1$, representing the fact that the outgoing radiation can have lower intensity due to absorption and loss in the load.
\end{Answer}

In the next exercise we treat the finite-length transmission line as a two-port linear network itself investigate its impedance and scattering matrix. This is of interest if we would connect the resonator to another transmission line or other non-linear elements or qubits.

\begin{Exercise}[title={Impedance matrix of a finite two-port transmission line},label=exc:2port-TL]
\begin{figure}[htb]
 \centering
 \includegraphics[scale=0.2]{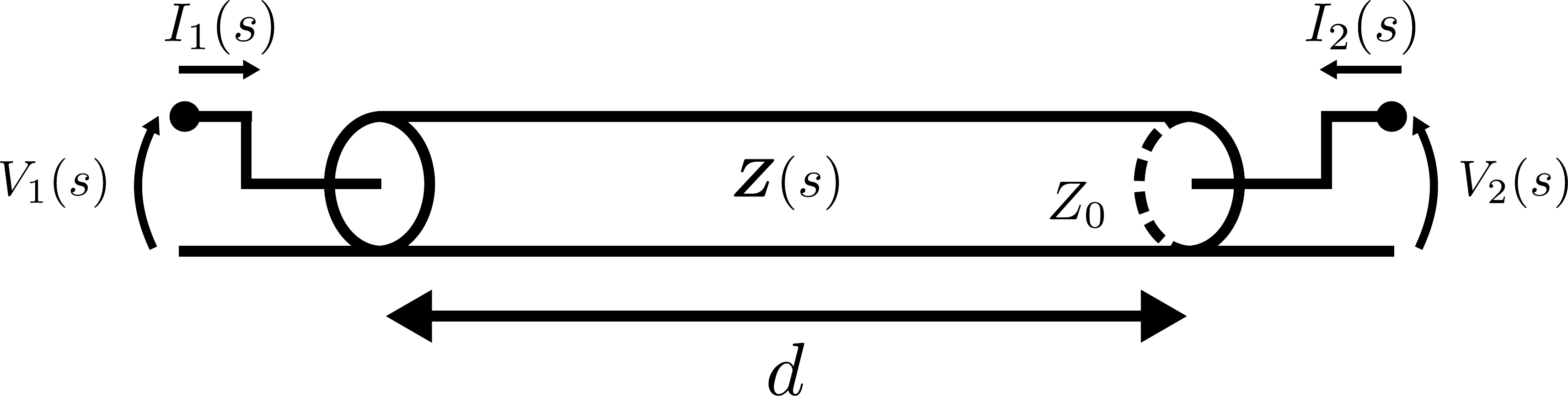}
 \caption{Two-port transmission-line network. A fictitious linear network represented by a $Z_{\rm load}(s)$ can be attached at, say, the right end.}
 \label{fig::2PTL}
 \end{figure}

In Fig.~\ref{fig::2PTL} you see a schematic of a two-port transmission-line resonator of length $d$ with characteristic impedance $Z_0$, i.e., the transmission line itself can be represented by the network in Fig.~\ref{fig:inftl}.
\Question Using Eqs.~\eqref{eq:char-imp} and \eqref{eq:lr-volt} show that the voltage and current in the Fourier domain (only consider $\omega > 0$) along the line are given by
 \begin{subequations}
\begin{equation}\label{eq::V}
 V(x, s=i\omega) = \vv_0^{\rightarrow}(\omega \sqrt{lc}) \biggl(e^{i\sqrt{\ell c} \omega x} + \Gamma(i\omega) e^{-i\sqrt{\ell c} \omega x} \biggr),
 \end{equation}
 \begin{equation}\label{eq::I}
 I(x, s=i\omega) = \frac{\vv_0^{\rightarrow}(\omega \sqrt{lc})}{Z_0} \biggl(e^{i\sqrt{\ell c} \omega x} - \Gamma(i\omega) e^{-i\sqrt{\ell c} \omega x} \biggr),
 \end{equation}
 \end{subequations}
 where $x$ denotes the distance from a fictitious load on the right-hand-side of the line (located at $x=0$), $\Gamma(s)$ is the reflection coefficient at this load and $Z_0$ is the characteristic impedance of the transmission line. These formulas hold more generally in the Laplace domain.
 \Question Which properties do you expect for the impedance matrix $\bm{\mathcal{Z}}(\omega)$ (in the Fourier domain) of this network?
  \Question Obtain the $2 \times 2$ impedance matrix $\bm{Z}(s)$ in the Laplace domain, using the definition in \cref{eq:def_impedance_mat_elems} and the fact that a zero current at a port can be modeled by an infinite load $Z_{\rm load}(s)$. Estimate the diagonal and off-diagonal elements of this matrix separately, using Eqs.~\eqref{eq::V} and \eqref{eq::I} and the expression for the reflection coefficient $\Gamma(s)$ in the answer of Exercise \ref{exc:refl} depending on $Z_{\rm load}(s)$. 
  \Question Obtain the scattering matrix $\bm{S}(s)$ in the Laplace domain as well. Which properties do you expect it to have in the Fourier domain?  
  \end{Exercise}

\begin{Answer}[ref={exc:2port-TL}]
\Question 
Fourier-transforming Eqs.~\eqref{eq:lr-volt} and considering only $\omega > 0$ (the $\omega < 0$ information is redundant):
\begin{align}
    V^{\rightarrow}(x,s=i\omega)=\vv_0^{\rightarrow}(\omega \sqrt{\ell c})e^{-i \sqrt{\ell c}\omega x}, \\
    V^{\leftarrow}(x,s=i\omega)= \,\vv_0^{\leftarrow}(\omega \sqrt{\ell c}) e^{i \sqrt{\ell c} \omega x}.
\end{align}

Due to the fictitious load at $x=0$ we need to have the reflected voltage $V^{\leftarrow}(x=0,s)=\Gamma(s) V^{\rightarrow}(x=0,s)$ and hence $\vv_0^{\leftarrow}(\omega \sqrt{lc})=\Gamma(i\omega)\vv_0^{\rightarrow}(\omega \sqrt{lc})$, relating the incoming and outgoing amplitudes.
Since the total voltage in the Fourier domain is $V(x,i\omega)=V^{\rightarrow}(x,i\omega)+V^{\leftarrow}(x,i\omega)$, and we are considering a position $x \in [-d,0]$, we obtain Eq.~\eqref{eq::V} where $x > 0$ now signifies the distance. Using the relation between $\vv^{\rightarrow}(x,t)$ and $\ii^{\rightarrow}(x,t)$ in Eq.~\eqref{eq:char-imp} which also holds in the Fourier domain, we get $I(x,i\omega)=I^{\rightarrow}(x,i\omega)+I^{\leftarrow}(x,i\omega)=\frac{1}{Z_0}(V^{\rightarrow}(x,i\omega)-V^{\leftarrow}(x,i\omega))$ from which Eq.~\eqref{eq::I} follows.
\Question We expect the impedance matrix in the Fourier domain to be symmetric and purely imaginary, since the network is reciprocal and lossless.
\Question As defined in~\cref{eq:def_impedance_mat_elems}, the matrix elements of the impedance matrix are given by for $k,l=1,2$
\begin{equation}
Z_{kl}(s)= \frac{V_{k}(s)}{I_{l}(s)} \biggl \rvert_{I_{k}=0, k \neq l}.
\end{equation}
Having zero current at a port (open circuit) is equivalent to having a load impedance $Z_\mathrm{load} \rightarrow + \infty$. We thus perform two thought experiments to determine the diagonal and the off-diagonal terms of the impedance matrix:
\begin{itemize}
\item Diagonal term for $k=1$ (at the left)
\begin{equation}
Z_{11}(s)= \frac{V(d,s)}{I(d,s)}= \lim_{Z_\mathrm{load} \rightarrow + \infty} Z_{0}  \frac{Z_{\rm }(s)+Z_0 \tanh (\sqrt{\ell c} s d)}{Z_0+Z_{\rm load}(s) \tanh (\sqrt{\ell c} s d)}= Z_0 \coth	(\sqrt{\ell c} s d).
\end{equation}
Note that by symmetry $Z_{22}(s)=Z_{11}(s)$ (i.e., by imagining the load attached to the left at port 1).
\item Off-diagonal terms 
\begin{equation}
Z_{21}(s)= \frac{V(0,s)}{I(d, s)}= \lim_{Z_\mathrm{load} \rightarrow + \infty} Z_0 \frac{1+\Gamma(s)}{e^{\sqrt{\ell c} s d} - \Gamma(s) e^{-\sqrt{\ell c} s d}} = Z_0 \, \mathrm{csch} (\sqrt{\ell c} s d).
\end{equation}
Again by symmetry $Z_{21}(s)=Z_{12}(s)$ (as  the transmission line is lossless).
\end{itemize}
Thus the impedance matrix is
\begin{equation}
\mat{Z}(s)= Z_0 \left(\begin{array}{cc}
\coth	(\sqrt{\ell c} s d) & \mathrm{csch} (\sqrt{\ell c} s d) \\
\mathrm{csch} (\sqrt{\ell c} s d) & \coth	(\sqrt{\ell c} s d)
\end{array}\right).
\end{equation}
\Question We expect the scattering matrix to be symmetric and unitary in Fourier domain. The scattering matrix can be obtained via plugging $\mat{Z}(s)$ in Eq.~\eqref{eq:scat-lap} and one gets
\begin{equation}
\bm{S}(s)= \left(\begin{array}{cc}
0 & e^{-\sqrt{\ell c} d s} \\
e^{-\sqrt{\ell c} d s}  & 0
\end{array}\right).
\end{equation}
Note that it is correctly symmetric and in Fourier domain ($s = i \omega$) it is unitary, as expected.
\end{Answer}

\begin{Exercise}[title={Impedance of a $\lambda/4$ resonator},label={exc:lambda_4}]
Consider a resonator which is grounded at one end, as discussed in Section~\ref{subsec:br}, while the other end functions as a single port. Show that this network has an impedance equal to
\begin{equation}
\mathcal{Z}(\omega)=i Z_0 \tan\left(\frac{d \omega}{v_p}\right)=i Z_0 \tan\left(\omega \sqrt{\ell c}d\right).
\label{eq:impl4}
\end{equation}
To derive this, you can use Eqs.~\eqref{eq::V} and \eqref{eq::I} on two-port transmission lines and impose the right condition representing the grounded port.
\end{Exercise}

\begin{Answer}[ref={exc:lambda_4}]
The grounded port can be represented by a `load' with $Z_{\rm load}(s)=0$ (short circuit) and hence $\Gamma(s)=-1$. This implies that
\[
\mathcal{Z}(\omega)=Z_{11}(s=i\omega)=Z_0 \frac{e^{i \sqrt{\ell c}\omega d}-e^{-i \sqrt{\ell c}\omega d}}{e^{i \sqrt{\ell c}\omega d}+e^{-i \sqrt{\ell c}\omega d}},
\]
from which the result follows.
\end{Answer}

\section{Black-box quantization}
\label{sec:bb}

\begin{figure}
\centering
\includegraphics[scale=0.15]{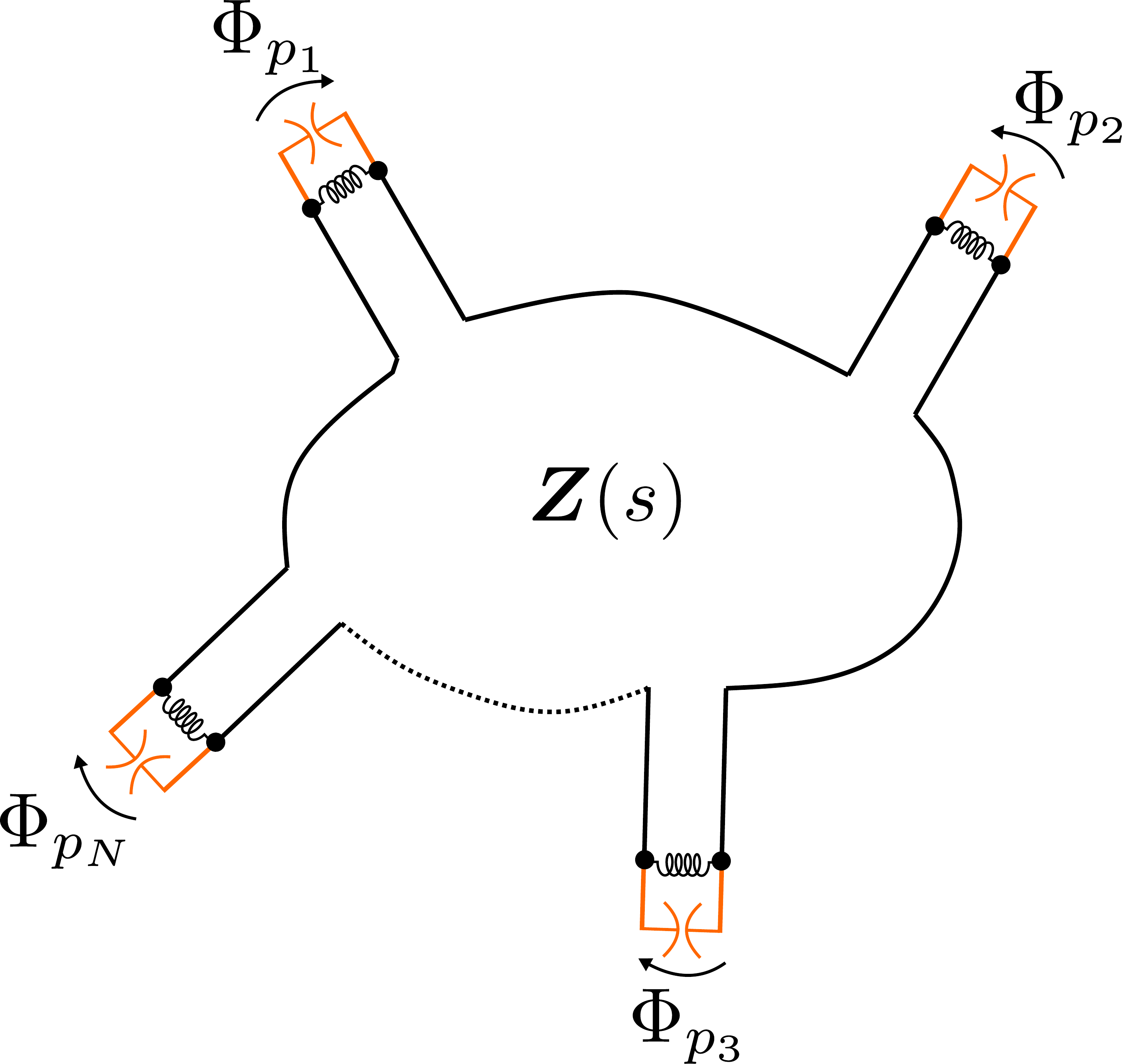}
\caption{Josephson junctions coupled to a generic microwave network. $\mat{Z}(s)$ is the impedance matrix in the Laplace domain of the linear network, including the linear part of the Josephson junction, i.e., the Josephson inductance.}
\label{fig:jj_net}
\end{figure}
The goal of black-box quantization is to obtain the Hamiltonian of a circuit composed of Josephson junctions coupled to a generic microwave network, and thus quantize the system. The underlying idea is that the only non-linear elements in this network are the Josephson junctions and that a representation of the system in terms of the degrees of freedom of the linear system (which are then affected by the Josephson junctions) is warranted. Of course, the analysis of the linear system is computationally efficient. In fact, as a quantum computational system, such a linear system has no power whatsoever beyond that of classical devices.

Black-box quantization gives a prescription to write a Hamiltonian in the normal mode basis of the linear system, even if we do not know the detailed circuit model of the network, but we only have access to its impedance or scattering matrix. These matrices can be the result of classical electromagnetic simulations of the linear part of the network which can be performed efficiently. 

Let us consider a system of Josephson junctions that are coupled to a passive, lossless and reciprocal microwave network. An example is the typical setup for transmon qubits on a superconducting chip, in which the transmons are coupled to bus and readout resonators. We can interpret the two terminals of a Josephson junction as a port of the linear microwave network. The original black-box approach of Ref.~\cite{Nigg.etal.2012:BlackBoxCqed} includes all the linear parts of the circuit in the linear network, even the linear part of the Josephson potential, i.e., the Josephson inductance. This is depicted in Fig.~\ref{fig:jj_net} for the case of $N$ junctions, where the black part of the system is described by the $N \times N$ impedance matrix $\mat{Z}(s)$ of the linear network, while the nonlinear part of the junction  is left out. More precisely, the potential of the $n$th Josephson junction is written as
\begin{equation}
U_{Jn} = -E_{Jn} \cos\biggl(\frac{2 \pi}{\Phi_0}  \Phi_{p_n}\biggr) = U_{\mathrm{spider}, n}(\Phi_{p_n}) + \frac{\Phi_{p_n}^2}{2 L_{Jn}} ,
\end{equation}
where $L_{Jn}$ is the Josephson inductance and $\Phi_{p_n}$ is the flux across the junction port, i.e., 
\begin{equation}\label{eq:port_fluxes}
\Phi_{p_n}(t) = \int_{-\infty}^t dt' \vv_{p_n}(t').
\end{equation}
The potential $U_{\mathrm{spider}, n}$ is the nonlinear part of the Josephson potential, which we depict with the spider element in orange in Fig.~\ref{fig:jj_net} while the inductors (in black) are added to the definition of $\mat{Z}(s)$.
\par 
The idea of black-box quantization using the normal-mode approach is the following. As we show in detail in Appendix~\ref{app:norm_mode} and Fig.~\ref{fig::cauer_circuit}, the linear part of the network is essentially a system of coupled oscillators, which can always be diagonalized at the classical level by finding the so-called normal modes. Thus, if we have $M$ normal modes, it must be possible to write the linear part of the Hamiltonian as
\begin{equation}
H_{\mathrm{lin}} = \sum_{m=1}^M \hbar \omega_m \hat{a}_m^{\dagger} \hat{a}_m,
\end{equation}
where $\omega_m$ is the resonant frequency of mode $m$ and $\hat{a}_m, \hat{a}_m^{\dagger}$ its annihilation and creation operators, respectively. In addition, it is possible to write the port flux operators $\hat{\Phi}_{p_n}$ as a linear combination of normal-mode fluxes $\hat{\Phi}_m$, see Eq.~\eqref{eq::phi_exp}, i.e.

\begin{equation}\label{eq:phi_port_norm}
\hat{\Phi}_{p_n} = \sum_{m=1}^M t_{mn} \hat{\Phi}_{m} =  \sum_{m=1}^M t_{mn} \sqrt{\frac{\hbar Z_m}{2}} \bigl(\hat{a}_m  + \hat{a}_m^{\dagger}\bigr),
\end{equation}
where $Z_m$ is the characteristic impedance of mode $m$ given by
\begin{equation}\label{eq:z0m}
Z_m = \frac{1}{C_0 \omega_m} \equiv \sqrt{\frac{L_m}{C_0}},
\end{equation}
where $C_0$ is a reference capacitance (needed for proper dimensionality), which is taken to be equal for all modes. This is equivalent to defining an effective inductance of the mode $m$ as 
\begin{equation}
\label{eq:eff_lm}
L_m = \frac{1}{C_0 \omega_m^2}.
\end{equation}
For later convenience we collect the coefficients $t_{mn} \in \mathbb{R}$ in Eq.~\eqref{eq:phi_port_norm} in $N$-dimensional real column vectors
\begin{equation}\label{eq:vec_turn}
\bm{t}_m = \left(\begin{array}{c}
t_{m1} \\ \vdots \\ t_{mN}
\end{array}\right), \quad m \in \{1, \dots, M\}.
\end{equation}
For reasons that will become clear later, these are called the vectors of turn ratios. \par 

Including the nonlinear potential, the total Hamiltonian of the system in the normal-mode basis equals
\begin{multline}
\label{eq:h_nm}
H = H_{\mathrm{lin}} + \sum_{n=1}^N U_{\mathrm{spider}, n}(\Phi_{p_n}) 
= \sum_{m=1}^M \hbar \omega_m \hat{a}_m^{\dagger} \hat{a}_m 
+ \sum_{n=1}^N U_{\mathrm{spider}, n}\biggl[\sum_{m=1}^M t_{mn} \sqrt{\frac{\hbar Z_m}{2}} \bigl(\hat{a}_m  + \hat{a}_m^{\dagger}\bigr)\biggr].
\end{multline}  
Thus, the linear part of the system is solved exactly, while the nonlinear part provides coupling and the possibility to generate entanglement between the modes. In addition, in order to get the Hamiltonian in the normal-mode basis, we need to extract:
\begin{itemize}
\item the normal mode frequencies $\omega_m$;
\item the vectors of turn ratios $\vect{t}_{m}$;
\item the characteristic impedances $Z_m$;
\item the Josephson energies $E_{Jn}$. One cannot obtain $E_{Jn}$ by classical electromagnetic simulations, but one could estimate each $E_{Jn}$ by estimating the critical current $I_c$ of the $n$th junction. Alternatively, one can estimate $E_J$, for a single transmon given the fabrication process of its junction, from its spectrum (using e.g. Eqs.~\eqref{eq:freqt} and \eqref{eq:anharmont}, and extrapolate this to identically-fabricated junctions.
\end{itemize}
This should be done at least for the modes that are relevant in our problem, while modes that are far away from the frequencies one is interested in could be excluded from the analysis as an approximation. 

The Hamiltonian in the normal-mode basis in Eq.~\eqref{eq:h_nm} can be obtained for any superconducting circuit involving transmons, flux qubits or fluxoniums. However, as we will argue in Section~\ref{sec:network-tr}, the normal-mode basis is more convenient when treating systems of transmon qubits, or more generally qubits with a weak nonlinearity. This is because the normal modes, with the addition of the weak nonlinearity, are, to a very good approximation, the degrees of freedom associated with the qubits. Nonetheless, the normal-mode approach has also been employed to study circuits with fluxonium qubits \cite{smith2016}. \par 

There are two typical situations in which we would like to employ black-box quantization and write the Hamiltonian as in Eq.~\eqref{eq:h_nm}. In the first case we have a lumped-element representation of the circuit. In this case, while we could use standard circuit quantization and obtain the Hamiltonian, we may still prefer to work in the normal-mode basis. This could be because the low energy spectrum of the circuit can be described by weakly anharmonic modes, for which the normal-mode basis is convenient, or because we would like to work systematically with a Hamiltonian of the form of Eq.~\eqref{eq:h_nm}. In the second case, we do not have a lumped-element representation of the circuit, but instead we have its impedance matrix in a certain frequency range, obtained, for instance, through electromagnetic simulations.

\begin{figure}
\centering
\includegraphics[height=5cm]{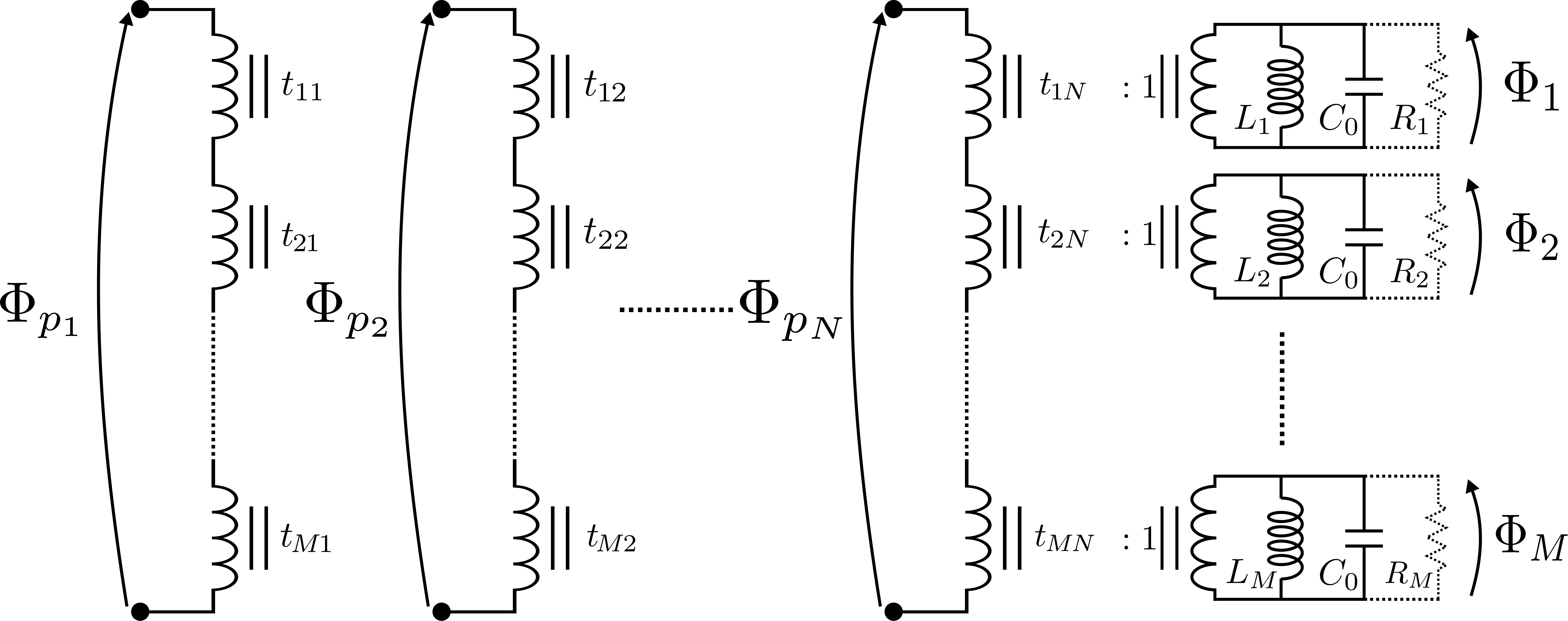}
\caption{Circuit equivalence of a $N$-port linear network, described in Chapter~$7$ of \cite{newcomb}. The inductances associated with the normal modes are given by $L_m = \frac{1}{C_0 \omega_m^2}$. When $N=1$ we see that the network can be represented by a series of $M$ LC oscillators with inductances $L_m$. When the network is weakly lossy, we can approximate it by including a resistor $R_m$ in parallel with each $m$th LC oscillator, as shown by the additional dotted branches (see the chapter of E. R. Beringer in \cite{purcell_microwave}).}
\label{fig::cauer_circuit}
\end{figure}

In order to write the Hamiltonian in the normal-mode basis, we make use of Cauer's construction that we depict in Fig.~\ref{fig::cauer_circuit}. Algebraically, this theorem implies that the impedance matrix of a passive, lossless and reciprocal LTI network can always be written as
\begin{equation}\label{eq:z_foster}
\mat{Z}(s) = \sum_{m=1}^M \frac{1}{2 C_0} \biggl(\frac{1}{s + i \omega_m} + \frac{1}{s-i \omega_m} \biggr)\mat{R}_m,
\end{equation}
where $C_0$ is the aforementioned reference capacitance. When all frequencies $\omega_m \geq 0$ are distinct, the matrices $\mat{R}_m$ ($\mat{R}$ for ``residue") are symmetric, rank-1, $N \times N$ matrices, which are fixed by the vectors of turn ratios $\vect{t}_m$, i.e.,
\begin{equation}
\mat{R}_m = \bm{t}_m \vect{t}_m^T, \;\; (\mat{R}_m)_{kl}=t_{mk}t_{ml}.
\label{eq:turn}
\end{equation}
We note that $\mat{R}_m$ leaves the overall sign of the vector, $\pm {\bm t}_m$, free, and these $M$ $\pm$ signs, one for each $\bm{t}_m$, are not physically meaningful.

We derive Cauer's construction using the Lagrangian formalism in Appendix~\ref{app:norm_mode} (see also Chapter~$7$ in~\cite{newcomb}), showing that obtaining the form of the impedance matrix in Eq.~\eqref{eq:z_foster} is, de facto, equivalent to obtaining the normal modes of the circuit.  We refer the reader to Refs.~\cite{newcomb, purcell_microwave, russer} for further discussions on this construction and microwave networks in general. \par 

Let us now explain how to extract the parameters needed to write the Hamiltonian in normal-mode form if the impedance matrix in the frequency domain $\mat{\mathcal{Z}}(\omega)$ is available. The resonant frequencies can be found by finding the poles of $\mat{\mathcal{Z}}(\omega)$ or equivalently the zeros of the admittance matrix $\mat{\mathcal{Y}}(\omega)= \mat{\mathcal{Z}}^{-1}(\omega)$.

Once we know the $M$ resonant frequencies $\omega_m$, we can also determine the vectors $\vect{t}_m$ in two ways. First of all, note that, using Eqs.~\eqref{eq:z_foster} and \eqref{eq:turn}, we have
\begin{equation}
\label{eq:turn_lim}
t_{mk} t_{ml} = \lim_{\omega \to \omega_m} 2 C_0(i \omega - i \omega_m) \mathcal{Z}_{kl}(\omega),
\end{equation}
which allows us to determine the elements of the vectors of turn ratios (up to the overall sign), assuming we are able to evaluate the limit. Given an analytical expression for $\mathcal{Z}_{kl}(\omega)$, one can use this expression to determine the $t_{km}$ coefficients.
Alternatively, for $\omega \approx \omega_m$ the  elements of the impedance matrix are
\begin{equation}\label{eq:znnp}
\mathcal{Z}_{kl}(\omega) \approx \frac{1}{2 C_0} \frac{t_{mk} t_{ml}}{i (\omega-\omega_m) }.
\end{equation}
Let $\tilde{\mathcal{Y}}_{k l}(\omega)=\frac{1}{\mathcal{Z}_{kl}(\omega)}$ (which is {\em not} the matrix element $\mathcal{Y}_{kl}( \omega)$ of the admittance matrix!). From Eq.~\eqref{eq:znnp} we obtain that
\begin{equation}
\label{eq:turn_adm}
t_{mk} t_{ml} = \frac{2 C_0}{ \mathrm{Im}\left(\frac{d\tilde{\mathcal{Y}}_{k l}}{d\omega}\Bigl|_{\omega=\omega_m}\right)}.
\end{equation}

\subsection{The energy-participation ratio approach to black-box quantization}
\label{subsec:epratio}
In this section, we discuss the energy-participation ratio method for black-box quantization introduced in Ref.~\cite{Minev2021}, which is also implemented in the Python package Qiskit Metal developed by IBM. The method is in principle completely equivalent to what we wrote in the previous section. However, it does not require extracting the vectors of turn ratios from the admittance or impedance matrix elements as in Eqs.~\eqref{eq:turn_lim} \eqref{eq:znnp}, \eqref{eq:turn_adm}. 
Instead, it shows how to determine the coefficients $\frac{2\pi}{\Phi_0} t_{mn} \sqrt{\frac{\hbar Z_m}{2}}$ in Eq.~\eqref{eq:h_nm} as an energy participation ratio of the mode $m$ into the $n$th Josephson junction obtained directly by the classical EM simulations.

 In order to highlight the fact that the energy-participation ratio is a classical concept that can thus be derived via standard electromagnetic simulations, we keep the discussion completely classical. Each normal mode flux $\Phi_{m}(t)$ satisfies the differential equation of a harmonic oscillator with frequency $\omega_m$:

\begin{equation}
\frac{d^2 \Phi_m(t)}{dt^2} = - \omega_{m}^2 \Phi_{m}(t), \quad m=1, \dots, M.
\end{equation}

Let us now consider the situation in which only the normal mode $m$ is initially ``excited". Mathematically, this translates into the following initial conditions

\begin{subequations}
\begin{equation}
\Phi_{m}(0) = \Phi_{\mathrm{in}}, \quad \frac{d \Phi_m}{dt} \biggl \lvert_{t=0} = \Phi_{\mathrm{in}}'(t),
\end{equation}
\begin{equation}
\Phi_{m'}(0) = 0, \quad \frac{d \Phi_{m'}}{dt} \biggl \lvert_{t=0} = 0, \quad m' \neq m,
\end{equation}
\end{subequations}

which give the solution

\begin{subequations}
\begin{equation}
\Phi_{m}(t) = \Phi_{\mathrm{in}} \cos(\omega_m t) + \frac{\Phi_{\mathrm{in}}'}{\omega_m} \sin(\omega_{m} t),
\end{equation}
\end{subequations}
while $\Phi_{m'}(t) =0$ for $m' \neq m$. Let us now define the energy-participation ratio $p_{m n}$ of mode $m$ at junction $n$ as the ratio between the total inductive energy stored at junction $n$ due to only mode $m$, and the total inductive energy of mode $m$ averaged over a period $2 \pi/ \omega_m$:

\begin{equation}
p_{m n} = \frac{\frac{1}{2 L_{Jn}}\int_{0}^{2 \pi/\omega_m} dt\, \Phi_{p_n}^2(t)}{\frac{1}{2 L_{m}}\int_{0}^{2 \pi/\omega_m} dt\, \Phi_{m}^2(t)} = \frac{L_m}{L_{Jn}} t_{mn}^2 = \frac{t_{mn}^2}{C_0 \omega_m^2 L_{Jn}},
\end{equation}

where we used the definition of the effective inductance $L_m$ of mode $m$ in Eq.~\eqref{eq:eff_lm}. This leads to the identification

\begin{equation}
t_{mn} = s_{mn} \omega_{m} \sqrt{C_{0} L_{Jn} p_{mn}},
\end{equation}
where $s_{mn} = \pm 1$. Thus, we see that that the energy-participation ratios allow us to obtain the turn ratios (and again in the $s_{mn}$ there is an overall irrelevant sign freedom). Ref.~\cite{Minev2021} discusses how to calculate the signs $s_{mn}$ from the classical field solutions. We refer the reader to Ref.~\cite{Minev2021b} for a more modular approach to black-box quantization based on the energy-participation ratio method.

\subsection{Single-port black-box quantization}
\label{subsec:ssr}

As an example, we obtain the transmon qubit Hamiltonian, the Hamiltonian of the circuit in Fig.~\ref{fig:LCb}, by viewing the transmon qubit as a single-port ($N=1$) network composed of an LC oscillator to which a Josephson junction is attached. The inductance of the Josephson junction $L_J$ is lumped together with the shunting capacitor to represent an LC oscillator. 

For the LC oscillator with fundamental frequency $\omega_J = 1/\sqrt{L_J C}$, the admittance $\mathcal{Y}(\omega)$ is given in Eq.~\eqref{eq:singleLC} and thus 
\begin{equation*}
\frac{d\mathcal{Y}(\omega)}{d\omega}\biggl|_{\omega=\omega_J}=2i C,
\end{equation*}
implying that $t_{11}^2=C_0/C$. Note that the irrelevant reference capacitance $C_0$ just drops out when we then apply Eq.~\eqref{eq:h_nm} with $Z_m$ in Eq.~\eqref{eq:z0m}, leading to
\begin{equation}
    H=\hbar \omega_J \hat{a}^{\dagger}\hat{a}+U_{\rm spider}\left(\sqrt{\frac{\hbar}{2C \omega_J}}(\hat{a}+\hat{a}^{\dagger})\right),
    \label{eq:single-p-ham}
\end{equation}
as a transmon Hamiltonian.

More generally, when we have single Josephson junction port $N=1$ attached to a network with $M$ frequencies, we have 
\begin{equation}
   t_{m1}=\sqrt{\frac{2C_0}{{\rm Im}(\frac{d\mathcal{Y}}{d\omega})\vert_{\omega=\omega_m}}}.
\end{equation}

We can introduce an effective mode capacitance $\tilde{C}_m$ so that the Hamiltonian reads
\begin{equation}
   H=\sum_m \hbar \omega_m \hat{a}_m^{\dagger}\hat{a}_m+U_{\rm spider}\left(\sum_{m=1}^M \sqrt{\frac{\hbar}{2\omega_m \tilde{C}_m}}(\hat{a}_m^{\dagger}+\hat{a}_m)\right),
   \label{eq:hcm}
\end{equation}
where the effective mode capacitance is thus defined as 
\begin{equation}
    \tilde{C}_m=\frac{1}{2}{\rm Im}\left(\frac{d\mathcal{Y}}{d\omega}\biggl\vert_{\omega=\omega_m}\right).
    \label{eq:eff-c}
\end{equation} Hence, instead of using a reference capacitance $C_0$, we also can rephrase the dependence on dimensionless turn ratios using the frequencies $\omega_m$ and effective capacitances $\tilde{C}_m$. Turn ratios are, however, naturally used in the statement of Cauer's construction (see Appendix~\ref{app:norm_mode}). Fig.~\ref{fig::smlossy_foster} shows the equivalence between these two approaches in the single-port case. We note that the steeper the slope is at $\mathcal{Y}(\omega_m)$, the weaker the effect of the nonlinear spider term is on the mode $m$, and thus the less anharmonic mode $m$ is. By plotting ${\rm Im}(\mathcal{Y})$ as a function of $\omega$ one can see which modes inherit most of the nonlinearity of the Josephson junction.

\begin{figure}
\centering
\includegraphics[height=5cm]{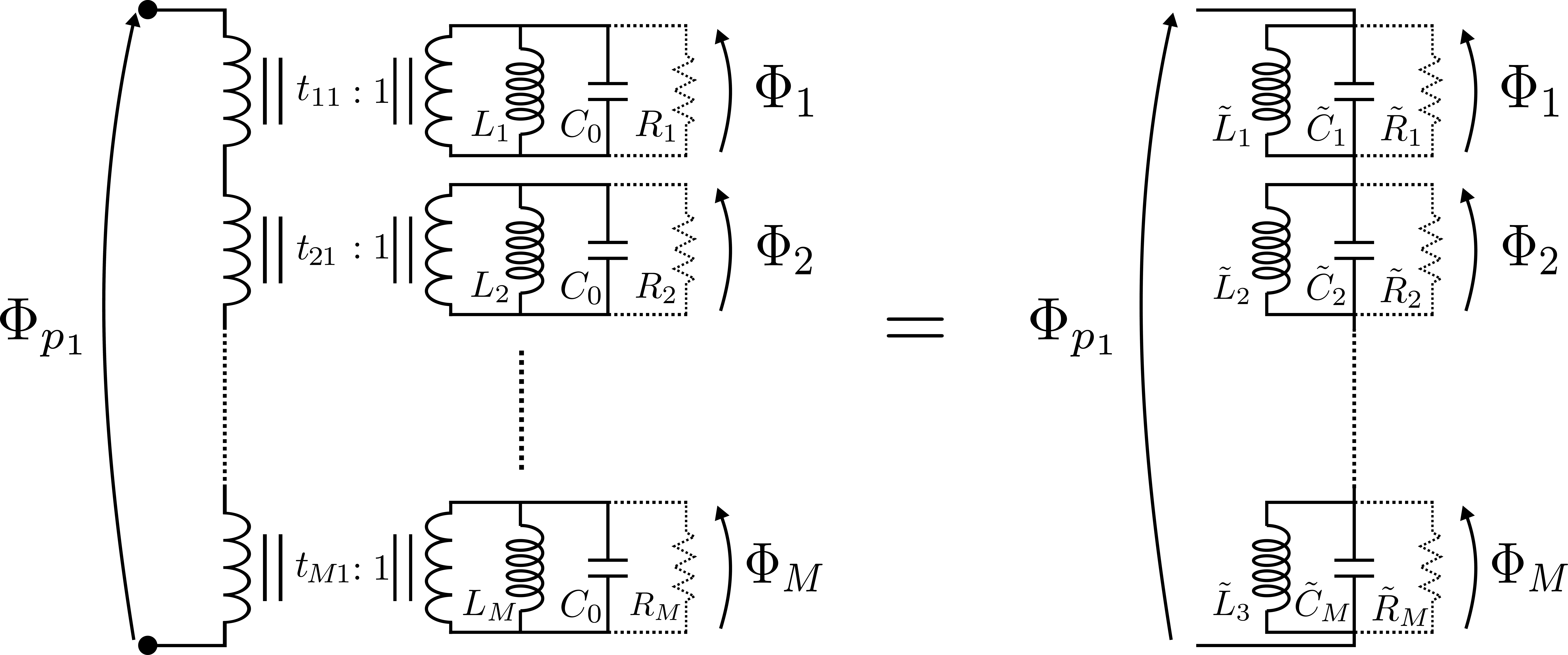}
\caption{Equivalence between the single-port Cauer circuit and the Foster circuit. The parameters in the circuits are related as follows: $\tilde{C}_m = C_0/t_{m1}^2$, $\tilde{L}_m = L_m t_{m1}^2 =\frac{1}{C_m \omega_m^2}$. In the weakly lossy case, the resistances (dotted branches) $R_m$ and $\tilde{R}_m$ are related as $\tilde{R}_m = R_m t_{m1}^2$.}
\label{fig::smlossy_foster}
\end{figure}

\section{Lossy networks}
\label{sec:loss}

When the network is lossy, the entries of its impedance matrix $\mat{\mathcal{Z}}(\omega)$ are no longer purely imaginary. The loss can be included in a lumped-element description of the network by including resistors. Note that resistors may not necessarily model all sources of loss or noise for superconducting devices, as noise can be frequency-dependent, like $1/f$ flux noise (see Chapter~\ref{chap:noise}). An example of loss which can be modeled as a resistor is the dielectric loss inside a capacitor (like the capacitor of the Josephson junction). To model this, we can place a resistor $R$ in parallel with the capacitor. In this case, a large resistance $R$ means, of course, low loss.

When the electric field supported by any of the computationally-relevant modes is present in some of the dielectric on the chip (an oxide, or the silicon substrate), dielectric loss will occur. The amount of loss in a mode is determined by what fraction of the electric field energy associated with the mode is contained in the particular dielectric material (a quantity that we call the participation ratio $r_i$ for the $i$th material \footnote{Not to be confused with the energy-participation ratios of the mode at the junctions discussed in Section~\ref{subsec:epratio}.}), as well as the loss characteristics of the material itself. The latter quantity is determined by the dielectric constant $\epsilon$ of the material, which for lossless materials is real, but has an imaginary component for a lossy material \cite{book:jackson}. The loss tangent of a material is given as 

\begin{equation}
\tan(\delta)=\frac{{\rm Im}(\epsilon)}{{\rm Re}(\epsilon)}=\frac{1}{Q},
\end{equation}
where $Q$ is the intrinsic quality factor of the material due to the dielectric loss. Clearly, $Q$ diverges in case of no loss when ${\rm Im}(\epsilon)=0$. If we have several materials, contributing to the dielectric loss, we have
\begin{equation}
    \tan(\delta_{\rm tot})=\sum_i r_i \tan(\delta_i)\equiv \frac{1}{Q_{\rm tot}};
    \label{eq:tot-loss}
\end{equation}
see e.g. \cite{martinis:decoh, oconnell:di-loss, wang:surface} for more information.

A simple model of a lossy network, that is accurate when the network is not too lossy, is a lossy Cauer network in which a resistor is placed in parallel with each LC oscillator in Fig.~\ref{fig::cauer_circuit}. This model does not represent a completely general linear, passive lossy network as was shown in Refs.~\cite{SAD:bb, solgun2015}. These works prescribe a fully accurate network analysis, based on the work by Brune (as described in, e.g., \cite{guillemin1957synthesis}). Unlike in Eq.~\eqref{eq:z_foster}, for a lossy network the poles of the $\mat{Z}(s)$ (or the zeros of $\mat{Y}(s)$), occur at generally complex $s=i\omega+\sigma$. In what follows, we consider only the simple single-port, lossy Cauer network in detail, but a generalization to multi-port networks is straightforward.

\subsection{Single-port case}
\label{subsec:sportcase}

\subsubsection{RLC circuit}

Consider first the simple example of a single LC~oscillator in parallel with a resistor $R$, the RLC oscillator.  In the case of lossy elements, it is most convenient to work with admittances and impedances in the Laplace domain.

The resonant frequency of the bare LC oscillator is $\omega_r = 1/\sqrt{LC}$ and we define the decay rate of the oscillator as

\begin{equation}
\label{eq:kapparlc}
    \kappa = \frac{1}{RC}.
\end{equation}

Since the admittances add in parallel, we have
\begin{equation}
Y(s)=s C+\frac{1}{s L}+\frac{1}{R}= \frac{C}{s} \left(s^2 + s\kappa + \omega_r^2 \right).
  \label{eq:admRLC} 
\end{equation}
Hence, $Y(s)=0$ implies 
\begin{equation}
s=s_{\pm}=-\frac{\kappa}{2}\pm i\omega_r \sqrt{1-\frac{\kappa^2}{4 \omega_r^2}}.
\label{eq:spoles}
\end{equation}
The resonant frequency $\tilde{\omega}_r =\omega_r \sqrt{1 - \frac{\kappa^2}{4 \omega_r^2}}$ is slightly shifted away from that of the ``bare" LC oscillator. Additionally, the decay rate $\kappa$ defined in Eq.~\eqref{eq:kapparlc} can be identified as $\kappa = -2 \mathrm{Re}(s_{\pm})$. This decay rate (in angular frequency) determines the quality factor $Q=\omega_r/\kappa$ of the mode, i.e., the number of oscillations before damping. Hence, if we connect this back to the definition of the total loss tangent of the mode in Eq.~\eqref{eq:tot-loss}, then we obtain consistency by choosing the resistance $R$ as
\begin{equation}
    R=\frac{1}{C\omega_r \tan(\delta_{\rm tot})}.
\end{equation}

It is also instructive to rewrite the impedance of the parallel RLC circuit in the Laplace domain in terms of its partial fraction expansion and compare it to the general lossless form in Eq.~\eqref{eq:z_foster}. We obtain

\begin{equation}
    Z(s) = \frac{1}{Y(s)} = \frac{1}{C} \left(\frac{s_+}{s_+ - s_-} \frac{1}{s -s_+} - \frac{s_-}{s_+ - s_-} \frac{1}{s -s_-} \right).
\end{equation}
When $\kappa/\omega_r$ is small, we can Taylor expand the poles $s_{\pm}$ in Eq.~\eqref{eq:spoles} to first order to get

\begin{equation}
    s_{\pm} = -\frac{\kappa}{2} \pm i\omega_r + \mathcal{O}\left(\frac{\kappa^2}{\omega_r^2} \right). 
\end{equation}

Thus, in this limit, we can approximate the impedance as

\begin{equation}
\label{eq:zrlcapprox}
    Z(s) \approx \frac{1}{2 C \omega} \left[\frac{\omega + i \kappa/2}{s - i (\omega + i \kappa/2)} + \frac{ \omega - i \kappa/2}{s + i (\omega - i \kappa/2)} \right].
\end{equation}

Now, close to the poles, i.e., for $s \approx \pm i \omega_r$, we can further neglect the $i \kappa/2$ contributions in the numerator of Eq~\eqref{eq:zrlcapprox}, to obtain \footnote{To see this, one can just show that, close to the poles, Eqs.~\eqref{eq:zrlcapprox} and\eqref{eq:zrlcapproxfoster} yield the same first-order Taylor expansion in $\kappa/\omega_r$.} 

\begin{equation}
\label{eq:zrlcapproxfoster}
    Z(s) \approx \frac{1}{2 C } \left[\frac{1}{s - i (\omega_r + i \kappa/2)} + \frac{1}{s + i (\omega_r - i \kappa/2)} \right].
\end{equation}

In principle, Eq.~\eqref{eq:zrlcapproxfoster} is valid only close to the poles. However, far away from the poles the impedance is approximately zero anyway, and thus Eq.~\eqref{eq:zrlcapproxfoster} is, in this sense, a good approximation when the dissipation is small, $\kappa/\omega_r \ll 1$. Eq.~\eqref{eq:zrlcapproxfoster} is the so-called lossy Foster approximation of the impedance of a single RLC oscillator. The circuit construction in Fig.~\ref{fig::cauer_circuit} generalizes this to the lossy multi-mode multi-port case.  

\begin{Exercise}[label=exc:kap]
Give the admittance for an LC oscillator with a resistor $R$ placed {\em in series} with the capacitor, determine $\kappa$ and show that, using the definition for the quality factor $Q=\omega_r/\kappa$ and Eq.~\eqref{eq:tot-loss}, we can identify
\begin{equation}
    R \approx \frac{\tan(\delta_{\rm tot})}{C \omega_r},
\end{equation}
when the resistance $R$ is small.
\end{Exercise}

\begin{Answer}[ref=exc:kap]
We have
\begin{equation}
    Y(s)=\frac{1}{s L}+\frac{1}{R+\frac{1}{s C}},
\end{equation}
for which $Y(s)=0$ at $s_{\pm }=\frac{-R}{2L}\pm \frac{i}{\sqrt{LC}} \sqrt{1-\frac{C R^2}{4L}}$ and hence $\kappa=R/L=\omega_r/Q$. Thus $R=\tan(\delta_{\rm tot}) L \omega_r \approx \frac{\tan(\delta_{\rm tot})}{C \omega_r}$, using the fact that the new resonant frequency ${\rm Im}(s_+)=\frac{1}{\sqrt{LC}}\sqrt{1-\frac{CR^2}{4L}} \approx \omega_r$ when $R$ is small.
\end{Answer}

\subsubsection{Single-port ($N=1$) with $M>1$ lossy network}

The expression in Eq.~\eqref{eq:zrlcapproxfoster} can be shown to extend to a general single- or multi-port network: the poles of $\mat{Z}(s)$ in Eq.~\eqref{eq:z_foster} acquire a real part $\kappa/2$, see Eq.~\eqref{eq:z_lossyfoster} in Appendix~\ref{app:norm_mode}. We discuss the single-port case $N=1$ here. 

 Consider the circuit in Fig.~\ref{fig::smlossy_foster} (right) and let $\omega_m$ be the frequency of the $m$th RLC oscillator and let $\kappa_m$ be its decay rate. For such a network, let $\{s_m\}$ be the $M$ complex zeros of the admittance, i.e.,
\begin{equation}
   \forall s_m, \quad Y(s_m)=0. 
\end{equation} 
Since the network is a series of RLC oscillators, we know that 
\begin{equation}
Y(s)=\frac{1}{\sum_{m=1}^M \frac{1}{Y_m(s)}}=Y_n(s) \frac{1}{1+\sum_{m \neq n} \frac{Y_n(s)}{Y_m(s)}},
\label{eq:y-rlc}
\end{equation}
with $Y_m(s)$ the admittance of the $m$th RLC oscillator as in Eq.~\eqref{eq:admRLC}. Thus a $s=s_n$ such that $Y_n(s_n)=0$ implies that $Y(s_n)=0$ and vice versa: $Y(s)=0$ only if there is some $n$ for which $Y_n(s)=0$. Hence
\begin{equation}
    \omega_m={\rm Im}(s_m), \kappa_m=-2\,{\rm Re}(s_m).
\end{equation}
If we connect the single port to other circuit elements, such as a Josephson junction, we would also need the effective mode capacitance $\tilde{C}_m$ in Eq.~\eqref{eq:eff-c}. From Eq.~\eqref{eq:y-rlc} we can see that  
\begin{equation}
\frac{dY}{ds}\biggl\vert_{s=s_m}=\frac{dY_m}{ds}\biggl\vert_{s=s_m} \approx \frac{dY_m}{ds}\biggl\vert_{s=i/\sqrt{\tilde{L}_m \tilde{C}_m}}=2 \tilde{C}_m, 
\end{equation}
as expected. 

As we mentioned before, $Y(s)$ could have been obtained from numerical electromagnetic simulations and the goal is to fit the data with a series of RLC oscillators. However, only data for pure imaginary $s=iz$, i.e. only $\mathcal{Y}(z)$, may be available while $\kappa_m$ relates to the real part of the complex zero $s_m$. In that case ---as a low-loss approximation--- we can determine the real roots $z_m$ for which ${\rm Im}(\mathcal{Y}(z_m))=0$ and let $\omega_m=z_m$. The decay rate of the mode of frequency $\omega_m$ is, by analogy with a single RLC oscillator, given by $\kappa_m=\frac{1}{\tilde{R}_m \tilde{C}_m}$ where $\tilde{C}_m=\frac{1}{2}{\rm Im}(\mathcal{Y}'(\omega_m))$ and $\tilde{R}_m=1/{\rm Re}(\mathcal{Y}(z_m))$, leading to the expression
\begin{equation}
    \kappa_m \approx \frac{2 {\rm Re}(\mathcal{Y}(\omega_m))}{{\rm Im}(\mathcal{Y}'(\omega_m))} \ge 0.
    \label{eq:approx-kappa}
\end{equation}

In Section~\ref{sec:lindblad} we discuss how to use the decay rates $\kappa_m$ and the Hamiltonian $H$ in Eq.~\eqref{eq:hcm} to model the dynamics of the network by using a Lindblad master equation, see Eq.~\eqref{eq:lindblad}.

\section{Using the QuCAT software package}
\label{sec:qucat}

In this section, you will get acquainted with the QuCAT software package through two exercises. QuCAT stands for Quantum Circuit Analyzer Tool and it is an open source Python library which provides standard analysis tools for superconducting circuits, built around Josephson junctions, developed by Mario Gely while working at Delft University of Technology (Steele Lab). QuCAT features an intuitive graphical or programmatical interface to create circuits, the ability to compute their Hamiltonian, and a set of complimentary functionalities such as calculating dissipation rates or visualizing current flow in the circuit. More information on what can be done using QuCAT can be found at \url{https://qucat.org} and Ref.~\cite{Gely_2020}.

The first exercise is about understanding how to represent a transmon qubit by a different circuit. The second exercise \ref{exc:qucat2} has further analysis of this transmon and its coupling to a lossy resistor. 

\begin{Exercise}[title={Black-box quantization of a modified transmon circuit I},label=exc:qucat]
\begin{figure}[htb]
\centering
\begin{subfigure}[t]{0.4\textwidth}
\centering
\includegraphics[height=0.4\textwidth]{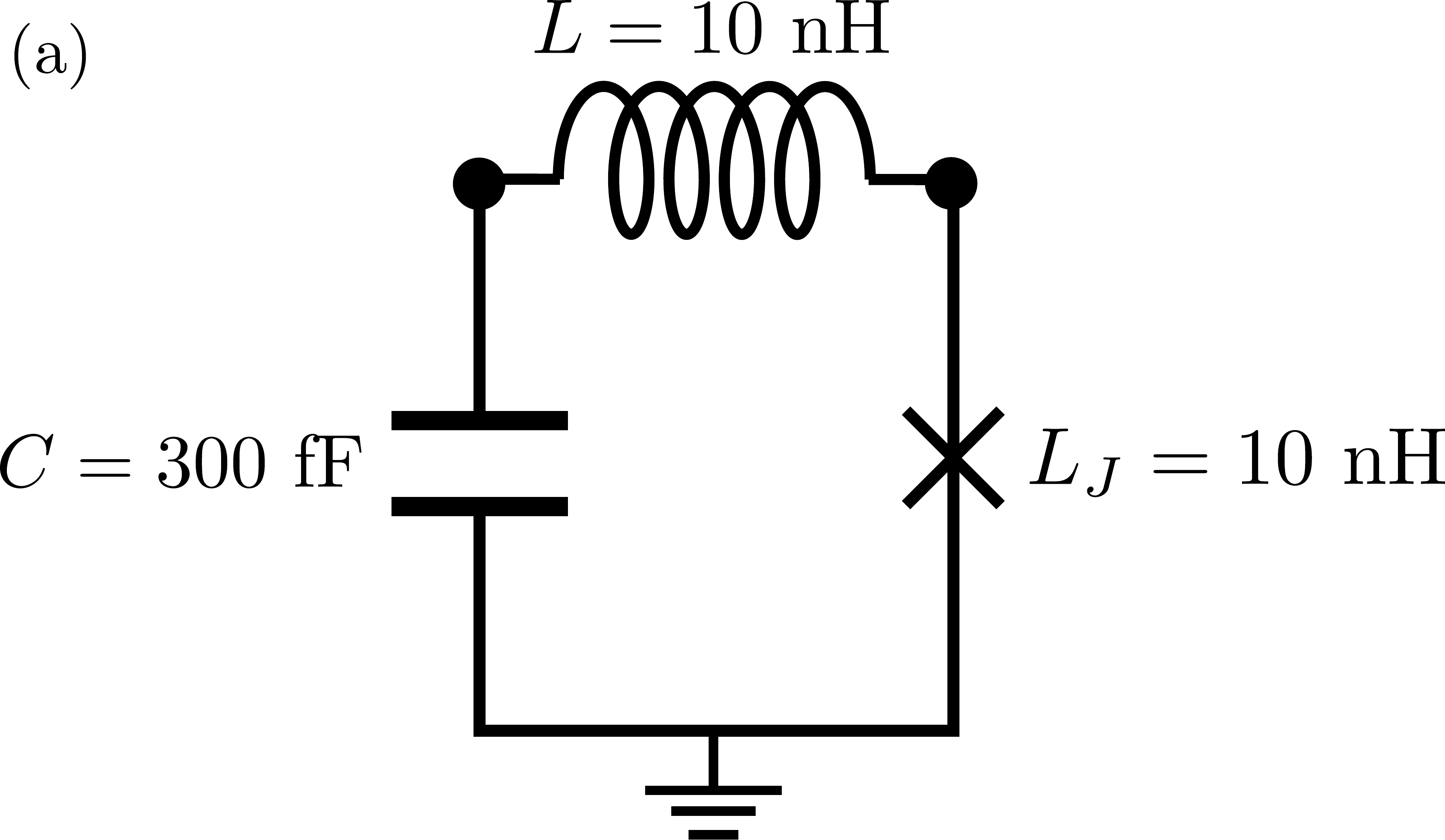}
\end{subfigure}
\begin{subfigure}[t]{0.4\textwidth}
\centering
\includegraphics[height=0.4\textwidth]{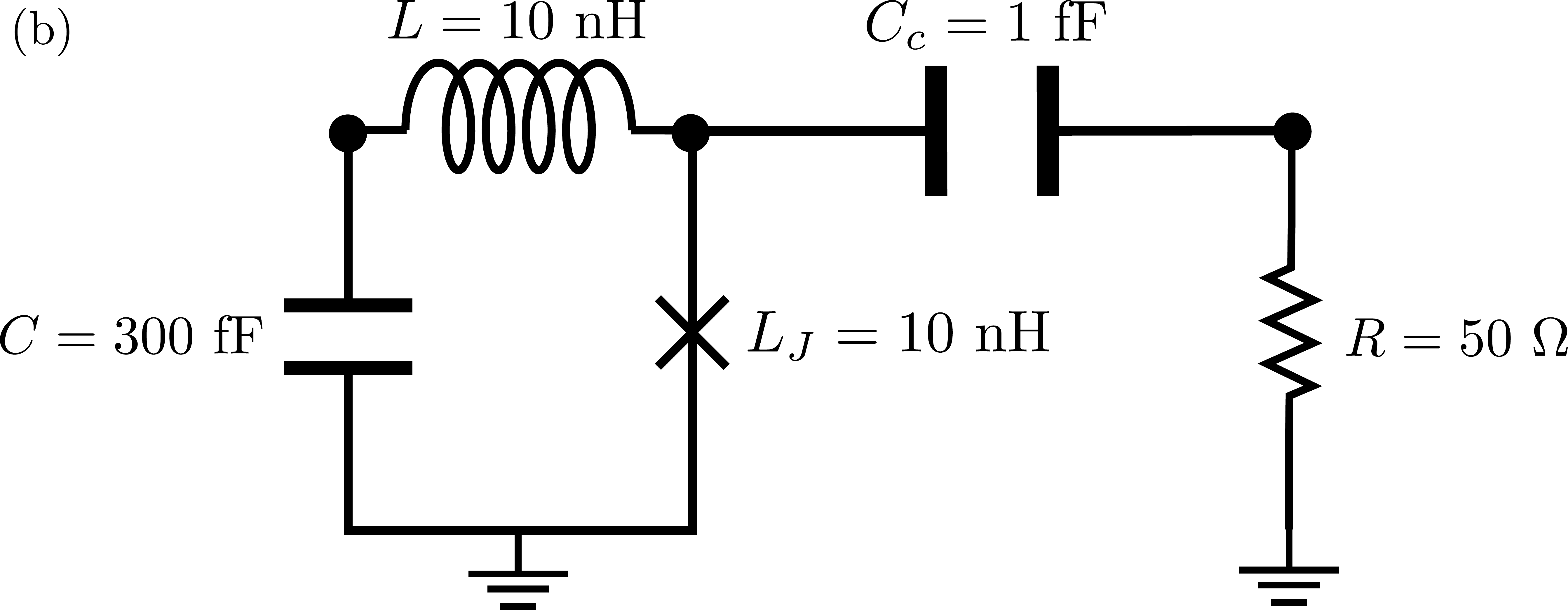}
\end{subfigure}
\caption{Modified transmon circuit in (a) with a coupling to a feedline which functions as a source of loss modeled as a resistor in (b).}
\label{fig:gely-ex}
\end{figure}
\Question Consider the circuit in Fig.~\ref{fig:gely-ex}(a). Temporarily replace the junction of this circuit with an inductor, with inductance equal to the Josephson inductance $L_J$. Write down the admittance $\mathcal{Y}(\omega)$ at the nodes of the inductor of this network.
\Question Determine the resonance frequency $\omega_m$ of this circuit at which $\mathcal{Y}(\omega_m)=0$. {\em Hint: this is easier if you first try to write $\mathcal{Y}$ as a rational function in $\omega$}.
\Question Prove that the effective capacitance of the circuit is 
\begin{align}
   C_m = \text{Im}\left(\mathcal{Y}'(\omega=\omega_m)\right)/2 =  C\frac{(L+L_J)^2}{L_J^2}, 
\end{align} see Eq.~\eqref{eq:eff-c}.
{\em Hint: you can simplify calculations by remembering that $\omega_m$ cancels the numerator of $\mathcal{Y}$ written in rational form.} 
\Question Using the resonance frequency and the effective capacitance, write down the Hamiltonian of the circuit as in Eq.~\eqref{eq:single-p-ham} with its spider contribution. The only operators in the Hamiltonian should be annihilation and creation operators. 
\Question Write the Hamiltonian that would result from expanding the cosine potential to fourth order and subsequently neglecting all non-diagonal terms in the Fock basis. Express the non-harmonic part of the Hamiltonian as a function of the anharmonicity
\begin{equation}
  \tilde{\delta} = -\frac{e^2}{2C}\frac{L_J^3}{(L+L_J)^3}.
\end{equation}
How does the anharmonicity of this circuit differ from the anharmonicity of the regular transmon (without an inductor) as derived in Section~\ref{sec:tr_approx}?
\end{Exercise}

\begin{Answer}[ref=exc:qucat]
\Question 
The capacitor and inductor are connected in series, so their impedances add up. This series combination is in parallel to the junction (which we have temporarily replaced with an inductor $L_J$), and so their admittances add up. As a consequence the total admittance at the nodes of the junction is
\begin{equation}
  \mathcal{Y}(\omega) = \frac{1}{iL_J\omega}+\frac{1}{iL\omega+\frac{1}{iC\omega}}.
\end{equation}
\Question We want to write $\mathcal{Y}$ as a rational function in $\omega$, so we rewrite $\mathcal{Y}(\omega)$ as
\begin{equation}
  \mathcal{Y}(\omega) = \frac{1-(L+L_J)C\omega^2}{iL_J\omega(1-LC\omega^2)}.
  \label{eq:derivative_from}
\end{equation}
We see that the numerator, and thus the admittance, is 0 when 
\begin{equation}
  \omega=\omega_m =\frac{1}{\sqrt{(L+L_J)C}}.
\end{equation}
\Question The derivative can be calculated from Eq.~\eqref{eq:derivative_from}. Applying standard formulas for the derivative of fractions, we get (without simplification)
\begin{equation}
  \mathcal{Y}'(\omega) = \frac{(-2(L+L_J)C\omega)(iL_J\omega(1-LC\omega^2))-\frac{\partial((iL_J\omega(1-LC\omega^2))}{\partial\omega}(1-(L+L_J)C\omega^2)}{(iL_J\omega(1-LC\omega^2))^2}.
\end{equation}
As we evaluate this function at $\omega_m$, the right part of the numerator will be equal to zero, so we do not need to calculate the derivative fully. Further simplification yields
\begin{equation}
  \mathcal{Y}'(\omega_m) = \frac{(-2(L+L_J)C\omega_m)}{(iL_J\omega_m(1-LC\omega_m^2))}.
\end{equation}
Filling in our expression for the resonance frequency and simplifying the expression further, we get
\begin{equation}
  \mathcal{Y}'(\omega_m) = 2iC\frac{(L+L_J)^2}{L_J^2}.
\end{equation}
Taking the imaginary part and dividing by two yields the effective capacitance of the circuit:
\begin{equation}
  C_m = C\frac{(L+L_J)^2}{L_J^2}.
\end{equation}
\Question The Hamiltonian is given by
\begin{equation}
  H = \hbar\omega_m\hat b^\dagger \hat b -E_J\left(\cos(\phi_\text{zpf}(\hat b+\hat b^\dagger))+\frac{(\phi_\text{zpf}(\hat b+\hat b^\dagger))^2}{2}\right),
\end{equation}
where
\begin{equation}
  \phi_\text{zpf} = \frac{\hbar}{2e}\sqrt{\frac{\hbar}{2\omega_m C_m}},
\end{equation}
identical to its definition in Eq.~\eqref{eq:zpf-reduced} (using the definition of $\Phi_0$).
\Question 
The expected form of the Hamiltonian is
\begin{equation}
   H/\hbar = \big(\omega_m+\tilde{\delta}\big)\hat b^\dagger \hat b+\frac{\tilde{\delta}}{2}\hat b^\dagger\hat b^\dagger \hat b \hat b.
\end{equation}
The anharmonicity expression only needs its parameters filled in:
\begin{multline}
  \tilde{\delta} = -\frac{E_J}{2\hbar}\phi_\text{zpf}^4 
  = -\frac{\left(\frac{\hbar}{2e}\right)^2}{2L_J\hbar}\left(\frac{2e}{\hbar}\right)^4\left(\frac{\hbar}{2\frac{1}{\sqrt{(L+L_J)C}} C\frac{(L+L_J)^2}{L_J^2}}\right)^2 
  =-\frac{e^2}{2C \hbar}\frac{L_J^3}{(L+L_J)^3}=-\frac{E_C}{\hbar} \frac{L_J^3}{(L+L_J)^3}.
\end{multline}
Comparing with Eq.~\eqref{eq:anharmont}, we see that $\tilde{\delta}$ has an additional factor $\frac{L_J^3}{(L+L_J)^3}$ as compared to $\delta$. Note also that the resonant frequency is modified.
\end{Answer}

In order to complete the next Exercise, you will have to download and install Python 3 on your computer (we recommend installing \href{https://www.anaconda.com/products/individual}{Anaconda}), as well as two packages, QuTiP and QuCAT. To install the two packages, simply open a command prompt and run:

\inline{pip install qutip}

and then

\inline{pip install qucat}

\begin{Exercise}[title={Black-box quantization of a modified transmon circuit II},label=exc:qucat2]
\Question Using QuCAT, build the circuit in Fig.~\ref{fig:gely-ex}(a) either \href{https://qucat.org/tutorials/basics.html#Building-a-circuit-with-the-GUI}{graphically} or \href{https://qucat.org/tutorials/basics.html#Building-the-circuit-programmatically}{programmatically},  using the indicated \textbf{numerical} circuit parameters. Verify your analytical calculations of the resonance frequencies and the anharmonicities using the functions \href{https://qucat.org/tutorials/basics.html#Eigen-frequencies}{\inline{eigenfrequencies}} and  \href{https://qucat.org/tutorials/basics.html#Anharmonicity}{\inline{anharmonicities}}.
\textit{Note 1: QuCAT gives all quantities in frequency units (not angular frequencies). This means that the function \inline{eigenfrequencies} will give $\omega_m/2\pi$, while the function \inline{anharmonicities} will give $\delta/2\pi$. Note that QuCat uses the letter $A_m$ for the anharmonicity of the mode at angular frequency $\omega_m$ instead of the symbol $\delta$.} 
\textit{Note 2: There have been reports of the QuCAT graphical user interface malfunctioning in recent versions of MacOS. The exercise can still be completed using a purely programmatic approach.} 
\Question In practice we are only interested in the first three levels of such a qubit. Give the analytical and numerical expression for the transition frequencies between these levels, using the diagonal Hamiltonian previously obtained analytically. A transition \textit{frequency} is the difference in eigen\textit{energies} of two levels, divided by Planck's constant.
\Question Using the \href{https://qucat.org/tutorials/basics.html#Hamiltonian,-and-further-analysis-with-QuTiP}{\inline{hamiltonian}} function of QuCAT (see also \href{https://qucat.org/API/circuit_functions/hamiltonian.html#qucat.Qcircuit.hamiltonian}{here}), calculate these transition frequencies in the case where we would not have discarded non-diagonal terms. This corresponds to the Hamiltonian returned by the \inline{hamiltonian} function with the Taylor expansion order specified to be 4 (\inline{order=4}). Include ten excitations for these calculations (\inline{excitations = 10}), you do not need to specify which modes to use in this exercise. Then, calculate the transition frequencies with the cosine Taylor expanded to order 6 (\inline{order=6}). Compare the different values you obtained with the perturbative, analytical calculation of the previous question by writing sufficient decimal points in order to see a difference between them. This should allow you to see the validity of the approximations made in our treatment of the cosine term. \textit{Note: QuCAT returns the Hamiltonian in units where h=1, such that eigenenergies subsequently computed with QuTiP will be in units of hertz.}
\Question Now, consider the case where this circuit would be capacitively connected to a lossy 50$\Omega$ feedline, represented by a resistor in the circuit of Fig.~\ref{fig:gely-ex}(b). Without carrying out any calculations, briefly explain how you would proceed to analytically determine the resulting loss rate.
\Question Construct the circuit of Fig.~\ref{fig:gely-ex}(b) in QuCAT and provide a numerical value for the \href{https://qucat.org/tutorials/basics.html#Loss-rates}{loss rate}. 
\end{Exercise}

\begin{Answer}[ref=exc:qucat2]
\Question 
\begin{lstlisting}
import qucat as qc 
import numpy as np
from scipy.constants import h,e

C = 300e-15
L = 10e-9
Lj= 10e-9

# Analytical results first 
print("analytical results")
print("w_m = %.2e Hz"%(1/np.sqrt((L+Lj)*C)/2/np.pi))
print("A = %.2e Hz"%(e**2/2/C*Lj**3/(L+Lj)**3/h))

# Now simulating the problem in QuCAT
cir = qc.Network([
  qc.C(0,1,C),
  qc.L(1,2,L),
  qc.J(2,0,Lj)
  ])
print("QuCAT results")
print("w_m = %.2e Hz"%(cir.eigenfrequencies()[0]))
print("A = %.2e Hz"%(cir.anharmonicities()[0]))
\end{lstlisting}

This gives the following results:

\begin{lstlisting}
analytical results
w_m = 2.05e+09 Hz
A = 8.07e+06 Hz
QuCAT results
w_m = 2.05e+09 Hz
A = 8.07e+06 Hz
[Finished in 5.3s]
\end{lstlisting}
\Question 
We are interested in the first three levels $|0\rangle$,$|1\rangle$,$|2\rangle$ of the system. The eigenenergies are going to be the expectation values of the Hamiltonian for these three states, namely $0$, $\hbar\omega_m-A$,and $2\hbar\omega_m-3A$. The transition frequencies are then (in frequency units) $\omega_m/2\pi-A/h$  and $\omega_m/2\pi-2A/h$.
Numerical values are \inline{2.04661e+09 Hz, 2.03854e+09 Hz}
\Question 
The following code addresses the questions:
\begin{lstlisting}
import qucat as qc 
import numpy as np
from scipy.constants import h,e

C = 300e-15
L = 10e-9
Lj= 10e-9

fm = 1/np.sqrt((L+Lj)*C)/2/np.pi
A = e**2/2/C*Lj**3/(L+Lj)**3/h

# Analytical results first 
print("analytical results")
print("%.5e Hz, %.5e"%(fm-A,fm-2*A))

# Now simulating the problem in QuCAT
cir = qc.Network([
  qc.C(0,1,C),
  qc.L(1,2,L),
  qc.J(2,0,Lj)
  ])
H = cir.hamiltonian(
    taylor = 4,
    excitations = 10)
ee = H.eigenenergies()
print("4th order diagonalization")
print("%.5e Hz, %.5e"%(ee[1]-ee[0],ee[2]-ee[1]))

H = cir.hamiltonian(
    taylor = 6,
    excitations = 10)
ee=np.real(H.eigenenergies())
print("6th order diagonalization")
print("%.5e Hz, %.5e"%(ee[1]-ee[0],ee[2]-ee[1]))
\end{lstlisting}

This yields the result
\begin{lstlisting}
analytical results
2.04661e+09 Hz, 2.03854e+09
4th order diagonalization
2.04655e+09 Hz, 2.03834e+09
6th order diagonalization
2.04661e+09 Hz, 2.03853e+09
[Finished in 9.9s]
\end{lstlisting}
\Question 
We would temporarily replace the junction of the circuit with an inductor (with an inductance equal to $L_J$). We would then write the admittance $\mathcal{Y}(\omega)$ at the nodes of this inductor $L_J$. We would then determine the complex number $\zeta_m$ which satisfies $\mathcal{Y}(\zeta_m)=0$. The loss rate is given by $\kappa_m = 2\,\text{Im}[\zeta_m]$.
\Question 

\begin{lstlisting}
import qucat as qc 

C = 300e-15
L = 10e-9
Lj= 10e-9
Cc = 1e-15
R = 50

cir = qc.Network([
  qc.L(1,2,L),
  qc.J(0,2,Lj),
  qc.C(1,0,C),
  qc.C(2,3,Cc),
  qc.R(3,0,R)
  ])
print("Loss rate = %.2e Hz"%(cir.loss_rates()[0]))
\end{lstlisting}

This gives:

\begin{lstlisting}
Loss rate = 1.11e+03 Hz
[Finished in 4.4s]
\end{lstlisting}
\end{Answer}

\section{Networks of transmon qubits in the dispersive regime}
\label{sec:network-tr}
In Section~\ref{sec:bb}, we have seen that the Hamiltonian of any microwave network involving $N$ Josephson junctions can be written in the normal-mode basis as in Eq.~\eqref{eq:h_nm}.  We now consider a circuit of CPBs capacitively coupled to a linear network \footnote{So here we exclude the case in which there is an inductive shunt, since otherwise we would not have a simple CPB, but a flux qubit or a fluxonium.}. We further assume that all degrees of freedom in the circuit are transmon-like and thus it is a good approximation to Taylor expand the cosine of the Josephson potential up to fourth~order, similar to the approach in Section~\ref{sec:tr_approx}. For the case of a single, isolated transmon we know that this expansion was justified when $E_J/E_C \gg 1$. However, in our case, in which transmons are connected to general networks, it is not immediately clear which capacitive energy scale we should consider for each port. This capacitance should be independent of the particular lumped-element representation of the circuit and thus cannot simply be the capacitance directly shunting the Josephson junction. In fact, the shunting capacitance that we should associate with port $n$ is the total so-called equivalent Th\'evenin capacitance seen by this port. This is defined as the equivalent capacitance that the port sees when all the inductances are removed, i.e., substituted with open circuits. In order to understand how to obtain the Th\'evenin capacitances, we can consider the general Cauer circuit in Fig.~\ref{fig::cauer_circuit} without the inductances. Assuming that the load at each port is also purely inductive, we should also assume that all the ports are open-circuited, apart from the port we are insterested in. In this case, the port $n$ effectively sees $M$ capacitances in series with capacitance $t_{m n}^2 C_0$ with $m=1, \dots, M$. Thus, the Th\'evenin capacitance seen by the $n$th port is given by the series combination of all these capacitances

\begin{equation}
    C_{Tn} = \biggl(\sum_{m=1}^M \frac{1}{t_{m n}^{2} C_0} \biggr)^{-1}, \quad n=1, \dots, N. 
\end{equation}
   
The Th\'evenin equivalent charging energy is then
\begin{equation}
E_{Cn} = \frac{e^2}{2 C_{Tn}}.
\end{equation}
All degrees of freedom are in the transmon regime if $E_{Jn}/E_{Cn} \gg 1$. \par 

By Taylor expanding Eq.~\eqref{eq:h_nm} up to fourth order, we obtain
\begin{equation}
H \approx  \sum_{m =1}^M \hbar \omega_m \hat{a}_m^{\dagger} \hat{a}_m 
 - \sum_{n=1}^N \frac{E_{Jn}}{24} \biggl(\frac{2 \pi}{\Phi_0} \biggr)^4 \biggl[ \sum_{m =1}^M t_{mn} \sqrt{\frac{\hbar Z_m}{2}} \bigl(\hat{a}_m + \hat{a}_m^{\dagger} \bigr) \biggr]^4.
\end{equation}
Now, in close analogy to what was done in Section~\ref{sec:tr_approx} and Exercise \ref{exc:qucat2}, we may opt to neglect all the contributions in the nonlinear term that do not preserve the number of excitations in each mode $m$. Assuming that the modes are not degenerate, this can be seen again as an instance of first-order perturbation theory. As an example, we will neglect terms proportional to $\hat{a}_m^{\dagger} \hat{a}_{m'} \hat{a}_{m'}^{\dagger} \hat{a}_{m'}$. Note that this is allowed only if we assume that the modes $m$ and $m'$ are sufficiently detuned \footnote{A less aggressive pruning of terms could be to only keep terms which preserve the total number of excitations in all or a subset of modes; for example, keep a term like $\hat{a}_m^{\dagger}\hat{a}_k^{\dagger}\hat{a}_l \hat{a}_n$, which is more proper when some frequencies are close together or matching.}. We will only keep resonant (energy-preserving) terms like $ \hat{a}_m^{\dagger} \hat{a}_{m} \hat{a}_{m'}^{\dagger} \hat{a}_{m'}$. There are only two kinds of terms which are kept
\begin{subequations}
\begin{equation}
\bigl(\hat{a}_m + \hat{a}_m^{\dagger} \bigr)^4 \overset{\rm RWA}{\approx}12 \hat{a}_m^{\dagger} \hat{a}_m + 6  \hat{a}_m^{\dagger} \hat{a}_m^{\dagger} \hat{a}_m \hat{a}_m, 
\end{equation}
\begin{equation}
\bigl(\hat{a}_m + \hat{a}_m^{\dagger} \bigr)^2 \bigl(\hat{a}_{m'} + \hat{a}_{m'}^{\dagger} \bigr)^2 \overset{\rm RWA}{\approx} 2 \hat{a}_m^{\dagger} \hat{a}_m + 2 \hat{a}_{m'}^{\dagger} \hat{a}_{m'} + 4 \hat{a}_m^{\dagger} \hat{a}_m \hat{a}_{m'}^{\dagger} \hat{a}_{m'}.
\end{equation}
\end{subequations}
Here, on the right-hand side, we have neglected constant terms and made use of the bosonic commutation relation $[\hat{a}_m, \hat{a}_{m'}^{\dagger}] = \delta_{m m'}\mathds{1}$ to reorder some terms. We can thus write 
\begin{equation}\label{eq:h_tr_circuit}
\frac{H}{\hbar} \approx \sum_{m =1}^M \Omega_m \hat{a}_m^{\dagger} \hat{a}_m + \frac{\delta_m}{2} \hat{a}_m^{\dagger} \hat{a}_m^{\dagger} \hat{a}_m \hat{a}_m + \sum_{\mathrm{all \, pairs}, \, m \neq m'}  \chi_{m m'}   \hat{a}_m^{\dagger} \hat{a}_m\hat{a}_{m'}^{\dagger} \hat{a}_{m'},
\end{equation}
where we have defined the effective mode frequencies
\begin{equation}
\Omega_m = \omega_m + \delta_m + \frac{1}{2} \sum_{m' =1, m' \neq m}^M \chi_{m m'},
\end{equation}
the anharmonicities 
\begin{equation}
 \delta_m = 
- \biggl(\frac{2 \pi}{\Phi_0} \biggr)^4 \sum_{n=1}^N \hbar \frac{E_{Jn}}{8} t_{nm}^4 Z_m^2, 
\end{equation}
and the cross-Kerr coefficient
\begin{equation}\label{eq::chi_4rwa}
\chi_{m m'} =
- \biggl(\frac{2 \pi}{\Phi_0} \biggr)^4 \hbar \sum_{n=1}^N \frac{E_{Jn}}{4} t_{nm}^2 t_{nm'}^2 Z_m Z_{m'}.
\end{equation}

It is worth making some comments on the Hamiltonian in Eq.~\eqref{eq:h_tr_circuit}. First of all, at first sight, it might seem unclear which of the modes should be considered a transmon and which a resonator (when the original coupled network contained resonators), since all modes are treated on the same footing. In this picture all modes inherit the anharmonicity due to the mutual coupling, but we can usually identify transmon-like modes to be those with the highest anharmonicity, for which the associated vector of turn ratios $\bm{t}_m$ has a large element corresponding to one of the ports. The other modes, with much weaker anharmonicity, can be considered resonator-like modes. 

\par A common alternative approach to describe coupled transmon-resonator systems is to introduce `local' transmon or resonator modes with annihilation and creation operators $\hat{b}_m, \hat{b}_{m}^{\dagger}$ and then add their capacitive couplings. In this case these local modes are coupled by a linear term like $-\hbar g(\hat{b}_m - \hat{b}_m^{\dagger})(\hat{b}_{m'} -\hat{b}_{m'}^{\dagger})$, with $g$ the coupling strength, as in Eq.~\eqref{eq:tr-osc-first}. These terms are not present in the Eq.~\eqref{eq:h_tr_circuit}, since we directly diagonalize them by obtaining the normal modes with annihilation and creation operators $\hat{a}_m, \hat{a}_{m}^{\dagger}$. 
It is interesting to note that when the normal modes are sufficiently detuned and we can use Eq.~\eqref{eq:h_tr_circuit} to describe the dynamics of the circuit, we are using a diagonal Hamiltonian, which is very easy to use. Also, if we consider a transmon-like and a resonator-like mode with annihilation operators $\hat{a}_t$, $\hat{a}_r$, respectively, the cross-Kerr term $\chi_{t r} \hat{a}_t^{\dagger} \hat{a}_t \hat{a}_r^{\dagger} \hat{a}_r$ can be understood as a dispersive interaction that can be used to measure the state of the transmon mode. Projecting the cross-Kerr interaction onto the first two levels of the transmon we get a term $ \frac{\chi_{t r }}{2} \sigma_z \hat{a}_r^{\dagger} \hat{a}_r$, which can be interpreted as a transmon-state-dependent shift by $\chi = \chi_{tr}/2$ of the resonator frequency, as in Eq.~\eqref{eq:dispersH} \footnote{Be mindful of the factor of $2$ between the definition of the cross-Kerr coefficients in Eqs.~\eqref{eq:h_tr_circuit} and \eqref{eq::chi_4rwa} and the dispersive shift in Eq.~\eqref{eq:dispersH}.}. By probing the (readout) resonator, we are able to identify the state of the normal, transmon mode, as discussed in Exercise \ref{exc:disp-meas}. 

If instead one wants to use `circuit-localized' modes to study the case in which each mode is far away in frequency from the other, then the standard approach would be to use a Schrieffer-Wolff (SW) transformation discussed in Section~\ref{sec:SW}. That analysis again leads to a diagonal dispersive Hamiltonian~\cite{Blais_2021} and the SW transformation \emph{delocalizes or dresses} the modes. We can thus see that the normal-mode approach and the local mode plus SW approach lead to similar expressions. The normal-mode approach, however, is efficiently generalized to circuits with many degrees of freedom, as we have seen, while the use of SW transformations with many subsystems is generally rather cumbersome. \par 
As a final comment, consider the cross-Kerr interaction between two transmon-like modes $1$ and $2$, i.e. $\chi_{12} \hat{a}_{1}^{\dagger}\hat{a}_{1} \hat{a}_{2}^{\dagger}\hat{a}_{2}$. If we project both modes onto their qubit subspaces, replacing $\hat{a}_1^{\dagger}\hat{a}_1$ by $\frac{1}{2}(\mathds{I}+Z_1)$ etc., we obtain a term $\frac{1}{4}\chi_{12} Z_1 Z_2$. This is the term responsible for the so-called $ZZ$ crosstalk (see Ref.~\cite{andersen2020} for instance). Hence, the crosstalk effect is naturally explained in the normal-mode approach.
Two-qubit gates between transmon gates, such as the CZ gate via external flux-tuning~\cite{Martinis_2014} to the avoided $\ket{11} \leftrightarrow \ket{02}$ crossing, can in principle also be analysed using the normal-mode approach ~\cite{msthesis-olivia}. In this picture, the entangling power and the on/off ratio of the CZ gate must be achieved by a strengthening of the entangling cross-Kerr interaction on the normal modes, whose value also depends on the external flux.

\chapter{Nonreciprocity}
\label{chap:nonrec}

The kit of circuit elements that we have used in these notes is incomplete in a significant way. All linear circuit elements that we have seen so far are {\em reciprocal}. We have formally defined reciprocity in Eq.~\eqref{eq:def-rec} and will discuss it in more intuitive detail in this chapter. In optical physics, reciprocity is captured by the saying, “We see the eyes that see us" (see Popular Summary of \cite{PhysRevX.3.031001}). This has a corresponding meaning for the impedance of multi-terminal circuits.  But the truth of the adage in wave physics is dependent on the medium through which the rays pass.  For certain media (e.g., those exhibiting the Faraday effect) we do {\em not} see the eyes that see us, and then the system is nonreciprocal.  

Reciprocity is generally not taken to be an applicable concept for nonlinear systems. A very good standard treatment of this may be found in Ref.~\cite{desoer1969basic}, Chapter $16.4$, especially pp. $694$-$696$. The treatment of Newcomb~\cite{newcomb}, Definition $2$-$7$, p. $29$, is unique in that it gives a more general treatment in which the definition of reciprocity can also be applied to nonlinear systems. When using the definition given by him, almost any nonlinear system, or networks containing sources, is non-reciprocal.  In particular, the single Josephson junction is non-reciprocal. One cannot rule out that a nonlinear system can be cleverly synthesized to mimic a linear system, and then exhibit nonreciprocity.  Videoconferencing on the internet is certainly an example of a very nonlinear optical transmission system. You can see that it does not mimic a reciprocal system by considering that the adage above is false: by looking at someone's eye on the screen, you are not making eye contact. 

Confining our attention to linear electrical circuits, as we do in the black-box quantization method in Chapter~\ref{chap:ln}, we will see that we have really missed something up until now. But to state more mathematically what we mean by the reciprocity (or not) of a linear electrical circuit, we must first digress and introduce a few further observations about multi-terminal electrical networks.

\section{Ports, terminals, and an important multi-terminal device}
\label{sec:mtn}

We have already discussed multi-port networks in  Section~\ref{sec:mport}. A port is a chosen pair of nodes to be connected to the outside world. Sometimes two different ports can share a single node, which is often the ground node.  Usually not all network nodes participate in ports. The subset of nodes which can participate in ports is referred to as {\em terminals}, or in some literature as {\em poles} \footnote{In some languages, a port is referred to as a two-pole (e.g., in Dutch {\em tweepool}).}. 

Terminals are nodes that are deemed suitable for communication with the rest of the world. A node will attain this status in two ways: 1) It is to be connected to a source, either in the form of a standard lumped current or voltage source, or in the form of a transmission-line cable. 2) Since we will focus here on the linear parts of networks, we can also consider a node to which a Josephson junction is attached to be a terminal, as we have done in Section~\ref{sec:bb}.

Another feature of a terminal pair comprising a port is that the instantaneous current entering one of these terminals from the outside world is equal to the current leaving the other terminal. Also, the response of our network should not depend on the voltage difference (or flux difference) between two terminals that are {\em not} in the same port. If there is such a dependence, then this pair should define a new port. Most of these constraints are obvious in actual examples, like the ones we will be considering shortly. 

 If we have singled out $N$ nodes in our linear network to be terminals, we can form an $N\times N$ response matrix in the following way. Consider a separate ground node not in the network.  Attach a voltage source, described by $V_l(s=i \omega)$, between terminal $l$ and the ground node, so that each node effectively represents a port as in Fig.~\ref{fig:el_network}.
 Recall from Section~\ref{subsec:LF} that the argument ``$s$'' indicates that we are in the Laplace domain, and that by setting $s=i \omega$ we are restricting to the Fourier domain. Connect a short-circuit wire from each of the other terminals to ground. Measure the currents $I_k(s=i\omega)$ in the Fourier domain in each of the wires (the particular current $I_k$ is just the Fourier transform of the current flowing through the wire at terminal $k$. For terminal $k=l$ it is the Fourier transform of the current through the voltage source). Due to the linearity of the system, we will get a well-defined number (complex and $\omega$-dependent), by taking the ratio:
\begin{equation}
    \frac{I_k(i \omega)}{V_l(i \omega)}\equiv Y_{kl}(i \omega).
\end{equation}
Thus, $Y_{kl}(i\omega)$ are the entries of the admittance matrix as in Section~\ref{sec:mport}. 
For ports we usually say that the impedance matrix is the inverse of admittance, see Eq.~\eqref{eq:inverseY}, but in fact this {\em terminal} admittance matrix never has an inverse (can you see why?). 

Using the terminal admittance matrix, we can straightaway state the condition for circuit reciprocity, generally in the Laplace domain:
\begin{equation}
\mbox{circuit reciprocity}\Longleftrightarrow Y_{ij}(s)=Y_{ij}(s).\label{eq:recipdef}
\end{equation}
In other words, in a reciprocal network, $Y$ is a complex {\em symmetric} matrix at all frequencies. It can be confirmed~\cite{desoer1969basic} that any network made of resistive, inductive and capacitive branches, with mutual inductances, and with transformers, is reciprocal.  Note that lossiness is perfectly compatible with reciprocity.

\section{The circulator and the gyrator}
\label{sec:gyr}

\subsection{Nonreciprocity in action}

\begin{figure}
    \centering
    \begin{subfigure}[t]{0.9 \textwidth}
\centering
\includegraphics[scale=0.40]{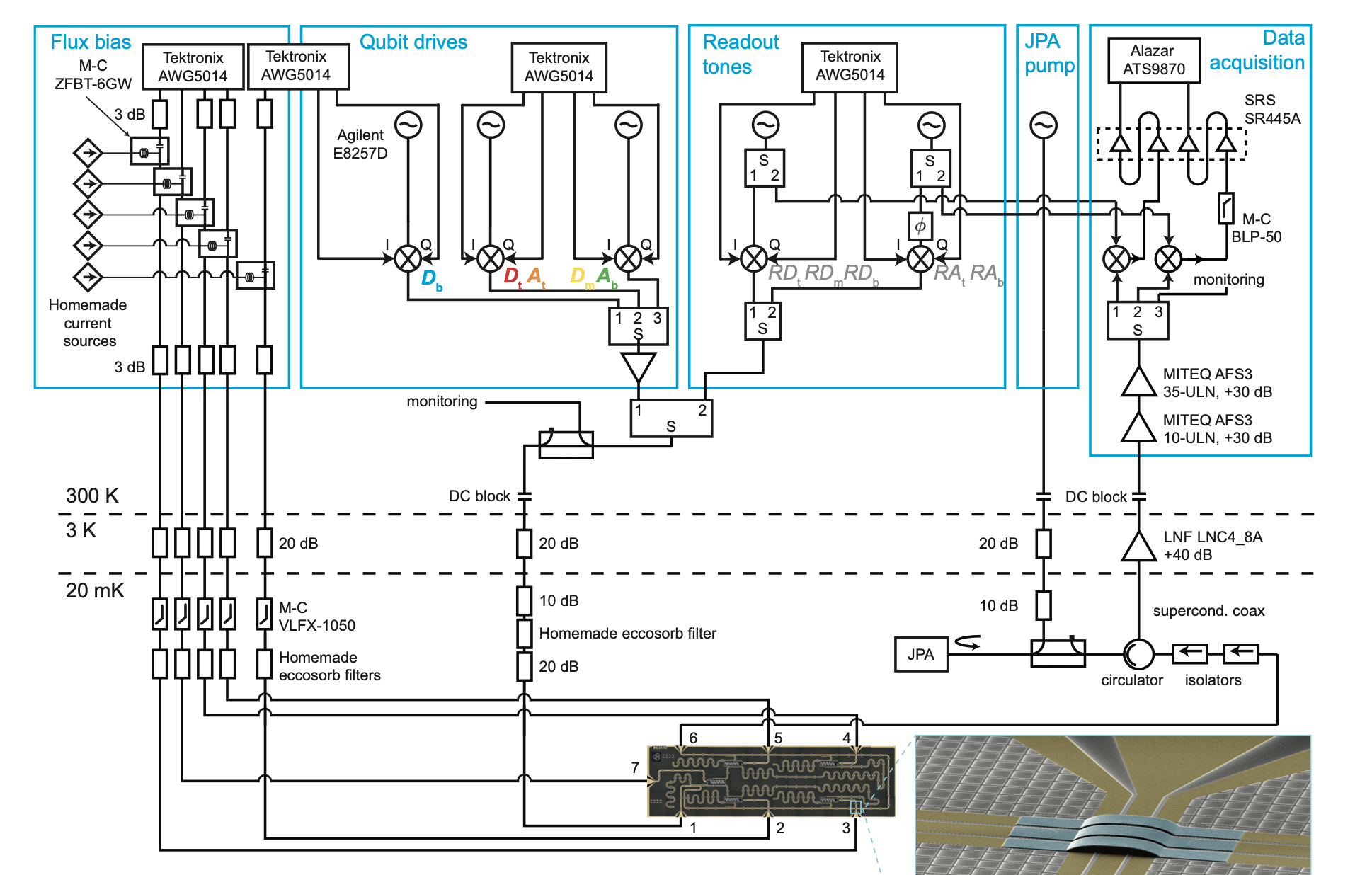}
\subcaption{Schematic diagram taken from the 5-transmon qubit experiment in Ref.~\cite{Riste.etal.2015:BitflipChip} showing a circulator and isolators besides many other electronic components. A circulator where one port is terminated by a resistor is called an isolator; it ensures that incoming noise/radiation into the readout line for outgoing signals is directed into this resistor and hence dissipates.}
    \label{fig:chip-with-isol}
    \end{subfigure}
    \begin{subfigure}[t]{0.9 \textwidth}
\centering
    \includegraphics[scale=0.25]{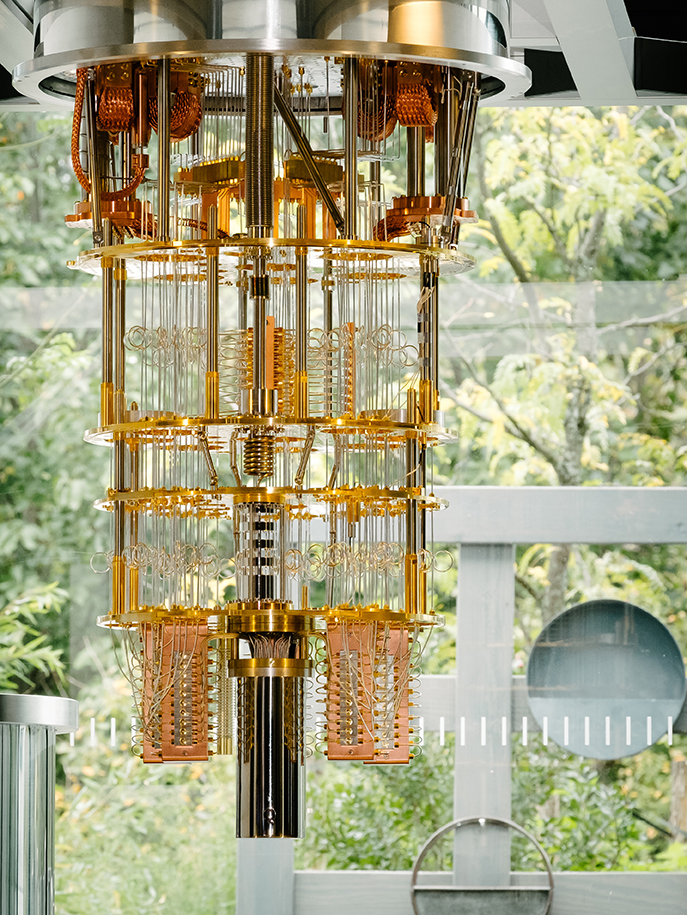}
    \subcaption{IBM cryostat hosting a 50-qubit chip (50Q system) from 2017 with stacks of circulators/isolators at the bottom left and right. Picture taken from  \href{https://admin02.prod.blogs.cis.ibm.net/blogs/think/2017/11/44283/}{IBM Quantum Blog ``The Future is Quantum" (Nov. 2017)}.}
    \label{fig:IBM-50}
\end{subfigure}
\caption{Where non-reciprocal elements are needed.}
\end{figure}

Virtually every solid-state qubit experiment at a quantum technology lab contains several instances of a microwave part called a {\em circulator}, see Fig.~\ref{fig:chip-with-isol}. Such circulators are used on each readout line and each readout line itself is used for the readout, in multiplexed-mode, of several transmon qubits. For multi-qubit chips this can amount to having tens or more circulators, see Fig.~\ref{fig:IBM-50}. 

Each of the circulators in use in these circuits is basically a rectangular box with a typical edge dimension of a few centimeters. While very large compared with inductors and capacitors, it nevertheless functions as, very nearly, a three-port, or six-terminal lossless device in some frequency band. The physical principle of these circulators is the propagation of microwave radiation in a cylindrical gallery containing a magnetized (ferrite) material. The field equations showing the nonreciprocal effect are thoroughly worked out in \cite{pozar}.

Several other approaches have been discussed for the physical implementation of the circulator. A clear motivation for this is that the present-day circulator is inconveniently large; in a quantum computer with a million qubits, one can estimate that the circulators will take up a volume, in the coldest part of the cryostat, approaching a cubic meter -- a very expensive architectural feature. 

One scheme for an alternative circulator is, on the face of it, very appealing. It involves a simple circuit, a ring of three Josephson junctions, with a non-zero external flux through the ring \cite{ 
PhysRevA.82.043811, navarathna}. The three nodes define the three ports of the device, with the other terminal of the ports in each case being ground. While experiments have now been done on this concept, it remains unclear whether this can satisfy the requirements of being a circulator; it will be clearly nonlinear, but there has been no convincing analysis or measurement that indicates that it will successfully mimic the desired linear device. The initial calculations show that when handling a single photon, it will route it in a way similar to a circulator. But this is far from demonstrating full functionality.

There are other proposed approaches using Josephson junctions that have been more fully elaborated, but still not demonstrated. They are parametrically-driven devices, closely related to amplifiers (amplifiers are only treated in the exercises in Chapter~\ref{chap:add} in these notes). They are designed so as to explicitly exhibit linear, nonreciprocal behavior. An unfortunate feature is that they would circulate only in a narrow band of frequencies; in the textbook circulator to be discussed in the next section, the device should work over an ideally infinite bandwidth.

Finally, there are a few concepts for getting real devices that are genuinely passive. Perhaps the most intriguing but least analyzed is a kind of transformer concept, with a highly nonreciprocal material response \cite{TellegenPatent2}, introduced by Tellegen, the inventor of the gyrator (see below). This is perhaps the only device concept that has a chance, in principle, of going to the ideal limit of infinite bandwidth, with the same issues similar to those for the transformer itself. The idea of using special materials also comes up in the concept of using a capacitive coupling to a quantum Hall, or quantum anomalous Hall, material to achieve a wide-bandwidth passive gyrator (but with a very large characteristic impedance $Z_0$); see Ref.~\cite{PhysRevX.4.021019} for a survey of this approach.
    
\subsection{Formalism of the circulator and the gyrator}

We will discuss the circulator by introducing first a fundamentally new three-terminal lossless device, with the ($s$-independent) terminal admittance matrix
\begin{equation}
    \mat{Y}_{c}(s)=G_c \left(\begin{array}{rrr}0&1&-1\\-1&0&1\\1&-1&0\end{array}\right),
    \label{circ}
\end{equation}
where $G_c$ is a constant which has units of inverse impedance.
For the purpose of these Lecture Notes, we use a special symbol for this three-terminal device shown in Fig.~\ref{fig:3t_circ}. We highlight the fact that the admittance matrix $\mat{Y}_{c}(s)$ is not invertible (the sum of all rows is zero), and so it does not admit an impedance representation. 

\begin{figure}
    \centering
    \includegraphics[scale=0.15]{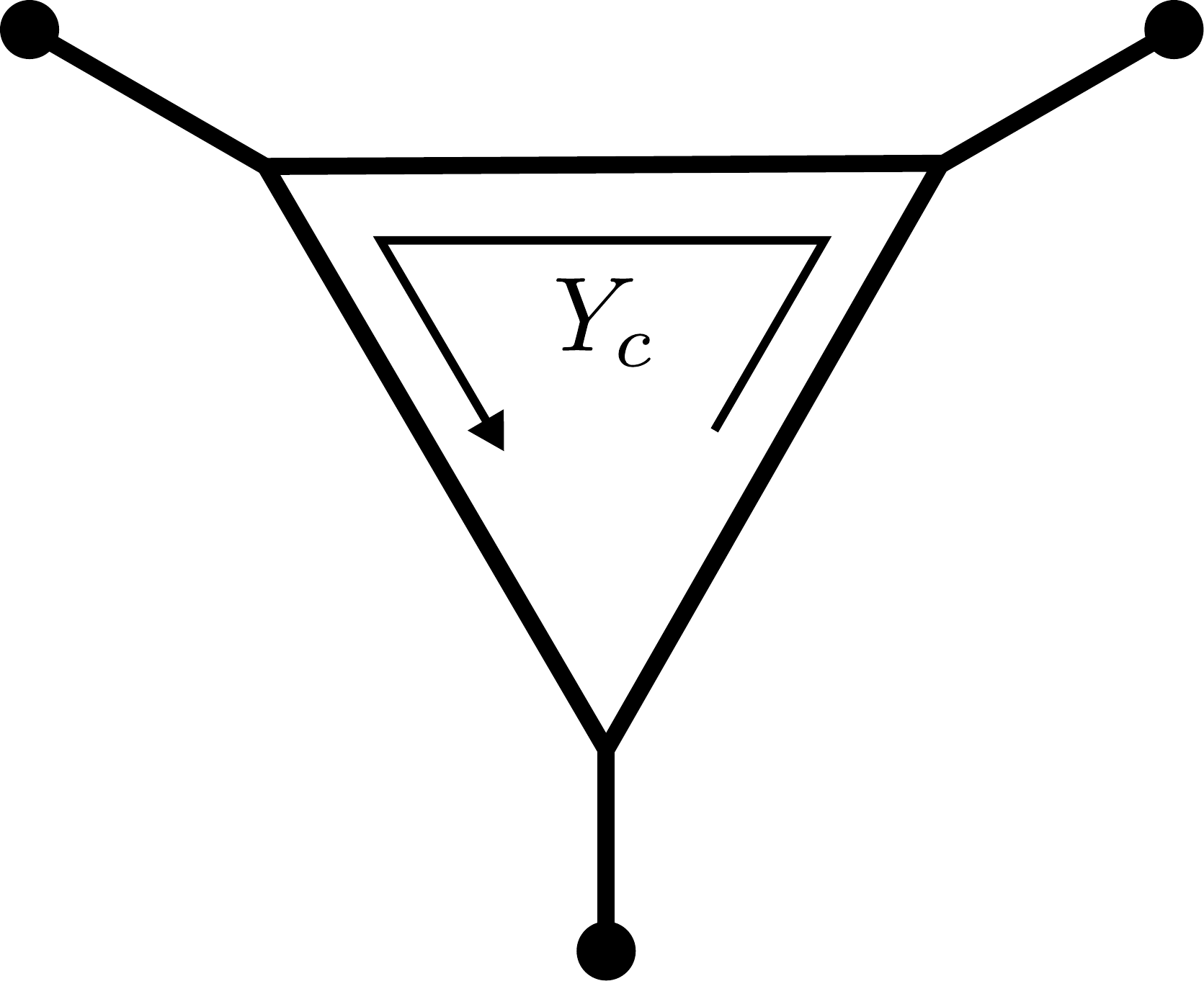}
    \caption{Symbol for the three-terminal element with response $\mat{Y}_c$.}
    \label{fig:3t_circ}
\end{figure}

Far from satisfying the reciprocity condition Eq.~\eqref{eq:recipdef}, this matrix is {\em anti}-symmetric, so we can see this as the most extreme possible violation of reciprocity -- a kind of anti-reciprocity. In the rest of this section, we will explore the consequences, for circuit-Hamiltonian theory, of the existence of such circuit elements.

In electromagnetism it is understood that reciprocity is only satisfied if the dielectric and diamagnetic material responses satisfy certainly symmetry properties, which are violated, for example, in magnetized materials (leading to the Faraday effect). Thus, it should not be surprising that discrete electrical circuits can also violate reciprocity under some circumstances. Still, it was not until 1947~\cite{tellegen} that Tellegen recognized theoretically that electrical networks were in need of such an element. 

\begin{figure}
    \centering
    \includegraphics[scale=0.15]{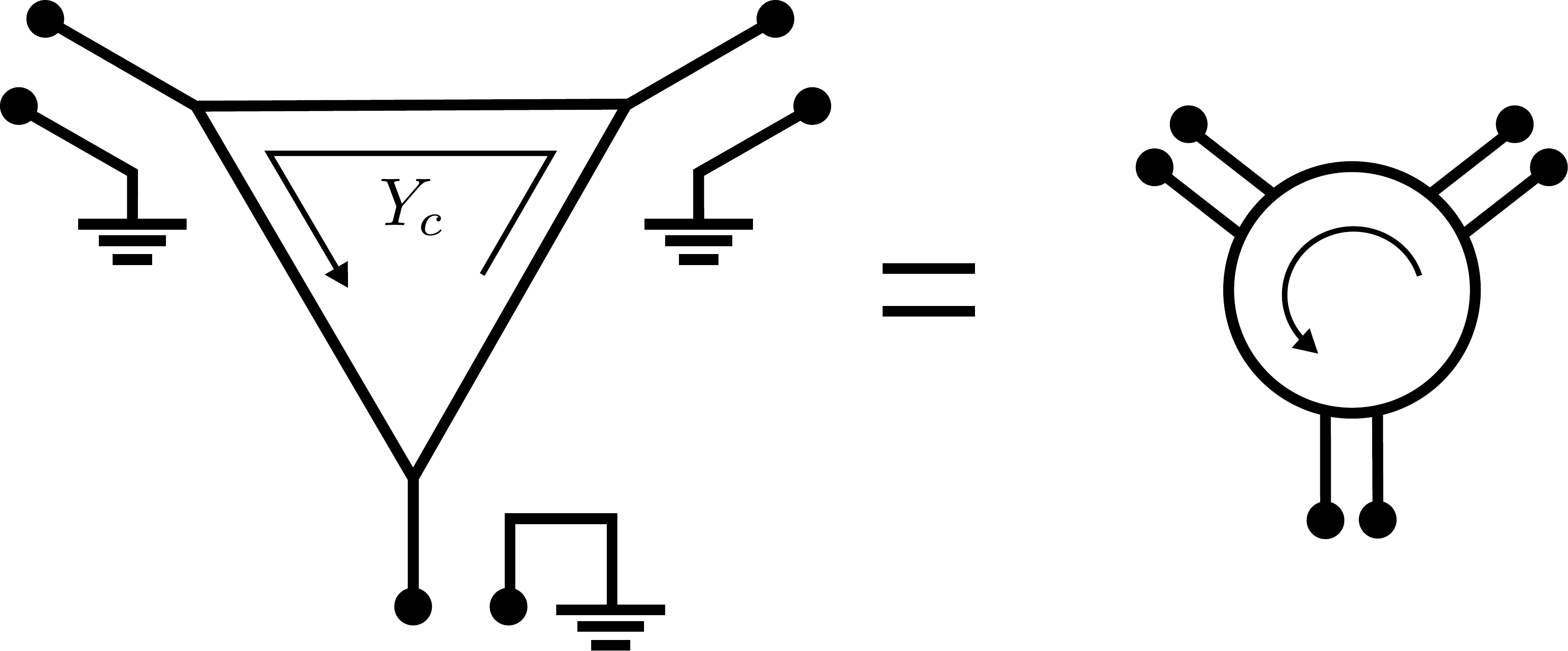}
    \caption{Defining the circulator by converting each terminal in  $\mat{Y}_{\mathrm{c}}$ to a port.}
    \label{fig:circ}
\end{figure}

Subsequently, the circulator was realized experimentally in 1952~\cite{Hogan}. Circulators in their present-day form were available by around 1960. A modern circulator is matched to the transmission system to which it is connected (i.e., $Z_0=50\, \mathrm{\Omega}$), and exhibits the response corresponding to $\mat{Y}_{c}$, in a sense that we will discuss shortly, over a wide but finite frequency range, e.g., $2$-$8 \, \mathrm{GHz}$. In principle, there is no physical reason why the circulator response could not occur for all frequencies, and the circuit theory assumes this. Thus, we will augment our circuit QED theory with the idealized device response of Eq.~\eqref{circ}. Modeling the finite-bandwidth response of the real circulator requires augmenting the model with conventional components (inductors and capacitors)~\cite{Bosco}.

To understand what exactly is ``circulating" in a circulator, it is handy to compute the scattering matrix $\mat{S}(s)$ matrix from the $\mat{Y}_c(s)$ matrix above. First we consider each of the three terminals as ports; this means pairing each of them with a ground, which is done quite naturally in the laboratory circulator, in which each port is a coaxial connection, see Fig.~\ref{fig:circ}. Then one can imagine connecting all three ports to transmission lines with characteristic impedance $Z_0$ (as in Fig.~\ref{fig:el_network_tl} for a general $N$-port network). If the circulator and the transmission lines are impedance matched, i.e., $G_c^{-1}=Z_0$, then using Eq.~\eqref{eq:scat-lap}, we obtain the scattering matrix of the circulator (verify this!)
\begin{equation}
    \mat{S}_{\rm c}(s)=\left(\begin{array}{rrr}0&0&1\\1&0&0\\0&1&0\end{array}\right).
    \label{eq:circs}
\end{equation}
The $\mat{S}_{\rm c}(s)$ matrix says that a signal from port $3$ exits at port $1$, $1$ exits at $2$, and $2$ exits at $3$, so, like in a roundabout, one can only circulate one way. This is the information conveyed in the conventional symbol for the circulator, as seen in practical engineering work shown in Fig.~\ref{fig:neg_res} (as described in the Wikipedia article \href{https://en.wikipedia.org/wiki/Circulator}{Circulator}.)

\begin{figure}
    \centering
\includegraphics[scale=0.15]{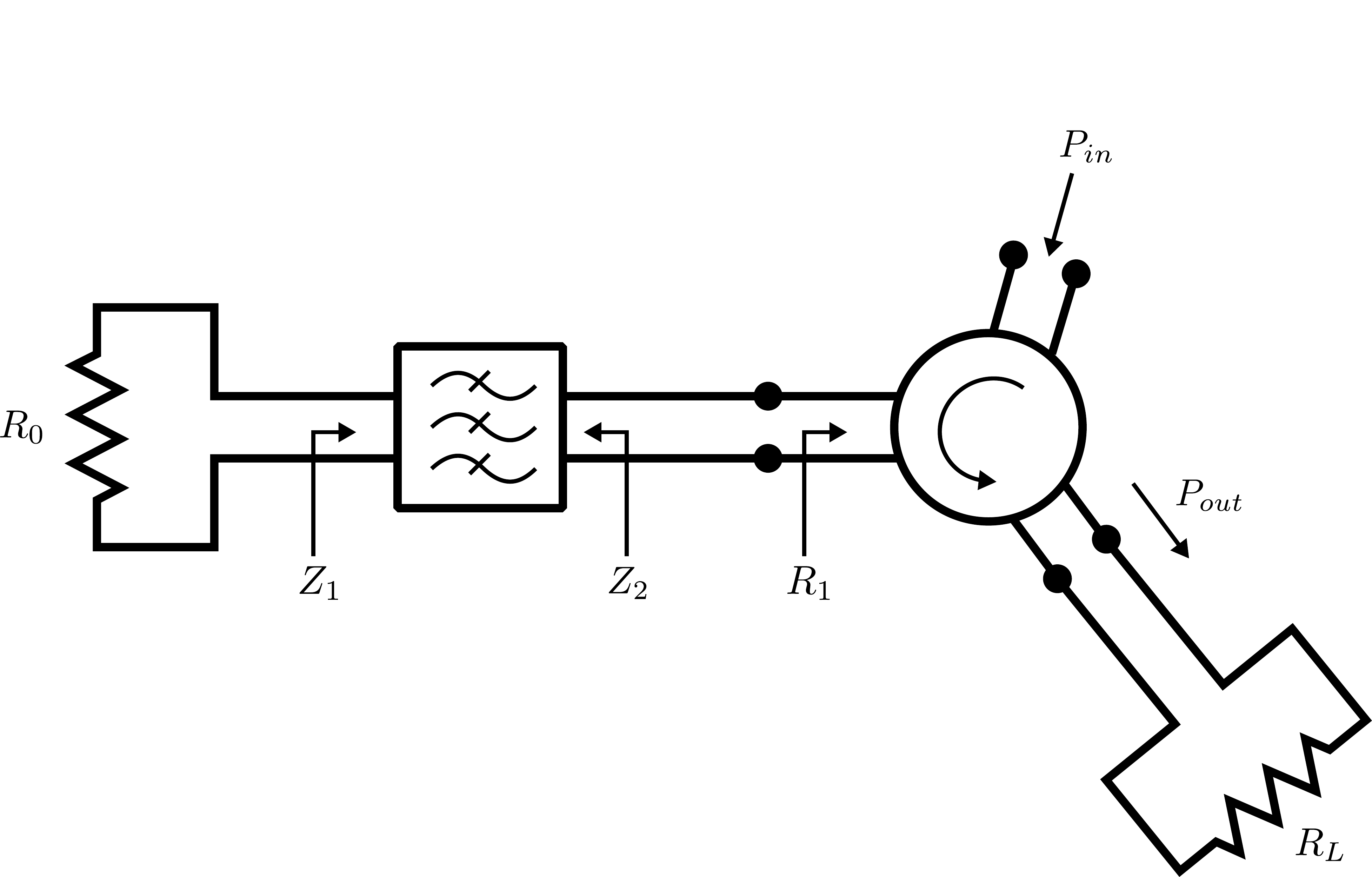}
    \caption{A circulator attached to three transmission lines. Example from real-life microwave engineering practice of the application of a circulator in a negative resistance amplifier. Here $R_0$ is a differential negative resistance, and the two-port provides impedance matching from $Z_1$ to $Z_2$. See \href{https://en.wikipedia.org/wiki/Negative_resistance\#Reflection_amplifier}{Wikipedia article on negative resistance} for details.} 
    \label{fig:neg_res}
\end{figure}

Tellegen, in fact, introduced a slightly different component, namely the {\em gyrator}, as the new electrical circuit element. The symbol for the gyrator is shown in Fig.~\ref{fig:gyr}. The word circulator, and the idea behind it, did not appear until four years after his work. We assert that these are two different manifestations of the same basic circuit response described by $\mat{Y}_{c}$. The difference is in how they are associated with a multi-port response. $\mat{Y}_{c}$ becomes a circulator if each of the three terminals is paired separately with ground, as was shown in Fig.~\ref{fig:circ}. In a gyrator, see 
Fig.~\ref{fig:gyr_circ}, one terminal of $\mat{Y}_{c}$ is singled out (shown here as ground),
and two ports are defined by a transformer coupling (recall Fig.~\ref{fig:mag_coupled_circuits}, and the discussion there). This construction has a feature that is often implied, but rarely stated as a feature of a multi-port device, namely that the ports are isolated from each other electrically.

Let us take a moment to comment on this idea of ``isolation". Mathematically the result of the insertions of transformers is that all four of the terminals are independent variables ---there is no common ground. In experiments, this permits one part of the circuit to have a ``floating voltage" with respect to the other. Depending on the experimental situation, it can be very advantageous to have this electrical isolation; on the other hand, sometimes the requirements of the experiment make it {\em crucial} that there be a common ground everywhere in the experiment.  Theorists, make sure to ask your experimentalist friends, when you are modeling their circuit, whether they employ a floating circuit or not! (Sometimes they are too shy to tell you.) See Exercise \ref{exc:grounding} for different groundings for transmon qubits, and also \href{https://en.wikipedia.org/wiki/Floating_ground}{Wikipedia entry on floating ground} for more on this issue.

\begin{figure}
    \centering
    \includegraphics[height=4cm]{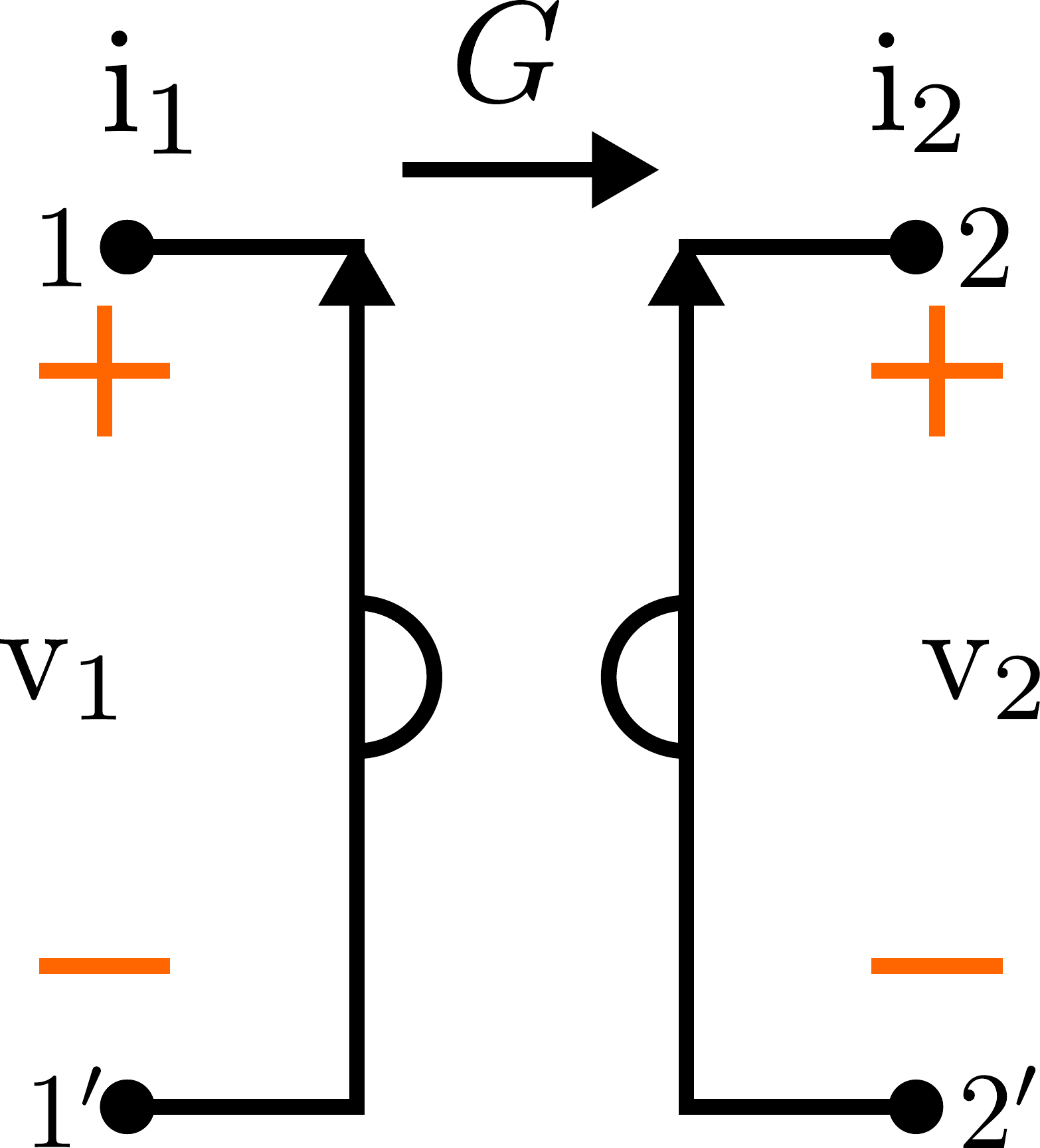}
    \caption{Circuit symbol of a gyrator.}
    \label{fig:gyr}
\end{figure}

\begin{figure}
    \centering
    \includegraphics[scale=0.13]{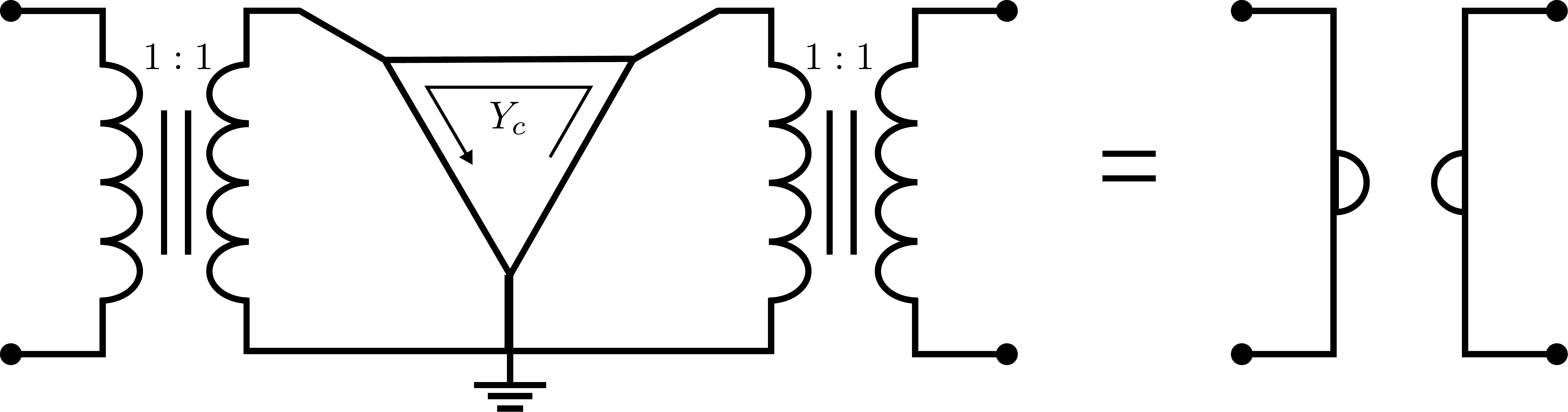}
    \caption{Obtaining a gyrator from $\mat{Y}_{\mathrm{c}}$ using ideal transformers.}
    \label{fig:gyr_circ}
\end{figure}

In Fig.~\ref{fig:gyr_circ_sym}, we show the symbols that you will often see for the circulator and the gyrator in microwave-circuit diagrams. They require some interpretation from our point of view: the single lines imply ports, not terminals, and you must find out for yourself what sort of grounding or isolation scheme is meant.

\begin{figure}
    \centering
    \includegraphics[scale=0.15]{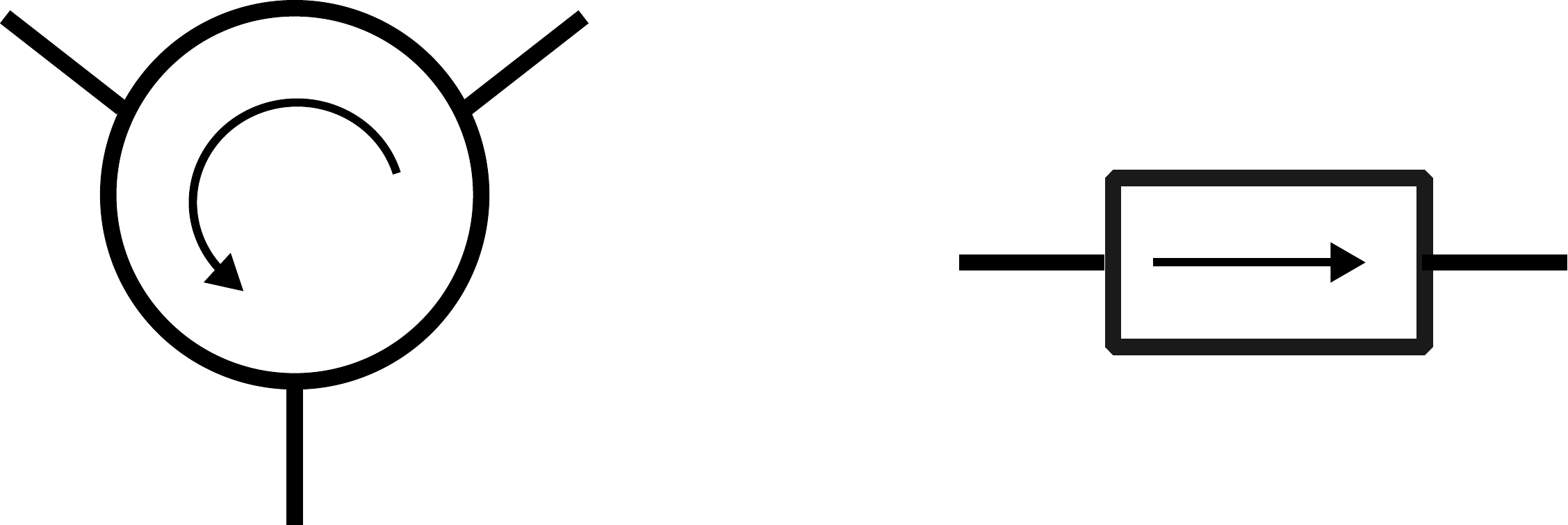}
    \caption{Representation of circulator (left) and gyrator (right), which are common in microwave-circuit diagrams.}
    \label{fig:gyr_circ_sym}
\end{figure}

This will have an influence on our main task, which we will embark on next, namely to introduce terms into the circuit Lagrangian and Hamiltonian that represent the effect of the gyrator or the circulator. In the fully-isolated case, we have remarked that the gyrator has four independent node fluxes. With additional transformer isolation, i.e., one adds transformers to each port in Fig.~\ref{fig:circ}, the circulator has six independent node fluxes. 

In many applications it is no problem if there is no port isolation, so that, for example, the lower two nodes of the gyrator are tied together. This is usually, but not always, shown explicitly in the gyrator circuit drawing when this is the intention. In the notation of microwave signal flows, where the conventional symbol for the gyrator is the one shown in Fig.~\ref{fig:gyr_circ_sym}, there is {\em no way} to show this distinction. To reiterate: this tying together is quite relevant for us, in that it can reduce the number of independent node fluxes in the circuit.

From this point onwards we will stick with the gyrator. The case of the circulator or the basic $\mat{Y}_{c}$ element can be covered similarly.

\section{Admittance matrix of the gyrator}

The presence of a gyrator clearly affects the dynamics of a circuit, and so it must somehow influence the construction of the Lagrangian and the Hamiltonian. However, we run into a paradoxical situation: When Tellegen introduced the gyrator in his first paper, one of the first things he proved about it is that it is lossless, but that it {\em stores no energy.} This means that we cannot apply the principle that we have used in Section~\ref{sec:ind-cap}, namely that we construct the Lagrangian by adding the energy contributions of each circuit element. We need another principle to deduce its contribution to the Lagrangian, as we did in the case of the transformer discussed in Section~\ref{subsec:ideal-trafo}.
We will see that it is possible to add a term that agrees with the observation that it does not influence the energy, in a classical sense, and yet has a real effect on the circuit dynamics, and also a real effect on the quantization of the circuit.

To proceed further, let us first state again what conditions the gyrator imposes on the circuit dynamics. Deducing from $\mat{Y}_{c}$ and using the representation of Fig.~\ref{fig:gyr_circ}, we find that the gyrator has a two-port impedance matrix
\begin{equation}
    \mat{Y}_{\mathrm{gyr}}(s)=\mat{Y}_{\mathrm{gyr}}=G \left( \begin{array}{rr} 0&-1\\1&0\end{array}\right).\label{Ygyr}
\end{equation}
with $G=G_c$. Note that being the admittance matrix of two ports, it is invertible, i.e.,
\begin{equation}
\mat{Z}_{\mathrm{gyr}}(s)=\mat{Y}_{\mathrm{gyr}}^{-1}(s) \propto \mat{Y}_{\mathrm{gyr}}(s)
\label{eq:Zgyr}
\end{equation}

\begin{Exercise}[label=exc:gyr-imp]
Verify Eq.~\eqref{Ygyr} using Fig.~\ref{fig:gyr_circ} and the action of ideal transformers given in Section~\ref{sec:mi} and $G=G_c$ (use the convention that the first row of the matrix $\vect{Y}_{\mathrm{gyr}}(s)$ labels the left port in Fig.~\ref{fig:gyr_circ} and the second row the right port).
\end{Exercise}

\begin{Answer}[ref=exc:gyr-imp]
Taking $\vect{I}(s)=\vect{Y}_c(s)\vect{V}(s)$, restricting the input $\vect{V}(s)$ to be of the form $\vect{V}^T(s)=\!\left( V_1(s), V_2(s), 0 \right)$ (note that the turns ratio is 1 and circulation goes from $1 \rightarrow 2 \rightarrow 3 \rightarrow 1$), gives the vector $\vect{I}^T(s)=G_c\left( V_2(s), - V_1(s), V_1(s)-V_2(s)\right)$ from which we can deduce Eq.~\eqref{Ygyr} by a change of port labeling.
\end{Answer}

For completeness, we write Eq.~\eqref{Ygyr} in components:
\begin{eqnarray}
I_1(s)&=&-G V_2(s),\nonumber\\
I_2(s)&=&\,\,\,\,\,G V_1(s).
\end{eqnarray}
This is a kind of a transformer characteristic, in which the input (port $1$) is ``stepped up" by $G$ to the output (port $2$). But it is in fact quite distinct from a transformer, see Eqs.~\eqref{eq:transform} in Section~\ref{sec:mi}, in that it ``gyrates" from the current to the voltage in going from input to output. 

\begin{Exercise}[label=exc:gyr0]
Give the admittance matrix of the gyrator in the time domain and show this implies 
\begin{eqnarray}
\ii_1(t)&=&-G \vv_2(t),\nonumber\\
\ii_2(t)&=&\,\,\,\,\,G \vv_1(t).
\label{eq:gyr-rel}
\end{eqnarray}
\end{Exercise}

\begin{Answer}[ref={exc:gyr0}]
If we take $s=i \omega$, we can use the inverse Fourier transform (see Eq.~\eqref{eq:laplace}) \begin{equation*}
\vect{Y}_{\mathrm{gyr}}(t)=\frac{1}{2\pi}\int_{-\infty} d\omega \,e^{i \omega t} \vect{Y}_{\mathrm{gyr}}(s=i \omega),
\end{equation*} 
to get
\begin{equation}
\vect{Y}_{\mathrm{gyr}}(t) = \delta(t)\mat{Y}_{\mathrm{gyr}},
\end{equation}
immediately implying Eq.~\eqref{eq:gyr-rel}.
\end{Answer}

If the gyrator is part of an impedance-matched circuit ---take $Z_0=G^{-1}$---  with two transmission lines at its ports as in Fig.~\ref{fig:el_network_tl}, then its scattering matrix $\mat{S}_{\rm gyr}(s)$ is very simple, namely
\begin{equation}
    \mat{S}_{\mathrm{gyr}}(s)=\left( \begin{array}{rr} 0&1\\-1&0\end{array}\right).\label{Sgyr}
\end{equation}
This can be derived, following the pattern for the circulator, from Eq.~\eqref{eq:scat-lap} applied to $\mat{Y}_{\rm gyr}(s)$ (verify this!).

\begin{figure}
    \centering
    \includegraphics[scale=0.35]{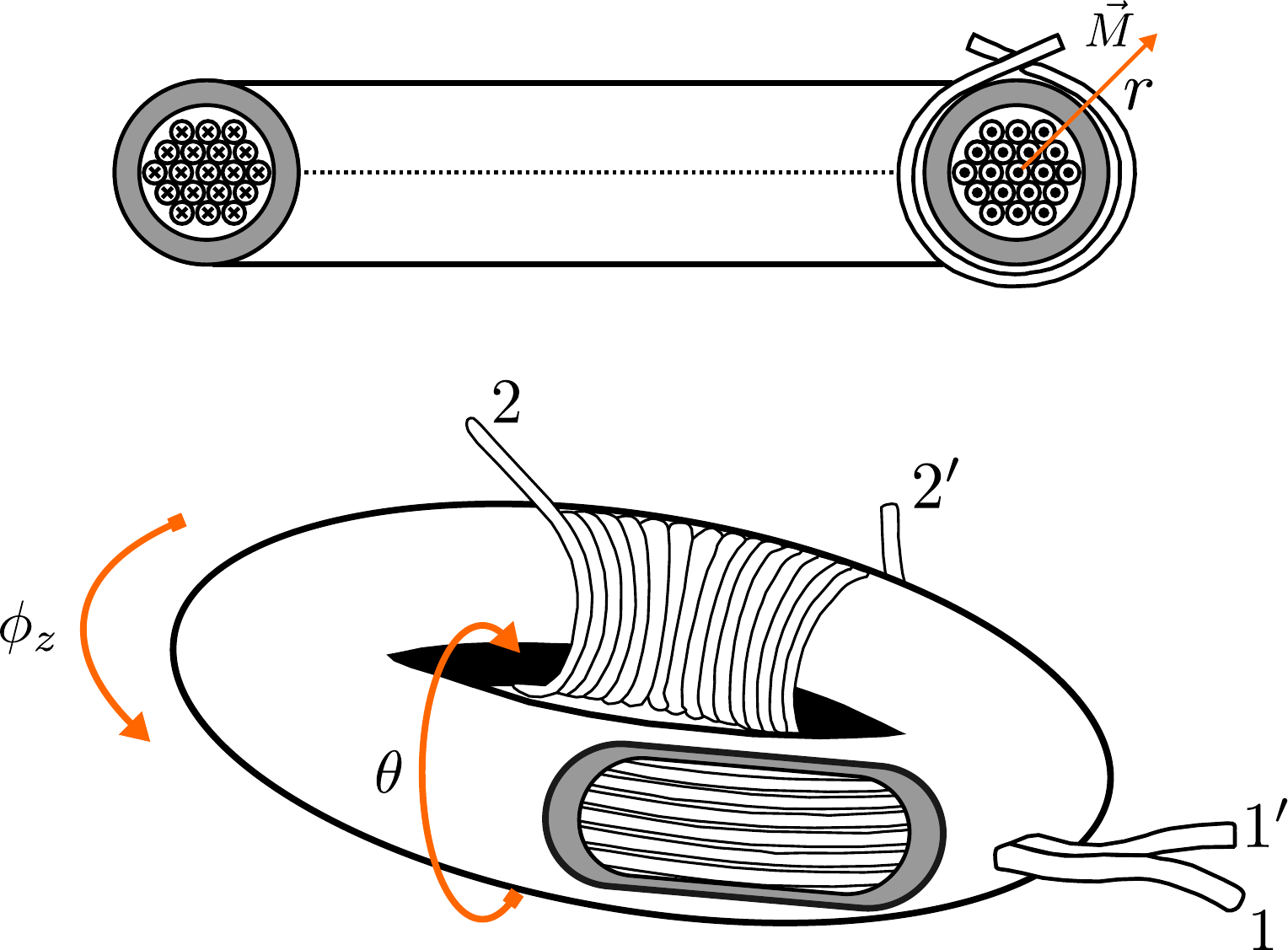}
    \caption{Tellegen's patented coil gyrator, reproduced from \cite{TellegenPatent2}.}
    \label{fig:coil_gyrator}
\end{figure}

\begin{Exercise}[label=exc:gyrT]
In this exercise we will perform a qualitative analysis of a physical gyrator proposed in 1954 by Tellegen \cite{TellegenPatent2}, see Fig.~\ref{fig:coil_gyrator}, with radial direction $\hat{r}$ and angular directions $\hat{\theta}$, $\hat{\phi}_z$ (toroidal coordinate system).
\Question Consider the two coils in Fig.~\ref{fig:coil_gyrator}. Using arguments of magnetostatics, show that the mutual inductance between coil $1$-$1'$ (port $1$) and coil $2$-$2'$ (port $2$) is zero, shown in the bottom Figure.
\Question Given the special properties of the material comprising the shell of the torus (in grey), there will in fact be a coupling between the two coils. For this we assume that material is an {\em anomalous quantum Hall} material. This means that the material is magnetized as indicated in the right part of Fig.~\ref{fig:coil_gyrator}(top), with the magnetization  vector $\vec{M}$ pointing in the $\hat{r}$ direction. As a consequence, the material exhibits the Hall effect: if an electric field is applied tangentially to the thin material, a current will flow perpendicular to this field, with a direction given by the right-hand rule. 

Given all this, recall that to analyze the off-diagonal components of the $2 \times 2$ impedance matrix $\mat{Z}(i\omega)$, one should apply a current at the first port, and determine an open-circuit voltage at the second port. Consider a time-dependent current $\sin(\omega t)$. Analyze the vectorial character of the induced fluxes, electromotive forces (electric fields) and induced currents. Do this for both choices of in and out ports, and show that $Z_{12}(i\omega)=-Z_{21}(i\omega)$. Also confirm that $Z_{12}$ is real, for any choice of $\omega$ \footnote{However, note that $Z_{11}(i\omega)$ and $Z_{22}(i\omega)$ are not zero due to self-inductance, so this does not represent the ideal gyrator of Eq.~\eqref{eq:Zgyr}.}. 
\end{Exercise}

\begin{Answer}[ref={exc:gyrT}]
\Question Coil $1$-$1'$ generates a magnetic field in the $\hat{\theta}$-direction which is orthogonal to the $\hat{\phi}_z$-direction and hence does not provide a magnetic field to which coil $2$-$2'$ can couple (and vice versa). 
\Question Suppose we inject current in port $1$ (coil $1$-$1'$) of the form $\sin(\omega t)$. This will produce a magnetic flux in the $\hat\theta$-direction, also proportional to $\sin(\omega t)$. This will induce an EMF in the Hall material in the $\hat\phi_z$-direction; due to the time derivative, this voltage will be of the form $\cos(\omega t)$. Due to the response tensor of the Hall material, there will be a current in the toroidal conducting shell, in the $\hat\theta$-direction, of functional form $\cos(\omega t)$. This current will induce a time-varying flux in the $\hat\phi_z$-direction of the form $\cos(\omega t)$. This will produce an EMF on the coil of the output port $2$ (coil $2$-$2'$), pointing in the $\hat\theta$-direction; with another time derivative, this is of the form $-\sin(\omega t)$. This is the open-circuit voltage at port $2$ (coil $2$-$2'$). The ratio of output voltage to input current, $Z_{21}(i\omega)=V_2(i\omega)/I(i\omega)$, is in-phase (so that $Z_{21}$ is real). 

The chain of causation is similar in the reverse direction: a current $\sin(\omega t)$ at port $2$ (coil $2$-$2'$) produces flux $\sin(\omega t)$ in the $\hat\phi_z$-direction, producing EMF $\cos(\omega t)$ in the $\hat\theta$-direction, producing Hall current $\cos(\omega t)$ in the $-\hat\phi_z$ direction, producing flux $\cos(\omega t)$ in the $-\hat\theta$-direction, producing an EMF $-\sin(\omega t)$ in the $-\hat\phi_z$-direction, or $\sin(\omega t)$ in the $\hat{\phi}_z$-direction, and this is the open-circuit voltage at port $1$ (coil $1$-$1'$). This gives $Z_{12}$ in the same way as $Z_{21}$, but with an extra minus sign arising from the anti-reciprocity of the Hall response.  
\end{Answer}

\section{The Lagrangian and Hamiltonian of the gyrator}
\label{sec:gyrator_Lagrangian}

Here we show that an acceptable prescription for the Lagrangian of a gyrator, when it appears in a superconducting circuit, is to add to the total Lagrangian the term
    \begin{equation}
        \lagrangian_{\mathrm{gyr}}=\frac{G}{2}\left[(\Phi_1-\Phi_{1'})(\dot{\Phi}_2-\dot{\Phi}_{2'})-(\Phi_2-\Phi_{2'})(\dot{\Phi}_1-\dot{\Phi}_{1'})\right],
        \label{eq:gyrdef}
    \end{equation}
with the labeling of the node fluxes $\Phi_1,\Phi_{1'}, \Phi_2, \Phi_{2'}$ as in Fig.~\ref{fig:gyr}.
This is clearly a new kind of object: we have not seen this first-order coupling term between `velocity' and `position' before. If you have done a lot of mechanics you may have seen it, but we will in any case explain below what it does. Whether all four of the node fluxes appearing in this Lagrangian term are independent will be determined by other parts of the network; this will be important for the Hamiltonian that we eventually get (as we only define momenta for independent variables), but not for the Lagrangian.

We will show the correctness of this prescription in a fairly general (but not the most general) setting; see more details in Refs.~\cite{rymarz:msc, placke:bsc, parra:gyrator}. Consider Fig.~\ref{fig:two_networks_gyrator} in which we have two general electrical networks with Lagrangians $\mathcal{L}_1$ and $\mathcal{L}_2$ which are coupled by a gyrator. We take the nodes $\Phi_{g1}$ and $\Phi_{g2}$ to be reference nodes in each of the two graphs (left and right) that compose the circuit and set $\Phi_{g1}=\Phi_{g2}=0$. 
The Lagrangians $\mathcal{L}_k$ of the networks are not just functions of $\Phi_k$ and $\dot{\Phi}_k$ but, in general, also of some additional internal degrees of freedom, which are, however, of no further importance in this analysis.
If the ideal gyrator has two of the ports connected to a common ground, i.e.,~$\Phi_1'=\Phi_2'=\dot{\Phi}_{1'}=\dot{\Phi}_{2'}=0$, Eq.~\eqref{eq:gyrdef} reduces to
\begin{equation}\label{eq:gyrator_lagrangian}
	\lagrangian_{\mathrm{gyr}} =
	 \frac{G}{2} (\Phi_1 \dot{\Phi}_2 - \dot{\Phi}_1 \Phi_2),
\end{equation}
so that the fluxes $\Phi_1, \Phi_2$ are the only independent node variables. 
 
\begin{figure}[h]
	\centering
	\includegraphics[height=4cm]{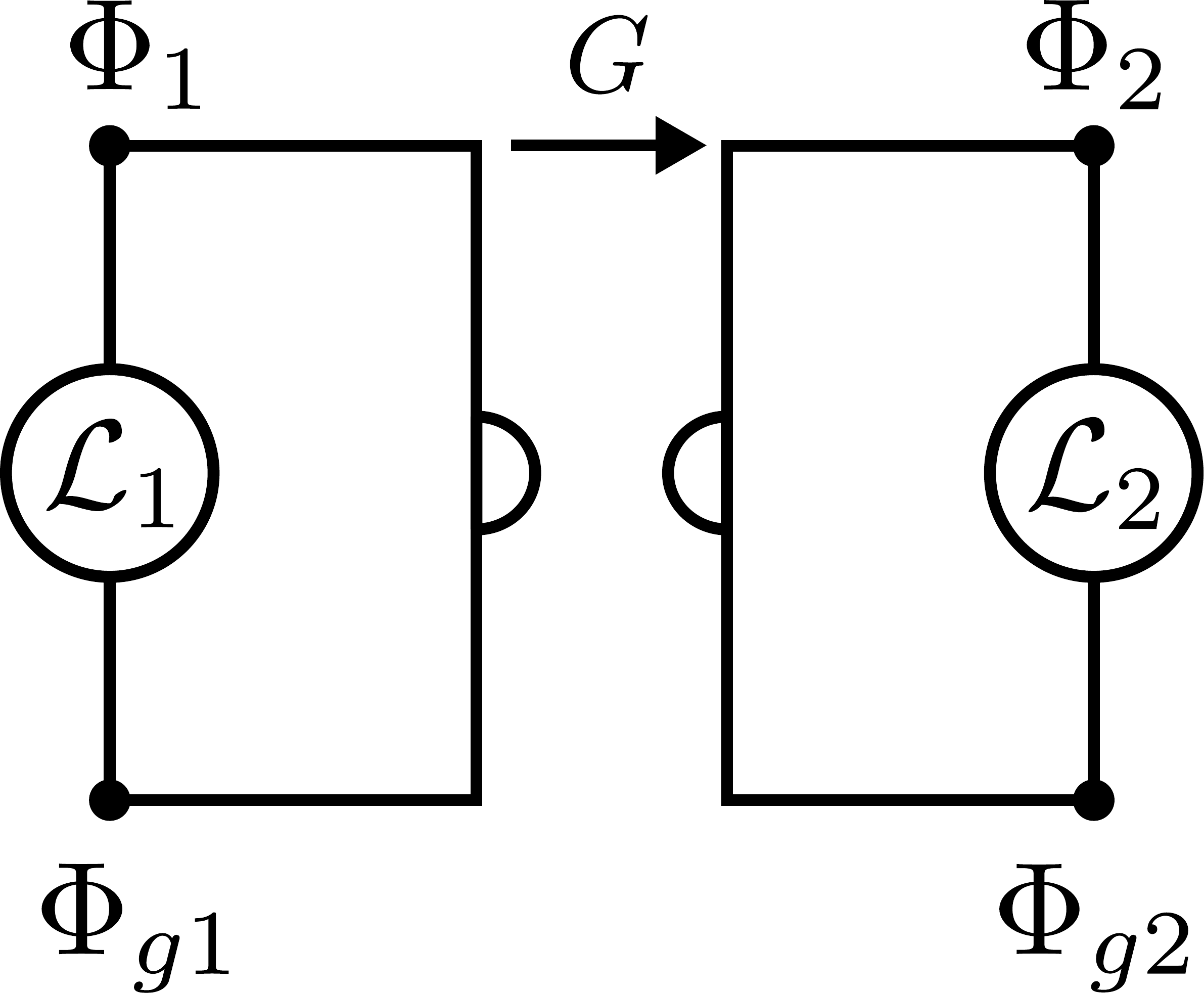}
	\caption{Two general electrical networks with Lagrangians $\mathcal{L}_1, \mathcal{L}_2$ which are coupled by a gyrator.}
	\label{fig:two_networks_gyrator}
\end{figure}

In this configuration, the total Lagrangian thus reads
\begin{equation}\label{eq_general_gyrator_network}
	\mathcal{L} = \mathcal{L}_1 + \mathcal{L}_2 + \mathcal{L}_\text{gyr} .
\end{equation}
 First of all, doing the classical mechanics, we note that the current passing through the electrical network with Lagrangian $\mathcal{L}_k$ is given by (do you see why?)
\begin{equation}
	\ii_{\lagrangian_k}(t) = \frac{d}{dt} \left( \frac{\partial \mathcal{L}_k}{\partial \dot{\Phi}_k} \right)
	- \frac{\partial \mathcal{L}_k}{\partial \Phi_k} = -\ii_k(t),
\end{equation}
with $\ii_k$ the current passing through the corresponding gyrator branch. By definition, its associated voltage reads
$\vv_k = \dot{\Phi}_k$. From that, we can conclude that the Euler-Lagrange equation of the total Lagrangian in Eq.~\eqref{eq_general_gyrator_network} with respect to the $\Phi_1$-variable yields
\begin{equation}
\begin{split}
	\frac{d}{dt} \left( \frac{\partial \mathcal{L}}{\partial \dot{\Phi}_1} \right) 
	- \frac{\partial \mathcal{L}}{\partial \Phi_1}
	&= 
	0, 	\\	
	\Leftrightarrow	\qquad
	\frac{d}{dt} \left( \frac{\partial \mathcal{L}_1}{\partial \dot{\Phi}_1} \right) 
	- \frac{\partial \mathcal{L}_1}{\partial \Phi_1}
	&=
	-\frac{d}{dt} \left( \frac{\partial \mathcal{L}_\text{gyr}}{\partial \dot{\Phi}_1} \right) 
	+ \frac{\partial \mathcal{L}_\text{gyr}}{\partial \Phi_1}, 				\\
	\Leftrightarrow	\hspace{81.5pt}
	 \ii_1(t) 
	&= - G \vv_2(t).
\end{split}	
\end{equation}

A similar evaluation of the Euler-Lagrange equation for $\Phi_2$ gives
\begin{equation}
	\ii_2(t) =  G \vv_1(t),
\end{equation}
such that both Euler-Lagrange equations together directly recover the admittance matrix characterizing the classical ideal gyrator in the time-domain in Eq.~\eqref{eq:gyr-rel}. Therefore, it is proven that the Lagrangian stated in Eq.~\eqref{eq:gyrator_lagrangian} indeed describes an ideal (frequency-independent) gyrator. \\

After this introduction of the gyrator Lagrangian description, we would like to recall that the Lagrangian formalism allows the addition of a total time derivative to the Lagrangian without affecting the classical equations of motion, see Exercise~\ref{exc:gauge} in Appendix~\ref{app:cc}. Therefore, it is instructive to consider a more general Lagrangian describing the ideal gyrator obtained by the transformation
\begin{equation}\label{equation_Lagrangian_timederivative}
	\lagrangian_{\mathrm{gyr}}
	\rightarrow \lagrangian_{\mathrm{gyr}} + a \frac{d}{dt} \left( \frac{G}{2} \Phi_1 \Phi_2 \right)
	= \frac{G}{2} \left[ (1+a) \Phi_1 \dot{\Phi}_2 - (1-a) \dot{\Phi}_1 \Phi_2 \right],
\end{equation}
with an arbitrary constant $a \in \mathds{R}$. We thus see that there is a continuum of related representations of the gyrator Lagrangian. In the next section, this will be interpreted as a {\em gauge} freedom in the description of the dynamics.

\subsection{``Magnetic field does no work" -- gyrator as magnetic field}

Since all terms in Eq.~\eqref{equation_Lagrangian_timederivative} are linear in $\dot{\Phi}_i$, the Lagrangian $\lagrangian_{\mathrm{gyr}}$ can be written and interpreted as the Lagrangian contribution of a `magnetic' vector potential
\begin{equation}
	\lagrangian_{\mathrm{gyr}} = \vect{A}(\vect{\Phi}) \cdot \dot{\vect{\Phi}}
	,	\qquad
	\vect{A}(\vect{\Phi})
	=
	\frac{G}{2}
	\begin{pmatrix}
		(a-1)\Phi_2		\\	
		(a+1)\Phi_1
	\end{pmatrix}.
\end{equation}
The two-dimensional vector potential $\vect{A}(\vect{\Phi})$ gives rise to a resulting magnetic field analog $\mat{B}=\vect{\nabla} \times \vect{\vect{A}}$, which is orthogonal to the $\Phi_1 , \Phi_2$ -plane with strength
\begin{equation}\label{eq_arbitrary_gauge}
	\vect{B}_3=(\vect{\nabla} \times \vect{\vect{A}(\vect{\Phi})})_3 
	= \frac{\partial A_2(\vect{\Phi})}{\partial \Phi_1} - \frac{\partial A_1(\vect{\Phi})}{\partial \Phi_2} = G,
\end{equation}
which is the gyration conductance. Given this analogy between the gyration conductance and a uniform magnetic field, the transformation considered in Eq.~\eqref{equation_Lagrangian_timederivative} is equivalent to a gauge transformation of the vector potential $\vect{A} \rightarrow \vect{A}+\vect{\nabla}\lambda$. The special cases of $a=\pm1$ and $a=0$ can be identified as two different Landau gauges and a symmetric gauge, respectively. Indeed, adding the total time derivative of any function $\lambda(\vect{\Phi})$ to $\lagrangian_{\mathrm{gyr}} $ is equivalent to adding $\vect{\nabla}\lambda(\vect{\Phi})$ to the vector potential, which does not affect its curl, since $\vect{\nabla} \times \vect{\nabla}\lambda(\vect{\Phi})=0$. Furthermore, the nonreciprocity of the passive gyrator is consistent with the interpretation that the gyrator acts formally like a static magnetic field, which breaks time-reversal symmetry and does not change the energy of (i.e., do any work on) the mechanical system to which it is coupled.

\subsection{Hamiltonian of the gyrator}\label{subsec:general_Legendre}
The Lagrangian contribution of the gyrator is the first step towards a quantized theoretical description of any electrical circuit containing any number of gyrators. Since we want to impose the quantization in the Hamiltonian formalism, we need to derive the corresponding Hamiltonian first. The Lagrangian of a general electrical network with $N+1$ nodes ---of which there is one ground node and $N$ independent ones---  which is built out of capacitances, any sort of inductances (linear, non-linear and mutual) and gyrators can be written as
\begin{equation}\label{eq_general_network_Lagrangian}
	\mathcal{L} 
	= \frac{1}{2} \dot{\vect{\Phi}}^{T} \mat{C} \dot{\vect{\Phi}} 
	+ \dot{\vect{\Phi}}^T \vect{A} (\vect{\Phi}) 
	- U(\vect{\Phi}).
\end{equation}
We assume that the capacitance matrix is invertible, meaning that $N$ node variables are independent; see Chapter~\ref{chap:cq-app} and Proposition \ref{lem:captree} for a discussion on invertibility. We point to Ref.~\cite{rymarz:msc} for a discussion of the many aspects of the singular non-invertible case with gyrators.
We note that with the inclusion of the gyrator contribution, the Lagrangian is still a convex function of the variables $\vect{\dot{\Phi}}$, so we can proceed with the Legendre transformation to obtain a Hamiltonian. 
The vector of conjugate variables is 
\begin{eqnarray}\label{eq_phi_dot_prime}
	\vect{Q} &=& \mat{C} \dot{\vect{\Phi}} + \vect{A} (\vect{\Phi} ), \nonumber\\
	\Rightarrow \vect{\dot\Phi} &=&  \mat{C}^{-1}(\vect{Q} - \vect{A} (\vect{\Phi}) ).
\end{eqnarray}
We have not seen this form of conjugate variable before, but, again, it is analogous to the mathematical description of the dynamics of an electron in a magnetic field with its conjugate variable, the momentum, given by $\vect{p}=m \dot{\vect{x}}+q\vect{A}$.

Completing the Legendre transformation, we get a general expression for the Hamiltonian:
\begin{equation}\label{eq_general_Legendre_transformation}
\begin{split}
	\hamiltonian
	&= \dot{\vect{\Phi}}^{T} \vect{Q} - \mathcal{L}										\\
	&= \dot{\vect{\Phi}}^{T} {\vect{Q}} 
		-\frac{1}{2} \dot{\vect{\Phi}}^{T} \mat{C} \dot{\vect{\Phi}}
		-	\dot{\vect{\Phi}}^T \vect{A} (\vect{\Phi})
		+ U ( \vect{\Phi} )													\\
	&= \frac{1}{2} ( \vect{Q} - \vect{A} (\vect{\Phi}) )^T \mat{C}^{-1} 
			( \vect{Q} - \vect{A} (\vect{\Phi}) )
		+ U ( \vect{\Phi} ),
\end{split}
\end{equation}
Exercise \ref{exc:gyr1} will show that this Hamiltonian, when worked out for the case of two shunting capacitors, is indeed just that of an electron confined to two dimensions in a uniform magnetic field.

In current research, we see many strange possibilities that arise if the capacitance is singular in this situation; for instance, by shunting the ports of the gyrator with Josephson junctions, one can get a kinetic energy of the form $\cos (\frac{\vect{Q}}{2e} )$~\cite{rymarz2021,PhysRevA.100.062321}. We believe that understanding the Hamiltonian description of such unusual circuits may also lead to suggestions for novel physical realizations of such circuits, which can then exhibit new physics.

\begin{Exercise}[title={Ideal gyrator},label=exc:gyr1]
\begin{figure}[htb]
\centering
\includegraphics[height=4cm]{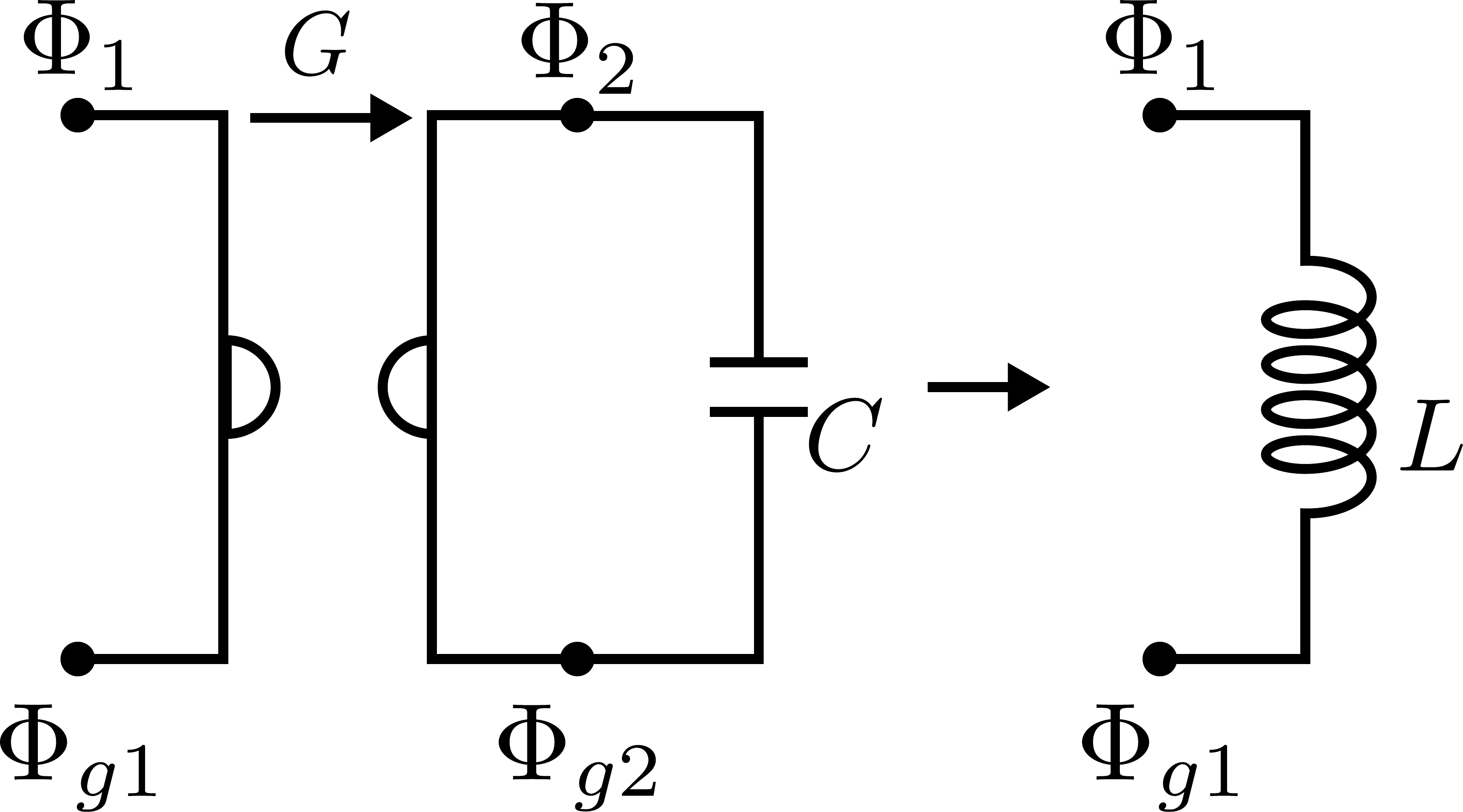}
\caption{Transforming a capacitor into an inductor using a gyrator.}
\label{fig:cgl}
\end{figure}
\Question Consider the circuit in Fig.~\ref{fig:two_networks_gyrator} with the particular case in which the Lagrangians $\mathcal{L}_1$ and $\mathcal{L}_2$ are the Lagrangians of LC~oscillators with inductance $L$ and capacitance $C$ in parallel. Assume the same inductance and the same capacitance on the two sides. Obtain the Hamiltonian in this case. What is a mechanical equivalent of this system? 
\Question Use the Lagrangian formalism to show that it is possible to transform a capacitor into an inductor using a gyrator as shown in Fig.~\ref{fig:cgl}. We can assume that $g1$ and $g2$ are reference nodes in each disconnected graph and set $\Phi_{g1}=\Phi_{g2}=0$. Obtain the value of the equivalent inductance $L$.
\end{Exercise}

\begin{Answer}[ref={exc:gyr1}]
\Question In the case of two LC oscillators the Lagrangian of Fig.~\ref{fig:two_networks_gyrator} reads
\begin{equation}
\mathcal{L}= \frac{C}{2} \dot{\Phi}_1^2 -\frac{\Phi_1^2}{2 L}+\frac{C}{2} \dot{\Phi}_2^2 -\frac{\Phi_2^2}{2 L} +\frac{G}{2} (\Phi_1 \dot{\Phi}_2-\dot{\Phi}_1 \Phi_2).
\end{equation}
To obtain the Hamiltonian we define the conjugate variables
\begin{subequations}
\begin{equation}
Q_1 = \frac{\partial \mathcal{L}}{\partial \dot{\Phi}_1}= C \dot{\Phi}_1-\frac{G}{2} \Phi_2,
\end{equation}
\begin{equation}
Q_2 = \frac{\partial \mathcal{L}}{\partial \dot{\Phi}_2}= C \dot{\Phi}_2+\frac{G}{2} \Phi_1.
\end{equation}
\end{subequations}
By taking the Legendre transform, we get
\begin{equation}
H = Q_1 \dot{\Phi}_1 +Q_2 \dot{\Phi}_2 -\mathcal{L}= \frac{(Q_1+\frac{G}{2} \Phi_2)^2}{2 C} +\frac{(Q_2-\frac{G}{2} \Phi_1)^2}{2 C} +\frac{\Phi_1^2}{2 L}+\frac{\Phi_2^2}{2 L}.
\end{equation}
We notice that this Hamiltonian has the same structure as the Hamiltonian of a charged particle moving in two dimensions in a constant magnetic field perpendicular to the plane, with, in addition, a paraboloid potential. This can be seen even more clearly if we define the fictitious vector potential
\begin{equation}
\vect{A}^T= (A_1,A_2)= \frac{G}{2} (
-\Phi_2,\Phi_1),
\end{equation}
and rewrite the Hamiltonian as
\begin{equation}
H= \frac{\lvert \vect{Q}-\vect{A} \rvert^2}{2 C}  +\frac{\Phi_1^2}{2 L}+\frac{\Phi_2^2}{2 L},
\end{equation}  
with $\vect{Q}=(Q_1 \quad Q_2)^T$. 
\Question
The Lagrangian of the circuit in Fig.~\ref{fig:cgl} reads
\begin{equation}\label{eq:lagr_cgl}
\mathcal{L} = \frac{G}{2} \bigl(\Phi_1 \dot{\Phi}_2 - \dot{\Phi}_1 \Phi_2 \bigr) + \frac{C}{2} \dot{\Phi}_2^2.
\end{equation}
The Euler-Lagrange equation associated with $\Phi_2$ reads
\begin{equation}
\frac{d}{dt} \biggl(\frac{\partial \mathcal{L}}{\partial \dot{\Phi}_2} \biggr) - \frac{\partial \mathcal{L}}{\partial \Phi_2} = G \dot{\Phi}_1 + C \ddot{\Phi}_2 = 0 \implies \dot{\Phi}_2 = -\frac{G}{C} \Phi_1 + \alpha,
\end{equation}
with $\alpha \in \mathbb{R}$. Substituting this into the Lagrangian Eq.~\eqref{eq:lagr_cgl}, we obtain
\begin{multline}
\mathcal{L} = \frac{G}{2} \Phi_1 \biggl(-\frac{G}{C} \Phi_1 + \alpha \biggr) - \frac{G}{2} \dot{\Phi}_1 \Phi_2 +\frac{C}{2} \dot{\Phi}_2 \biggl(-\frac{G}{C} \Phi_1 + \alpha \biggr) = \\
-\frac{1}{2 C/G^2} \Phi_1^2 -\frac{G}{2} \frac{d(\Phi_1 \Phi_2)}{dt} + \frac{\alpha G}{2} \Phi_1 + \frac{\alpha C}{2} \dot{\Phi}_2 = -\frac{1}{2 C/G^2} \Phi_1^2 -\frac{G}{2} \frac{d(\Phi_1 \Phi_2)}{dt} + \frac{C}{2} \alpha^2,
\end{multline}
which is the Lagrangian of an inductor with $L = C/G^2$ plus a total time derivative and a constant. The total time derivative does not change the Euler-Lagrange equations, see Exercise \ref{exc:gauge}, and can thus be omitted.
\end{Answer}

\begin{Exercise}[title={A circulator from a gyrator}, label=exc:circ-gyr]
\begin{figure}[htb]
\centering
\includegraphics[height=5cm]{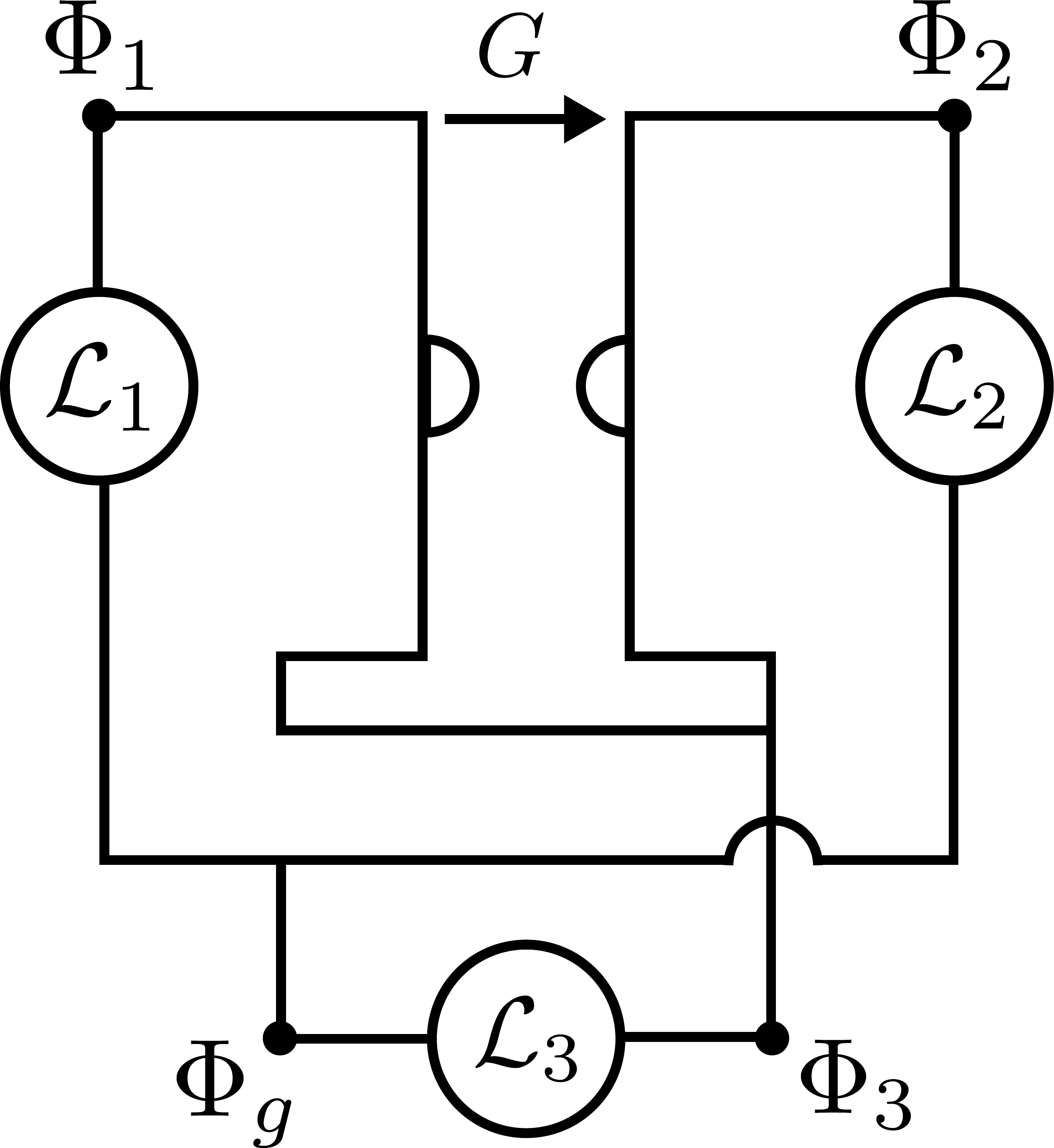}
\caption{Carlin's construction of a circulator. Here the three terminals are labelled and are supposed to identify with the terminals in Fig.~\ref{fig:3t_circ}.}
\label{fig:carlinfig}
\end{figure}

Let us consider the circuit depicted in Fig.~\ref{fig:carlinfig}. It represents a way to construct a circulator: the three-terminal element in Fig.~\ref{fig:3t_circ}, from a gyrator, known as Carlin's construction, connected to an external circuit with some Lagrangians $\lagrangian_i$. It is thus more or less the opposite of Fig.~\ref{fig:gyr_circ}, in which one obtains a gyrator from the element $\mat{Y}_c$.
\Question Starting from the Lagrangian of a gyrator introduced in Eq.~\eqref{eq:gyrator_lagrangian}, and setting $\Phi_g=0$ as reference node flux, show that the Lagrangian of the circulator can be written as
\begin{equation}
\lagrangian_{\mathrm{cir}}= G (\Phi_1 \dot{\Phi}_2+\Phi_2 \dot{\Phi}_3+\Phi_3 \dot{\Phi}_1).
\label{eq:lagr_circulator}
\end{equation} \par 
\emph{Hint: remember that summing a total time derivative to the Lagrangian does not change the equations of motion, see Exercise \ref{exc:gauge}.} 
\Question Building on the intuition developed in Exercise~\ref{exc:gyr1}, derive the admittance matrix $\vect{Y}_c(s)$ of a circulator (as in Eq.~\eqref{circ}) from the equations of motion. 
\Question From the admittance matrix, obtain the scattering matrix, assuming that all the transmission lines at the ports of the circulator have characteristic impedance $Z_0$. Obtain the gyration constant $G$ for which the scattering matrix becomes the scattering matrix of an ideal impedance-matched circulator $\mat{S}_c$ in Eq.~\eqref{eq:circs}. What happens if the transmission line is not impedance-matched?
\end{Exercise}

\begin{Answer}[ref=exc:circ-gyr]
\Question The Lagrangian of the circulator part of the circuit in~\cref{fig:carlinfig} follows immediately from the Lagrangian $\mathcal{L}_\mathrm{gyr}$ in Eq.~\eqref{eq:gyrdef}
\begin{equation}
\label{lcirc1}
\mathcal{L\sp{\prime}}_\mathrm{cir}=   \frac{G}{2} \left(\Phi_1 \dot{\Phi}_2 - \dot{\Phi}_1 \Phi_2+\Phi_2 \dot{\Phi}_3 - \dot{\Phi}_2\Phi_3+\Phi_3 \dot{\Phi}_1 - \dot{\Phi}_3 \Phi_1\right).
\end{equation}
We see that it is not of the form of~\cref{eq:lagr_circulator}, but it can be brought into that form by adding to the Lagrangian a total time derivative; i.e., an equivalent circulator Lagrangian is
\begin{equation}
\mathcal{L}_\mathrm{cir} = \mathcal{L\sp{\prime}}_\mathrm{cir} +\frac{G}{2} \frac{d}{dt}(\Phi_1 \Phi_2+ \Phi_2 \Phi_3+\Phi_3 \Phi_1).
\end{equation}
\Question We can now proceed to find the admittance matrix. Consider the total Lagrangian
\begin{equation}
\mathcal{L}= \mathcal{L}_1 + \mathcal{L}_2+\mathcal{L}_3 +\mathcal{L}_\mathrm{cir},
\end{equation}
with $\lagrangian_k= \lagrangian_k(\{\Phi_{k}, \dot{\Phi}_k\}), \, k=\{1,2,3\}$. The currents~$\ii_{\lagrangian_k}$ passing through the element with Lagrangian~$\lagrangian_k$ are defined as 
\begin{equation}
\ii_{\lagrangian_k} = \frac{d}{dt} \biggl(\frac{\partial \lagrangian_k}{\partial \dot{\Phi}_k} \biggr) - \frac{\partial \lagrangian_k}{\partial \Phi_k}.
\end{equation}
At the same time we have
\begin{equation}
    0=\frac{d}{dt} \biggl(\frac{\partial \lagrangian}{\partial \dot{\Phi}_k} \biggr) - \frac{\partial \lagrangian}{\partial \Phi_k}=\ii_{\lagrangian_k}+\frac{d}{dt} \biggl(\frac{\partial \lagrangian_{\rm cir}}{\partial \dot{\Phi}_k} \biggr) - \frac{\partial \lagrangian_{\rm cir}}{\partial \Phi_k}.
\end{equation}
Taking the convention that $\ii_k=\ii_{\lagrangian_k}$ as current going into the terminal, the Euler-Lagrange equations associated to the variables $\Phi_1, \Phi_2$ and $\Phi_3$ then read respectively
\begin{subequations}
\begin{equation}
\ii_1 +G \dot{\Phi}_3-G \dot{\Phi}_2 =0,
\end{equation}
\begin{equation}
\ii_2 +G \dot{\Phi}_1-G \dot{\Phi}_3 =0,
\end{equation}
\begin{equation}
\ii_3 +G \dot{\Phi}_2-G \dot{\Phi}_1 =0,
\end{equation}
\end{subequations}
which we can write compactly as
\begin{equation}
\left(\begin{array}{c}
\ii_1(t) \\
\ii_2(t) \\
\ii_3(t)
\end{array}\right)= 
\left(\begin{array}{cccc}
0 & G & -G \\
-G & 0 & G \\
G & -G & 0
\end{array} \right)
\left(\begin{array}{c}
\vv_1(t) \\
\vv_2(t) \\
\vv_3(t)
\end{array}\right).
\end{equation}
We can thus identify the admittance matrix of the circulator (in the Laplace domain) as
\begin{equation}
\mathbf{Y}_\mathrm{c}(s) =G \left(\begin{array}{cccc}
0 & 1 & -1 \\
-1 & 0 & 1 \\
1 & -1 & 0
\end{array}\right).
\end{equation}
\Question To get the scattering matrix from the admittance matrix, we use Eq.~\eqref{eq:scat-lap} where $Z_0$ is the characteristic impedance of the transmission lines. We get
\begin{equation}
\mat{S}_{\rm c}=\frac{1}{3G^2Z_0^2+1}
\left(
\begin{array}{ccc}
 1-G^2 Z_0^2 & 2 G Z_0 (G Z_0-1) & 2 G Z_0 (G Z_0+1) \\
 2 G Z_0 (G Z_0+1) & 1-G^2 Z_0^2 & 2 G Z_0 (G Z_0-1) \\
 2 G Z_0 (G Z_0-1) & 2 G Z_0 (G Z_0+1) & 1-G^2 Z_0^2 \\
\end{array}
\right),
\end{equation}
which only for $G= Z_0^{-1}$ becomes an impedance-matched circulator without any unwanted reflection as in Eq.~\eqref{eq:circs}. If there is no perfect impedance matching, we see that there is transmission from, say, port 1 to both port 2 and 3, as well as reflection.
\end{Answer}

\chapter{Noise, or all that can go wrong}
\label{chap:noise}

The physical sources of noise, dissipation and limited control of superconducting devices are manifold and a subject beyond this book. In this chapter we first discuss the modeling of noise in Sections~\ref{sec:lindblad} and \ref{sec:losscircuit}, and then, in Section~\ref{sec:ns}, some aspects of noise sensitivity and noise protection in superconducting qubits.

\section{Lindblad master equation model}
\label{sec:lindblad}

As we have seen in Section~\ref{sec:loss} and Appendix~\ref{app:norm_mode}, it is possible to obtain the Hamiltonian $H$ and the (nonnegative) decay rates $\kappa_m$associated with the modes ($\hat{a}_m$) of an $N$-port linear reciprocal network to which $N$ Josephson junctions are attached. A simple dynamical model in which to use these parameters is the Lindblad master equation:
\begin{equation}
\dot{\rho}=-\frac{i}{\hbar} [H,\rho]+\sum_{m=1}^M \kappa_m \mathcal{D}[\hat{a}_m](\rho),
    \label{eq:lindblad}
\end{equation}
using the definition ${\cal D}[A](\rho)=A\rho A^{\dagger}-\frac{1}{2}\{ A^{\dagger} A, \rho \}$ for any operator $A$. Here $H$ is the Hamiltonian such as in Eq.~\eqref{eq:h_nm} and $\hat{a}_m$ is the annihilation operator of mode $m=1,\ldots, M$ of the network.

Such a master equation implicitly assumes that the noise is Markovian: the environment which introduces the noise is sufficiently large to be unchanged by the weak interaction with the system whose internal dynamics is governed by $H$. If this approximation is not warranted, one has to use dynamical equations which include both system and environment generally. In some cases one can build models in which the environment is entirely classical, but non-Markovian, characterized by long temporal correlations, such as in random telegraph noise. In the Fourier domain, noise with long temporal correlations has strong low-frequency components. Generically, such noise is called $1/f$ noise, as its amplitude is often seen to scale as $1/f$. Examples are the noise due to an environment of so-called two-level defects, noise due to quasi-particle dynamics, or other temporally-correlated flux or charge noise. For temporally-correlated noise, one does not typically get an exponential decay of qubit coherence, but decay of the form $\exp(-\gamma t^{\alpha})$ with $\alpha > 1$; a power-law decay is also possible, see e.g. \cite{oliver:juelich, popp:solomon}.

In the master equation above, only energy {\em loss} is modeled, but if the frequency of the mode is relatively low compared to the dilution fridge temperature $T \approx 20 \, \mathrm{mK}$, one should include both loss and gain, which leads to
\begin{equation}
\dot{\rho}=-\frac{i}{\hbar} [H,\rho]+\sum_{m=1}^M \kappa_{m,-} \mathcal{D}[\hat{a}_m](\rho)+\sum_{m=1}^M \kappa_{m,+} \mathcal{D}[\hat{a}_m^{\dagger}](\rho).
    \label{eq:lindT}
\end{equation}
with $\kappa_{m,-}=\kappa_m(1+\overline{n}(\omega_m))$ and $\kappa_{m,+}=\kappa_m \overline{n}(\omega_m)$ with thermal occupation number
\begin{align}
\overline{n}(\omega_m)=\frac{1}{e^{\beta \hbar\omega_m}-1}, 
\label{eq:ntherm}
\end{align}
and thermodynamic beta $\beta = \frac{1}{k_B T}$.
Even though for a transmon qubit at, say, $5 \, \mathrm{GHz}$, $\overline{n}$ is tiny from a thermal perspective, additional sources of excitations are known to exist, see Exercise \ref{exc:temp}.

Now imagine that there is a single Josephson junction $N=1$ and hence we want to consider the decoherence of the qubit associated with that Josephson junction. Then Eq.~\eqref{eq:lindT} does not look like a very direct and efficient way of doing this, as the noise comes in through the $M$ modes to which the qubit is coupled. A more direct model starts from a qubit circuit connected to a lossy circuit branch and will be considered in Section~\ref{sec:losscircuit} and Appendix~\ref{app:lcys}. 

In the remainder of this section we discuss some aspects of a generic qubit Lindblad equation ---assuming it can be derived to model noise--- through two exercises.

A qubit version of such Lindblad equation in Eq.~\eqref{eq:lindT} with $H=-\frac{\hbar\Omega}{2} Z$ reads
\begin{equation}
\dot{\rho}=-\frac{i}{\hbar} [H, \rho]+\kappa_- \mathcal{D}[\sigma^-](\rho)+\kappa_+ \mathcal{D}[\sigma^+](\rho),
    \label{eq:lindQ}
\end{equation}
with $\kappa_+=\kappa_-e^{-\beta \hbar \Omega}$.
Solving the dynamics of the qubit Lindblad equation is straightforward, see Exercise \ref{exc:lind}. The qubit state $\rho(t)$ will equilibrate to the Gibbs state $\rho_{\beta}=e^{-\beta H}/{\rm Tr}(e^{-\beta H})$ at sufficiently long time $t$. The relaxation time $T_1$ sets the rate of decay of the diagonal elements $\bra{0} (\rho(t)-\rho_{\beta}) \ket{0} \propto e^{-t/T_1}$ where $T_1$ can be related to $\kappa_{\pm}$ as $1/T_1=\kappa_++\kappa_-$. The off-diagonal elements of $\rho(t)$ will decay as $\bra{0} \rho(t) \ket{1} \propto e^{-t/T_2}$ with $T_2=2T_1$.  
To model other sources of (fast, Markovian) dephasing than just dephasing due to relaxation, one can introduce a pure dephasing time $T_{\phi}$ defined through
\begin{align}
    \frac{1}{T_2}=\frac{1}{2T_1}+\frac{1}{T_{\phi}},
\end{align}
where $T_2$ can be taken as the `experimentally measured' dephasing time. In the qubit Lindblad equation, one can then model such additional dephasing by including an additional pure dephasing term $\gamma \mathcal{D}[Z]$ with $\gamma=\frac{1}{2T_{\phi}}$. In practice, the measured $T_1$ and $T_2$ times of superconducting qubits ---$T_2$ is often measured with and without applied echo or dynamical decoupling pulses which remove slow-varying noise components--- fluctuate considerably over time, and when a transmon is flux-tuned to different frequencies; see e.g. \cite{bylander:bench} and studies on the impact of cosmic rays \cite{thorbeck:cosmic} and references there. Clearly, much more detailed physical models are necessary to capture the non-equilibrium effects at varying time scales \cite{FC:quasip}.

\begin{Exercise}[title={Dephasing and relaxation in the qubit Lindblad equation},label=exc:lind]
\Question Observe that $\sigma^-=\ket{0}\bra{1}=\frac{1}{2}(X+iY)$.
Determine $\frac{d}{dt}\left({\rm Tr} \, \sigma^- \rho(t)\right)={\rm Tr}\, \sigma^- \dot{\rho}(t)$ with $\rho(t)$ the solution of the Lindblad equation in Eq.~\eqref{eq:lindQ}. Then solve for ${\rm Tr}\,\sigma^-(\rho(t)$ given ${\rm Tr}\,X \rho(t=0)=1$ and ${\rm Tr}\,Y \rho(t=0)=0$, i.e. $\rho(t=0)=\ket{+}\bra{+}$ . Verify that $\frac{2}{T_2}=\kappa_++\kappa_-$.
\Question Include the additional dephasing term $\gamma \mathcal{D}[Z]$ in the Lindblad equation and verify how $T_{\phi}$ relates to $\gamma$ by again examining $\frac{d}{dt}\left({\rm Tr} \, \sigma^- \rho(t)\right)$.
\Question Solve for ${\rm Tr}\,\ket{0}\bra{0} \rho(t)$ as well, confirming that $T_2=2T_1$.
\end{Exercise}

\begin{Answer}[ref={exc:lind}]
\Question We have 
\[
{\rm Tr}\, X \dot{\rho}(t)=\frac{i\Omega}{2}
{\rm Tr}\, X [Z,\rho(t)]-\frac{1}{2}(\kappa_-+\kappa_+){\rm Tr}\,X \rho(t)=
 \Omega {\rm Tr}\,Y \rho(t)-\frac{1}{2}(\kappa_-+\kappa_+){\rm Tr} \,X \rho(t),
\]
and
\[
{\rm Tr}\, Y \dot{\rho}(t)=\frac{i \Omega}{2}
{\rm Tr}\, Y [Z,\rho(t)]-\frac{1}{2}(\kappa_-+\kappa_+){\rm Tr} Y \rho(t)=
-\Omega {\rm Tr}X \rho(t)-\frac{1}{2}(\kappa_-+\kappa_+){\rm Tr}Y \rho(t),
\]
or
\[
\frac{d}{dt}\left({\rm Tr}\, \sigma^- \rho(t)\right)=(-i \Omega-\frac{1}{2}(\kappa_-+\kappa_+)){\rm Tr}\,(\sigma^- \rho(t)),
\]
which gives the solution ${\rm Tr}\,\sigma^- \rho(t)=\frac{1}{2}e^{-i\Omega t-(\kappa_++\kappa_-)t/2}$ with the given initial condition, showing the exponential decay with $T_2=\frac{2}{\kappa_++\kappa_-}$.
\Question We have $\mathcal{D}[Z](\rho(t))=Z\rho Z-\rho$ and thus
$\gamma{\rm Tr} X \mathcal{D}[Z](\rho)=-2\gamma {\rm Tr} X\rho(t)$ and the same for $\gamma{\rm Tr} \, Y \mathcal{D}[Z](\rho)=-2\gamma {\rm Tr} Y \rho(t)$. This gives an additional decay $e^{-2 \gamma t}\equiv e^{-t/T_{\phi}}$ so $T_{\phi}=\frac{1}{2\gamma}$.
\Question We have 
\[
\frac{d}{dt}({\rm Tr} \ket{0}\bra{0} \rho(t))=\kappa_- {\rm Tr}\ket{1}\bra{1}\rho(t)-\kappa_+ {\rm Tr} \ket{0}\bra{0} \rho(t)=\kappa_--(\kappa_++\kappa_-) {\rm Tr}\ket{0}\bra{0} \rho(t),
\]
or ${\rm Tr} \ket{0}\bra{0} \rho(t)=C e^{-(\kappa_++\kappa_-)t}+\frac{\kappa_-}{\kappa_++\kappa_-}$ for some constant $C$, confirming that $\frac{1}{T_1}=\kappa_++\kappa_-$.  Note that ${\rm Tr}\ket{0}\bra{0}\rho_{\beta}=\frac{\kappa_-}{\kappa_++\kappa_-}$.
\end{Answer}

\begin{Exercise}[title={Thermal equilibration?},label=exc:temp]
\Question A transmon qubit at $f=\Omega/2 \pi= 5$GHz is left alone to idle and we imagine that it equilibrates to the temperature of the dilution fridge which is kept at 20mK. What is the probability to find the transmon qubit in $\ket{0}$ after it equilibrates? Do the analysis both by treating the transmon as a qubit as well as a purely harmonic system (what is the difference in answers?). 
You can use that $50 {\rm mK} \approx 1 {\rm GHz}$. When you compare your answers with an experimentally estimated probability to find the transmon qubit in $\ket{0}$ - see e.g. Fig.~1 in \cite{RBLD:feedback} where the ground-state occupancy is estimated as $80\%$ - you see that this thermal model is far too simple: there are other sources of noise which make that the qubit is not measured as $\ket{0}$.
\Question Same question for a fluxonium qubit, treating it as a qubit at $f=\Omega/2\pi=300$MHz. What does this tell you about fluxonium qubit initialization? 
\end{Exercise}

\begin{Answer}[ref={exc:temp}]
\Question We imagine that the qubit equilibrates to the Gibbs state $\rho=\frac{e^{-\beta H}}{{\rm Tr} (e^{-\beta H})}$ with $\beta=\frac{1}{k_B T}$. If we treat the transmon as a qubit, we have $H=-\frac{\hbar \Omega}{2}Z$, $\hbar \Omega= hf$ and $\mathbb{P}(0)=\bra{0} \rho \ket{0}=\frac{1}{1+e^{-\beta \hbar \Omega}}=\frac{1}{1+e^{-5/(4/5)}} \approx 0.998$. If we treat the transmon as a harmonic system with $H=\hbar \Omega(\hat{n}+\frac{1}{2})$, we have $\mathbb{P}(0)=1-e^{-\beta \hbar \Omega}=1-e^{-5/(4/5)} \approx 0.998$ which is identical within this precision. 
\Question For the fluxonium qubit $\mathbb{P}(0)=\bra{0} \rho \ket{0}=\frac{1}{1+e^{-(3/10)/(4/5)}}=\frac{1}{1+e^{-3/8})} \approx 0.593$. Due to its low frequency, a fluxonium qubit needs to be actively measured (reset) to initialize the qubit to $\ket{0}$, as its thermal population in $\ket{1}$ is rather high.
\end{Answer}

\section{A qubit circuit with a lossy circuit branch}
\label{sec:losscircuit}

We imagine that in one of the qubit circuits we have been discussing in this book, there is also some lossy element, which can be described as an additional circuit branch between the qubit nodes  with admittance $\mathcal{Y}(\omega)=1/\mathcal{Z}(\omega)$. Since it is lossy, $\mathcal{Y}$ and $\mathcal{Z}$ will have non-zero real parts, as we have seen in the single-port discussion in Section~\ref{sec:loss}. This branch may be just a simple resistor $R$, but it can also be composed of several circuit elements. 
We will confine ourselves to the case where ${\rm Re}(\mathcal{Y})$ is small (so that $\mathcal{Z}$, or $R$, is large); this makes sense because we will only attempt to make a qubit in cases where the loss is small, so that the lossy branch is only a small perturbation on the circuit.

We are thus considering a scenario such as in Fig.~\ref{fig:rescirc} where an LC oscillator, or any other type of qubit like a transmon or flux qubit, is weakly capacitively coupled to such a resistive element. This resistive element can also physically be an (unobserved) transmission line. 

We observe that the `lossy Foster' analysis in Section~\ref{sec:loss} in which each LC oscillator has a resistance in parallel, does not fit the set up in Figure \ref{fig:rescirc} very well. The claim of the lossy Foster circuit is that one can write any $Z(s)$, for example $Z(s)$ in parallel with the LC oscillator in Fig.~\ref{fig:rescirc}, as a series of RLC circuits (see Fig.~\ref{fig::cauer_circuit} with $N=1$). However, in this case this is a somewhat singular representation, as we are missing some capacitors and inductors. A more direct approach, explained in Appendix~\ref{app:lcys} for an LC oscillator in parallel with an arbitrary, but small admittance, is more suitable. 
The physics approach to model the decoherence of the qubit is 
to treat the resistive element as a continuum bath of oscillators, with a certain spectral density, in a thermal state at a given temperature, and consider the system-bath Hamiltonian. This is a so-called Caldeira-Leggett model originally studied in Ref.~\cite{caldeiraleggett} (see also Refs.~\cite{Vool2016, parra2022} for the application to electrical circuits).  
Then to such a full system + bath Hamiltonian, one can apply a Born-Markov analysis \cite{petruccione} through which one obtains a Lindblad master equation for the qubit system only, treating the bath as a large, unperturbed, weakly-coupled environment. The dissipative terms in this Lindblad master equation which drive the qubit to a thermal state then set the relaxation time $T_1$ as in Eq.~\eqref{eq:lindQ}. For example, without going through the derivation, we quote a result from \cite{nguyen2019}, Appendix 3 (see its derivation as Eq.~(131) of \cite{BKD:circuit} and Section~$3$ of \cite{schoelkopf2003qubits}) for the decay rate:
\begin{equation}
    T_1^{-1}=\frac{1}{\hbar}|\langle 0|\hat\phi_i-\hat\phi_j|1\rangle|^2\left( \frac{\Phi_0}{2\pi}\right)^2 \Omega\mbox{Re}(\mathcal{Y}(\Omega))\left(\coth \frac{\beta \hbar\Omega}{2}+1\right)\label{BKDT1:alt}
\end{equation}
where $\hbar \Omega$ is the energy splitting between the state $\ket{0}$ and $\ket{1}$ of the qubit.
The nodes $i$ and $j$ are those between which the lossy impedance (with admittance $\mathcal{Y}(\omega)$) is attached as a branch.
Note that $\coth{\beta \hbar\Omega/2}=1+2\overline{n}(\Omega)$ with $\overline{n}(\Omega))$ with thermal occupation number as in Eq.~\eqref{eq:ntherm}. We see appearing here the operators for the dimensionless node fluxes of nodes $i$ and $j$; if one node is ground, the corresponding operator is simply omitted. 

Eq.~\eqref{BKDT1:alt} also holds when several qubit circuits in series are inserted between these nodes $\phi_i$ and $\phi_j$ (and it could be that the qubit wavefunctions are largely independent of $\phi_i$ and $\phi_j$). But it is worthwhile to consider the application of this equation to the very simplest example, that of a single LC circuit shunted by a resistance $R$. The evaluation of the matrix element proceeds using Eqs.~\eqref{eq:chargeflux_aadag},\eqref{phidef}, yielding
\begin{equation}
\langle 0|\hat\phi|1\rangle=\frac{2\pi}{\Phi_0}\sqrt{\frac{\hbar}{2}\sqrt{\frac{L}{C}}}.
\end{equation}
Plugging this into Eq.~(\ref{BKDT1:alt}), with the additional replacements $\Omega \mapsto \omega_r=1/\sqrt{LC}$ and $\mbox{Re}(\mathcal{Y}(\omega_r))=1/R$, and in the low-temperature limit ($\coth(\cdot)+1\rightarrow 2$), we get the simple result
\begin{equation}
\label{eq:t1rc}
T_1^{-1}=\frac{1}{RC}.
\end{equation}
This straightforward identification of $T_1$ with ``RC decay" is discussed further in \cite{PhysRevLett.101.080502}. In Appendix~\ref{app:lcys} we show that this identification holds more generally for LC circuits in parallel with a weakly perturbing admittance (see Eq.~\eqref{eq:t1ys}). 

It is illuminating to collect the constants differently in Eq.~(\ref{BKDT1:alt}) so that it is rewritten as
\begin{equation}
    T_1^{-1}=|\langle 0|\hat\phi_i-\hat\phi_j|1\rangle|^2
    \frac{\Omega}{2\pi} R_Q
    \mbox{Re}(\mathcal{Y}(\Omega))\left(\coth\frac{\beta\hbar\Omega}{2}+1\right).\label{BKDT1simp}
\end{equation}
where $R_Q$ is the quantum of resistance in Eq.~\eqref{eq:resist-quant}. This rewrite makes the units of the expression very transparent, and it makes clear that one aspect of $T_1^{-1}$ being ``small" is that $\mbox{Re}\,{\mathcal Y}(\Omega)\ll R_Q^{-1}$.

\begin{figure}
\centering
\begin{subfigure}[t]{1.0 \textwidth}
\centering
\includegraphics[height=4cm]{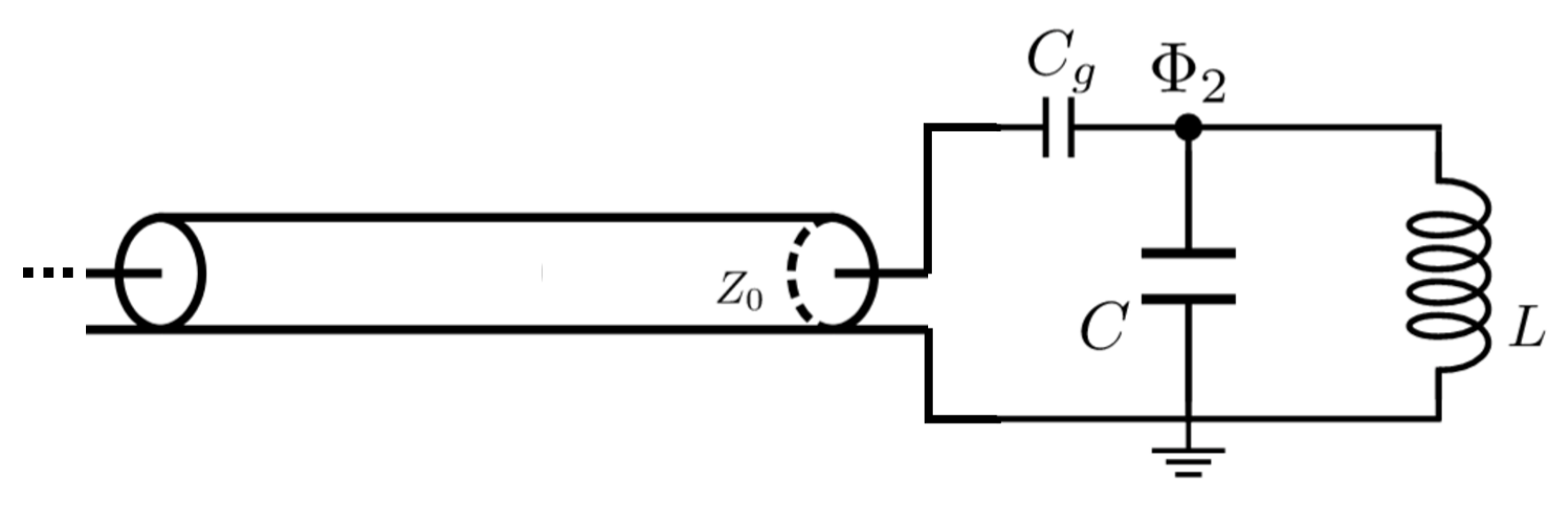}
\subcaption{}
\label{fig:rescirc_a}
\end{subfigure}
\begin{subfigure}[t]{1.0 \textwidth}
\centering
\includegraphics[height=4cm]{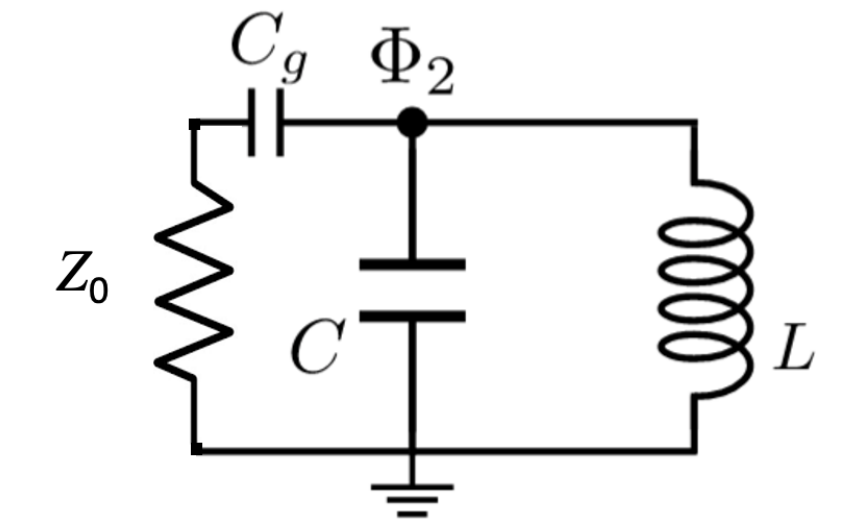}
\subcaption{}
\label{fig:rescirc_b}
\end{subfigure}
\caption{An unobserved transmission line (not driven) with a characteristic impedance $Z_0$ in (a) can be replaced by a resistor with resistance $Z_0$ in (b).}
\label{fig:rescirc}
\end{figure}

\begin{Exercise}[title={Qubit relaxation caused by a long transmission line},label=exc:transloss]
In Fig.~\ref{fig:driven_LC} we showed a simple resonant circuit driven by a voltage source. The real situation is usually more like that shown in Fig.~\ref{fig:rescirc}(a). The voltage source is actually many meters away, connected to the experiment by a long length of transmission line (cf.~Fig.~\ref{fig::2PTL}). As seen by the resonant circuit, it is well known that an infinite transmission line with wave impedance $Z_0$ has the same action as an actual resistor of that same resistance $Z_0$ (see e.g. \cite{pozar}). Of course, the transmission line is not infinite, but there is in reality, after about a meter of transmission line, a filter with a resistor with resistance $Z_0$ in it. So the situation is well represented by Fig.~\ref{fig:rescirc}(b). The actual resistor is still in the coldest part of the cryostat, so we can say $T=20 \, \mathrm{mK}$. Many aspects of this problem resemble the analysis given in Ref.~\cite{catpara:resist}, where an experiment is also presented.
Using the temperature given and the following parameters: $Z_0=50 \, \mathrm{\Omega}$, $E_C/h=250 \, \mathrm{MHz}$ (see Eq.~\eqref{eq:char_en}), $E_L/h=12.5 \, \mathrm{GHz}$ (see Eq.~\eqref{eq:ind_en}), and $C_g/C=0.005$ (reasonable, following \cite{doi:10.1063/1.3010859}), compute $T_1$, and compare it with the period of oscillation of the resonator. Treat the RC series combination on the left of Fig.~\ref{fig:rescirc}(b) as a single composite branch.
\end{Exercise}

\begin{Answer}[ref={exc:transloss}]
Note that $\omega_r=1/\sqrt{LC}$. The relevant quantities for our lossy branch are
\begin{eqnarray}
\mathcal{Z}(\omega_r)&=&Z_0+\frac{1}{i\omega_r C_g}, \notag\\
\mathcal{Y}(\omega_r)&=&\frac{i\omega_r C_g}{1+i\omega_r C_gZ_0}, \notag \\
\mbox {Re}(\mathcal{Y}(\omega_r))&=&\frac{\omega_r^2 C_g^2Z_0}{1+\omega_r^2 C_g^2Z_0^2}, \notag \\
\mbox {Re}(\mathcal{Y}(\omega_r))&=&\frac{C_g^2Z_0/CL}{1+ C_g^2Z_0^2/CL}.
\end{eqnarray}
Note that the factor in the denominator of the last equation can be expressed
\begin{equation}
 \frac{Z_0^2C_g^2}{LC}=Z_0^2\frac{C}{L}\left(\frac{C_g}{L}\right)^2=\frac{(2\pi)^2}{8}\left(\frac{Z_0}{R_Q}\right)^2\left(\frac{C_g}{C}\right)^2\frac{E_L}{E_C}\ll 1.
\end{equation}
This is true even though $E_L>E_C$, because the other factors are very small. Thus, also noting the magnitude of the temperature, we can approximate
\begin{equation}
\mbox {Re}({\mathcal Y}(\omega_r))\approx \frac{C_g^2Z_0}{CL}=\frac{(2\pi)^2}{8}\frac{1}{R_Q}\frac{Z_0}{R_Q}\left(\frac{C_g}{C}\right)^2\frac{E_L}{E_C},\,\,\,\,\mbox{ and }\,\,\,\,\coth{\beta \hbar\omega_r/2}\approx 1
\end{equation}
To evaluate the matrix element, we use the harmonic oscillator result from Eq.~\eqref{eq:zpf-reduced}, so that
\begin{equation}
   |\langle 0|\hat\phi_2|1\rangle|^2=\sqrt{\frac{2E_C}{E_L}}. 
\end{equation}
Collecting everything together, using Eq.~\eqref{BKDT1simp} gives
\begin{eqnarray}
 T_1^{-1}&=&\sqrt{\frac{2E_C}{E_L}}\cdot \frac{\omega_r}{2\pi} \cdot R_Q\cdot\frac{(2\pi)^2}{8}\frac{1}{R_Q}\frac{Z_0}{R_Q}\left(\frac{C_g}{C}\right)^2\frac{E_L}{E_C}\cdot 2 \notag \\
 &=&\frac{(2\pi)^2}{2\sqrt{2}}\sqrt{\frac{E_L}{E_C}}\frac{Z_0}{R_Q}\left(\frac{C_g}{C}\right)^2 \frac{\omega_r}{2\pi} \notag \\
 &=&0.00019\cdot 6 {\rm GHz}\approx 1.1{\rm MHz}.
\end{eqnarray}
This gives a reasonable, but not state of the art, decay rate of about $1 \, \mathrm{\mu s}$. Longer lifetimes are achieved by using smaller values of the coupling capacitance $C_g$, and sometimes by the use of a ``Purcell filter'', which strongly reduces ${\rm Re}(\mathcal{Y})$ at the frequency $\omega_r$.

\end{Answer}

\section{Noise sensitivity and protection}
 \label{sec:ns}

Imagine we wish to engineer a Hamiltonian $H_{\rm ideal}$ of a single superconducting qubit, using the conjugate operators $\hat{\phi}$ and $\hat{q}$. In practice, this system is coupled to other systems, quantum or classical, whose dynamics affects the system introducing noise, and let the Hamiltonian including the coupling to the environment be $\tilde{H}$. One can roughly classify noise as either charge noise or flux noise, meaning that it couples to either the charge variable $\hat{q}$ or the flux variable $\hat{\phi}$ in $\tilde{H}$. 

Both flux noise and charge noise can contribute to {\em both} qubit dephasing and qubit relaxation: what noise is dominant depends on coupling strengths and the character and symmetry of the qubit states. A simple model for the rate $1/T_1$ is given by Fermi's golden rule which says that
\begin{equation}
    \frac{1}{T_1}=\frac{2\pi}{\hbar}|\bra{0} \tilde{H} \ket{1}|^2 \rho(\Omega),
    \label{eq:fermi}
\end{equation}
where $\rho(\Omega)$ is the density of states of the environmental degrees of freedom at the qubit frequency $\Omega$. One can compare this expression with Eq.~\eqref{BKDT1:alt}, where the coupling is expressed in terms of the branch flux operator, the admittance $\mathcal{Y}(\Omega)$ and the thermal properties of the environment at frequency $\Omega$.
Clearly, whether flux and charge noise contribute to relaxation via Eq.~\eqref{eq:fermi} depends on whether terms which involve the charge $\hat{q}$ and/or flux operator $\hat{\phi}$ in $\tilde{H}$ induce transitions between $\ket{0}$ and $\ket{1}$.

\subsubsection{Noise sensitivity}

As we know, the transmon qubit is designed such that it is less sensitive to charge noise. This is expressed in the plots in Fig.~\ref{fig:energyng} where one sees a reduced sensitivity towards noise in the offset charge $n_g$ of the energies of all eigenstates $\ket{\psi_k}$. Let $\tilde{H}$ be the Hamiltonian when the offset charge $n_g$ is slightly different than some ideal value $n_{g}^{\rm ideal}$, i.e., $\tilde{H}(n_g=n_g^{\rm ideal})=H_{\rm ideal}$ and so the insensitivity is expressed as
\begin{equation}
   \bra{\psi_m} \tilde{H} \ket{\psi_m}\approx \bra{\psi_m} H_{\rm ideal} \ket{\psi_m},
\end{equation}
where $\ket{\psi_m}$ are eigenstates of $H_{\rm ideal}$. For the transmon this comes about due to the weak dependence on $n_g$ in $E_m(n_g)$ in Eq.~\eqref{eq:mathieu} for large $E_J/E_C$. This insensitivity lengthens the dephasing time ---limiting the $Z$ error strength--- of the qubit due to charge noise, since 
it implies that
\begin{equation}
   \bra{0} \tilde{H} \ket{0}-\bra{1} \tilde{H} \ket{1} \approx \bra{0} H_{\rm ideal} \ket{0}-\bra{1} H_{\rm ideal} \ket{1}.
\end{equation}
If a qubit is insensitive to dephasing errors, one expects that the dephasing $T_2$ is long and dominated by relaxation-induced dephasing, i.e., $\frac{1}{T_2}\gtrapprox \frac{1}{2T_1}$.

Observe that a different question is whether the matrix element $|\bra{0}\tilde{H}\ket{1}|$ is small. Note that by definition one has $|\bra{0}H_{\rm ideal}\ket{1}|=0$ since $\ket{0},\ket{1}$ are the ideal eigenstates, and one thus considers the effect of the (time-varying, stochastic) perturbation in $\tilde{H}$ which couples the qubit to the environment. Just from the fact that the transmon is in the regime $E_C \ll E_J$ and any charge noise comes with energy scale $E_C$, charge noise causing transitions between $\ket{0}$ and $\ket{1}$ ---bit flips $X$--- will be suppressed. For these general reasons, a flux-type qubit where the contributions for capacitance terms scaling with $E_C$ are relatively small is naturally protected against charge noise.

\subsection{Flux sweet spots}
\label{subsec:fluxss}

Any loop in a superconducting lumped-circuit representation of a device through which a current can run is affected by (stray, fluctuating) magnetic fields and can lead to flux noise. 

A common form of protection against classical flux noise is obtained when the ideal Hamiltonian $H_{\rm ideal}$ operates at a so-called external flux sweet spot. Imagine setting the potential $U(\phi,\phi_{\rm ext})$ to the point $\phi_{\rm ext}=\phi_{\rm ext}^{\rm ideal}$.  Given a sufficiently small amount of noise on the external flux $\phi_{\rm ext}$, we can then write
\[
\tilde{H}(\phi_{\rm ext})=H_{\rm ideal}+(\phi_{\rm ext}-\phi_{\rm ext}^{\rm ideal})\frac{\partial U}{\partial \phi_{\rm ext}}\biggl\vert_{\phi_{\rm ext}^{\rm ideal}}+\frac{1}{2}(\phi_{\rm ext}-\phi_{\rm ext}^{\rm ideal})^2 \frac{\partial^2 U}{\partial\phi_{\rm ext}^2}\biggl\vert_{\phi_{\rm ext}^{\rm ideal}}+\ldots
\]
Now we can consider how fluctuations in $\phi_{\rm ext}$ affect the qubit frequency, i.e.
\begin{multline*}
\bra{0} \tilde{H}(\phi_{\rm ext})\ket{0}-\bra{1} \tilde{H}(\phi_{\rm ext}) \ket{1}=\hbar \Omega^{\rm ideal} \notag \\
+(\phi_{\rm ext}-\phi_{\rm ext}^{\rm ideal})\left(\bra{0}\frac{\partial U}{\partial\phi_{\rm ext}}\biggl\vert_{\phi_{\rm ext}^{\rm ideal}}\ket{0}-\bra{1} \frac{\partial U}{\partial\phi_{\rm ext}}\biggl\vert_{\phi_{\rm ext}^{\rm ideal}} \ket{1}\right) +\mathcal{O}((\phi_{\rm ext}-\phi_{\rm ext}^{\rm ideal})^2),
\end{multline*}
where $\Omega^{\rm ideal}$ is the targeted transition frequency.  
We say that $\phi_{\rm ext}^{\rm ideal}$ is a flux sweet spot when 
\begin{equation}
\bra{0}\frac{\partial U}{\partial\phi_{\rm ext}}\biggl\vert_{\phi_{\rm ext}^{\rm ideal}}\ket{0}-\bra{1} \frac{\partial U}{\partial \phi_{\rm ext}}\biggl\vert_{\phi_{\rm ext}^{\rm ideal}} \ket{1}=0.
\label{eq:ss}
\end{equation}
A common case is when $U$ depends on $\cos(\phi-\phi_{\rm ext}^{\rm ideal})$ with $\phi_{\rm ext}^{\rm ideal}=k \pi$ for $k\in \mathbb{Z}$, such that
$\bra{\psi}\frac{\partial U}{\partial \phi_{\rm ext}}\big\vert_{\phi_{\rm ext}^{\rm ideal}} \ket{\psi} \propto \bra{\psi} \sin(\phi)\ket{\psi}=0$ for any eigenstate $\ket{\psi}$ due to the parity symmetry of the ideal Hamiltonian, as discussed in Section~\ref{sec:sym-prot}. Thus $\phi_{\rm ext}^{\rm ideal}=k \pi$ are then flux sweet spots where one expects the longest $T_2$ time.

\begin{Exercise}[label=exc:ss]
Show that small fluctuations in $\phi_{\rm ext}$ at a flux sweet spot, for which Eq.~\eqref{eq:ss} holds, can still generate bit-flip ($X$) errors and thus lower $T_1$, i.e.,
\begin{equation}
|\bra{0}\tilde{H}(\phi_{\rm ext})\ket{1}|,
\label{eq:offdiag}
\end{equation}
is not necessarily small when $U$ is parity-symmetric at $\phi_{\rm ext}^{\rm ideal}$, $U(\phi,\phi_{\rm ext}=\phi_{\rm ext}^{\rm ideal})= U(-\phi,\phi_{\rm ext}=\phi_{\rm ext}^{\rm ideal})$ {\em and} we wish to be able to drive the $0-1$ transition of the qubit.
\end{Exercise}

\begin{Answer}[ref=exc:ss]
To first order we have 
\begin{equation}
|\bra{0}\tilde{H}(\phi_{\rm ext})\ket{1}|=(\phi_{\rm ext}-\phi_{\rm ext}^{\rm ideal})
|\bra{0}\frac{\partial U}{\partial \phi_{\rm ext}}\biggl\vert_{\phi_{\rm ext}^{\rm ideal}}\ket{1}|.
\end{equation}
Since $H_{\rm ideal}$ is parity-symmetric, $\frac{\partial U}{\partial \phi_{\rm ext}}$
is an odd function $f_{\rm odd}(-\phi)=-f_{\rm odd}(\phi)$. The states $\ket{0}$ and $\ket{1}$ should have different parity symmetry to be able to drive the qubit (due to Eqs.~\eqref{eq:Q0} and \eqref{eq:sym0-2}), hence there is no {\em symmetry} or sweet-spot reason why Eq.~\eqref{eq:offdiag} should be small.
\end{Answer}

\subsection{Built-in protection}
\label{sec:bip}

Here we discuss forms of noise protection which are based on shaping the form of the wavefunction as in a quantum error-correcting code. For a general perspective on protected qubits we refer to Ref.~\cite{gyenis2021moving}.

One can generally write 
\[
\tilde{H}=H_{\rm ideal}+f(\hat{\phi})+g(\hat{q}),
\]
where $f(.)$ and $g(.)$ are some functions which can also depend on additional quantum or classical environment degrees of freedom, or be time-dependent. For the purpose of the next arguments, the form of this dependence is not relevant, so we don't specify it exactly.

Imagine we define a qubit with $\psi_0(\phi)$ and $\psi_1(\phi)$ the wavefunctions of the state $\ket{0}$ resp. $\ket{1}$, which are eigenstates of $H_{\rm ideal}$. If the wavefunctions $\psi_0(\phi)$ and $\psi_1(\phi)$ have approximately {\em disjoint} support, then 
\begin{equation}
    \bra{0} f(\hat{\phi}) \ket{1}=\int_{-\infty}^{\infty} d\phi\, \psi_0^*(\phi) f(\phi) \psi_1(\phi) \approx 0,
    \label{eq:disjoint-1}
\end{equation}
Similarly, we have for the function $g(\hat{q})=\hat{q}^k$ for any finite $k$:
\begin{equation}
    \bra{0} g(\hat{q}) \ket{1}=\int_{-\infty}^{\infty} d\phi\, \psi_0^*(\phi) \left(-i \frac{\partial}{\partial\phi}\right)^k \psi_1(\phi) \approx 0.
    \label{eq:disjoint-2}
\end{equation}
Note, however, that if $g(\hat{q})=e^{i \epsilon \hat{q}}$, that is, a displacement of sufficient magnitude $\epsilon$, then the support of $\psi_0(\phi)$ could be displaced to that of $\psi_1(\phi)$ and $\bra{0} g(\hat{q}) \ket{1}\neq 0$. However, we don't naturally expect such displacement terms in the Hamiltonian, since $\hat{q}$ represents the charge operator. 

Eqs.~\eqref{eq:disjoint-1} and \eqref{eq:disjoint-2} imply that the noisy Hamiltonian $\tilde{H}$ cannot cause any transition from $\ket{0}$ to $\ket{1}$ and vice versa, leading to the qubit being protected against bit-flip $X$ errors induced by charge or flux noise. In other words, energy exchange with the environment inducing $\ket{0} \leftrightarrow \ket{1}$ is suppressed, as we cannot couple into these transitions and one expects that the qubit relaxation time $T_1$ is very long for this reason.

An example of a qubit with such `built-in' $T_1$-protection \footnote{Of course, if we happen to call the states $\ket{0}$ and $\ket{1}$, $\ket{+}$ and $\ket{-}$ instead, then one would call this built-in protection against phase-flip $Z$ errors, leading to a long dephasing time.} is the $0$-$\pi$ qubit discussed in Section~\ref{subsec:0pi}, as the wavefunctions for $\ket{0}$ and $\ket{1}$ are localized in the well around $\phi=0$ and the well around $\phi=\pi$. It is clear that the requirement of having disjoint support requires the potential well structure to be sufficiently broad. For a qubit like the transmon, the wavefunctions definitely do not have disjoint supports (such as approximately two lowest-energy wavefunctions of the harmonic oscillator which involve Hermite functions with index 0 and 1).

\subsection{Double protection?}

We have seen that disjoint support of wavefunctions induces a protection against $X$ errors, but is it possible to shield a qubit from both $X$ as well as $Z$ errors in this manner? 

At first sight, this seems difficult as it would require that 
\begin{align}
    \bra{+} f(\hat{\phi}) \ket{-}=\int_{-\infty}^{\infty} d\phi\, \psi_+^*(\phi) f(\phi) \psi_-(\phi) \overset{?}{\approx} 0, \label{eq:trouble} \\
   \bra{+} g(\hat{q}) \ket{-}=\int_{-\infty}^{\infty} d\phi\, \psi_+^*(\phi) g\left(-i \frac{\partial}{\partial\phi}\right) \psi_-(\phi) \overset{?}{\approx} 0, \label{eq:fine}
\end{align}
where $\psi_{\pm}(\phi) \approx (\psi_0(\phi)\pm \psi_1(\phi))/\sqrt{2}$ are the wavefunctions of the states $\ket{\pm}=(\ket{0}\pm \ket{1})/\sqrt{2}$. It is clear that the states $\psi_{\pm}(\phi)$ generally do {\em not} have disjoint support, as they are superpositions. But still Eqs.~\eqref{eq:trouble} and \eqref{eq:fine} can be zero for another reason, as follows.  We can examine the $\ket{\pm}$ states using the wavefunctions $\tilde{\psi}_{\pm}(q)$ in the conjugate $\hat{q}$-basis. In fact, it is possible to design wavefunctions $\psi_{0,1}(\phi)$ such that the wavefunctions $\tilde{\psi}_{\rm \pm}(q)$ also have nonoverlapping support: this type of qubit is called a GKP qubit after its inventors Gottesman, Kitaev and Preskill in 2001 \cite{Gottesman.etal.2001:GKPcode}. For the GKP qubit, $\psi_0(\phi)$ is a wavefunction sharply peaked at $\phi=2k\alpha\sqrt{\pi}$
for $k\in \mathbb{Z}$, while $\psi_1(\phi)$ is supported near $\phi=(2k+1)\alpha\sqrt{\pi}$ for some chosen bias $\alpha$.
Fourier transforming implies that $\tilde{\psi}_+(q)$ has support on $q=2k \sqrt{\pi}/\alpha$ while $\tilde{\psi}_-(q)$ has disjoint support on $q=(2k+1) \sqrt{\pi}/\alpha$, hence with a gap of $\Delta q=\sqrt{\pi}/\alpha$ between the supports. Such disjoint support then suffices to claim that Eq.~\eqref{eq:fine} holds, i.e. we will have
\begin{equation}
    \bra{+} g(\hat{q}) \ket{-}=\int_{-\infty}^{\infty} dq\, \tilde{\psi}_+^*(q) g(q)\tilde{\psi}_-(q) \approx 0.
\end{equation}
However, {\em even} given disjoint support in the $\hat{q}$-basis, when, say, $f(\hat{\phi})\propto \cos(\hat{\phi})=(e^{i\hat{\phi}}+e^{-i\hat{\phi}})/2$, the support of $\tilde{\psi}_+(q)$ can get shifted by $\pm 1$ since
\begin{equation}
    e^{i\hat{\phi}}\ket{q}=\ket{q+1}.
\end{equation}
This implies that one should choose the bias $\alpha$ such that $\Delta q=\sqrt{\pi}/\alpha> 1$ in the presence of some additional $\cos \phi$ noise term. However, the crafting of a potential which has the GKP wavefunctions as ground states could already lead to terms which effectively act as, say, $\cos(2 \phi)$ etc.,  generating larger charge fluctuations, which require a small $\alpha$. Hence, it is not clearly possible to fulfill these conditions with the appropriate squeezing bias $\alpha$; what is possible also depends on the inductive and capacitive energy scales which play into the relative strength of the zero point fluctuation bias $q_{\rm zpf}$ versus $\phi_{\rm zpf}=\frac{1}{2q_{\rm zpf}}$ in Eq.~\eqref{eq:zpf_resc}. We refer to \cite{rymarz2021} for the construction of a Hamiltonian with approximate GKP qubit eigenstates using a gyrator. The active (driven) realization of GKP qubit states in a superconducting 3D cavity resonator was first realized in \cite{GKP:exp}.
  
  In general, we note that a built-in protection of a qubit against noise also makes the qubit harder to couple to.  Ideally, a Hamiltonian can be tuned from a protected to a less-protected regime to allow for gates to take place \cite{brooks2013}.

 When realizing a protected superconducting qubit in a system of many degrees of freedom, one can ask about how many degrees of freedom should be involved to make logical transitions from $\ket{0} \leftrightarrow \ket{1}$ and $\ket{+} \leftrightarrow \ket{-}$. For example, the M\"obius strip qubit \cite{Kitaev06}, partially discussed in Section~\ref{subsec:mobius}, can have a protection which grows with $N$ as it involves $N$ degrees of freedom to make some logical transitions. This quantum error-correcting code perspective is further developed in \cite{VCT:homo}.

\chapter{Models of superconducting amplifiers}
\label{chap:add}

In this chapter we give three `stand-alone' exercises on the topic of amplifiers. Amplifiers built from circuit QED components can be close to achieving the quantum limit in terms of adding noise. These amplifiers play an important role in the dispersive measurement of qubits, since the microwave signal which is used for the measurement is of low intensity and requires low-noise amplification before it can read as a classical voltage at room temperature. 
We refer to \cite{RD:amplifiers, CDGM:noiseRMP} for more background.

\section{Exercises on amplifiers}

\begin{Exercise}[title={Haus-Caves limit of a bosonic amplifier},label=exc:HC]
In this exercise we will derive the so-called {\em quantum limit} of a linear, phase-insensitive bosonic amplifier. A very simple model of such an amplifier is the relation between the operators of an output mode $\hat{b}$ and an input mode $\hat{a}$ as 
\begin{equation}
    \hat{b}=\sqrt{G} \hat{a}+\sqrt{G-1}\hat{c}^{\dagger}, 
    \label{eq:phase-insens}
\end{equation}
with gain $G$ and $\hat{c}^{\dagger}$ is the creation operator of a bath mode at some relevant frequency $\omega$ obeying $[\hat{c},\hat{c}^{\dagger}]=\mathds{1}$; $\hat{a}$ similarly obeys $[\hat{a},\hat{a}^{\dagger}]=\mathds{1}$. Such a relation between input and output fields could have been obtained using an actual circuit and the input-output formalism in Section~\ref{sec:HL_derive}, see Exercise \ref{exc:JRM}. We assume that the bath state is a thermal bosonic state with $\langle \hat{c} \rangle_T=0$ at all temperatures $T$. 
\Question Explain why quantum mechanics forbids the existence of an amplifier with $\hat{b}= \sqrt{G} \hat{a}$, where $G >1$, i.e., we need to have a contribution from another bath mode as in Eq.~\eqref{eq:phase-insens}. Verify that  $[\hat{b},\hat{b}^{\dagger}]=\mathds{1}$ from Eq.~\eqref{eq:phase-insens}.
\Question The reason that the amplifier is called phase-insensitive is that both input quadratures $X_1 =(\hat{a}+\hat{a}^{\dagger})/2$ and $X_2= i(\hat{a}^{\dagger}-\hat{a})/2$ (dropping hats on these new operators to simplify notation) get amplified by the same amount. Verify indeed that $\langle Y_1 \rangle^2=G \langle X_1 \rangle^2$ and similarly $\langle Y_2\rangle^2=G \langle X_2 \rangle^2$ with output quadratures $Y_1 =(\hat{b}+\hat{b}^{\dagger})/2$ and $Y_2= i(\hat{b}^{\dagger}-\hat{b})/2$. 
\Question Show that the variance of the $Y_1$ quadrature satisfies the inequality
\begin{equation}
\label{ineq}
\langle (Y_1-\langle Y_1 \rangle)^2 \rangle \equiv (\Delta Y_1)^2 \ge G (\Delta X_1)^2 + \frac{1}{4}\lvert G-1 \rvert.
\end{equation}
An analogous formula holds for $(\Delta Y_2)^2$. 
\Question  Show that the inequality in Eq.~\eqref{ineq} becomes an equality when the bath mode is in the vacuum state at $T=0$. 
{\em Comment}: A different way of expressing the Haus-Caves limit is obtained by dividing both sides of Eq.~\eqref{ineq} by $G$ such that
\[
\frac{(\Delta Y_1)^2}{G} \geq  (\Delta X_1)^2+A_N,
\]
with $A_N\geq \frac{1}{4}(1-G^{-1})$ the added noise number which is at least $1/4$. For a high-temperature bath in which thermally-added noise dominates, one can use the high-temperature approximation $k T \gg \hbar \omega$ to approximate $A_N(\omega) \approx \bar{n}(\omega) \approx \frac{kT}{\hbar \omega}$.
\end{Exercise}

\begin{Answer}[ref={exc:HC}]

\Question We know that the input and output modes are bosonic modes and so they must satisfy the bosonic commutation relation $[\hat{a}, \hat{a}^{\dagger}]=[\hat{b}, \hat{b}^{\dagger}]=\mathds{1}$. However, if $\hat{b}= \sqrt{G} \hat{a}$, then $[\hat{b}, \hat{b}^{\dagger}] = G\mathds{1}$, which forces $G=1$, i.e., there is no amplification nor de-amplification. To verify $[\hat{b}, \hat{b}^{\dagger}]=\mathds{1}$, we use the bosonic commutation relations of $\hat{a}$ and $\hat{c}$.
\Question We have $\langle Y_i \rangle=\sqrt{G} \langle X_i \rangle+\sqrt{G-1}\langle X_{i,{\rm bath}}\rangle_T=\sqrt{G} \langle X_i \rangle$ due to $\langle \hat{c} \rangle_T=0$ (using $X_{1,{\rm bath}}=(\hat{c}+\hat{c}^{\dagger})/2$ and $X_{2,{\rm bath}}=i(\hat{c}^{\dagger}+\hat{c})/2$).
\Question We consider $(\Delta Y_1)^2=\langle Y_1^2 \rangle-\langle Y_1 \rangle^2$ and $\langle Y_i^2 \rangle=G \langle X_i^2 \rangle +2\sqrt{G(G-1)} \langle X_i \rangle \langle X_{i,{\rm bath}}\rangle_T+(G-1) \langle X_{i,{\rm bath}}^2\rangle_T=G \langle X_i^2\rangle+(G-1) \langle X_{i,{\rm bath}}^2 \rangle_T$ while $\langle Y_i\rangle$ is given in the previous question. We have 
\[
\langle X_{i,{\rm bath}}^2 \rangle_T=\frac{1}{4}\langle \hat{c} \hat{c}^{\dagger}+\hat{c}^{\dagger} \hat{c}\rangle_T=\frac{1}{4}[1+2 \langle \hat{c}^{\dagger} \hat{c}\rangle_T]=\frac{1}{4}[1+ 2\bar{n}(\omega)],
\]
with thermal occupation number $\bar{n}(\omega)=\frac{1}{e^{\frac{\hbar \omega}{kT}}-1} \geq  0$. This leads to 
the lower bound in Eq.~\eqref{ineq}.
\Question In the vacuum $\bar{n}(\omega)=0$ for any $\omega$.
\end{Answer}

\begin{Exercise}[title={Degenerate parametric amplifier},label=exc:deg_par_amp]
\begin{figure}[htb]
\centering
\includegraphics[height=4cm]{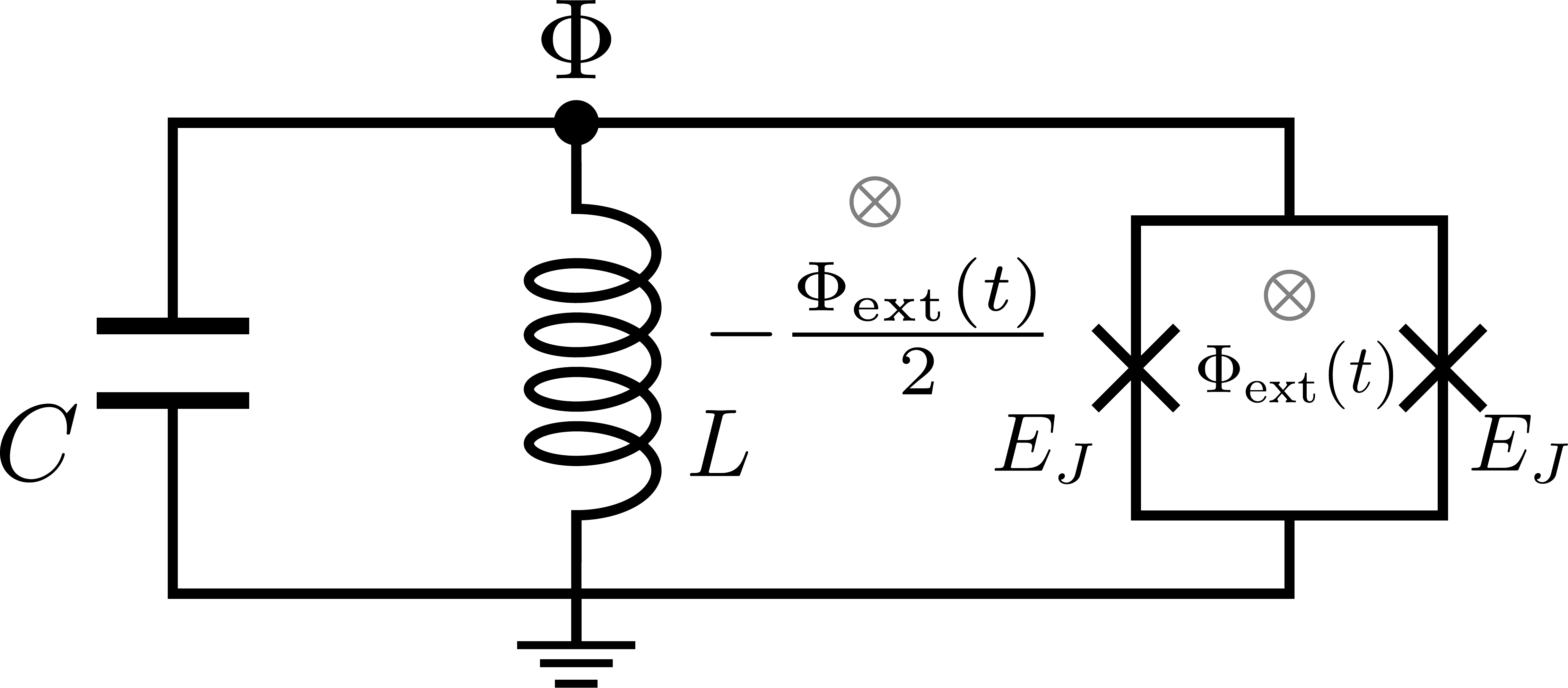}
\caption{LC oscillator with a SQUID loop threaded by two time-dependent fluxes. }
\label{fig:LCsquid}
\end{figure}

In this exercise we consider an amplifier which realizes phase-sensitive amplification or single-mode squeezing. Squeezing and other quantum optical transformations are well described in~\cite{book:Gerry.Knight:QuantumOptics}. Consider the circuit in Fig.~\ref{fig:LCsquid}. 
\Question Obtain the Lagrangian and the Hamiltonian of the system. 
\Question Consider a time-dependent flux $\Phi_{\mathrm{ext}}(t)= \mathcal{E}_P \cos \omega_P t$ and write the potential term in the Hamiltonian, using
the first-order 
\href{https://en.wikipedia.org/wiki/Jacobi}
{Jacobi-Anger expansion}
\begin{equation}
\cos(z \cos \theta) \approx	J_0(z)-2 J_2(z) \cos(2 \theta),
\end{equation}
where $J_0(z)$ (resp. $J_2(z)$) are the 0th (resp.~2nd) Bessel function of the first kind. Also, expand the potential to second order.
Then, introduce annihilation and creation operators for the mode. Show that for a specific value of $\omega_P$ (what value?), one gets the Hamiltonian of a degenerate parametric amplifier
\begin{equation}
\label{Hdpa}
H_{\rm DPA} = \frac{ \hbar \varepsilon}{2} (\hat{a}^2+\hat{a}^{ \dagger 2}),
\end{equation} 
in the interaction (rotating frame) picture of the resonant frequency of the mode. Determine $\varepsilon$. 
\Question The time evolution operator associated to $H_{\rm DPA}$ has the form of a unitary squeezing operator 
\begin{equation}
S(\xi)= \exp \biggl[\frac{1}{2}\biggl(\xi^{*} \hat{a}^2 - \xi \hat{a}^{\dagger 2} \biggr) \biggr],
\end{equation}  
where generically $\xi= r e^{i \theta}$ is a complex number. It can be shown, see e.g. Wikipedia for a  \href{https://en.wikipedia.org/wiki/Squeeze_operator}{proof}, or \cite{book:Gerry.Knight:QuantumOptics} that for such an operator one has
\begin{equation}
S^{\dagger}(\xi) \hat{a} S(\xi)= \hat{a} \cosh r - \hat{a}^{\dagger} e^{i \theta} \sinh r. 
\label{eq:sq}
\end{equation}
Using Eq.~\eqref{eq:sq}, compute the average and the variance of the two quadratures $X_{\theta/2}$ and $X_{\theta/2+\pi/2}$ (using the definition 
$X_{\varphi}= (\hat{a} e^{-i \varphi}+\hat{a}^{\dagger}e^{i \varphi})/2$)
 in a vacuum squeezed state $\ket{\xi}= S(\xi) \ket{0}$. 
 Why does this model realize a phase-sensitive amplifier?
  According to these results, what happens when the system starts in the vacuum and undergoes a time evolution with the Hamiltonian in Eq.~\eqref{Hdpa}? 
 \Question Consider now attaching a transmission line to the system. Write down the Heisenberg-Langevin equation for the annihilation and creation operators $\hat{a}(t)$ and $\hat{a}^\dagger(t)$ (as in Section~\ref{sec:HL_derive}), and solve the system in the steady state, i.e., express $\hat{b}_{\rm out}(t)$ in terms of the input fields. What parameters determine the amount of squeezing in $\hat{b}_{\rm out}$? Give an expression for $\cosh(r)$.
 {\em Comment}: This analysis can be improved by Fourier transforming the Langevin equations and considering $\hat{b}_{\rm out}[\omega]$ as a function of $\hat{b}_{\rm in}[\omega]$ and $\hat{b}_{\rm in}^{\dagger}[\omega]$, allowing one to understand the frequency-dependence of the squeezing gain.
\end{Exercise}

\begin{Answer}[ref={exc:deg_par_amp}]
\Question  The Lagrangian reads
\begin{equation}
\mathcal{L}= \frac{C}{2} \dot{\Phi}-\frac{\Phi^2}{2 L}+E_J \cos \biggl[\frac{2 \pi}{\Phi_0}\biggl(\Phi + \frac{\Phi_{\mathrm{ext}}(t)}{2} \biggr)\biggr]+E_J \cos \biggl[\frac{2 \pi}{\Phi_0}\biggl(\Phi -\frac{\Phi_{\mathrm{ext}}(t)}{2} \biggr)\biggr].
\end{equation}
Identifying the conjugate variable $Q= \partial  \mathcal{L}/ \partial \dot{\Phi}$ we obtain the Hamiltonian
\begin{equation}
H= \frac{Q^2}{2 C}+ \frac{\Phi^2}{2 L}- 2 E_J \cos \biggl [\frac{\pi}{\Phi_0} \Phi_{\mathrm{ext}}(t) \biggr] \cos \biggl [\frac{2 \pi}{\Phi_0} \Phi \biggr]. 
\end{equation}
\Question We now consider $\Phi_{\mathrm{ext}}(t)= \mathcal{E}_P \cos (\omega_P t)$. We have the potential
\begin{equation}
U= \frac{\Phi^2}{2 L} - 2 E_J \cos \biggl [\frac{\pi}{\Phi_0} \mathcal{E}_P \cos (\omega_P t)\biggr] \cos \biggl [\frac{2 \pi}{\Phi_0} \Phi \biggr],
\end{equation}
which, using the Jacobi-Anger expansion of $\cos(z \cos \theta)$, expanding the second cosine to second order in $\Phi$ and neglecting terms that do not alter the Euler-Lagrange equations, becomes
\begin{equation}
\label{Ujac}
U=  \frac{\Phi^2}{2 L} + 2 E_J \biggl[J_0 \biggl ( \frac{\pi \mathcal{E}_P}{\Phi_0} \biggr)-2 J_2 \biggl ( \frac{\pi \mathcal{E}_P}{\Phi_0} \biggr) \cos(2 \omega_P t)\biggr] \biggl( \frac{2 \pi}{\Phi_0}\biggr)^2 \Phi^2.
\end{equation}
In order to keep a compact notation, we define an equivalent inductance $L_H$ via the relation
\begin{equation}
\frac{1}{L_H}= \frac{1}{L}+4 E_J J_0 \biggl ( \frac{\pi \mathcal{E}_P}{\Phi_0} \biggr)\biggl( \frac{2 \pi}{\Phi_0}\biggr)^2,
\end{equation}
and a parameter 
\begin{equation}
\lambda= -4 E_J J_2 \biggl ( \frac{\pi \mathcal{E}_P}{\Phi_0} \biggr) \biggl( \frac{2 \pi}{\Phi_0}\biggr)^2,
\end{equation}
to rewrite the Hamiltonian as
\begin{equation}
H= \frac{Q^2}{2 C} + \frac{\Phi^2}{2 L_H}+\lambda \cos(2 \omega_P t) \Phi^2.
\end{equation}
We now introduce annihilation and creation operators ($Z= \sqrt{L_H/C}$, $\omega_r=1/\sqrt{L_H C}$)
\begin{subequations}
\begin{equation}
\Phi= \sqrt{\frac{\hbar Z}{2}} (\hat{a}+\hat{a}^{\dagger}),
\end{equation}
\begin{equation}
Q= i \sqrt{\frac{\hbar }{2 Z}} (\hat{a}^{\dagger}-\hat{a}),
\end{equation}
\end{subequations}
in terms of which the Hamiltonian reads
\begin{equation}
H= \hbar \omega_r \hat{a}^{\dagger}\hat{a} + \frac{ \hbar \varepsilon }{2} \biggl(e^{i 2 \omega_P t}+e^{-i 2 \omega_P t} \biggr)(\hat{a}^2 +\hat{a}^{\dagger 2}+1+2 \hat{a}^{\dagger}\hat{a}),
\end{equation}
where we define the parameter $\varepsilon= \lambda Z/2$. We now go to an interaction picture with $H_0= \hbar \omega_r \hat{a}^{\dagger} \hat{a}$. The Hamiltonian in the interaction picture becomes
\begin{align}
H_I(t)= H_{\rm DPA} = e^{i H_0 t/\hbar} He^{-i H_0 t/\hbar} -H_0= \frac{ \hbar \varepsilon }{2} \biggl(e^{i 2 \omega_P t}+e^{-i 2 \omega_P t} \biggr) \notag \\ \times \biggl (a^2 e^{-2 i \omega_r t}+a^{\dagger 2} e^{-2 i \omega_r t} +1+2 \hat{a}^{\dagger}\hat{a} \biggr).
\end{align}
In order to obtain the desired Hamiltonian, we need to choose a particular frequency $\omega_P$ so that only the desired term is time-independent, while the effect of the others averages out in a RWA. This frequency is clearly $\omega_P= \omega_r$. We finally get the Hamiltonian of a degenerate parametric amplifier (within the RWA approximation):
\begin{equation}
H \overset{\mathrm{RWA}}{=} \frac{\hbar \varepsilon}{2} (\hat{a}^2+\hat{a}^{\dagger 2}).
\end{equation}
\Question Using Eqs.~\eqref{eq:sq} and its Hermitian conjugate, it is immediately clear that the average of a generic quadrature in a vacuum squeezed state $\ket{\xi}= S(\xi) \ket{0}$ is zero. In fact, it is sufficient to show that the average of $\hat{a}$ is zero:
\begin{equation}
\bra{\xi} \hat{a} \ket{\xi} = \bra{0} S^{\dagger}(\xi) \hat{a} S(\xi) \ket{0}= \cosh r \bra{0} \hat{a} \ket{0}-e^{i \theta} \sinh r \bra{0} \hat{a}^{\dagger} \ket{0}=0,
\end{equation}
and similarly $\bra{\xi} \hat{a}^{\dagger} \ket{\xi}=0$. This implies that the variance of a generic quadrature $(\Delta X_{\varphi})^2$ is equal to the average of $X_{\varphi}^2$ in a vacuum squeezed state. In order to compute this variance, let us first derive how a generic quadrature $X_{\varphi}$ transforms under the action of the squeezing operator. Using again Eqs.~\eqref{eq:sq} and its Hermitian conjugate, we get
\begin{equation}
S^{\dagger}(\xi) X_{\theta/2} S(\xi) = [\cosh(r)-\sinh(r)]X_{\theta/2}, 
\end{equation}
and
\begin{equation}
S^{\dagger}(\xi) X_{\theta/2+\pi/2} S(\xi) = [\cosh(r)+\sinh(r)]X_{\theta/2+\pi/2}.
\end{equation}
We can now compute the variance of $X_{\theta/2}$ in a vacuum-squeezed state:
\begin{equation}
(\Delta X_{\theta/2})^2 = \bra{\xi} X_{\theta/2}^2 \ket{\xi}=[\cosh^2(r)+\sinh^2(r)-2\sinh(r)\cosh(r)]\bra{0} X_{\theta/2}^2 \ket{0}=\frac{1}{4}e^{-r}.
\end{equation}
Similarly, we can deduce the variance of the quadrature $X_{\theta/2+\pi/2}$
\begin{equation}
(\Delta X_{\theta/2+\pi/2})^2 = \frac{1}{4}e^r. 
\end{equation}
 We see that the variance of the quadrature $X_{\theta/2}$ is reduced (squeezed) compared to the vacuum, while the variance of the conjugate quadrature $X_{\theta/2+\pi/2}$ is increased (anti-squeezed). These two quadratures are special because they are the only quadratures for which the uncertainty relation $(\Delta X_{\theta/2})^2 (\Delta X_{\theta/2+\pi/2})^2\geq \frac{1}{16}$ is satisfied with equality sign for the squeezed vacuum, as you can check. The amplification is phase-sensitive since different quadratures (with different phases) are squeezed to a different extent.
 \par 
Let us consider the time evolution operator associated to the Hamiltonian Eq.~\eqref{Hdpa},
\begin{equation}
\label{timeEvOp}
U_I(t) = \exp \biggl[- \frac{i \varepsilon t}{2} (\hat{a}^2+\hat{a}^{\dagger 2})\biggr],
\end{equation} 
which is a squeezing operator with $\xi(t)= i \varepsilon t=\epsilon t e^{i \pi/2}$, hence $\theta= \pi/2$. We conclude that if we start from the vacuum, the time evolution operator of Eq.~\eqref{timeEvOp} generates a vacuum squeezed state in which the quadrature $X_{\pi/4}$ is squeezed, while the quadrature $X_{3 \pi/4}$ is anti-squeezed. The modulus $r$ of the squeezing parameter $\xi$ increases linearly with time.
\Question
The Langevin equations associated to $\hat{a}(t)$ and $\hat{a}^{\dagger}(t)$ (dropping time dependence for simplicity) are
\begin{subequations}
\begin{equation}
\frac{d \hat{a}}{d t} = \frac{1}{i \hbar} [\hat{a}, H_{\rm DPA}]-\frac{\kappa}{2}\hat{a}+\sqrt{\kappa} \hat{b}_{\rm in},
\end{equation}
\begin{equation}
\frac{d \hat{a}^{\dagger}}{d t} = \frac{1}{i \hbar} [\hat{a}^{\dagger}, H_{\rm DPA}]-\frac{\kappa}{2}\hat{a}^{\dagger}+\sqrt{\kappa} \hat{b}_{\rm in}^{\dagger},
\end{equation}
\end{subequations}
which give
\begin{subequations}
\begin{equation}
\frac{d \hat{a}}{d t} = -i \varepsilon \hat{a}^{\dagger}-\frac{\kappa}{2}\hat{a}+\sqrt{\kappa} \hat{b}_{\rm in},
\end{equation}
\begin{equation}
\frac{d \hat{a}^{\dagger}}{d t} =i \varepsilon \hat{a}-\frac{\kappa}{2}\hat{a}^{\dagger}+\sqrt{\kappa} \hat{b}_{\rm in}^{\dagger},
\label{eq:input-mode}
\end{equation}
\end{subequations}
with input-output relation
\begin{equation}
\label{inOut}
\hat{b}_{\rm out}= \sqrt{\kappa} \hat{a}-\hat{b}_{\rm in}.
\end{equation}
In the steady state they become algebraic equations
\begin{subequations}
\begin{equation}
-i \varepsilon \hat{a}^{\dagger}-\frac{\kappa}{2}\hat{a}+\sqrt{\kappa} \hat{b}_{\rm in}=0,
\end{equation}
\begin{equation}
i \varepsilon \hat{a}-\frac{\kappa}{2}\hat{a}^{\dagger}+\sqrt{\kappa} \hat{b}_{\rm in}^{\dagger}=0.
\end{equation}
\end{subequations}
Solving for $\hat{a}$ we get
\begin{equation}
\hat{a}= -\frac{\kappa}{\kappa^2-4 \varepsilon^2} \frac{4 i \varepsilon}{\sqrt{\kappa}} \hat{b}_{\rm in}^{\dagger}+ \frac{2 \kappa^2}{\kappa^2-4 \varepsilon^2} \frac{1}{\sqrt{\kappa}} \hat{b}_{\rm in},
\end{equation}
and using the input-output relations Eq.~\eqref{inOut} we obtain the output field annihilation operator
\begin{equation}
\label{outSteady}
\hat{b}_{\rm out}(t) = \frac{\kappa^2+4 \varepsilon^2}{\kappa^2-4 \varepsilon^2} \hat{b}_{\rm in}(t)-i \frac{4 \varepsilon \kappa}{\kappa^2-4 \varepsilon^2} \hat{b}_{\rm in}^{\dagger}(t).
\end{equation}
You can verify that this equation is of the form $\hat{b}_{\rm out}=\hat{b}_{\rm in} \cosh(r)-\hat{b}_{\rm in}^{\dagger} e^{i \theta}\sinh(r)$ and $\cosh(r)=\frac{\kappa^2+4\epsilon^2}{\kappa^2-4\epsilon^2}$, depending on the strength of the squeezing term $\epsilon$ and the coupling (decay) strength into the transmission line $\kappa$. Note that the amount of squeezing $r$ does not increase linearly in time, but is limited by the strength of $\kappa$ versus $\epsilon$.
\end{Answer}

\begin{Exercise}[title={Nondegenerate parametric amplifier: Josephson Ring Modulator},label=exc:JRM]
\begin{figure}[htb]
\centering
\includegraphics[height=8cm]{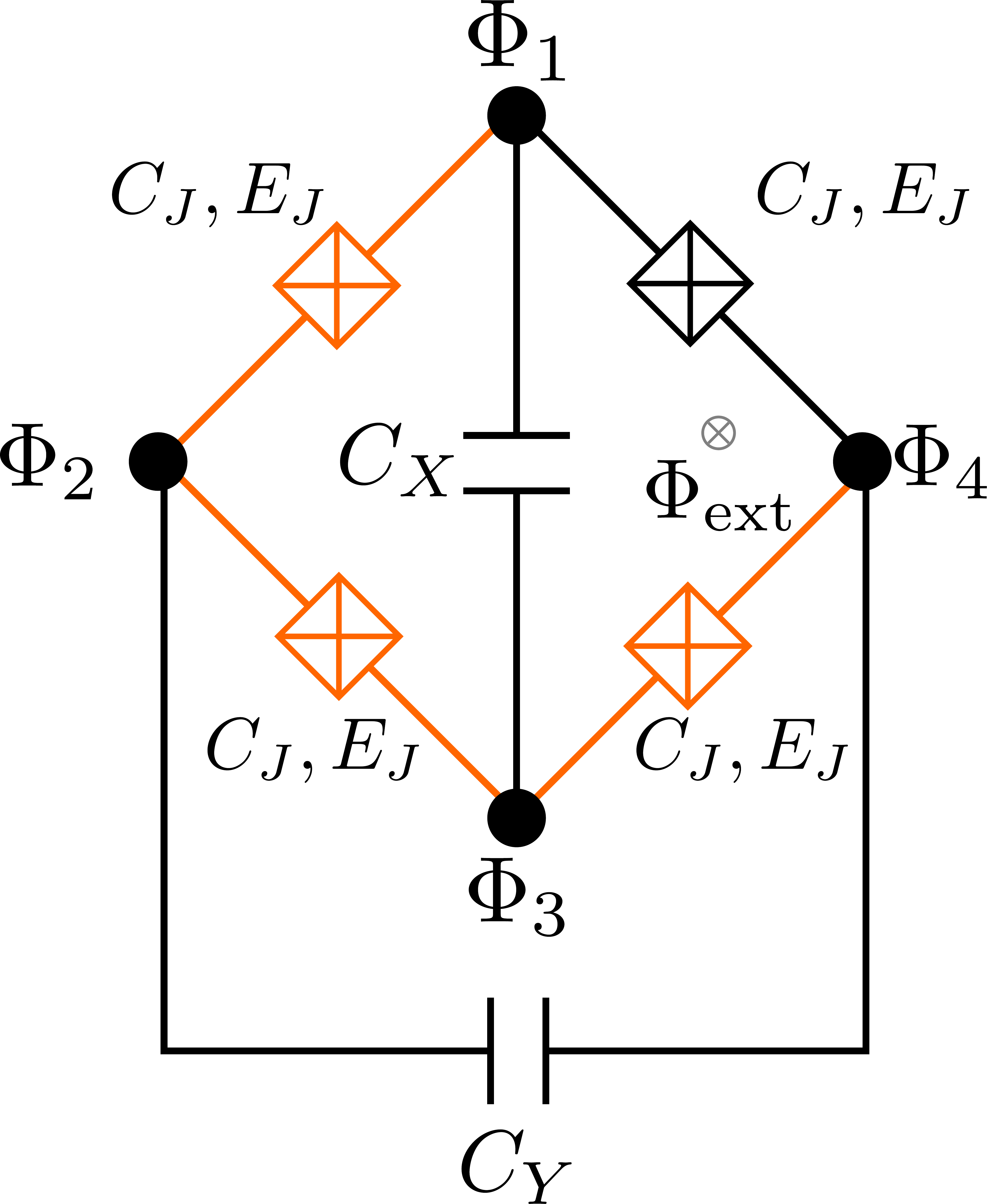}
\caption{Circuit of a Josephson Ring Modulator.}
\label{fig:JRM}
\end{figure}
Consider the circuit of a Josephson Ring Modulator (JRM) in Fig.~\ref{fig:JRM} with equal junctions and a loop threaded by a constant external magnetic flux $\Phi_{\mathrm{ext}}$. 
The goal of this exercise is to show how it can be used to realize the Hamiltonian of a non-degenerate parametric amplifier in Eq.~\eqref{eq:ndpa}. While the degenerate amplifier with Hamiltonian in Eq.~\eqref{Hdpa} in the previous Exercise enacts single-mode squeezing, the non-degenerate amplifier in Eq.~\eqref{eq:ndpa} enacts two-mode squeezing, that is, the two modes on which it acts are different/non-degenerate.
\Question Obtain the Lagrangian of this system, taking the tree shown in orange in Fig.~\ref{fig:JRM}, assuming $\Phi_3 =0$ as reference ground node. Then, introduce two consecutive changes of variables. First 
\begin{subequations}
\begin{equation}
\Phi_1- \Phi_2 = \tilde{\Phi}_1- \tilde{\Phi}_2+ \frac{\Phi_{\rm ext}}{4},
\end{equation}
\begin{equation}
\Phi_2 = \tilde{\Phi}_2+\frac{\Phi_{\rm ext}}{4},
\end{equation}
\begin{equation}
\Phi_4 = \tilde{\Phi}_4- \frac{\Phi_{\rm ext}}{4},
\end{equation}
\end{subequations} 
and afterwards
\begin{subequations}
\begin{equation}
\Phi_X=\tilde{\Phi}_1,
\end{equation}
\begin{equation}
\Phi_Y= \tilde{\Phi}_4-\tilde{\Phi}_2,
\end{equation}
\begin{equation}
\Phi_Z = \tilde{\Phi}_1- \tilde{\Phi}_2-\tilde{\Phi}_4.
\end{equation}
\end{subequations} 
Obtain the Hamiltonian in terms of $\Phi_X, \Phi_Y, \Phi_Z$ and their conjugate variables. By setting $\Phi_{\rm ext}= \Phi_0/2$ (half a flux quantum), show that the system approximately realizes the Hamiltonian of three harmonic oscillators coupled by the coupling term $H_c= \nu \Phi_X \Phi_Y \Phi_Z$: give the parameter $\nu$. \emph{Hint: You can expand the potential up to second order.}
Finally, introduce annihilation and creation operators and show that by a careful choice of the characteristic frequencies $\omega_X$, $\omega_Y$, $\omega_Z$ (what choice?) one obtains the three-wave mixing Hamiltonian 
\begin{equation}\label{eq::3waveEq}
H_{3 W}/\hbar \approx \omega_X \hat{a}_{X}^{\dagger}\hat{a}_{X}+  \omega_Y \hat{a}_{Y}^{\dagger}\hat{a}_{Y}+  \omega_Z \hat{a}_{Z}^{\dagger}\hat{a}_{Z}+\eta(\hat{a}_X^{\dagger}\hat{a}_{Y}^{\dagger}\hat{a}_{Z}+\mathrm{h.c.}),
\end{equation} 
with $\eta$ the three-wave mixing parameter. 
\Question Consider the Hamiltonian in Eq.~\eqref{eq::3waveEq}. In what follows we call mode $X$ the signal mode, i.e., $\hat{a}_{X}=\hat{a}_{S}$, mode $Y$ the idler mode, i.e., $\hat{a}_{Y}=\hat{a}_{I}$, and mode $Z$ the pump mode, i.e., $\hat{a}_Z= \hat{a}_P$. We drive the pump mode with a coherent drive at frequency $\omega_P$ with large amplitude so that we can substitute $\hat{a}_P$ with some time-dependent (approximate) expectation value $\langle \hat{a}_P \rangle=i \alpha_P e^{-i \omega_P t}$ with some real amplitude $\alpha_P$. Convert the Hamiltonian to the interaction picture (rotating frame) at the frequencies of the modes $X$ and $Y$ to get the two-mode squeezing Hamiltonian
\begin{equation}
\tilde{H}_{I 3W}/\hbar =  i \lambda \biggl(\hat{a}_{S}^{\dagger} \hat{a}_{I}^{\dagger} -\hat{a}_{S} \hat{a}_{I}\biggr).
\label{eq:ndpa}
\end{equation}
\Question We now assume the signal mode to be coupled to a transmission line with input field $\hat{b}_{S, \rm in}(t)$ (coupling set by decay rate $\kappa_S$), while the idler mode is coupled to another transmission line with $\hat{b}_{I, \rm in}(t)$ (coupling set by decay rate $\kappa_I$). Write down the Heisenberg-Langevin equations of motion for $\hat{a}_{S}(t)$ and $\hat{a}_{I}^{\dagger}(t)$ as in Section~\ref{sec:HL_derive}, and obtain the steady-state solution $\hat{b}_{S,\rm out}(t)$ as a function of input fields of signal and idler. Discuss the effect of the idler mode when it starts in the vacuum state (see Exercise \ref{exc:HC}).
{\em Comment}: This analysis can be improved by Fourier transforming the Langevin equations and considering $\hat{b}_{S,\rm out}[\omega]$ as a function of $\hat{b}_{S,\rm in}[\omega]$ and $\hat{b}_{I,\rm in}^{\dagger}[\omega]$, allowing one to understand the frequency dependence, see \cite{CDGM:noiseRMP}.
\end{Exercise}

\begin{Answer}[ref={exc:JRM}]
\Question 
The Lagrangian in terms of node variables reads
\begin{multline}
\mathcal{L}= \frac{C_J}{2} \dot{\Phi}_2^2 + \frac{C_J}{2} \dot{\Phi}_4^2 + \frac{C_J}{2} (\dot{\Phi}_1-\dot{\Phi}_2)^2+\frac{C_J}{2} (\dot{\Phi}_1-\dot{\Phi}_4)^2 +\frac{C_X}{2} \dot{\Phi}_1^2 +\frac{C_Y}{2}(\dot{\Phi}_4-\dot{\Phi}_2)^2 \\ + E_J \cos\biggl[\frac{2 \pi}{\Phi_0}\Phi_2 \biggr]+E_J \cos\biggl[\frac{2 \pi}{\Phi_0}\Phi_4 \biggr]+E_J \cos\biggl[\frac{2 \pi}{\Phi_0}(\Phi_1-\Phi_2) \biggr]+E_J \cos \biggl[\frac{2 \pi}{\Phi_0}(\Phi_1-\Phi_4-\Phi_{\mathrm{ext}}) \biggr].
\end{multline}
The first suggested change of variables results in
\begin{multline}
\mathcal{L}= \frac{C_J}{2} \dot{\tilde{\Phi}}_2^2 + \frac{C_J}{2} \dot{\tilde{\Phi}}_4^2 + \frac{C_J}{2} (\dot{\tilde{\Phi}}_1-\dot{\tilde{\Phi}}_2)^2+\frac{C_J}{2} (\dot{\tilde{\Phi}}_1-\dot{\tilde{\Phi}}_4)^2 +\frac{C_X}{2} \dot{\tilde{\Phi}}_1^2 +\frac{C_Y}{2}(\dot{\tilde{\Phi}}_4-\dot{\tilde{\Phi}}_2)^2 \\ + E_J \cos\biggl[\frac{2 \pi}{\Phi_0}\biggl(\tilde{\Phi}_2+\frac{\Phi_{\mathrm{ext}}}{4} \biggr) \biggr]+E_J \cos\biggl[\frac{2 \pi}{\Phi_0}\biggl(\tilde{\Phi}_4-\frac{\Phi_{\mathrm{ext}}}{4} \biggr)\biggr]+E_J \cos\biggl[\frac{2 \pi}{\Phi_0}\biggl(\tilde{\Phi}_1-\tilde{\Phi}_2+\frac{\Phi_{\mathrm{ext}}}{4} \biggr) \biggr] \\+E_J \cos \biggl[\frac{2 \pi}{\Phi_0}\biggl(\tilde{\Phi}_1-\tilde{\Phi}_4-\frac{\Phi_{\mathrm{ext}}}{4} \biggr) \biggr],
\end{multline}
while the second change of variables diagonalizes the kinetic part of the Lagrangian and brings it into the simple form
\begin{multline}
\mathcal{L} = \frac{C_J+C_X}{2} \dot{\Phi}_X^2 + \frac{C_J+C_Y}{2} \dot{\Phi}_Y^2 + \frac{C_J}{2} \dot{\Phi}_Z^2 \\ + 4  E_J\cos \biggl[ \frac{2 \pi}{\Phi_0} \frac{\Phi_{\mathrm{ext}}}{4} \biggr]\cos \biggl[ \frac{2 \pi}{\Phi_0} \frac{\Phi_X}{2} \biggr] \cos \biggl[ \frac{2 \pi}{\Phi_0} \frac{\Phi_Y}{2} \biggr] \cos \biggl[ \frac{2 \pi}{\Phi_0} \frac{\Phi_Z}{2} \biggr] \\ + 4 E_J \sin \biggl[ \frac{2 \pi}{\Phi_0} \frac{\Phi_{\mathrm{ext}}}{4} \biggr]\sin \biggl[ \frac{2 \pi}{\Phi_0} \frac{\Phi_X}{2} \biggr] \sin \biggl[ \frac{2 \pi}{\Phi_0} \frac{\Phi_Y}{2} \biggr] \sin \biggl[ \frac{2 \pi}{\Phi_0} \frac{\Phi_Z}{2} \biggr].
\end{multline}
From this Lagrangian it is easy to obtain the Hamiltonian
\begin{multline}
H= \frac{Q_X^2}{2 C_{\Sigma X}}+ \frac{Q_Y^2}{2 C_{\Sigma Y}} + \frac{Q_Z^2}{2 C_J} -4 E_J \cos \biggl[ \frac{2 \pi}{\Phi_0} \frac{\Phi_{\rm ext}}{4} \biggr]\cos \biggl[ \frac{2 \pi}{\Phi_0} \frac{\Phi_X}{2} \biggr] \cos \biggl[ \frac{2 \pi}{\Phi_0} \frac{\Phi_Y}{2} \biggr] \cos \biggl[ \frac{2 \pi}{\Phi_0} \frac{\Phi_Z}{2} \biggr] \\ - 4 E_J \sin \biggl[ \frac{2 \pi}{\Phi_0} \frac{\Phi_{\rm ext}}{4} \biggr]\sin \biggl[ \frac{2 \pi}{\Phi_0} \frac{\Phi_X}{2} \biggr] \sin \biggl[ \frac{2 \pi}{\Phi_0} \frac{\Phi_Y}{2} \biggr] \sin \biggl[ \frac{2 \pi}{\Phi_0} \frac{\Phi_Z}{2} \biggr],
\end{multline}
with the conjugate variables $Q_X= C_{\Sigma X} \dot{\Phi}_X$, $Q_Y= C_{\Sigma Y} \dot{\Phi}_Y$, $Q_Z= C_{J} \dot{\Phi}_Z$ and the abbreviations $C_{\Sigma X}= C_X+ C_J$, $C_{\Sigma Y}= C_Y+ C_J$. Setting $\Phi_{\mathrm{ext}}=\Phi_0/2$ and expanding the potential up to second order, as suggested, one obtains
\begin{equation}
H= \frac{Q_X^2}{2 C_{\Sigma X}}+ \frac{Q_Y^2}{2 C_{\Sigma Y}} + \frac{Q_Z^2}{2 C_J}+\frac{1}{2 L_{JRM}} (\Phi_X^2+\Phi_Y^2+\Phi_Z^2)+\nu \Phi_X \Phi_Y \Phi_Z,
\end{equation}
where we define a characteristic inductance of the JRM, $L_{JRM}= \Phi_0^2/(4 \sqrt{2} E_J \pi^2)$, and the coupling parameter $\nu=-2 \sqrt{2} E_J \pi^3/\Phi_0^3$. \\
In order to pass to a quantum description, we introduce annihilation and creation operators for the modes $X, Y$ and $Z$, as
\begin{subequations}
\begin{equation}
\Phi_K = \sqrt{\frac{\hbar Z_K}{2}} (\hat{a}_K+\hat{a}_K^{\dagger}),
\end{equation}
\begin{equation}
Q_K = i \sqrt{\frac{\hbar }{2 Z_K}} (\hat{a}_K^{\dagger}-\hat{a}_K),
\end{equation}
\end{subequations}
with $K=\{X, Y, Z\}$ and $Z_{K}$ the characteristic impedance of mode $K$, and we rewrite the Hamiltonian as
\begin{equation}
\label{H3c}
H= \hbar \omega_X \hat{a}_X^{\dagger}\hat{a}_X+\hbar \omega_Y \hat{a}_Y^{\dagger}\hat{a}_Y +\hbar \omega_Z \hat{a}_Z^{\dagger}\hat{a}_Z + \hbar \eta (\hat{a}_{X}+\hat{a}_{X}^{\dagger})(\hat{a}_{Y}+\hat{a}_{Y}^{\dagger}) (\hat{a}_{Z}+\hat{a}_{Z}^{\dagger}),
\end{equation}
where we introduce the characteristic frequencies of the resonators and the parameter
\begin{equation}
\hbar \eta = \nu \biggl(\frac{\hbar}{2} \biggr)^{3/2} \sqrt{Z_X Z_Y Z_Z}. 
\end{equation}
Let us take a closer look at Eq.~\eqref{H3c}. We would like to have an argument for keeping only the term $a_{X}^{\dagger} a_{Y}^{\dagger} a_{Z}$ and its Hermitian conjugate. Of course, what we need to impose is that this term is the only energy-preserving (resonant) term among all the terms that appear in the expansion of the coupling Hamiltonian. This is the case if the frequency of the mode $Z$ (the pump) is equal to the sum of the other two. So we set the condition $\omega_Z= \omega_X+\omega_Y$.
\Question 
Let us rewrite the Hamiltonian in the new notation, keeping only the three-wave-mixing terms
\begin{equation}
H_{3 W}/\hbar\approx \omega_S \hat{a}_S^{\dagger}\hat{a}_S+\omega_I \hat{a}_I^{\dagger}\hat{a}_I + \omega_P \hat{a}_P^{\dagger}\hat{a}_P + \eta (\hat{a}_{S}^{\dagger} \hat{a}_{I}^{\dagger} \hat{a}_{P}+\mathrm{h.c.}),
\end{equation} 
with the three-wave mixing condition $\omega_P= \omega_I+\omega_S$. By assuming that we strongly drive the pump, we replace the operator $\hat{a}_P$ with its average $i \alpha_P e^{-i \omega_P t}$ to get
\begin{equation}
H_{3 W}/\hbar \approx  \omega_S \hat{a}_S^{\dagger}\hat{a}_S+ \omega_I \hat{a}_I^{\dagger}\hat{a}_I + i  \lambda \biggl(\hat{a}_{S}^{\dagger} \hat{a}_{I}^{\dagger} e^{-i \omega_P t}-\hat{a}_{S} \hat{a}_{I} e^{+i \omega_P t} \biggr),
\end{equation}
where we called $\lambda=\eta \alpha_P$. Going to an interaction picture at the resonators' frequencies, i.e., with $H_0/hbar=  \omega_S \hat{a}_S^{\dagger}\hat{a}_S+\omega_I \hat{a}_I^{\dagger}\hat{a}_I$, 
we have $\tilde{H}_{3W}=e^{i H_0 t} H_{3 W}e^{-i H_0 t} -H_0$, obtaining Eq.~\eqref{eq:ndpa}.
\Question
We have the following Heisenberg-Langevin equations for $\hat{a}_{I}^{\dagger}$ and $\hat{a}_S$:
\begin{subequations}
\begin{equation}
\frac{d \hat{a}_I^{\dagger}}{dt} = \frac{1}{i \hbar}[\hat{a}_I^{\dagger}, \tilde{H}_{I 3W}] -\frac{\kappa_I}{2} \hat{a}_I^{\dagger} +\sqrt{\kappa_I} \hat{b}_{I, \rm in}^{\dagger},
\end{equation}
\begin{equation}
\frac{d \hat{a}_S}{dt} = \frac{1}{i \hbar}[\hat{a}_S, \tilde{H}_{I 3W}] -\frac{\kappa_S}{2} \hat{a}_S +\sqrt{\kappa_S} \hat{b}_{S, \rm in},
\end{equation}
\end{subequations} 
with input-output relations $\hat{b}_{I, \rm out}= \sqrt{\kappa_I} \hat{a}_I -\hat{b}_{I, \rm in}$ and $
\hat{b}_{S, \rm out}= \sqrt{\kappa_S} \hat{a}_S -\hat{b}_{S, \rm in}$. Computing the commutators we get
\begin{subequations}
\begin{equation}
\frac{d \hat{a}_I^{\dagger} }{dt} = \lambda 
\hat{a}_{S}-\frac{\kappa_I}{2} \hat{a}_I^{\dagger} +\sqrt{\kappa_I} \hat{b}_{I, \rm in}^{\dagger},
\end{equation}
\begin{equation}
\frac{d \hat{a}_S}{dt} = \lambda \hat{a}_{I}^{\dagger} -\frac{\kappa_S}{2} \hat{a}_S +\sqrt{\kappa_S} \hat{b}_{S,\rm in}.
\end{equation}
\end{subequations} 
The steady-state solutions are readily obtained as
\begin{subequations}
\begin{equation}
\hat{a}_S(t)= \frac{2 \lambda}{\kappa_S} \hat{a}_{I}^{\dagger}+\frac{2}{\sqrt{\kappa_S}} \hat{b}_{S, \rm in},
\end{equation}
\begin{equation}
\hat{a}_I^{\dagger}(t)= \frac{2 \lambda}{\kappa_I} \hat{a}_{S}+\frac{2}{\sqrt{\kappa_I}} \hat{b}_{I, \rm in}^{\dagger},
\end{equation}
\end{subequations}
which imply 
\begin{equation}
\hat{a}_S(t)= \frac{1}{1-Q^2} \frac{4 \lambda}{\kappa_S \sqrt{\kappa_I}}\hat{b}_{I, \rm in}^{\dagger}(t) + \frac{1}{1-Q^2}  \frac{2}{\sqrt{\kappa}_S} \hat{b}_{S, \rm in}(t),
\end{equation}
with $Q= 2 \lambda/(\sqrt{\kappa_1 \kappa_2})$. Using $\hat{b}_{S, \rm out}= \sqrt{\kappa_S} \hat{a}_S -\hat{b}_{S, \rm in}$ we can write the output field of the signal as
\begin{equation}
\hat{b}_{S, \rm out}(t)=  \frac{1+Q^2}{1-Q^2}\hat{b}_{S, \rm in}(t)+ \frac{2 Q}{1-Q^2} \hat{b}_{I, \rm in}^{\dagger}(t).
\end{equation}
We see that the idler mode will generate extra noise even when it is in the vacuum state, as discussed in the analysis of Eq.~\eqref{eq:phase-insens} in Exercise \ref{exc:HC}.
\end{Answer}

\appendix

\chapter{A review of canonical quantization}
\label{app:cc}

In this appendix, we review the construction of a Hamiltonian from a Lagrangian and the formalism of canonical quantization of classical conservative (energy-conserving) dynamics. We refer the reader to specialized texts on classical mechanics such as Ref.~\cite{goldstein} for more details.

\section{Principle of minimal action and gauge invariance}

We start with a phase space of $n$ independent variables $x_1, \ldots, x_n$ with $\dot{x_i}=d x_i /dt$ and each $x_i \in \mathbb{R}$. In a mechanical scenario the variables $x_1,\ldots x_n$ could be the positions of $n$ particles. In the electrical circuit case these could be a set of independent node fluxes, or independent node charges.
For an electrical circuit, we have freedom in this choice of variables, as long as we have verified that they form an independent set of variables.

Given is the Lagrangian $\lagrangian (x_1, \ldots, x_n, \dot{x}_1,\ldots, \dot{x}_n)$.  For simplicity, we assume that $\lagrangian$ has no explicit time dependence. For a general mechanical system, in the absence of a magnetic field, the Lagrangian will be a sum of the kinetic energy $T$ minus the potential energy $U$, i.e., 
\begin{equation}\label{eq:lagr_sc}
\lagrangian (x_1, \ldots, x_n, \dot{x}_1,\ldots, \dot{x}_n) = T(\dot{x}_1,\ldots,\dot{x}_n) - U(x_1, \ldots, x_n).
\end{equation} 
With the Lagrangian we can define the classical action as 
\begin{equation}
\mathcal{S} =\int_{t_1}^{t_2} dt\, \lagrangian (x_1, \ldots, x_n, \dot{x}_1, \ldots, \dot{x}_n).
\end{equation}
The principle of minimal action, i.e., $\delta S=0$, implies the classical equations of motion, called the {\em Euler-Lagrange equations}:
\begin{equation}\label{eq:EL}
\frac{d}{d t} \left(\frac{\partial \lagrangian }{\partial \dot{x}_i}\right)-\frac{\partial \lagrangian}{\partial x_i}=0, \quad \forall i = 1, \dots, n.
\end{equation}
Thus, the construction of the Lagrangian is determined by the classical dynamical equations and vice versa. However, the Lagrangian is not unique and it can be changed by a gauge transformation, namely 
\begin{equation}\label{eq:gauge}
   \lagrangian \rightarrow \lagrangian+\frac{df(x_1, \ldots, x_n)}{dt}, 
\end{equation} with $f(x_1, \ldots, x_n)$ an arbitrary function, without this affecting the equations of motion. Clearly, such change maps the action $\mathcal{S} \rightarrow \mathcal{S}+f(t_2)-f(t_1)$, which indeed does not change trajectory minimizes the action.

\begin{Exercise}[label=exc:gauge]
Verify that Eq.~\eqref{eq:EL} follows from 
$\frac{d}{d t} \left(\frac{\partial \lagrangian' }{\partial \dot{x}_i}\right)-\frac{\partial \lagrangian'}{\partial x_i}=0$ with $\lagrangian'=\lagrangian+\frac{df(x_1, \ldots, x_n)}{dt}$, using $\frac{df}{dt}=\sum_i \frac{\partial f}{ \partial x_i} \dot{x}_i$ as $f(.)$ only depends on $x_1, \ldots, x_n$.
\end{Exercise} 

\begin{Answer}[ref={exc:gauge}]
Consider 
$\frac{d}{dt} \frac{\partial}{\partial \dot{x}_i}\frac{df}{dt}=\frac{d}{dt}\frac{\partial f}{\partial x_i}=\sum_j \frac{\partial^2 f}{\partial x_j \partial x_i}\dot{x}_j=\frac{\partial }{\partial x_i}\frac{df}{dt}$. 
\end{Answer}

\section{Legendre transformation: the Hamiltonian}

To define a Hamiltonian, one introduces, for each independent variable $x_i$, a {\em conjugate} variable $p_i$ defined as 
\begin{equation}
p_i=\frac{\partial \lagrangian }{\partial \dot{x}_i}.
\end{equation}
For a simple mechanical system with kinetic energy given by $T =  \sum_i m_i \dot{x}_i^2/2$, the conjugate variables are momenta given by $p_i=m_i \dot{x}_i$. 

One can then define a new function, namely the Hamiltonian, which depends on the variables $\{x_i\}$ and their conjugates $\{p_i\}$, but no longer on $\{\dot{x}_i\}$. More precisely, the Hamiltonian $\hamiltonian (x_1, \dots, x_n, p_1, \dots, p_n)$ is defined as the Legendre transformation of the Lagrangian, i.e.,
\begin{equation}
\hamiltonian(x_1, \ldots, x_n,p_1, \ldots, p_n)=\sup_{\dot{x}_1, \ldots, \dot{x}_n} \left[\sum_{i=1}^n p_i \dot{x_i}-\lagrangian (x_1,\ldots,x_n, \dot{x}_1,\ldots, \dot{x}_n)\right].
\label{eq:defH}
\end{equation}

When $\lagrangian(x_1, \ldots, x_n, \dot{x}_1, \ldots, \dot{x}_n)$ is a {\em convex} function of the variables $\dot{x}_1, \ldots, \dot{x}_n$, the supremum in Eq.~\eqref{eq:defH} occurs at a unique point $\dot{x}_1^{\rm sup}, \ldots, \dot{x}_n^{\rm sup}$ since $\sum_i p_i \dot{x}_i-\lagrangian$ is concave in $\dot{x}_i$. At this supremum
\begin{equation}
\frac{\partial }{\partial \dot{x}_i} \left[ \sum_{i=1}^n p_i \dot{x}_i-\lagrangian(x_1,\ldots,x_n, \dot{x}_1,\ldots, \dot{x}_n)\right]=0 \implies p_i=\frac{\partial \lagrangian}{\partial \dot{x}_i}, \quad \forall i.
\end{equation}
Then we can derive the Hamiltonian in Eq.~\eqref{eq:defH} by determining $p_i$ and solving for $\dot{x}_i$ to eliminate the $\dot{x}_i$ variables in $\hamiltonian$.

As an example, we consider the general case in which the initial Lagrangian has a quadratic, but potentially `off-diagonal' kinetic energy. Let $\bm{x} = \left(
x_1, \dots, x_n\right)^T$. The Lagrangian reads
\begin{equation}
\lagrangian (\bm{x},\dot{\bm{x}})= \frac{1}{2} \dot{\bm{x}}^T \mat{C} \dot{\bm{x}} - U(\bm{x}) =\frac{1}{2}\sum_{i,j=1}^n \mat{C}_{ij} \dot{x}_i  \dot{x}_j-U(x_1, \ldots, x_n),
\label{eq:standard-form}
\end{equation}
where the matrix $\mat{C}$ is symmetric and we assume $\mat{C} > 0$, i.e., $\mat{C}$ has eigenvalues $d_i > 0$, $i=1, \ldots,n$. In this case, we can determine that $\bm{p} = \mat{C} \dot{\bm{x}}$ so that $\dot{\bm{x}}= \mat{C}^{-1}  \bm{p}$ at the supremum. Inserting this equality in Eq.~\eqref{eq:defH} gives compactly
\begin{equation}
\hamiltonian (\bm{x}, \bm{p} )=\frac{1}{2} \bm{p}^T \mat{C}^{-1} \bm{p} + U(\bm{x}).
\label{eq:Hnondiag}
\end{equation}

 In a different analysis of the same problem we could have first switched to new independent variables $\bm{y}=\mat{S} \bm{x}$, with $\mat{S}$ the matrix that diagonalizes $\mat{C}$, in order to obtain
\begin{equation}
\lagrangian (\bm{y}, \dot{\bm{y}})=\frac{1}{2} \dot{\bm{y}}^T \bm{D} \dot{\bm{y}} - \underbrace{U(\mat{S}^{-1} \bm{y})}_{\tilde{U}(\bm{y})}  = \frac{1}{2}\sum_{i=1}^n d_i \dot{y}_i ^2-\tilde{U}(y_1, \dots, y_n) ,
\end{equation}
with the diagonal matrix $\bm{D} = 
\mat{S}^T \mat{C} \mat{S}$ with diagonal entries $d_i$. Only then we could have introduced the conjugate variable as $\bm{p}=\partial \lagrangian/\partial \dot{\bm{y}} = \mat{D} \dot{\bm{y}}$ so that 
\begin{equation}
\hamiltonian(\bm{y}, \bm{p})= \frac{1}{2}\bm{p}^T \mat{D}^{-1} \bm{p} + \tilde{U}(\bm{y}) = \sum_{i=1}^n \frac{p_i^2}{2d_i} + \tilde{U}(y_1, \ldots, y_n).
\label{eq:Hdiag}
\end{equation}

It is instructive to discuss the case when $\mat{C} \ge 0$, meaning that some of the eigenvalues $d_i$ of $\mat{C}$ are zero, thus preventing $\mat{C}$ from being invertible. In this case, the associated variable $y_i$ (with $d_i=0$) only occurs in the potential $\tilde{U}$. Such cases can classically be solved by realizing that all coordinates $y_i$ without kinetic energy contribution are simply to be put at a value for which $\partial \tilde{U}(y_1,\ldots,y_n)/\partial y_i=0$ as this is the Euler-Lagrange equation for those coordinates. When the potential is convex, this configuration is one which minimizes the potential $\tilde{U}(y_1, \ldots,y_n)$ energy; more generally, one sets these values to a local minimum \footnote{Local maxima are unstable when $d_i$ is only slightly non-zero.}.

\begin{Exercise}[label=exc:relat]The Lagrangian of a relativistic particle with mass $m$ and spatial coordinate ${\bf x}=(x_1, x_2, x_3)$, is \[\lagrangian=-m\mathfrak{c}^2 \sqrt{1-\frac{\sum_{i=1}^3\dot{x}_i^2}{\mathfrak{c}^2}}-U(x_1, x_2, x_3),
\]
where $\mathfrak{c}$ is the speed of light and $U(x_1, x_2, x_3)$ is some potential. Verify that $\lagrangian$ is convex in $\dot{x_i}$ and derive the Hamiltonian. 
\end{Exercise}

\begin{Answer}[ref={exc:relat}]
By plotting $f(x)=-\sqrt{1-x^2}$ for $-1 \leq x \leq 1$ one can verify $f(x)$ is convex and so is $\lagrangian$. We have relativistic momenta $p_i=\frac{\partial \lagrangian}{\partial \dot{x}_i}=\frac{m \dot{x}_i}{\sqrt{1- v^2/\mathfrak{c}^2}}$ or $\sum_{i=1}^3\dot{x}_i^2=v^2=\frac{p^2}{m^2+p^2/\mathfrak{c}^2}$, so that $\hamiltonian=\frac{m v^2}{\sqrt{1-v^2/\mathfrak{c}^2}}+m\mathfrak{c}^2 \sqrt{1-v^2/\mathfrak{c}^2}+U(x_1, x_2, x_3)=\sqrt{p^2 \mathfrak{c}^2+m^2 \mathfrak{c}^4}+U(x_1, x_2, x_3)$.
\end{Answer}

In various physical situations, the Lagrangian is explicitly time-dependent, i.e., $\lagrangian({\bf x}, {\bf \dot{x}}, t)$. For example, the potential is one that is induced by adding a time-dependent driving force $F(t)=A\cos(\omega t)$ to a mass-spring system with position $x$. The potential will be $U(x)=\frac{1}{2} K x^2 - F(t) x$ corresponding to the total force $F_{\rm tot}(t)=-\partial U/\partial x=-K x+F(t)$. In such cases, when the explicit time dependence occurs in the potential energy $U$, it is simply carried over in the Hamiltonian $\hamiltonian(t)$, i.e.~for the driven mass-spring system we have $\hamiltonian=\frac{p^2}{2m}+\frac{1}{2} K x^2 - F(t) x$.

On the other hand, when the explicit time dependence affects the kinetic energy, it enters the Euler-Lagrange equation, Eq.~\eqref{eq:EL}, in its total time-derivative $d/dt$. For example, a mechanical system with a time-dependent mass $m(t)$ and a potential $U(x)$ obeys the Euler-Lagrange equation $m(t)\ddot{x}+\dot{x} \dot{m}= d U/dx$. In addition, when we switch from Lagrangian to Hamiltonian, such time dependence can enter the definition of the conjugate variables and through this makes the definition of the Hamiltonian time-dependent. For example, when treating time-dependent fluxes, the conjugate charge in Eq.~\eqref{eq:cc-timedep} in Section~\ref{sec:ex-flux} is explicitly time-dependent through $\Phi_{\rm ext}(t)$. For conjugate variables which are time-dependent (time is simply a classical parameter), one can still write down the Poisson bracket as in Eq.~\eqref{eq:poiss} and apply quantization.

\section{Poisson bracket and quantization}

Using conjugate variables one can define a Poisson bracket for functions $f(\bm{x}, \bm{p})$ and 
$g(\bm{x}, \bm{p})$ as follows
\begin{equation}
\{f, g \}=\sum_{i=1}^n \left[\frac{\partial f}{\partial x_i} \frac{\partial g}{\partial p_i}-\frac{\partial f}{\partial p_i} \frac{\partial g}{\partial x_i}\right].
\label{eq:poiss}
\end{equation}
It follows from this definition that
\begin{equation}
\{x_i, p_j\}=\delta_{ij}, \{x_i, x_j\} =\{p_i, p_j\} =0.
\label{eq:cc}
\end{equation}
One can also apply the Poisson bracket to the Hamiltonian function defined in Eq.~\eqref{eq:defH} and obtain  
\begin{equation}\label{eq:hameq1}
\{x_i, \hamiltonian\}=\frac{\partial \hamiltonian}{\partial p_i}=\dot{x}_i,
\end{equation}
i.e., the Poisson bracket with the Hamiltonian determines the time dynamics. 
The last equality follows from how $\hamiltonian$ is defined as a function of $p_i$. Similarly, we have 
\begin{equation}\label{eq:hameq2}
\{p_i, \hamiltonian\}=-\frac{\partial \hamiltonian}{\partial x_i}=\dot{p}_i,
\end{equation}
where we have used that $-\partial \hamiltonian/\partial x_i=\partial \lagrangian/\partial x_i$. In the last equality we used the Euler-Lagrange equation and the definition of $p_i$. Eqs.~\eqref{eq:hameq1} and~\eqref{eq:hameq2} are known as Hamilton's equations.

Quantization means that we promote the conjugate variables $x_i \in \mathbb{R}$ and $p_j \in \mathbb{R}$ to operators $x_i \rightarrow \hat{x}_i$, $p_j \rightarrow \hat{p}_j$, $\hamiltonian \rightarrow H$. Paul Dirac came up with the idea that the Poisson bracket should be replaced by the commutator due to `their great similarity', see minutes 15.30-23.00 in \href{https://www.youtube.com/watch?v=2GwctBldBvU}{Dirac Lecture 1 of 4 on Quantum Mechanics}. The mapping is
\begin{equation}
\{f,g\} \rightarrow \frac{1}{i \hbar} [\hat{f}, \hat{g}].
\end{equation}
This directly gives the canonical commutation relation 
\begin{equation}
[\hat{x}_i,\hat{p}_j]=i\hbar \delta_{ij}\mathds{1}.
\end{equation}
through Eq.~\eqref{eq:cc}. In addition, the quantization of Eq.~\eqref{eq:hameq1} gives $[\hat{x}_i(t), H]=i \hbar \frac{d \hat{x}_i(t)}{dt}$ where $\hat{x}_i(t)$ is the time-evolved Heisenberg operator $\hat{x}_i(t) \equiv e^{i t H/\hbar} \hat{x}_i e^{-i t H/\hbar}$.

\begin{Exercise}[label=exc:field]For a scalar (continuum) field of variables $\xi(x_{\mu})$ depending on space-time coordinates $x_{\mu}=(t,x,y,z)$, the (non-relativistic) action equals $\mathcal{S}=\int d^3x \int dt\, \lagrangian(\xi(x_{\mu}),\frac{\partial \xi}{\partial x_{\mu}})$.
Show that the principle of minimal action, $\frac{\delta \mathcal{S}}{\delta \xi(x_{\mu})}=0$ for the field configuration $\xi(x_{\mu})$, implies that $\xi(x_{\mu})$ obeys the equation of motion
\begin{equation}
\frac{\partial \lagrangian}{\partial \xi}-\frac{\partial}{\partial x_{\mu}}\frac{\partial \lagrangian}{\partial (\frac{\partial \xi}{\partial x_{\mu}})}=0.
\label{eq:ELfield}
\end{equation}
(Comment: non-relativistic means that we make no distinction between $x_{\mu}$ and $x^{\mu}$ while these are usually related as $x_{\mu}=g_{\mu \nu} x^{\nu}$ with summation implied and $g_{\mu \nu}$ is the metric, nor is the action here manifestly Lorentz-invariant.)
\end{Exercise}

\begin{Answer}[ref={exc:field}]
If we vary the field locally, i.e.,~$\xi'(x_{\mu})=\xi(x_{\mu})+\delta \xi(x_{\mu})$, it leads to \[
\mathcal{S}'=\int d^3x \int dt\,\lagrangian\left(\xi(x_{\mu})+\delta \xi,\frac{\partial (\xi+\delta \xi)}{\partial x_{\mu}}\right)=\mathcal{S}+\int d^3x \int dt\,\left[\frac{\partial \lagrangian}{\partial \xi}\delta \xi +\sum_{\mu} \frac{\partial \lagrangian}{\partial(\frac{\partial \xi}{\partial x_{\mu}})}\frac{\partial (\delta \xi)}{\partial x_{\mu}}\right],
\]
and partial integration of the last term,  assuming that $\delta \xi$ (and $\xi(x_{\mu})$) vanish at the boundaries of the space-time coordinates, we get
\[
\delta \mathcal{S}=\int d^3 x\int dt \,\delta \xi(x_{\mu}) \left[\frac{\partial \lagrangian}{\partial \xi}-\frac{\partial}{\partial x_{\mu}}\frac{\partial \lagrangian}{\partial (\frac{\partial \xi}{\partial x_{\mu}})}\right], 
\]
so that $\delta \mathcal{S}=0$ for any variation $\delta \xi(x_{\mu})$ leads to Eq.~\eqref{eq:ELfield}.
\end{Answer}

\chapter[Harmonic systems and beyond]{Harmonic systems and beyond: elimination of high-energy variables}
\label{app:anharmonic-born-opp}

In this appendix, we discuss how to extract the uncoupled normal modes from a classical Lagrangian of coupled modes in a quadratic potential. Such Lagrangians corresponding to `linear' systems ---the equations of motion are linear--- are ubiquitous in circuit QED. 
If the potential energy admits an expansion as a polynomial in the degrees of freedom and the terms beyond the quadratic potential are weak, one can first
 first solve the harmonic system and express the remaining nonlinear terms of the potential in terms of the normal modes of the harmonic system.
 
  In Section~\ref{sec:elim} we go beyond this and discuss the Born-Oppenheimer approximation, which allows one to eliminate `high-frequency' or `fast' variables to focus on obtaining low-energy effective dynamics with fewer degrees of freedom. This method can also be used to eliminate degrees of freedom in the Lagrangian which have very large kinetic energy (capacitance or `mass' going to zero).
  
\section{Normal modes of a harmonic system}

We first argue, through one simple Exercise, that any linear dependence in the potential energy on the variables $x_i$ can be absorbed in a quadratic dependence by a change of variables.

\begin{Exercise}
Given a time-independent Lagrangian $\lagrangian(\vect{x},\vect{\dot{x}})$ with potential
\begin{equation}
    U(\vect{x})= \vect{a}^T \vect{x}+\frac{1}{2} \vect{x}^T \mat{K} \vect{x},
\end{equation}
with symmetric matrix $\mat{K} > 0$ and any real vector $\vect{a}$, show how to write this as
\begin{equation}
    U'(\vect{x}')=\frac{1}{2} \vect{x}'^T \mat{K} \vect{x}',
\end{equation}
up to a constant, using $\vect{x}'=\vect{x}+\vect{b}$ for some $\vect{b}$.
\label{ex:lin}
\end{Exercise}

\begin{Answer}[ref={ex:lin}]
Transform $\vect{x}'=\vect{x}+\vect{b}$ so that $\vect{\dot{x}}'=\vect{\dot{x}}$ with $\vect{b}=\mat{K^{-1}}\vect{a}$, using the fact that $\mat{K}$ is invertible and symmetric, and drop the constant term $\frac{1}{2}\vect{b}^T \vect{b}$.
\end{Answer}

Now consider some linear dependence on $\vect{\dot{x}}$ in the Lagrangian $\lagrangian$, i.e.,~
\begin{equation}
    \lagrangian= \frac{1}{2} \vect{\dot{x}'}^T \mat{C} \vect{\dot{x}'}+\vect{a}^T \vect{\dot{x}}-U(\vect{x}),
    \label{eq:a}
\end{equation}
where $U(\vect{x})$ is arbitrary. By choosing $\lagrangian'=\lagrangian-\vect{a}^T \vect{\dot{x}}$, or $f(x_1, \ldots, x_n)=-\sum_i a_i x_i$ in Eq.~\eqref{eq:gauge}, such dependence can be gauged away entirely.

Hence besides special cases, e.g. when $\vect{a}$ in Eq.~\eqref{eq:a} is explicitly time-dependent or there are magnetic fields or gyrators at play, we can assume that the Lagrangian is of the following form
\begin{equation}
\label{eq::lagr_in}
\lagrangian (\vect{x}, \vect{\dot{x}}) = \frac{1}{2} \vect{\dot{x}}^T \mat{C} \vect{\dot{x}}-U(\vect{x}).
\end{equation}

The potential is assumed to be of the form
\begin{equation}
 U(\vect{x})=\frac{1}{2}\sum_{i,j} K_{ij} x_i x_j+U_{\rm nl}(x_1, \ldots, x_n),   
 \label{eq:upot}
\end{equation} where $\mat{K}\geq 0$ and $U_{\rm nl}$ is a higher-order (cubic or more) polynomial in the $x_i$. 
We thus assume that the potential $U(\vect{x})$ can be expanded as a polynomial in the variables $x_i$ and that we have transformed away any linear part.

\subsection{Diagonalization}
\label{sec:diag}

If $U_{\rm nl}$ in Eq.~\eqref{eq:upot} is zero, the system is harmonic. We then show how to transform the variables $\vect{\dot{x}}$ and $\vect{x}$ such that they represent a set of uncoupled harmonic oscillators, each with resonant frequency $\omega_i$, see Eq.~\eqref{eq:lagr_lin}.
We will perform the appropriate change of variables at the Lagrangian level so we don't have to worry about the preservation of canonical commutation relations, which will be automatically imposed when we define the conjugate variables.

We can write the quadratic part of the Lagrangian compactly as
\begin{equation}\label{eq::lagrl}
\lagrangian_{\rm lin}(\vect{x}, \vect{\dot{x}}) = \frac{1}{2} \vect{\dot{x}}^T \mat{C} \vect{\dot{x}} - \frac{1}{2} \vect{x}^T \mat{K} \vect{x}.
\end{equation}
The Euler-Lagrange equation, Eq.~\eqref{eq:EL}, associated with the Lagrangian $\lagrangian_{\rm lin}$ in Eq.~\eqref{eq::lagrl}, can be written as
\begin{equation}\label{eq::lagrl_el}
\mat{C} \frac{d^2 \vect{x}}{dt^2} = - \mat{K} \vect{x} \implies  \frac{d^2 \vect{x}}{dt^2} = -\mat{C}^{-1} \mat{K} \vect{x},
\end{equation}
and is manifestly linear in the variables (hence the nomenclature linear or harmonic systems).

Solving Eq.~\eqref{eq::lagrl_el} is equivalent to diagonalizing the matrix $\mat{C}^{-1} \mat{K}$. We remark that $\mat{C}^{-1} \mat{K}$ is not necessarily symmetric. We make use of the following theorem (adapted from Ref.~\cite{horn-johnson}, page 485) to diagonalize $\mat{C}^{-1} \mat{K}$:
\begin{thm}\label{th::diag}
Let $\mat{A}, \mat{B}$ be symmetric, real $n \times n$ matrices. 
\begin{enumerate}
\item If $\mat{A} > 0$, then there is a non-singular, real $n \times n$ matrix $\mat{S}$ such that $\mat{A} = \mat{S} \mat{S}^T$ and $\mat{B} = (\mat{S}^{-1})^{T} \mat{\Lambda} \mat{S}^{-1}$, where $\Lambda$ is a real, diagonal matrix.
\item If $\mat{A} > 0$, then $\mat{AB}$ is diagonalizable and its diagonal matrix is $\mat{\Lambda}$.
\end{enumerate} 
\end{thm}
\begin{proof}
The matrix $\mat{A}$ is positive definite and thus there exists an orthogonal matrix $\mat{O}_A$ that diagonalizes $\mat{A}$, i.e., $\mat{O}_A^{T} \mat{A} \mat{O}_A = \mat{D}$ where $\mat{D}$ is a diagonal matrix with positive diagonal entries. Defining the matrix $\mat{T} = \mat{O}_A \mat{D}^{1/2}$, we have
$\mat{T}^{-1} \mat{A} (\mat{T}^{-1})^T = \mathds{1}$. The matrix $\mat{T}^T \mat{B} \mat{T}$ is real and symmetric and thus can also be diagonalized by an orthogonal matrix $\mat{O}$ such that $\mat{O}^{T} \mat{T}^T \mat{B} \mat{T} \mat{O} = \mat{\Lambda}$. Defining the nonsingular matrix $\mat{S}= \mat{T} \mat{O}$, we see that $\mat{S}^{-1} \mat{A} (\mat{S}^{-1})^T = \mathds{1}$ and $\mat{S}^T \mat{B} \mat{S} = \mat{\Lambda}$. The inversion of the previous formulas completes the proof of point 1. \\
Point 2 simply follows from point 1. In fact, we can write $\mat{A} = \mat{S} \mat{S}^T$ and $\mat{B} = (\mat{S}^{-1})^{T} \mat{\Lambda} \mat{S}^{-1}$ from which we get $\mat{A} \mat{B} = \mat{S} \mat{\Lambda} \mat{S}^{-1}$. 
\end{proof}
In addition, it can be shown that if $\mat{B}\geq 0$, then $\mat{\Lambda}$ has nonnegative eigenvalues (if $\mat{B} > 0$, $\mat{\Lambda}$ has positive eigenvalues). We omit the proof of this fact here, but we refer the interested reader to Theorems 4.5.8 and 7.6.1 in Ref.~\cite{horn-johnson}. Our matrices $\mat{C}^{-1}$ ($=\mat{A}$) and $\mat{K}$ ($=\mat{B}$) satisfy the assumptions of Theorem~\ref{th::diag} and let's call $\mat{\Lambda}=\mat{\omega^2}$ with nonnegative eigenvalues. Let $\mat{S}$ be the matrix that diagonalizes $\mat{C}^{-1} \mat{K}$, that is
\begin{equation}
\mat{S}^{-1} \mat{C}^{-1} \mat{K} \mat{S} = \mat{\Lambda},
\end{equation}
and such that $\mat{C}^{-1} = \mat{S} \mat{S}^T/C_0$ where $C_0>0$ is an arbitrary parameter with the same dimensions as the elements of $\mat{C}$ (we introduce the parameter $C_0$ to ensure that $\mat{S}$ is dimensionless). Notice that since $\mat{C}^{-1} \mat{K}$ is not necessarily symmetric, $\mat{S}$ is not necessarily an orthogonal matrix, although it must be invertible. 
Note that in case $\mat{K}$ does not have full rank, i.e.,~the quadratic potential is flat `in some directions', there will be corresponding zero eigenvalues in $\mat{\Lambda}$.

We now define new normal-mode variables as
\begin{equation}
\vect{X} = \mat{S}^{-1} \vect{x}.
\label{eq:mode-tr}
\end{equation}

Using the properties of Theorem~\ref{th::diag} the quadratic part of the Lagrangian Eq.~\eqref{eq::lagrl} can be written in terms of the normal-mode variables $\vect{X}$ as 
\begin{equation}\label{eq:lagr_lin}
\lagrangian_{\rm lin} (\vect{X}, \vect{\dot{X}}) = \frac{C_0}{2} \vect{\dot{X}}^T \vect{\dot{X}} - \frac{C_0}{2} \vect{X}^T \mat{\Lambda} \vect{X} =  \sum_{i=1}^n  \left[\frac{C_0}{2} \dot{X}_i^2 - \frac{C_0}{2} \omega_i^2 X_i^2\right],
\end{equation}
which is the Lagrangian of a collection of uncoupled harmonic oscillators with resonant frequencies $\omega_i$. 

Defining the conjugates to these normal modes as 
\begin{equation}
\vect{P} = \frac{\partial \mathcal{L}}{\partial \vect{\dot{X}}} = C_0 \vect{\dot{X}},
\end{equation}
we also immediately obtain the Hamiltonian in terms of the normal modes
\begin{equation}
\hamiltonian (\vect{X}, \vect{P}) = \frac{1}{2 C_0} \vect{P}^T \vect{P} + \frac{1}{2} C_0 \vect{X}^T \mat{\Lambda} \vect{X}. 
\end{equation}

Quantization implies that $[\hat{X}_i, \hat{P}_j]=i\hbar \delta_{ij}$. 
By introducing annihilation operators \[\hat{a}_i=\sqrt{\frac{C_0 \omega_i}{2\hbar}}\left(\hat{X}_i+\frac{i}{C_0 \omega_i}\hat{P}_i\right)
\]
with $[\hat{a}_i, \hat{a}_j^{\dagger}]=\delta_{ij} \mathds{1}$ we can write
the quantized Hamiltonian as
\begin{equation}
    H=\sum_{i=1}^n \hbar \omega_i \left(\hat{a}_i^{\dagger} \hat{a}_i+\frac{1}{2}\right).
\end{equation}

Now we can include the non-linear term directly into the classical Hamiltonian. Using the normal modes and their conjugate momenta, we have
\begin{equation}
\hamiltonian (\vect{X}, \vect{P}) = \frac{1}{2 C_0} \vect{P}^T \vect{P} + \frac{1}{2} C_0 \vect{X}^T \mat{\Lambda} \vect{X} + U_{\rm nl}(\vect{S} \vect{X}).
\end{equation}
Using $\hat{X}_i=\sqrt{\frac{\hbar}{2C_0 \omega_i}}(a_i +a_i^{\dagger})$ we can write the nonlinearity $U_{\rm nl}(\sum_j S_{ij} X_j)$ in terms of creation and annihilation operators; see also Section~\ref{sec:bb} and Appendix~\ref{app:norm_mode} for the use of this diagonalization method in black-box quantization in circuit QED.

\section{Eliminating high-frequency modes}
\label{sec:elim}

Having determined the normal modes of a system, 
we imagine that for some subset of modes, the frequencies $\omega_i$ are much larger than for the other modes. We can call these the high-energy modes $X_{i, {\rm high}}$ (versus low-energy modes $X_{i, {\rm low}}$). Then, as a very simple approximation, one may choose to put these high-energy modes directly in their ground vacuum state, and replace the operators $X_{i, {\rm high}}$ and $P_{i, {\rm high}}$ by taking expectation values with respect to the vacuum state. We then work in a reduced Hilbert space where one considers only the remaining low-energy modes. Such a substitution will of course also affect the nonlinearity $U_{\rm nl}(\mat{S}\vect{X})$ in Eq.~\eqref{eq:upot}, which will reduce to some nonlinearity on the remaining `active' modes. In this simple approach the effect of the nonlinearity is not included in determining the frequencies of the high-energy or the low-energy modes.

A more accurate approximation which goes by the name of Born-Oppenheimer approximation is as follows \footnote{See Ref.~\cite{bornOppenheimer1927} for the original work by Born and Oppenheimer (in German).}. This approximation is used in molecular physics where the `fast' high-energy (light mass) variables are the motional degrees of freedom of the electrons and the slow (heavy) variables are those of the nuclei. In what follows, we adapt it to our case of harmonic oscillators coupled by a nonlinear potential.
We first identify a set of low-frequency modes with positions $\{X_{i,{\rm low}}\}$ and high-frequency modes to be eliminated with positions $\{X_{i, {\rm high}}\}$. Then we return to the Lagrangian, obtained by applying Eq.~\eqref{eq:mode-tr} to Eq.~\eqref{eq:upot} and collecting the linear term in Eq.~\eqref{eq:lagr_lin}
\begin{equation}
    \lagrangian (\vect{X}, \vect{\dot{X}}) = \frac{C_0}{2} \vect{\dot{X}}^T \vect{\dot{X}} - \frac{C_0}{2} \vect{X}^T \mat{\omega}^2 \vect{X} -U_{\rm nl}(\mat{S} \mat{X}).
    \label{eq:laghighlow}
\end{equation}
We expand this Lagrangian to second order in the variables $X_{i,{\rm high}}$, keeping all dependence on $X_{j,{\rm low}}$:
\begin{equation}
   U_{\rm nl}(\mat{S} \mat{X})=\sum_j c_j(\{X_{i, {\rm low}}\}) X_{j,{\rm high}}+\frac{1}{2} \vect{X}_{\rm high}^T \mat{N}(\{X_{i, {\rm low}}\})\vect{X}_{\rm high}+\tilde{U}_{\rm nl}(\{X_{i,\rm low}\}) +O(|\vect{X}_{\rm high}|^3),
   \label{eq:nl-exp}
\end{equation} 
where $c_j(\{X_{i, {\rm low}\}}\})$ are some coefficients which depend on $X_{i, {\rm low}}$.
Here we have collected all terms which solely depend on $X_{i, {\rm low}}$ in $\tilde{U}_{\rm nl}(\{X_{i,\rm low}\})$. Thus the `positions' of the low-frequency (slow) modes create a potential for the high-frequency (fast) modes \footnote{If one expands around a point where $\mat{N}$ has negative eigenvalues, one has to include higher-order nonlinear terms to have a confining potential.}. With this expansion, the form of the Lagrangian for the vector of variables $\vect{X}_{\rm high}$ is that of a harmonic system, as the linear term in $\vect{X}_{\rm high}$ can be transformed away into a quadratic term (see Exercise~\ref{ex:lin}), assuming that the symmetric matrix $\mat{N}$ is invertible. More generally, in the Born-Oppenheimer approximation one solves for the ground-state problem of the high-energy variables (a possibly non-harmonic system).

Neglecting the higher-order dependence on $\vect{X}_{\rm high}$, we can diagonalize the Lagrangian for the high-energy modes and obtain a form as in Eq.~\eqref{eq:lagr_lin}, where the frequencies $\tilde{\omega}_k$ of these new uncoupled high-energy modes will depend on $\vect{X}_{\rm low}$. Then when we put these high-frequency modes in their ground state, they will contribute their vacuum energy $\frac{1}{2}\sum_{k, {\rm high}}  \hbar \tilde{\omega}_k(\{\hat{X}_{i, {\rm low}}\})$. The final Hamiltonian in this approximation reads
\begin{equation}
    H \approx \sum_{i, {\rm low}} \hbar \omega_i \left(\hat{a}_{i,{\rm low}}^{\dagger} \hat{a}_{i,{\rm low}}+\frac{1}{2}\right)+\frac{1}{2}\sum_{k, {\rm high}}  \hbar \tilde{\omega}_k(\{\hat{X}_{i, {\rm low}}\})+ \tilde{U}_{\rm nl}(\{\hat{X}_{i, {\rm low}}\}).
    \label{eq:B1H}
\end{equation}
Here the first term comes about through the first two terms in Eq.~\eqref{eq:laghighlow} only for the variables $X_{i,{\rm low}}$ with their associated operators $\hat{a}_{i,{\rm low}}$. 
Thus the high-energy modes are eliminated, inducing an effective (nonlinear) coupling between the previously decoupled low-energy modes. This method was used for handling electrical circuits in Ref.~\cite{DBK:BO}.

\begin{figure}[h]  
	\centering
	\includegraphics[height=6cm]{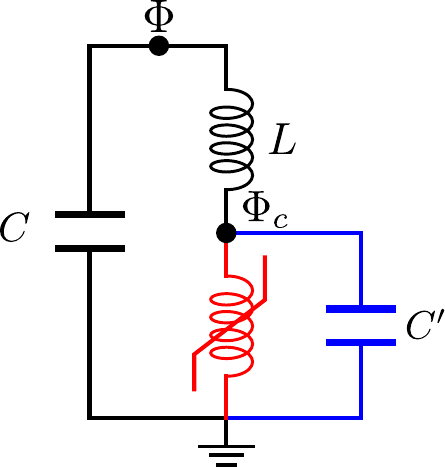}
	\caption{Series combination of a linear inductance $L$ and a nonlinear inductor (red) in parallel to a shunting capacitance $C$. The blue branch highlights the intrinsic capacitance $C'$ of the nonlinear inductor, which we consider to be either vanishingly small (regular case) or absent (singular case).}
	\label{fig_Non_Linear_Inductance}
\end{figure}

\subsection{Illustration of the Born-Oppenheimer method}
\label{sec:elim-BO}

We now go on to a simple example that illustrates the application of the Born-Oppenheimer approach. We consider the series combination of a linear inductance $L$ and a generic nonlinear inductor with intrinsic capacitance $C'$, all in parallel with a total shunting capacitance $C$; see Fig.~\ref{fig_Non_Linear_Inductance}. If $C'=0$ we have a singular circuit for which a Hamiltonian cannot be straightforwardly derived, as one degree of freedom only enters the potential energy  (see the discussion in Section~\ref{subsec:patho}). 

Here we imagine $C'$ is non-zero but very small. Then, we can straightforwardly transform the Lagragian 
\begin{equation}\label{eq_Lagrangian_non_linear_inductance}
	\mathcal{L} = \frac{C\dot{\Phi}^2}{2} + \frac{C'\dot{\Phi}_c^2}{2} - \frac{(\Phi-\Phi_c)^2}{2L} - U_{\rm nl}(\Phi_c),
\end{equation}
to obtain the quantized, circuit Hamiltonian:
\begin{equation}\label{eq_Hamiltonian_r}
	H_r = \frac{\hat{Q}^2}{2C} + \frac{\hat{Q}_c^2}{2C'} + \frac{(\hat{\Phi}-\hat{\Phi}_c)^2}{2L} + U_{\rm nl}(\hat{\Phi}_c).
\end{equation}
 The divergence of the second term of Eq.~\eqref{eq_Hamiltonian_r} for $C' \rightarrow 0$ indicates that $\Phi_c$ will be the `fast' high-energy variable.

The Born-Oppenheimer approximation \cite{DBK:BO} will allow us to derive an effective low-energy Hamiltonian as a function of $\Phi$ and $Q$ only. To this end, we first solve the stationary Schr\"odinger equation associated with the fast degree of freedom, $\Phi_c$, for fixed values of $\Phi$ and $Q$. Thus, we identify the high-energy part of $H_r$, exhibiting fast dynamics classically, as
\begin{equation}\label{eq_def_H_fast}
    H_\text{fast} = \frac{\hat{Q}_c^2}{2C'} + \frac{(\hat{\Phi}-\hat{\Phi}_c)^2}{2L} + U_{\rm nl}(\hat{\Phi}_c),
\end{equation}
and we solve
\begin{equation}\label{eq_fast_SE_BO}
	H_\text{fast} \psi_{\Phi,n}(\Phi_c) = E_{\Phi,n} \psi_{\Phi,n}(\Phi_c),
\end{equation}
for the eigenfunctions $\psi_{\Phi,n}(\Phi_c)$ and the associated eigenenergies $E_{\Phi,n}$, which both are labeled by $n \in \mathbb{N}_0$ and parametrized by $\Phi$. The ground-state energy ($n=0$) is then considered as an effective low-energy potential for the low-energy (or `slow') variable $\Phi$, whose dynamics is captured by the effective Hamiltonian
\begin{equation}\label{eq_H_r_eff}
	H_{r,\text{eff}} = \frac{\hat{Q}^2}{2C} + U_\text{BO}(\hat{\Phi}),
\end{equation}
with the Born-Oppenheimer potential given by the fast-variable ground-state energy:
\begin{equation}\label{eq_def_BO_potential}
    U_\text{BO}(\Phi) = E_{\Phi,n=0}-E_{\Phi=0,n=0}.
\end{equation}
It is handy to introduce here an energy offset of $U_\text{BO}(\Phi)$ such that $U_\text{BO}(0) = 0$ in order to avoid additive constants that are divergent as $C'\rightarrow 0$.

We will shortly make a particular choice for $U_{\rm nl}$. Of course, the nonlinear potential that we know is $-E_J\cos(\Phi_c)$. Since we are here only trying to illustrate the mathematical treatment of Born-Oppenheimer theory, we will make another choice that makes the analysis easier. For many situations, including when we have the Josephson potential, Born-Oppenheimer becomes a purely numerical exercise, which we want to avoid here. At the end of this section we will explain the relation of our simplified $U_{\rm nl}$ to the realistic case of the Josephson potential in Fig.~\ref{fig:patho}.

In what follows, it is convenient to think about the $\Phi$ degree of freedom as a classical variable, while only $\Phi_c$ is quantized. Thus, we will not use hats on the $\Phi$ variable, but we will quantize it only at the end of the procedure. Before fixing a particular $U_{\rm nl}$, it is informative to rescale the fast Hamiltonian. We introduce the ${\rm LC'}$ resonator frequency and a flux zero-point fluctuation parameter, defined as
\begin{equation}\label{eq_defs_ZPF_and_omega_LC}
    \omega'_r = \frac{1}{\sqrt{LC'}},
    \qquad
    \delta\Phi=\sqrt{\hbar}\biggl(\frac{L}{C'} \biggr)^{1/4},
\end{equation}
respectively, and we express $H_\text{fast}$ as
\begin{equation}\label{eq_H_fast_rewritten}
    H_\text{fast} =
    \hbar \omega'_r
    \left[
    \frac{\hat{p}^2+\hat{z}^2}{2} + \frac{U_{\rm nl}(\delta\Phi  \hat{z}+ \Phi)}{\hbar \omega'_r}
    \right].
\end{equation}
Note that the parameter $\delta\Phi$ is essentially $\Phi_{\rm zpf}$ of Eq.~\eqref{eq:zpf}, differing by factor $\sqrt{2}$.
The dimensionless conjugate variables $\hat{z}$ and $\hat{p}$ are defined as $\hat{z}=(\hat{\Phi}_c- \Phi)/\delta\Phi$ and $\hat{p} = \hat{Q}_c \delta\Phi / \hbar$, and they satisfy the canonical commutation relation $[\hat{z}, \hat{p}]=i\mathds{1}$.

We define the parameter $\epsilon = 1 / \sqrt{\hbar\omega'_r} =\sqrt[4]{LC'}/\sqrt{\hbar}= \sqrt{L}/\delta\Phi$, and we divide out the prefactor in Eq.~\eqref{eq_H_fast_rewritten}, obtaining
\begin{equation}\label{eq_eps2_H_fast_type_1}
    \epsilon^2 H_\text{fast} = H_0  + \epsilon^2 U_{\rm nl} \left( \frac{\sqrt{L}}{\epsilon} \hat{z} +\Phi \right),
\end{equation}
with the dimensionless harmonic-oscillator Hamiltonian $H_0 = (\hat{p}^2+\hat{z}^2)/2$. Note that we are looking to determine the $\Phi$ dependence of the ground state of this Hamiltonian, which appears only in the final term. Furthermore, we note that since $\epsilon$ can be viewed as a small parameter when we consider $C'\rightarrow 0$, we may be able to progress by treating the final term in Eq.~\eqref{eq_eps2_H_fast_type_1} as a perturbation of the ground state of $H_0$ (as in the expansion in Eq.~\eqref{eq:nl-exp}). The ground-state energy of $H_0$ is just $1/2$, so we will seek small corrections in $\epsilon$ to this. Note that we will have to divide out the $\epsilon^2$ from the left-hand side of Eq.~\eqref{eq_eps2_H_fast_type_1} in the end, so that the ground-state energy of $H_{\rm fast}$ will diverge; this is correct, and due to zero-point fluctuations. But the $\Phi$ dependence of this ground-state energy will be in the perturbative corrections, and will not diverge for $U_{\rm nl}$ functions of interest.

Of course, the case of most interest is the Josephson potential, but the perturbative calculation is not feasible to work out in closed form in this case, so we take another model $U_{\rm nl}$ potential, namely
\begin{equation}
       U_{\rm nl}(\Phi)=E_{\rm nl}\left|\frac{\Phi}{\Phi_0}\right|^{\nicefrac{3}{2}}.
\end{equation}
For our choice, the perturbation in Eq.~\eqref{eq_eps2_H_fast_type_1} becomes explicitly
\begin{equation}
  \epsilon^2 U_{\rm nl} \left( \frac{\sqrt{L}}{\epsilon} \hat{z}+ \Phi \right)=\epsilon^{\nicefrac{1}{2}} E_{\rm nl} \frac{L^{\nicefrac{3}{4}}}{\Phi_0^{\nicefrac{3}{2}}}   \left|\hat{z}+\frac{\epsilon }{\sqrt{L}} \Phi \right|^{\nicefrac{3}{2}}.  
\end{equation}
To proceed with first-order perturbation theory, we compute

\begin{multline}
\bra{0}\epsilon^{\nicefrac{1}{2}} E_{\rm nl} \frac{L^{\nicefrac{3}{4}}}{\Phi_0^{\nicefrac{3}{2}}}   \left|\hat{z}+\frac{\epsilon }{\sqrt{L}} \Phi\right|^{\nicefrac{3}{2}} \ket{0}=\epsilon^{\nicefrac{1}{2}} E_{\rm nl}\frac{L^{\nicefrac{3}{4}}}{\pi^{\nicefrac{1}{2}}\Phi_0^{\nicefrac{3}{2}}}\int_{-\infty}^\infty dz  e^{-z^2}  \left|z+\frac{\epsilon}{\sqrt{L}} \Phi \right|^{\nicefrac{3}{2}} \\
=\epsilon^{\nicefrac{1}{2}} E_{\rm nl}\frac{L^{\nicefrac{3}{4}}}{\pi^{\nicefrac{1}{2}}\Phi_0^{\nicefrac{3}{2}}} \frac{\pi\sqrt{\epsilon\Phi}}{2\sqrt[4]{L}} e^{-\frac{\epsilon^2\Phi^2}{4L}} \left[\frac{\epsilon^2\Phi^2}{L}
   I_{\frac{3}{4}}\left(\frac{\epsilon^2\Phi^2}{4L}\right)+\left(\frac{\epsilon^2\Phi^2}{L}+1\right)
   I_{-\frac{1}{4}}\left(\frac{\epsilon^2\Phi^2}{4L}\right)\right].
   \label{reallan}
\end{multline}
Here we see the appearance of modified Bessel functions $I_n$ (integrals courtesy of Wolfram Mathematica software). The Bessel functions here are non-analytic (fourth root) in their arguments, but because these arguments go like $\Phi^2$, and because of the additional $\sqrt{\Phi}$ factor, Eq.~(\ref{reallan}) is analytic in $\Phi^2$. Expanding in a Taylor series, we get
\begin{equation}
\bra{0}\epsilon^{\nicefrac{1}{2}} E_{\rm nl}\frac{L^{\nicefrac{3}{4}}}{\Phi_0^{\nicefrac{3}{2}}}\left|\hat{z}+\frac{\epsilon}{\sqrt{L}} \Phi \right|^{\nicefrac{3}{2}}\ket{0}=\epsilon^{\nicefrac{1}{2}} E_{\rm nl}\frac{L^{\nicefrac{3}{4}}}{\pi^{\nicefrac{1}{2}}\Phi_0^{\nicefrac{3}{2}}}\left({\mbox{const.}}+\frac{3\pi}{4\sqrt[4]{2}\Gamma(3/4}\frac{\epsilon^2\Phi^2}{L}+O\left((\epsilon\Phi)^4\right)\right).
\end{equation}
So, dividing by $\epsilon^2$, we get the following form for the Born-Oppenheimer potential:
\begin{equation}
    U_\text{BO}(\Phi) = E_{\Phi,0}-E_{\Phi=0,0}\approx \frac{3\pi^{\nicefrac{1}{2}}}{4\sqrt[4]{2}\Gamma({\nicefrac{3}{4}})}\frac{E_{\rm nl}}{L^{\nicefrac{1}{4}}\Phi_0^{\nicefrac{3}{2}}}\epsilon^{\nicefrac{1}{2}}\Phi^2=\frac{3\pi^{\nicefrac{1}{2}}}{4\sqrt[4]{2}\Gamma({\nicefrac{3}{4}})}\frac{E_{\rm nl}(C')^{\nicefrac{1}{8}}}{\hbar^{\nicefrac{1}{4}}L^{\nicefrac{1}{8}}\Phi_0^{\nicefrac{3}{2}}}\Phi^2.
\end{equation}
By comparing with Eq.~\eqref{eq:chargeflux_aadag}, the reader can see that this expression for the Born-Oppenheimer potential is dimensionally consistent. Examination of the Taylor expansion shows that this quadratic expression is valid over a range that diverges as $\epsilon\rightarrow 0$. Thus, one observes that the anharmonic potential has, due to quantum fluctuations, been turned into a harmonic effective potential for the slow coordinate. But perhaps more significant is that, due to the final $C'$ dependence, this whole potential is vanishing (very slowly!) as $C'\rightarrow 0$. Thus, the whole circuit from node $\Phi$ to ground is going to an open circuit (i.e., effectively infinite inductance). This is the most notable consequence of quantum fluctuations. In Ref.~\cite{rymarz:sing}, it was proven that this open-circuit behavior is universal, and applies also to the case where a Josephson junction is present as the nonlinear element in Fig.~\ref{fig_Non_Linear_Inductance}. The only thing that is required for this universal behavior to hold is that the nonlinear potential increases more slowly than $\Phi_c^2$ for large $|\Phi_c|$.  

\section{Normal modes and Cauer's construction}\label{app:norm_mode}

Building on the analysis in Appendix~\ref{sec:diag}, we now show that an arbitrary lossless, reciprocal impedance matrix can always be expanded as in Eq.~\eqref{eq:z_foster}, a result that goes by the name of Foster's theorem when we only have a single port~\cite{russer, foster1924}.  We provide a derivation of this result based on the Lagrangian formalism for electrical circuits that we have given in this book. A similar approach has been put forward also in Ref.~\cite{egusquiza2022}. The derivation makes clear the equivalence between the normal modes of a linear circuit and the Foster expansion of the impedance matrix. The Lagrangian formulation also allows the straightforward derivation of the Hamiltonian and consequently the quantization of the circuit. In addition, we will point out how nonlinear effects, due for instance to Josephson junctions coupled at the ports, can be straightforwardly introduced in the model, leading to the general Hamiltonian in Eq.~\eqref{eq:h_tr_circuit}. \par 
We start our derivation by considering the general case of the $N$-port network pictorially depicted in Fig.~\ref{fig:el_network}. As we assume the network to be LTI, it is completely characterized by its impulse-response matrix, i.e., its impedance $\bm{Z}(s)$ in the Laplace domain, and the network must consist only of interconnections of capacitances, inductances and mutual inductances (or ideal transformers). We follow the general procedure for obtaining the Lagrangian of a circuit described in Section~\ref{sec:canq_el}. We make the following observations: 
\begin{itemize}
\item We can always choose the port fluxes $\Phi_{p_n}$, $n=1, \dots, N$ defined in Eq.~\eqref{eq:port_fluxes} as the variables in the problem;
\item Additionally, there will be internal degrees of freedom that will generally be defined as the integral of the voltages across some arbitrary branches. These internal degrees of freedom depend on the topology of the specific circuit. We denote them as $\Phi_{I n'}$, with $ n' = 1, \dots, N_I$ and $N_I$ the total number of independent internal degrees of freedom.
\end{itemize} 

Given these facts, the Lagrangian of a linear, reciprocal and lossless network can always be written as
\begin{equation}\label{eq:lagrLTI}
\lagrangian_{\mathrm{lin}}(\vect{\dot{\theta}}; \vect{\theta}) = \frac{1}{2} \vect{\dot{\theta}}^T \mat{C} \vect{\dot{\theta}} - \frac{1}{2}\vect{\theta}^T \mat{K} \vect{\theta},
\end{equation} 
where we define the total vector of fluxes
\begin{equation}
\vect{\theta} = \begin{bmatrix}
\Phi_{p_1}, \dots, \Phi_{p_N}, \Phi_{I_1}, \dots, \Phi_{I_{N_I}}
\end{bmatrix}^T,
\end{equation}
i.e., $\bm{\theta}$ is a vector of length $M = N + N_I$ with $M$ the total number of independent variables in the problem. 

 We now note that Eq.~\eqref{eq:lagrLTI} is of the form of the Lagrangian in Eq.~\eqref{eq::lagrl}, and thus we can simply follow the derivation there. Defining the normal-mode variables (cf. Theorem \ref{th::diag})
\begin{equation}\label{eq:phi_theta}
\vect{\Phi} = \mat{S}^{-1} \vect{\theta},
\end{equation}
we can readily write 
\begin{equation}
\lagrangian_{\mathrm{lin}}(\dot{\vect{\Phi}}, \vect{\Phi}) = \sum_{m=1}^M  \left[\frac{C_0}{2} \dot{\Phi}_m^2 - \frac{C_0}{2} \omega_m^2 \Phi_m^2\right]. 
\end{equation}
As expected, the Lagrangian expressed as a function of the normal mode fluxes can be interpreted as a collection of $M$ uncoupled harmonic oscillators with frequency $\omega_m$ and (arbitrary) characteristic impedance given in Eq.~\eqref{eq:z0m}, using the reference capacitance $C_0$.

The previous derivation shows that such a network can always be replaced by an equivalent circuit shown in Fig.~\ref{fig::cauer_circuit}. In fact, the circuit has a simple interpretation. We have several uncoupled parallel LC oscillators that are associated with the normal modes. Eq.~\eqref{eq:phi_theta} implies that the port fluxes (or voltages) can be expressed as a linear combination of those associated with the normal modes, specifically
\begin{equation} \label{eq::phi_exp}
\Phi_{p_n} = \sum_{m =1}^M t_{mn} \Phi_m, \quad n=1, \dots, N,
\end{equation}
where the $t_{mn}=S_{mn} \in \mathbb{R}$ have the interpretation of (signed) turn ratios of the transformers in the circuit in Fig.~\ref{fig::cauer_circuit} and are the elements of the vectors $\bm{t}_m$ in Eq.~\eqref{eq:vec_turn}. 

It is now a simple exercise to show that the impedance matrix associated with the circuit in Fig.~\ref{fig::cauer_circuit} can be expanded as in Eq.~\eqref{eq:z_foster}. The impedance of the $m$th LC oscillator in the circuit is given by
\begin{equation}
Z_m (s) = \frac{1}{2 C_0} \biggl (\frac{1}{s + i \omega_m} + \frac{1}{s - i \omega_m} \biggr ).
\end{equation}
Let us now apply a generic vector of currents at the ports
\begin{equation}
\bm{I}(s)= \begin{bmatrix}
I_{p_1}(s) & I_{p_2}(s) & \dots & I_{p_N}(s)
\end{bmatrix}^T.
\end{equation}
The total current passing through the $m$th LC oscillator is given by
\begin{equation}
I_{m}(s) = \sum_{n=1}^{N} t_{mn} I_{P_n}(s) =\bm{t}_m^T \bm{I}(s).
\end{equation}
As a consequence we get the following voltage across the $m$th LC oscillator:
\begin{equation}
V_{m}(s) = Z_{m}(s) I_{m}(s) = \frac{1}{2 C_0} \biggl (\frac{1}{s + i \omega_m} + \frac{1}{s - i \omega_m} \biggr ) \bm{t}_m^T \bm{I}(s).
\end{equation}
The voltage drop across the $n$th port is then given as
\begin{equation}
V_{p_n}(s) = \sum_{m=1}^{M} t_{mn} V_{m}(s)= 
 \sum_{m=1}^{M} \frac{1}{2 C_0} \biggl (\frac{1}{s + i \omega_m} + \frac{1}{s - i \omega_m} \biggr ) t_{mn}  \bm{t}_m^T \bm{I}(s),
\end{equation}
and so we write the vector of port voltages as
\begin{equation}
\bm{V}(s) = \underbrace{\biggl(\sum_{m=1}^{M} \frac{1}{2 C_0} \biggl (\frac{1}{s + i \omega_m} + \frac{1}{s - i \omega_m} \biggr ) \bm{t}_m \bm{t}_m^T \biggr)}_{\bm{Z}(s)} \bm{I}(s),
\end{equation}
with matrix $\mat{R}_m=\vect{t}_m \vect{t}_m^T$, which shows the result. 

After quantizing the circuit, annihilation and creation operators for the normal modes can be introduced as usual, to write
\begin{equation}
\hat{\Phi}_m = \sqrt{\frac{\hbar Z_m}{2}}\bigl(\hat{a}_m + \hat{a}_m^{\dagger} \bigr),
\end{equation}
which gives Eq.~\eqref{eq:phi_port_norm} and from which Eq.~\eqref{eq:h_nm} also follows.

Finally, as discussed in Section~\ref{sec:loss}, we also mention the weakly dissipative case which can be modeled by adding resistances $R_m$ with $m=1, \dots, M$ in parallel with the LC oscillators associated with each mode, as shown in Fig.~\ref{fig::cauer_circuit}. Weak dissipation corresponds to the resistance of the $R_m$s being large.

In this case, following a similar reasoning as that in Section~\ref{subsec:sportcase} for the single RLC oscillator, in the weakly dissipative limit, the impedance matrix can be approximated as

\begin{equation}
\label{eq:z_lossyfoster}
\mat{Z}(s) \approx \sum_{m=1}^M \frac{1}{2 C_0} \biggl(\frac{1}{s + i (\omega_m + i \kappa_m/2)} + \frac{1}{s+i (\omega_m - i\kappa_m/2)} \biggr)\mat{R}_m,
\end{equation}
Recall that the condition for having weak dissipation for each mode is
\begin{equation}
    \kappa_m = \frac{1}{R_m C_0} \ll \omega_m.
\end{equation}

In order to determine the decay rates $\kappa_m$, and thus the effective resistances $R_m$, in the circuit in Fig.~\ref{fig::cauer_circuit}, the procedure is completely equivalent to the one described in Section~\ref{subsec:sportcase} for the single-port case. Let us select an arbitrary port, say port $1$, and let us take all other ports to be open-circuited. If we let a current $I_1(s)$ flow through port $1$, we get

\begin{equation}
    \vect{V}(s) = \begin{pmatrix}
        V_{p1}(s) \\
        V_{p2}(s) \\
        \vdots \\
        V_{pn}(s)  
    \end{pmatrix} = \begin{pmatrix}
        Z_{11}(s) \\
        Z_{21} (s) \\
        \vdots \\
        Z_{M 1} (s)
    \end{pmatrix} I_1(s).
\end{equation}

Note that in this setup port $1$ does not see the effect of the other ports, and thus the circuit can be treated as the single-port case in Fig.~\ref{fig::smlossy_foster}. Thus, port $1$ sees the effective admittance

\begin{equation}
\label{eq:yz11}
    Y(s) = \frac{I_{1}(s)}{V_{p1}(s)} = \frac{1}{Z_{11}(s)}. 
\end{equation}

This shows that, in the multi-port case, we can obtain the approximate damping rate $\kappa_m$ of mode $m$ using Eq.~\eqref{eq:approx-kappa} with the admittance $Y(s)$ in Eq.~\eqref{eq:yz11} in the frequency domain.

\section{LC circuit shunted by a small admittance}
\label{app:lcys} 

\begin{figure}[htbp]
    \centering
    \includegraphics[height=3cm]{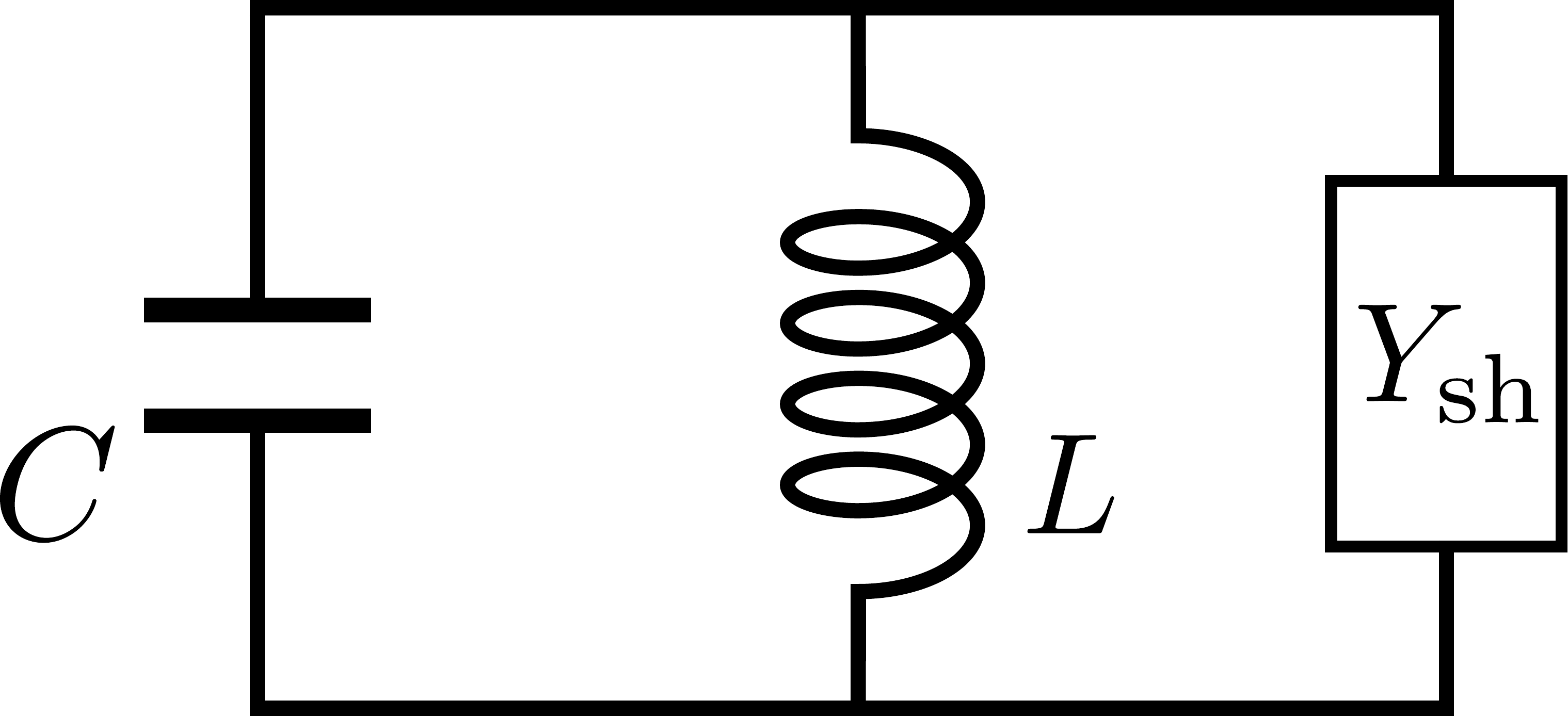}
    \caption{LC oscillator in parallel with a shunting admittance.}
    \label{fig:lcy}
\end{figure}

In this section, we consider the circuit in Fig.~\ref{fig:lcy} focusing on the case of small shunting admittance $Y_{\mathrm{sh}}(s)$. Our goal is two-fold. First, we want to show how the small admittance modifies the resonant frequency of the circuit and induces a decay rate. Second, we highlight the fact that the decay rate we obtain with this simple classical reasoning matches the quantum one in Eq.~\eqref{BKDT1:alt}. In fact, we have already seen in Section~\ref{sec:losscircuit} that Eq.~\eqref{BKDT1:alt} reproduces the decay rate of a parallel RLC oscillator and we wish to confirm that the same holds for a weak shunting admittance. Let us start by considering the total admittance in the Laplace domain of the circuit in Fig.~\ref{fig:lcy} that reads

\begin{equation}
    Y(s) = Y_{\mathrm{LC}}(s) + Y_{\mathrm{sh}}(s),
\end{equation}
with the admittance of the LC oscillator given by
\begin{equation}
    Y_{\mathrm{LC}}(s) = \frac{1}{Z_{\mathrm{LC}}(s)} = \left[\frac{1}{2 C} \left(\frac{1}{s + i \omega_r} + \frac{1}{s - i \omega_r}  \right) \right]^{-1},
\end{equation}
with resonant frequency $\omega_r = 1/\sqrt{LC}$. Close to the pole $s \approx  i \omega$ the admittance of the LC oscillator is approximately
\begin{equation}
    Y_{\mathrm{LC}}(s) \approx 2 C (s - i \omega_r). 
\end{equation}
Thus, within this approximation, the total admittance becomes
\begin{equation}
    Y(s) \approx 2 C (s - i\omega_r) + Y_{\mathrm{sh}}(s) \approx 2 C \left\{s - i \left[\omega - \frac{\mathrm{Im}(Y_{\mathrm{sh}}(i \omega_r))}{2 C} \right] + \frac{\mathrm{Re}(Y_{\mathrm{sh}}(i \omega_r))}{2 C} \right \}.
\end{equation}
Thus, we obtain a new pole
\begin{equation}
    \tilde{p} = i \left[\omega_r - \frac{\mathrm{Im}(Y_{\mathrm{sh}}(i \omega_r))}{2 C} \right] - \frac{\mathrm{Re}(Y_{\mathrm{sh}}(i \omega_r))}{2 C} = i \tilde{\omega}_r - \frac{\kappa_{\mathrm{sh}}}{2},
\end{equation}
where we define the shifted resonant frequency
\begin{equation}
    \tilde{\omega}_r = \omega_r - \frac{\mathrm{Im}(Y_{\mathrm{sh}}(i \omega_r))}{2 C},
\end{equation}
and the decay rate caused by the shunting admittance
\begin{equation}
\label{eq:ksh}
    \kappa_{\mathrm{sh}} =  \frac{\mathrm{Re}(Y_{\mathrm{sh}}(i \omega_r))}{ C}.
\end{equation}
Notice that this procedure is valid self-consistently as long as the new pole $\tilde{p}$ does not deviate too much from the original one at $i \omega_r$, a condition that is satisfied if
\begin{equation}
  \biggl | \frac{\mathrm{Im}(Y_{\mathrm{sh}}(i \omega_r))}{2 C}  \biggr| \ll \omega_r, \quad \biggl | \frac{\mathrm{Re}(Y_{\mathrm{sh}}(i \omega_r))}{2 C}  \biggr| \ll \omega_r. 
\end{equation}
Repeating the procedure close to $s \approx -i \omega_r$, and using the fact that $\forall \omega \in \mathbb{R}$ $\mathrm{Re}(i \omega) = \mathrm{Re}(- i \omega)$ and $\mathrm{Im}(i \omega) = - \mathrm{Im}(- i \omega)$ (see Chapter $4$ in \cite{pozar} for these properties), we also get another pole $\tilde{p}^*$ which is simply the complex conjugate of $\tilde{p}$. Thus, following a similar reasoning as in Section~\ref{subsec:sportcase}, we get that the impedance of the circuit in Fig.~\ref{fig:lcy} can be approximated as

\begin{equation}
    Z(s) = \frac{1}{Y(s)} \approx \frac{1}{2 C} \left( \frac{1}{s - \tilde{p}} + \frac{1}{s - \tilde{p}^*} \right).
\end{equation}

Up to now it is not clear why the coefficient $\kappa_{\mathrm{sh}}$ in Eq.~\eqref{eq:ksh} is called a ``decay rate". In order to understand this we need to analyze the problem in the time domain. Let us suppose a Dirac delta current $\ii(t) = \ii_0 \delta(t)$ flows through the circuit. In the Laplace domain $I(s) = \ii_0$ and thus we get

\begin{equation}
    V(s) = Z(s) I(s) \approx \frac{\ii_0}{2 C} \biggl( \frac{1}{s - \tilde{p}} + \frac{1}{s - \tilde{p}^*} \biggr).
\end{equation}
Now, remember that the Laplace transform of $g(t) = e^{p t} u(t)$ with $u(t)$ the Heaviside step function is given by

\begin{equation}
    G(s) = \int_{- \infty}^{+\infty} dt \,e^{(p - s) t} u(t) =  \int_{0}^{+\infty} dt \,e^{(p - s) t} = \frac{1}{s - p},
\end{equation}
when $\mathrm{Re}(p) < \mathrm{Re}(s)$. Thus, in the time domain we get that for $t \ge 0$

\begin{equation}
    \vv(t) = \frac{\ii_0}{2 C} \left(e^{\tilde{p}t} + e^{\tilde{p}^* t}\right) = \frac{\ii_0}{2 C} \left(e^{i \tilde{\omega}_r t} + e^{-i \tilde{\omega}_r t} \right) e^{- \frac{\kappa_{\mathrm{sh}}}{2} t},
\end{equation}
which shows why $\kappa_{\mathrm{sh}}/2$ is interpreted as a decay rate. 

Finally, as for the RLC circuit analyzed in Section~\ref{sec:losscircuit} (see the derivation of Eq.~\eqref{eq:t1rc}) the quantum formula for the decay rate in Eq.~\eqref{BKDT1:alt} at zero temperature matches the classical derivation in Eq.~\eqref{eq:ksh} for a harmonic oscillator:

\begin{equation}
\label{eq:t1ys}
    \frac{1}{T_1} = \frac{\mathrm{Re}(Y_{\mathrm{sh}}(i \omega_r))}{C} = \frac{\mathrm{Re}(\mathcal{Y}_{\mathrm{sh}}(\omega_r))}{C}=\kappa_{\mathrm{sh}}.
\end{equation}


\bibliographystyle{apsrev4-1}
\bibliography{circuit-quant.bib, revtex-custom}

\end{document}